\documentclass[a4paper,12pt,openany]{book}
\usepackage[inner=30mm,outer=30mm,top=30mm,bottom=30mm]{geometry}
\usepackage{pdfpages}
\usepackage{amsmath, amsthm, amssymb}
\usepackage{yhmath}
\usepackage{graphicx}
\usepackage{epstopdf}

\newtheorem{Theorem}{Theorem}[section]

\newtheorem{remark}{Remark}[section]
\newtheorem{Assumption}{Assumption}[section]
\newtheorem{definition}{Definition}[section]
\newtheorem{notation}{Notation}[section]
\newtheorem{Claim}{Claim}[section]

\usepackage{threeparttable}
\usepackage{caption}
\usepackage{subcaption}
\usepackage{multirow}
\usepackage{algorithm}
\usepackage{algpseudocode}

\usepackage{fancyhdr}
\usepackage{minitoc}
\usepackage{acronym}
\usepackage{setspace}

\begin{document}

\frontmatter
\begin{titlepage}
\centering
{\bf \Large Performance Improvement by Introducing Mobility in Wireless Communication Networks}
\\
\vspace{1cm}
{Guided by:}
\\
\vspace{1cm}
{\bf \Large Hailong Huang}

\vspace{3cm}
{2017.12}
\end{titlepage}

\chapter*{Abstract}
\addcontentsline{toc}{chapter}{Abstract}
Communication technology is a major contributor to our lifestyles. Improving the performance of communication system brings with various benefits to human beings.

This report covers two typical such systems: wireless sensor networks and cellular networks. Both relate to people’s lives closely. For example, people can use wireless sensor networks to get a better understanding of the environment; and use cellular networks to contact with others. We study the influence of mobility in these two networks. In wireless sensor networks, we consider the usage of mobile sinks to collect sensory data from static nodes. The mobile sinks can be attached to robots or vehicles. For the former case, we consider the non-holonomic model and propose a path planning algorithm. The generated paths are smooth, collision free, closed, and letting all the sensory data be collected. For the latter case, we design data collection protocols for a single mobile sink which aim at balancing the energy consumption among sensor nodes to improve the network lifetime. We also design an algorithm for multiple mobile sinks to collect urgent sensory information within the allowed latency.

Regarding the mobility in cellular networks, we mean the service providers are mobile. Conventional approaches usually consider how to optimally deploy static base stations according to a certain metric, such as throughput. However, due to the explosive demands, it will be difficult to satisfy the user requirements by the current facilities. Base station densification is a solution, but it is not cost efficient, because of the high prices of renting sites and backhaul links. Introducing mobility is to let the base stations fly in the sky, resulting in less investment. We consider one of the key issues of using drones to serve cellular users: drone deployment. From simple to complex, we study how to deploy a given number of drones in the area of interest to maximize the served user number; and what is the minimum number of drones and their placements to serve all the users.

We have implemented all the proposed algorithms by either simulations or experiments, and the results have confirmed the effectiveness of these approaches.

\dominitoc
\tableofcontents
\listoffigures \addcontentsline{toc}{chapter}{List of Figures}                   
\listoftables \addcontentsline{toc}{chapter}{List of Tables}

\mainmatter
\setcounter{mtc}{5}
\chapter{Introduction}
\minitoc
\section{Communication Networks}
Communication technology is a major contributor to our lifestyles. With regard of communication, data travels from devices to information sinks, and vice verse. Existing techniques to realize such communication pattern include wireless sensor networks, cellular network, etc. A traditional wireless sensor network consists of one or several static sinks and a set of static sensor nodes. With the years of development, it is found that the hop spot issue is in essence unavoidable in the traditional wireless sensor networks. In order to tackle this issue, mobility (mobile sinks) has been introduced to wireless sensor networks. Cellular network is another type of communication networks, where cellular users are in essence mobile. Conventional cellular networks usually use base stations to provide service to users. Users may experience bad service if a large amount of users request data simultaneously. One promising solution is to deploy Unmanned Aerial Vehicles (UAVs), represented by drones, airships or balloons, to serve as flying base stations. Consider both wireless sensor networks with mobile sinks and cellular networks with drones under the framework of mobile networks, the mobile sinks and drones can be regarded as servers, and the sensor nodes and cellular users can be treated as users. The major difference between these networks is that the users in wireless sensor networks are sensor nodes, which are mostly static; while the users in cellular networks are mobile. This report studies how to improve system performance by using mobility in these networks.
\subsection{Wireless Sensor Networks}
Wireless sensor networks (WSNs), a multidisciplinary research area, were primarily motivated by military applications. Benefiting from technological development, the production cost of sensor nodes created interest in diverse applications, such as habitat monitoring \cite{mainwaring2002habitat}, important machine operation monitoring \cite{kim2007health}, industrial process control \cite{zhao2011wireless}, intrusion detection \cite{butun2014survey}, disaster management \cite{saha2007disaster}, etc. A WSN, consisting of a number of distributed autonomous sensor nodes, has been regarded as a promising means to collect diverse information sources from the physical world, such as temperature, motion, seismic waves, and many others, which definitely helps people to have a better understanding of the environment.

Research activities conducted on WSNs can be categorized into three groups: application domain, hardware design and software development. According to various application domains, the requirements of hardware and software in different WSNs may be varying from each other.

In terms of hardware, the focus is on the production of the main components of a sensor node such that it can reliably work. Typically, an autonomous sensor node is composed of: sensor(s); memory; controller; transceiver; and power source. The main component of a sensor node is the sensor(s), which is to detect the nearby environment and the type(s) of the sensor(s) is application depended. A sensor node usually periodically senses the environment and the information can be transmitted out to others or stored locally in the memory. The controller performs the pre-configured tasks, processes the sensory data and controls the functionality of other components in the node. The transceiver is a single device, which is used to both transmit or receive data packets. The power resource supports power for sensing, communication, data processing and other activities. 

Researchers also focus on the software in the sensor nodes. When a WSN is deployed, it is expected to work reliably and automatically as long as possible. Suppose all the hardware can support this, then one of the key points is the management of the power source. Although the renewable battery has been seen in the market, it is still in development. So, how to make fully use of the limited battery is a significant problem in designing a WSN.
\subsection{Cellular Networks}
The considerable growth in demand for higher data rates and other services in cellular networks have accelerated the need to develop more innovative communications infrastructures. Neither the conventional macro cell facilities nor the small cells are able to address such issue in an cost efficient way. Because deploying more of them not only increases the cost including the equipments and site rental, but also brings with other issues such as that in the non-peak period, there will be a high percentage of facilities under low load. In such a scenairo, the drones can be integrated into the current communication systems \cite{UAV2016survey} to enhance service in areas with dense traffic, which is believed to be a cost efficient solution.

There are various problems we have to consider to introduce drones into cellular networks. In terms of system architecture, how will these drones collaborate the existing base stations? Regarding only the drone layer, where should they be to serve users? How long should they work? How to recharge the batteries? etc. These questions need to be answered before practically using drones to assist serving cellular users.
  
\section{Research Question}
The main topic of this report is how to make use of mobility to improve the performance of mobile networks. Regarding these two types of mobile systems, we study the below questions:

Considering the fact that the technology of renewable energy in WSNs is not mature \cite{jiang2005perpetual, kansal2007power, park2008energy} and the difficulty of recharging the distributed sensor nodes, minimizing energy consumption and improving wireless sensor network lifetime is usually the concern of system designers. Applying mobile sinks is promising in saving energy resource for the sensor nodes. So our first research question is how to make use of mobile sinks to improve system energy efficiency and network lifetime for wireless sensor networks.

Since the outdoor cellular data for personal use has been rising steadily, and the existing infrastructures usually have capacity limitations, cellular users may experience bad service especially when a large number of users request data simultaneously. Deploying more BSs is able to meet the increasing traffic demand. This solution, however, not only results in more cost, but also brings with other problems. Dense BSs may lead to high interference and a high percentage of BSs may have low utility in non-peak period. In such a scenario, the utilization of drones, which work as flying BSs, is a preferred solution, comparing to that of BS densification. Our second research question is how to make use of flying BSs to improve the user experience.
\section{Contributions}
Having the above questions in mind, we conduct extensive research from various aspects. In particular, for the first research question, we consider the controllable mobility and constrained mobility. Controllable mobility refers to that the mobile sinks can be fully controlled by network designers without any constrained. Generally, such mobility is easy to handle (see Chapter \ref{path_planning}). Constrained mobility refers to that the mobile sinks have certain constraints. This model is more practical in some applications, such as collecting data in urban environment WSNs (see Chapter \ref{cluster_ms}, \ref{cluster_cs} and \ref{cluster_um}). Besides the work proposed to make use of mobility to improve the network performance, we study the sink tracking issue, which can provide the network with the current locations of mobile sinks (see Chapter \ref{tracking}).For the second research question, we also consider the constrained mobility. Different from Chapter \ref{path_planning}, \ref{cluster_ms}, \ref{cluster_cs} and \ref{cluster_um}, in Chapter \ref{drone} we employ drones to serve cellular users. The main contributions of this report have been summarized as follows:

\begin{itemize}
\item Focusing on \textit{\textbf{controllable}} mobility, the first contribution is the proposal of a path planning approach which resolves several practical issues not having been sufficiently tackled yet when controllable mobile sinks are used to collect sensory data from sensor nodes. Different from existing methods which simply regard the mobile sinks as moving points, we use the unicycle robots (with constant line speed and limited angular velocity). Such model is closer to practice than the point-wise model. Taking into account this model, we need to design paths which are: smooth, collision free from sensor nodes and obstacles, closed, and letting the sinks read all the sensory data from the sensor nodes. Regarding these features, we define the term of viable path and propose two approaches for a single mobile sink and a set of mobile sinks respectively, which are named as: Shortest Viable Path Planning (SVPP) and $k$-Shortest Viable Path Planning ($k$-SVPP). We have shown that SVPP and $k$-SVPP are effective to design viable paths for unicycle mobile sinks with bounded angular velocity and can save much energy for the nodes compared to the multihop transmission. 
\item The second contribution is an approach using a single \textit{\textbf{constrained}} mobile sink with fixed path to collection data from sensor nodes. The proposed protocol aims at balancing the energy consumption, including energy expenditure to transmit data packet and network overhead across the network, to make the network operate as long as possible with all nodes alive. We design an energy-aware unequal clustering algorithm and an energy-aware routing algorithm. Theoretical analysis and simulation results confirm the effectiveness of the proposed approach against the alternative methods. 
\item Similar to the framework of the second contribution, the third one also considers the scenario of using a single \textit{\textbf{constrained}} mobile sink with fixed path. The difference lies in that we combine the technique of compressive sensing (CS) and clustering: within clusters, raw reading is transmitted; while CS measurement is transmitted between clusters and MS. We present an analytical model to describe the energy consumed by the nodes, based on which we figure out the optimal cluster radius. We present two distributed implementations, whose message complexities at a node are both $O(1)$. We conduct extensive simulations to investigate their performance and compare with existing work in terms of the energy efficiency. 
\item The forth contribution is the proposal of an algorithm which targets on delivery unusual message to the mobile sinks within the allowed latency. Same to the second and third contribution, the \textit{\textbf{constrained}} mobile sinks are attached to vehicles with fixed trajectories, e.g., public transportation vehicles. Instead of single mobile sink, we use multiple mobile sinks which are attached to the buses. The proposed data collection consists of sensor nodes, bus stop nodes (which work as the interface between sensor nodes and mobile sinks), and mobile sinks amounted on the buses. Upon detecting any unusual message, the source sensor node transmits the information to a set of selected target bus stop nodes. When buses pass by, the information is uploaded. The key issue here is how the source node selects the bus stop nodes. We formulate it as an integer programming problem. We take into account the realistic features of the buses, such as the timetable and uncertainties in the arrival time as well as the stopping duration. We conduct simulations and also experiments to test our approach. We show the proposed approach can deliver the unusual message to mobile sinks within the allowed latency with higher reliability and efficiency than the alternatives. 
\item The final contribution of this report lies in the study of drone deployment problems in cellular networks. We formulate the \textit{\textbf{constrained}} drone placement problems based on a novel street graph model associated with the UE density function (built up based on the real dataset). The performance in terms of serving UEs is competitive with existing work. The advantage is that the 2D projections of drones obtained by our approaches are always valid, since they will be on the streets. Furthermore, we provide solutions to the multiple drone placement problem and the problem of minimum number of drones to achieve the required QoS level from the point of Internet Service Provider (ISP), which can serve suitable guidelines in practice.
\end{itemize}

\section{Organization}
The organization of the rest content is briefly outlined: Chapter \ref{literature} reviews the related work. Chapter \ref{tracking} studies a basic problem when mobility is used to serve WSNs, i.e., sink tracking. Chapter \ref{path_planning} to \ref{drone} present the main work of the report. Specifically, Chapter \ref{path_planning} considers the scenario of using \textit{\textbf{controllable}} mobile sinks to collect data from sensor nodes. The problem we focus on is the path planning for the mobile sinks which are modelled as dubins car. Chapter \ref{cluster_ms} and \ref{cluster_cs} consider the scenario of using a \textit{\textbf{constrained}} mobile sink for data collection. The focus is how to efficiently collect data from the sensor nodes such that the network lifetime can be maximized. Chapter \ref{cluster_um} considers the scenario of transmitting the detected information about urgent events to \textit{\textbf{constrained}} mobile sinks within the allowed latency. Chapter \ref{drone} studies the \textit{\textbf{constrained}} drone deployment problems. Finally, Chapter \ref{conclusion} of this report summarize the key results and highlights the possible future research directions for the problems and solutions presented in the report.
\chapter{Related Work}\label{literature}
\minitoc
\section{Overview}
There are lots of existing work in literature on the topics of wireless sensor networks and cellular networks. In this chapter, we only present a survey of work related to our studied problems, i.e., how to use mobile sinks to improve network energy efficiency and/or network lifetime in WSNs, and how to improve the user experience by drones in cellular networks. Note, there are other options to improve the system performance of WSNs, i.e., adding energy to the system by either energy harvesting or wireless charging techniques. Because of the breadth of this report, we refer readers to \cite{gilbert2008comparison, seah2009wireless, vullers2010energy, xie2013wireless, xie2012making} and the references therein for more comprehensive reviews. 

In this chapter, besides discussing the most related work to our problem, we will also highlight the positions of our proposed approaches in the literature. The reviewed approaches fall into the structure shown in Fig. \ref{structure}. 

\begin{figure}[t]
\begin{center}
{\includegraphics[width=0.9\textwidth]{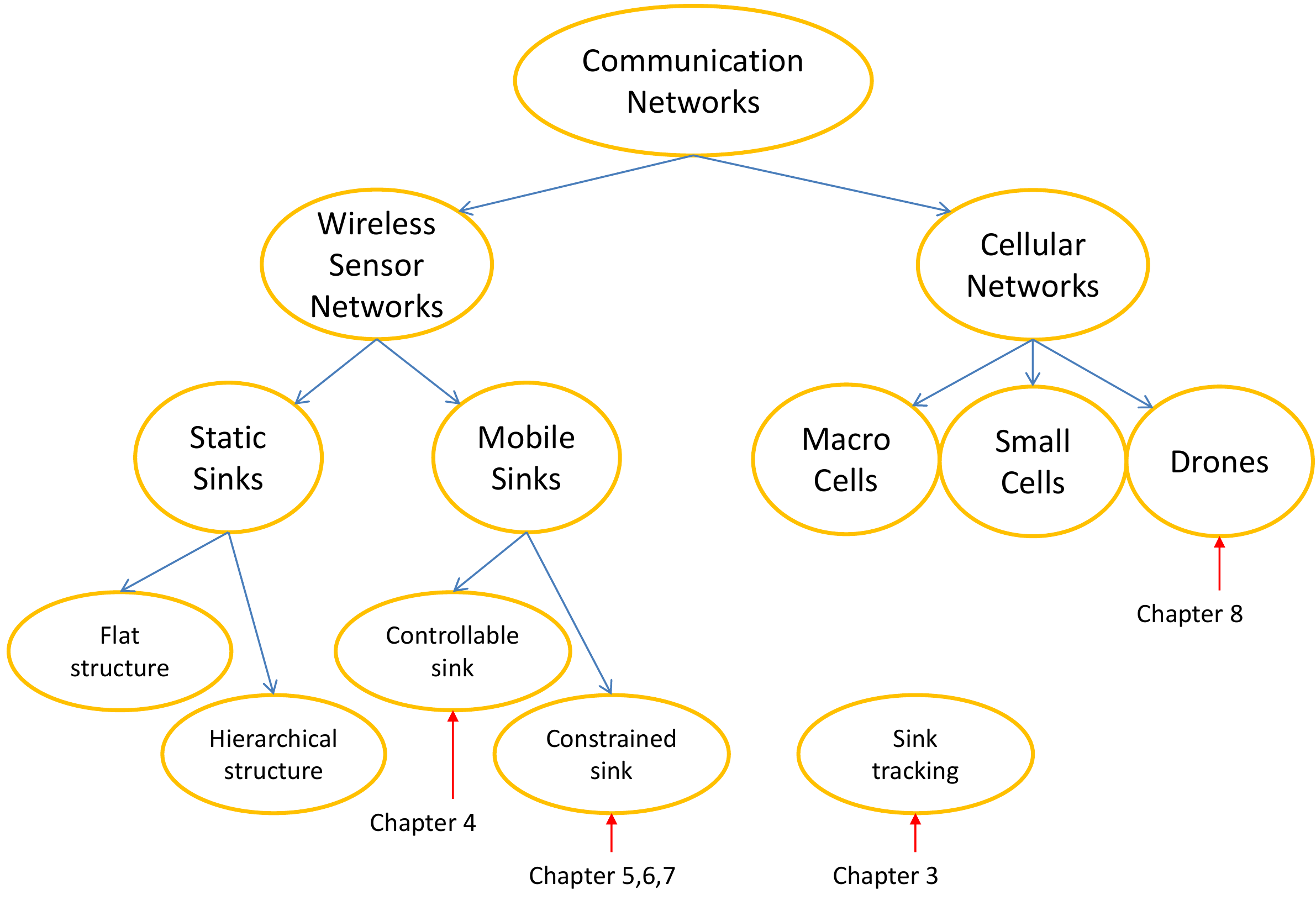}}
\caption{The catogery of the reviewed approaches and the positions of our contributions.}\label{structure}
\end{center}
\end{figure}
\section{Wireless Sensor Networks}
A wireless sensor network (WSN) consists of a set of wireless sensor nodes which work together to achieve single or multiple goals, e.g., environmental morning, intrusion detection, target tracking, etc. The wireless sensor nodes usually have on-board sensors, transmitter, receiver, a micro computer to process the sensory data and a battery. Driven by advances in manufacturing of high density electronics, the wireless sensor node becomes tiny. It handles with various types of sensors and the abilities of communication and data processing have been improved significantly. However, one bottleneck is the slow development in battery technologies. Although renewable energy has been introduced, its application in wireless sensor networks is still not mature. Thus, the energy constrained nature of sensor nodes is a major certain in the development of data collection protocols. Based on the mobility of base stations, we classify data collection approaches in wireless sensor networks into two categories: static sinks and mobile sinks. Below, we present a brief review of these two categories of approaches.

Recently, equipping sensor nodes with mobility can bring many advantages. 
For example, the number of nodes is decreased dramatically to guarantee the coverage of a given area due to the mobility \cite{cheng2009comlett, cheng2011sweep}.
Moreover, exploiting mobile sensor nodes to achieve a barrier coverage of the sensing field can be used to minimize the probability of undetected intrusion in intrusion detection applications and a sweep coverage can be used to maximize the event detection rate and in the meantime minimize the missed detections \cite{cheng2011barriersweep}.
Considering the fact that most existing WSN application still use static sensor nodes, those approaches which employ mobile sensor nodes, are out of the scope of this report.
Interested readers are refered to \cite{cheng2009comlett, cheng2011sweep, cheng2011barriersweep} and the references therein. Throughout this report, the wireless sensor nodes considered are static. 

\subsection{Traditional data collection approaches}
Traditional data collection approaches refer to those which use static sinks. The typical scenario is that a set of sensor nodes is deployed in the area of interest. The sensor nodes sense ambient conditions, transform the measurements into certain signals that can be processed to reveal the characteristics about the phenomena of interest, then send the data together with the location information to the static sink. At the sink, a large amount of sampling data at various positions for the same timestamp constructs a map of the area of interest. The convention way to transmit the sensory data to the static sink is through multi-hop communication, i.e., a sensor node transmits its sampling data to another node which is closer to the sink. It is easy to imagine that the nodes nearby the sink are overloaded than those far away. If the sensor nodes all have the same initial energy, the former will run out of battery earlier than the latter. The sink is isolated from the rest of the sensor nodes if the nearby nodes die. Approaches arming at improving energy efficiency and network lifetime can be classified into two categories according to the network structure: flat-based routing and hierarchical-based routing.

\subsubsection{Flat-based routing protocols}
Traditional flat-based routing protocols include flooding protocol \cite{flood02} or gossiping-based routing \cite{gossip06}. The shortcomings of flooding based protocol include implosion, which is caused by duplicate packets sent to the same node, and resource blindness without consideration for energy constraints. Gossiping avoids such issue by randomly selecting node to send the data packet rather than broadcasting the data packet. However, the propagation of data causes high delays. To improve the energy utilization and network lifetime, the redundancy within the sensory data must be exploited \cite{al2004routing}. Directed diffusion \cite{intanagonwiwat2000directed} aggregates the data coming from different sources by removing redundancy, reduces the number of transmissions, thus saves network energy and prolongs the network lifetime. Assuming the fixed BS, i.e., the direction of routing is always known, Minimum Cost Forwarding Algorithm \cite{ye2001scalable} allows a node only to maintain the least cost estimate from itself to the BS, instead of a routing table. When a node transmits a data packet, it broadcasts it to the neighbour nodes. Only that which is on the least cost path between the BS and source node rebroadcasts it. Gradient-based routing \cite{schurgers2001energy} memorizes the number of hops when the interest is diffused through the entire network. In particular, each node manages the depth of the node, which is the minimum number of hops to reach the BS. The difference between two nodes' depths is defined as the gradient on that link. A data packet is forwarded on a link with the largest gradient. 

Compressive Data Gathering (CDG) is proposed in \cite{CDG09}. Consider the sensor nodes form a particular tree with the BS as root. Rather than transmitting the raw sensory data, a well defined measurement is relayed to BS.  Benefiting from Compressive Sensing (CS) theory \cite{CS06_1, CS06_2}, the original readings can be recovered from the well defined measurement. Thus, the number of transmitted data packets is significantly reduced. In \cite{CDG09}, the leaf nodes initiate data transmission. Any node multiplies a measurement item to its raw reading, adds it to the sum of the measurements received from all its children, and sends the final sum to its parent. Then, all nodes send the same number of messages regardless of their hop distance to the BS. More importantly, the transmission load is uniform in the entire network, which avoids the energy hole issue.

\subsubsection{Hierarchical-based routing protocols}
In hierarchical-based routing protocols, the roles of sensor nodes are different. Some of them act as cluster heads (CHs), which can handle the cluster members (CMs) and execute local aggregation. So this kind of protocols is also known as cluster protocols. Low Energy Adaptive Clustering Hierarchy (LEACH). \cite{heinzelman2000energy} randomly selects CHs and rotates the role to evenly distribute the energy load among the entire network. CMs are to perform the sensing while CHs are to process the received data packets and transmit the aggregated data packet to BS. Unlike randomly selecting the CHs, \cite{estrin1999next} selects the CHs according to their residual energy and degree. Hybrid Energy-Efficient Distributed Clustering \cite{younis2004heed} is another cluster protocol. It tries to evenly distribute the CHs. The probability that two closely located nodes both becoming CHs is much smaller than LEACH. \cite{chen2009unequal} proposes Unequal Routing Clustering (URC). The authors point out that the sizes of the clusters near the sink should be smaller than those far away. Thus, the energy consumption on intra-cluster communication of CHs near the sink is reduced, and these CHs can spend more energy on inter-cluster communication, i.e., relaying data.

CS based routing approaches can also be extended to the hierarchical-based routing protocols. In CDG \cite{CDG09}, every node transmits the same number of packets for one reading. However, when multihop communication is used, the leaf nodes only need to send one packet containing its original reading. Thus, one disadvantage of CDG is the increasing energy consumption at the nodes close to the leaf nodes. Motivated by this observation, the hybrid CS approach has been proposed \cite{xiang11}. In the hybrid CS method, the nodes close to the leaf nodes transmit the raw readings without using CS method, while the nodes close to the sink transmit CS measurements. In this way, the overall message complexity is further reduced. The authors of \cite{jia14} proposes a clustering method that uses the hybrid CS. Within clusters, cluster members (CMs) send raw readings to their cluster head (CH). All the CHs and the sink form spanning trees with the sink as the root and CHs transmit CS measurements to the sink. This approach further reduces the message complexity in the entire network. The authors of \cite{CDG16} also use routing trees to collect data, but they claim that the links on the routing trees should be scheduled for transmissions such that adjacent transmissions do not cause harmful interference on one another (thus corrupting the compressed measurements) while maintaining a maximum spatial reuse of the wireless spectrum.

Although these approaches are able to improve the network lifetime, the energy hole issue is still unavoidable. In the next part, we present a review of work where mobility is used.
\subsection{Data collection approaches based on mobile sinks}\label{ME_collection}
Based on the mobility of sinks, we can divide the approaches using mobile sinks into three groups: random, controllable, and constrained. The random mobility refers to that the mobile sinks randomly move in the sensing area. The controllable mobility means that the mobile sinks can be fully controlled to visit any one of the sensor nodes. constrained mobility refers to the case that the mobile sinks have some frequently visited positions or follow some trajectories. One extreme case of the constrained mobility is that the mobile sinks are amounted on some vehicles, such as buses, which paths have been predetermined. Here, we only consider the controllable and constrained mobile sinks.
\subsubsection{Using controllable mobile sinks}
When the mobile sinks are fully controllable, the question of how the mobile sinks move is often asked. Thereby, lots of publications focus on the path planning for the mobile sinks such that some certain metrics is met.

One of the most considered metric is the path length. Path length is a reflect of data collection delay. Traditional Travelling Salesman Problem \cite{TSP85, TSP94} is a suitable formulation for the data collection problem using a mobile robot. Given a point set (the set of sensor nodes) in a plane, the objective is to find a minimum-length path that visits each node exactly once. Here the state of the robot can be represented by two-tuple $(x,y)\in \mathbb{R}^2$ representing the robot's position. A variant of TSP is called Asymmetric TSP (ATSP) \cite{L92} where the distances between two nodes may be different in two directions.

With the cost of the robot movement, the sensor nodes' transmission energy dissipation can be reduced. On the other hand, it leads to a long collection time due to the robot physical speed. To tackle such issue, a class of energy dissipation and delivery delay trade-off approaches has been proposed. For example, in TSPN based approaches\cite{TSPN07}, the point set is extended to a region set and communication between the node and robot is available once the robot is within the region. Considering the fact that the communication regions of sensor nodes may overlap with each other in dense networks, the path planning problem can be formulated as Generalized TSP (GTSP) \cite{TSPN13CSS}. It was shown GTSP can be transformed to ATSP through the Noon and Bean transformation \cite{NB91}.

One point worth mentioning is that the TSP based approaches focus more on the high level path planning while pay less attention to the mobile sinks' mobility constraint. For example, when a non-holonomic robot carries a mobile sink to collect data from the sensor nodes, the mobility constraint must be considered. As a result, another direction of research has been created. A typical approach is TSP for Dubins vehicles (DTSP), introduced in \cite{SFB05}, is to find the minimum-length path satisfying the bounded curvature, given a point set in a plane. Various research effort has been put onto the problem of DTSP, see e.g., \cite{tang2005motion, mcgee2006path, rathinam2007resource, savla2008traveling, cons2014integrating, shima2014motion, isaiah2015motion}. Following the idea of DTSP, the authors in \cite{visitingcircle2015} proposed a Smooth Path Construction (SPC) scheme to plan path for the robot based on a turning circle model. The produced path is smooth; however, if applied to grounded robots, the robots may collide with sensor nodes. As a variant of DTSP, TSP for Dubins vehicles with neighbourhoods (DTSPN) has been studied. 

Another aspect needs to be accounted in the consideration of low level path planning for realistic robots, i.e., obstacle avoidance. There are lots of existing robot navigation approaches with obstacle avoidance. Temporarily following the boundary of an obstacle is a standard method employed by many obstacle avoidance algorithms. The basic idea is to switch between two modes \cite{lumelsky1987path}: (a) the movement to-wards the objective and (b) the movement around an obstacle. \cite{lumelsky1990incorporating} proposes an improvement, which is an improvement to the condition that the robot uses to stop contouring an obstacle and resume the movement to the destination the so called leaving condition. Such enhancement leads to shorter path from origin to destination. 
\cite{kamon1997sensory} makes an alteration on the leaving condition which allows the robot to quit the  obstacle boundaries as soon as the global convergence is guaranteed, based on the information of whether the destination is in the free range direction.
\cite{langer2007k} further improves the leaving condition by exploiting the local sensing: the robot leaves the origin-destiny line when a new collision is detected.

The basic idea of switching between two modes also applies in moving obstacle avoidance. \cite{van2008planning} discusses the problem of safe path planning among unpredictably mobile obstacles. Specifically, the starting point as well as the destination are given. The sizes of the obstacles grow over time. It presents an approach to compute the minimal time cost path between origin and destination in the plane that avoids these growing obstacles. 
\cite{toibero2009stable} considers the wall-following task and presents a continuous controller for wheeled mobile robots. The proposed control system has three modes, and each one is designed to solve a specific instance of the navigating task: re-orientation (to avoid collisions), wall-following (for walls with nearly constant contour) and circle-performer (to recover the wall information in open corners.) The controller switches among these modes based on a switching logic which depends on the odometry and distance information.
\cite{MTS11} considers the guidance and control of an autonomous vehicle in problems of  border patrolling and avoiding collisions with moving and deforming obstacles. 
It consists of properly switching between moveing towards target in straight line, and bypassing obstacles at a pre-defined distance following a sliding mode control law.
The obstacles here are assumed to be static and convex.
\cite{matveev2012real} extends \cite{MTS11} to the case of dynamic environments which are cluttered with arbitrarily shaped obstacles. The ability of reaching destination globally in such environments was illustrated by experiments with a real wheeled robot as well as simulations.
A biologically inspired algorithm is proposed in \cite{savkin2013simple}.
Mathematically analysis of this algorithm is provided for the case of round obstacles which move with constant velocities. Simulations together experiments demonstrate that the algorithm performs better than some well-known methods such as artificial potential field and navigation laws based  on the  obstacle velocity.
\cite{savkin2014seeking} proposes an algorithm for collision free navigation of a non-holonomic robot in unknown complex dynamic environments with mobile obstacles. This algorithm is based on an integrated representation of the information about the environment that does not require to separate obstacles and approximate their shapes by discs or polygons or any information on the obstacles' velocities. 
\cite{matveev2015globally} proposes a reactive algorithm to navigate a planar mobile robot in densely cluttered environments with unpredictably moving and deforming obstacles. It uses omnidirectional vision of the scene up to the nearest reflection point. Apart from access to the desired azimuth, it assumes no further sensing capacity or
knowledge of the scene configuration.
Many other available schemes can be found in a recent survey \cite{HMS15} and book \cite{MSHW15}.

The above mentioned approaches focus more on robotics. From the point of sensor networks, some features should also be considered.

In a WSN, the data loads of the nodes may be different when event-driven sensors are used. If the nodes' storage size is fixed, visiting every node with the same frequency leads to the issue of buffer overflow for the nodes whose data generation rates are large. Considering this, the authors in \cite{MES04} studied the Mobile Element Scheduling (MES) problem. Different from TSP framework, one node may be visited multiple times depending on its data generation rate. To address MES problem, three algorithms, i.e., Earlist Deadline First (EDF), EDF with k-lookahead, and Minimum Weight Sum First (MWSF) were presented. They extended the work to the case where there are multi mobile robots to be scheduled  \cite{MES07}. The authors in \cite{PBS05} also considered MES problem and proposed a Partitioning-Based Scheduling (PBS) algorithm.

For the long data delivery delay issue, TSPN based approaches can reduce the latency while its potential is limited. Another approach from sensor networks is to select a subset of nodes as Rendezvous Points (RPs) \cite{rendezvous07, rendezvous12, adaptivestop14, MA07}. The other nodes forward the extracted data to RPs, and the mobile robots just visit RPs. Thus the data delivery delay can be reduced significantly. \cite{rendezvous07} formulated the Minimum-Energy Rendezvous Planning (MERP) problem and proposed two algorithms to address it. \cite{rendezvous12} proposed a Mobicluster protocol dealing with some practical issues, for example how to handle the case that RPs run out of energy, which were not covered by \cite{rendezvous07}. \cite{adaptivestop14} falls into the same category. But instead of searching the sensor nodes to select RPs, the authors created a lattice graph of the field and set the vertices as RPs. Two traversal schemes were proposed: Deterministic Walk (DW) and Biased Random Walk (BRW). DW guides the robot to traverse the network following a fixed visiting order, while BRW adapts the robot to the dynamic environment. \cite{MA07} also belongs to this category. The authors proposed a heuristic mechanism to determine RPs and then constructed path for SenCar. Note, the RPs in \cite{adaptivestop14, MA07} may not necessarily be the sensor nodes, which differs from the work \cite{rendezvous07, rendezvous12}. 
Comparing to TSPN, the rendezvous-based approaches are able to cutback the  path length decidedly and consequently shorten the collection time. 

One implied assumption behind the already discussed schemes is that the robot is able to move freely in the field. However, in many realistic applications, the movement of the robot is restricted \cite{pathconstrained11IEEE}, such as using public transport vehicles \cite{pathconstrained07}, \cite{pathconstrained11}. \cite{pathconstrained07} considered the application of highway traffic surveillance using sparsely deployed sensors. Assuming the locations and data loads of the nodes are known, the authors designed a Transmission Scheduling Algorithm (TSA) to determine the time slots for each sensor node. Different from \cite{pathconstrained07}, the authors in \cite{pathconstrained11} studied the data collection problem in a large scale network using path-constrained mobile robots. Similar in spirit to the rendezvous-based approaches, a subset of nodes which are geographically within the communication ranges of the mobile robots serve as RPs. 

Some other studies have also considered the case of utilizing multiple robots. Vehicle Routing Problem (VRP) \cite{DR59}, a generalization of Travelling Salesman Problem (TSP), is to design paths for a fleet of vehicles. Also many variants of VRP have been proposed taking into account varying factors. More details of VRP can be found in the recent book \cite{TV14}. Generally, the objective of the VRP framework is to minimize the total cost of all the vehicles' paths. Applying this kind of approaches may lead to the situation: some paths' costs are greater than the others.

Another approach is the cluster-based approach, where the sensor nodes are divided into a set of clusters and one robot serves one cluster. The basic consideration of this approach is that usually the number of robots is smaller than that of the nodes. Then the problem can be regarded as a combination of source assignment and path planning. Many existing clustering approaches have been used to separate the sensor nodes and then the problem turns to solving the problem with single robot in each cluster \cite{visitingcircle2015, TO05, RSD07}. For example, K-means was used in \cite{visitingcircle2015, TO05} and minimum spanning tree algorithm was used in \cite{RSD07}. It is easy to understand that the cost in each cluster depends on the clustering results.

To obtain a set of paths with similar costs, $k$-Travelling Salesman Problem ($k$-TSP) was studied and the authors in \cite{FHK78} proposed several heuristics including $k$-NEARINSERT, $k$-NEARNEIGHBOR, and $k$-SPLITOUR, among which $k$-SPLITOUR is the simplest and performs far superior. Differing from VRP, the goal of $k$-TSP is to minimize the length of the maximum path and the $k$-SPLITOUR algorithm starts from a TSP path and then splits it into $k$ subpaths with similar lengths. As an extension, the authors in \cite{BI09} introduced the contact time for downloading data from sensor nodes. As the generalization of \textit{k}-TSP, the authors in \cite{KAUWWT12, KUAWWT14} studied $k$-Travelling Salesman Problem with Neighbourhoods ($k$-TSPN), where the robots have different initial locations. 
\subsubsection{Using constrained mobile sinks}
One shortcoming of using fully controllable mobile sinks is the increased data collection delay, due to the low physical movement. A number of approaches using trajectory constrained MSs to collect data in WSNs have been proposed. The MSs can collect data either through single hop communication \cite{SONG07, MEHRABI16, CHAK03} or multihop communication \cite{SOMA06, GAO11}. The authors of \cite{CHAK03} consider the scenario where a MS is installed on a bus which moves on its fixed path periodically and collects data from a set of sensor nodes deployed near the path. They propose a queuing formulation to model the process of data collection and show that using constrained mobility can lead to large energy saving over convention static WSNs. Further they propose a communication protocol to assist gathering data by MS. Under the similar context, \cite{SONG07} focuses on the scheduling problem in node-sink transmission and a trade-off between the probability of successful information retrieval and node energy consumption is studied. Different from \cite{SONG07} which considers sparsely deployed network, the authors of \cite{MEHRABI16} focus on dense networks. They consider the maximization of data collection throughput by dividing the traversing time into a set of time slots with equal duration and studying the problem of assigning nodes to the time slots. One defect of these approaches is that they all use single hop communication, which requires that the sensor nodes are deployed within the communication range of MS when it moves on the trajectory.

In \cite{SOMA06}, the assumption, i.e., all the nodes are located within in the communication range of MS, is removed while multi-hop communication is used. During the movement, MS broadcasts a packet continuously within its communication range. The nodes that can hear this packet are called subsinks. Subsinks then forward the packet within their communication ranges. 
When MS finishes its trip, every node knows its shortest hop distance to a subsink as well as the shortest path towards the subsink. From MS's next trip, nodes transmit their sensory data to the corresponding subsinks and when MS comes, the subsinks upload the data. Compared to the single hop case, using multi-hop communication improves the applicability and scalability of system. The authors of \cite{GAO11} consider the same problem as \cite{SOMA06}. They first point out that the shortest path based routing leads to an unbalanced assignment of nodes to subsinks, which further results in that some subsinks having long contacting time with MS is associated with a small number of nodes. Thus, the subsinks may not manage to upload all its buffered data. Considering this, they formulate a constrained assignment problem such that each subsink is associated with an appropriate number of nodes, which enables more data can be collected. By doing this, the throughput is improved. However, some nodes may be associated to a far away subsink, thus relaying their data consumes more energy than the shortest path routing.

The above approaches using either single hop or multihop communication are based on the flat network structure. While the deployment of large scale WSNs and the need for aggregation necessitate the efficient management of the network topology to balance the load and prolong the network lifetime. Clustering, as opposed to direct single hop communication schemes, has been shown to be an effective approach for organizing the network, which improves energy efficiency, reduces packet collisions, and results in increased network throughput under high load \cite{MAMALIS09}. Researchers have proposed many well known clustering approaches, such as LEACH \cite{heinzelman2000energy}, EEHC \cite{bandyopadhyay2003energy}, HEED \cite{younis2004heed}, and their extensions. However, they all focus on static networks. The authors of \cite{KONS12} consider the same context as \cite{SOMA06, GAO11}, but clustering is introduced. The data collection protocol presented in \cite{KONS12}, MobiCluster, is based on a clustering algorithm, called Unequal Routing Clustering (URC)  \cite{CHEN09}. The authors of \cite{CHEN09} point out that given a network with a static sink, the sizes of the clusters near the sink should be smaller than those far away. The reason is as follows. The CHs near the sink have heavier burden of relaying data compared to the CHs far away. CHs also need to collect data from their CMs within cluster and aggregate the collected data. Thus, if the cluster size is identical across the network, the CHs near the sink may run out of energy much more quickly than those far away (funnelling effect). To avoid this, an effective approach is to construct unequal clusters, i.e., the clusters near the sink have smaller sizes while the clusters far away have larger sizes. Thus, the energy consumption on intra-cluster communication of CHs near the sink is reduced, and these CHs can spend more energy on inter-cluster communication, i.e., relaying data. The authors of \cite{KONS12} adopt a simplified version of URC where only two sizes are considered, and apply it to the scenario of using MS. One drawback of URC is it assumes that the sink is able to broadcast a packet to all the nodes. Similar to the defect of \cite{SONG07, MEHRABI16, CHAK03}, this assumption limits the scalability of URC \cite{CHEN09}, so as MobiCluster \cite{KONS12}. Another drawback is the cluster size depends on a distance information derived from signal strength, which may not be reliable in harsh environment. With regard to the disadvantages of URC, the authors of \cite{WEI11} propose an energy-efficient clustering algorithm (EC), where the cluster size relates to the hop distance from a CH to the sink, rather than the Euclidean distance. Apparently, compared to URC, EC is more appropriate for large-scale networks. \cite{KONS12, CHEN09, WEI11} assume the nodes are uniformly deployed. However, in some applications the nodes are not uniformly deployed and the approaches in \cite{KONS12, CHEN09, WEI11} are not guaranteed to work well. As shown in \cite{wu2008avoiding}, the unbalanced energy depletion among all the nodes in the circular multihop sensor network (modeled as concentric coronas) is unavoidable, due to the inherent many-to-one traffic pattern. The authors propose a strategy which distributes sensor nodes in a non-uniform manner, to achieve better balance of energy depletion across the network by regulating the number of nodes in each subarea. The authors of \cite{cardei2008non} also focus on nonuniform deployment of sensor nodes. They create a sensor movement plan to achieve the desired sensor densities for uniform energy depletion. 

\subsection{Sink tracking}\label{tracking_related}
One fundamental requirement in the many aforementioned approaches is the knowledge of the locations of sensor nodes and mobile sinks. Generally, these locations refer to physical coordinates, which are obtained by Global Positioning System (GPS) \cite{gps04}, Received Signal Strength Indication (RSSI) \cite{RSSI08}, Angle of Arrive (AOA) \cite{aoa06}, Time of Arrival (TOA) \cite{toa03}, and Time Difference of Arrival (TDOA) \cite{tdoa05} or some other technologies. Obviously, such means is not cost efficient. The reasons to require the location information of sensor nodes and mobile sinks are different. For sensor nodes, the location information is simply to indicate where the sensory data comes from; while for mobile sinks, it is used to efficiently route data packets dynamically to mobile sinks. Instead of relying on GPS device, an alternative to get the location information of sensor nodes is to use topological maps \cite{TPM14} based only on connectivity Information. 

The Topology Preserving Map (TPM) \cite{TPM14}, a recent technique, is generated from the network's Virtual Coordinates (VCs). VC reflects the connectivity of the sensor nodes. A subset of nodes is selected as anchors and each node's VC vector is represented by the minimum hop distance from itself to the set of anchors. The number of anchors is the dimensionality of nodes' VC vectors. The problem of anchor selection is discussed in \cite{anchor11}. By making use of an SVD-based dimensionality reduction scheme, nodes' Topological Coordinates (TCs) can be derived, thus regaining the directional information lost in VCs to derive TPM. TPM preserves the topological information of physical sensor networks but takes out the physical distances between sensors. TPM is derived from the hop-distances between sensors and anchors, which is easier to get than physical locations. TPM provides an alternative tool for physical coordinates. \cite{gunathillake2016maximum, gunathillake2017topology} presents a Maximum Likelihood-Topology Maps (ML-TM) that provide a more accurate physical representation, by using the probability of signal reception, an easily measurable parameter that is sensitive to the distance. Approach is illustrated using a mobile robot that listens to signals transmitted by sensor nodes and maps the packet reception probability to a coordinate system using a signal receiving probability function. It has been shown that the topology coordinates can replace the physical coordinates in some applications. 

Regarding the location information of mobile sinks, we can treat them as some targets and then many existing target tracking approaches can be used: cluster-based methods, tree-based approaches, Particle Filter, Kalman Filter and their variations. In \cite{cluster03}, a cluster-based architecture is used to determine the location of a mobile sink. The selected cluster head sensors carry out tracking algorithm and member sensors in the cluster collect information about the mobile sink. In \cite{tree04}, a tree-based approach simplifies collaboration among sensors in a localized manner. Particle Filter (PF) is investigated for tracking moving sinks \cite{particle10}. It provides an analysis of the effect of various design and calibration parameters on the accuracy of PF. PF’s extensions, Auxiliary Particle Filter (APF) and Cost-Reference Particle Filter (CRPF) are also applied for mobile sink tracking \cite{binary08}. For Kalman Filter, an adaptive tracking method, which is based on Kalman Filter (KF), is proposed for linear systems \cite{disk05}. The tracking region size adapts accordingly to the mobile sink's acceleration by waking up different numbers of sensors. As an extension of KF, Extended Kalman Filter (EKF) \cite{ekf79} is obtained by making use of linearization at the current state and using KF to solve the linear problem. One deviation is displayed in \cite{likelihood68}. For nonlinear systems, EKF has been utilized in \cite{ekf09energy}, \cite{ekf12missing}, \cite{ekf12sparsity}, etc. EKF is applied to an energy efficient tracking system where the sampling interval is adaptive based on the users' requirements \cite{ekf09energy}. In \cite{ekf12missing}, EKF is explored to deal with the problem of missing measurements for time-varying systems with stochastic nonlinearities. In \cite{ekf12sparsity}, the problem of mobile sink tracking based on energy readings of sensors is studied by using an EKF to minimize the estimation error. As an extension of EKF, the time continuous version of Robust Extended Kalman Filter (REKF) is completely formulated in \cite{rekf99} and the time discrete version is described in \cite{rekf09}. It focuses on the uncertainties of the system and achieves more robust performance. REKF is applied to derive an estimate of the mobile sink's location in networks where sensors are mobile \cite{rekf04location}. Instead of tracking the mobile sink, REKF is also used for sensor localization in Delay-Tolerant Sensor Networks \cite{pathirana2005node}. REKF is used to estimate the sensor positions for a DTN by a mobile robot, and it is computationally more efficient and robust in comparison to EKF implementation. In \cite{binary11} and \cite{binary08}, Binary Sensor Networks, where the sensor nodes are incapable of measuring distances, are discussed. The authors assume the mobile sink is always emitting a signal that can be detected in a range of the circular radius. Under this assumption, the sensors' output is 1 or 0 representing whether the mobile sink is within their sensing ranges. However, to make use of these binary outputs, they also assume the sensors know their own locations \cite{binary05}, which may raise the cost of sensor placement. Removing the need for sensors' locations will result in simpler nodes and less sensitive algorithms, and facilitate the large-scale deployment of networks. 

In \cite{TPMtrack13}, TPMs is used for tracking a target. It uses the TCs instead of any physical distance measurements in WSNs to locate the mobile sinks. However, the accuracy of this strategy is not guaranteed. It proposes the mobile sink's VCs are obtained by taking an average of neighbouring sensors' VCs. Therefore, the accuracy of the mobile sink's TCs highly depends on the accuracy of this averaging approximation. \cite{gunathillake2017mobile} makes use of the topology map to navigate a mobile robot towards an emergency source. \cite{ashanie16} proposes a real time target search scheme using topology map. The objective is to employ a mobile robot to catch the target in the shortest time. \cite{ashanie17} considers the problem of tracking targets. The proposed approach is a distributed one which consists of a Robust Kalman filter combined with a nonlinear least-square method, and Maximum Likelihood Topology Maps. The primary input for estimating target location and direction of motion is provided by time stamps recorded by the sensor nodes when the target is detected within their sensing range. An autonomous robot following the target collects this information from sensors in its neighborhood to determine its own path in search of the target. This work is the most relevant one to what we will present in Section \ref{tracking}, while the Extended Robust Kalman Filter is adopted there.

\section{Cellular Networks}
Cellular networks are the major means for people to communicate with others. The study of cellular networks has a long history. Traditionally, researchers usually focus on the macro cell management. The positions of macro cells  are determined mainly based on the long-term traffic behaviour \cite{bor2016new}. Due to the explosive demands in recent years, users usually have bad experience especially during crowed events \cite{shafiq2016characterizing}. Internet Service Providers (ISPs) have been developing more efficient strategies to meet the requirements of users \cite{nakamura2013trends}. Generally, under a common umbrella of network densification \cite{bhushan2014network, ge20165g}, two main approaches have been studied: the usage of small cells and drones. We respectively present brief reviews on these two kinds of approaches.
\subsection{Serving users by small BSs}
Underlaid in a traditional (macro) cellular networks, small BSs can be deployed in traffic hotspots \cite{bennis2013cellular, hwang2013holistic}. The key aspect of deploying small BSs is the elaborate site planning, which not only impacts on the capacity performance but also influences the coverage of the users. \cite{ghazzai2016optimized} studies how to deploy base stations to satisfy both cell coverage and capacity constraints. Meta-heuristic algorithms: PSO and GWO, are presented to address the optimization problem. User density is considered.
\cite{zhao2017approximation} considers the problem of deploying macro cells and small BSs. Each small area is associated with a required data rate. The spectral efficiency on a bandwidth cannot be lower than a threshold. Approximation algorithms are proposed.

\cite{arnold2010power} investigates the power consumption models of macro and small BSs in heterogeneous cellular networks. The authors claim that the energy efficiency of any deployment is impacted by the power consumption of each individual network element and the dependency of transmit power and load.
\cite{ashraf2011sleep} introduces sleep mode algorithms  for small cell base stations to reduce cellular networks' power consumption. They design small cell driven, core network driven, and user equipment driven algorithms, such that the energy consumption is modulated over the variations in traffic load. 
Similarly, the sleep mode idea is also used in macro cells \cite{marsan2009optimal, soh2013energy}. They formulate the power consumption minimization and energy efficiency maximization problems, and determine the optimal operating regimes for macro cells. 
\cite{hossain2014evolution} studies the interference management challenge (e.g. power control, cell association) in these networks with shared spectrum
access (i.e. when the different network tiers use the same spectrum). The authors claim that the existing interference management strategies are not able to deal with the interference management issue in the prioritized 5G multi-tier networks, where UEs in different tiers have various priorities to access available channels.
\cite{ye2013user} focuses on the problem of user association between macro cell and small cells. They find that although macro cell can provide strongest downlink signal, users should be actively pushed onto the small BSs (offloading), which will often be lightly loaded and so can provide a higher rate over time by offering the mobile many more resource blocks than the macro cell.
\cite{singh2014joint} studies the issue of joint resource partitioning and offloading. The authors show that resource partitioning is required in conjunction with offloading to improve the rate of cell edge users in co-channel heterogeneous networks.
\subsection{Serving users by UAVs}
Considering the fact that deploying dense small BSs is not efficient in cost (includes the equipments, the payments for backhaul and site rental, etc), and involves an elaborate site planning, an alternative solution is the utility of UAVs as intermediate points between the existing BSs and User Equipments (UEs). Using drones may not only be a cost efficient solution, but also be able to handle emergencies, which cannot be realized by traditional network infrastructure. Some world leading companies have already started to consider the utility of UAVs, such as Facebook \cite{facebook2014} and Google \cite{google2014}.

There is a growing number of works on the topic of drone BSs in cellular network. They consider various challenges, including optimal 3D deployment of drones, energy limitations \cite{chae2015iot}, interference management \cite{mozaffari2015drone} and path planning \cite{chi2012civil}. 

Regarding the optimal 3D deployment, the wireless communication model between drones and ground users is usually the first concern, which is quite different from conventional ground BSs to ground users. Different from the schemes of determining the positions of BSs, drones have another variable, i.e., altitude. This variable makes the wireless communication model between seriver and user very different from the 2D cases. In \cite{al2014modeling}, air-to-ground pathloss is modelled. It shows that there are two main propagation situations: UEs have Line-of-Sight (LoS) with drones and UEs have non-Line-of-Sight (NLoS) with drones due to reflections and diffractions. A closed form expression for the probability of LoS between drones and UEs is developed in \cite{al2014optimal}. The authors of \cite{mozaffari2015drone} investigate the maximum coverage by two drones in the presence and absence of interference between them. 

Based on this LoS probability model, two main two problems have been addresses: how to deploy a given number of drones, and how to figure out the minimum number of drones to serve all users. \cite{bor2016efficient} considers the problem of deploying one drone in 3D to maximize the number of the served users, and in the meantime, the users receive acceptable service. 
Exploring the relationship between vertical and horizontal dimensions, \cite{alzenad20173d} solves the problem in \cite{bor2016efficient} by turning the problem into a circle placement problem. \cite{sharma2016uav} studies the problem of deploying multiple drones. It develops an approach of mapping the drones to high traffic demand areas via a neural-based cost function. \cite{mozaffari2016efficient} derives the downlink coverage probability for drones as a function of the altitude (all drones are assumed to fly at the same altitude) and the antenna gain; and determines the locations of drones to maximize the total coverage area and the coverage lifetime using circle packing theory. The authors of \cite{kalantari2016number} consider the issue of the number of drones and where to deploy them in 3D environment; and design a PSO based heuristic algorithm, such that the users receive acceptable downlink rate. \cite{rohde2012interference} proposes a interference aware positioning approach for UAVs for cell overload and outage compensation.

There are also many publications on some other interesting aspects of drones. In \cite{yang2017proactive}, a proactive drone-cell deployment framework to alleviate the workload of existing ground BSs is proposed. It involves traffic models for three typical social activities: stadium; parade and gathering. Also, it presents a traffic prediction scheme based on the models. Further, it discusses an operation control method to evenly deploy the drones. 
\cite{becvar2017performance} studies the scenario where a drone serves users moving along a street. It compares with approaches of using small cell base stations and it shows that the flying BS introduces a significant gain in channel quality and outperforms that using dense BSs in terms of throughput.
\cite{chen2017caching} proposes a proactive deployment of cache-enabled drones to improve the quality of experience at users. The drones cache some popular content based on a prediction model. Such cache is able to reduce the data packet transmission delay. 

It is also worth mentioning that the drone networks studied can be viewed as a special class of networked multi-agent control systems \cite{savkin2004coordinated, matveev2003problem, matveev2004problem, matveev2007analogue, ren2008distributed, dimarogonas2012distributed, cao2013overview, azuma2013broadcast, cassandras2013optimal,shamma2008cooperative}.

\section{Summary}
In this chapter, we present a brief review on the existing work related to our studied problems. Some typical reviewed approaches will be used for comparisons in the next chapters. For more in depth reviews, readers are referred to the survey papers and books \cite{ian02survey, mobilesinksurvey11, OR15survey, feng2013survey, Gu16, savkin2015book, hasan2011green, damnjanovic2011survey, aliu2013survey}. The literature review about using mobile sinks to collection data in Section \ref{ME_collection} are presented in \cite{huang17survey}. From the next chapter, we will present the main contributions of this report.

\chapter{Sink Tracking in a Wireless Sensor Network using a Topology Map and Robust Extended Kalman Filter}\label{tracking}

\minitoc
\section{Motivation}
When mobility is adopted in WSNs, one fundamental challenge is the location update of the mobile sink, which impacts on the data routing from sensor nodes to the mobile sink. 
This chapter visits the problem of how to maintain the current location of the mobile sink. We treat it as a target tracking problem. 

A common framework of target tracking works like this: when a target moves in a sensing field, the sensor nodes nearby the target can detect it. A centralized or distributed algorithm is employed to estimate the target location and/or predict the future position based on the sensor nodes' location information and their measured values such as the distances between sensor nodes and the target. There are two challenges in this common framework. The first one is the sensor nodes' location information may be costly to get. Second, the devices to measure distances are usually expensive and may be too large for the tiny nodes. Therefore, the methods which can do not depend on the nodes' locations and measuring distances are attractive for the networks which consist of a great number of inexpensive sensor nodes.

We propose an approach for such scenario and related operations using sensor nodes incapable of distance measurement. Topology Preserving Map (TPM) \cite{TPM14, gunathillake2016maximum, gunathillake2017topology} is a new technology for WSNs. It is a rotated and/or distorted version of the real physical node map to account for connectivity information inherent in Virtual Coordinates (VCs). Making use of the hop distances (VCs) between sensors and anchors instead of physical distance measurements, a set of Topological Coordinates (TCs) is generated by Singular Value Decomposition (SVD) method. The TCs provided by TPMs is an alternative for physical coordinates for some applications depending on connectivity and location information. TPM preserves the topology of the WSNs. Specifically, one sensor node's neighbour nodes are the same no matter in the physical domain or in TPM. This is the foundation of target tracking using TPM.

The basic idea of this chapter can be stated as follows. When a mobile sink moves into the sensing field, a subset of the sensor nodes detects the sink and keeps track of receiving probabilities (formally defined in the following section) from the mobile sink. Using the TCs of these sensor nodes, which can be obtained after the network is deployed, and the measured values, we employ an estimation and prediction algorithm to estimate the mobile sink location and predict its position in the sensor field m time steps into the future. We use this prediction to wake up a new subset of sensor nodes that are near the predicted position to detect the sink. This new subset of sensor nodes will detect the mobile sink when it arrives and the tracking process continues. 

Different from the common framework, we track the mobile sink in the sensing field using TPM. The expected benefit is that TPM releases the requirement of nodes' location as it provides an alternative tool to represent locations which are called TCs. However, it also poses its own challenge: distortion, compared to physical coordinate based maps. As shown in \cite{TPM14}, some sensor nodes are located along a straight line. But in its corresponding TPM, the formation of these sensor nodes may be a curve. Therefore, the transformation of movement from the physical domain to topology domain brings uncertainties, which we have to deal with to track the mobile sink successfully in topology domain. In this chapter, we use a robust algorithm for estimation and prediction to tackle the uncertainties.

We provide a comprehensive study on the evaluation of the mobile sink tracking performance using TPM. The performance is compared with the physical tracking where the sensor nodes are location aware and capable of measuring distances. As a distorted version of the physical map, the distance between two nodes in topology map is not the same as that in physical map. Thus, it is challenging to compare them directly. Three methods are proposed to evaluate the performance and they all avoid the problem of inconsistent units. The first method focuses on the prediction errors. Our solution is to transform the movement in topology map back to the physical domain. By doing this, we compare the prediction errors under the same unit. The second one is to investigate the number of the required sensor nodes. As the predicted position cannot be 100\% precise no matter in the physical map or topology map, we need to select a certain set of sensor nodes to detect the mobile sink each time slot. We argue the larger the prediction error is, the more sensor nodes are needed. So comparing the number of these sensor nodes is another option to evaluate the tracking performance. Thirdly, we pay attention to the effectiveness of the selected sets of sensor nodes. We will show in this chapter, the physical tracking is generally more accurate than that in topology domain. Thus, we investigate the percentage that the topology set can cover the physical set.

The main contribution of this chapter is an approach for sink tracking when sensor nodes are incapable of measuring distance, which is competitive with themes that are based on distance measurements. It includes:

\begin{itemize}
\item We argue the benefits of utilizing TPM for mobile sink tracking: the release of sensor nodes' location and distance measurement. Using sensor nodes' TCs and their measured values, we analyze the possibility of tracking the mobile sink in TPM and demonstrate it by simulations.
\item We employ Robust Extended Kalman Filter \cite{rekf99} to estimate the mobile sink's location and predict its future positions in TPM using a rough estimated motion model.
\item We provide extensive simulations in a network which consists of 550 sensor nodes. We target on different kinds of motions ranging from constant or varying speed to random direction movement where the speed can be constant as well as varying.
\item We propose three methods to evaluate the proposed work. These methods avoid the challenging problem of different units in the physical domain and topology domain.
\end{itemize}

The rest of this chapter is organized as follows. 
Section \ref{tracking_statement} formulates the problem. It describes the Topological Preserving Maps first and then discusses the sink mobility model and the sensors' measurement model. Section \ref{tracking_kf} demonstrates the Robust Extended Kalman Filter (REKF) that we use to solve the problem. In Section \ref{tracking_topo}, we demonstrate the possibility of tracking using TPM. We provide the estimation performance of our proposed method in constant speed movement as well as varying speed movement. Section \ref{trackiing_compare} focuses on the prediction performance of our method. We target on a general movement: random movement. To evaluate the proposed algorithm, we also employ REKF based method in the physical domain. Further we propose three methods to compare the proposed method with the physical case. Finally, Section \ref{conclusion_tracking} concludes the chapter.

\section{Problem Statement}\label{tracking_statement}
\subsection{Problem Statement}
In this section, we formally describe the studied problem. The problem of mobile sink tracking in TPM involves two types of domains: physical domain and topology domain respectively.

We have a sensing field covered by a wireless sensor network ($S$) which consists of $n$ sensor nodes. Using the technology of TPM, every sensor node has their own TCs. A mobile sink moves in this field. The sensor nodes can output a measured value when the mobile sink moves near them. As we have no knowledge of the sensor nodes' location information, tracking the target in the physical domain is difficult. Thus, we introduce a transformation, i.e. the mobile sink's movement transformation from the physical domain to topology domain. This transformation requires a location estimation method to find out the mobile sink's location in TPM based on the detecting sensor nodes' TCs and their measured values. After having the location in TPM, the next sub-problems are: 1) how to predict its future position in TPM; 2) how to activate sensor nodes to detect the mobile sink. The final task is to use the activated sensor nodes to track the mobile sink.
\subsection{Topology Preserving Maps}
A Virtual Coordinate System (VCS) is based on anchors, which are selected randomly or by a well defined strategy \cite{anchor11}. Each node in the network is represented by a VC vector, indicating the shortest hops to each anchor. Directional information is not involved in VCS. \cite{tpm10topology} presents a scheme to extract directional information from VCs in the form of TCs. A network's TPM is based on TCs. It has been shown to preserve relative position information. In this chapter, we consider dealing with mobile sink tracking in the topology domain rather than the physical domain.

Consider a network with $n$ sensors and $m$ of them are selected as anchors. Thus, each sensor is represented by a VC vector with $m$ elements. Let $P$ be a $n\times m$ matrix containing VCs of all $n$ nodes. Considering the high computation cost of Singular Value Decomposition (SVD) in extracting TCs from $P$, calculating SVD items from anchor set is an alternative for large-scale networks. Let $AN$ be a $M\times M$ matrix containing VCs of only the anchor nodes. The way of producing TCs from $AN$ is presented below:
\begin{equation}\label{eq1}
AN=U_A \cdot S_A\cdot V_A^T
\end{equation}
\begin{equation}\label{eq2}
P_{SVD}=P \cdot V
\end{equation}
\begin{equation}\label{eq3}
[X_{TC},Y_{TC}]=[P_{SVD}^{(2)},P_{SVD}^{(3)}]
\end{equation}
where $P_{SVD}^{(2)}$ and $P_{SVD}^{(3)}$ in (\ref{eq3}) are the second and third columns of $P_{SVD}$ in (\ref{eq2}), and the set $\{(X_{TC},Y_{TC})\}$ consists of  coordinate pairs which form the TPM and play a fundamental role in our work. More details about the TPMs generalization method can be found in \cite{TPM14}.
\subsection{Motion Model in Physical Domain}
We focus on the following time-continuous mobility model and measurement model in physical domain
\begin{equation}\label{eq4_system_model}
\begin{cases}
\dot{x}=A(x)+Bw_n \\
y=C(x)+v_n\\
\end{cases}
\end{equation}
where $x$ is the state of the system. It may include the position of the mobile sink as well as the velocity and acceleration. $y$ is the measured output values of the nodes.  $A(\cdot)$ is the dynamic model, $B$ is a given matrix. $C(\cdot)$ is the measurement function. $w_n$ is the process noise and $v_n$ is measurement noise.
The dynamic model can describe various movements. For example, when $A(\cdot)$ is a linear function of the state $x$, we can use this model to represent nearly constant speed movement or nearly constant acceleration movement, both of which can be influenced by the noise. When $A(\cdot)$ is nonlinear function, it can describe more complex mobility.

$C(\cdot)$ depends on the measurement type of the sensor nodes. In this chapter, we use the sensor nodes which can output a receiving probability based on the distance between the mobile sink and themselves. So it is a nonlinear function of the system state. We provide more details in the following part.
\subsection{Measurement model}
In this part, we introduce a function $C(d)$ for sensor nodes which describes the probability of getting a signal from the mobile sink when the mobile sink is at distance $d$ from the sensor node. This function satisfies the following constraints:
\begin{equation}\label{eq5}
  \begin{aligned}
  0\leq C(d)\leq 1,\forall d\\
  C(d_1)\leq C(d_2),\forall d_1 \geq d_2 \\
  C(d)=0,\forall d>d_0
  \end{aligned}
\end{equation}
where $d_0$ is some given distance. Such function is called the receiving probability function. An extreme example of such a function is
\begin{equation}\label{eq6}
C(d)=
\begin{cases}
0,d>d_0 \\
1,d\leq d_0 \\
\end{cases}
\end{equation}

Obviously function (6) satisfies all the conditions (5) and this function has the same formulation of binary sensor nodes mentioned in the related work part. In binary sensor nodes, if the distance between the mobile sink and the node is no larger than $d_0$, it is often assumed that the node can detect the mobile sink. However, function (6) is discontinuous at $d_0$ which is difficult to be implemented in our algorithm. Thus we will use the following example of such a function
\begin{equation}\label{eq7}
C(d)=
\begin{cases}
\left( \frac{1}{d_0}-\frac{1}{a} \right) d+1,& 0 \leq d\leq a \\
\frac{a}{d_0(a-d_0)}d-\frac{a}{a-d_0},& a<d\leq d_0 \\
0,& d>d_0
\end{cases}
\end{equation}
where $0<a<d_0$. If $d=\sqrt{(X-X_i)^2+(Y-Y_i)^2}$ is the distance between the mobile sink and the sensor $i$, the sensing range of sensors, and $(X_i,Y_i)$ is the TCs of sensor $i$. As $d$ is a function of the mobile sink’s state, the receiving probability function $C(d)$ can also be formulated as $C(x)$. Thus the output of the sensor nodes is represented by adding the additive noise to $C(x)$, i.e. $y=C(x)+v_n$. Fig. 1 demonstrates the illustrative example of the receiving probability function. 

We can see from Fig. 1, when $a\rightarrow d_0$ , measurement function (7) tends to be the binary model (6).
\begin{figure}[t!]
\begin{center}
{\includegraphics[width=0.4\textwidth,height=0.3\textwidth]{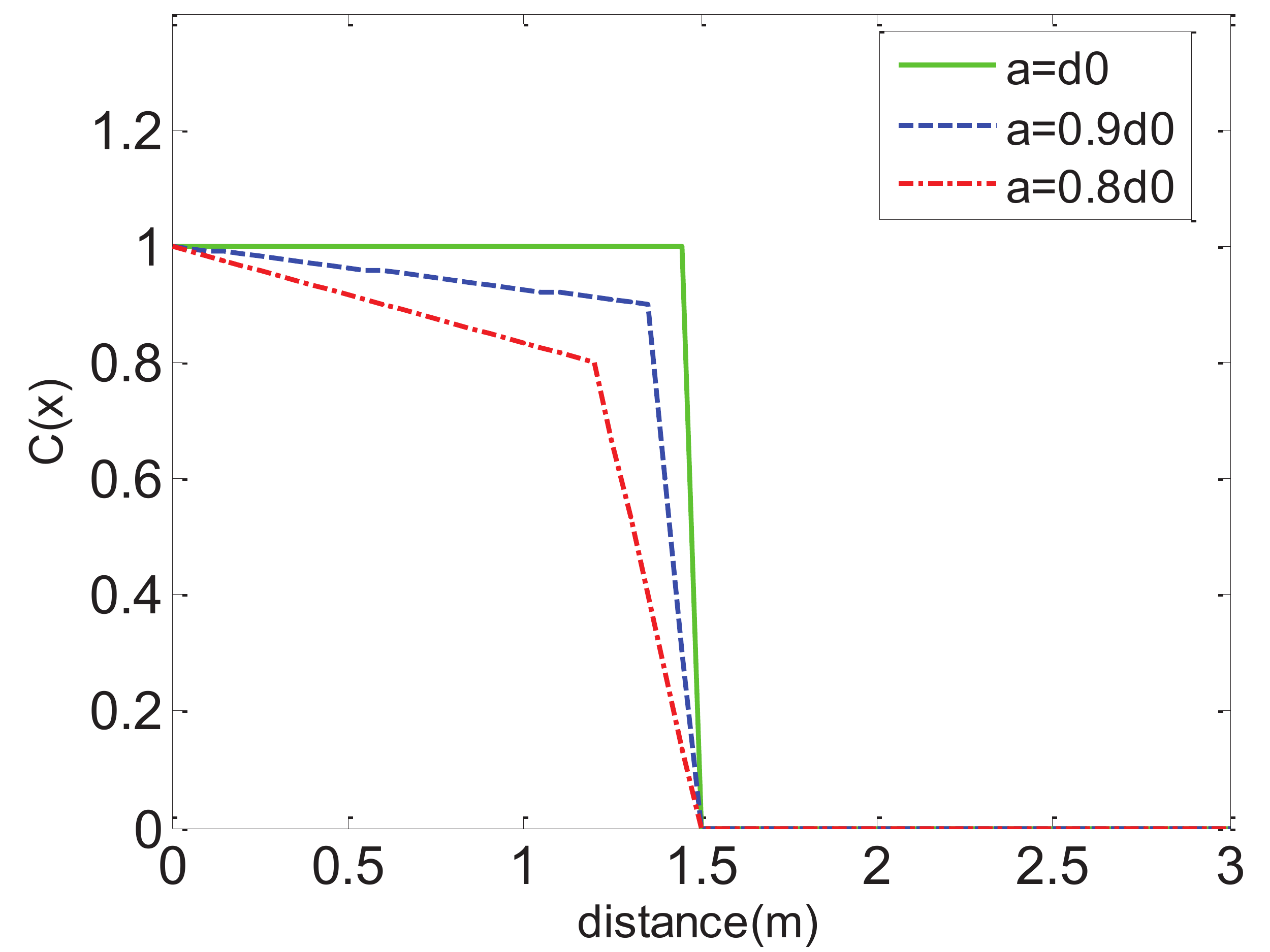}}
\caption{Measurement functions}
\end{center}
\end{figure}
\subsection{Movement Generation in Topology Domain}
Before providing the motion model in topology domain, we will first look at what the difference is between physical domain and topology domain. An illustrative example is shown in Fig. \ref{figure2}. Fig. \ref{figure2}(a) displays the trajectory from middle bottom to middle top in physical domain and Fig. \ref{figure2}(b) shows the corresponding trajectory from right to left in topology domain. We select 5 anchors following the Extreme Node Search (ENS) method \cite{anchor11}.

As mentioned before, TPM is a distorted version of the real physical node map. Various examples in \cite{TPM14} can demonstrate this character. Thus, the mobile sink's corresponding movement in topology map may not be transformed linearly from the physical domain. Therefore, in order to track the mobile sink in TPM, we have to provide a way to generate the mobile sink's movement. In this chapter, we follow the same method used in \cite{TPMtrack13}. Consider the mobile sink is surrounded by  $l$ neighbour nodes. Let $MT$ (a 1-by-$m$  vector) be the mobile sink's VCs and $P_{N_i}$(a 1-by-$m$ vector) be the $i^{th}$ neighbour node's VCs. So the mobile sink's TCs can be generated based on (\ref{eq2}) and (\ref{eq3}):
\begin{equation}\label{eq8}
  \begin{aligned}
  &MT_{SVD}=\left( \sum_{i=1}^l P_{N_i}/l \right) V_A \\
  &X_T=MT_{SVD}^2 \\
  &Y_T=MT_{SVD}^3 \\
  \end{aligned}
\end{equation}

We also consider the speeds along the two axes as state variables. Then the movement in topology domain can be formulated as follows
\begin{equation}\label{eq9_1}
\begin{cases}
\dot{x}_T=A_T(x_T)+B_T w_T \\
y_T=C_T(x_T)+v_T\\
\end{cases}
\end{equation}
where $x_T=[X_t,Y_t,\dot{X_t},\dot{Y_t}]$ is the state variable in TPM, $\dot{X_t}$ and $\dot{Y_t}$ are the velocity along each direction in TPM. $w_T$ is the process uncertainty and $v_T$ is measurement uncertainty.
\begin{figure}[t!]
\begin{center}
{\includegraphics[width=0.4\textwidth,height=0.6\textwidth]{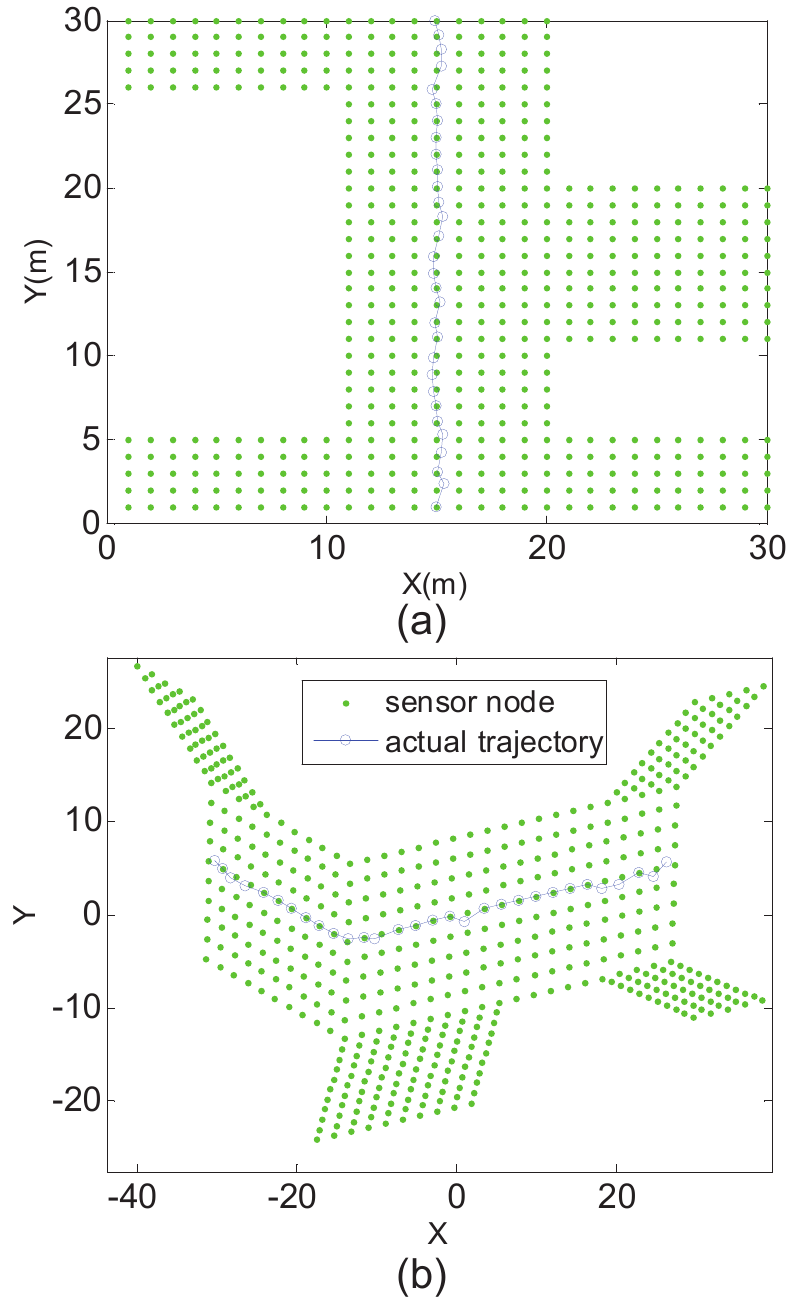}}
\caption{Illustration of the difference of trajectory in physical domain and topology domain. (a) In physical domain. (b) In topology domain.}\label{figure2}
\end{center}
\end{figure}

\section{Robust Extended Kalman Filter}\label{tracking_kf}
In the last section, we express the mobility model and measurement model with TPM. In this part, we will describe the robust mobile sink tracking method and the Robust Extended Kalman Filter is taken from \cite{rekf99} and originally this framework was presented in \cite{james1998nonlinear}. Instead of (\ref{eq4_system_model}), we consider a more general nonlinear uncertain system of the form
\begin{equation}\label{eq9}
\begin{cases}
\dot{x}=A(x,u)+Bw \\
z=K(x,u)\\
y=C(x)+v\\
\end{cases}
\end{equation}
where $z$ is the uncertainty output of this system and $K\geq0$ is the uncertainty model. The uncertainties  can be defined by a nonlinear integral constraint shown below:
\begin{equation}\label{eq10}
\Phi(x(0))+ \int_{0}^s L_1(w(t),v(t))dt \leq e+ \int_{0}^s L_2(z(t))dt
\end{equation}
where $\Phi$, $L_1$, and  $L_2$ are all bounded nonnegative functions with continuous partial derivatives satisfying growth conditions of the type
\begin{equation}\label{eq11}
  \begin{aligned}
\parallel \phi(x)-\phi(x') \parallel \leq \sigma \left( 1+ \parallel x \parallel +\parallel x' \parallel \right) \parallel x-x' \parallel
\end{aligned}
\end{equation}
where $e$ is a positive real number, $\parallel\cdot\parallel$ is the Euclidean norm with $\sigma>0$, and $\phi=\Phi,L_1,L_2$. Uncertainty input $w(\cdot)$ and $v(\cdot)$ satisfying this condition are called admissible uncertainties. Note, the constraint (\ref{eq10}) is a general case of the so-called integral quadratic constraint; see e.g. \cite{peterson2000robust}.

We consider the problem of characterizing the set of all possible states $X_s$ of (\ref{eq9}) at time $s\geq0$ which are consistent with a given control input $u^0(\cdot)$ and a given output path $y^0(\cdot)$; i.e., $x\in X_s$ if and only if there exist admissible uncertainties such that if $u^0(t)$ is the control input and $x(\cdot)$ and $y(\cdot)$ are resulting trajectories, then $x(s)=x$, $y(t)=y^0(t)$, for all $0\leq t\leq s$.

\subsection{State Estimator}
The state estimation set $X_s$ is characterized in terms of level sets of the solution $V(x,s)$ of the Partial Differential Equation (PDE) (\ref{eq12}).
\begin{equation}\label{eq12}
\begin{aligned}
&\frac{\partial}{\partial t} V+ \smash{\displaystyle\max_{w\in R^m}} [\nabla_x V_.(A(x,u^0)+bw) \\
&-L_1(w,y^0-C(x))+L_2(K(x,u^0))]=0
\end{aligned}
\end{equation}

The PDE (\ref{eq12}) can be viewed as a filter, taking observations $u^0(\cdot)$, $y^0(\cdot)$, $0\leq t\leq s$ and producing the set $X_s$ as output. The state of this filter is the function; thus $V$ is an information state for the state estimation problem.

Theorem: Assume the uncertain system (\ref{eq9}), (\ref{eq10}) satisfies the assumptions given above. Then the corresponding set of possible states is given by
\begin{equation*}
X_s=\{x\in R^n : V(x,s)\leq d\} 
\end{equation*}
where $V(x,t)$ is the unique viscosity solution of (\ref{eq11}) in $C(R^N\times [0,s])$.

Proof can be seen in \cite{rekf99}.

Here we consider an approximation to the PDE (\ref{eq12}), which leads to a Kalman filter-like characterization of the set $X_s$. Reference \cite{rekf99} presents this as an extended Kalman filter solution to the set-value state estimation problem for a linear plant with the uncertainty described by an Integral Quadratic Constraint (IQC). This IQC is also presented as a special case of (\ref{eq10}). We consider the uncertain system described by (\ref{eq9}) and an IQC of the form
\begin{equation}\label{eq13}
\begin{aligned}
&\left( x(0)-x_0 \right)'N \left(x(0)-x_0 \right) \\
&+\frac{1}{2}\int_0^s \left( w(t)'Q(t)w(t)+v(t)'Rv(t)\right)dt \\
&\leq e+ \frac{1}{2}\int_0^s z(t)'z(t) dt \\
\end{aligned}
\end{equation}
where $N>0$, $Q>0$ and $R>0$. For (\ref{eq9}) and (\ref{eq10}) the PDE (\ref{eq12}) can be written as
\begin{equation}\label{eq14}
\begin{aligned}
& \frac{\partial}{\partial t}V+ \nabla V_.A(x,u^0)+\frac{1}{2} \nabla_xV_.BQ^{-1}B\nabla_xV_.'\\
& -\frac{1}{2} \left( y^0-C(x) \right)' R \left( y^0-C(x) \right)\\
& +\frac{1}{2} K(x-u^0)'K(x-u^0)=0\\
& V(x,0)=\left( x(0)-x_0 \right)'X_0 \left( x(0)-x_0 \right)\\
\end{aligned}
\end{equation}

We now consider a function $x(t)$ defined as
\begin{equation*}
\hat{x}(t)\triangleq \smash{\displaystyle\max_{w\in R^m}} V(x,t)
\end{equation*}
Then it follows
\begin{equation}\label{eq15}
\nabla_x V \left( \hat{x}(t),t \right)=0
\end{equation}
And
\begin{equation}\label{eq16}
\nabla_x^2 V\left( \hat{x}(t),t \right)\dot{\hat{x}}(t)+ \frac{\partial}{\partial t}\nabla_x V\left( \hat{x}(t),t \right)\dot{\hat{x}}(t)'=0
\end{equation}

If we take the gradient with respect to $x$ of PDE (\ref{eq14}) and evaluate the yielding equation at $x=\hat{x}(t)$, and use (\ref{eq15}) and (\ref{eq16}), we obtain 
\begin{equation*}
\begin{aligned}
\nabla_x^2 V\left( \hat{x}(t),t \right)\dot{\hat{x}}(t)=& \nabla_x^2 V\left( \hat{x}(t),t \right) A(\hat{x}(t),u^0)\\
& +\nabla_x C(\hat{x}(t))'R(y^0-C(\hat{x}(t)))\\
& +\nabla_x K(\hat{x}(t),u^0)'K(\hat{x}(t),u^0)\\
& \hat{x}(0)=x_0
\end{aligned}
\end{equation*}

If the matrix $\nabla_x^2 V\left( \hat{x}(t),t \right)$ is nonsingular for all $t$, we can rewrite this equation as
\begin{equation*}
\begin{aligned}
\dot{\hat{x}}(t)=&A(\hat{x}(t),u^0)+\\
&\left[\nabla_x^2 V\left( \hat{x}(t),t \right) \right]^{-1} [\nabla_x C(\hat{x}(t))'R(y^0-C(\hat{x}(t)))+\\
& \nabla_x K(\hat{x}(t),u^0)'K(\hat{x}(t),u^0)]\\
& \hat{x}(0)=x_0
\end{aligned}
\end{equation*}

We earn a rough solution to the PDE (\ref{eq14}) by approximating $V(x,t)$ with a function of the form
\begin{equation}\label{eq17}
\widetilde{V}(x,t)=\frac{1}{2}\left(x-\widetilde{x}(t)\right)'N(t)\left(x-\widetilde{x}(t)\right)+\phi(t)
\end{equation}

The resulting (\ref{eq18}), (\ref{eq19}), and (\ref{eq20}) define our approximate solution to PDE (\ref{eq14})
\begin{equation}\label{eq18}
\begin{aligned}
\dot{\widetilde{x}}(t)=&A(\hat{x}(t),u^0)+\\
& X^{-1}[\nabla_x C(\hat{x}(t))'R(y^0-C(\hat{x}(t)))+\\
& \nabla_x K(\hat{x}(t),u^0)'K(\hat{x}(t),u^0)]\\
& \hat{x}(0)=x_0
\end{aligned}
\end{equation}
where $X$ is defined as the solution to the Riccati Differential Equation (RDE)
\begin{equation}\label{eq19}
\begin{aligned}
&\dot{X}+\nabla_x A(\widetilde{x},u^0)'X+X\nabla_x A(\widetilde{x},u^0)\\
&+XBQ^{-1}BX -\nabla_x C(\widetilde{x})'R\nabla_x C(\widetilde{x}) \\
&+\nabla_x K(\widetilde{x},u^0)'\nabla_x K(\widetilde{x},u^0)=0 \\
& X(0)=N
\end{aligned}
\end{equation}
and
\begin{equation}\label{eq20}
\begin{aligned}
\phi(t)=&\frac{1}{2}\int_{0}^t [(y^0-C(x))'R(y^0-C(x)) \\
&-K(\widetilde{x},u^0)'K(\widetilde{x},u^0)]d\tau \\
\end{aligned}
\end{equation}

Hence, it follows from Theorem 1 that an approximate formula for the set $X_s$ is given by
\begin{equation*}
\widetilde{X}_s=\{x\in R^n:\frac{1}{2}(x-\widetilde{x}(s))'X(s)(x-\widetilde{x}(s))\leq e-\phi(t)\}
\end{equation*}

This amounts to the so-called Robust Extended Kalman Filter (REKF) generalization presented in \cite{rekf99}.

Comparing system models (\ref{eq4_system_model}) and (\ref{eq9}), we can simplify the REKF formulas as follows.

Substitute $P(t)=X^{-1}(t)$ in (\ref{eq19}), the corresponding Riccati Differential Equation for the mobile sink's motion model is:
\begin{equation}\label{eq21}
\begin{aligned}
&\dot{P}=A'P+PA+BQ^{-1}B' \\
&+P[K'K-\nabla_x C(\hat{x})'RC(\hat{x})]P \\
&P(0)=N^{-1}
\end{aligned}
\end{equation}
And from (\ref{eq17}), the state estimator for our system is:
\begin{equation}\label{eq22}
\begin{aligned}
&\dot{\hat{x}}=A\hat{x}+P[\nabla_x C(\hat{x})'R(y-C(\hat{x}))+K'K\hat{x}]\\
&\hat{x}(0)=x_0\\
\end{aligned}
\end{equation}

\section{Sink Tracking in Topology Domain}\label{tracking_topo}
In this part, we provide the possibility for tracking mobile sink using REKF in topology domain. We use two different motion models to demonstrate the proposed method.
\subsection{Mobility Models}
We use two mobility models to evaluate the proposed method. Firstly we consider a simple movement (Motion 1) in physical domain, i.e. constant speed movement. This motion model falls into some realistic applications, e.g. tracking a car which is moving along a straight road. In this case the system state variable is specified as $x=[X,\dot{X},Y,\dot{X}]$, where $X$ and $Y$ are the mobile sink's physical coordinates, $X$ and $Y$ are the speeds along each axis. The dynamic matrix is as follows
\begin{equation*}
A=
\left(
  \begin{array}{cccc}   
    0 & 1 & 0 & 0 \\
    0 & 0 & 0 & 0 \\
    0 & 0 & 0 & 1 \\
    0 & 0 & 0 & 0 \\
  \end{array}
\right)
\end{equation*}

The second model (Motion 2) is more complex than Motion 1, i.e. varying speed movement. In this case the system state variable is $x=[X,\dot{X},\ddot{X},Y,\dot{Y},\ddot{Y}]$, where $X$ and $X$ are the accelerations along each axis. If the accelerations keep the same, the dynamic matrix is
\begin{equation*}
A=
\left(
  \begin{array}{cccccc}   
    0 & 1 & 0 & 0 & 0 & 0 \\
    0 & 0 & 1 & 0 & 0 & 0 \\
    0 & 0 & 0 & 0 & 0 & 0 \\
    0 & 0 & 0 & 0 & 1 & 0 \\
    0 & 0 & 0 & 0 & 0 & 1 \\
    0 & 0 & 0 & 0 & 0 & 0 \\
  \end{array}
\right)
\end{equation*}

Both of these two motion models are only used to describe the movements.
\subsection{Uncertainties}
In realistic mobile sink tracking applications, the precise mobility model of the mobile sink and the measurement model are often unavailable to users. The motion model and measurement model in TPM cannot be exactly known either. Therefore, the uncertainty in these models is a challenge in the implementation of our tracking algorithm. In this chapter, we consider both the process uncertainty and measurement uncertainty.

Our task is to construct the model (\ref{eq8}) to implement REKF. In realistic applications, we usually have no knowledge about the precise motion model of the mobile sink. Besides, the movement transformation (\ref{eq8}) from physical map to topology domain may be a nonlinear function, which is also unknown. Therefore, to estimate the movement in topology map, we roughly employ a nearly constant speed motion model: $\dot{x}_T=A_Tx_T+B_Tw_T$, where 
\begin{equation*}
A_T=
\left(
  \begin{array}{cccc}   
    0 & 1 & 0 & 0 \\
    0 & 0 & 0 & 0 \\
    0 & 0 & 0 & 1 \\
    0 & 0 & 0 & 0 \\
  \end{array}
\right)
\end{equation*}

We consider an uncertainty output model which is a linear function of the system state, i.e. $z_T=Kx_T$, where $K=\varepsilon I$, $\varepsilon$ is a positive scalar and $I$ is a unit matrix. The process uncertainty is modelled as $w_T=\gamma z_T+w_n$, where $\gamma \leq 1$ is a given parameter and $w_n$ is the white noise. Substitute $z_T$ we have
\begin{equation}\label{eq23}
w_T=\gamma \varepsilon Ix_T+w_n
\end{equation}

Thus the nearly constant speed motion model for estimation in topology domain becomes:
\begin{equation}\label{eq24}
\dot{x}_T=A_Tx_T+B_Tw_T=A_Tx_T+B_T(\gamma \varepsilon Ix_T+w_n)
\end{equation}
Now, we introduce the uncertain measurement function. We add a sine function and the white noise to (\ref{eq7}). The estimated measurement function is as follows:
\begin{equation}\label{eq25}
C_T(x)=C(x)+\eta sin(\xi d)+v_n
\end{equation}
where $\eta$ and $\xi$ are parameters and $v_n$ is the white noise. To compare (\ref{eq25}) with (\ref{eq7}), we display the two functions in the Fig. \ref{measurementfunction}.
\begin{figure}[t!]
\begin{center}
{\includegraphics[width=0.4\textwidth,height=0.3\textwidth]{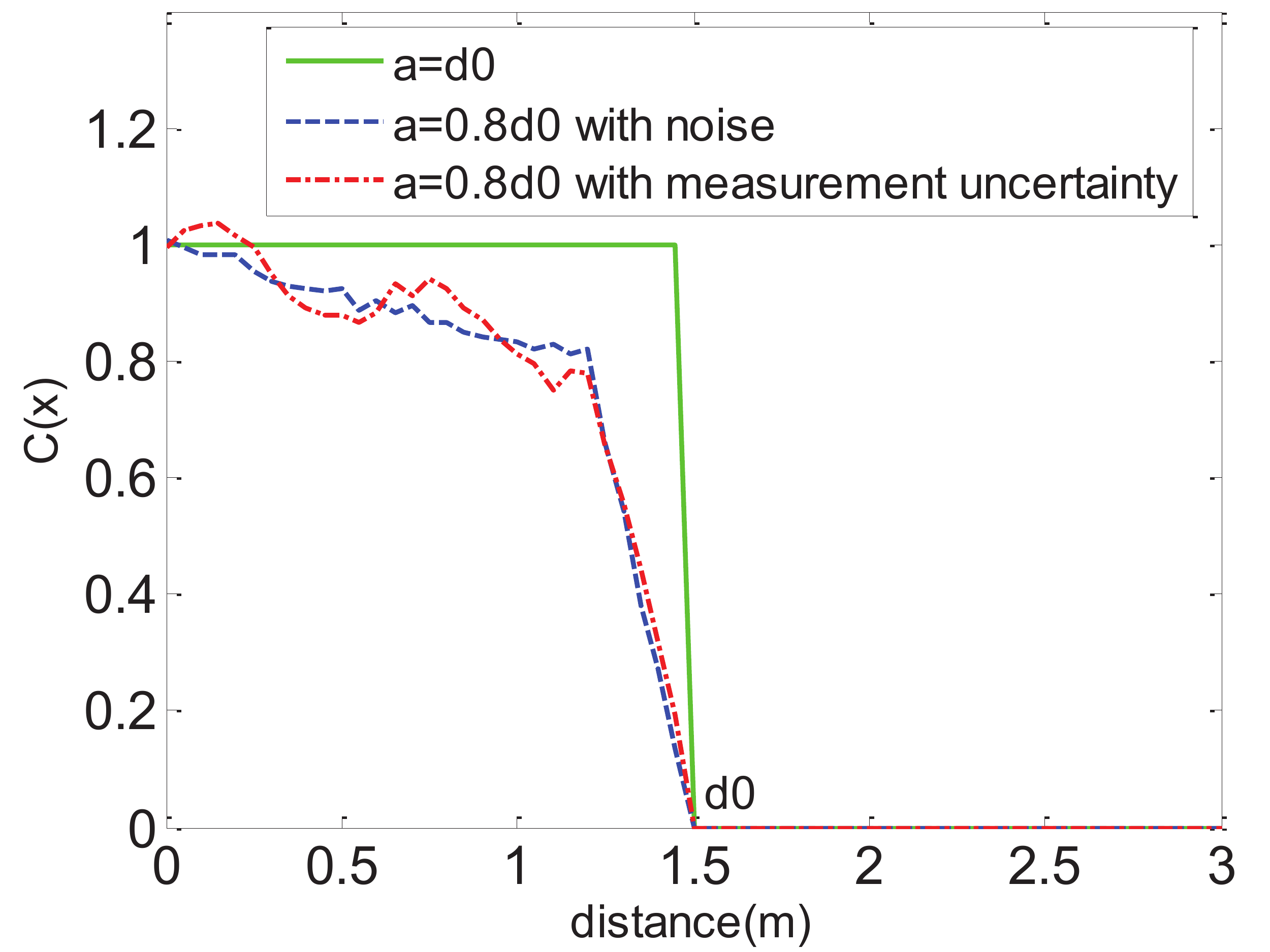}}
\caption{Measurement function with uncertainty.}\label{measurementfunction}
\end{center}
\end{figure}
\subsection{Simulations}
To examine the performance of the REKF in TPMs, we use a WSN with 550 nodes which are deployed in a 30$m$ * 30$m$ sensing field. The parameters and the uncertainty details for the simulation can be seen in Table \ref{simpara}.
\begin{table}
\begin{center}
\caption{Simulation Parameters and Uncertainty Details}\label{simpara} 
  \begin{tabular}{| c | c | c |}
    \hline
     \textbf{Parameter} & \textbf{Value} & \textbf{Comment} \\ \hline
    $n$ & 550 & Number of nodes\\ \hline
    $m$ & 5 & Number of anchors\\ \hline
    $s$ & 30$s$/60$s$ & Simulation time\\ \hline
    $d_0$ & 1.5$m$ & Sensing range\\ \hline
    $\bigtriangleup t$ & 1$s$ & Sampling interval\\ \hline
    $N$ & 0.1$I_4$ & Weighting on initial viscosity solution\\ \hline
    $w_n$ & 0.5$I_4$ & Covariance of process noise\\ \hline
    $v_n$ & 0.01$I_n$ & Covariance of measurement noise\\ \hline
    $Q$ & $I_4$ & Weighting on process uncertainty\\ \hline
    $R$ & $I_n$ & Weighting on measurement uncertainty\\ \hline
    $\gamma$ & 0.7 & Uncertainty parameter for motion model\\ \hline
    $\eta$ & 0.05 & Uncertainty parameter for measurement model\\ \hline
    $\xi$ & 10 & Uncertainty parameter for measurement model\\ \hline 
  \end{tabular}
\end{center}
\end{table}

Firstly, we provide a simulation using the illustrative example shown in Fig. \ref{figure2}. It is a nearly constant speed movement and the tracking result is shown in Fig. \ref{figure4}. The estimation error (defined as the distance between the actual position and the estimated position) of each step is shown in Fig. \ref{figure5}. The largest estimation error in this case is about 3.7 units and this is due to the initial condition. The estimation error in stable stages keeps at about 0.5 units in average.

The second simulation is to track the constant acceleration movement. Suppose a mobile sink makes a U-turn in the sensing field. The movement is displayed in Fig. \ref{figure6}. The mobile sink starts from (1,1) in the physical domain and its initial speeds along the two axis are 1.2$m/s$ and 0.48$m/s$ respectively. To make a U-turn as shown in Fig. \ref{figure6}, the acceleration along the two axis are set as -0.04$m/s^2$ and 0. It takes 60 $seconds$ to move to the top left. The movement trajectory in topological domain and the estimated trajectory are shown in Fig. \ref{figure7}. The estimation error in TPM is demonstrated in Fig. \ref{figure8}. In the stable stage, the estimation error also achieves 0.5 units in average.

From the above simulation results, the proposed method is able to track the movements with constant speed and varying speed as well in topology domain. In the next section, we evaluate the tracking performance in details.
\begin{figure}[t]
    \centering
    \begin{subfigure}[h]{0.45\textwidth}
        \includegraphics[width=\textwidth]{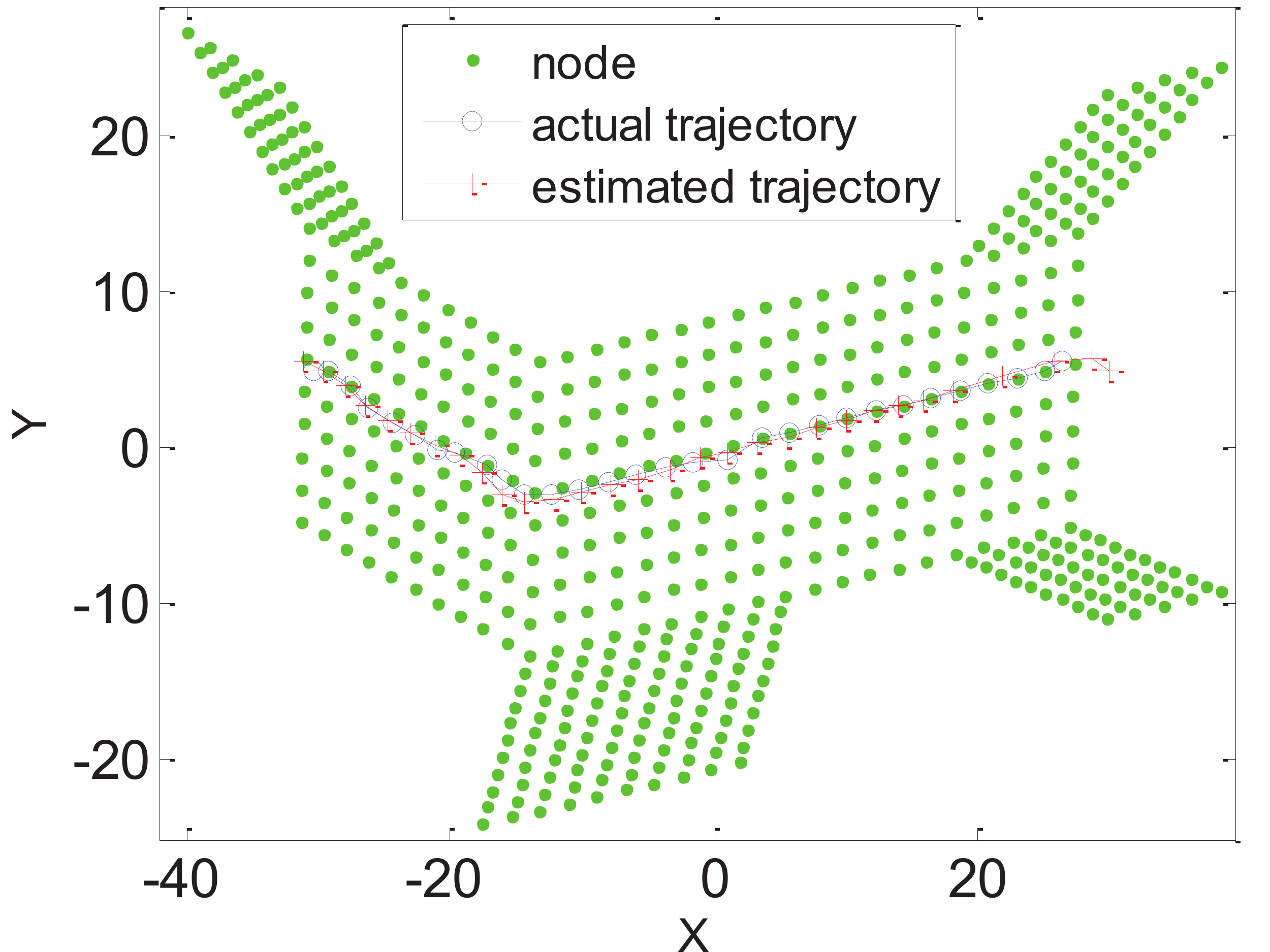}
        \caption{}
        \label{figure4}
        \end{subfigure}
    \begin{subfigure}[h]{0.45\textwidth}
        \includegraphics[width=\textwidth]{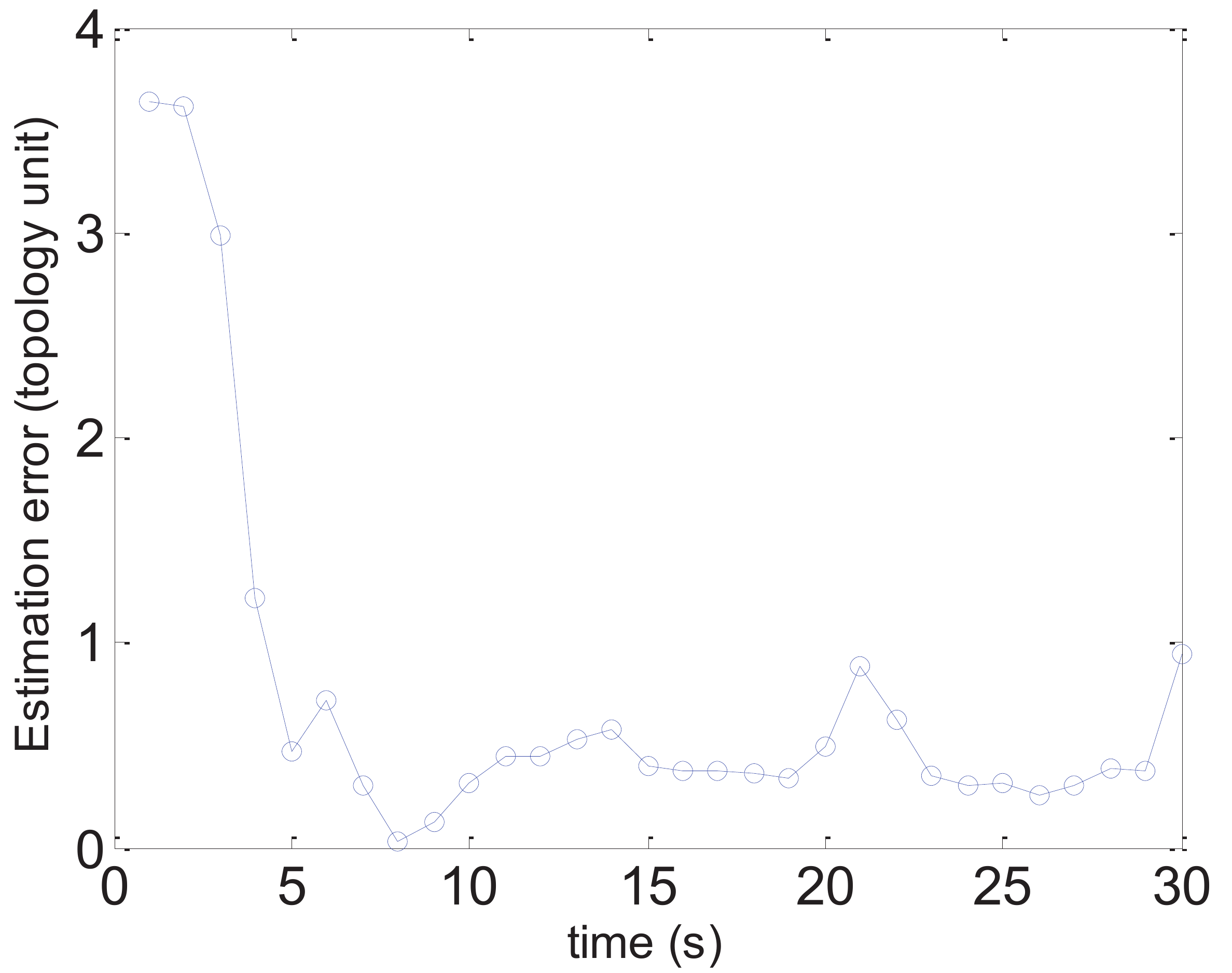}
        \caption{}
        \label{figure5}
    \end{subfigure}
    \caption{Motion 1. (a) Tracking result in TPM; (b) Errors.}
\end{figure}


\begin{figure}[t]
    \centering
    \begin{subfigure}[h]{0.45\textwidth}
        \includegraphics[width=\textwidth]{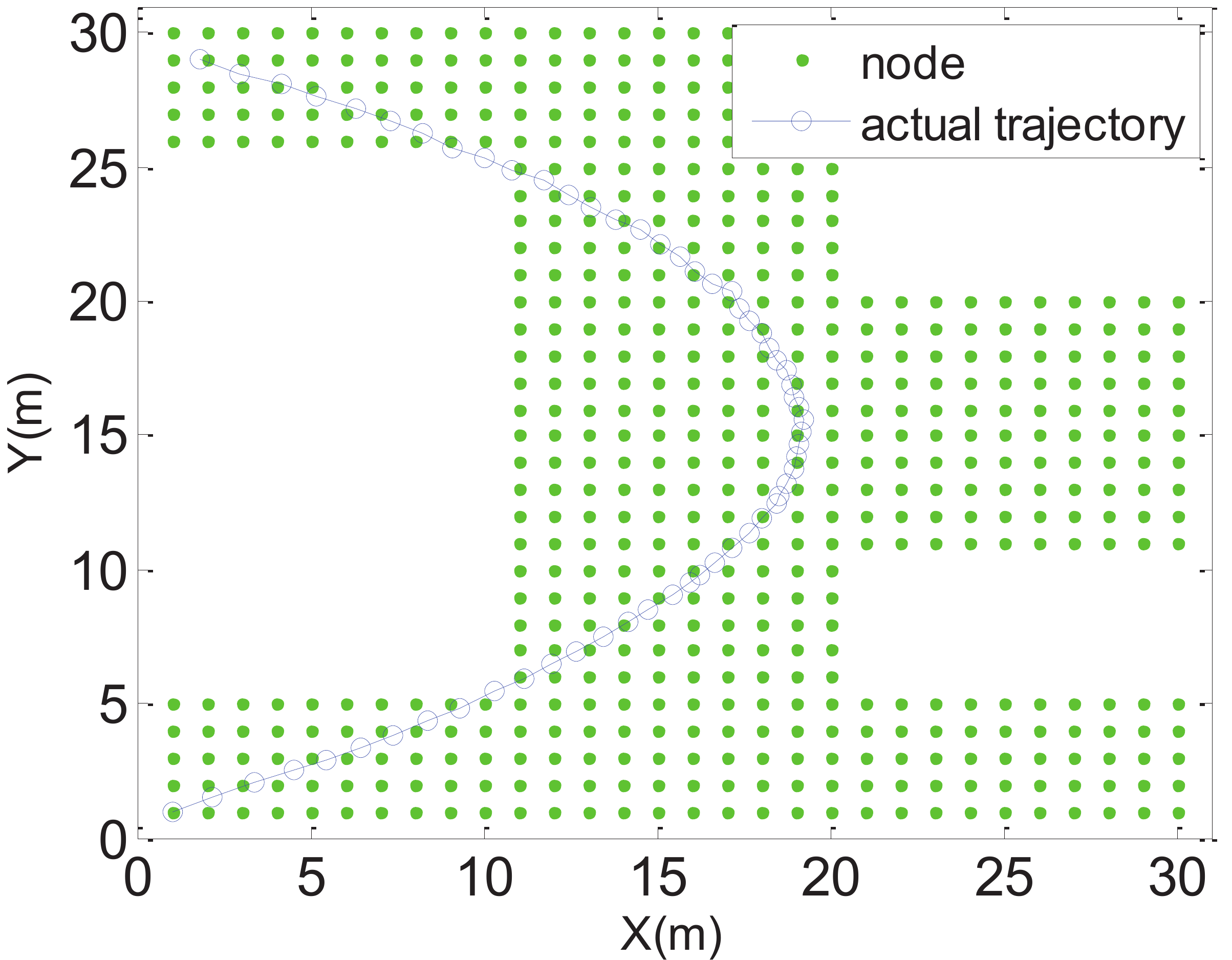}
        \caption{}
        \label{figure6}
        \end{subfigure}
    \begin{subfigure}[h]{0.45\textwidth}
        \includegraphics[width=\textwidth]{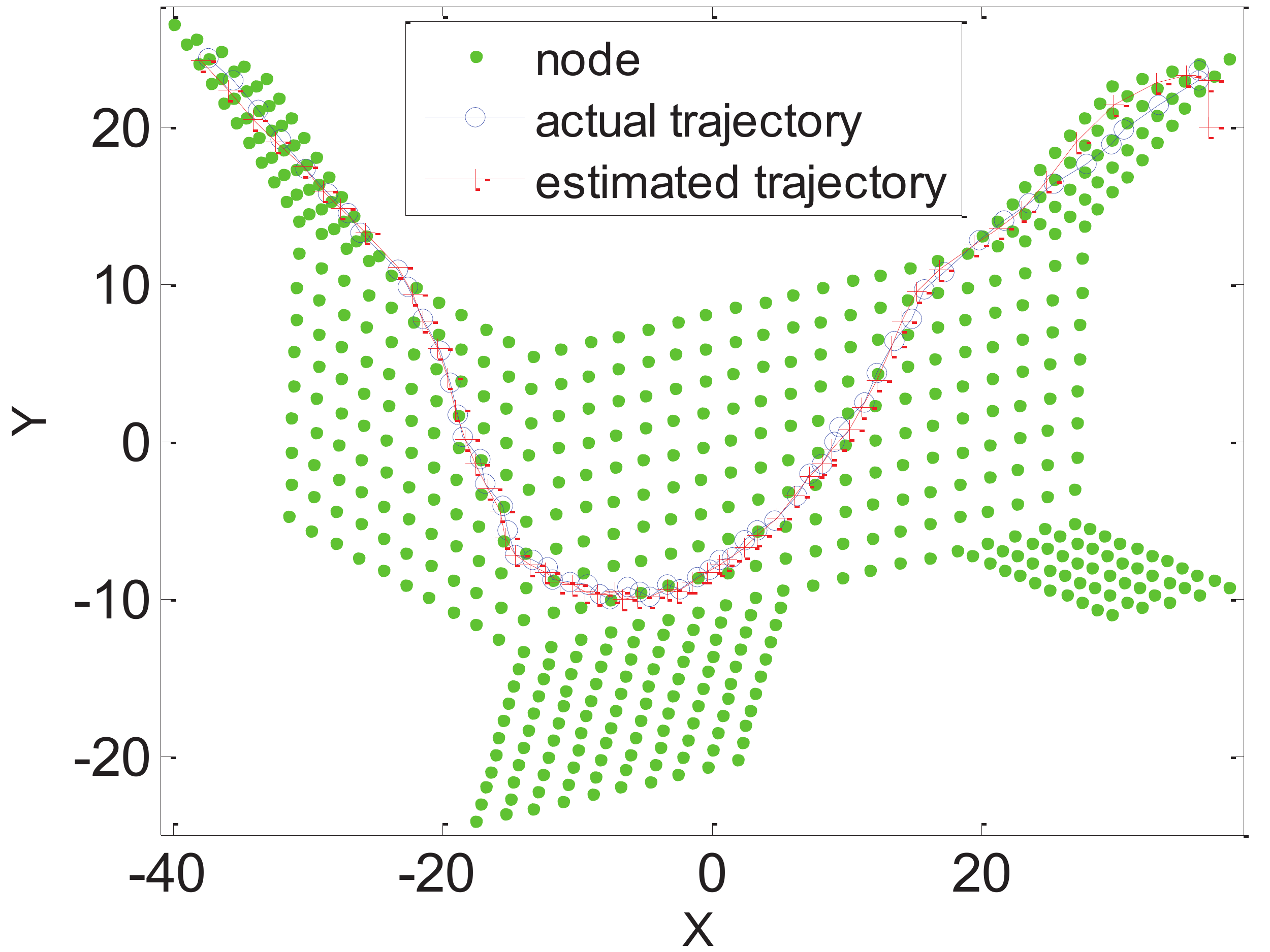}
        \caption{}
        \label{figure7}
    \end{subfigure}
    \begin{subfigure}[h]{0.45\textwidth}
        \includegraphics[width=\textwidth]{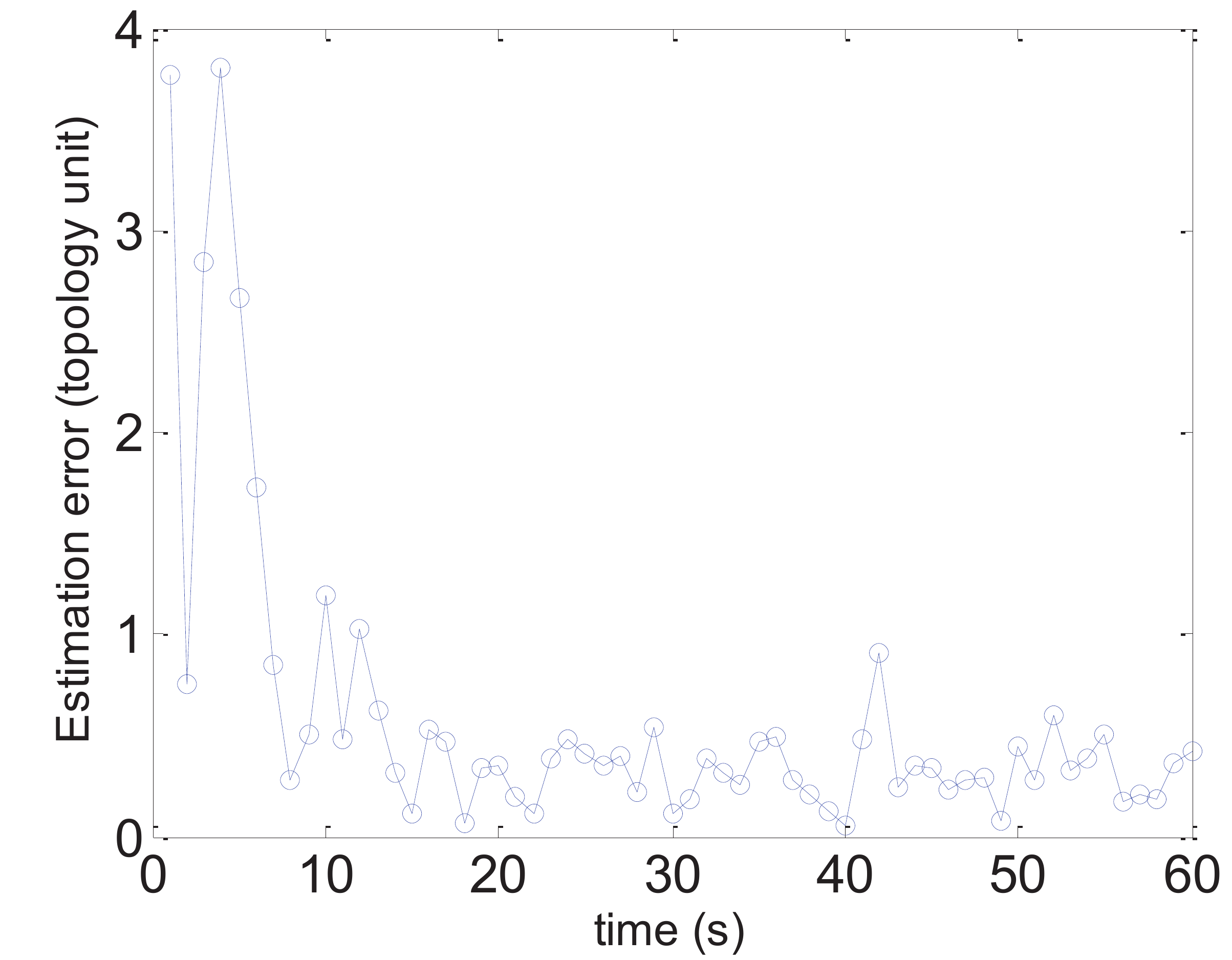}
        \caption{}
        \label{figure8}
    \end{subfigure}
    \caption{Motion 2. (a) Trajectory in physical domain; (b) Tracking result in TPM; (c) Errors.}
\end{figure}

\section{Topology Tracking VS Physical Tracking}\label{trackiing_compare}
In this part, we focus on the comparison of the tracking performances in topology domain and in physical domain. The last section demonstrates the effectiveness of REKF based tracking method in estimation of the mobile sink, while in this section we pay more attention to the prediction performance.
\subsection{Mobility Models}
We use a random movement \cite{mobility97} (Motion 3) to evaluate the proposed method. The speed and direction of movement in a new time period have no relation to the previous values. This model can generate mobile behaviour of pedestrians such as sharp turns or sudden stops. It has been studied in \cite{TPMtrack13} to evaluate the proposed method therein. But the speed of the mobile sink is required to be constant. In this chapter, we release this requirement and let the mobile sink move with varying speeds which have an upper bound. Besides, the mobile sink can change its direction by itself and it is assumed it must make a sharp turn when the mobile sink reaches the boundary of the sensing field.

In this case the system state variable can be represented by $x=[X,\dot{X},\ddot{X},Y,\dot{Y},\ddot{Y}]$, where $\ddot{X}$ and $\ddot{Y}$ are the accelerations along each axis.

In Motion 3 the different speeds may have two situations: 1) in different segments of the random movement, the speeds can distinguish from each other; but in each segment the speed keeps the same; 2) the speed in each segment of the movement is varying. These two cases are considered as a general version of the above Motion 1 and 2 respectively. In this section, we consider both of them.

\begin{figure}[t]
    \centering
    \begin{subfigure}[h]{0.45\textwidth}
        \includegraphics[width=\textwidth]{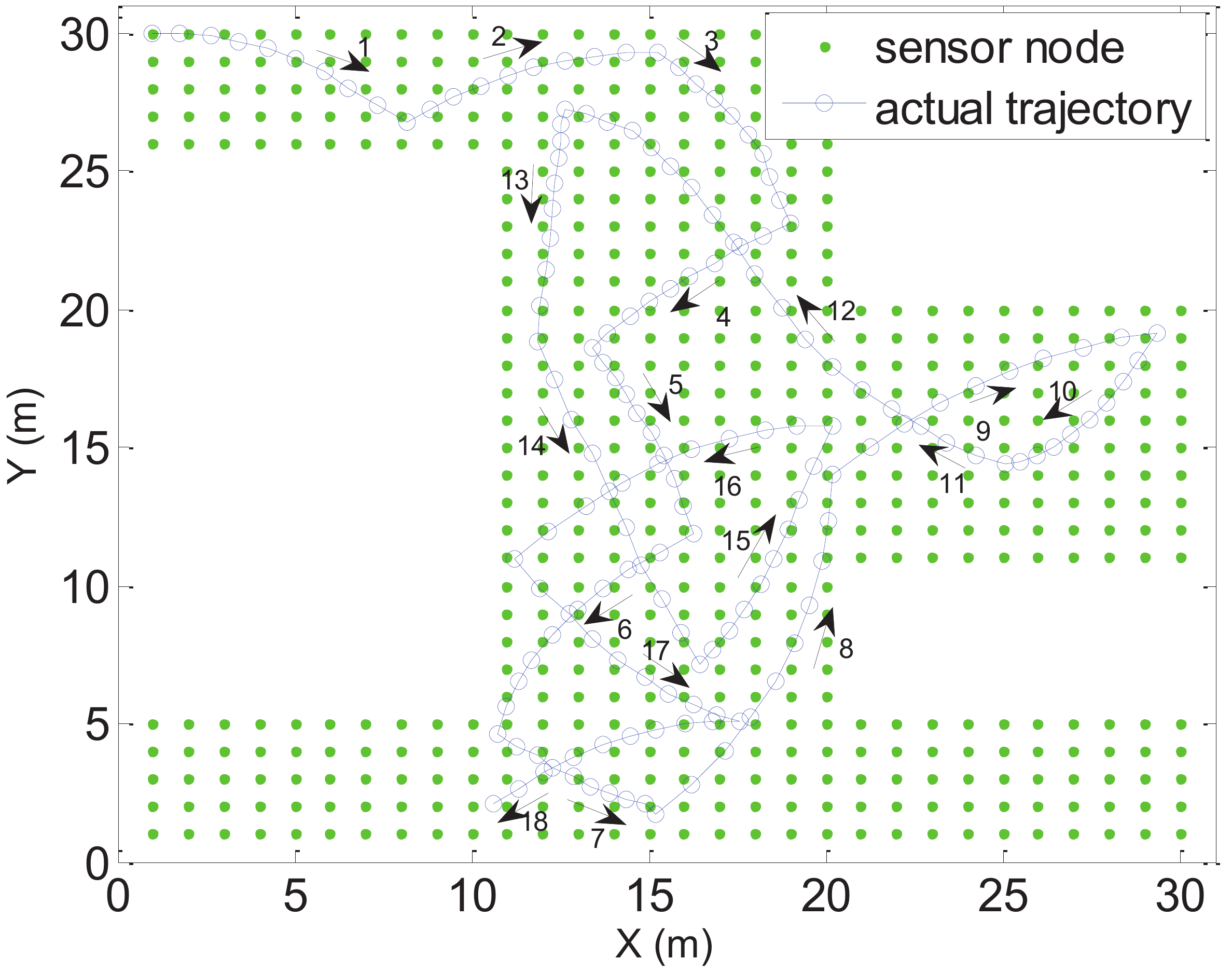}
        \caption{}
        \label{figure9}
        \end{subfigure}
    \begin{subfigure}[h]{0.45\textwidth}
        \includegraphics[width=\textwidth]{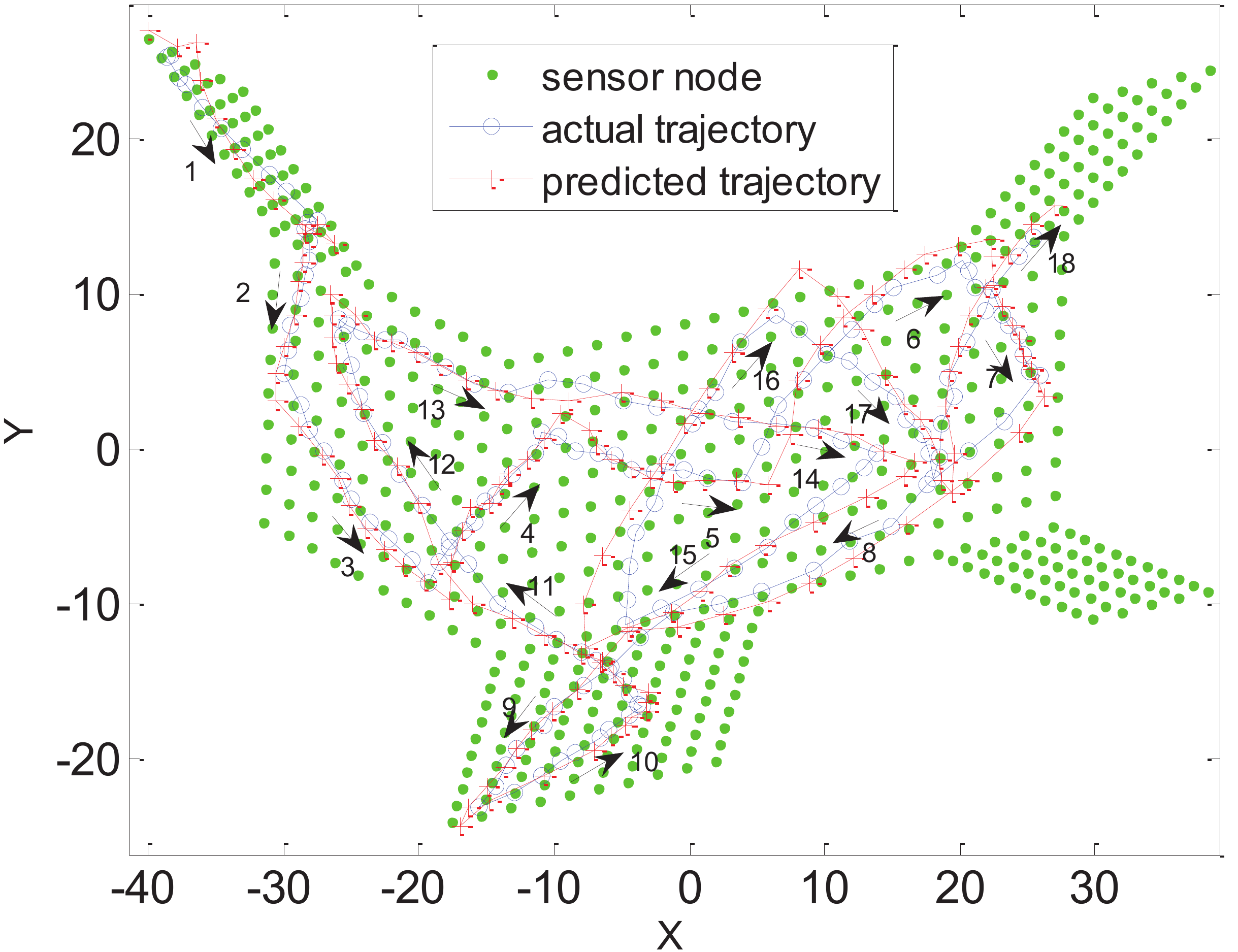}
        \caption{}
        \label{figure10}
    \end{subfigure}
    \begin{subfigure}[h]{0.45\textwidth}
        \includegraphics[width=\textwidth]{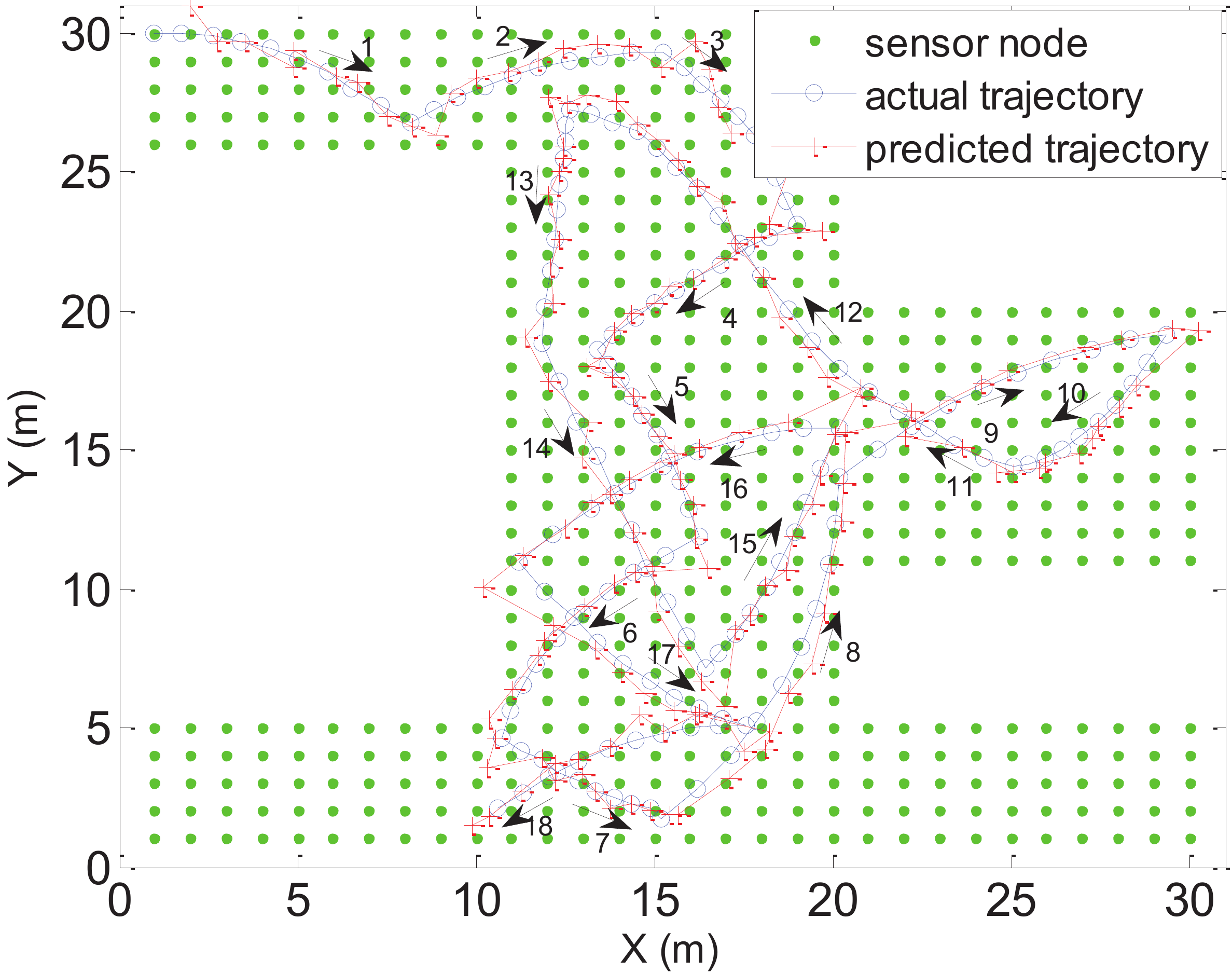}
        \caption{}
        \label{figure11}
    \end{subfigure}
    \caption{Random movement. (a) In physical domain; (b) Tracking in TPM; (c) Tracking in physical domain.}
\end{figure}

\subsection{Simulations}
Fig. \ref{figure9} shows a random movement in physical domain. The mobile sink moves in a random way (turning, accelerating or decelerating) and the number besides the arrow represents the moving order.

To track the mobile sink in topology domain, we first need to transform the movement into topology map. Then we need to implement the REKF based tracking algorithm. We still use a nearly constant speed movement model to estimate the mobile sink's location and predict its future position in topology map. The corresponding movement in topology map and the tracking result are both shown in Fig. \ref{figure10}. 

Different from Section \ref{tracking_topo}, we focus on the prediction performance here. So we only demonstrate the predicted trajectory in Fig. \ref{figure10}. From Fig. \ref{figure10} we can see that the REKF based estimation and prediction algorithm can make a good prediction for the mobile sink's future position. However, when the mobile sink makes a sudden turn, the prediction is not as accurate as the normal moving situations.

To evaluate the tracking performance using topology map, we also implement the REKF based algorithm in physical domain. We assume the sensor nodes in physical map are capable of measuring distance and their locations are known. The tracking result is shown in Fig. \ref{figure11}.

\subsection{Evaluations}
In this section, we compare the tracking performances demonstrated above in depth.

An intuitionist way to compare these two tracking results is based on the magnitudes of the prediction errors. No matter in which domain, the prediction error is defined as the gap between the actual position and the predicted position. We display the prediction errors in physical domain as well as in topology domain in Fig. \ref{figure12}. 

From Fig. \ref{figure12}, we can see the magnitudes of the prediction errors in topology domain are about twice of that in physical tracking. However, as the topology map is a distorted version of the physical coordinate based map, the map size and the space between two nodes may be different in two maps. Thus, we cannot evaluate the tracking performance by comparing the magnitude of the prediction error directly. One trial is to relate the two units of distance in the two domains. But we find this relation is not consistent across the whole network. Let's take the network used in this chapter as an example (Fig. \ref{figure2}). In the physical domain, all the sensor nodes are uniformly deployed. The space between any two sensor nodes is the same (Fig. \ref{figure2}(a)). However, the situation in its topology map is different (Fig. \ref{figure2}(b)). In some parts such as the top left part, it is denser than some other parts such as the middle part. It is possible to use 6 functions to relate the two units for this example. But this method is not feasible for large-scale and complex networks such as non-uniformly deployed networks. Therefore, we need to find a bridge to link these two maps.

In this chapter we provide three methods to evaluate the tracking performances in the two domains. The basic idea of the first method is: since topology domain has a different set of units, we may transform one result into the other's domain. Thus, these two tracking results can be compared under the same unit. We will first discuss the possibility of this method and then provide the comparison result.

As shown in \cite{TPM14}, any two adjoined nodes in physical domain are still next two each other in topology domain. In addition, any points on the link of these two nodes in physical domain are still between them in topology domain. So any predicted position in topology domain has a corresponding point in the physical domain. Now the problem is to find a way which transforms a point in topology domain to physical domain.

\begin{figure}[t]
    \centering
    \begin{subfigure}[h]{0.45\textwidth}
        \includegraphics[width=\textwidth]{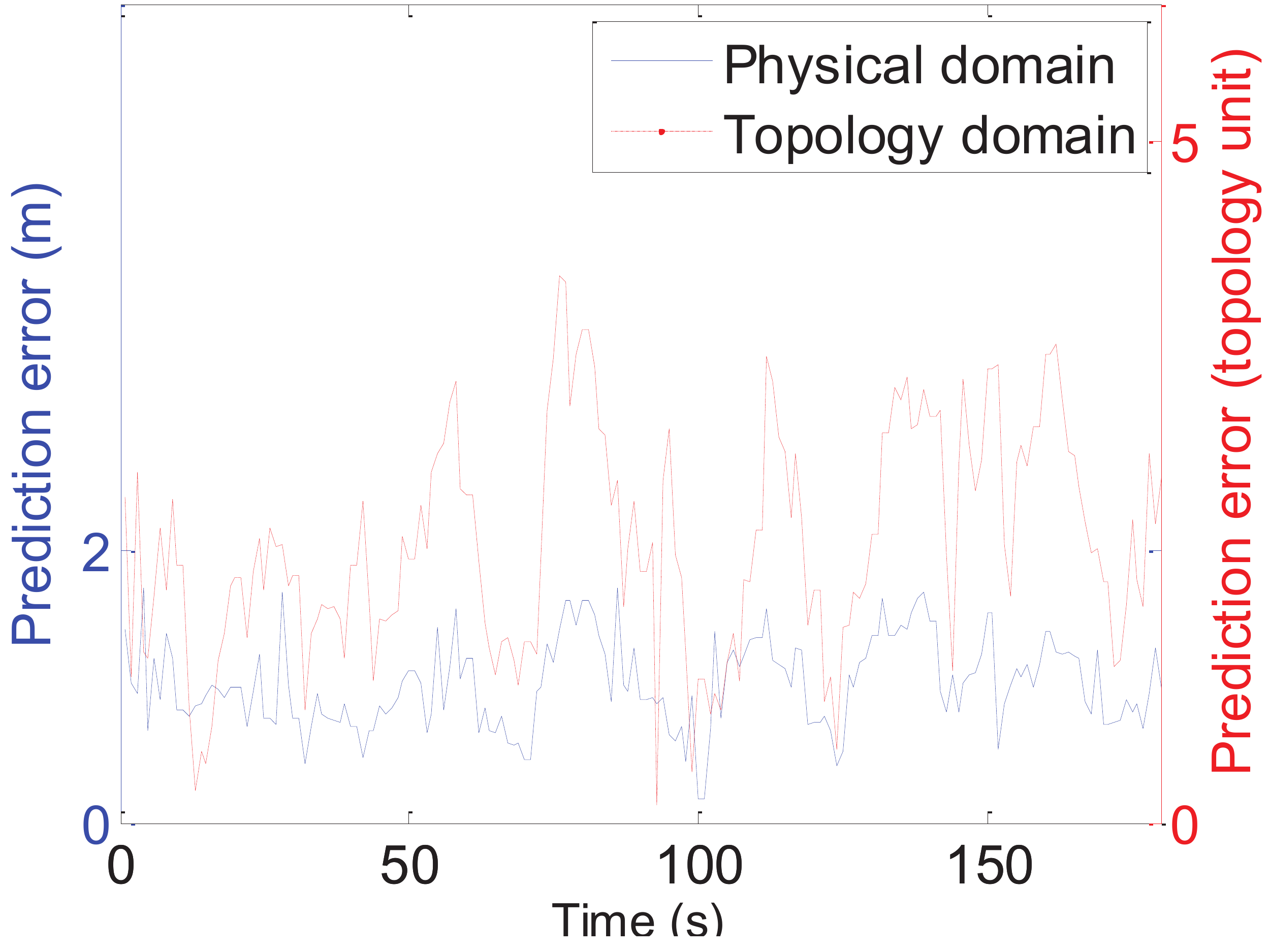}
        \caption{}
        \label{figure12}
    \end{subfigure}
     \begin{subfigure}[h]{0.45\textwidth}
        \includegraphics[width=\textwidth]{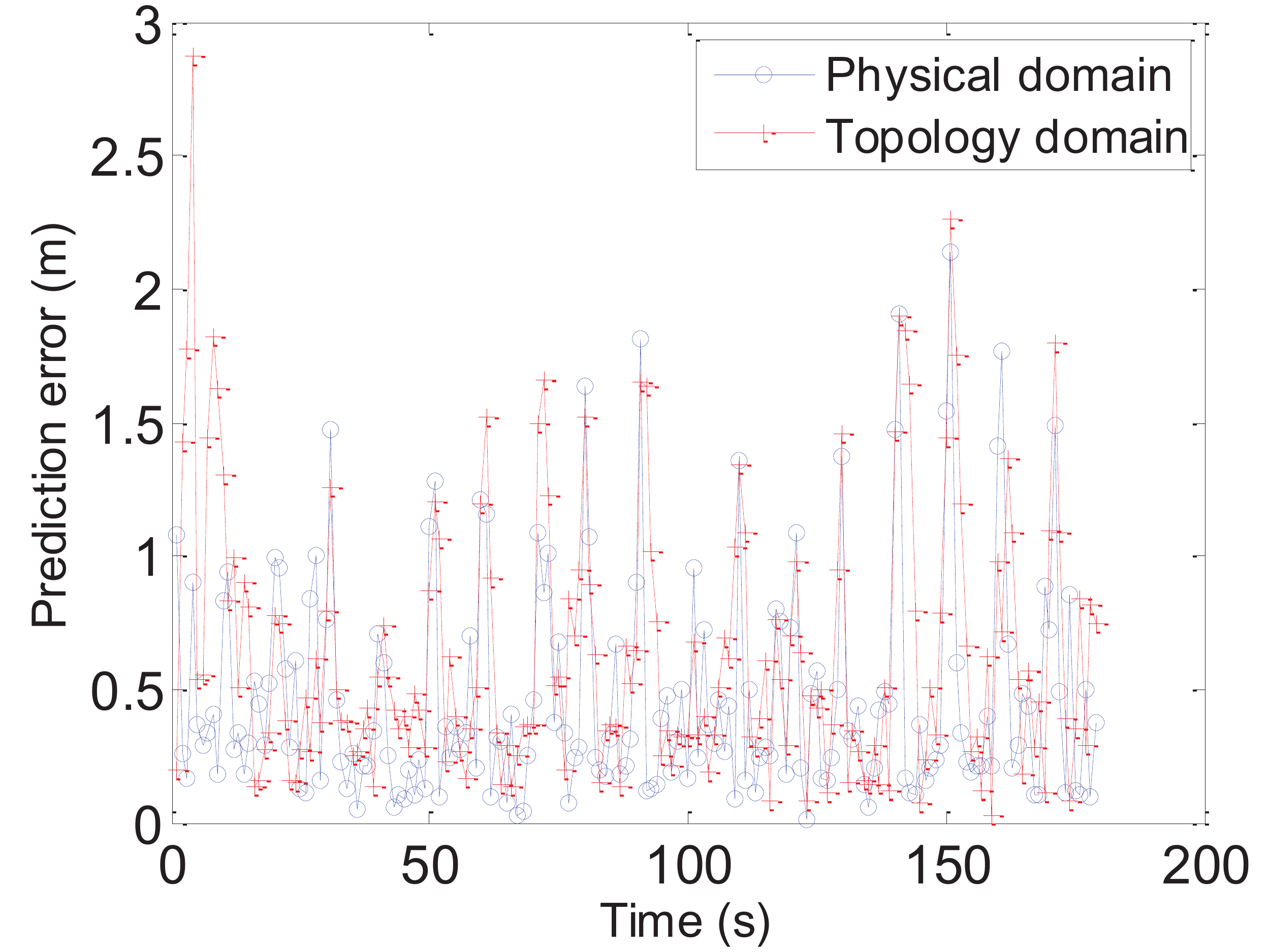}
        \caption{}
        \label{figure13}
    \end{subfigure}
     \begin{subfigure}[h]{0.45\textwidth}
        \includegraphics[width=\textwidth]{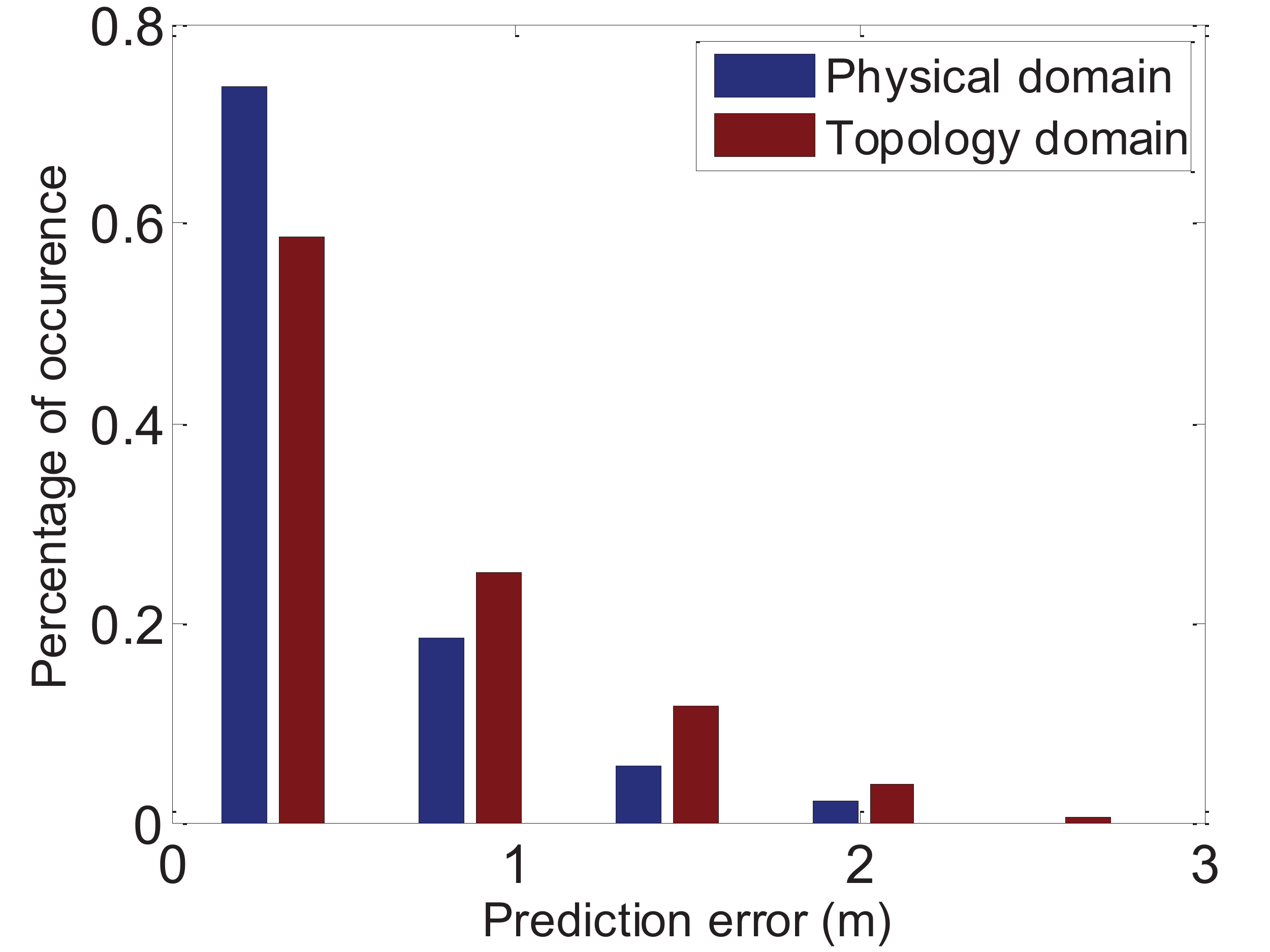}
        \caption{}
        \label{figure14}
    \end{subfigure}
    \caption{Comparisons on the random movement. (a) Errors in TPM; (b) Errors in physical domain; (c) Percentage of occurrence of prediction errors.}
\end{figure}
The topology map is generated from the virtual coordinates of the nodes by SVD. But so far we have no ideas on how to transform the topology map back to the physical map. We propose a rough method to compute the corresponding position in physical domain of the certain point in topology map. Inspired by the preserving characteristic, we use some nearby sensor nodes to locate a point's corresponding position in physical domain. Suppose a point is located at $(X_T,Y_T)$ in topology domain. We can select $k$ nearest sensor nodes by calculating the distances between this point and the all the sensor nodes' TCs and the set of these $k$ sensor nodes is represented by
\begin{equation}\label{eq26}
\begin{aligned}
S_k=\{s_i:&|(X_T,Y_T),(X_{Ts_i},Y_{Ts_i})|\leq \\
&|(X_T,Y_T),(X_{Ts_k},Y_{Ts_k})|,s_i,s_k\in S\} \\
\end{aligned}
\end{equation}
The corresponding position is calculated by taking the average of the physical coordinates of set $S_k$
\begin{equation}\label{eq27}
(X_P,Y_P)=\frac{\sum\limits_{s_i\in S_k}\left(X_{P_{s_i}},Y_{P_{s_i}}\right)} {k}
\end{equation}
Through (\ref{eq26}) and (\ref{eq27}), we can obtain a set of corresponding points of the predicted positions in topology domain tracking and these points are named as corresponding predicted positions. Following the definition of prediction error, we provide the comparison of prediction errors and the corresponding prediction errors in $meter$. The result is shown in Fig. \ref{figure13} and the occurrence of the prediction errors are shown in Fig. \ref{figure14}. From Fig. \ref{figure13} and \ref{figure14} we can see the corresponding prediction errors from topology domain are a bit larger than that of physical tracking. The average of the corresponding prediction errors is 0.65 $meters$ while that of physical tracking is only 0.47 $meters$. Therefore, tracking in physical domain outperforms that in topology map.

Now we propose the second method to evaluate the tracking performance. It looks at the problem from a different view compared to the first method. The basic idea is to evaluate the performance using the number of sensor nodes.

$Definition 1$: the required sensor nodes to do the tracking task are the ones that locate in a certain circle with the predicted position as center and the prediction error as radius.

Definition 1 is illustrated in Fig. \ref{figure15}. For each time slot, our algorithm predicts a future position. However there exists a gap between the prediction and the actual position. To detect the mobile sink, we select the sensor nodes locating in the red circle whose center is the prediction position and radius is the gap. These 8 nodes are the required sensor nodes to guarantee the mobile sink detection.
\begin{figure}[t]
    \centering
    \begin{subfigure}[h]{0.35\textwidth}
        \includegraphics[width=\textwidth]{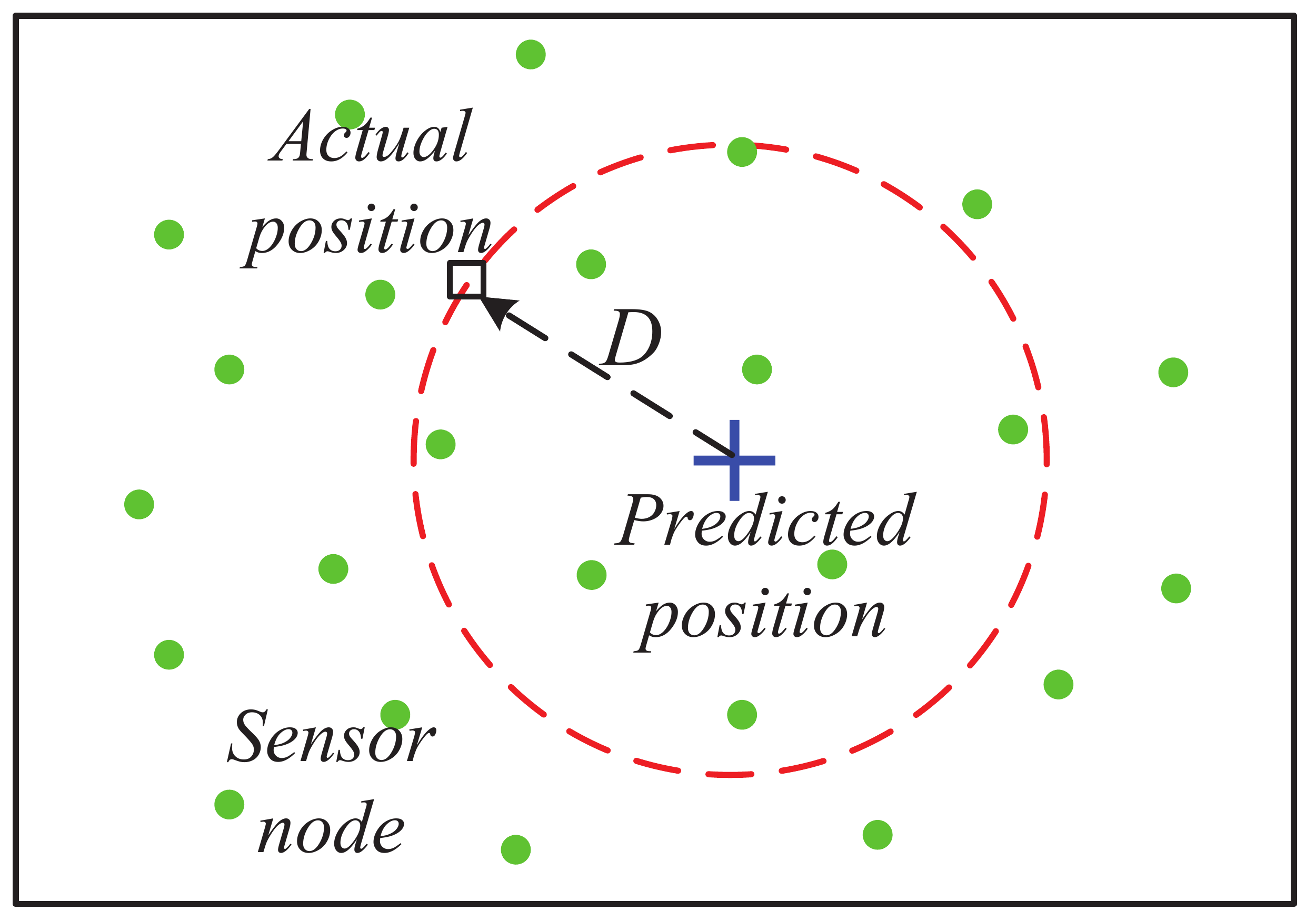}
        \caption{}
        \label{figure15}
        \end{subfigure}
    \begin{subfigure}[h]{0.45\textwidth}
        \includegraphics[width=\textwidth]{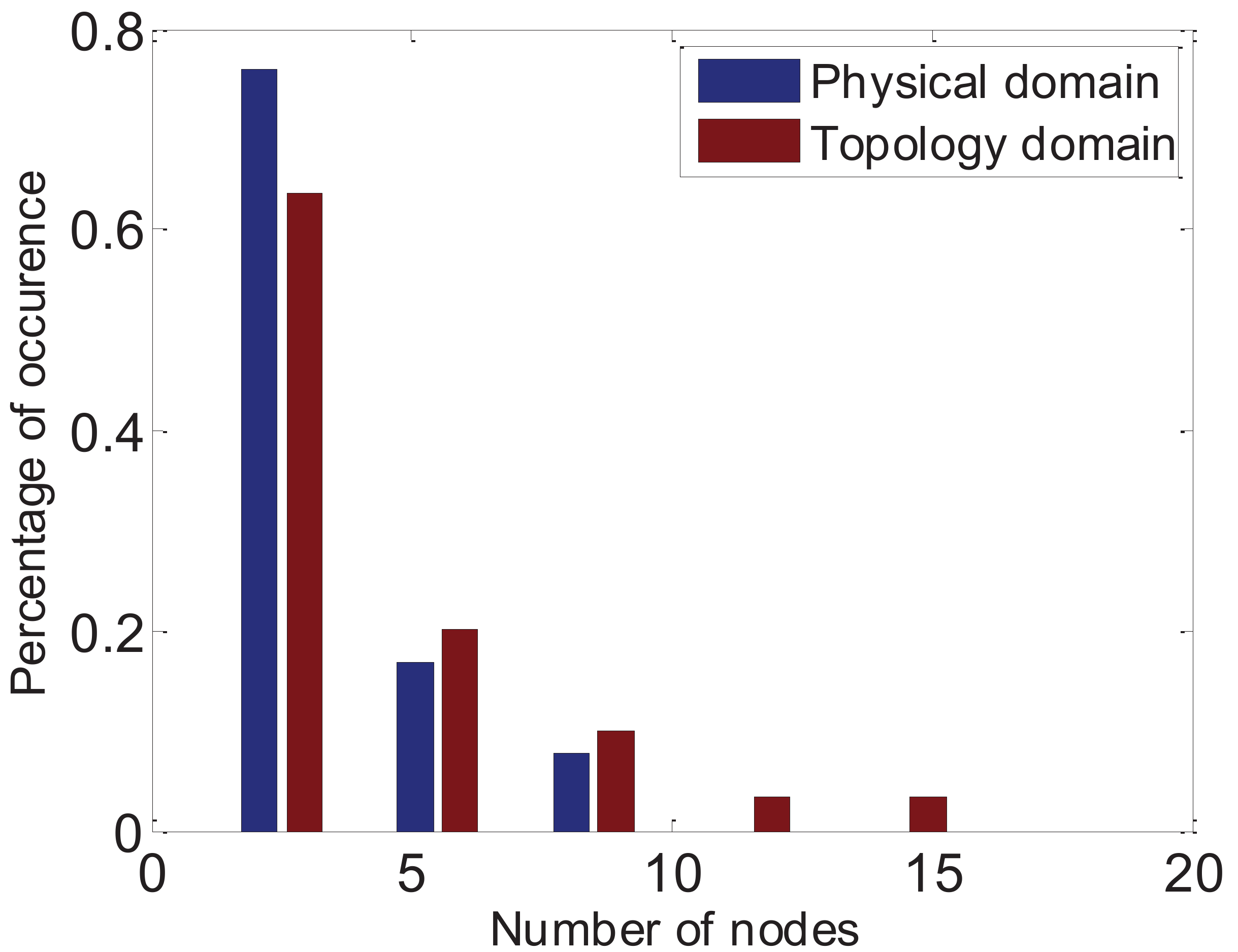}
        \caption{}
        \label{figure16}
    \end{subfigure}
    \caption{(a) Demonstration of required sensor nodes; (b) Comparison of the number of nodes to do detection in two domains.}
\end{figure}


There may be some cases where the prediction is very accurate. In other words, the gap between the predicted position and the actual position is too small to find one sensor node inside the certain circle. In these situations, we select the sensor node which is the nearest one to the predicted position.
We compare the number of the required sensor nodes in the two domains because this metric avoids the inconvenience of diverse distance units in physical domain and topology domain. The result is shown in Fig. \ref{figure16}. From Fig. \ref{figure16} we can see when we select the same number of sensor nodes in the two domains, for example 6, it has 92\% possibility to detect the mobile sink in physical domain. However, with this number of sensor nodes in topology domain, the detection possibility is 84\%. To achieve the same detection possibility, about 9 sensor nodes are needed in topology domain. To have 100\% possibility of detection for the whole movement, 9 and 15 are required in the two domains respectively. 

From the above analysis, the tracking performance in physical domain outperforms that in topology domain. The prediction error in the former is about 0.18 $meters$ averagely less than that of the latter and the number of sensor nodes to guarantee the detection of the mobile sink in physical domain is smaller than that in topology domain. The most significant reason is in physical domain, we assume the sensor nodes are able to measure the distance between themselves and the mobile sink, while in topology domain the sensor nodes can only output a receiving probability. Therefore, it is reasonable to use more sensor nodes to do detection in topology domain as the cost for the sensor nodes used there is lower than the nodes which are capable of measuring distance.

Finally, we focus on the effectiveness of tracking in topology domain by considering the above two conclusions. We mobile sink on another evaluation metric: the sensing area coverage in the following part. It can be seen from the first conclusion, the tracking performance using sensor nodes which are capable of distance measurement is generally better than that using topology map. In other words, the sensor nodes selected according to the prediction in physical domain generally have a higher possibility to detect the mobile sink. In topology domain, we can also select a subset of the sensor nodes to do detection. Instead of only comparing the number of these two sets as we do in the second method, we investigate how much these two sets overlap with each other. We provide our analysis result in physical domain which is more practical and we test the percentage that the topology set can cover the physical set.

As we see from Fig. \ref{figure13}, the prediction error of physical tracking is usually more accurate than that of topology tracking but there are some cases where the physical prediction is worse than the topology prediction. In those cases, the topology set of sensor nodes is more effective than the physical one, which means it is unnecessary to compute the coverage percentage by the topology set. So we only compare the sensing area coverage when the physical prediction error is less than the corresponding prediction error.

We select 9 and 15 sensor nodes in physical domain and topology domain respectively and compute the sensing area coverage. We provide two examples to demonstrate this coverage and they are shown in Fig. \ref{figure17}. In most time slots, the topology set of sensor nodes can fully cover the physical one (Fig. \ref{figure17}((a)), which means the topology set has no less possibility for detecting the mobile sink compared to the physical one. However, there are some time slots in which the topology set cannot 100\% cover the physical one (Fig. \ref{figure17}(b)). If the mobile sink appears in the area which is uncovered by the topology set, the sensor nodes selected by the topology tracking may fail to detect the mobile sink. 

As we only compare the sets when the physical prediction error is smaller than the topology one, the sensing area coverage can be defined as follows: 
\begin{equation}
Coverage=\frac{|T\cap P|}{|P|}
\end{equation}
where $P$ and $T$ are two sets of sensor nodes and $|\cdot|$ calculates the element number. The statistic result of this sensing area coverage is shown in Fig. \ref{figure18}.

From Fig. \ref{figure18} we can see, in all the time slots we test, the topology set of sensor nodes can at least cover the physical set by 50 percent and there are 86\% cases where this coverage can reach more than 80 percent. Therefore, the set of sensor nodes selected in topology domain also has a high quality in terms of detecting the mobile sink.
\begin{figure}[t!]
\begin{center}
{\includegraphics[width=0.45\textwidth,height=0.6\textwidth]{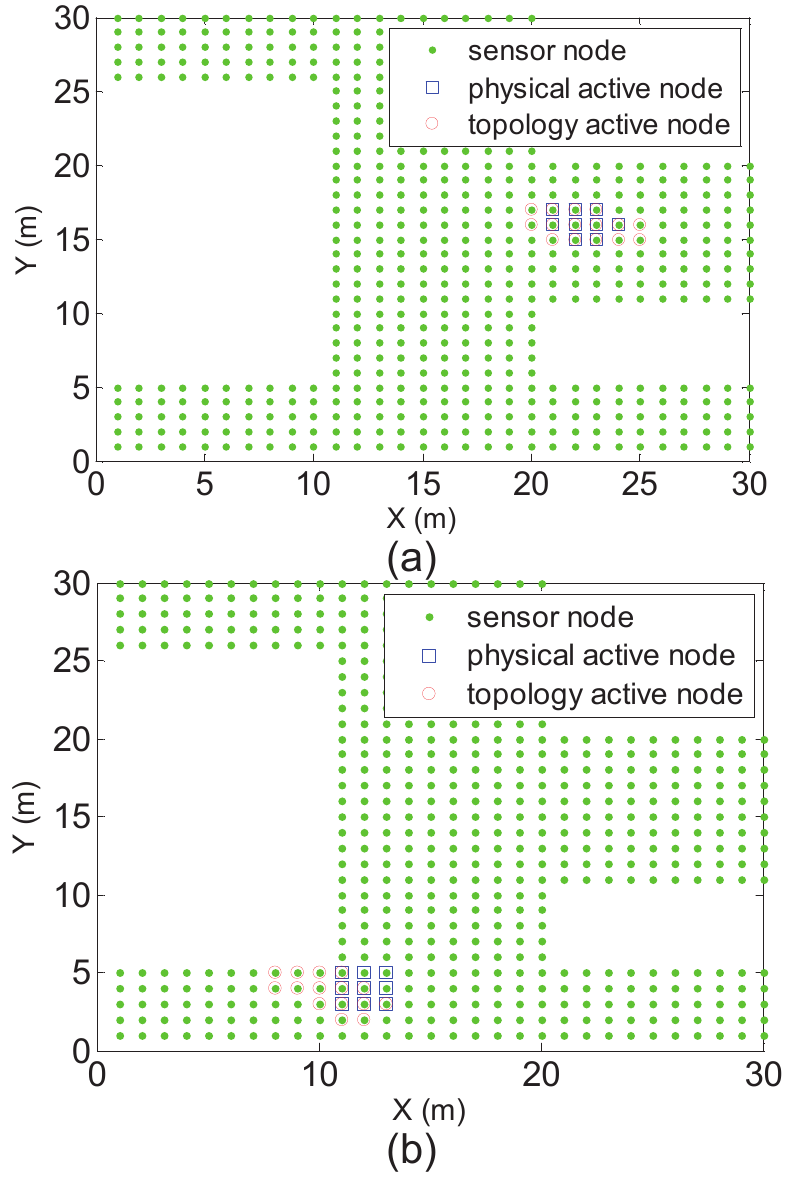}}
\caption{Two illustrative examples for the sensing area coverage. (a) Completed coverage. (b) Uncompleted coverage}\label{figure17}
\end{center}
\end{figure}
\begin{figure}[t!]
\begin{center}
{\includegraphics[width=0.4\textwidth,height=0.3\textwidth]{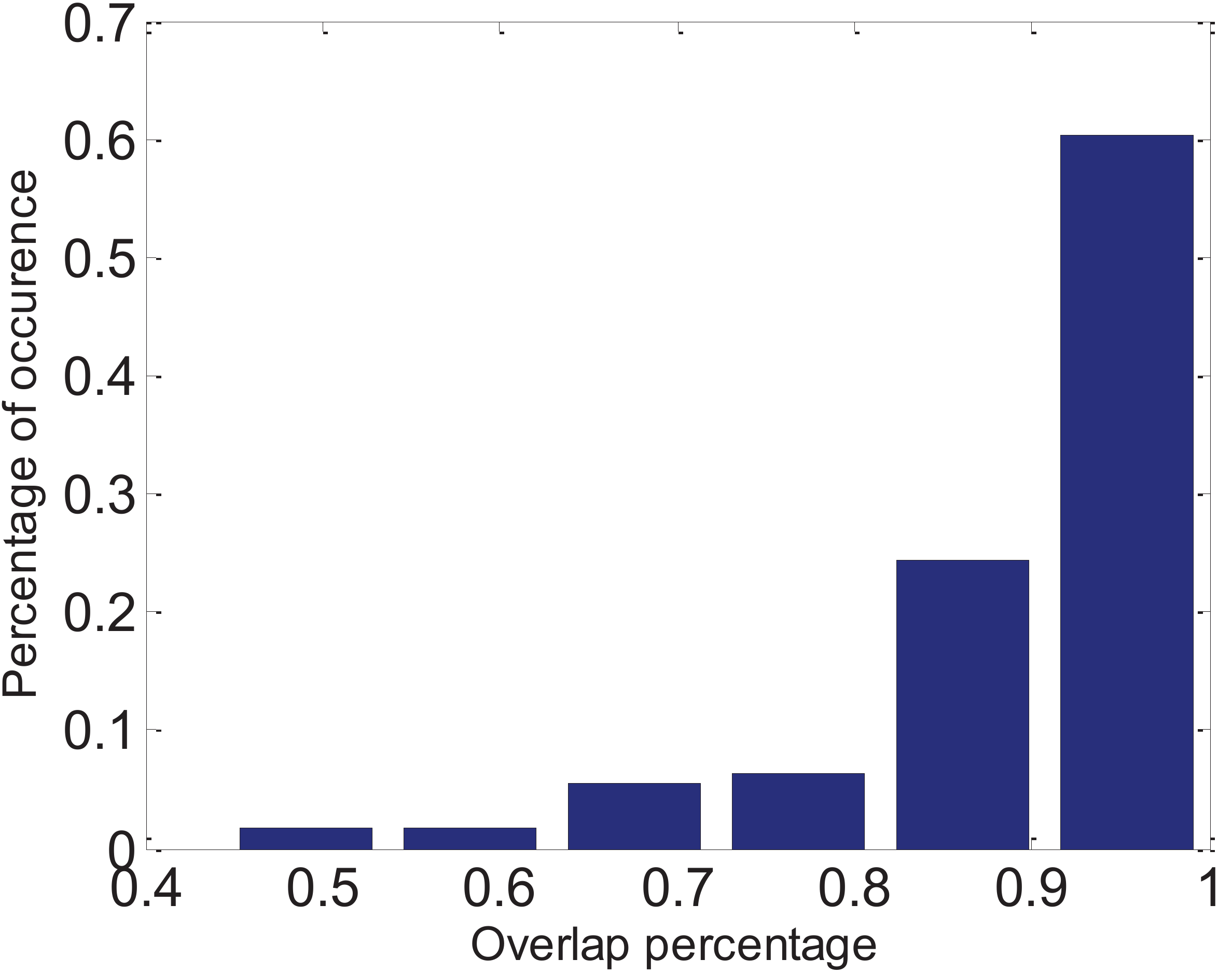}}
\caption{Overlap of the two sets of active sensor nodes in two domains}\label{figure18}
\end{center}
\end{figure}


\section{Summary}
\label{conclusion_tracking}
In this chapter, we tackled the problem of tracking a mobile sink in a wireless sensor network where the sensor nodes are incapable of distance measurement and location unaware. We utilized the sensor nodes' topology coordinates and their receiving probabilities to assist the implementation of Robust Extended Kalman Filter in topology domain. The REKF based algorithm can deal with the uncertainties in topology map. Various mobility models have been studied, ranging from simple models such constant or varying speed models to random motion models, to demonstrate the effectiveness of the proposed estimation and prediction method. Moreover, we proposed three metrics to evaluate the tracking performance. We compared our work in the topology map with that in the physical map. Our results indicated the proposed method is competitive with themes that are based on distance measurements. Our framework is costless and is more feasible for large-scale networks with inexpensive sensor nodes. 

The work presented in this chapter plays a fundamental role in the approaches where mobile sinks are used to serve WSNs, since the performance of them may depend on the current location of mobile sinks.
\chapter{Viable Path Planning for Data Collection Robots in a Sensing Field with Obstacles}\label{path_planning}
\minitoc
\section{Motivation}
This chapter presents the first main contribution of the report, which falls into the category of using \textit{\textbf{controllable}} mobility to collect data in WSNs.

Mobile sink has been regarded as a promising method to prolong the network lifetime. By physical movement, mobile sinks can save much energy resource at sensor nodes, since the communication between them can be done in a relatively short distance. As a coin has two sides, one defect of such scheme is that the data delivery delay is also increased due to the relative slow physical movement. Thus, one important research topic is how to design the path for mobile sinks such that all the sensory data can be collected and at the same time, the delivery delay is minimized. One basic approach is to view such problem as the traditional Travelling Salesman Problem (TSP) \cite{TSP85, TSP94}. The mobile sink is regarded as a salesman and the sensor nodes are regarded as the delivery destinations. Then, the problem is to find the shortest path in length such that every sensor node is visited exactly once. When multiple mobile sinks are used, the corresponding problem can be viewed as the Vehicle Routing Problem (VRP) \cite{VRP92}. 

Unlike the point-wise model in TSP and VRP based approaches, in this chapter, the mobile sinks are modelled as a unicycle moving at a constant speed with bounded angular velocity. This model is also called the Dubins car \cite{dubins1957} and it is well known that the motion of many wheeled robots and unmanned aerial vehicles can be described by this model \cite{MTS11, andrey11unicycle_2, andrey11unicycle_3, savkin2013reactive}. Recent studies have proposed various solutions to the path planning problem in data collection. However, several issues still need further investigation. One is that many existing approaches set the sensor nodes' locations as target positions for the robots, e.g., see \cite{visitingcircle2015}, which leads to the collisions with the sensor nodes. Another issue is that most studies assume the sensing field is obstacle-free, e.g., see \cite{obstaclefree10}, which is quite ideal in practice. Further, the dynamic constraint of robots is rarely taken into account in recent algorithms, then the produced paths may have non-smooth corners at the intersections \cite{rendezvous08}. This is a severe limitation in practice because the paths cannot be applied to some robots such as the considered unicycle robots with bounded angular velocity. Some approaches assume the data transfer time from the source sensor nodes to the robots is negligible \cite{MES04, MES07}. It is reasonable for the light data load nodes, but not for the nodes with heavy data loads, such as the nodes equipped with cameras to snapshot over-speed vehicles \cite{pathconstrained07}. 

With the mentioned concerns, we define a \textit{viable} path which is smooth, collision-free with sensor nodes/base station and obstacles, closed, and provides enough contact time with all the sensor nodes. The viable path takes into account properties of both robotics and sensor networks, which is close to reality. The main objective is to design the shortest viable path for the considered robots. We formulate the problem as a variant of Dubins Travelling Salesman Problem with neighbourhoods (DTSPN) \cite{OOD10, IH13}. To solve the problem, we propose a Shortest Viable Path Planning (SVPP) algorithm. In essence, our algorithm is based on a roadmap: Tangent Graph, which is then modified by a reading adjustment to provide enough contact time for each node. With the modified Tangent Graph, it determines a permutation of sensor nodes by solving an Asymmetric TSP (ATSP) instance. Having the permutation, the modified Tangent Graph is simplified and converted to a tree-like graph, where the edges and vertices unrelated to the permutation are removed. Finally, the shortest viable path is figured out by searching the tree-like graph using a dynamic programming based method. 

The viable paths produced by SVPP algorithm can be travelled by the considered unicycle robots periodically. However, such paths are designed for single mobile sink. For the large scale network, we further consider the situation where multiple mobile sinks are available, which is promising to reduce the average path length for each sink and shorten the delivery delay. We consider $k$ identical mobile sinks and target on designing viable paths for them such that the lengths of the $k$ paths are more or less equal. Then the data delivery delays on these paths can be similar. We first discuss an algorithm by introducing viable path to $k$-SPLITOUR \cite{FHK78} (denoted as viable $k$-SPLITOUR), which constructs a whole path using SVPP and then splits it into $k$ subpaths of more or less equal lengths. We point out that the generated $k$ viable paths are not guaranteed to be optimal. With this regard, we conduct a further operation, i.e., reconstruct the $k$ paths respectively using SVPP and this algorithm is referred as $k$-SVPP. It is easy to understand $k$-SVPP achieves no worse performance than viable $k$-SPLITOUR, and in many cases $k$-SVPP performs better.

For performance evaluation, extensive simulations are conducted. First, SVPP is applied to figure out viable paths for the network instances with different numbers of nodes and topologies where obstacles exist. We investigate the influences of two factors: robot speed and data load on system metrics: path length and collection time. Second, the performance of $k$-SVPP is demonstrated by comparing with Viable $k$-SPLITOUR. Third, the comparison of the proposed methods with a multihop communication algorithm: Shortest Path Routing is also provided. We study the influence of data load distributions on data collection performance in terms of energy consumption. We find that using mobile robots to collect data saves around 95\% energy for sensor nodes compared to multihop communication, which definitely increase the network lifetime.  

Moreover, we use the idea of tangent graph to navigate a flying robot to avoid detections by the ground sensors, which is called safe mission planning in various threat environments. 

The key contributions of this chapter include:

\begin{enumerate}
\item
We propose the concept of viable path which combines the concerns of robotics and sensor networks.
\item
We provide a formulation for the studied problem as a variant of DTSPN and an algorithm called Shortest Viable Path Planning (SVPP). 
\item
We provide a k-SVPP algorithm to design $k$ viable paths which have approximately equal lengths.
\item
We present extensive simulations to demonstrate the effectiveness and advantages of our algorithms.
\item
We further consider the problem of navigating a military aircraft in a threat environment to its final destination while minimizing the maximum threat level and the length of the aircraft path. The proposed method to construct optimal low-risk aircraft paths involves a simple geometric procedure and is very computationally efficient.
\end{enumerate}

The remaining parts of this chapter are structured as follows. Section \ref{problemstatement} formally describes the studied problem. Section \ref{algorithm_basic} provides the suggested SVPP algorithm and a \textit{k}-SVPP algorithm dealing with path planning for $k$ robots is proposed in Section \ref{algorithm_extend}. Section \ref{simulation_path} demonstrates the performance of our algorithms by simulations and comparisons with the alternatives. Some discussions are provided in Section \ref{discussion}. Furthermore, the extension of our approach to the problem of safe mission planning is discussed in Section \ref{navigation}. Finally, Section \ref{conclusion_path} concludes the chapter. The publications related to this chapter include \cite{huang2017viable}, \cite{savkin2017optimal}, \cite{huang2016path}, \cite{savkin2016problem}.
\section{System Model and Problem Statement}
\label{problemstatement}
This section describes the system model, gives the basic assumptions and states our problem formally. The main notations used in this chapter are listed in Table \ref{table:notation}.

\begin{table}[h!]
\begin{center}
\caption{Notations}\label{table:notation} 
  \begin{tabular}{| c |l |}
    \hline
    \textbf{Notation} & \textbf{Comments}\\ \hline
    $X$ & Configuration of robot \\ \hline
    $(x,y)$ & Position of robot\\ \hline
    $\theta$ & Heading of robot\\ \hline
    $u$ & Angular velocity of robot \\ \hline
    $v$ & Speed of robot\\ \hline
    $u_M$ & Maximum angular velocity of robot\\ \hline
    $R_{min}$ & Minimum turning radius of robot\\ \hline
    $n$ & The number of base station and nodes\\ \hline
    $s_i$ & Base station or sensor node\\ \hline
    $m$ & The number of obstacles\\ \hline
    $o_j$ & Obstacle\\ \hline
    $d_{safe}$ & Obstacle safety margin\\ \hline
    $\partial o_j$ & Boundary of the convex hull\\ \hline
    $C_i$ & Visiting circle\\ \hline
    $G(V,E)$ & (Modified) Tangent Graph\\ \hline
    $\cal L$& Viable path construction function\\ \hline
    $g$& Data load\\ \hline
    $r$& Data transmission rate\\ \hline
    $e$& Transmission energy consumption rate\\ \hline
    $\delta$ & Required contact time\\ \hline
    $\Sigma$& Permutation of visiting with length $n$\\ \hline
    $\sigma_i$ & The $i^{th}$ element of $\Sigma$\\ \hline
    $G(V',E')$ & Simplified Tangent Graph\\ \hline
    $\cal P$& Projection function from $SE(2)$ to $\mathbb{R}^2$\\ \hline
    $K$ & The blocking number\\ \hline
    $\Sigma'$& Extended permutation $(n'=n+K)$\\ \hline
    $\sigma_i'$ & The $i^{th}$ element of $\Sigma'$\\ \hline
    $\cal L'$& Viable path construction function\\ \hline
    $T$& tree-like graph, $T_i$ is the $i^{th}$ layer of $T$\\ \hline
    $L$& Path length\\ \hline
    $z$& Configuration variable\\ \hline
    $P$& Viable path\\ \hline
    $k$& The number of robots\\ \hline
    $len$& Path length function\\ \hline
    \end{tabular}
\end{center}
\end{table}

Consider a planar robot modelled as a unicycle, whose state can be represented as the configuration $X=(x,y,\theta)\in SE(2)$, where $(x,y)\in \mathbb{R}^2$ is the robot's position and $\theta \in \mathbb{S}^1$ is its heading. The robot travels with a constant speed $v$ and is controlled by the angular velocity $u$. The dynamics of the robot can be described as

\begin{equation}
\label{dubinscar}
\begin{array}{l}
\dot{x}(t) = v \cos \theta (t)
\\
\dot{y}(t) = v \sin \theta (t)
\\
\dot{\theta}(t) = u(t)\in [-u_M, u_M]
\end{array}
\begin{array}{l}
\end{array} 
\end{equation}
where $u_M$ is the given maximum angular velocity. Such a model describes a planar motion of many ground robots, missiles, UAVs, underwater vehicles etc \cite{tang2005motion, rathinam2007resource, MTS11, savkin2013reactive, manchester2006circular, kim2011minimum, MSF08, piazzi2007mmb, ST10, zarchan2012tactical, shima2011intercept, gottlieb2017multi, yang2002optimal}. The standard non-holonomic constraint can be represented as:
\begin{equation}
\label{nonh}
|u(t)|\leq u_M.
\end{equation}
Then the minimum turning radius of the robot is
\begin{equation}
\label{Rmin} R_{min}= \frac{v}{u_M}.
\end{equation}

Any path  $(x(t),y(t))$ of the robot (\ref{dubinscar}) is a plane curve satisfying the following constraint on its so-called average curvature (see \cite{dubins1957}): let $P(a)$ be this path parametrized by arc length, then
\begin{equation}
\label{ac}
\|P'(a_1)-P'(a_2)\|\leq \frac{1}{R_{min}}|a_1-a_2|.
\end{equation}
Here $\|\cdot\|$ denotes the standard Euclidean vector norm. Notice that we use the constraint
(\ref{ac}) on average curvature because we cannot use the standard definition of curvature from differential geometry  \cite{struik2012lectures} since the curvature may not exist at some points of the robot path.

Consider a sensor network deployed in a cluttered environment. It consists of a base station ($s_1$) and $n-1$ sensor nodes ($s_i,i\in [2,n]$). We overload $s_i$ ($i\in [1,n]$) as a sensor node/base station and its location. $s_1$ executes the path planning task and the robot downloads the path before it depart. Consider $m$ disjoint and smooth obstacles ($o_j,j\in [1,m]$) in the field. The locations of all the stationary sensor nodes and the obstacles as well as the shapes of the obstacles are known. Let $d_{safe}$ be the safety margin for the obstacles and then we can get the safety boundary of $o_j$, see Figure \ref{fig:convex_hull} for an example. Since $o_j$ can be a non-convex set, its safety boundary can also be non-convex. As discussed below, our objective is pass the obstacles as fast as possible. Then, moving along the boundary of the convex hull of $o_j$ saves time comparing to moving along the winding safety boundary. Let $\partial o_j$ ($j\in [1,m]$) be the boundary of the convex hull of $o_j$'s safety boundary. The way to construct such convex hull is simple. For the non-convex part, a common tangent line is placed, see Figure \ref{fig:convex_hull} for an example. Regarding constraint (\ref{ac}), we make the following assumption on $\partial o_j$, $j\in [1,m]$.
\begin{Assumption}\label{curvature}
Any $\partial o_j$ ($j\in [1,m]$) is a smooth curve with the curvature $c(p)$ at any point $p$ satisfying $c(p)\leq \frac{1}{R_{min}}$.
\end{Assumption}

The objective is to design the shortest viable path for robot (\ref{dubinscar}). We define the viable path as follows:

\begin{definition}\label{viable}
A path $P$ is viable if the following five Conditions are satisfied: 1) smooth, collision-free of 2) sensor nodes/base station and 3) obstacles, 4) closed and 5) offers enough contact time to read all the data from sensor nodes.
\end{definition}

The motivations behind the viable path include the following practical concerns. First, as aforementioned, since the considered robots' bounded turning radius, it is inappropriate for them to turn at a spot. Thus, the foremost condition of the viable path is smoothness. Second, the robot should not traverse the locations of the sensor nodes, otherwise collisions may happen. Besides, since obstacles are possible to exist in the sensing field, the robot should also be able to avoid them. Further, as many applications require the robot to execute the data collection task periodically, the closeness feature may facilitate this requirement. Finally, to read all the sensory data from sensor nodes, the produced path should provide enough time for the robot to communicate with sensor nodes. Next, we detail how to design viable paths.

With respect to Condition 2) and considering the dynamics (\ref{dubinscar}), we first describe our visiting circle model.

\begin{definition}
The visiting circle is centred at the location of a sensor node/base station with the radius of $R_{min}$.
\end{definition}

The radius of the visiting circle can be smaller than $R_{min}$. But as $R_{min}$ is the minimum turning radius, the robot cannot move along a circle whose radius is smaller than $R_{min}$. Being able to move along the visiting circles facilitates our work in the applications where the stored data cannot be uploaded in a short time. In this case, the robot can rotate around the sensor nodes/base station on the visiting circles until the data uploading is finished. Thus, the radius of the visiting circle is set as $R_{min}$. Similar to $d_{safe}$, $R_{min}$ can also be regarded as the safety margin of sensor nodes/base station. We use $C_i,i\in [1,n]$ to represent the visiting circle and we use the term \textit{element} to represent either a visiting circle $C_i$ ($i\in [1,n]$) or a boundary $\partial o_j$ ($j\in [1,m]$) in the rest of this chapter.

\begin{Assumption}\label{position}
Any two elements do not intersect. 
\end{Assumption}

Now we construct Tangent Graph \cite{savkin2013reactive}. One major component of Tangent Graph is the \textit{tangent}, defined as a straight line that is simultaneously tangent to two elements and not intersecting with any others. 
The common point of a tangent and an element is called the \textit{tangent point}. 
The curve between two tangent points on the same element is called the \textit{arc}. 
Let $G(V,E)$ be Tangent Graph, where vertex set $V$ consists of a finite set of tangent points and edge set $E$ consists of a finite set of tangents and arcs. Figure \ref{fig:graph} depicts an example of $G(V,E)$.

\begin{figure}[!h]
    \centering
    \begin{subfigure}[t]{0.5\textwidth}
        \includegraphics[width=\textwidth]{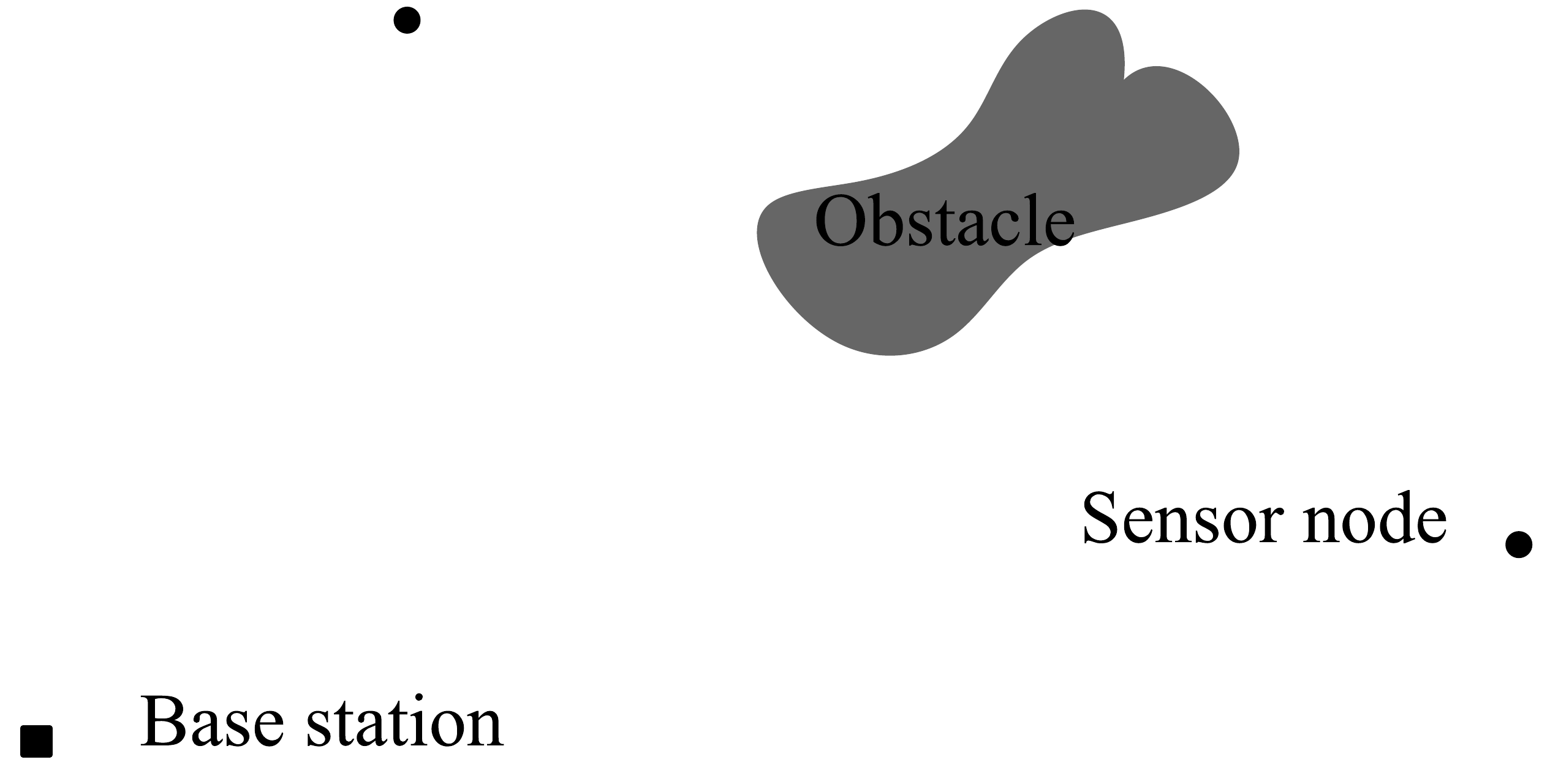}
        \caption{}
        \label{fig:puregraph}
    \end{subfigure}
    \begin{subfigure}[t]{0.22\textwidth}
        \includegraphics[width=\textwidth]{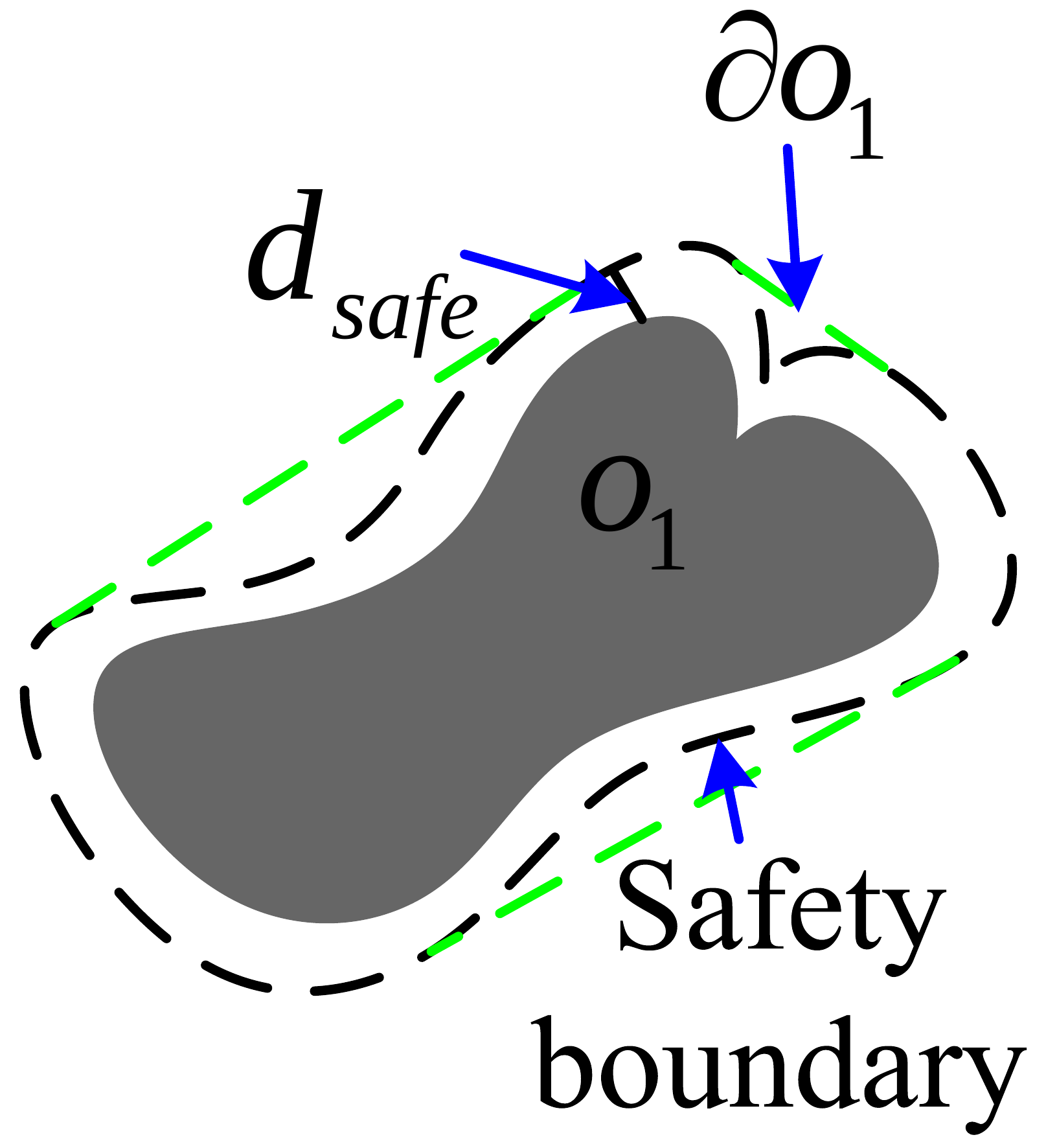}
        \caption{}
        \label{fig:convex_hull}
    \end{subfigure}
    \begin{subfigure}[t]{0.6\textwidth}
        \includegraphics[width=\textwidth]{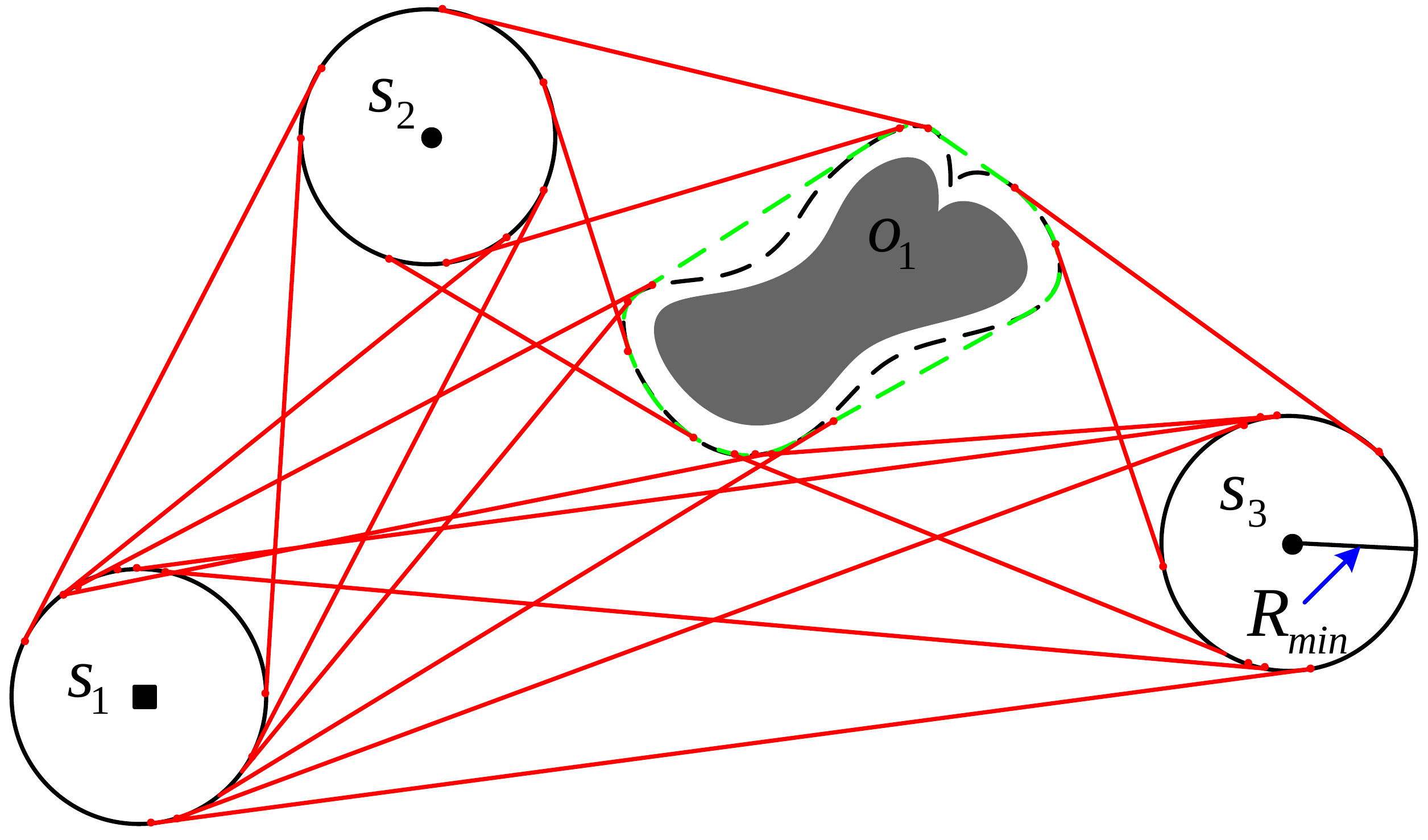}
        \caption{}
        \label{fig:tangentgraph_example}
    \end{subfigure}
    \caption{An illustrative example with one base station, two nodes and one non-convex obstacle. (a) Sensing field; (b) Convex hull construction; (c) Tangent Graph.}\label{fig:graph}
\end{figure} 

We claim that the path for the robot can be formed by a subset of $E$, since when the robot moves on $G(V,E)$, Condition 2) and 3) are satisfied definitely. However, such path is not viable. To make the path smooth, we further consider a heading constraint.

\begin{definition}\label{heading}
The heading constraint refers to that the robot's heading $\theta$ at the beginning of an edge should be equal to that at the ending of the last edge.
\end{definition}

Let $edge_1,edge_2\in E$ be two edges and they have a common tangent point $p$. Let $\theta_{edge_1}^p$ and $\theta_{edge_2}^p$ be the robot's headings at the ending of edge $edge_1$ and beginning of edge $edge_2$ respectively, i.e., at point $p$. Then, the heading constraint requires:

\begin{equation}\label{headingrequirement}
\theta_{edge_1}^p=\theta_{edge_2}^p.
\end{equation}

By carefully selecting the edges according to (\ref{headingrequirement}), the path can be smooth, i.e., Condition 1) is satisfied.

Now we consider Condition 5), which relates to the communication mode. With the visiting circle model, we consider the following assumption.

\begin{Assumption}\label{transmitrange}
The buffered data is transmitted only when the robot (\ref{dubinscar}) is on  $C_i,i\in [1,n]$.
\end{Assumption}

Notice as long as the robot is within the transmitting range of the sensor nodes, the sensor nodes can transmit data to the robot even before it arrives at the visiting circle and after it leaves. But it is practical to consider Assumption \ref{transmitrange}. One reason is that short distance communication guarantees a low data loss rate. Besides, it results in a relatively low transmission energy consumption.

With respect to Assumption \ref{transmitrange}, the node's data load and data transmission rate together determine the required contact time. However the arc edge constructed by two tangent points on a visiting circle does not guarantee such contact time. So we need to adjust the arc edge when necessary. Let $g$ be the data load of a node over a period $\cal T$ and $r$ be the data transmission rate. We consider the non-increasing staircase model \cite{ren2014data} where the data transmission rate $r(d)$ ($bit/s$) is a function of distance ($d$) between itself and the robot: 

\begin{equation}\label{transmitting}
r(d)=
\begin{cases}
r_1,0< d\leq d_1 \\
r_2,d_1< d<d_2 \\
\vdots\\
r_i,d_{i-1}< d\leq d_{i}\\
\vdots\\
0,d>d_{max}\\
\end{cases}
\end{equation}
where $d_{max}$ is the transmission range of a sensor node. Further, each data transmission rate is associated with an energy consumption rate $e(d)$ ($J/bit$), which is modelled as a non-decreasing staircase function of $d$.

For $s_i$, the required contact time $\delta_i$ can be calculated by

\begin{equation}\label{rtime}
{\delta}_i=\dfrac{g_i}{r(R_{min})}.
\end{equation}
With speed $v$, we obtain the minimum arc length $l_i$:
\begin{equation}\label{rlength}
l_i=\delta_i v=\dfrac{g_iv}{r(R_{min})}.
\end{equation}

We introduce two parameters to describe our arc edge adjustment model
\begin{equation}\label{qp}
\begin{aligned}
&a=\left\lfloor \dfrac{g_iv}{r(R_{min})2\pi R_{min}} \right\rfloor,\\
&b=\dfrac{g_iv}{r(R_{min})}-a\cdot 2\pi R_{min}.
\end{aligned}
\end{equation}

Given a visiting circle $C_i$, $i\in [1,n]$, and all the arcs on $C_i$, any arc $\wideparen{p_1p_2}$ can be adjusted by
\begin{equation}\label{p1p2}
|\wideparen{p_1p_2}|=
\begin{cases}
a \cdot 2\pi R_{min} +|\wideparen{p_1p_2}|,\ &if\ (\romannumeral 1)\\
(a+1) \cdot 2\pi R_{min} +|\wideparen{p_1p_2}|,\ &if\ (\romannumeral 2)\\
\end{cases}
\end{equation} 
where
\begin{equation*}
\begin{aligned}
(\romannumeral 1):\ &b \leq |\wideparen{p_1p_2}|\\
(\romannumeral 2):\ &b > |\wideparen{p_1p_2}|.\\
\end{aligned}
\end{equation*}

With the adjustment (\ref{p1p2}), $\wideparen{p_1p_2}$ is a valid arc edge allowing all the buffered data to be read by the robot. So we call (\ref{p1p2}) the reading adjustment. If $\wideparen{p_1p_2}$ is selected as a part of path, it means the robot arrives at $p_1$, makes $a$ or $a+1$ more round trips around $s_i$ and leaves from $p_2$. In the rest of this chapter, $G(V,E)$ represents the modified Tangent Graph, where $V$ keeps the same and the arc edges on visiting circles in $E$ are adjusted by (\ref{p1p2}). With this model, extracting a subset of $E$ satisfies Condition 5).

Finally, with respect to Condition 4) and the shortest path request, we introduce the problem to be addressed as DTSPN in \cite{IH13, VJ15}. Let $C=\{C_1,...,C_n\}$ be a set of goal regions to be visited and $\Sigma=\{\sigma_1,...,\sigma_n\}$ be an ordered permutation of $\{C_1,...,C_n\}$. Define a projection from $SE(2)$ to $\mathbb{R}^2$ as $\cal P$: $SE(2) \rightarrow \mathbb{R}^2$, i.e., ${\cal P} (X)=(x,y)$. The considered problem is an optimization over all possible permutations $\Sigma$ and configurations $X$. Stated more formally:
\begin{equation}\label{problem}
\begin{aligned}
&\min_{\Sigma,X}\  |{\cal L}(X_{\sigma_{n}},X_{\sigma_{1}})| +\sum_{i=1}^{n-1}|{\cal L}(X_{\sigma_{i}},X_{\sigma_{i+1}})|\\
& s.t.\ {\cal P}(X_{\sigma_{i})} \in \sigma_{i},\ i \in [1,n]
\end{aligned}
\end{equation}
where ${\cal L}(X_1,X_2)$ is the viable subpath (defined below) with minimum-length from configuration $X_1$ to $X_2$ (${\cal P}(X_1)\in C_{i_1}$, ${\cal P}(X_2)\in {C_{i_2}}$, $i_1,i_2\in [1,n]$ and $i_1 \neq i_2$) and $|{\cal L}(X_1,X_2)|$ gives the length. 

\begin{definition}\label{viablesubpath}
A viable subpath consists of a subset of $E$. Any two consecutive edges on the viable subpath have a common tangent point where the heading constraint (\ref{headingrequirement}) is satisfied.
\end{definition}

So far, since the function $\cal L$ makes the generated paths satisfy Condition 1), 2), 3) and 5) and formulation of problem (\ref{problem}) considers Condition 4) as well as the shortest request, the path produced by solving problem (\ref{problem}) is the shortest viable path for the data collection unicycle robot. Next section discusses how to address this problem.

\section{Shortest Viable Path Planning}\label{algorithm_basic}
Formulated as DTSPN like \cite{OOD10, IH13}, problem (\ref{problem}) is also NP-hard. But unlike DTSPN having infinite number of possible configurations in each goal region, the candidate configuration number on $C_i$, $i\in [1 ,n]$, in problem (\ref{problem}) is finite due to the Tangent Graph. So more appropriately, problem (\ref{problem}) is a sampled DTSPN. On the other hand, in the context of sensor networks, sensor nodes are usually deployed apart from each other. Then our considered problem is a standard version of sampled DTSPN, in which any two goal regions do not overlap according to Assumption \ref{position}. 

We propose an algorithm called Shortest Viable Path Planning (SVPP). The basic idea is similar to \cite{VJ15}, where the first stage is to determine the permutation $\Sigma$ based on $G(V,E)$. In the second stage, we simplify $G(V,E)$ into $G(V',E')$ and convert $G(V',E')$ to a tree-like graph $T$, and then search the shortest path in $T$. The main steps of SVPP are outlined as Algorithm \ref{wholealgorithm}. The details are given below.

\begin{algorithm}[h]
\caption{Shortest Viable Path Planning (SVPP)}\label{wholealgorithm}
\begin{algorithmic}[1]
\State Compute $\Sigma$ by solving ATSP instance based on $G(V,E)$.
\State Compute $\Sigma '$ by adding blocking safety boundaries to $\Sigma$.
\State Simplify $G(V,E)$ to $G(V',E')$ by remaining the edges and vertices related to $\Sigma'$ and deleting others.
\State Convert $G(V',E')$ to tree-like graph $T$.
\State Given an initial configuration, search the shortest path $P$ in $T$ .
\end{algorithmic}
\end{algorithm}

In Step 1, to compute $\Sigma$, we construct a directed graph. The reason for a directed graph instead of undirected graph lies in Condition 5) of Definition \ref{viable}, i.e., different sensor nodes may require different contacting time. We take $s_i$, $i\in [1,n]$ as the vertex set of the directed graph. We construct the edge set as follows. The length of the edge between two vertices takes into account two aspects: the length of the valid path between their visiting circles and the length of the adjusted arc on the latter vertex. Note, since there may be multiple valid paths between two visiting circles, see $s_1$ and $s_2$ in Figure \ref{fig:simplifiedgraph} for an example, here we use the average length of them. Thus, the length of the edge from $s_1$ to $s_2$ equals to the summation of the average length and the length of the adjusted arc on the visiting circle of $s_2$. In contrast, the length of the edge from $s_2$ to $s_1$ equals to the summation of the average length and the length of the adjusted arc on the visiting circle of $s_1$. With such directed graph, we use a ATSP solver \cite{FGM82} to calculate the permutation $\Sigma$.

Having $\Sigma$, $G(V,E)$ can be simplified by remaining the tangent edges connecting two successful visiting circles in $\Sigma$ and the corresponding arc edges. When any obstacle blocks any pair of visiting circles, the edges passing the obstacle safety boundaries are also remained. By doing this, we can get the blocking number $K\geq 0$ which counts the occurrence number of blocking two successful visiting circles in $\Sigma$. Since one obstacle can block more than one pair of successful visiting circles, the number of the blocking obstacles in $G'(V',E')$ may be smaller than $K$. 
We insert these boundaries to the proper positions in $\Sigma$ and get an extended permutation $\Sigma '=\{\sigma_1 ',...,\sigma'_{n'} \}$, whose length is $n'=n+K$. Obviously, $\Sigma \subseteq \Sigma '$. We call the new graph the Simplified Tangent Graph $G'(V',E')$, where $V'\subseteq V$ and $E'\subseteq E$. Given $\Sigma '$, the problem (\ref{problem}) can be reformulated as:
\begin{equation}\label{problem1}
\begin{aligned}
&\min_{X}\  |{\cal L'}(X_{\sigma'_{n'}},X_{\sigma_{1}'})| +\sum_{i=1}^{n'-1}|{\cal L'}(X_{\sigma'_{i}},X_{\sigma'_{i+1}})|\\
& s.t.\ {\cal P}(X_{\sigma'_{i})} \in \sigma'_{i},\ i \in [1,n']
\end{aligned}
\end{equation}
where ${\cal L'}={\cal L}$. 

\begin{figure}[t]
\begin{center}
{\includegraphics[width=0.5\textwidth]{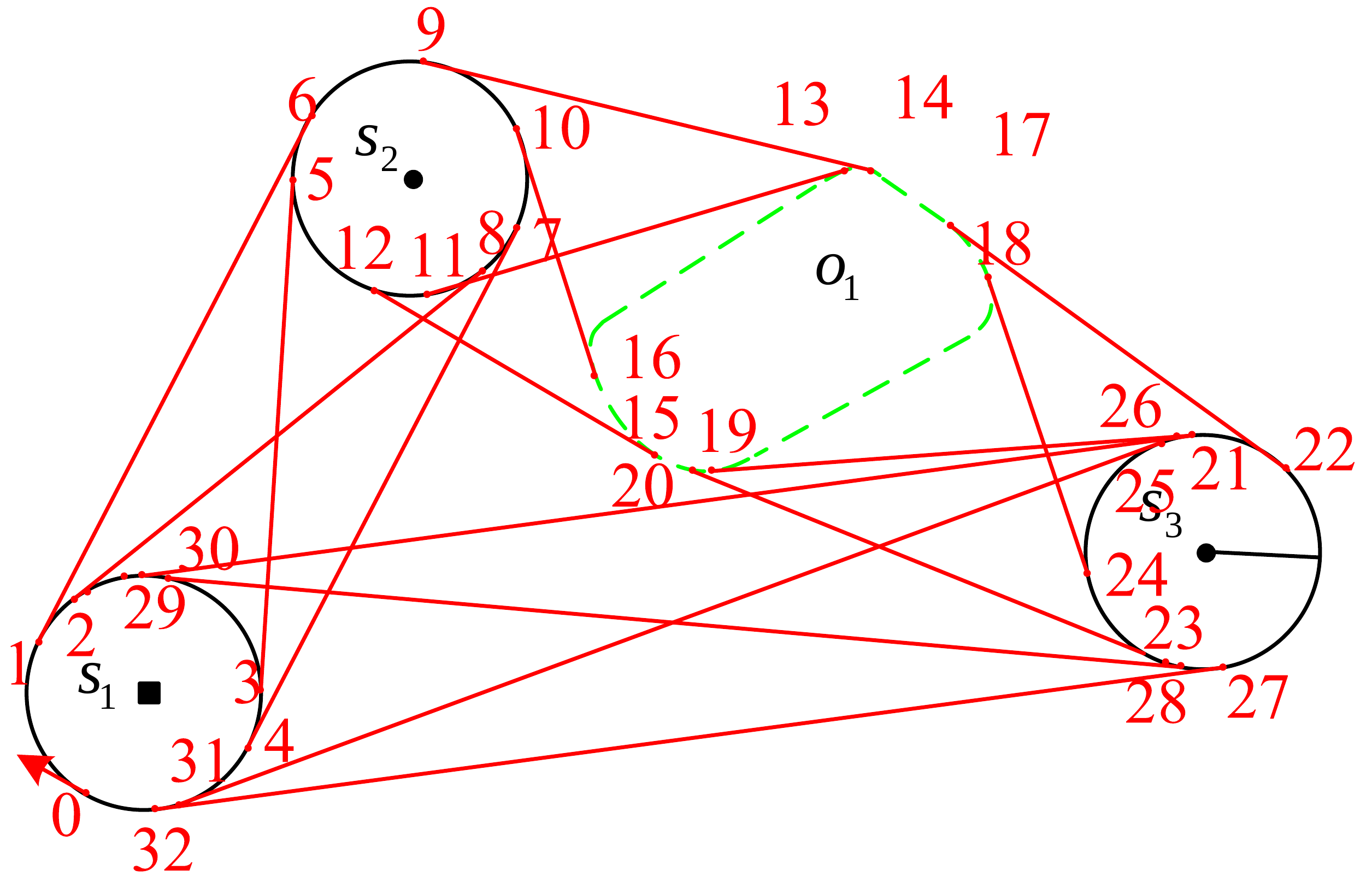}}
\caption{Simplified Tangent Graph $G(V',E')$ of $G(V,E)$ in Figure \ref{fig:tangentgraph_example}. Since $\Sigma=\{C_1,C_2,C_3\}$, the edges and vertices between $C_1$ and $\partial o_1$ are removed while the others are remained. Since $\partial o_1$ blocks $C_2$ and $C_3$, $\partial o_1$ is inserted into $\Sigma$, then we obtain $\Sigma '=\{C_1,C_2,\partial o_1,C_3\}$. The numbers near the vertices are the labels of robot configurations. For example, Label 0 represents $X_0$.}
\label{fig:simplifiedgraph}
\end{center}
\end{figure}

It is worth mentioning that if we name $X_1$ ($X_2$) as the arrival configuration on $C_{i_1}$ ($C_{i_2}$), there is a departure configuration, say $X_3$, on $C_{i_1}$ such that the heading constraint (\ref{headingrequirement}) is satisfied. Then ${\cal L'}(X_1,X_2)$ always consists of two edges: one arc edge $\wideparen{{\cal P}(X_1){\cal P}(X_3)}$ and one tangent edge $\overline{{\cal P}(X_3){\cal P}(X_2)}$. With such characteristic of ${\cal L'}$, the objective of problem (\ref{problem1}) can be stated to find the minimum-length path by selecting two configurations from each element in $G(V',E')$, subject to the heading constraint (\ref{headingrequirement}). 

Before we describe how to select such configurations, we analyse the number of possible paths in $G(V',E')$. We generate tangents between any pair of elements to construct $G(V,E)$. Our algorithm to construct tangents is given in \ref{app}. As the robots only need to move along the boundary of convex hull, we have 4 common tangents and 8 tangent points for each pair of successful elements in $G(V',E')$. Then one element has 8 configurations, 4 of which are arrival configurations and the other 4 are departure configurations. For example, in Figure \ref{fig:simplifiedgraph}, $X_5,X_6,X_7,X_8$ are the 4 arrival configurations of $C_2$ and $X_9,X_{10},X_{11},X_{12}$ are the 4 departure configurations. Also, we notice that from one arrival configuration of one element, we have 2 options to reach the arrival configurations of the next element considering the heading constraint (\ref{headingrequirement}). Then, from a given $X_0$ we have $2^{n'}$ paths to reach $\sigma'_{n'}$. Finally, the robot needs to return to $\sigma'_{1}$ from $\sigma'_{n'}$. Again due to the heading constraint (\ref{headingrequirement}), half of the arrival configurations on $\sigma'_{n'}$ cannot reach $X_0$, such as $X_{26}$ and $X_{28}$ in Figure \ref{fig:simplifiedgraph}. Therefore, the total number of paths starting and ending at $X_0$ is $2^{n'-1}$.

To better demonstrate the viable paths in $G(V',E')$, we convert it to a tree-like graph. Given $X_0$, we cut $\sigma_1'$ into two parts: one part contains the departure configurations and $X_0$, and the other contains the arrival configurations and $X_0$, then the circle-like graph $G(V',E')$ turns to be a tree-like graph, which is called tree-like graph $T$. Since $n'$ elements are on $G(V',E')$ and $\sigma_1'$ is divided into two parts, $T$ consists of $n'+1$ layers, where $T_i=\sigma_i'$ ($i\in [1,n']$) and $T_{n'+1}=\sigma_1'$. An example of $T$ is shown in Figure \ref{fig:treegraph}.

\begin{figure}[t]
\begin{center}
{\includegraphics[width=0.45\textwidth]{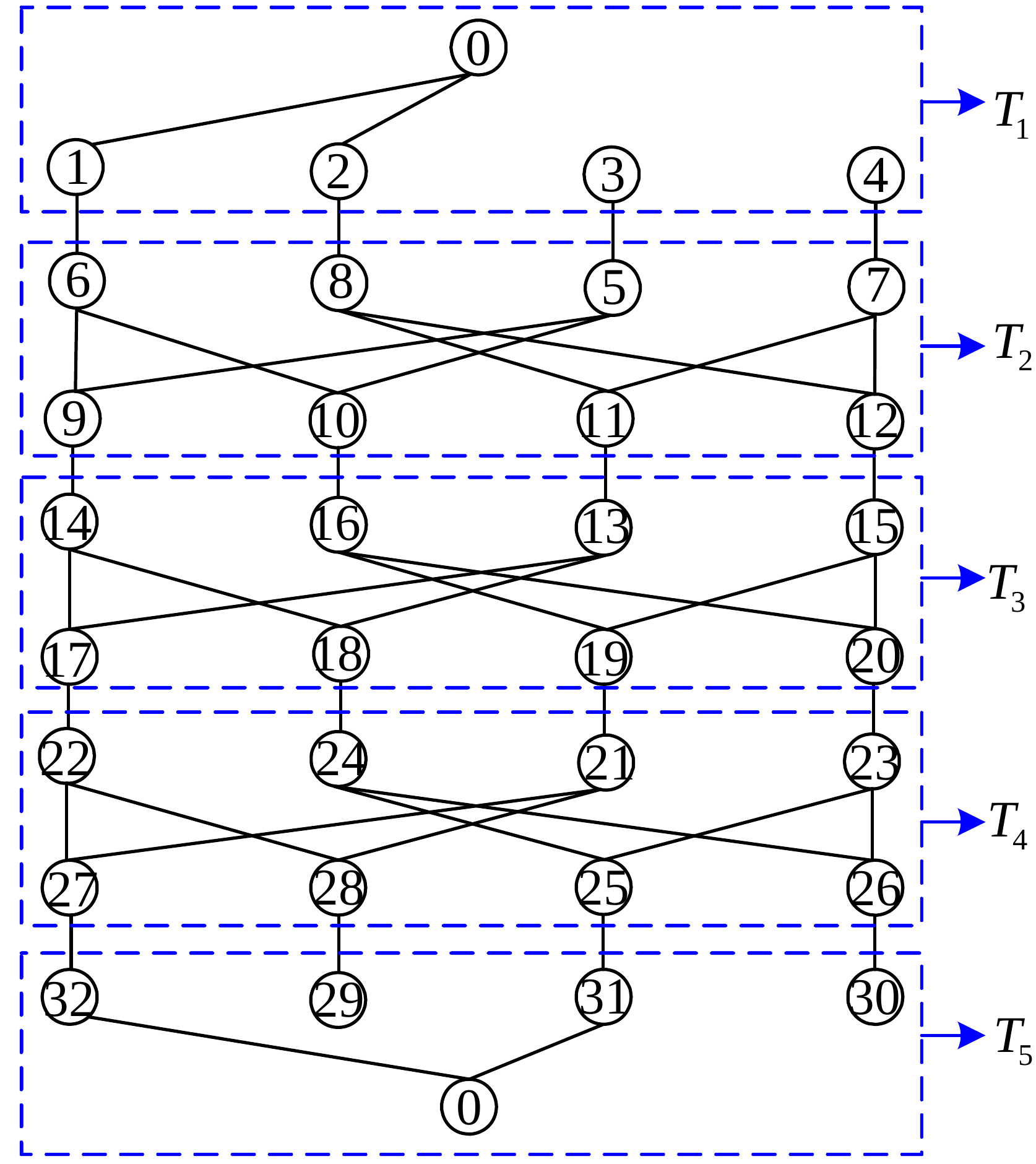}}
\caption{Tree-like graph of the example shown in Figure \ref{fig:simplifiedgraph}. There are 8 viable paths starting at $X_0$ and ending at $X_0$.}
\label{fig:treegraph}
\end{center}
\end{figure}

Evolved from $G(V',E')$, $T$ inherits the feature that given a departure configuration, moving from one layer to the next one has only one path. Define $z_i$ ($i\in [1,n'+1]$) as an arrival configuration variable of layer $T_i$. Define $L_{i,i+1}$ as the path length from an arrival configuration $z_i$ on $T_i$ to an arrival configuration $z_{i+1}$ on $T_{i+1}$, then we have

\begin{equation}
L_{i,i+1}= |{\cal L'}(z_{i},z_{i+1})|,1\leq i \leq n'.
\end{equation}

Suppose $L_{1,i+1}^*$ is the shortest path length from an arrival configuration on ${T}_{1}$ to an arrival configuration on ${T}_{i+1}$, $1\leq i \leq n'$. Since the initial value of $z_1$ is prescribed, i.e., $z_1=X_0$, it follows that $L_{1,i+1}^*$ is only a function of the viable $z_{i+1}$. Then
\begin{equation}\label{dynamicformulation}
L^*_{1,i+1}=\min\ \{L^*_{1,i}+L_{i,i+1}\},1\leq i \leq n'.
\end{equation}
where $L_{1,1}^*=0$.

Let $z_i^*$, $i\in [1,n'+1]$, denote the optimal configurations to problem (\ref{dynamicformulation}). We use a dynamic programming based method, given by Algorithm \ref{algorithm}, to solve problem (\ref{dynamicformulation}).

\begin{algorithm}[t]
\caption{Shortest Path Search}\label{algorithm}
\begin{algorithmic}[1]
\Require $T$, $\cal L'$
\Ensure $P$
\For{$i=2, i\leq n',i=i+1$}
\State For each feasible $z_i \in T_i$, choose feasible ${z}_{i-1} \in {T}_{i-1}$ to minimize $\{L^*_{1,i}+L_{i,i+1}\}$. Then $L^*_{1,i+1}$ is expressed as a function of $z_{i}$ only.
\EndFor
\State Use the destination configuration ${z}_1$ to select configurations $z_{n'+1}^*$, $z_{n'}^*$,..., $z_{2}^*$. $z_{1}^*=z_1$.
\State $P=\{{\cal L'}(z_{1}^*,z_{2}^*),...,{\cal L'}(z_{n'}^*,z_{n'+1}^*), \wideparen{{\cal P}(z_{n'+1}^*){\cal P}(z_{1}^*)}\}$.
\end{algorithmic}
\end{algorithm}

To end this section, we analyse the time complexity of SVPP. In Step 1, the complexity of the ATSP algorithm \cite{FGM82} is $O(n^3)$.  In Step 2, we check $n$ pairs of visiting circles to see whether they are blocked by any boundary of convex hull. In each checking, the worst case is to check all the $m$ obstacles and then the time complexity is $O(mn)$. In Step 3, we do a constant number of operations to each element in $\Sigma'$, then the time complexity of the simplifying procedure is $O(n+K)$. Converting $G(V',E')$ to $T$ costs $O(1)$ in Step 4. Finally, because the maximum number of arrival configurations is 4 on each layer of $T$, searching the shortest path in $T$ costs $O(n+K)$. Therefore, the overall time complexity of SVPP is $O(n^3)$. 

\section{$k$-Shortest Viable Path Planning}\label{algorithm_extend}
In Section \ref{algorithm_basic}, we have only considered the situation of using single robot. Applying SVPP to a large scale network may result in a path with unacceptable collection time. One possible solution to deal with the long collection time is to employ multiple robots. Given the sensor network, increasing the number of robots can decrease the average path length of the robots. In this scenario, we need to design a viable path for each robot. Given $k$ identical robots whose initial positions are at ${\cal P}(X_0)\in C_1$, we aim at finding $k$ viable paths $\{P_1,...,P_k\}$ such that:
\begin{enumerate}
\item
$P_l$, $l\in [1,k]$ starts from ${\cal P}(X_0)$ and ends at $C_1$;
\item
for each $i\in[2,n]$, $\exists l\in[1,k]$ such that a) $C_i \cap P_l \neq \emptyset$ and b) $C_i \cap P_j = \emptyset$, $j\in [1,k] \backslash \{l\}$, i.e., each visiting circle is visited by only one viable path;
\item
$\max_{l\in[1,k]}\{|P_l|\}$ is minimized, where the function $len(P_l)$ gives the length of $P_l$.
\end{enumerate}

Here the objective is to minimize the length of the longest path, denoted as $k$-length. This objective can make the paths have similar lengths. Since the $k$ robots move at the same speed, the collection times of these paths are more or less equal. It it easy to understand that this problem is a NP-hard problem. If the non-holonomic constraint is removed, $C_i$ is reduced to $s_i$. Then it is reduced to $k$-TSPN. If $k=1$, it is further reduced to TSPN, which is known as NP-hard. So this problem is also NP-hard. 

Three relevant approaches have been proposed for the case of multiple robots. The first is called Vehicle Routing Problem (VRP) \cite{DR59}. Since the VRP framework targets on minimization of the total length of all the paths, it may lead to the situation where the lengths of the produced paths are quite different from each other. The second one is the cluster-based approach which first divides the network into several clusters and then finds optimal path in each cluster. The cluster-based approach (e.g., K-means \cite{visitingcircle2015, TO05}) is known for effectively decreasing the problem scale, but the path lengths are still not guaranteed to be similar. The third one is the split-based approach (e.g., $k$-SPLITOUR \cite{FHK78}) which is first to construct a whole path for the single robot and then divided it into several parts by carefully selecting $k-1$ split positions. It is noted for generating paths of more or less equal lengths. However, we point the paths generated by $k$-SPLITOUR may not be guaranteed to be optimal since it simply connects the initial position to the selected split positions. 

Based on these considerations, we propose an algorithm, called $k$-Shortest Viable Path Planning ($k$-SVPP). $k$-SVPP is to make use of the profits of the cluster-based and split-based approaches to compensate their defect. Particularly, we use $k$-SPLITOUR to guide clustering and find optimal path in each cluster. The $k$-SVPP algorithm is given by Algorithm \ref{ksplit} and its main steps include 1) run SVPP to get a whole path; 2) use $k$-SPLITOUR to split the whole path into $k$ parts and get $k$ clusters; and 3) run SVPP to reconstruct the $k$ paths.  

\begin{algorithm}[t]
\caption{$k$-Shortest Viable Path Planning}\label{ksplit}
\begin{algorithmic}[1]
\State Run SVPP to get path $P$ with length $L=|P|$ and permutation $\Sigma=\{\sigma_1,...,\sigma_n\}$.
\State For each $l$, $l\in [1,k-1]$, find the last visiting circle $\sigma_{i(l)}$ along $P$ such that $len(\sigma_1,\sigma_{i(l)})\leq \frac{l}{k}(L-2L_{max})+L_{max}$. Construct $k$ clusters as follows: 

$Cluster_1=\{\sigma_1,\sigma_2,...,\sigma_{i(1)}\}$,

$Cluster_l=\{\sigma_1,\sigma_{i(l-1)+1},...,\sigma_{i(l)}\}$, $l\in [2,k-1]$, 

$Cluster_k=\{\sigma_1,\sigma_{i(k-1)+1},...,\sigma_{n}\}$.
\State Run SVPP to compute $P_l$ on $Cluster_l$, $l\in [1,k]$. 
\end{algorithmic}
\end{algorithm}

In Algorithm \ref{ksplit}, $L_{max}=\max_{i\in[1,n]}|{\cal L}(X_{\sigma_1},X_{\sigma_i})|$, i.e., the largest length between the configurations on $\sigma_1$ and the other visiting circles. Function $len()$ calculates the path length between two given configurations. $\cup_{l=1}^k Cluster_l=\{C_1,...,C_n\}$ and $Cluster_{l_1} \cap Cluster_{l_2}=C_1$, $l_1,l_2\in [1,k]$, $l_1\neq l_2$. Note the structure of $k$-SVPP is the same with $k$-SPLITOUR. The main differences between them are as follows. First, $k$-SPLITOUR is originally designed for $k$-TSP while in $k$-SVPP all the path segments have taken into account the concept of viable path. We can replace all the paths generated by $k$-SPLITOUR with viable paths and we call the corresponding method viable $k$-SPLITOUR. Second, in $k$-SPLITOUR the second (final) step is to construct $k$ permutations, i.e.,
\begin{equation*}
\begin{aligned}
&\Sigma_1=\{\sigma_1,\sigma_2,...,\sigma_{i(1)}\},\\
&\Sigma_l=\{\sigma_1,\sigma_{i(l-1)+1},...,\sigma_{i(l)}\}, l\in [2,k-1], \\
&\Sigma_k=\{\sigma_1,\sigma_{i(k-1)+1},...,\sigma_{n}\}.
\end{aligned}
\end{equation*}
while in $k$-SVPP we conduct a further operation (Step 3), i.e., reconstructing the path in each cluster. Suppose the lengths of the paths generated by viable $k$-SPLITOUR and $k$-SVPP are represented by $Length_l$ and $Length_l^*$, $l\in [1,k]$ respectively. We claim that $Length_l\geq Length_l^*,\forall l\in[1,k]$. Then we have $\max_l\{Length_l\}\geq \max_l\{Length_l^*\}$, i.e., the $k$-length of paths by $k$-SVPP is no greater than that of viable $k$-SPLITOUR. Behind this, the cost is to run SVPP for another $k$ times to get the $k$ reconstructed paths in $k$-SVPP while it is straightforward to construct paths based on the $k$ permutations in viable $k$-SPLITOUR.

\section{Simulation results}
\label{simulation_path}
This section is divided into three parts. The first part demonstrates the performance on some instances and investigate the influence of different factors on SVPP. The second part displays the performance of $k$-SVPP, and the third part provides the comparison with multihop communication.
\subsection{Performance of SVPP}\label{svpp}
We simulate a 200$m$ $\times$ 200$m$ virtual field with a set of disjoint obstacles. In this field, $n$ (10-50) sensor nodes are randomly deployed outside the obstacles. According to Assumption \ref{position}, 1) any two obstacles are at least $2d_{safe}$ from each other; 2) any two sensor nodes are at least $2R_{min}$ from each other; and 3) any sensor node and obstacle are $R_{min}+d_{safe}$ from each other. We further assume that the sensor network is connected, i.e., any node has at least one neighbour node within $d_{max}$ distance. This enables the nodes to transmit the data to the base via multihop communication. The sensor nodes are event driven and the data load $g$ is between 0 and 1$MB$. We consider $k$ identical ground robots, whose maximum angular velocity $u_M$ is set as $1rad/s$. The speed $v$ is between 1 and $6m/s$. Thus the corresponding minimum turning radius $R_{min}$ is between 1 and $6m$. 

This section displays the performance of SVPP. Since the base station is treated in the same way as sensor nodes, we present them using the same mark in the following figures. An network instance with 10 nodes is shown in Figure \ref{fig:node10_f} together with its Tangent Graph, Simplified Tangent Graph. The speed of the robot is taken as 6$m/s$. Applying the proposed algorithm SVPP, we get the shortest viable path. Here we demonstrate two shortest viable paths given two different initial headings. We also conduct simulations for the networks with 20, 30, 40 and 50 sensor nodes shown in Figure \ref{fig:node20_50}, where $v=3m/s$. In these simulations, we set a small value to $g$ such that the data can be loaded shortly.

\begin{figure}[t]
    \centering
    \begin{subfigure}[t]{0.45\textwidth}
        \includegraphics[width=\textwidth]{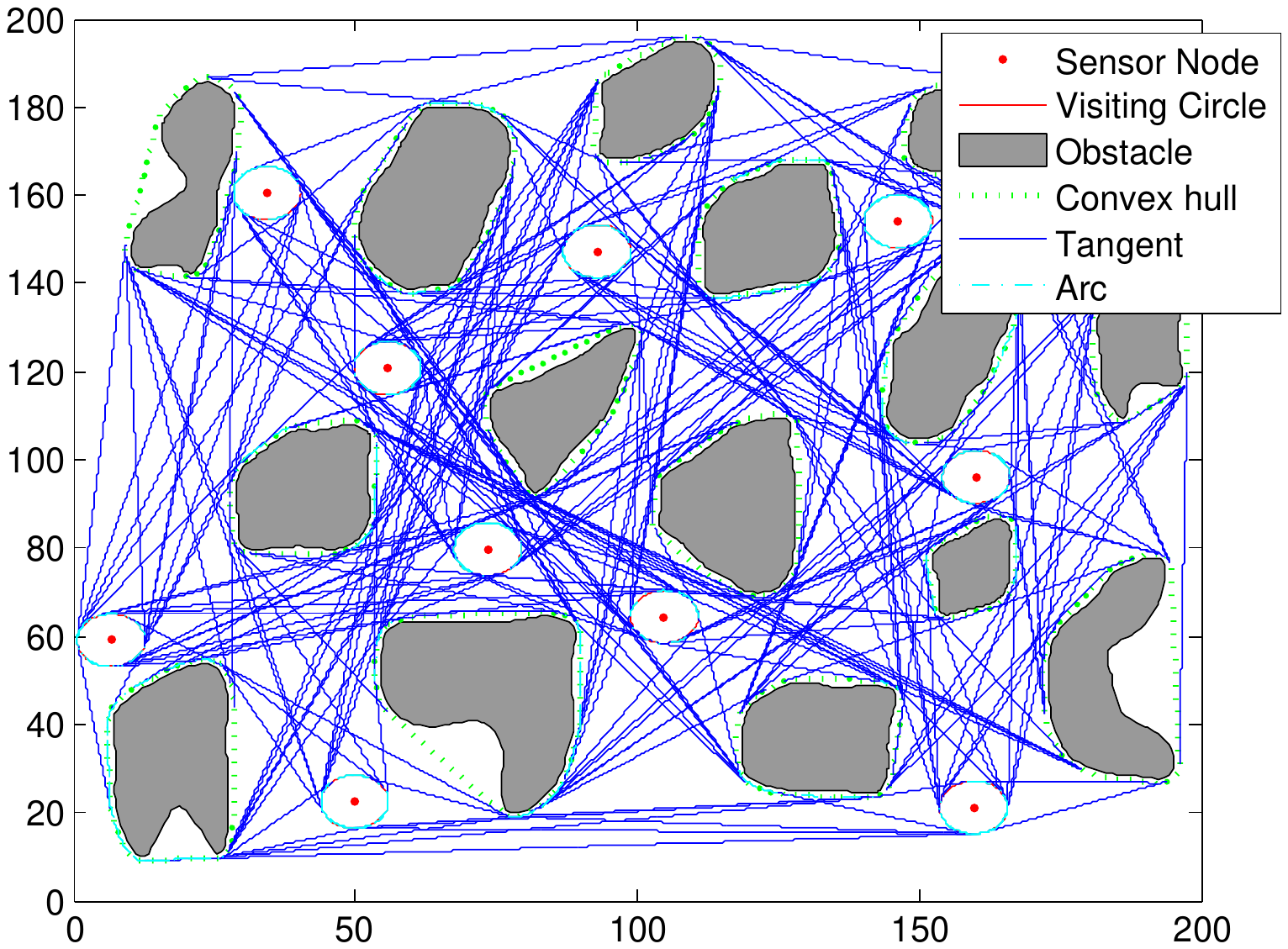}
        \caption{}
        \label{fig:tangentgraph}
    \end{subfigure}
    ~
    \begin{subfigure}[t]{0.45\textwidth}
        \includegraphics[width=\textwidth]{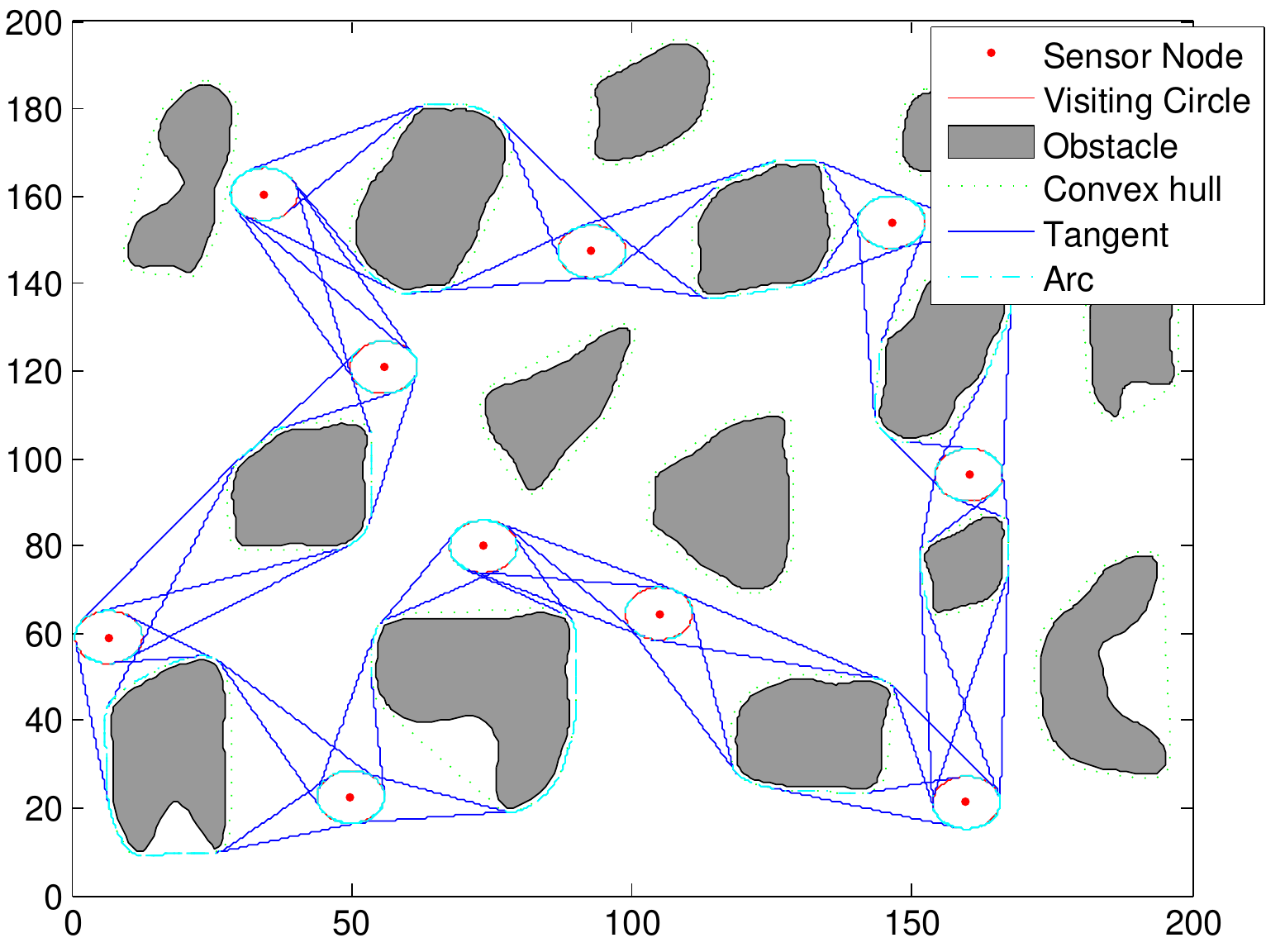}
        \caption{}
        \label{fig:simplified_tangent_graph}
    \end{subfigure}
    \\
    \begin{subfigure}[t]{0.45\textwidth}
        \includegraphics[width=\textwidth]{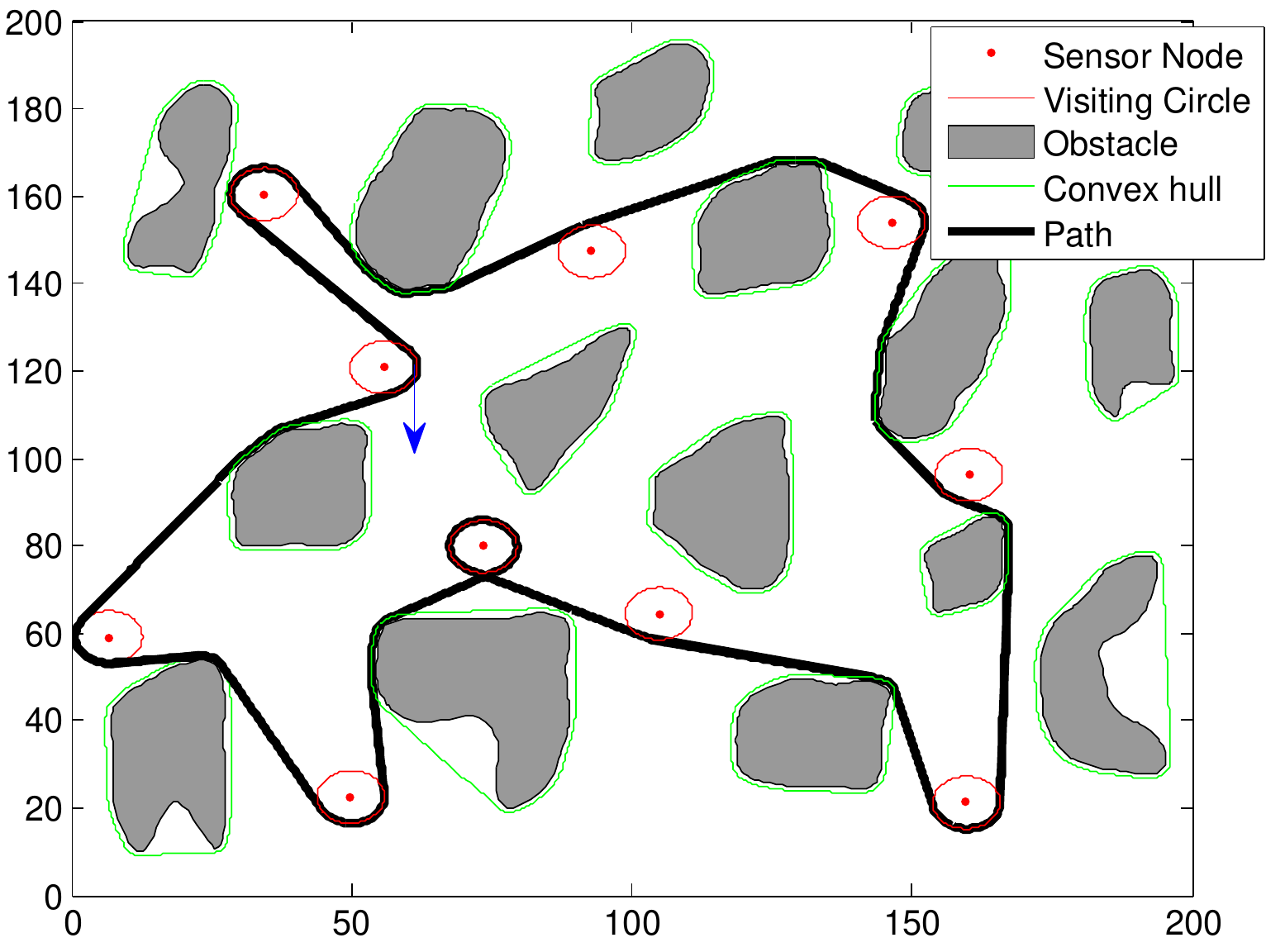}
        \caption{}
        \label{fig:node10_min}
    \end{subfigure}
    ~
    \begin{subfigure}[t]{0.45\textwidth}
        \includegraphics[width=\textwidth]{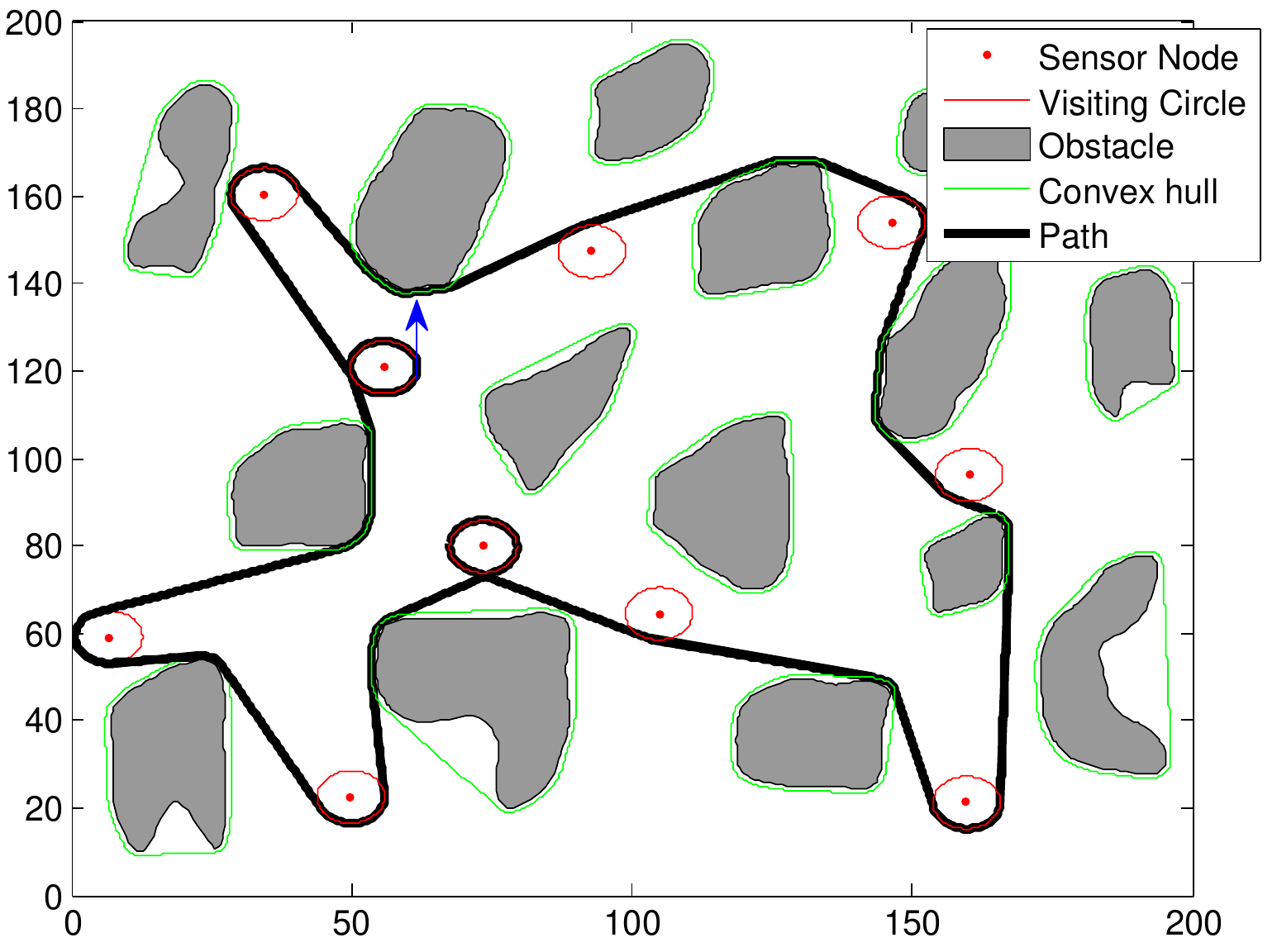}
        \caption{}
        \label{fig:node10_max}
    \end{subfigure}
    \caption{A demonstrative example with $n=10$: (a) Tangent Graph; (b) Simplified Tangent Graph; (c) Given a south facing initial heading, the shortest viable path length is 749.1$m$; (d) Given a north facing initial heading, the shortest viable path length is 786.7$m$}\label{fig:node10_f}
\end{figure}

\begin{figure}[t]
    \centering
    \begin{subfigure}[t]{0.45\textwidth}
        \includegraphics[width=\textwidth]{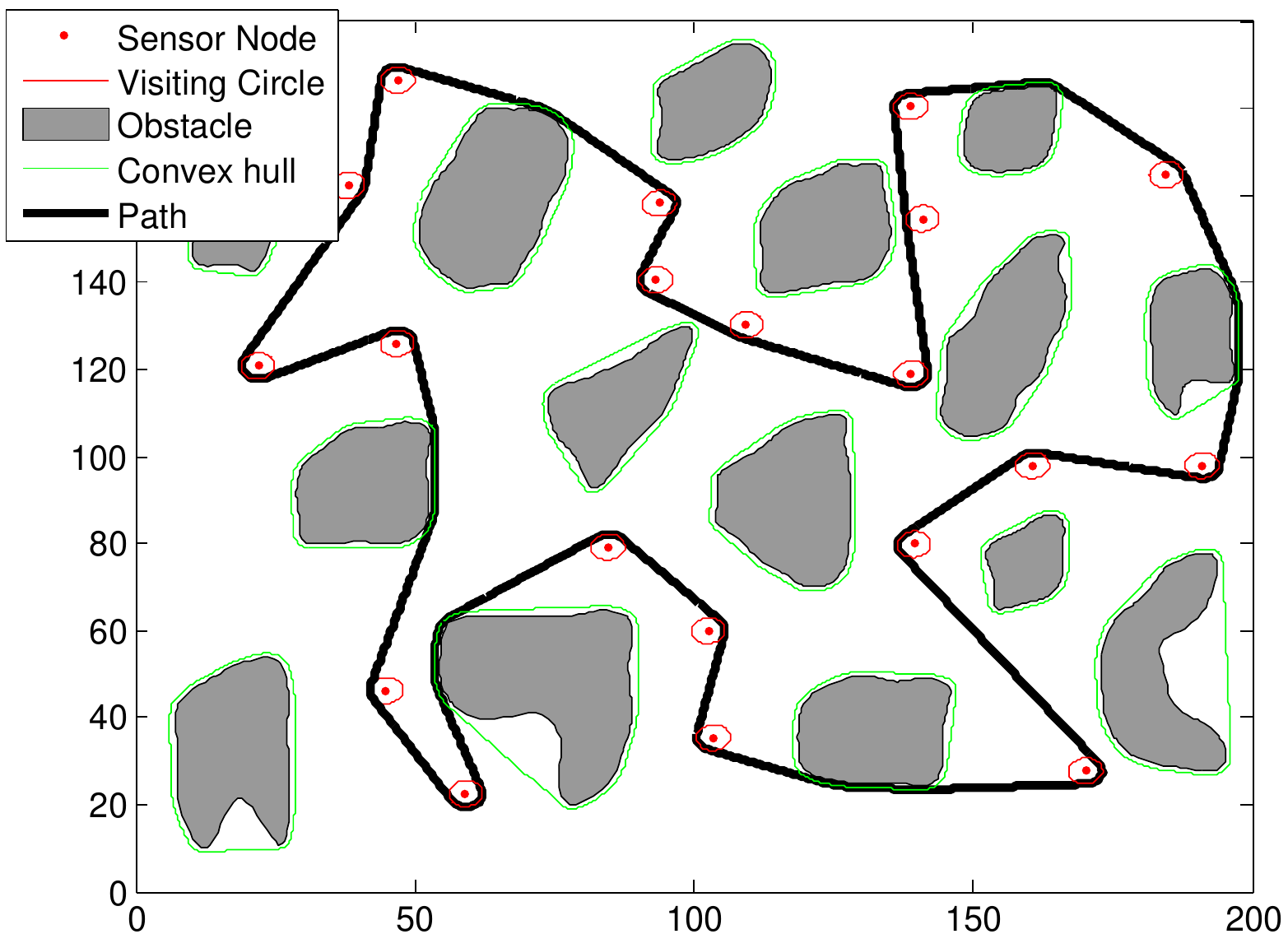}
        \caption{}
        \label{fig:node20}
    \end{subfigure}
    ~
    \begin{subfigure}[t]{0.45\textwidth}
        \includegraphics[width=\textwidth]{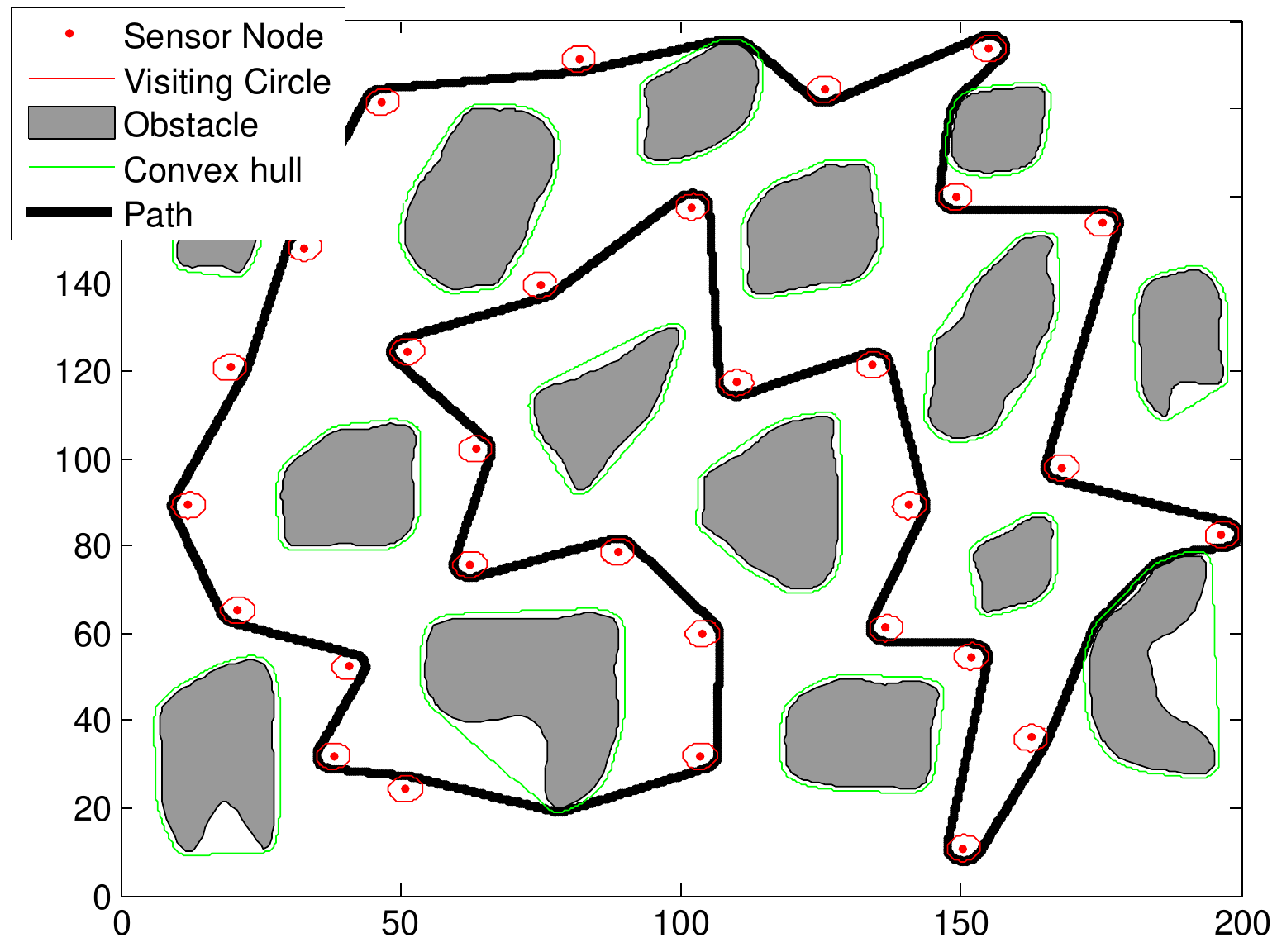}
        \caption{}
        \label{fig:node30}
    \end{subfigure}
    \\
    \begin{subfigure}[t]{0.45\textwidth}
        \includegraphics[width=\textwidth]{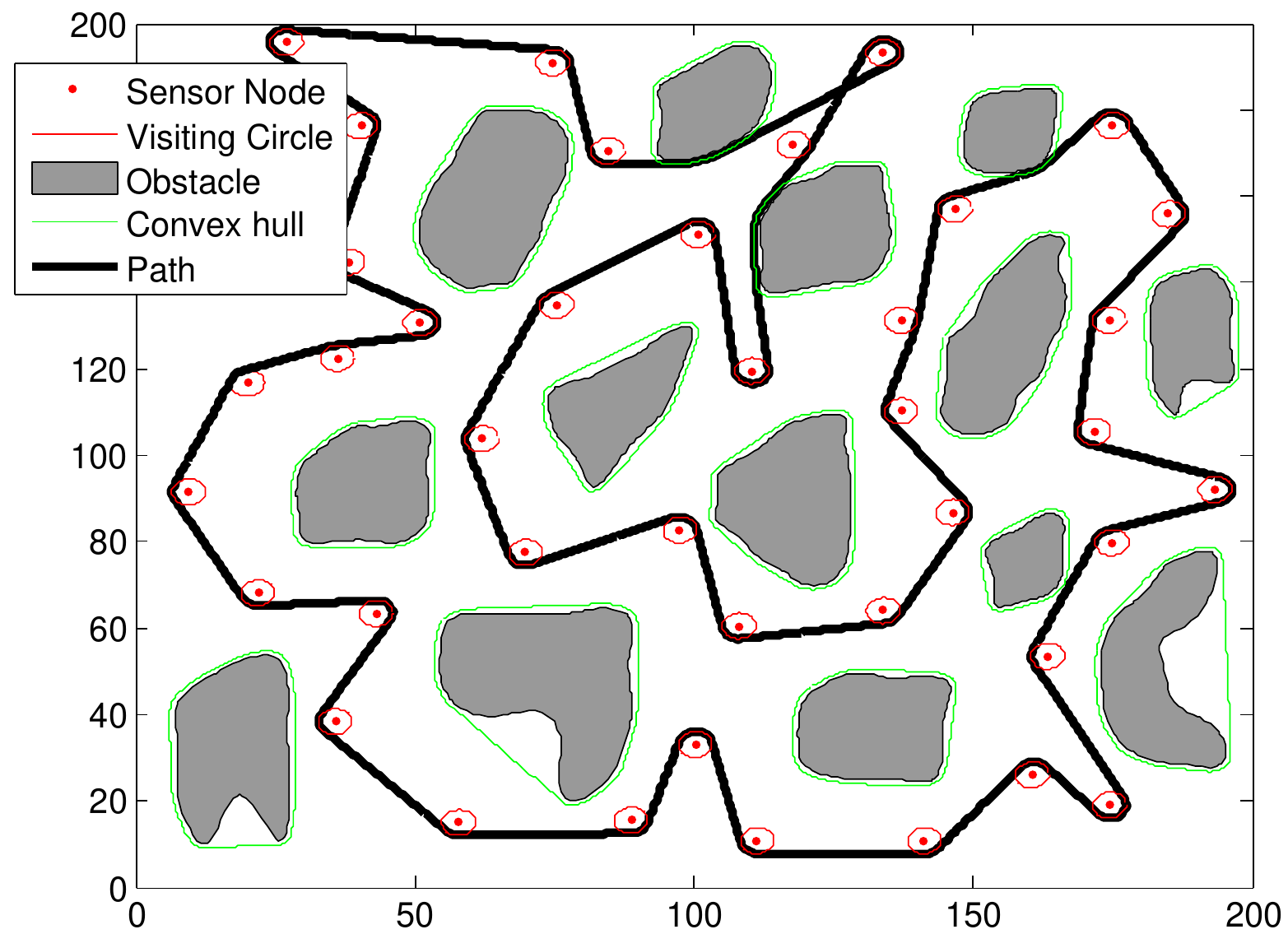}
        \caption{}
        \label{fig:node40}
    \end{subfigure}
    ~
    \begin{subfigure}[t]{0.45\textwidth}
        \includegraphics[width=\textwidth]{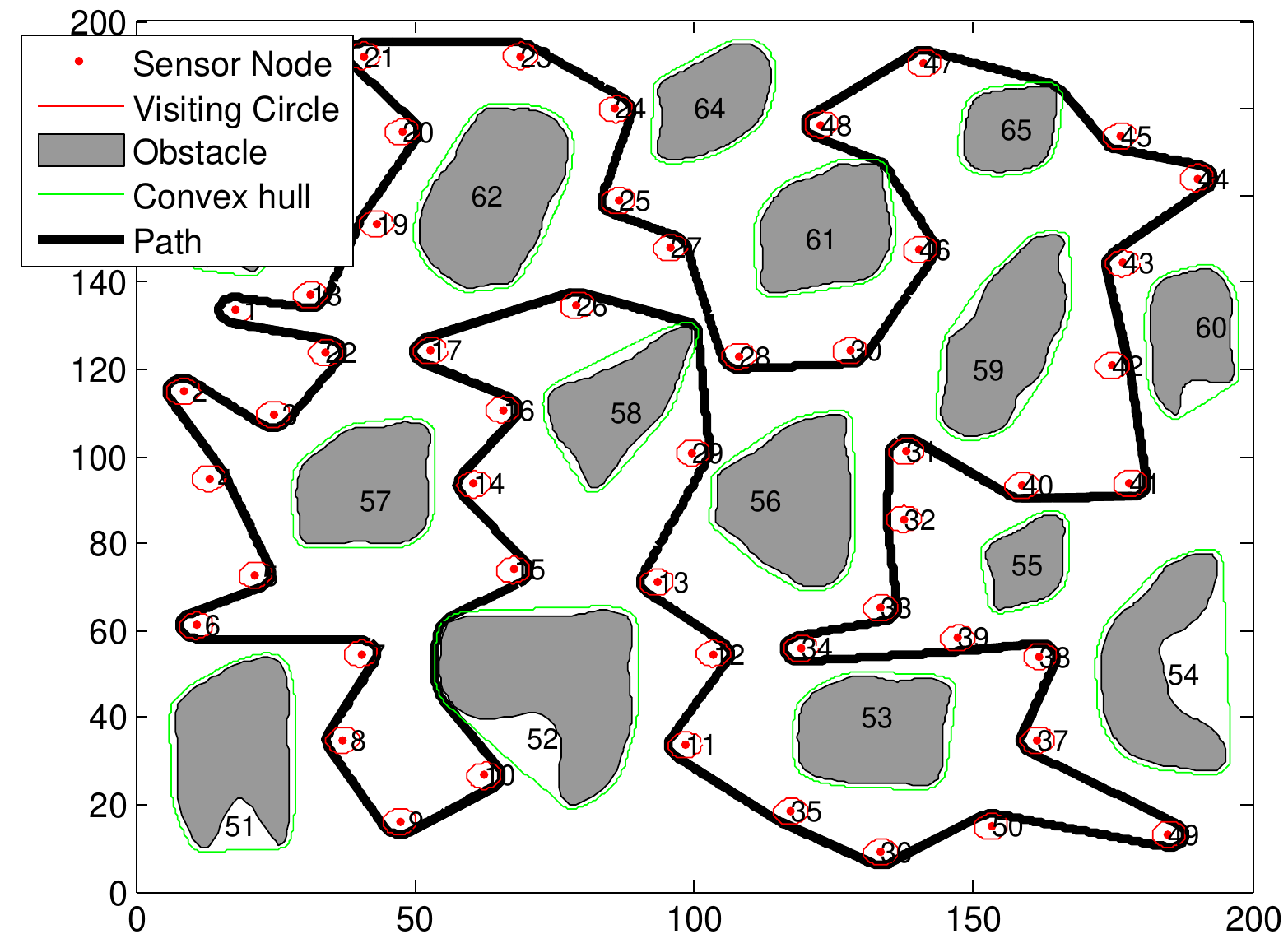}
        \caption{}
        \label{fig:node50}
    \end{subfigure}
    \caption{The viable paths for instances with different scale: (a) $n=20$ and the path length is 907.0$m$; (b) $n=30$ and the path length is 1078.4$m$; (c) $n=40$ and the path length is 1275.9$m$; (d) $n=50$ and the path length is 1352.4$m$.}\label{fig:node20_50}
\end{figure}

Next we investigate the influences of two key parameters, i.e., robot speed $v$ and data load $g$ on two system metrics: path length and collection time (path length/$v$). Here $v$= 1, 2, 3, 4, 5 and 6$m/s$ and $n=$10, 20, 30, 40, and 50. For each pair of $v$ and $n$, we simulate 20 independent instances and the results are shown in Figure \ref{fig:metrics1}. For a fixed $n$, the path length tends to increase with the increasing of $v$, due to the increasing of $R_{min}$. Extremely, if $v=0$, it turns to the case of TSP. We do not display the TSP paths as they are not viable. In Figure \ref{fig:timedelay1}, the collection time  decreases with the increasing of $v$. From these simulations we can see that although increasing the robot speed raises the path length, the collection time can be reduced. Note the data load $g$ still takes a small value. When large data load is accounted, such conclusion may not be appropriate.

\begin{figure}[t]
    \centering
    \begin{subfigure}[t]{0.45\textwidth}
        \includegraphics[width=\textwidth]{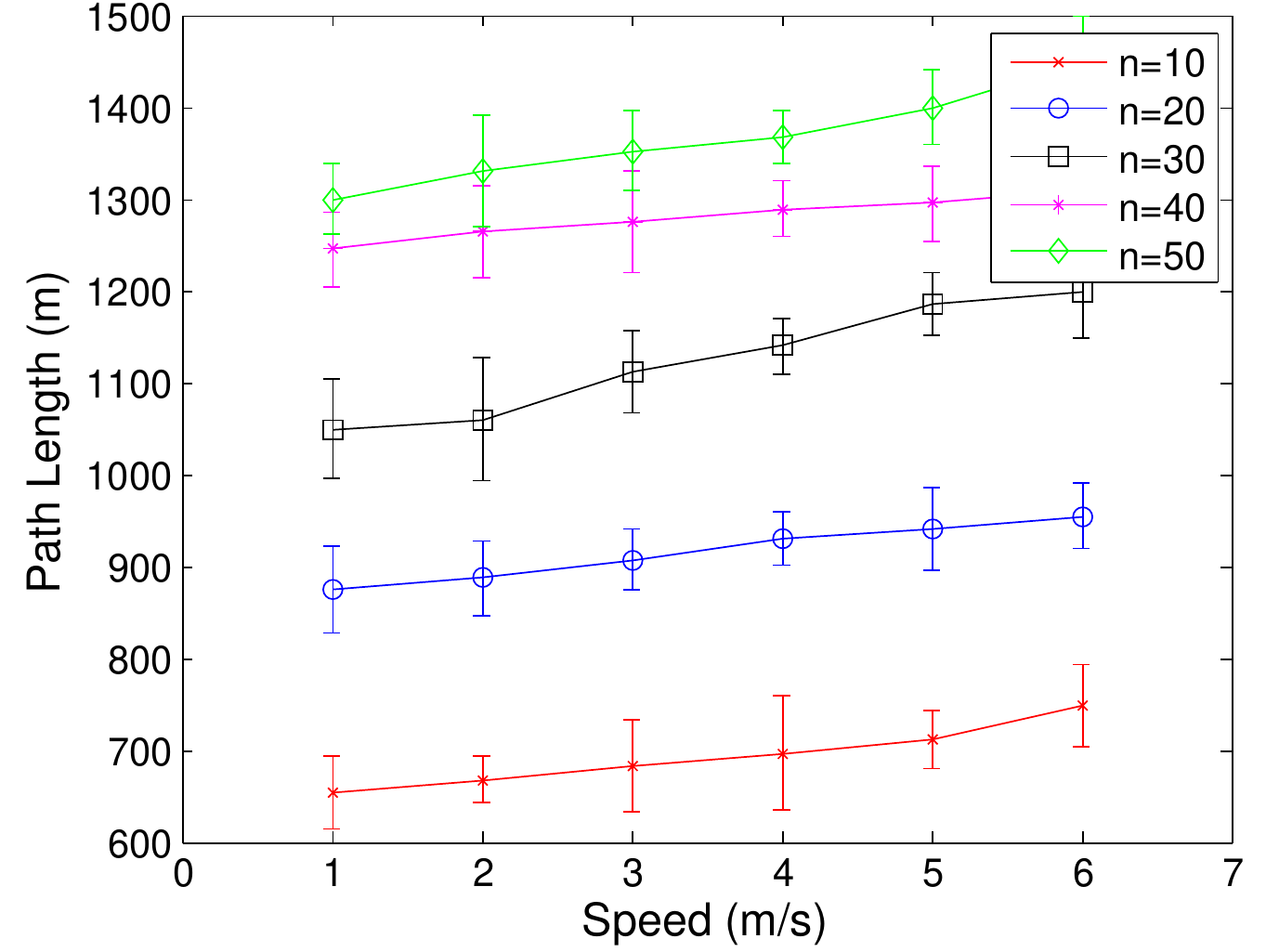}
        \caption{}
        \label{fig:pathlength1}
    \end{subfigure}
    \begin{subfigure}[t]{0.45\textwidth}
        \includegraphics[width=\textwidth]{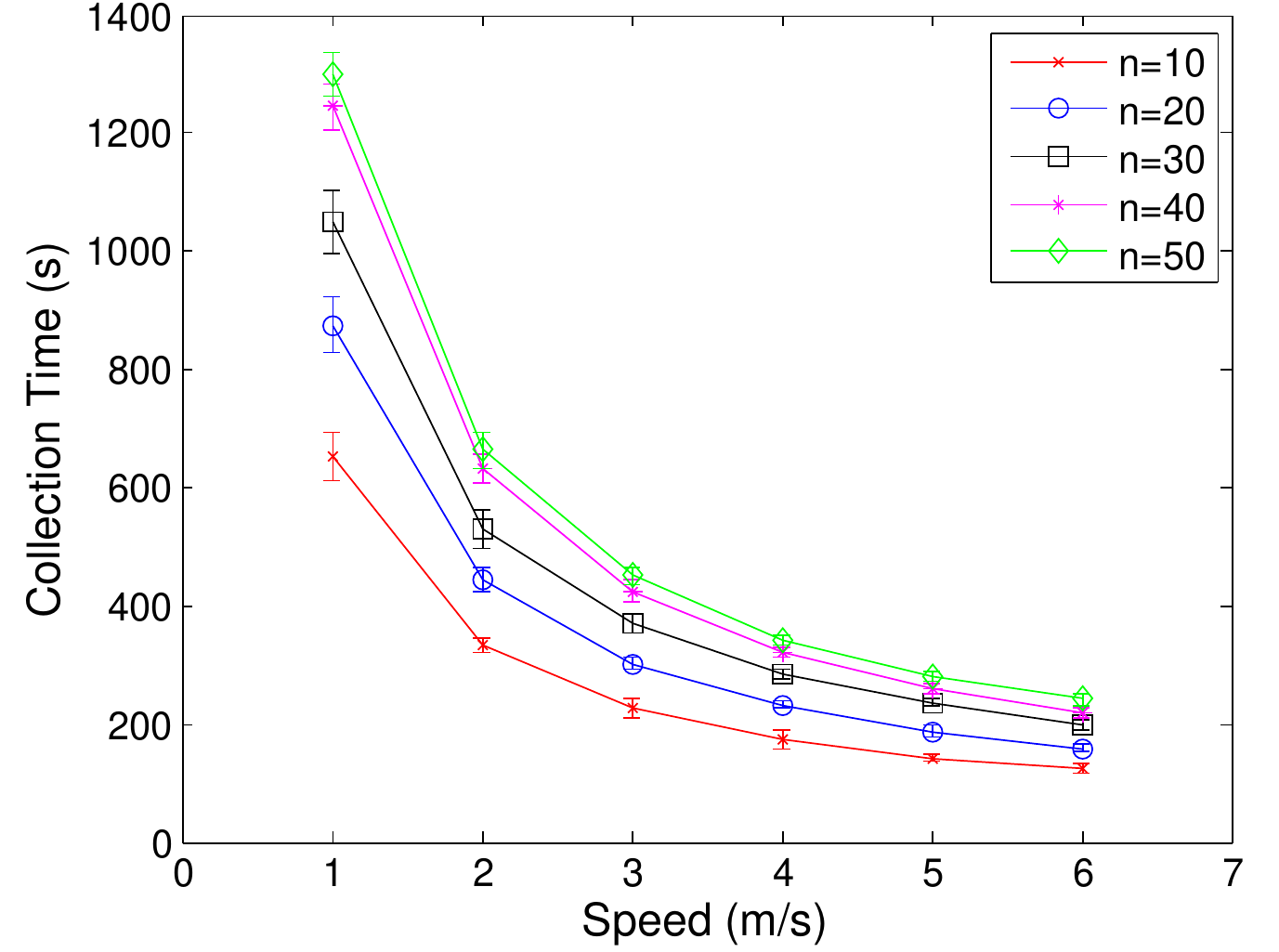}
        \caption{}
        \label{fig:timedelay1}
    \end{subfigure}
    \caption{The impacts of speed on path length and collection time: (a) Path length; (b) Collection time.}\label{fig:metrics1}
\end{figure} 

On these instances, we measure the computation time of SVPP on a 64-bit Windows machine with the processor of Inter(R) Core(TM) i5-4570CUP @3.20GHz. We display the average time for each network scale in Figure \ref{fig:computation_time}. It shows that the major time consumption is made on Step 1 to calculate the permutation, while the procedures (Step 2-5) to search the path takes relatively short time.

\begin{figure}[t]
\begin{center}
{\includegraphics[width=0.5\textwidth]{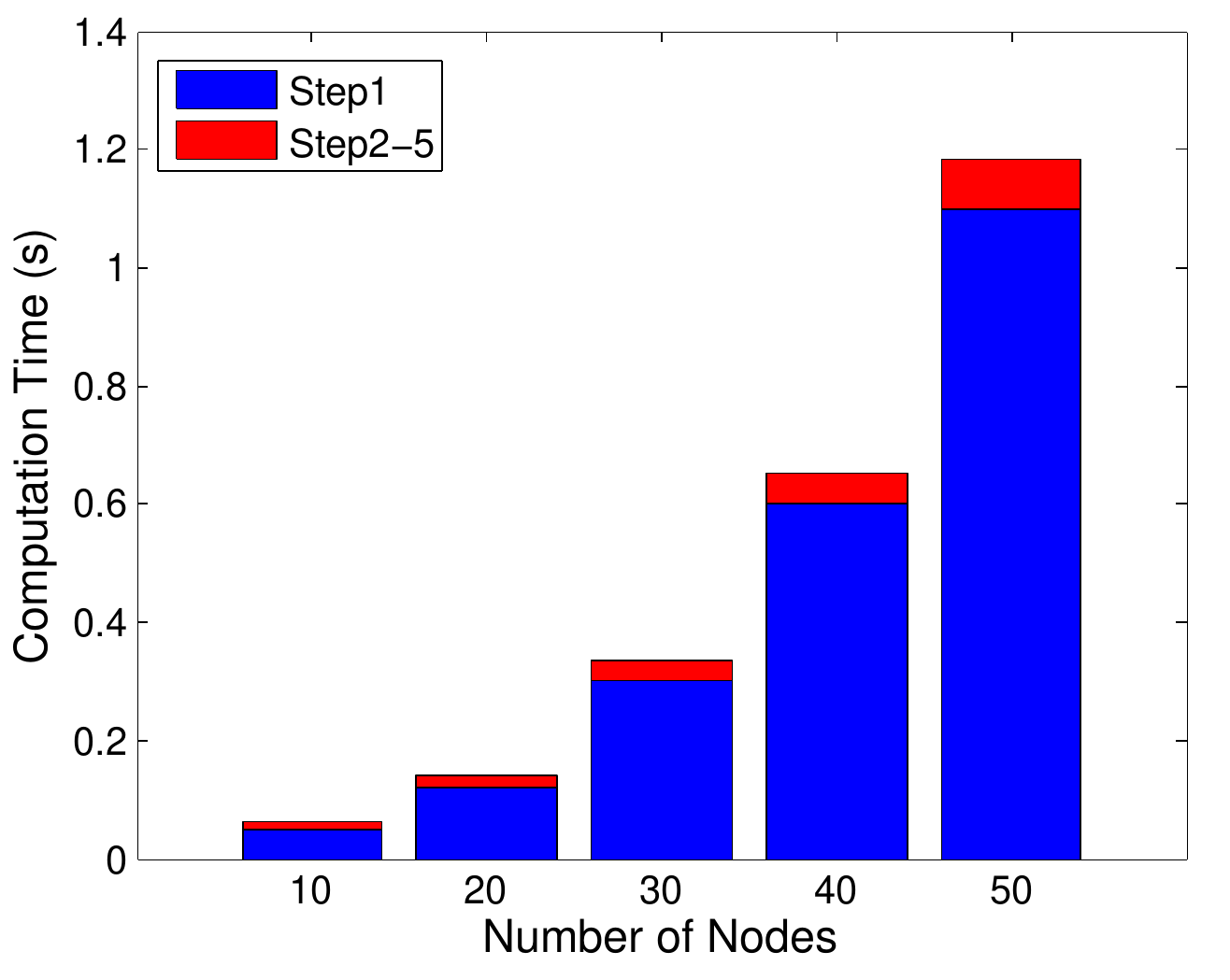}}
\caption{Average computation time of SVPP for different scale networks.}
\label{fig:computation_time}
\end{center}
\end{figure}

Now we investigate the impact of $g$ on system metrics and we focus on a network such as the one shown in Figure \ref{fig:node10_f}. Here, the data loads are uniform among all the nodes, which are between 0 and 1$MB$. For each pair of $g$ and $v$, we compute the shortest viable path by SVPP and the results are shown in Figure \ref{fig:metrics2}. We can learn that for a fixed $v$, the path length is non-decreasing with $g$ and it has an initial stabilization. For example, for $v=3m/s$, with the increasing of $g$, the path length remains at first and then starts to increase after $g=0.08MB$. The reason behind this is amoung the nodes' regular contact times, the shortest one allows to transmit up to $0.08MB$. When $g$ is larger than $0.08MB$, at the shortest regular contact time is not enough, then the reading adjustment applies, which leads to the raising of the path length. Another interesting phenomenon is the length of the path with higher $v$ may be shorter than that with lower $v$. For example, when $g=0.02MB$, the path length with $v=5m/s$ is longer than that with $v=6m/s$. This is because the regular contact times on the path with $v=6m/s$ are still enough for the data load, while several circles are added to the path with $v=5m/s$.
Figure \ref{fig:timedelay2} shows that the collection time increases with data load and increasing the robot speed can reduce the collection time.

\begin{figure}[t]
    \centering
    \begin{subfigure}[t]{0.45\textwidth}
        \includegraphics[width=\textwidth]{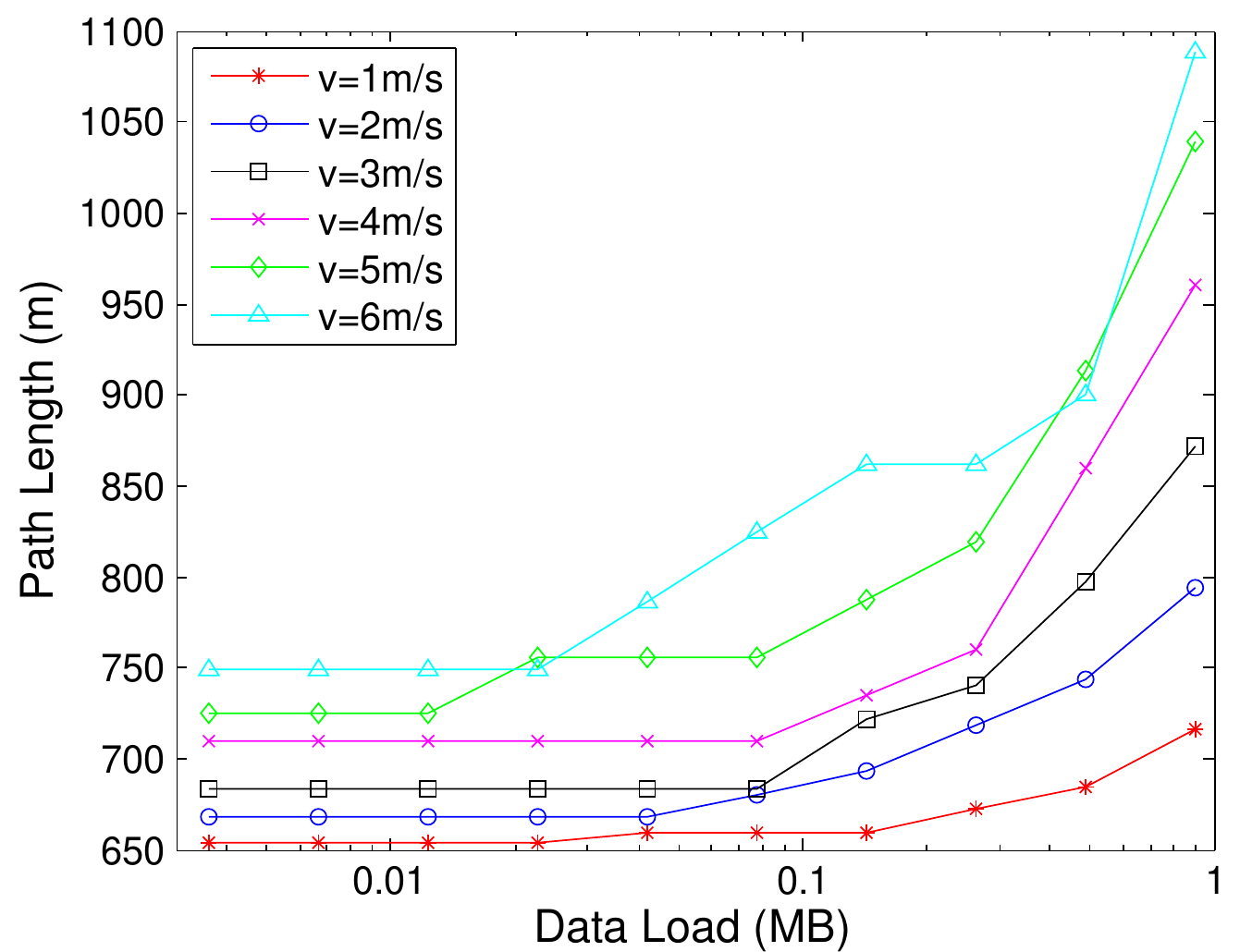}
        \caption{}
        \label{fig:pathlength2}
    \end{subfigure}
    \begin{subfigure}[t]{0.45\textwidth}
        \includegraphics[width=\textwidth]{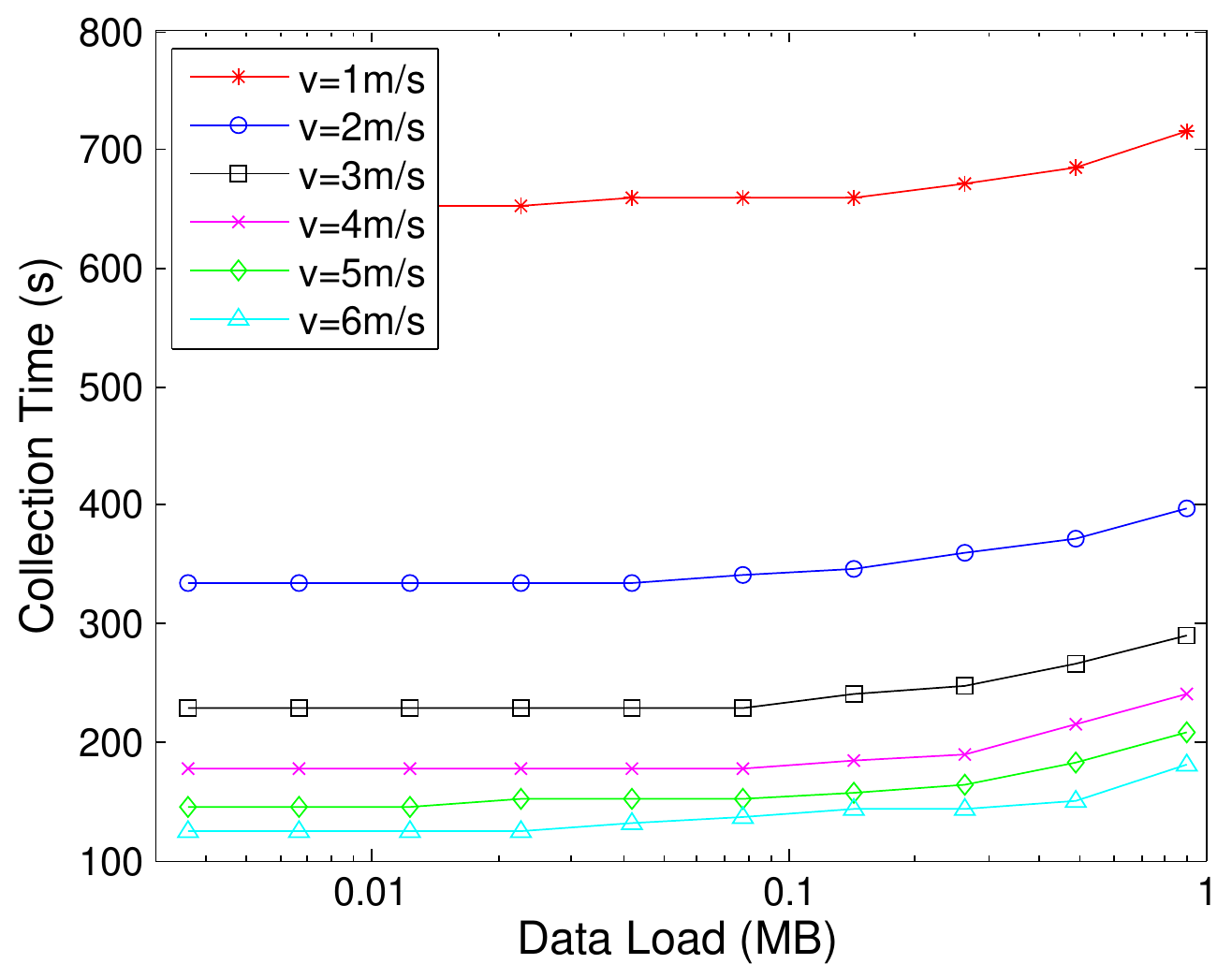}
        \caption{}
        \label{fig:timedelay2}
    \end{subfigure}
    \caption{The impacts of data generation rate on path length and collection time: (a) Path length; (b) Collection time.}\label{fig:metrics2}
\end{figure} 

\subsection{Performance of $k$-SVPP}\label{ksvpp}
This section investigates the performance of $k$-SVPP. 

We focus on a network shown in Figure \ref{fig:node50}. The node in the central of the field, i.e., $s_{29}$, is regarded as the base station, and the robots start and end at its visiting circle. We present the results of $k=2$ and $k=3$. For comparison, the paths generated by viable $k$-SPLITOUR are also displayed. For $k=2$, the split visiting circle is $C_{28}$. Then $\{C_{29}, C_{26},...,C_{28}\}$ forms the first cluster and $\{C_{29}, C_{30},...,C_{13}\}$ is the second. In this instance, $k$-SVPP and viable $k$-SPLITOUR generate the same paths, see Figure \ref{fig:node50_two_path}. For $k=3$, the two split visiting circles are $C_{18}$ and $C_{43}$. Then $\{C_{29}, C_{26},...,C_{18}\}$ forms the first cluster, $\{C_{29}, C_{19},...,C_{43}\}$ forms the second, and $\{C_{29}, C_{42},...,C_{13}\}$ is the third. The three paths generated by $k$-SVPP are shown in Figure \ref{fig:node50_three_path} and those of viable $k$-SPLITOUR is demonstrated by Figure \ref{fig:node50_three_path_compare}. Comparing these two results, we can see that the paths generated by $k$-SVPP outperform those by viable $k$-SPLITOUR. The $k$-length of $k$-SVPP paths is 528.8$m$ and that of viable $k$-SPLITOUR is 569.2$m$. Thus, in this instance, $k$-SVPP saves path length by 7.6\%. The reason behind this improvement is obvious, i.e., as discussed in Section \ref{algorithm_extend}, the split permutations are not optimal in the clusters. Thus, the paths constructed based on these permutations are not the shortest.

\begin{figure}[t]
\begin{center}
{\includegraphics[width=0.45\textwidth]{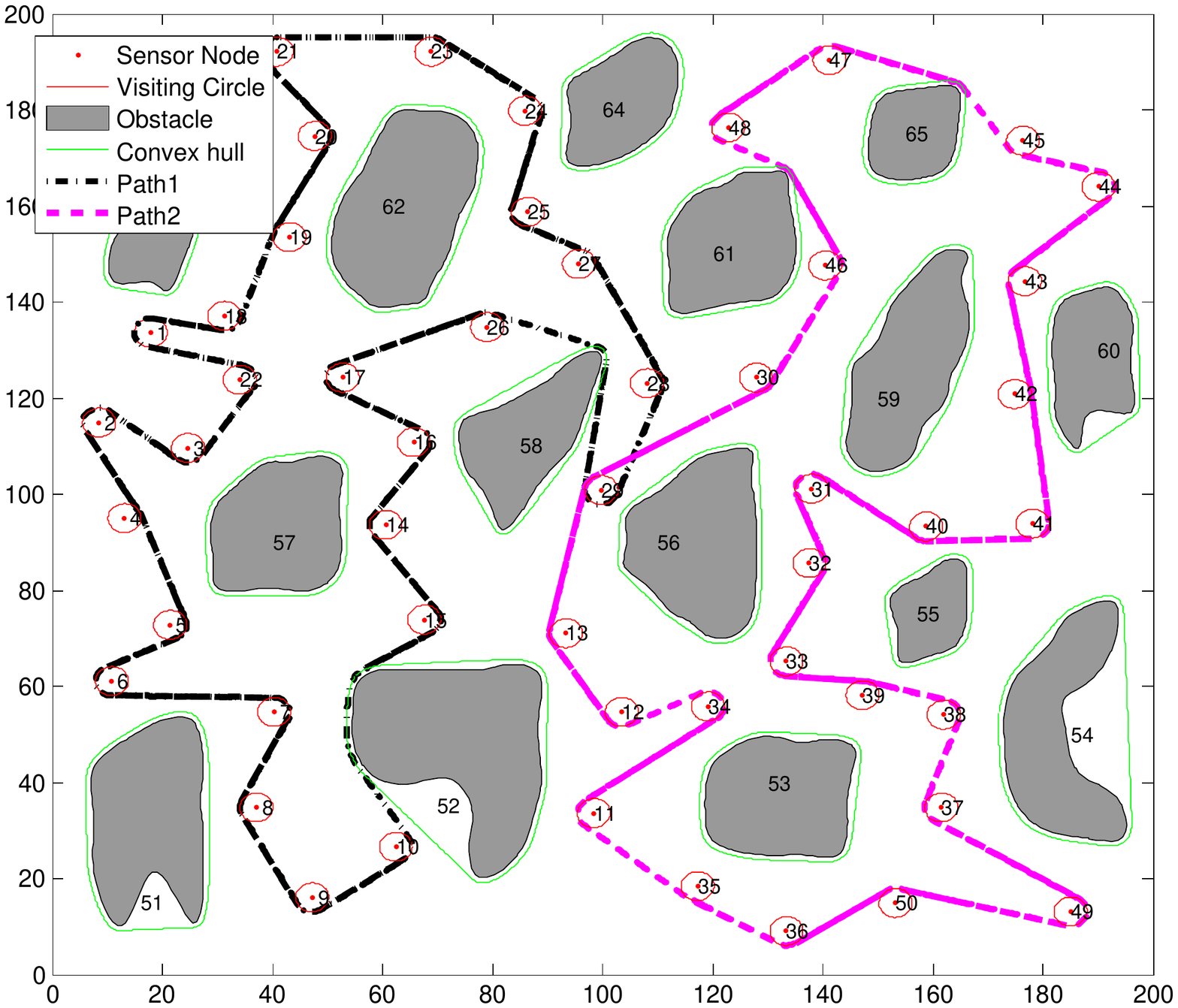}}
\caption{When $k=2$, the paths generated by $k$-SVPP and viable $k$-SPLITOUR on the instance shown in Figure \ref{fig:node50} are the same. The two path lengths are 710.2$m$ and 688.3$m$ respectively.}
\label{fig:node50_two_path}
\end{center}
\end{figure}

\begin{figure}[t]
    \centering
    \begin{subfigure}[t]{0.45\textwidth}
        \includegraphics[width=\textwidth]{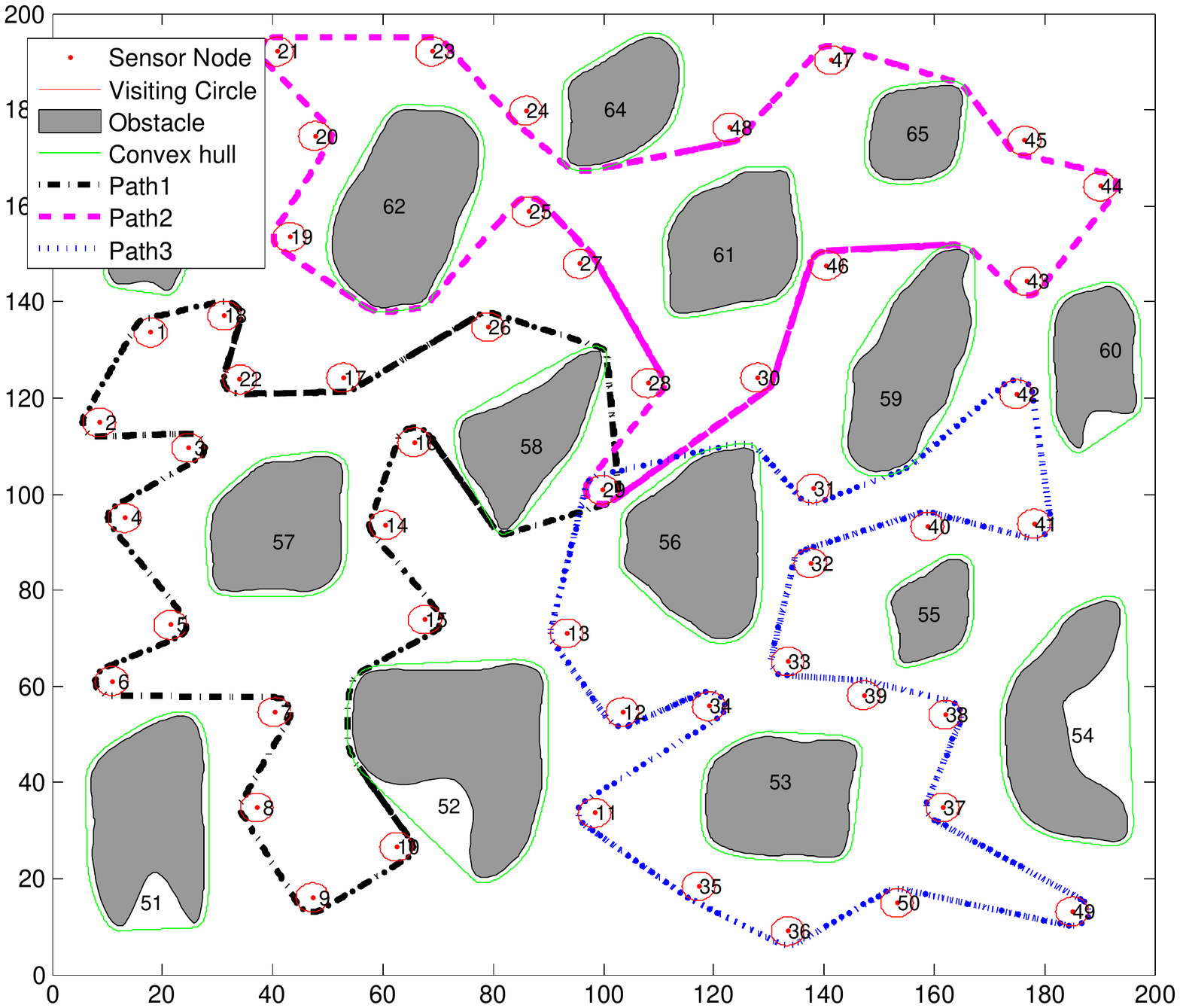}
        \caption{}
        \label{fig:node50_three_path}
    \end{subfigure}
    \begin{subfigure}[t]{0.45\textwidth}
        \includegraphics[width=\textwidth]{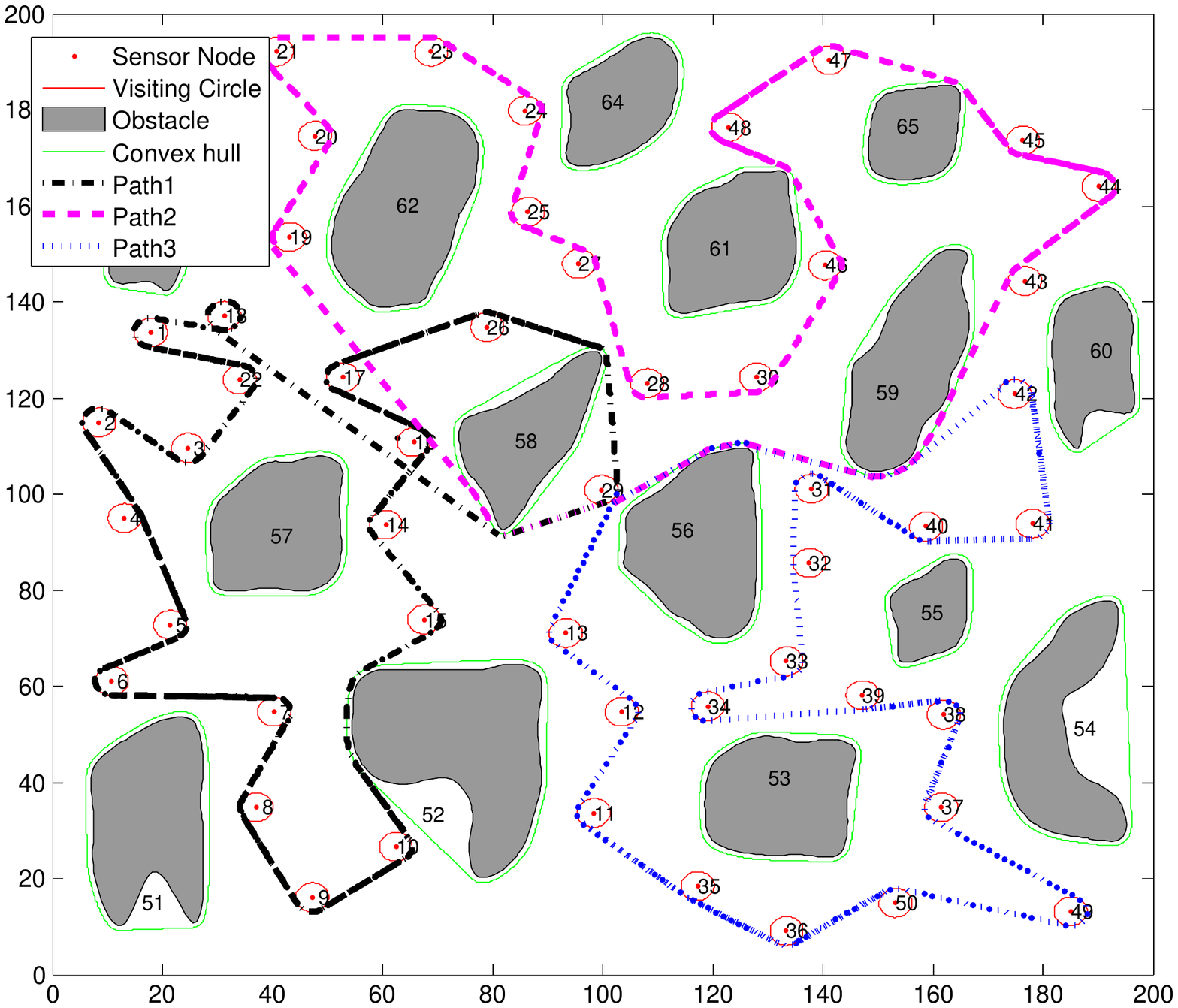}
        \caption{}
        \label{fig:node50_three_path_compare}
    \end{subfigure}
    \caption{Comparison of $k$-SVPP and viable $k$-SPLITOUR when $k=3$: (a) $k$-SVPP (path lengths are 528.8$m$, 517.2$m$ and 517.0$m$); (b) $k$-SPLITOUR (path lengths are 569.2$m$, 559.4$m$ and 552.8$m$).}\label{fig:node50_3}
\end{figure} 

Further, we apply $k$-SVPP to more network instances and compare the performance with viable $k$-SPLITOUR for $k=3$. The results are shown in Figure \ref{fig:klength} We can see that $k$-SVPP performs no worse than viable $k$-SPLITOUR and the former yields 5.5\% improvement on average.

\begin{figure}[t]
\begin{center}
{\includegraphics[width=0.5\textwidth]{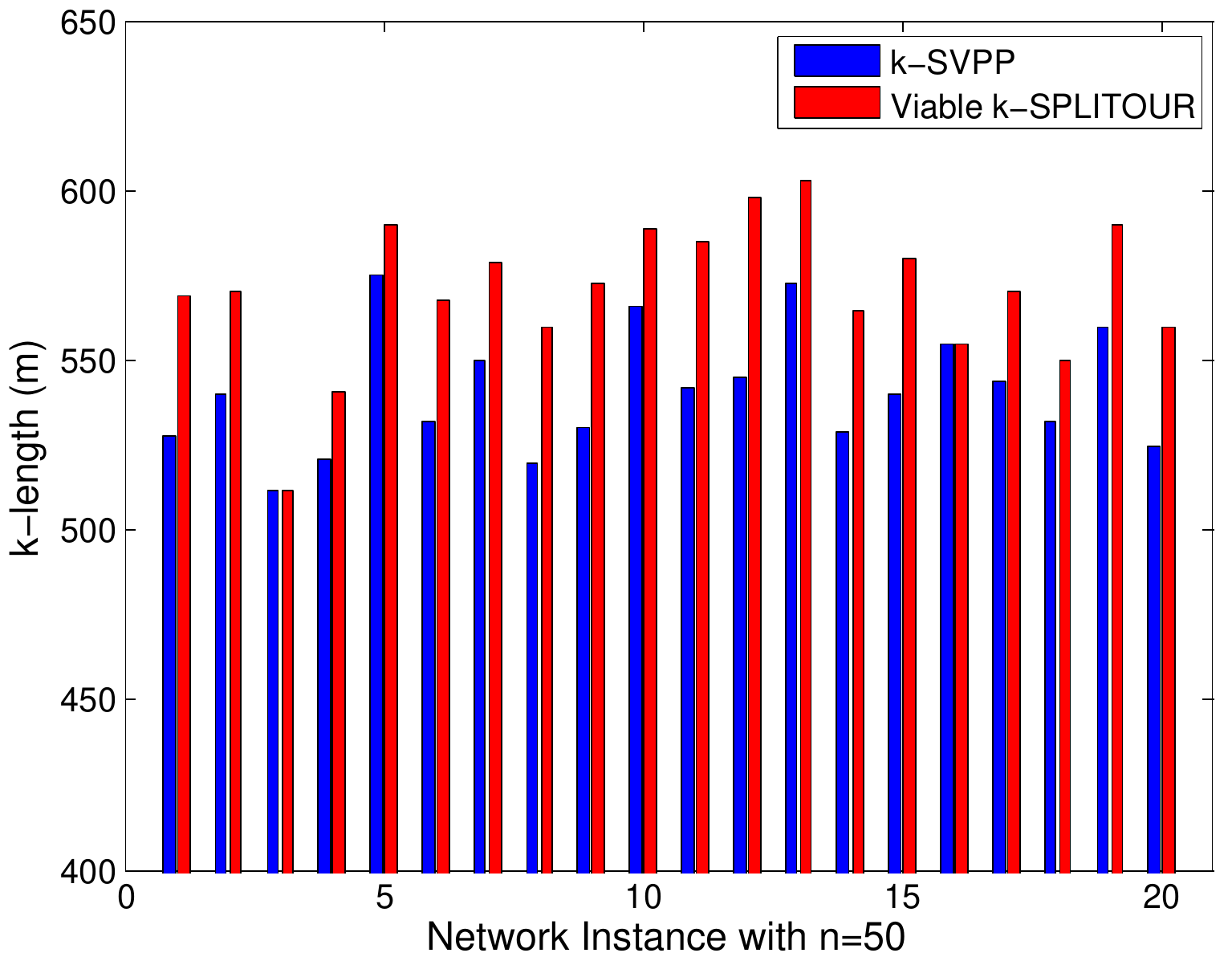}}
\caption{Comparison of $k$-SVPP and viable $k$-SPLITOUR when $k=3$ on the 20 network instances with $n=50$ generated in Section \ref{svpp}.}
\label{fig:klength}
\end{center}
\end{figure}

\subsection{Comparing with Multihop Communication}
This subsection compares our work with a multihop communication algorithm: Shortest Path Routing. We focus on energy consumption since it is an important system metric when the sensor nodes have limited power supplies.  We consider a data transmission rate model of two stairs and the corresponding parameters are $r_1=250KB/s$, $r_2=19.2KB/s$, $d_1=20m$, $d_{max}=d_2=50m$ \cite{ren2014data}. Associated with the transmission rates, the transmission energy consumption rates are $e_1=6.8\times 10^{-6}J/bit$ and $e_2=1.1\times 10^{-5}J/bit$ respectively. 

In Shortest Path Routing, every node transmits its data to the base through the path with shortest distance (such shortest path can be found by existing algorithms such as Dijkstra's algorithm \cite{Dijkstra}).  Since transmitting data usually consumes more energy than sensing and receiving, the energy consumption of the network using multihop communication can be simply expressed as:
\begin{equation}\label{energy_path}
E_{multihop}=\sum_{i}^n eh_ig_i
\end{equation}
where $h_i$ is the hop number from $s_i$ to the base, $g_i$ is the data load of $s_i$ and $e$ indicates the energy consumption for transmitting one bit data. More details about (\ref{energy_path}) can be found in \cite{GAO11}. If we assume the distance between any pair of nodes is larger than $d_1$, then $e=e_2$. Considering $R_{min}$, the sensor nodes can transmit data to the robots via single hop and the energy consumption rate is $e_1$. We use $E_{singlehop}$ to represent the energy consumption when robots are used.
\begin{equation}
E_{singlehop}=E_{node}+E_{robot}
\end{equation}
where $E_{node}=\sum_i^n e_1g_i$ and $E_{robot}=\lambda L$ \cite{energy_consumption}. $L$ gives the total path length of all the robots. $\lambda$ depends on the robot and its speed.
We will see the value of $\lambda$ influences the performance. Here, $\lambda$ is set to be 0.035.
 
In the following part, we focus on the network instances of 50 nodes shown in Section \ref{svpp} and \ref{ksvpp}. Unlike Section \ref{svpp} where we assume uniform data loads, the data loads in this part are distributed in the network, which is more practical. By Multivariate Gaussian Model, we generate the distributed data loads as shown in Figure \ref{fig:distribution}. We study four data load distributions. 1) Central: an interested event occurs at position (100,100) of the field, see Figure \ref{fig:central_distribution}. The closer a node to this position, the more data it generates. 2) Southwest: an interested event occurs at position (50,50), see Figure \ref{fig:southwest_distribution}. Also, we consider two random distributions. 3) R-low, the data loads across the network are generated \textit{Randomly} at a \textit{low} level (0-0.5$MB$). 4) R-high, the data loads are generated \textit{Randomly} at a \textit{high} level (0.5-1$MB$).

\begin{figure}[t]
    \centering
    \begin{subfigure}[t]{0.47\textwidth}
        \includegraphics[width=\textwidth]{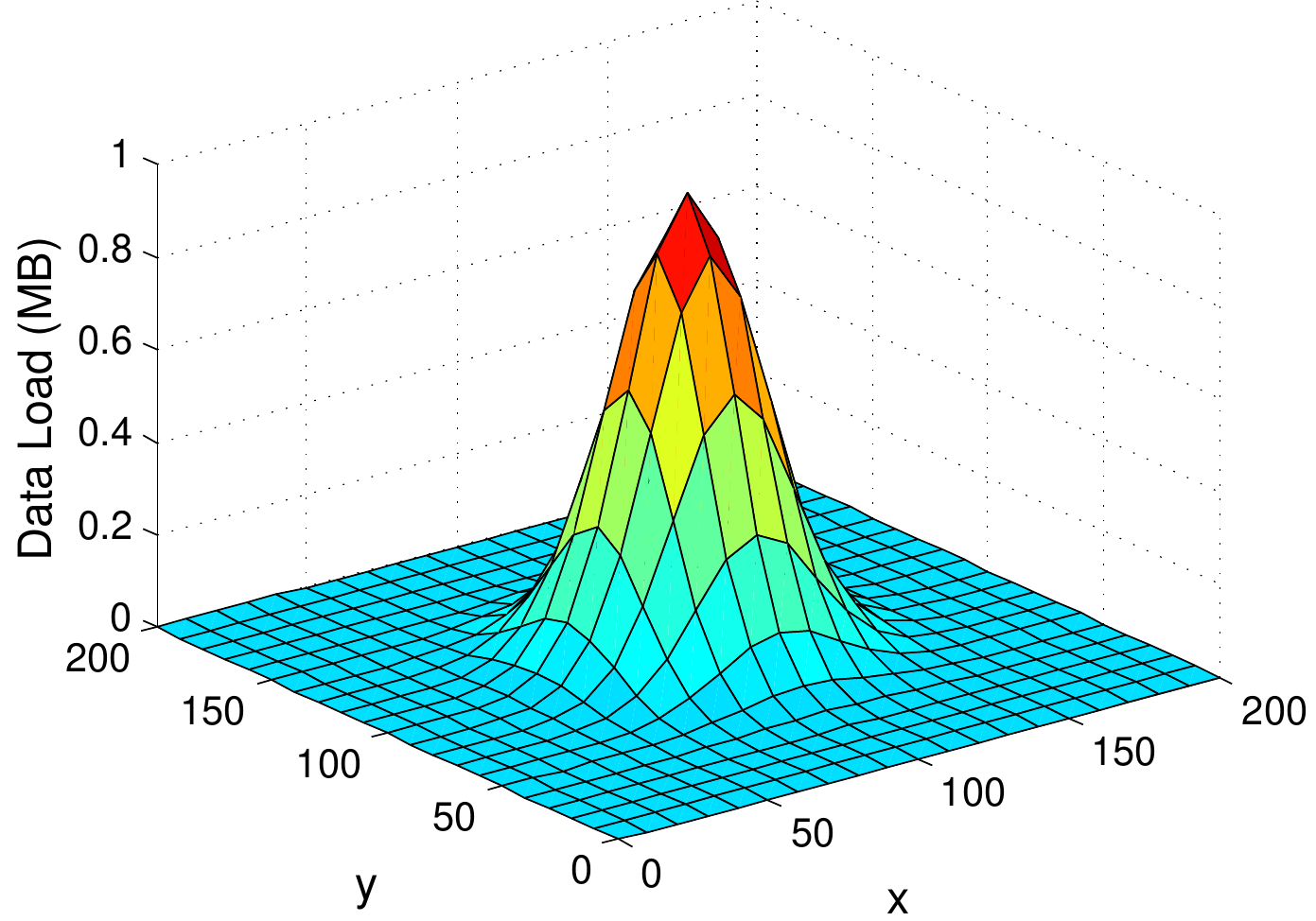}
        \caption{}
        \label{fig:central_distribution}
    \end{subfigure}
    ~
    \begin{subfigure}[t]{0.47\textwidth}
        \includegraphics[width=\textwidth]{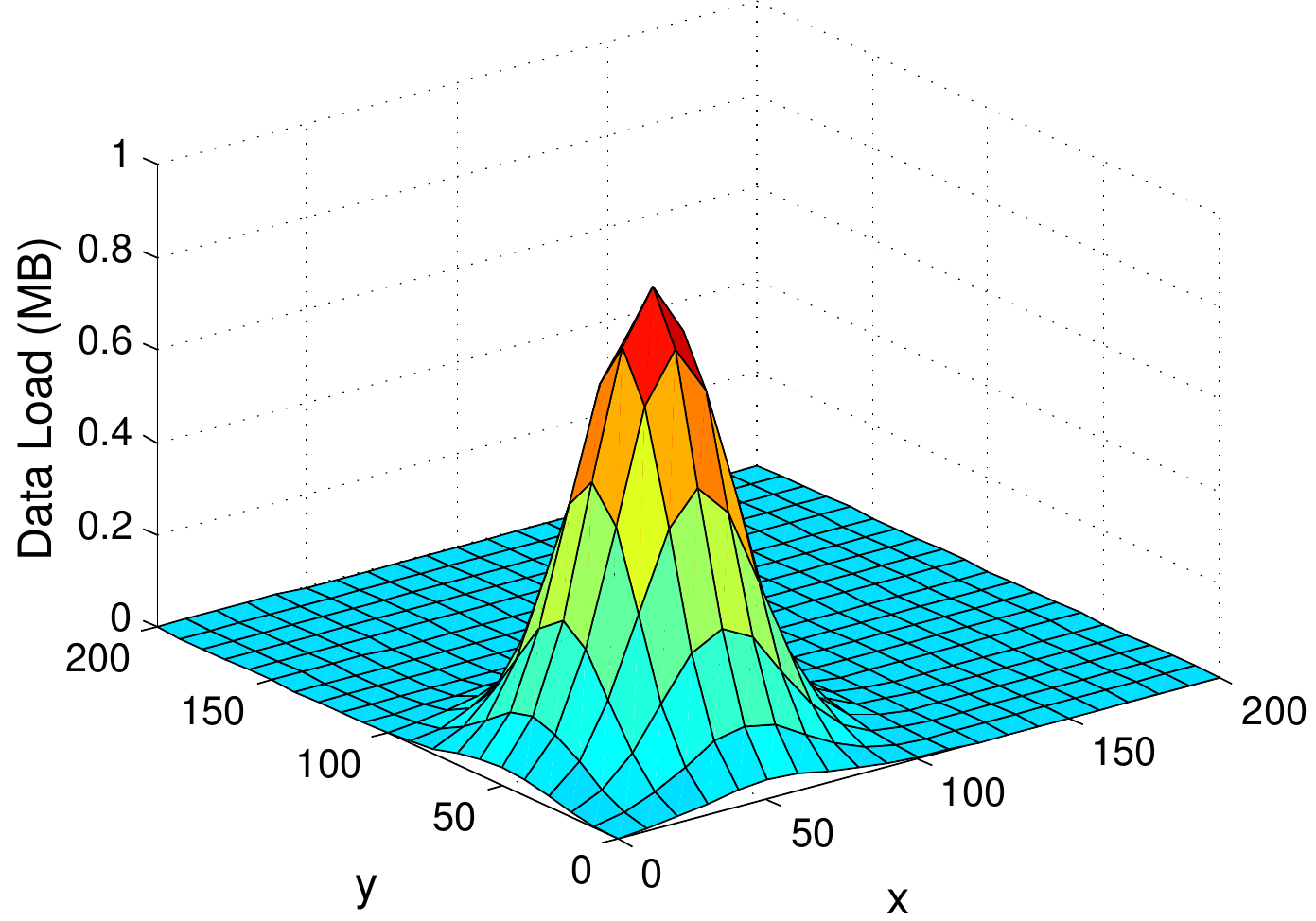}
        \caption{}
        \label{fig:southwest_distribution}
    \end{subfigure}
    \caption{Distributed data loads from Multivariate Gaussian Model where the covariance matrix is $[400,0;0,400]$. The node, nearest to the position of event, generates up to 1$MB$ data during $\cal T$. (a) Central; (b) Southwest.}\label{fig:distribution}
\end{figure}

We apply SVPP, $k$-SVPP, SVPP-no-adjustment, Shortest Paht Routing and Viable $k$-SPLITOUR on the network shown in Figure \ref{fig:node50} and the results are displayed in Figure \ref{fig:compare_energy}. Note, the results are mean values on the 20 independent instances. 
\begin{figure}[t]
    \centering
    \begin{subfigure}[t]{0.47\textwidth}
        \includegraphics[width=\textwidth]{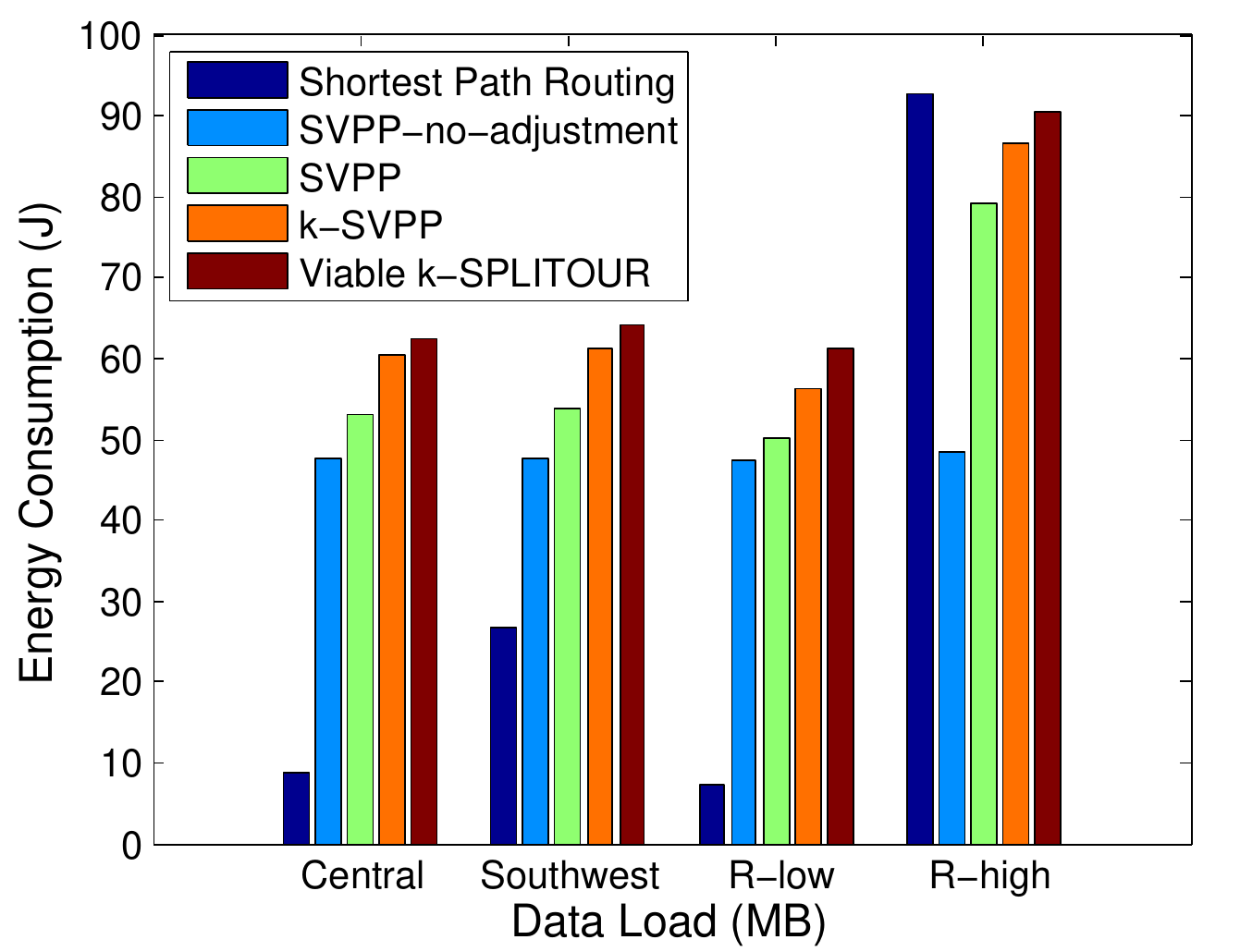}
        \caption{}
        \label{fig:compare_energy_robot}
    \end{subfigure}
    \begin{subfigure}[t]{0.47\textwidth}
        \includegraphics[width=\textwidth]{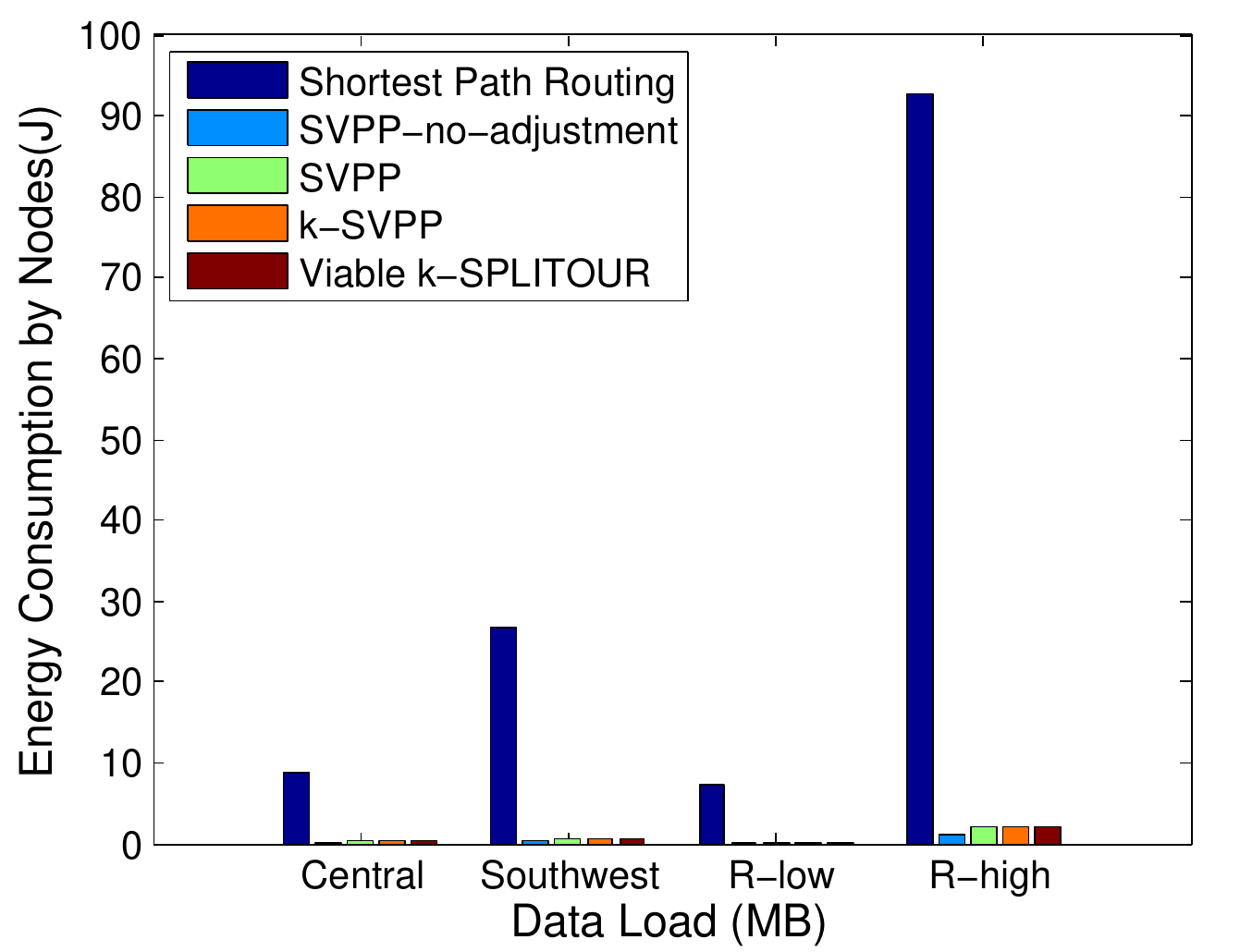}
        \caption{}
        \label{fig:compare_energy_node}
    \end{subfigure}
    \caption{Energy consumption by Shortest Path Routing, SVPP without adjustment, SVPP, $k$-SVPP, Viable $k$-SPLITOUR. (a) By sensor nodes and robots; (b) By sensor nodes.}\label{fig:compare_energy}
\end{figure} 

Figure \ref{fig:compare_energy_robot} demonstrates the energy consumptions by sensor nodes and robots to collect data in these four situations. We learn the following things. The energy consumption by Shortest Path Routing algorithm in Southwest is higher than Central. The reason is that as the base station locates in the central area, the sensor nodes in the Southwest case consume more energy on relaying data, although the total data loads of all the nodes are similar in these two cases. Shortest Path Routing algorithm spends the most energy in R-high and least in R-low due to the total amount of data loads. For the algorithms using robots, SVPP-no-adjustment and SVPP employ single robot while $k$-SVPP and Viable $k$-SPLITOUR both employ three robots. Figure \ref{fig:compare_energy_robot} shows that using single robot consumes less energy than using multiple robots. This is reasonable since the total path length of the multiple robots is longer than the single robot case, as shown in Section \ref{ksvpp}. The benefit of employing multiple robots is to shorten the collection time. Among the two schemes using single robot, SVPP costs more energy than SVPP-no-adjustment. The excess part is used on the rotating movement. We calculate the average collected data percentage of SVPP-no-adjustment in these four distributions: 55\%, 62\%, 95\% and 53\%. To collect all the data, the robot needs to do at least one more tour along the paths. Then SVPP-no-adjustment spends more energy than SVPP. Among the two schemes using multiple robots, Viable $k$-SPLITOUR costs more than $k$-SVPP. The reason is the same as that discussed in Section \ref{ksvpp}. Comparing these schemes on the four distributions, Shortest Path Routing consumes the least energy in the first three cases, since robot movement generally consumes more energy. To exclude the influence of movement energy consumption, we display the part consumed by sensor nodes only in Figure \ref{fig:compare_energy_node}. It shows in every case, using robots to collect data saves sensor nodes around 95\% energy compared to using multihop communication. Also, in terms of total energy consumption, using multihop communication may spend more energy when the data loads are large, such as the R-large case in Figure \ref{fig:compare_energy_robot}. 

We indicate that the results shown in Figure \ref{fig:compare_energy_robot} depends on $\lambda$. Increasing this factor may lead to large energy consumption by the schemes using robots. Although using robots may consume more than conventional multihop communication in terms of total energy consumption, the approaches of using robots are still promising since: the energy consumption by nodes can be saved prominently as shown in Figure \ref{fig:compare_energy_node}; besides, the robots can be recharged much easily than the sensor nodes.

\section{Discussion}\label{discussion}
This section provides some discussions of our work.
\subsection{Extension}\label{extension}
Our proposed algorithms SVPP and $k$-SVPP are designed for the data collection systems with two layers: sensor layer and base station layer. But both of them are not restricted to the systems of two layers. It can be extended to the systems of three layers: sensor layer, gateway layer and base station layer, by adding inner network communication protocols between sensor and gateway layers. The sensor nodes can be clustered locally into different groups and use inner network communication to transmit the extracted data to their corresponding gateways. Then our algorithms will serve the gateway and base station layers in the same way as we do above. As recent studies on the area of WSNs have proposed many network communication protocols, extending our work to the three layer systems will be efficient and effective. In Section \ref{navigation}, we present a specific extension of SVPP in details.

\subsection{Limitation}
We do not consider the computation capability of robots. In other words, we assumed that following the designed paths is able to avoid obstacles. However, this is based on the accuracy of the given information of obstacles such as shapes and positions. We notice that  that in reality such information may be inaccurate or unavailable, for example when the moveable obstacles exist in the sensing field. In this scenario, the proposed algorithms would not be appropriate. 
\subsection{Application}
The proposed work can be applied to the applications of using wheeled robots with bounded turning radius, to collect data from sensor nodes. It may save sensor nodes much energy in the scenario of large data loads, for example in wireless multimedia sensor networks.

The proposed algorithms are originally designed for the applications of collecting information from sensor nodes. They can also be applied in a reverse manner, i.e., to distribute information to the sensor nodes, such as charging the sensor nodes \cite{PLZQ10, ZWL12, DWCXL14} due to the availability of wireless charging technique \cite{kurs2007wireless}. The charging amount can be regarded as the data load, which is decided by the residual energy in the node and its capacity.

\section{Extension to Point to Point Navigation}\label{navigation}
In the above sections, we focus on the scenario of path planning for a or several mobile sinks to visit a number of sensor nodes and return to the origin. The approach can also be used in another way, i.e., navigating a robot or aircraft from a point to another.

Safe mission planning in various threat environments is an important objective for military aircraft
such as strike aircraft, cruise missiles and Unmanned Aerial Vehicles (UAVs). Safe mission planning for such aircraft typically seeks to construct a route from origin to destination that minimizes the risk imposed by enemy threats. Typical examples include minimizing the risk of aircraft detection by radars or other ground sensors \cite{IMM04,ZAB1,BA1,BA2}, minimizing the risk of being destroyed by surface-to-air missiles \cite{KH03,MOR}, and minimizing exposure to radiation while passing through a contaminated area. 

This class of problems was studies in a number of  publications where two main approaches were proposed. The first approach is based on minimizing the cumulative threat exposure when flying through the threat environment \cite{KH03,ZAB2}. In the second approach, the objective is to minimize the probability of being detected \cite{BA1,BA2}. An alternative approach to safe military aircraft routing in threat environments is to be presented below. This approach is based on an optimization model where the objective is to navigate an aircraft to its final destination while obtaining the shortest path that minimizes the threat level of the aircraft path. In some applications, this model is closer to reality. 

\subsection{Extended Problem Statement}
We consider an aircraft or a flying robot moving in a plane and modelled as (\ref{dubinscar}). The initial position of the aircraft is in a known region ${\mathcal M}$. Also, there is a point-wise aircraft final destination ${\mathcal F}$ which also belongs to the region ${\mathcal M}$. There are $n$ steady point-wise  threat agents  in the plane with known coordinates $P_1=(x_1,y_1), P_2=(x_2,y_2),\ldots, P_n=(x_n,y_n)$. Each threat agent $i, 1\leq i\leq n$ has a known threat level continuous function $f_i(d)$ which describes the level of threat for a point at standard Euclidean distance $d\geq 0$ from  the threat agent $i$ located at $(x_i,y_i)$. In real life applications, threat agents can represent surface-to-air missiles or geophysical obstacles such as mountains or volcanoes from
which the aircraft needs to keep a safe distance.

We assume that the threat level functions $f_i(d)$ satisfy the following assumption.

\begin{Assumption}
\label{Ast} 
For any $i=1,2,\ldots,n$, there exists a constant $R_i>0$ such that $f_i(d)=0$ for all $d\geq R_i$,
 $f_i(d)>0$ for all $0\leq d< R_i$, and $f_i(d_1)>f_i(d_2)$ for all $0\leq d_1<d_2< R_i$.  Moreover, we assume that $R_i>R_{min}$ for all $i=1,2,\ldots,n$. A typical threat level function is shown in Fig. \ref{fig1}. Further, it is worth mentioning that this formulation is different from the conventional risk-theoretic function which monotonically increases with the distance.
\end{Assumption}

\begin{figure}[t]
\begin{center}
{\includegraphics[width=0.4\textwidth]{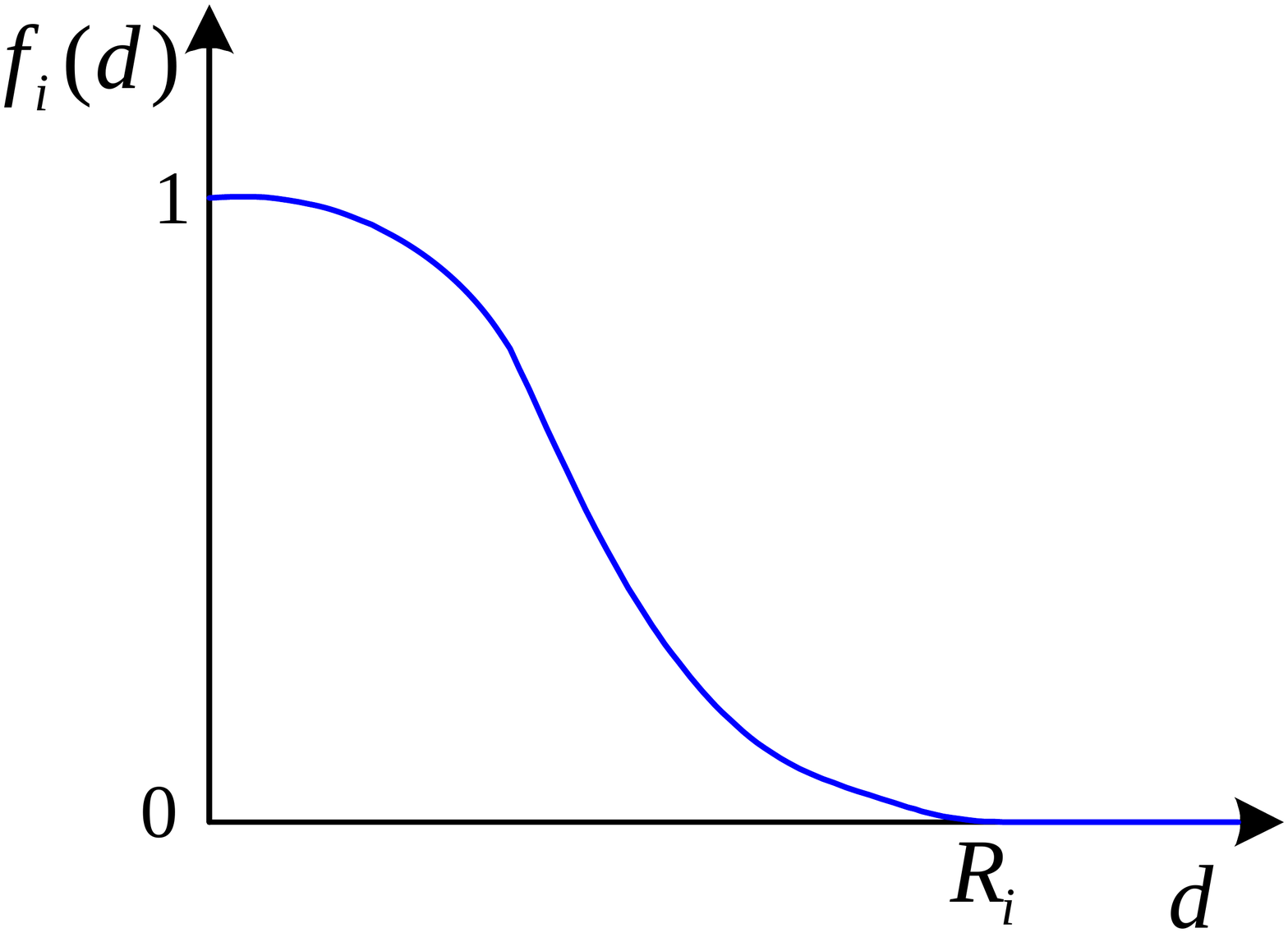}}
\caption{A typical threat level function.}
\label{fig1}
\end{center}
\end{figure}

Moreover, we suppose that the region ${\mathcal M}$ satisfies the following assumption.

\begin{Assumption}
\label{As-1} The region ${\mathcal M}$ is a convex bounded closed set. The boundary $\partial D({\mathcal M})$ of the region ${\mathcal M}$ is a smooth curve with the curvature $k(p)$ at any point $p$ satisfying $k(p)\leq \frac{1}{R_{i}}$ for all $i=1,2,\ldots, n$.
\end{Assumption}

Assumption \ref{As-1} guarantees that the shape and the boundary of the region ${\mathcal M}$ are relatively simple.
This greatly reduces the complexity of the proposed navigation algorithm.

\begin{definition}\label{define_2_1}
A path $p(t)=(x(t),y(t))$ of the aircraft (\ref{dubinscar}) is said to be {\em reaching the final destination} if there exists a time $t_f>0$ such that $p(t_f) ={\mathcal F}$ and $p(t)$ belongs to the region ${\mathcal M}$ for all $i , t\in [0,t_f]$. Furthermore, $T[p(\cdot)]$ denotes the maximum threat level to which the aircraft was exposed on the path $p(\cdot)$:
\begin{equation}
\label{mtl}
T[p(\cdot)]:=\max_{i=1,2,\ldots,n,t\in [0,t_f]}f_i(d_i(t)) 
\end{equation}
where $d_i(t)$ is the distance between the aircraft and the threat agent $i$ at time $t$: 
\[
d_i(t):=\sqrt{(x(t)-x_i(t))^2+(y(t)-y_i(t))^2}.
\]
Moreover, $L[p(\cdot)]$ denotes the length of the path $p(\cdot)$.
\end{definition}

\begin{definition}
A reaching of the final destination path  $p_0(t)=(x_0(t),y_0(t))$ of the aircraft (\ref{dubinscar}) is said to be {\em optimal} if the following two conditions holds:
\begin{enumerate}
\item
For any other reaching of the final destination path $p(t)=(x(t),y(t))$, the inequality 
\begin{equation}
\label{E1}
T[p_0(\cdot)]\leq T[p(\cdot)]
\end{equation}
holds.
\item
If $T[p_0(\cdot)]= T[p_{*}(\cdot)]$ for some other reaching of the final destination path $p_{*}(t)=(x_{*}(t),y_{*}(t))$, then 
\begin{equation}
\label{E2}
L[p_0(\cdot)]\leq L[p_{*}(\cdot)].
\end{equation}
\end{enumerate}
If $p_0(t)$ is an optimal path, then  the number $T[p_0(\cdot)]$ is called the minimum achievable threat level.
\end{definition}

\begin{remark}
An optimal aircraft path guarantees the minimum possible
threat level. Moreover, among different paths with the same minimum threat level, the optimal path has a shortest length.
\end{remark} 

{\bf Problem statement:} Our problem is to design an optimal
aircraft path. In other words, the objective is to navigate the aircraft to the final destination ${\mathcal F}$ inside the region ${\mathcal M}$ while minimizing the maximum threat level and the length of the aircraft path.

{\bf Example: Minimizing worst case probability of being hit:} We introduce an example in which the optimal solution
requires minimizing the maximum threat level. Suppose that
the threat agents $P_1,P_2,\ldots,P_n$ are known ground locations from which surface-to-air missiles can be fired to hit the aircraft. Let $N$ be a total number of missiles. 
The probability of each missile to hit the aircraft depends on the distance $d$ between the aircraft and the agent from which the missile is fired and is described by a known function $f(d)$. The distribution of missiles between the threat agents $P_1,P_2,\ldots,P_n$  are not known to us. Our objective is to minimize the worst case (for all missiles distributions between the agents and all enemy firing strategies) probability that at least one of the $N$ missiles will hit the aircraft.

Let $[0,t_f]$ be the time interval during which the aircraft
moves from origin to destination, $d_i(t)$ be the distance
between the aircraft and the agent $i$ at time $t$. Furthermore, let
\[
d_{min}:=\min_{i=1,2,\ldots,n, ~~t\in [0,t_f]}d_i(t).
\]
It is obvious that for any aircraft path, the worst case scenario for us is when all $N$ missiles are fired from the closest agents at the moments when the aircraft is at the shortest distance $d_{min}$. In this case, the probability that the aircraft is hit by at least one missile is
\[
p_h=1-(1-f(d_{min}))^N=1-(1-\theta)^N
\]
where $\theta$ is the maximum threat level of the trajectory.
Therefore, the minimum possible value for the probability
$p_h$ is delivered by an aircraft path  with the minimum possible threat level $\theta$.

\subsection{Optimal Aircraft Path}
Now, we build optimal aircraft paths that minimize the threat level and the length.

\begin{notation}\label{notation_3_1}
In this chapter, $d(w,h)$ denotes the standard Euclidean distance between the points $w$ and $h$. Moreover, 
$d(w,M)$ denotes the minimum  distance between the point $w$ and the closed set $M$.
\end{notation}

\begin{definition}\label{define_3_1}
Let $P_i$  be a threat agent such that the distance 
$d(P_i,\partial D({\mathcal M}))$ between $P_i$ and the boundary $\partial D({\mathcal M})$  of the region ${\mathcal M}$ satisfies $d(P_i,\partial D({\mathcal M}))<R_i$. Then it immediately follows from Assumption \ref{As-1}, that there exists a unique point $P_{*}$ at the  boundary of ${\mathcal M}$ such that  $d(P_i,\partial D({\mathcal M}))=d(P_i,P_{*})$. Then the point $P_{*}$ is called a critical point of type 1, and the number $\theta:=f_i(d(P_i,P_{*}))>0$ is called a critical threat level of type 1.
\end{definition}

\begin{definition}\label{define_3_2}
Let $P_i$ and $P_j$ be threat agents such that $d(P_i,P_j)<R_i+R_j$. Since the threat functions $f_i(\cdot), f_j(\cdot)$ are continuous and satisfy Assumption \ref{Ast}, there exists a unique point $P_{*}$ at the straight segment connecting $P_i$ and $P_j$ such that $f_i(d(P_i,P_{*}))=f_j(d(P_j,P_{*}))>0$. The point $P_{*}$ is called a critical point of type 2, and the number $\theta:=f_i(d(P_i,P_{*}))=f_j(d(P_j,P_{*}))>0$ is called a critical threat level of type 2.
\end{definition}

\begin{remark}\label{remark_3_1}
It is obvious that if $d(P_i,P_j)<R_i+R_j$ and the threat functions $f_i(\cdot)$ and $f_j(\cdot)$ are identical then the corresponding critical point $P_{*}$ is the middle of
the straight segment connecting $P_i$ and $P_j$.
\end{remark}

\begin{definition}
A number $\theta\geq 0$ is said to be a critical threat level
if one of the following conditions holds:
\begin{enumerate} 
\item 
$\theta =0$;
\item
$\theta$ is a critical threat level of type 1;
\item
$\theta$ is a critical threat level of type 2.
\end{enumerate}
Further, a point $P_{*}\in {\mathcal M}$ is said to be a critical point if it is either a critical point of type 1 or 
a critical point of type 2.
\end{definition}

Let $\theta\geq 0$ be a critical threat level. We now construct a non-convex closed set ${\mathcal M}(\theta)\subset {\mathcal M}$ as follows.
\begin{notation}
\label{note}
For any $i=1,2,\ldots,n$, introduce a radius $R_i(\theta)>0$ where $R_i(\theta):=R_i$ if $\theta =0$, and $R_i(\theta)$ satisfies $f_i(R_i(\theta))=\theta$ for $\theta>0$.
Notice that if $f_i(d)<\theta$ for all $d$, we do not introduce any $R_i(\theta)$. Now we construct the set ${\mathcal M}(\theta)$ by deleting from the region ${\mathcal M}$ all the open disks centred at the threat agents $P_i$ with the radius $R_i(\theta)$; see Fig. \ref{fig2}.
\end{notation}

\begin{figure}[t]
    \centering
    \begin{subfigure}{0.45\textwidth}
        \includegraphics[width=\textwidth]{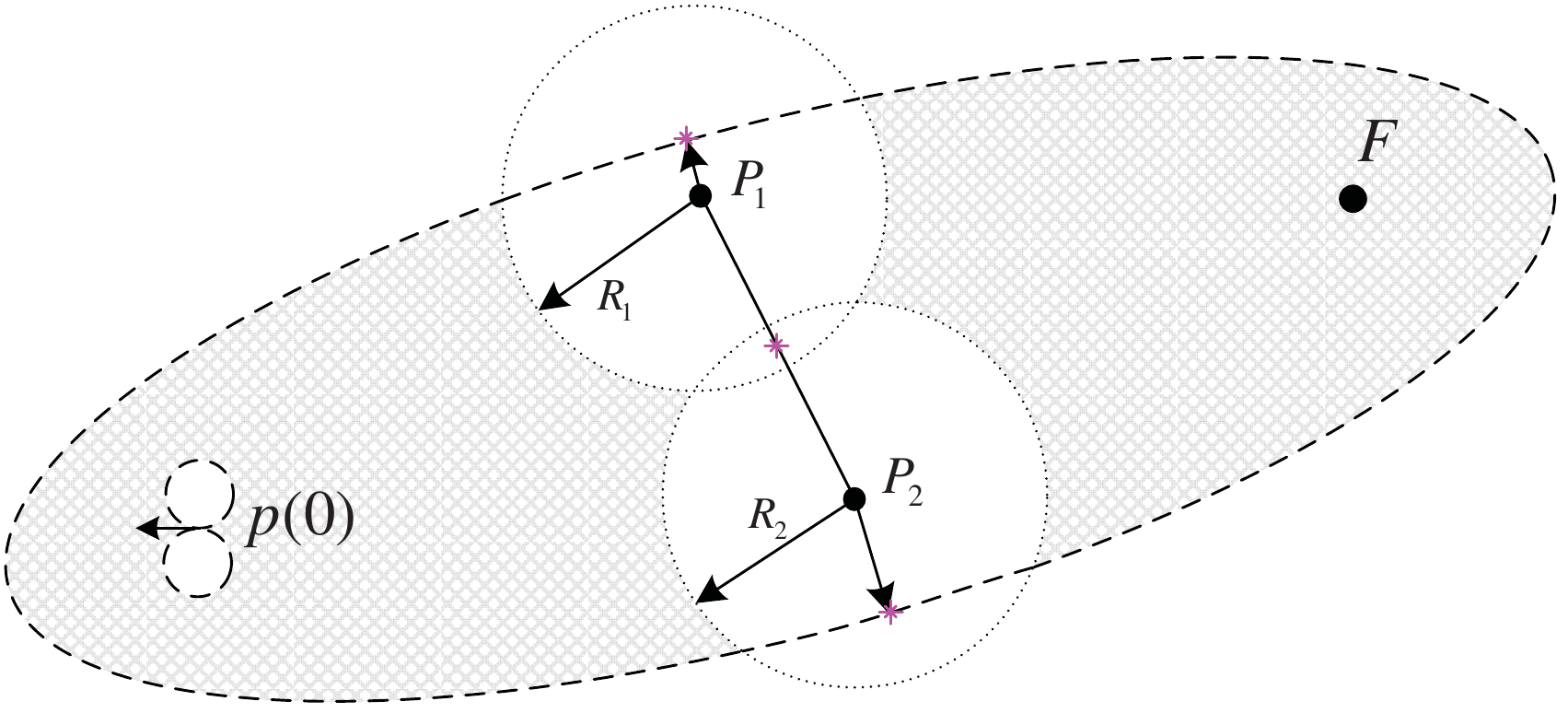}
        \caption{}
        \label{fig2a}
    \end{subfigure}
    \begin{subfigure}{0.45\textwidth}
        \includegraphics[width=\textwidth]{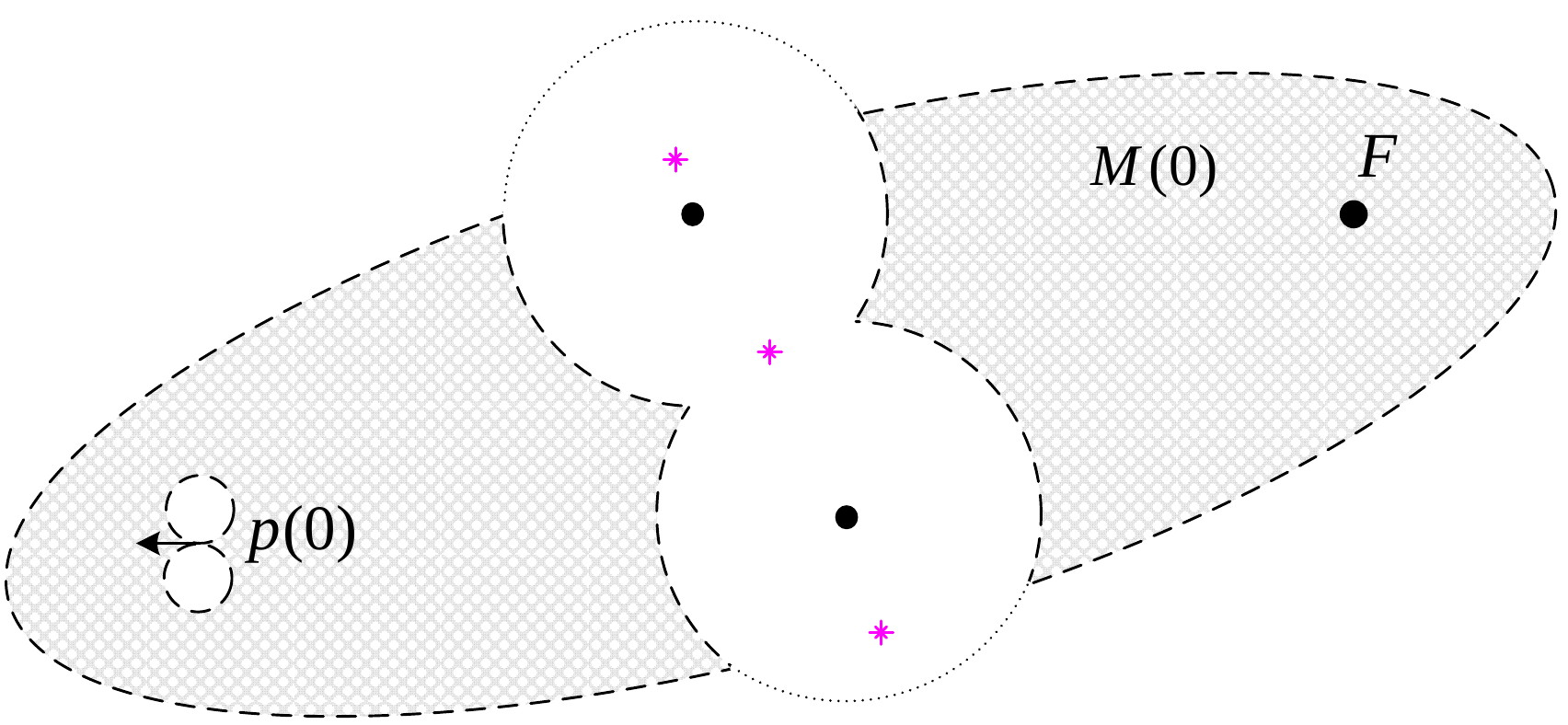}
        \caption{}
        \label{fig2b}
    \end{subfigure}
    \begin{subfigure}{0.45\textwidth}
        \includegraphics[width=\textwidth]{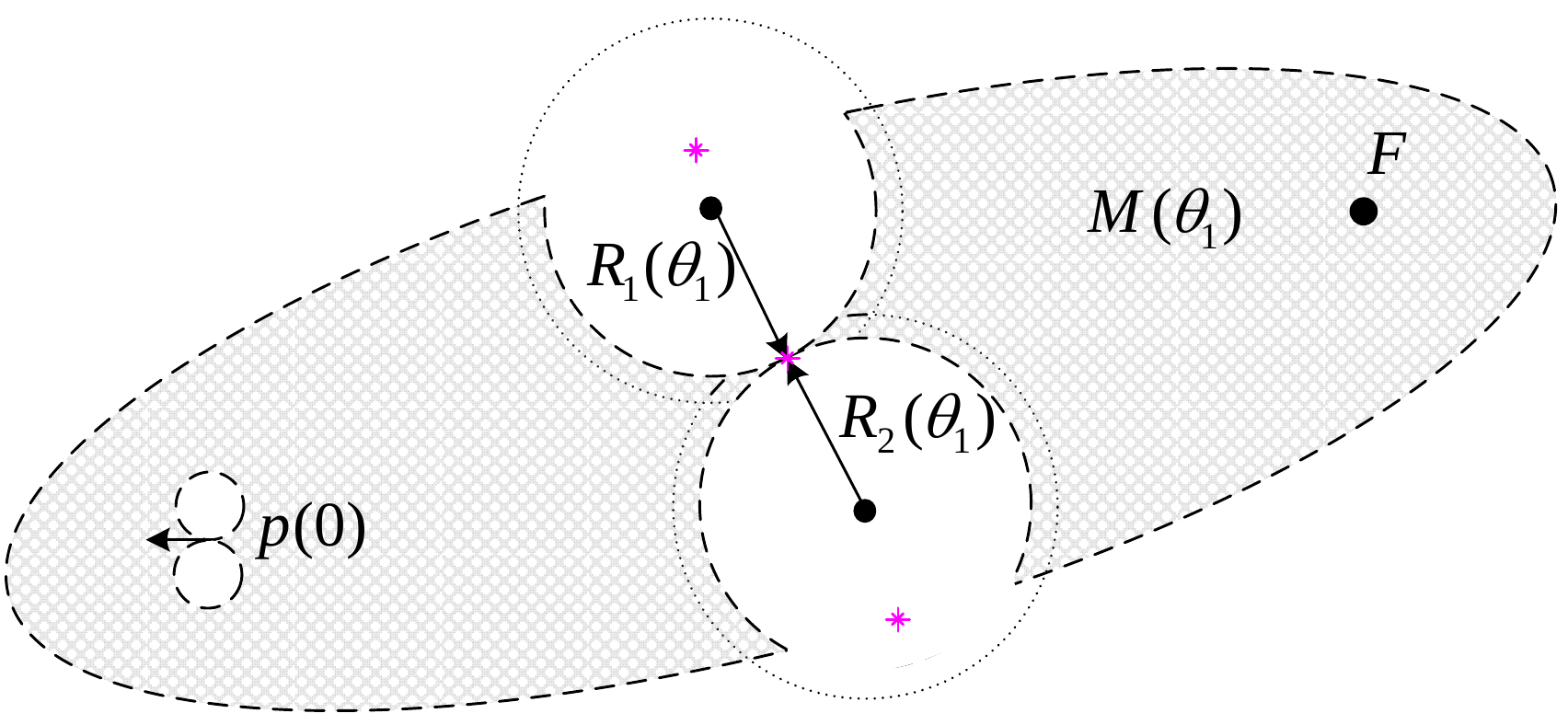}
        \caption{}
        \label{fig2c}
    \end{subfigure}
    \begin{subfigure}{0.45\textwidth}
        \includegraphics[width=\textwidth]{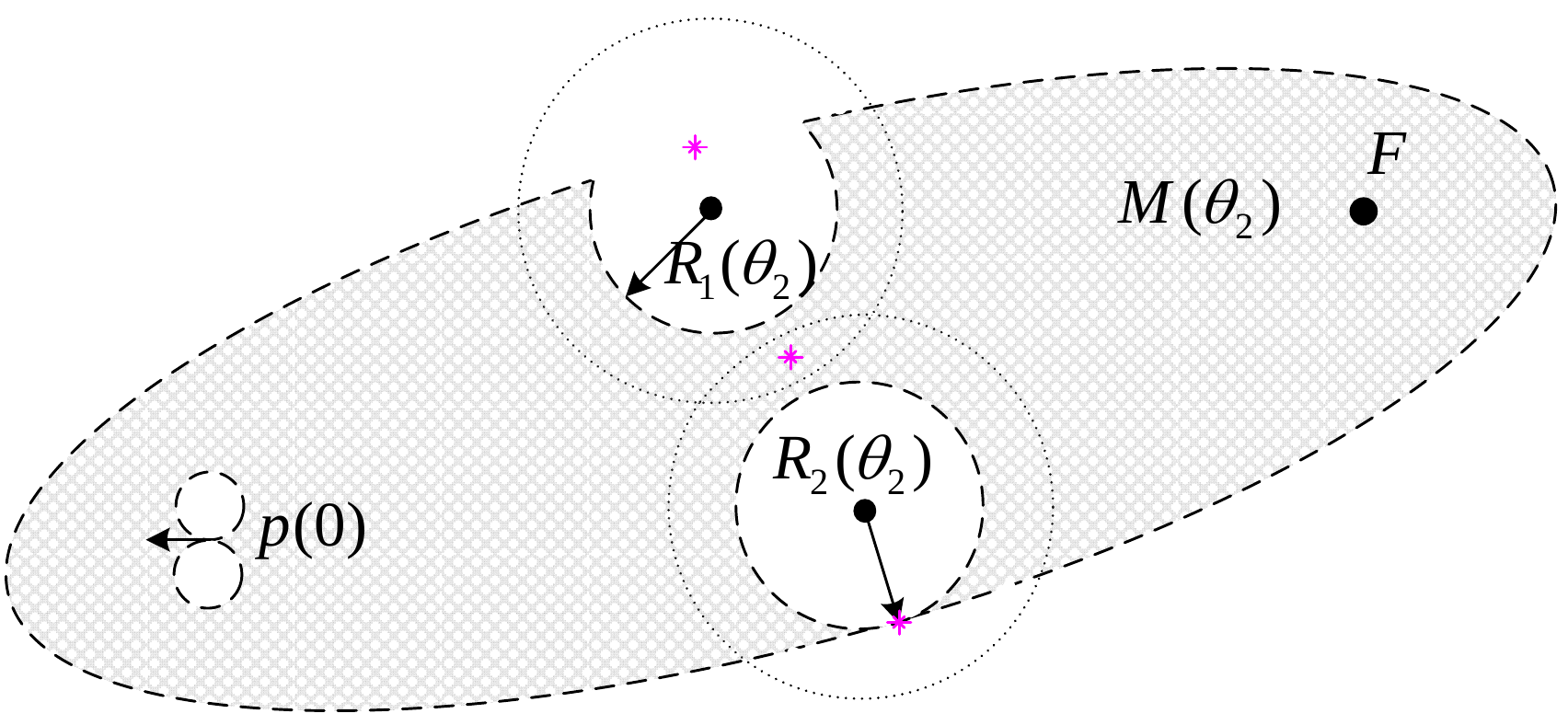}
        \caption{}
        \label{fig2d}
    \end{subfigure}
    \caption{Demonstration for constructing set ${\mathcal M}(\theta)$. (a) Initial figure with two agents locating at $P_1$ and $P_2$ with threat radius $R_1$ and $R_2$. (b) ${\mathcal M}(0)$. (c) ${\mathcal M}(\theta _1)$ by applying a critical point of type 2 and $\theta _1 =0.1$. (d) ${\mathcal M}(\theta _2)$ by applying a critical point of type 1 and $\theta _2 =0.3$.}\label{fig2}
\end{figure}  

It is clear that ${\mathcal M}(\theta)$ consists of all points ${\mathcal M}$ for which the threat level defined by
(\ref{mtl}) does not exceed $\theta$. It is also obvious that
${\mathcal M}(\theta_1)\subset {\mathcal M}(\theta_2)$ if 
$\theta_1\geq \theta_2$.

\begin{definition}
A critical threat level $\theta\geq 0$ is said to be 
a critical threat level with the connectivity requirement if
there exists a continuous curve belonging to ${\mathcal M}(\theta)$ and connecting the aircraft initial position $p(0)$ and the final destination ${\mathcal F}$. Moreover, a critical threat level $\theta$ with the connectivity requirement is called the minimum critical threat level with the connectivity requirement if the connectivity requirement does not hold for any critical threat level $\theta_{*}<\theta$.
\end{definition}

\begin{Assumption}
\label{Aso3-} 
Let $\theta$ be the minimum
critical threat level with the connectivity requirement.
We assume that $R_i(\theta)\geq R_{min}$ for all $i=1,2,\ldots,n$. 
\end{Assumption}

We also assume that the initial position $p(0)=(x(0),y(0))$ of the aircraft is far enough from the threat agents $P_1=(x_1,y_1), P_2=(x_2,y_2),\ldots, P_n=(x_n,y_n)$, the aircraft final destination ${\mathcal F}$ and the boundary of the region ${\mathcal M}$.
\begin{Assumption}
\label{Aso3} 
The following inequalities hold: $d({\mathcal F},p(0))> 8R_{min}$ and $d(P_i,p(0))> 8R_{min}+R_i$ for all $i=1,2,\ldots,n$. Finally, the distance between the aircraft initial position $p(0)$ and the boundary of the region ${\mathcal M}$ is greater than $8R_{min}$.
\end{Assumption}

Assumption \ref{Aso3} is a technical assumption that is necessary for the proof of our main theoretical result.

\begin{definition}
There are two circles with the radius $R_{min}$ that cross the initial aircraft position $p(0)$ and tangent to the aircraft initial heading $\alpha(0)$. We will call them the initial circles.
\end{definition}

\begin{definition}\label{define_3_6}
Let $\theta$ be the minimum critical threat level with the connectivity requirement. A straight line $S$ is said to be a tangent line if one of the following conditions holds:
\begin{enumerate}
\item
The line $S$ is simultaneously tangent to two circles centred
at $P_i$ and $P_j$ with radii $R_i(\theta)$ and $R_j(\theta)$, correspondingly, see Fig. \ref{fig3a} 
for an example.
\item
The line $S$ is simultaneously tangent to a circle centred
at $P_i$ with radius $R_i(\theta)$ and
an initial circle, see Fig. \ref{fig3b} 
for an example.
\item
The line $S$ is tangent to a  circle centred
at $P_i$ with radius $R_i(\theta)$ and crosses the final destination ${\mathcal F}$, see Fig. \ref{fig3c}
for an example.
\item
The line $S$ is tangent to an initial circle and crosses the final destination ${\mathcal F}$, see Fig. \ref{fig3d} 
for an example.
\end{enumerate}
\end{definition}

\begin{figure}[t]
    \centering
    \begin{subfigure}{0.35\textwidth}
        \includegraphics[width=\textwidth]{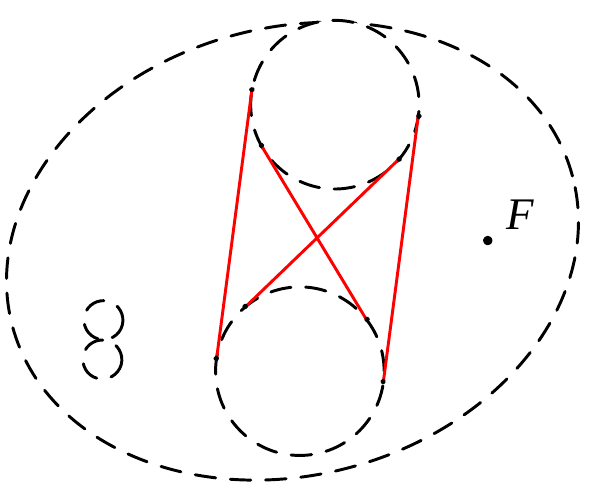}
        \caption{}
        \label{fig3a}
    \end{subfigure}
    \begin{subfigure}{0.35\textwidth}
        \includegraphics[width=\textwidth]{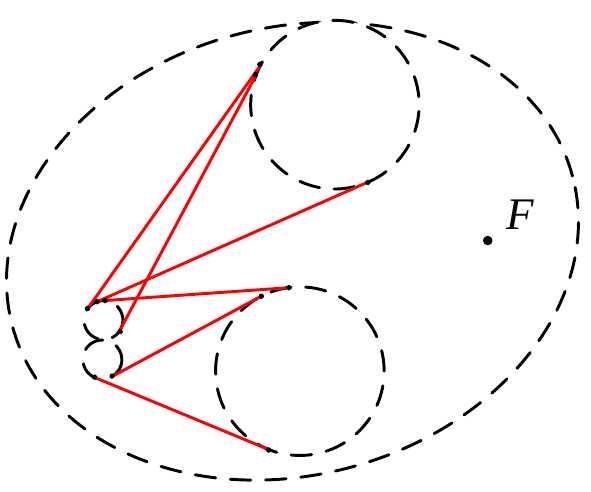}
        \caption{}
        \label{fig3b}
    \end{subfigure}
    \\
    \begin{subfigure}{0.35\textwidth}
        \includegraphics[width=\textwidth]{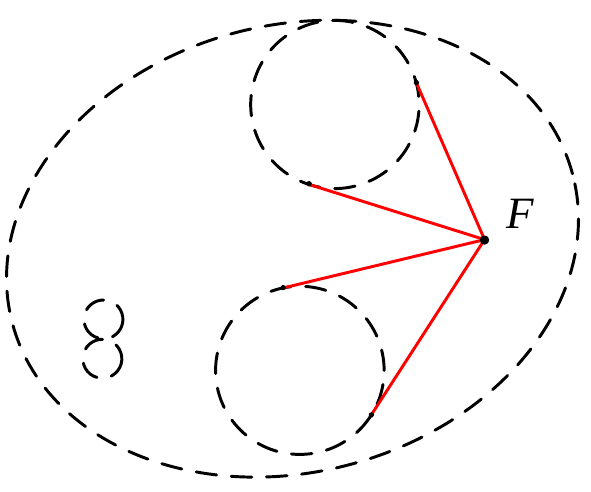}
        \caption{}
        \label{fig3c}
    \end{subfigure}
    \begin{subfigure}{0.35\textwidth}
        \includegraphics[width=\textwidth]{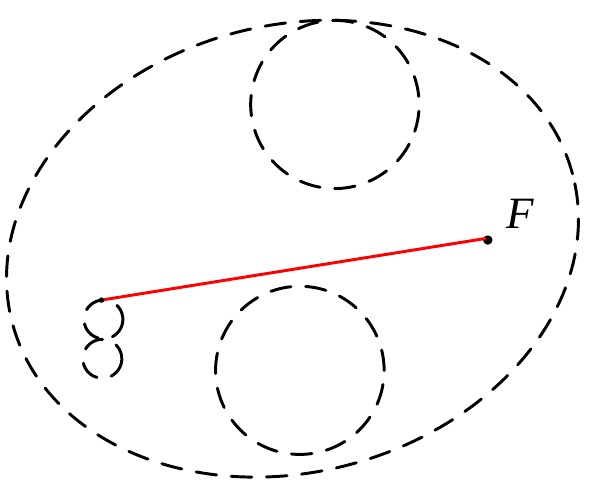}
        \caption{}
        \label{fig3d}
    \end{subfigure}
    \caption{Demonstration for tangent lines. (a) Tangent lines between two circles centered at two agents. (b) Tangent lines between a circle centered at an agent and an initial circle. (c) Tangent lines between a circle centered at a agent and  ${\mathcal F}$. (d) Tangent lines between an initial circle and ${\mathcal F}$.}\label{fig3}
\end{figure}

\begin{definition}\label{define_3_7}
Points of circles centred at $P_i$ and initial circles belonging to tangent lines are called tangent points. 
\end{definition}

We consider only finite segments of tangent lines between tangent points such that their interiors do not cross
the interiors of the initial circles. Moreover, we consider
only finite segments of tangent lines between tangent points that belong to ${\mathcal M}(\theta)$.

\begin{definition}\label{define_3_8}
The vertices of the graph $\mathcal{G}(\theta)$ are the aircraft initial position $p(0)$, the final destination $\mathcal{F}$ and the tangent points. The edges of the graph $\mathcal{G}(\theta)$ are the tangent segments, arcs of the circles centred at $P_i$ with radii $R_i(\theta)$ and the initial circles that connect the vertices of the graph. The graph ${\mathcal G}(\theta)$ is called the extreme graph (see e.g. Fig. \ref{fig4a}).
\end{definition}

\begin{definition}\label{define_3_9}
A path connecting the aircraft initial position $p(0)$ and the final destination ${\mathcal F}$ on the extreme graph is said to be viable if the heading at the end of each edge of the path is equal to the heading at the beginning of the next edge of the path (see Fig. \ref{fig4b}).
\end{definition}

\begin{figure}[t]
    \centering
    \begin{subfigure}{0.45\textwidth}
        \includegraphics[width=\textwidth]{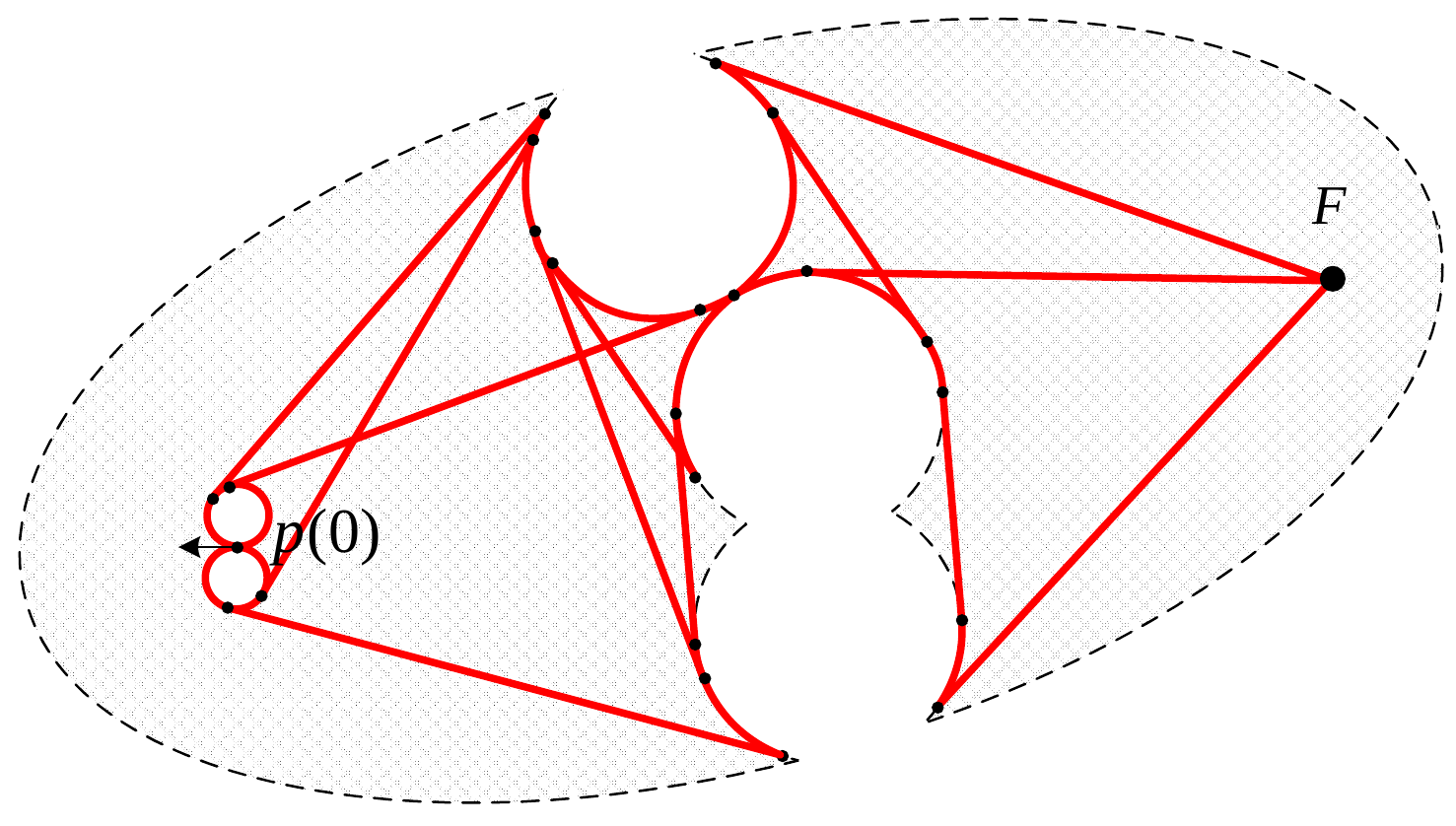}
        \caption{}
        \label{fig4a}
    \end{subfigure}
    \begin{subfigure}{0.45\textwidth}
        \includegraphics[width=\textwidth]{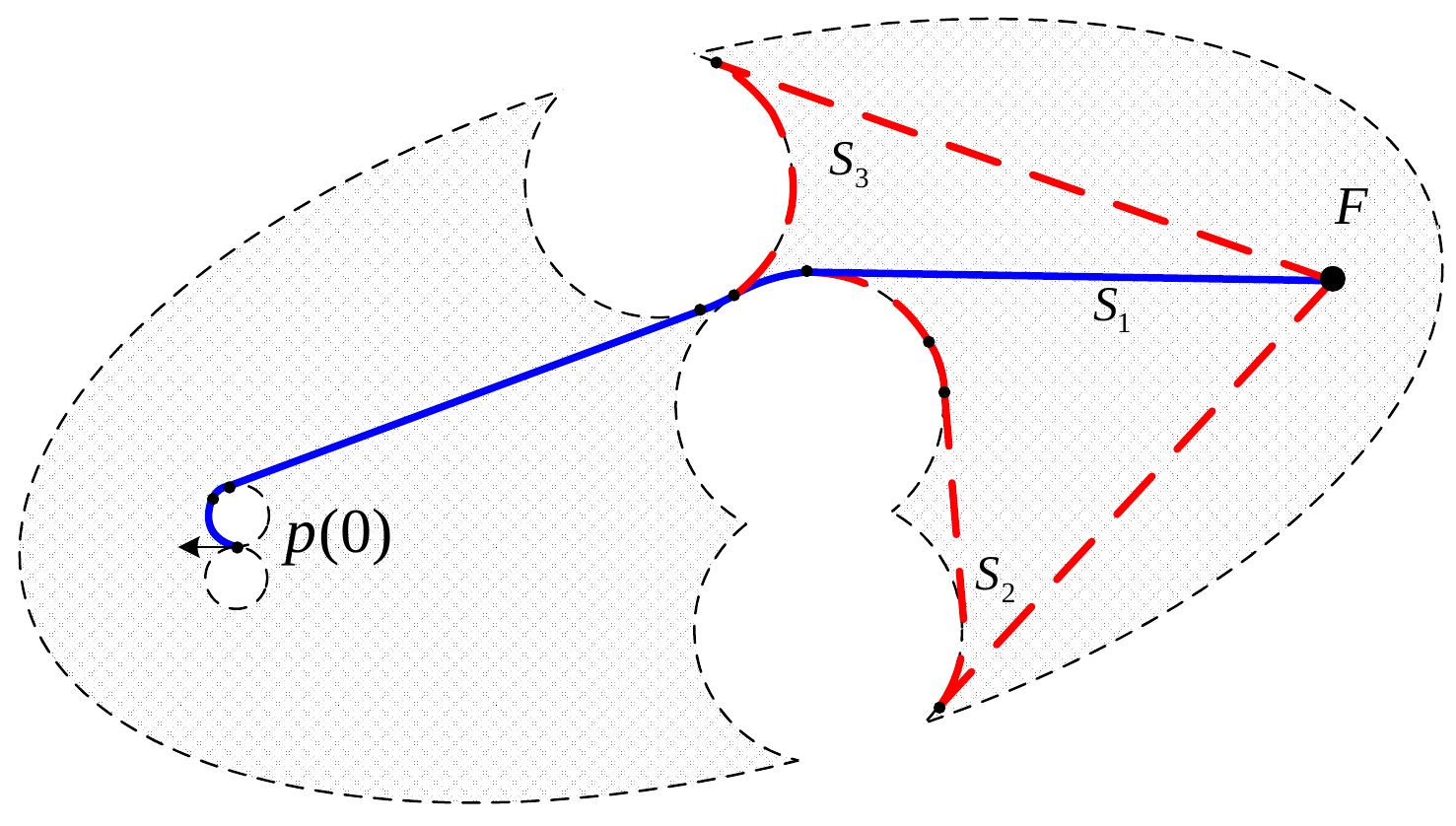}
        \caption{}
        \label{fig4b}
    \end{subfigure}
    \caption{The extreme graph and the viable path. (a) An illustrative example of ${\mathcal G}(\theta)$ consisting of the vertices representing by the black solid points and the edges representing by the red solid lines and arcs. (b) $S_1$ is a viable path satisfying the heading constraint while $S_2$ and $S_3$ are not as they dissatisfy the heading constraint.}\label{fig4}
\end{figure} 

\begin{remark}
If a path connecting the aircraft initial position $p(0)$ and the final destination ${\mathcal F}$ on the extreme graph is viable it can be a path of the aircraft (since $R_{min}< R_i$ for all $i$, it follows from Assumptions \ref{Ast} and \ref{As-1} that any viable path on the extreme graph ${\mathcal G}(\theta)$ satisfies the curvature requirement (\ref{ac})). Also, it is obvious that if $\theta$ is the minimum critical threat level with the connectivity requirement then the corresponding critical point is a tangent point. The proposed procedure for constructing optimal paths reminds in spirit optimal path planning for ground vehicles in complex cluttered environments \cite{savkin2013reactive}.
\end{remark}

Now we are in a position to present the main theoretical result.

\begin{Theorem}
\label{Theorem1} Suppose that Assumptions~\ref{Ast}, \ref{As-1},
\ref{Aso3-} and \ref{Aso3} are satisfied. Then the following statements hold:
\begin{enumerate}
\item
The minimum achievable threat level $\theta$ is equal to the minimum critical threat level  with the connectivity requirement.
\item
If $\theta$ is the minimum critical threat level  with the connectivity requirement, then an optimal aircraft path is a shortest viable path on the extreme graph ${\mathcal G}(\theta)$.
\end{enumerate}
\end{Theorem}

{\bf The proof of Theorem \ref{Theorem1}:} First, we prove that the minimum achievable threat level $\theta$ is equal to the minimum critical threat level  with the connectivity requirement. Indeed, as in \cite{dubins1957}, it  follows from Ascoli's Theorem (see e.g. \cite{KF99}) that the minimum achievable threat level exists. Furthermore, it is obvious that  the minimum achievable threat level is a critical threat level. Furthermore, if some critical threat level $\theta_{*}$ does not satisfy  the connectivity requirement then it is obvious that there does not exist any path in ${\mathcal M}(\theta_{*})$ connecting the aircraft initial position $p(0)$ and the final destination ${\mathcal F}$. Therefore, the minimum achievable threat level $\theta$ cannot be less
than the minimum critical threat level with the connectivity requirement. On the other hand, we prove that if $\theta$ is  the minimum critical threat level, then the set ${\mathcal M}(\theta)$ contains  a shortest reaching the final destination path. Indeed, as in \cite{DUB57}, it follows from Ascoli's Theorem (see e.g. \cite{KF99}). Furthermore, we prove that this shortest reaching the final destination path is a shortest viable path on the extreme graph ${\mathcal G}(\theta)$. Indeed, let $S$ be a shortest (minimum length) reaching the final destination path. We now prove that the path $S$ does not go inside of any of two initial circles. Indeed, consider the circle $C$ of radius $6R_{min}$ centred at $p(0)$. Since $R_i(\theta) \leq R_i$ for all $i$, it follows from Assumption \ref{Aso3} that this circle is inside of the set ${\mathcal M}(\theta)$. Let $p_1$ be the point of the circle at which the path $S$ leaves this circle, and $\alpha_1$ be the trajectory heading at $p_1$. It means that $S$ consists of two segments $(p(0),p_1)$ and $(p_1,{\mathcal F})$ where $(p_1,{\mathcal F})$ is outside of the circle $C$. Now let $S_1$ be a path of minimum length connecting the points $p(0)$ and $p_1$ such that it has headings $\alpha(0)$ and $\alpha_1$ at these points, respectively, and the curvature requirement (\ref{ac}) is satisfied. It follows from the main result of \cite{DUB57} that such a path exists and $S_1$ belongs to the disk of radius $8R_{min}$ centred at $p(0)$. Therefore, Assumption \ref{Aso3} and the fact that $R_i(\theta) \leq R_i$ for all $i$ imply that $S_1$ is inside of the set ${\mathcal M}(\theta)$. Since, $S_1$ is a minimum length path, its length is less or equal to the length of $(p(0),p_1)$. On the other hand, the length of $(p(0),p_1)$ cannot be greater than the length of $S_1$ (if it is not true, then the path consisting of $S_1$ and $(p_1,{\mathcal F})$ would be a shorter   path reaching the final destination than $S$ which contradicts to our assumption that $S$ is a minimum length path. Hence, $(p(0),p_1)$ is a shortest path satisfying the curvature constraint (\ref{ac}). Therefore, it follows from the main result of \cite{DUB57} that $(p(0),p_1)$ consists of segments of two or less  minimum radius circles and a straight line segment and $(p(0),p_1)$ does not cross the interiors of the both  initial circles. Hence, we have proved that that the path $S$ does not go inside of any of two initial circles.

\begin{figure}[t]
\begin{center}
{\includegraphics[width=0.4\textwidth]{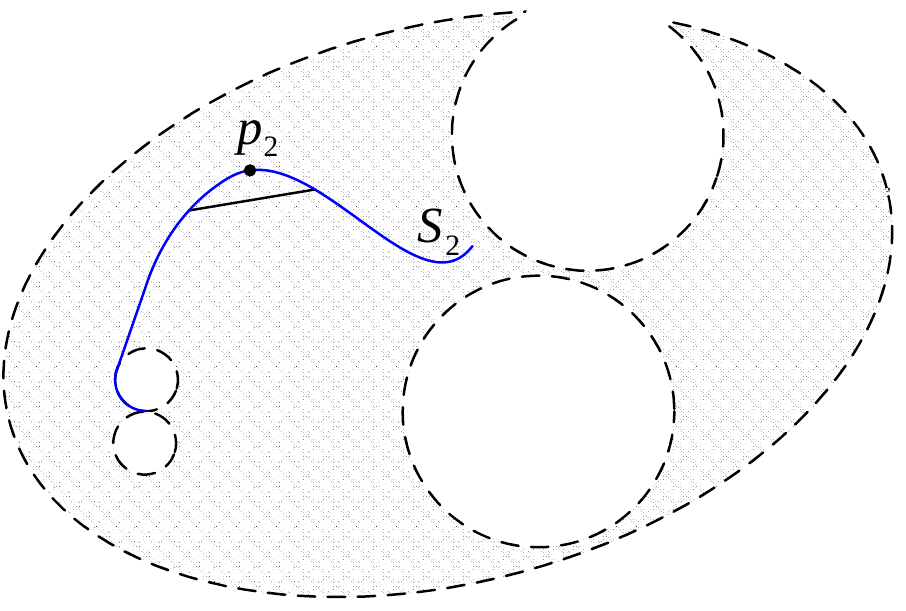}}
\caption{If a segment of a path is not a straight line segment, lies outside of the initial circles and does not cross the boundary of $\mathcal M(\theta)$, then there exists a shorter path.
}
\label{fig5}
\end{center}
\end{figure}

Now let $S_2$ be a minimum length path inside of the set ${\mathcal M}(\theta)$ connecting the points $p(0)$ and ${\mathcal F}$ which is not crossing two initial circles. Now we do not assume that $S_2$ satisfies condition (\ref{ac}). Since we have proved that the path $S$ does not go inside of any of two initial circles, the length of $S_2$ is less or equal to the length of $S$. Now we prove that $S_2$ consists of segments of the boundaries of the set ${\mathcal M}(\theta)$ and two initial circles, and straight line segments connecting two different points on these boundaries. Indeed, we prove it by contradiction. If this statement does not hold, then there exists a point $p_2\in S_2$ which is outside of  two initial circles and the boundary of the set ${\mathcal M}(\theta)$, and $p_2$ is not an interior point of a straight line segment of $S_2$. In this case, in some small neighbourhood of $p_2$ we can replace the segment of path $S_2$ containing $p_2$ by a straight line segment which does not intersect two initial circles and the boundary of the set ${\mathcal M}(\theta)$; see Fig. \ref{fig5}. Therefore, we can construct a path with shorter length than $S_2$ satisfying all the requirements. Furthermore, it follows from Theorem 5 of \cite{LA92} that any straight line segment of $S_2$ connects either two tangent points or a tangent point and the final destination point. Hence, the path $S_2$ may consist of segments of two initial circles, segments of the boundary of the set ${\mathcal M}(\theta)$ and segments of straight lines connecting tangent points or a tangent point and the target. Therefore, the path $S_2$ is a viable path on the extreme graph ${\mathcal G}(\theta)$. Finally, since $R_i(\theta) \leq R_i$ for all $i$, it follows from Assumptions \ref{Ast} and \ref{As-1} that any viable path on the extreme graph ${\mathcal G}(\theta)$ satisfies the curvature requirement (\ref{ac}). This completes the proof of Theorem \ref{Theorem1}.

\subsection{Implementation and Complexity Analysis}\label{Implementation}

In this section, we describe the implementation of the proposed navigation and present the complexity analysis.

\subsubsection{Implementation}\label{imp}
The input of our algorithm includes $p(0)$, $\mathcal{F}$, $P_i$, $R_i$, $f_i$ and $\partial D(\mathcal{M})$. Note $i=1,...,n$. The output is the minimum threat level and the shortest path. The overall procedure is outlined below:
\begin{enumerate}
\item Compute all the critical points and the corresponding threat levels; the sets of them are represented by $CP$ and $\Theta$ respectively;
\item Sort the critical points according to the ascending of threat level (without introducing more notations, $CP$ and $\Theta$ denote the sorted sets henceforce); let $k=1$;
\item Using $CP(k)$ to construct the corresponding non-convex set $\mathcal{M}(\Theta(k))$ and then the extreme graph $\mathcal{G}(\Theta(k))$;
\item Check whether $\mathcal{G}(\Theta(k))$ satisfies the connectivity requirement; if yes, output $\Theta(k)$ and the shortest path, and exit; otherwise, $k=k+1$ and repeat Step 3).
\end{enumerate}

According to Definition \ref{define_3_1} and \ref{define_3_2}, Step 1) requires the following information as input: $P_i$, $f_i$, $R_i$ and $\partial D(\mathcal{M})$. To calculate the critical points of type 1, we need to check whether an agent's threat range crosses the region boundary according to Notation \ref{notation_3_1}. Following Definition \ref{define_3_1}, the critical points of type 1 are the points on the region boundary which are the closest points to the related agents, see e.g. Fig. \ref{fig2}. To calculate the critical points of type 2, we need to check whether two agents' threat range overlap. 

We specify the threat level function $f_i(d)$ and threat radius $R_i(\theta)$ for all $i=1,...,n$ as follows:
\begin{equation}\label{eq3}
f_i(d)=
\begin{cases}
1-\frac{d}{R_i},\ \ 0\leq d < R_i \\
0,\ \ \ \ \ \ \ \ \ d \geq R_i \\
\end{cases}
\end{equation}
\begin{equation}\label{eq4}
R_i(\theta)=R_i\cdot(1-\theta)
\end{equation}
It is obvious the proposed piecewise linear threat level functions satisfy all the requirements of  Assumption \ref{Ast}. 

Considering Remark \ref{remark_3_1}, i.e., the threat functions are identical, the critical points of type 2 are the middle points of the connection of the related two agents, see e.g. Fig. \ref{fig2}. In the situation of non-identical threat functions, the critical point of type 2 is not the middle point. Details of computing such critical point are presented in Appendix \ref{AppendixB}. The algorithm to compute critical points and threat levels is given in Algorithm \ref{cp}. Here we assume $n>1$. If $n=1$, the critical points of type 2 do not exist. In this case, there are at most two critical points of type 1. Then, Line 3 of Algorithm \ref{cp} should be slightly modified by finding two $P_*$ on different sides, see e.g., Fig. \ref{fig9a}.

\begin{algorithm}[t]
\caption{Compute critical points and threat levels}\label{cp}
\begin{algorithmic}[1]
\Require $P_i$, $R_i$, $\partial D(\mathcal{M})$, $f_i$ 
\Ensure $CP$, $\Theta$
\State $CP\leftarrow \emptyset$, $\Theta \leftarrow \emptyset$
\For{$i=1:n$}
\State Find $P_*$ on $\partial D(\mathcal{M})$ such that $d(P_i,P_*)$ is minimized 
\If {$d(P_i,P_*)<R_i$}
\State Add $P_*$ to $CP$; add $f_i(d(P_i,P_*))$ to $\Theta$
\EndIf
\For{$j=i+1:n$}
\If {$d(P_i,P_j)<R_i+R_j$}
\State Compute $P_*$ by either $\frac{P_i+P_j}{2}$ or Eq. (\ref{cp_coordinates})
\State Add $P_*$ to $CP$ and the corresponding threat level to $\Theta$
\EndIf
\EndFor
\EndFor
\end{algorithmic}
\end{algorithm}

Step 2) sorts the critical points according to the threat levels. The idea behind is that if a smaller threat level leads to an extreme graph satisfying the connectivity requirement, there is no need to check the larger levels.

In Step 3), we build up the non-convex set $\mathcal{M}(\Theta(k))$ and then construct the extreme graph $\mathcal{G}(\Theta(k))$ based on the selected critical point $CP(k)$. To construct $\mathcal{G}(\Theta(k))$, we need to construct all the tangent lines according to Definition \ref{define_3_6}. There are two types of tangent lines between: 1) a circle and a point outside; and 2) two circles, both of which can be precisely computed. The formulations of common tangent lines for these two cases are given in Appendix \ref{AppendixC}. 

In Step 4), we need to verify whether $\mathcal{G}(\Theta(k))$ meets the requirement of connectivity by Breadth First Search (BFS) \cite{BFS}. If connected, the algorithm BFS can also figure out the shortest path from $p(0)$ to $\mathcal{F}$. We finally check whether the heading requirement is satisfied by all two successful edges. If we can obtain a viable path under the current threat level, then it is the minimum threat level and the path is the optimal one; otherwise, we need to test the next critical point which leads to a larger threat level, i.e., repeat Step 3).

\subsubsection{Complexity analysis}\label{complexity_analysis}
Now we analyse the complexity of the proposed algorithm.

Shown as Algorithm \ref{cp}, the time complexity of Step 1) is $O(n^2)$. For the sorting procedure, many algorithms can be finished with time complexity lower than $O(n^2)$, e.g., Heapsort, which is in $O(n\log(n))$. Both of Step 1) and 2) will be executed by only once.

We point out that our minimum threat level and the shortest path will be found within a finite number of rounds. This number equals to the size of $CP$. We first consider an extreme case to estimate the upper bound of $|CP|$ and then analysis the complexity of applying BFS in Step 4). 

Since any pair of agents have at most one critical point of type 2, the maximum number of critical points of type 2 is $\frac{n(n-1)}{2}$. Since any agent has at most one critical point of type 1 with the boundary, the maximum number of critical points of type 1 is $n$. So $|CP|$ is upper bounded by $\frac{n(n+1)}{2}$. 

As discussed in Section \ref{imp}, the tangent lines can be calculated analytically (the time complexity is $O(1)$). Any pair of circles can have up to four tangent lines, and there are at most two tangent lines between $\mathcal{F}$ and a circle. There are $n+2$ circles in the graph. In the extreme case, i.e., any agent circle has four common tangent lines with all the other circles, the number of tangent lines is $4C_{n+2}^2-4$. Note the tangent lines between two initial circles should be taken out. Between the $n$ threat circles and $\mathcal{F}$, there are $2n$ tangent lines. Thus, the total number of tangent lines is $2n^2+6n+2n$. The total number of tangent points in the graph is $2(2n^2+6n)+2n$. Note, the 2 here corresponds to the fact that there are two tangent points for the tangent line between the circles; while for the tangent line between the threat circles and $\mathcal{F}$, only the tangent points on the circles should be considered. Besides, using BFS to check the connectivity of $\mathcal{G}(\Theta(k))$ results in the time complexity of $O(n^2)$. 

Overall, in the extreme case, the time complexity of our algorithm is $O(n^2+nlog(n)+\frac{n(n+1)}{2}\times n^2)=O(n^4)$. 

However, we also point that the extreme case does not exist in real applications. In our considered scenario, the agents should be well placed. In other words, not any two agents can have common tangent lines. Besides, not all agents' threat areas overlap with the region boundary. Thus, the upper bound of $|CP|$ will be decreased significantly. Here, we consider at most $K$ (a constant) threat circles overlap with each other simultaneously. In this case, number of critical points of type 2 is at most $Kn/2$ and that of critical points of type 1 is no larger than $n$. Thus, the upper bound of $|CP|$ is reduced from $O(n^2)$ to $O(n)$. Further, the numbers of vertices and edges in $\mathcal{G}(\Theta(k))$ is decreased accordingly, both from $O(n^2)$ to $O(n)$.

Therefore, in general, the complexity of our algorithm is $O(n^2)$ instead of $O(n^4)$. In the next section, we will use some simulation to demonstrate the complexity difference between the extreme and general cases.

\subsection{Simulation Results}
\label{S4}

This section presents examples to demonstrate the efficiency of the proposed method. We first set up the simulation environment. Then, we present the simulations for identical threat radius and non-identical threat radii. Finally, we compare the proposed approach with a fuzzy logic algorithm.

We consider several scenarios and build optimal paths of the aircraft. The aircraft initial heading to the X-axis is $\alpha(0)=\pi$. The speed of the aircraft is set as $v=1.5m/s$ and its maximum angular velocity is $u_M=0.5rad/s$. $R_i$ is taken as $12m$ for Simulation 1, 2, 3 and 4, and $30m$ for Simulation 5. There is only one agent in Simulation 1 and 2, i.e., $n=1$, while $n=2,\ 6,\ 10$ for Simulation 3, 4 and 5 respectively. 
\subsubsection{Identical threat radius}
Fig. \ref{fig8} demonstrates a scenario where $\mathcal{M}$ is wide enough and the aircraft is able to reach the final destination with $0$ threat level. However in Fig. \ref{fig9}, the aircraft cannot arrive at the final destination with 0 threat level. In this case, the optimal path crosses a critical point of type 1. Another simulation where the optimal path crosses a critical point of type 2 is shown in Fig. \ref{fig10}.  A complex simulation is displayed in Fig. \ref{fig11}. 
 Six agents block the region and the aircraft arrives at its final destination with 0.08 threat level. A more challenging simulation is shown in Fig. \ref{fig12} where the maximum threat level is 0.10. In these cases, the optimal paths contain critical points of type 2.

\begin{figure}[t]
\begin{center}
{\includegraphics[width=0.45\textwidth]{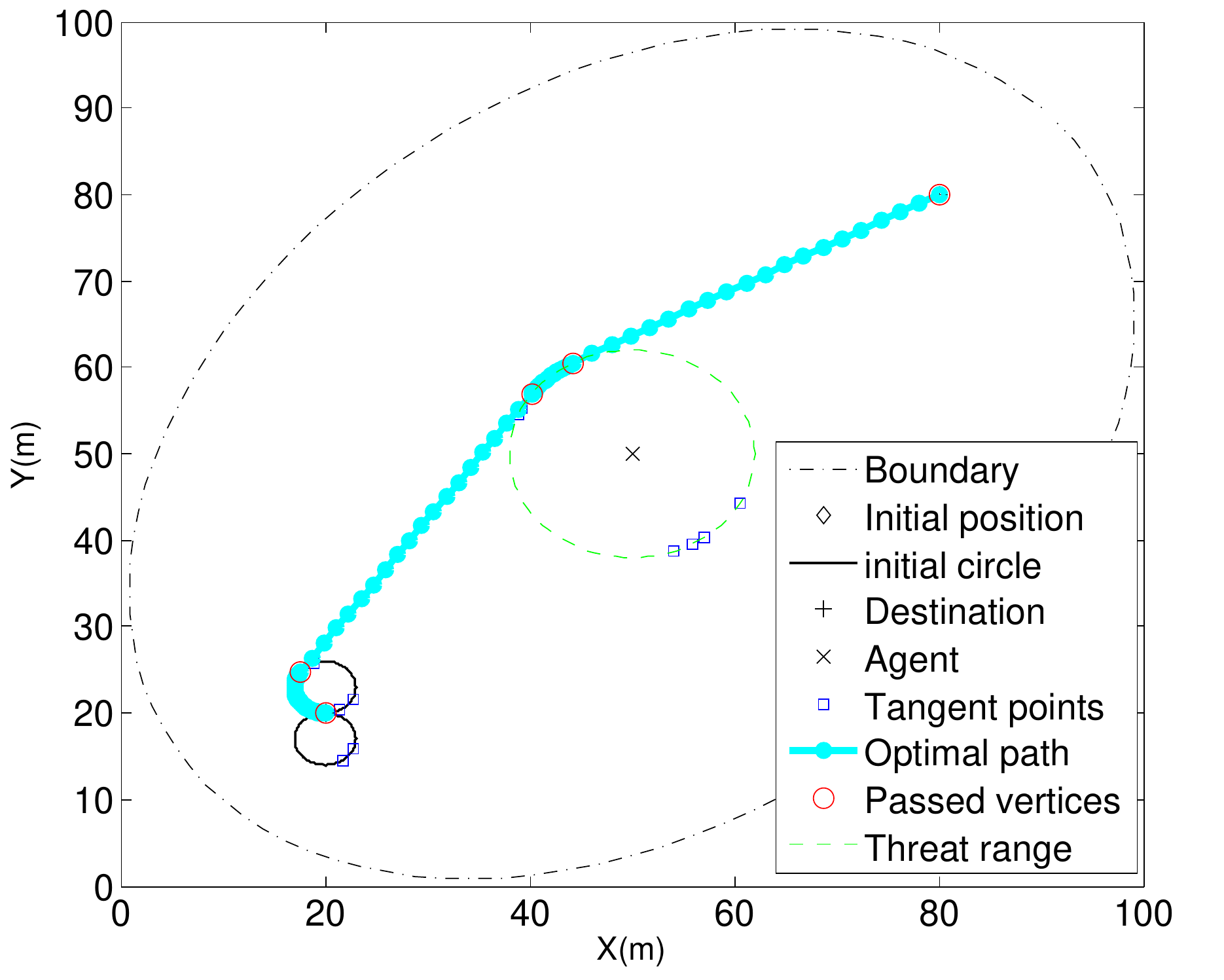}}
\caption{Simulation 1. $\theta=0.00$.}
\label{fig8}
\end{center}
\end{figure}

\begin{figure}[t]
    \centering
    \begin{subfigure}{0.45\textwidth}
        \includegraphics[width=\textwidth]{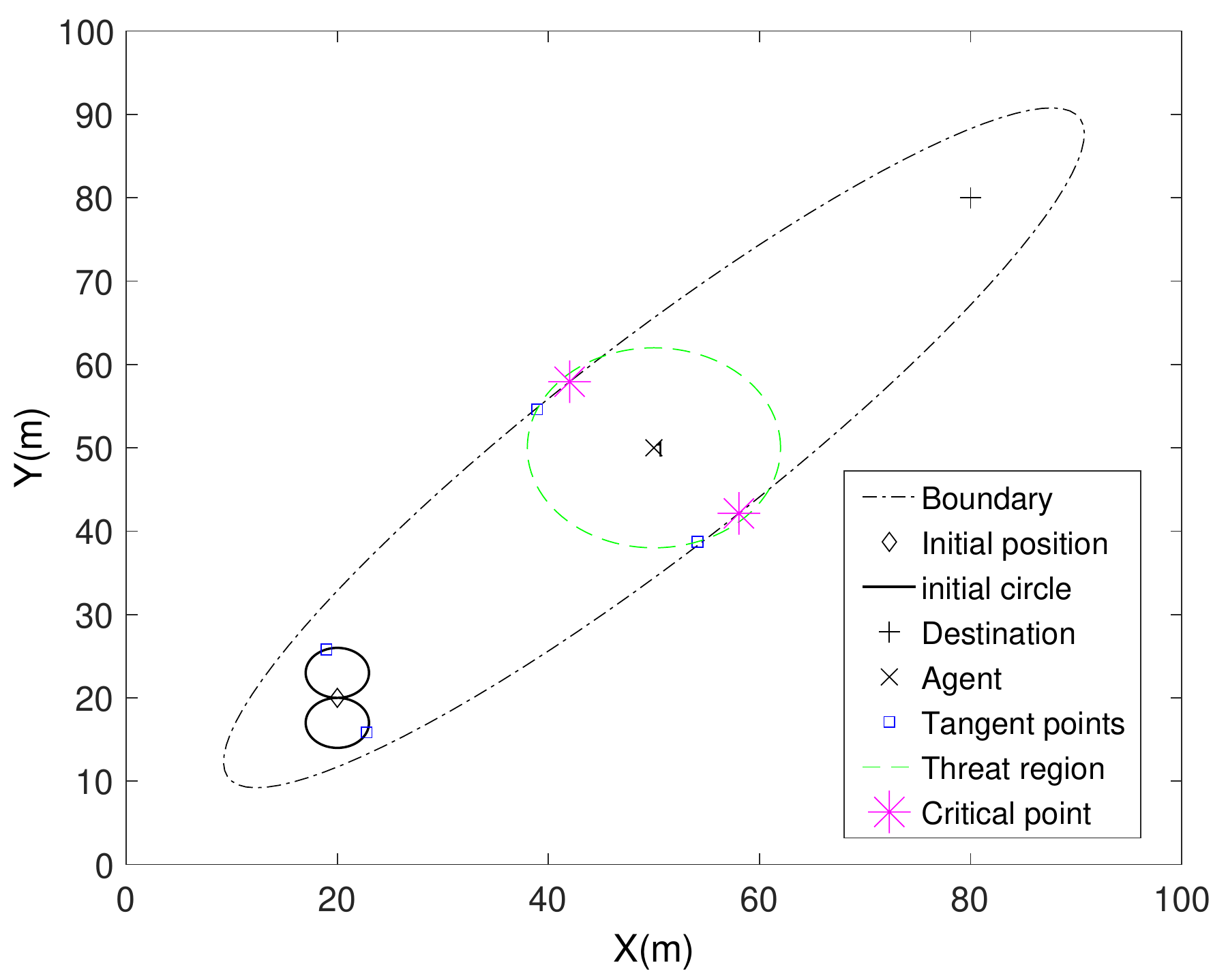}
        \caption{}
        \label{fig9a}
    \end{subfigure}
    \begin{subfigure}{0.45\textwidth}
        \includegraphics[width=\textwidth]{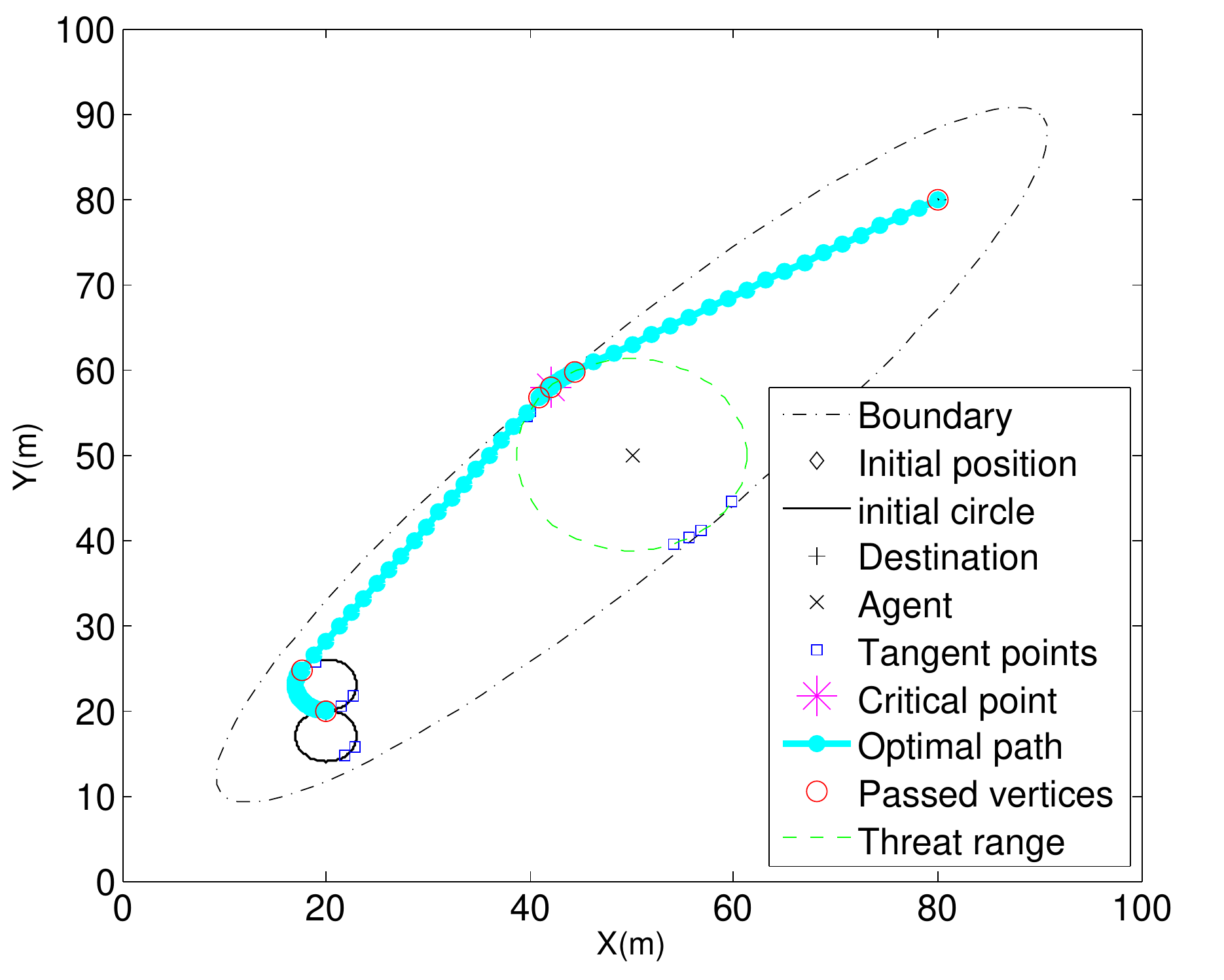}
        \caption{}
        \label{fig9b}
    \end{subfigure}
    \caption{Simulation 2 for crossing the critical point of type 1. (a) $\theta=0.00$. (b) $\theta=0.06$.}\label{fig9}
\end{figure}  

\begin{figure}[t]
    \centering
    \begin{subfigure}{0.45\textwidth}
        \includegraphics[width=\textwidth]{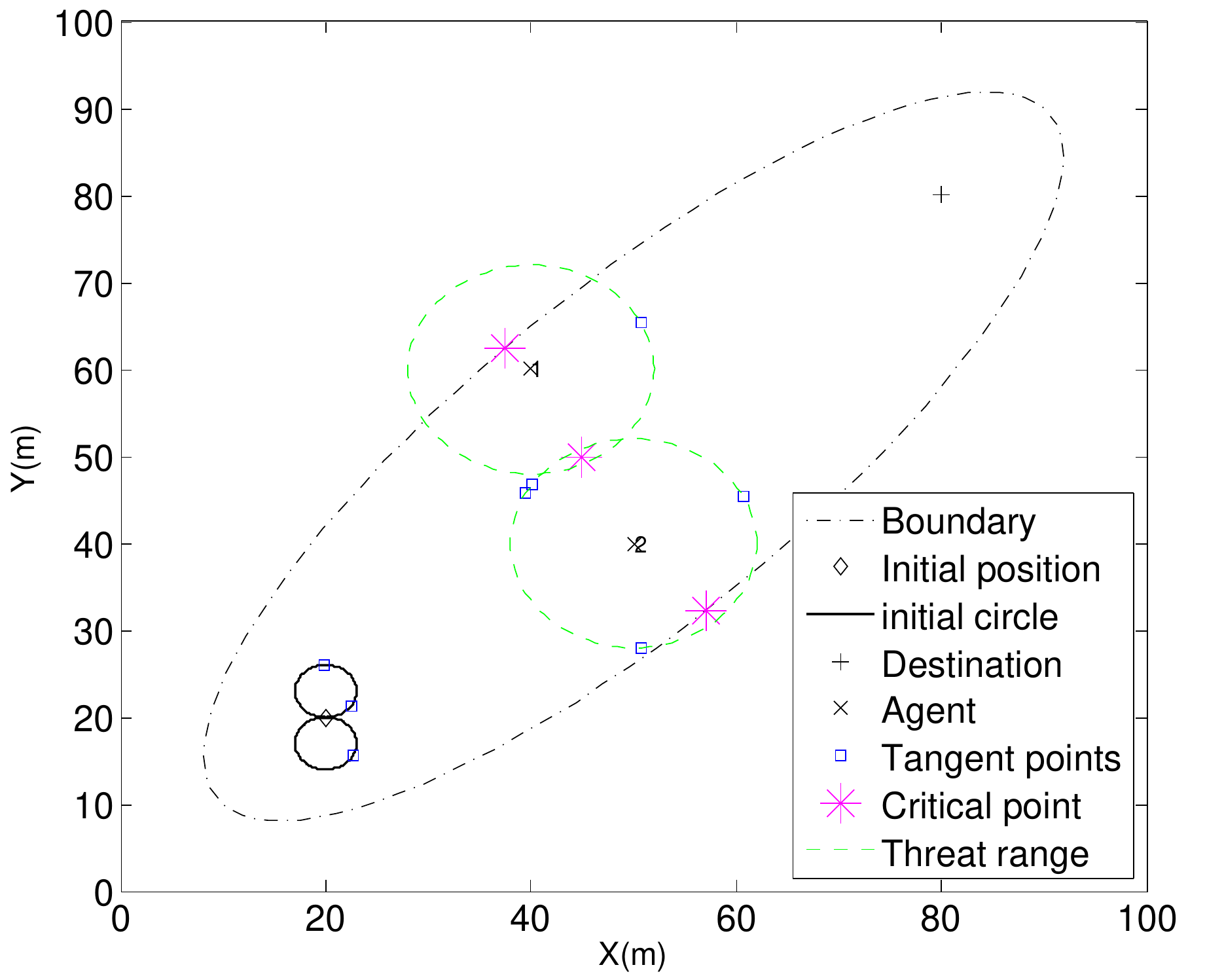}
        \caption{}
        \label{fig10a}
    \end{subfigure}
    \begin{subfigure}{0.45\textwidth}
        \includegraphics[width=\textwidth]{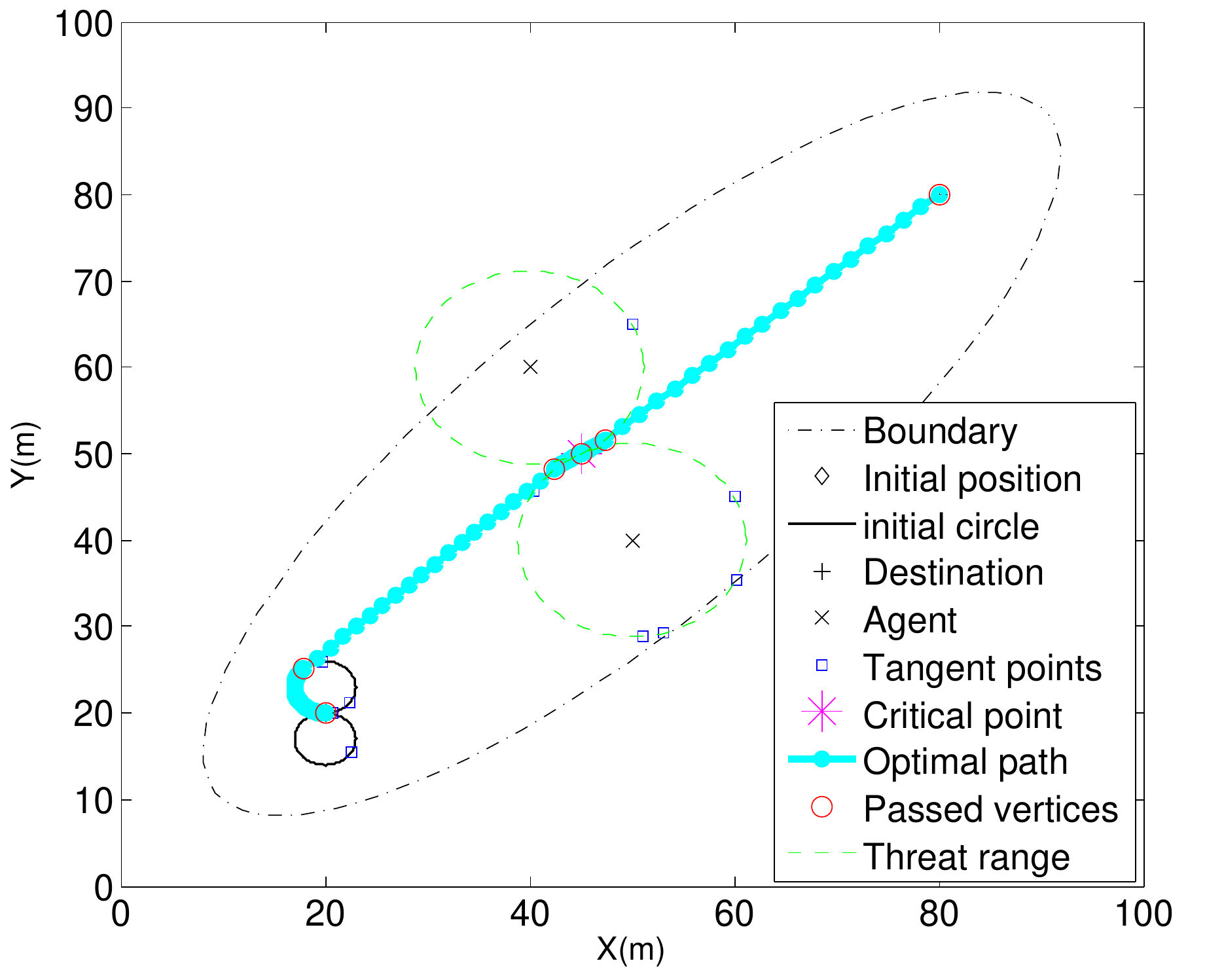}
        \caption{}
        \label{fig10b}
    \end{subfigure}
    \caption{Simulation 3 for crossing the critical point of type 2. (a) $\theta=0.00$. (b) $\theta=0.07$.}\label{fig10}
\end{figure}  

\begin{figure}[t]
    \centering
    \begin{subfigure}{0.45\textwidth}
        \includegraphics[width=\textwidth]{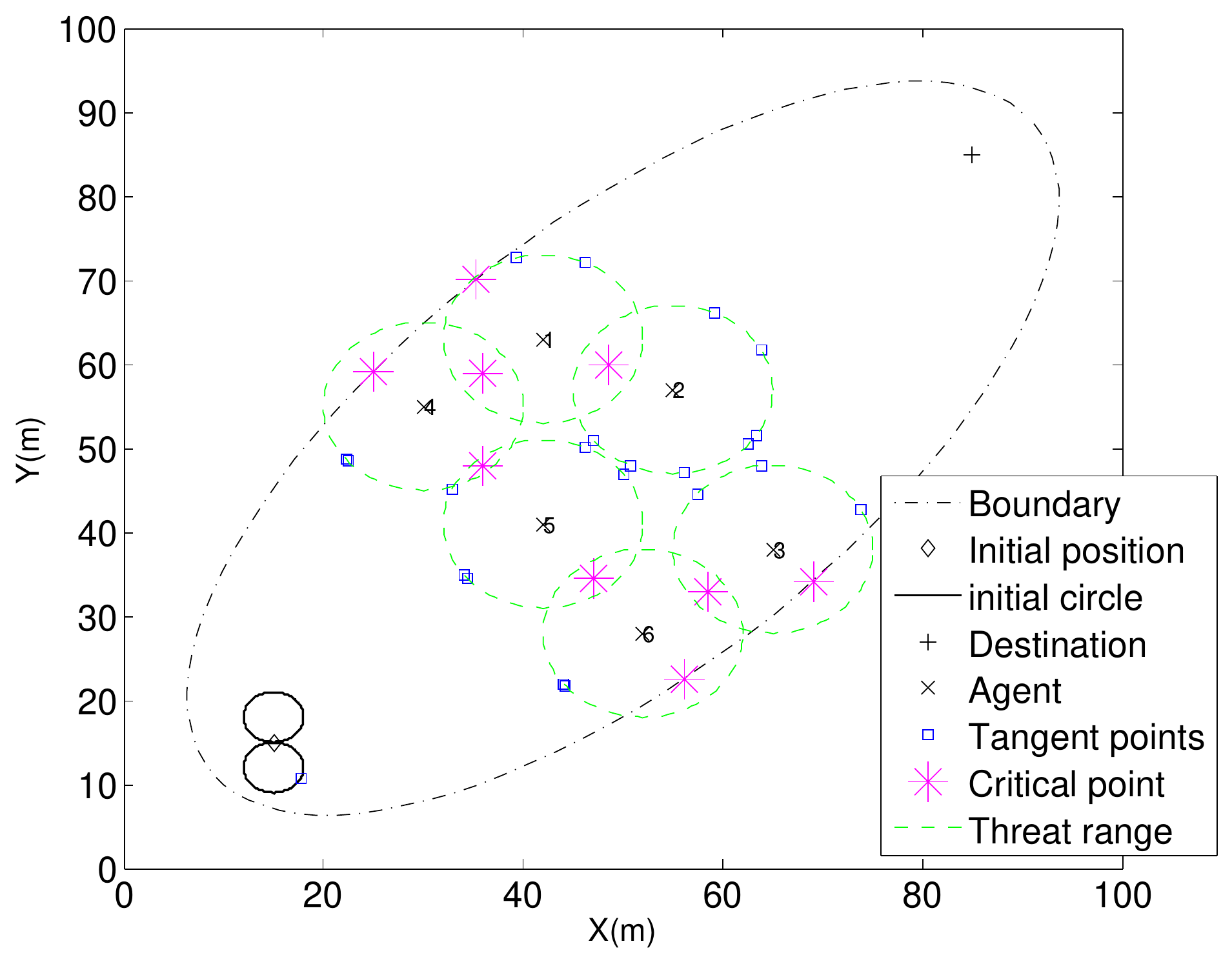}
        \caption{}
        \label{fig11a}
    \end{subfigure}
    \begin{subfigure}{0.45\textwidth}
        \includegraphics[width=\textwidth]{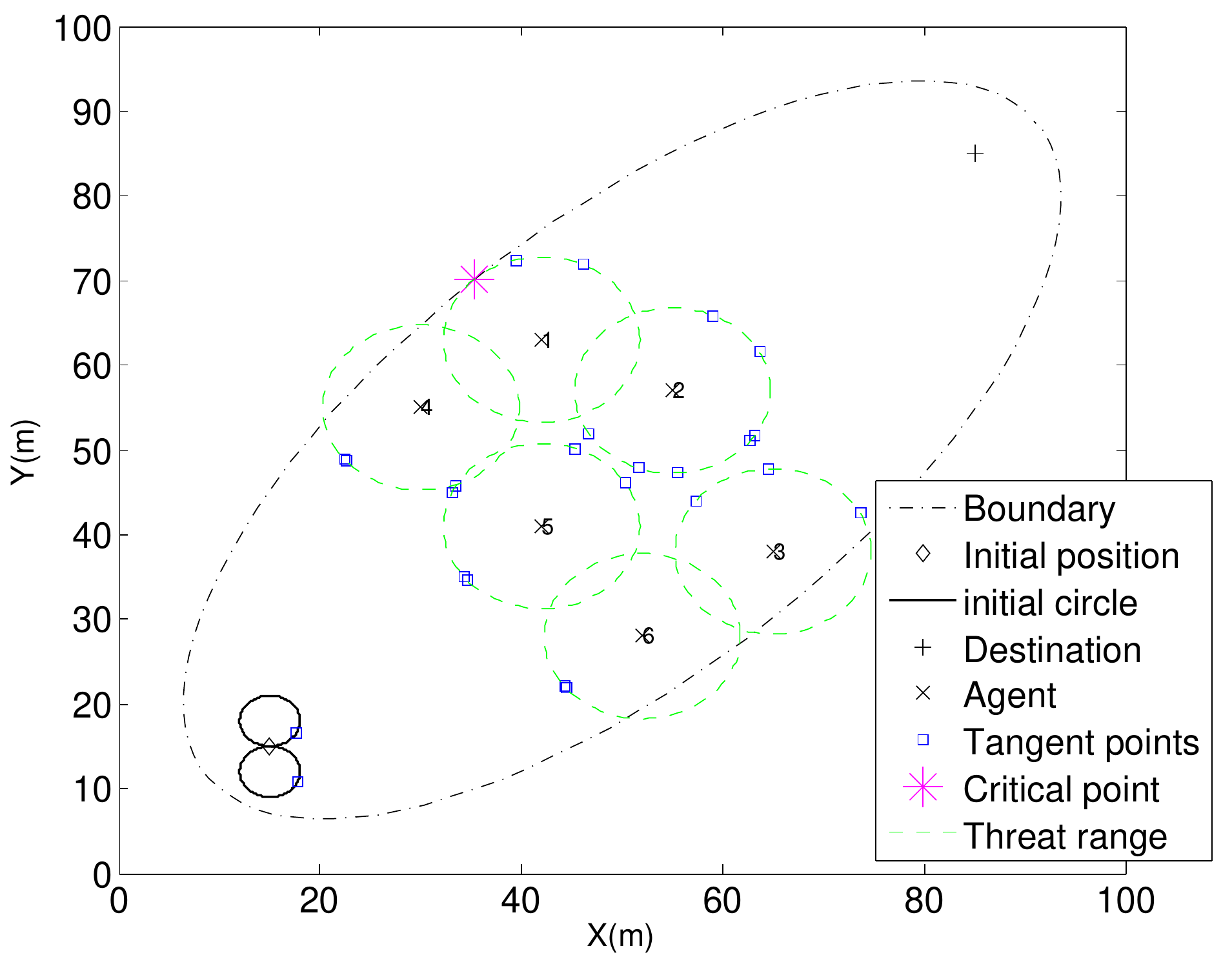}
        \caption{}
        \label{fig11b}
    \end{subfigure}
    \begin{subfigure}{0.45\textwidth}
        \includegraphics[width=\textwidth]{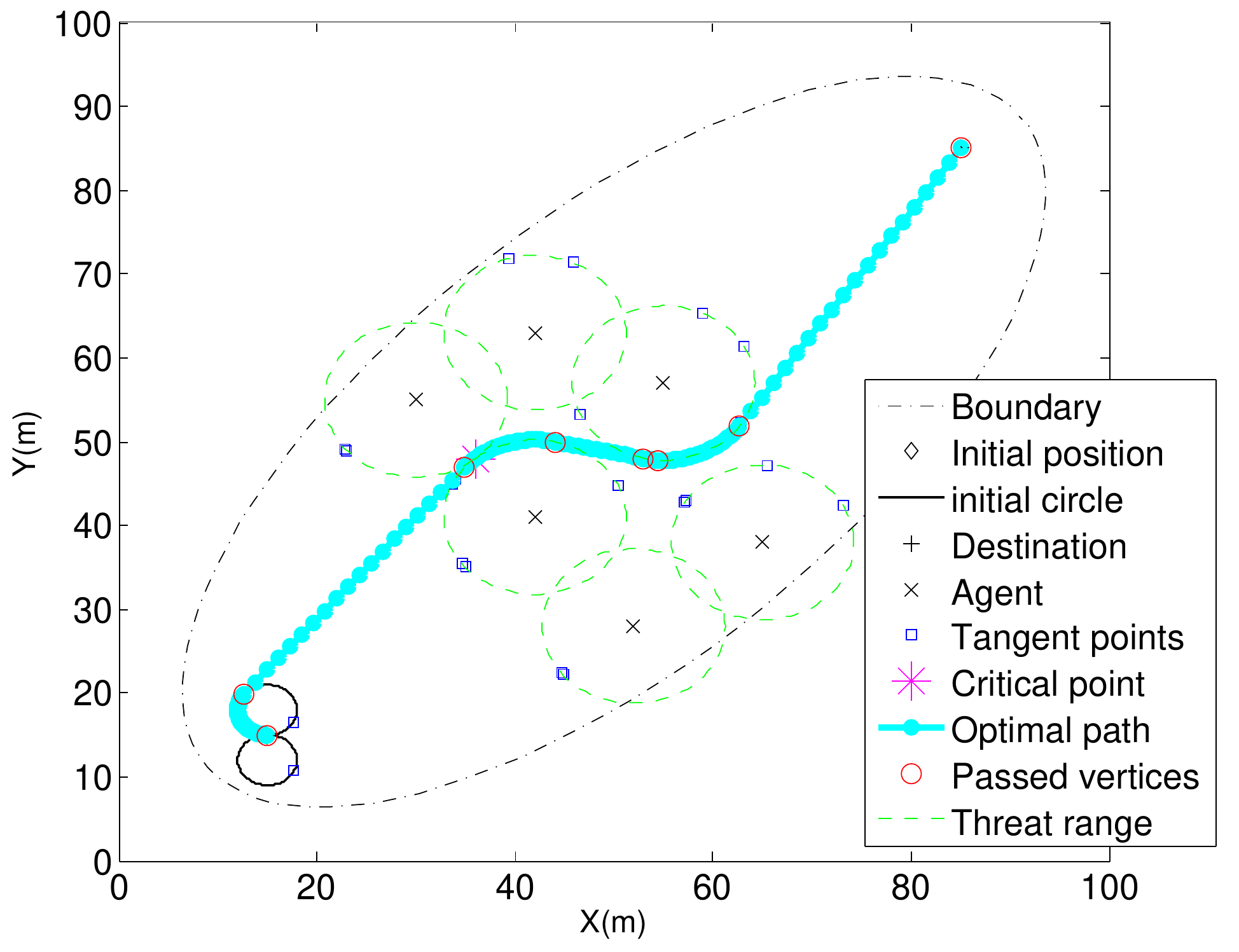}
        \caption{}
        \label{fig11c}
    \end{subfigure}
    \caption{Simulation 4 with 6 agents. (a) $\theta=0.00$. (b) $\theta=0.02$. (c) $\theta=0.08$.}\label{fig11}
\end{figure}

\begin{figure}[h]
    \centering
    \begin{subfigure}{0.45\textwidth}
        \includegraphics[width=\textwidth]{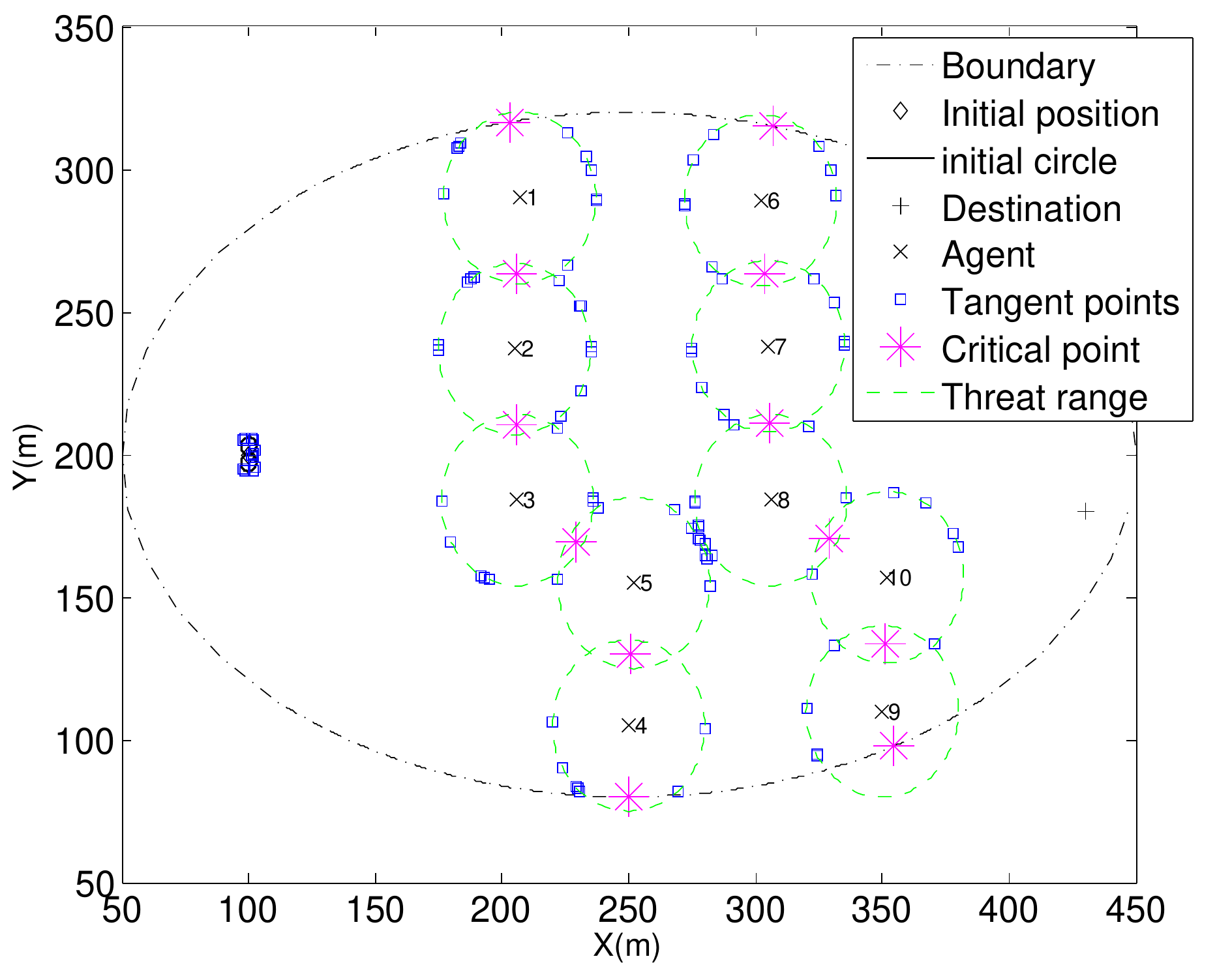}
        \caption{}
        \label{fig12a}
    \end{subfigure}
    \begin{subfigure}{0.45\textwidth}
        \includegraphics[width=\textwidth]{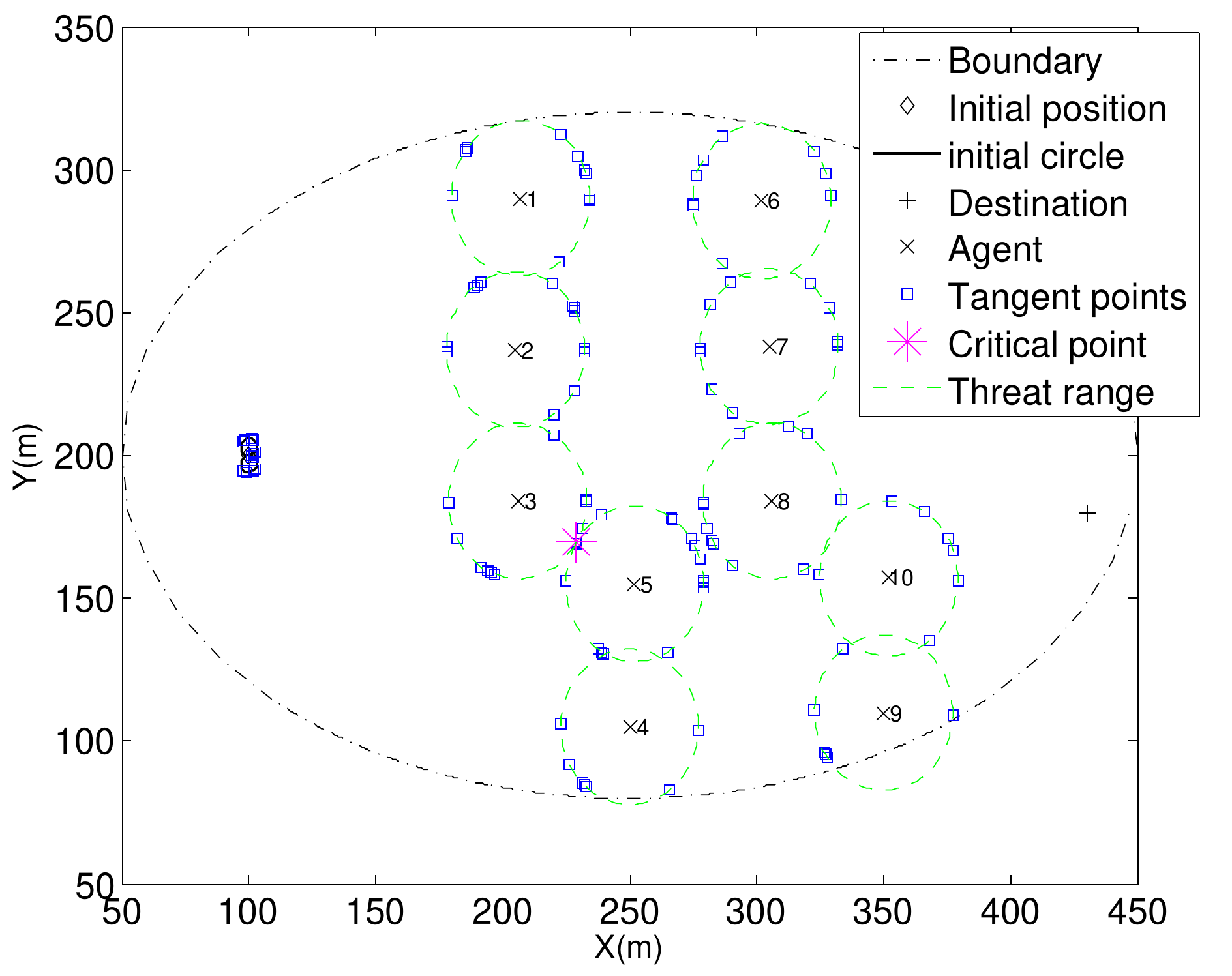}
        \caption{}
        \label{fig12b}
    \end{subfigure}
    \begin{subfigure}{0.45\textwidth}
        \includegraphics[width=\textwidth]{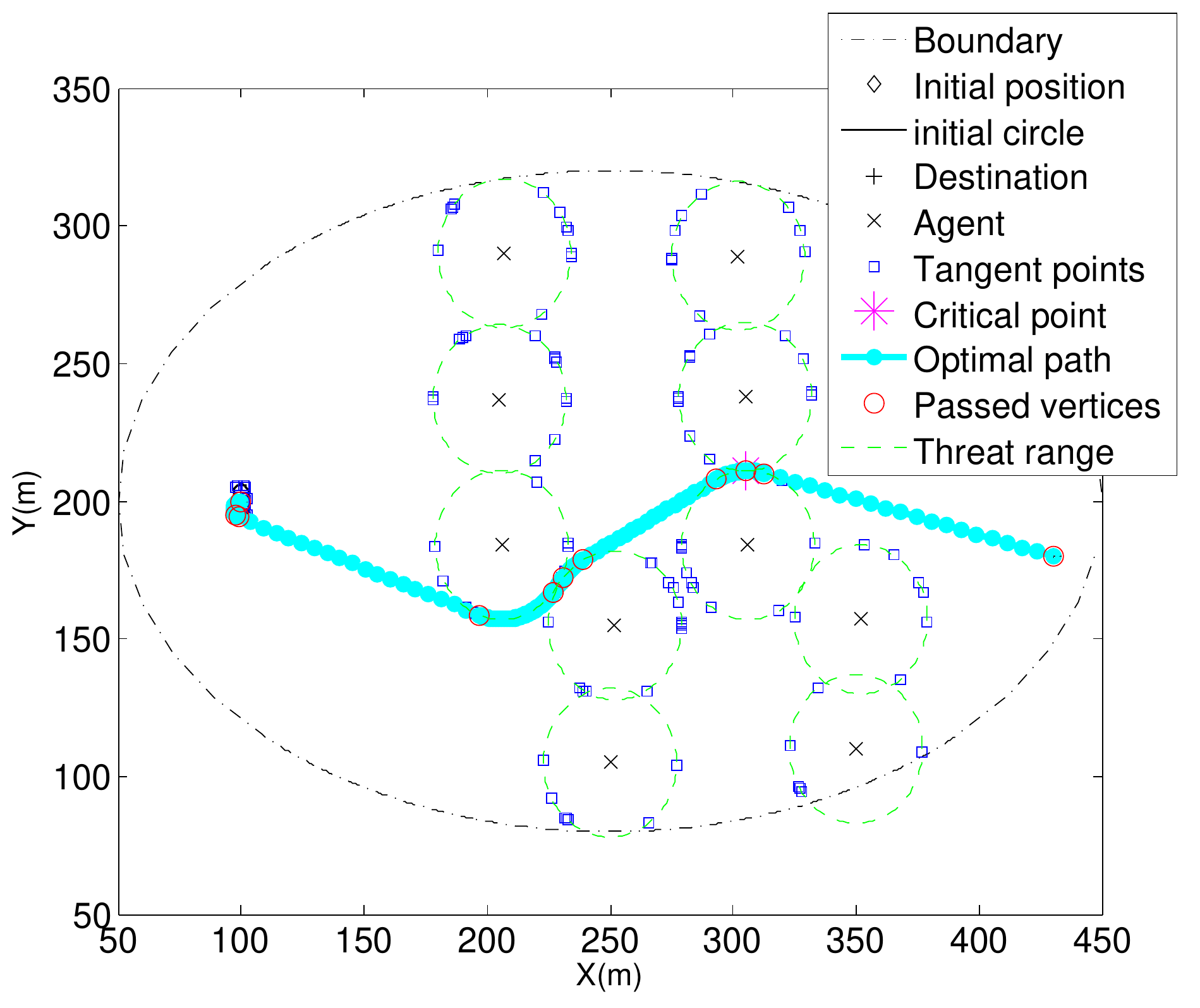}
        \caption{}
        \label{fig12c}
    \end{subfigure}
    \caption{Simulation 5 with 10 agents. (a) $\theta=0.00$. (b) $\theta=0.09$. (c) $\theta=0.10$.}\label{fig12}
\end{figure}

Now we test the complexity analysis in Section \ref{complexity_analysis} on Simulation 3, 4 and 5, where $n$ = 2, 6 and 10. We summarize the actual numbers of critical points, vertices and edges in the extreme graph in Table \ref{table:number}. Furthermore, we provide the estimated numbers of critical points, vertices and edges for both the general case and the extreme case. For the general case, considering the actual environment, $K$ takes 2, 2 and 3 respectively for Simulation 3, 4 and 5. As seen in Table \ref{table:number}, the estimations in the general case is more closer to reality than the extreme case.

\begin{table}[t]
\begin{center}
\caption{The number of critical points, vertices and edges}\label{table:number} 
  \begin{tabular}{| c | c | c | c | c|}
    \hline
    \multicolumn{2}{|c|}{} & $CP$ & Vertex & Edge\\ \hline
    \multirow{4}{*}{$n=2$} & Round 1& \multirow{2}{*}{3}& 10& 9\\
    \cline{4-5} 
    & Round 2& & 13&17\\ 
    \cline{2-5}
    & General & 3& 38&55\\
    \cline{2-5}
    & Extreme&3& 54&78\\ \hline
    \multirow{5}{*}{$n=6$} & Round 1& \multirow{3}{*}{9}& 24&26 \\ 
    \cline{4-5} 
    & Round 2& & 26&30\\ 
    \cline{4-5} 
    & Round 3& & 29&37\\
    \cline{2-5}
    & General& 12& 78& 115\\
    \cline{2-5}
    &Extreme& 21& 246&366\\ \hline
    \multirow{5}{*}{$n=10$} & Round 1& \multirow{3}{*}{12}& 96& 121\\ 
    \cline{4-5} 
    & Round 2& & 118&154\\ 
    \cline{4-5} 
    & Round 3& & 125&171\\
    \cline{2-5}
    & General& 25&166&247\\
    \cline{2-5}
    & Extreme& 55&566&846\\ \hline
    \end{tabular}
\end{center}
\end{table}

\subsubsection{Non-identical threat radii}\label{S52}
We further execute simulations for non-identical threat radii. In Simulation 4, the threat radii of the six agents are 12, 13, 12, 12, 11 and 10 meters respectively. In Simulation 5, the threat radii are 30, 28, 26, 30, 32, 31, 29, 27, 30 and 28 meters respectively. The results are demonstrated in Fig. \ref{fig13} and \ref{fig14}. 
Since the threat radii have been changed, correspondingly the achieved threat levels of passing the region are also different: 0.03 and 0.05  respectively in Simulation 4 and 5. 

\begin{figure}[h]
    \centering
    \begin{subfigure}{0.45\textwidth}
        \includegraphics[width=\textwidth]{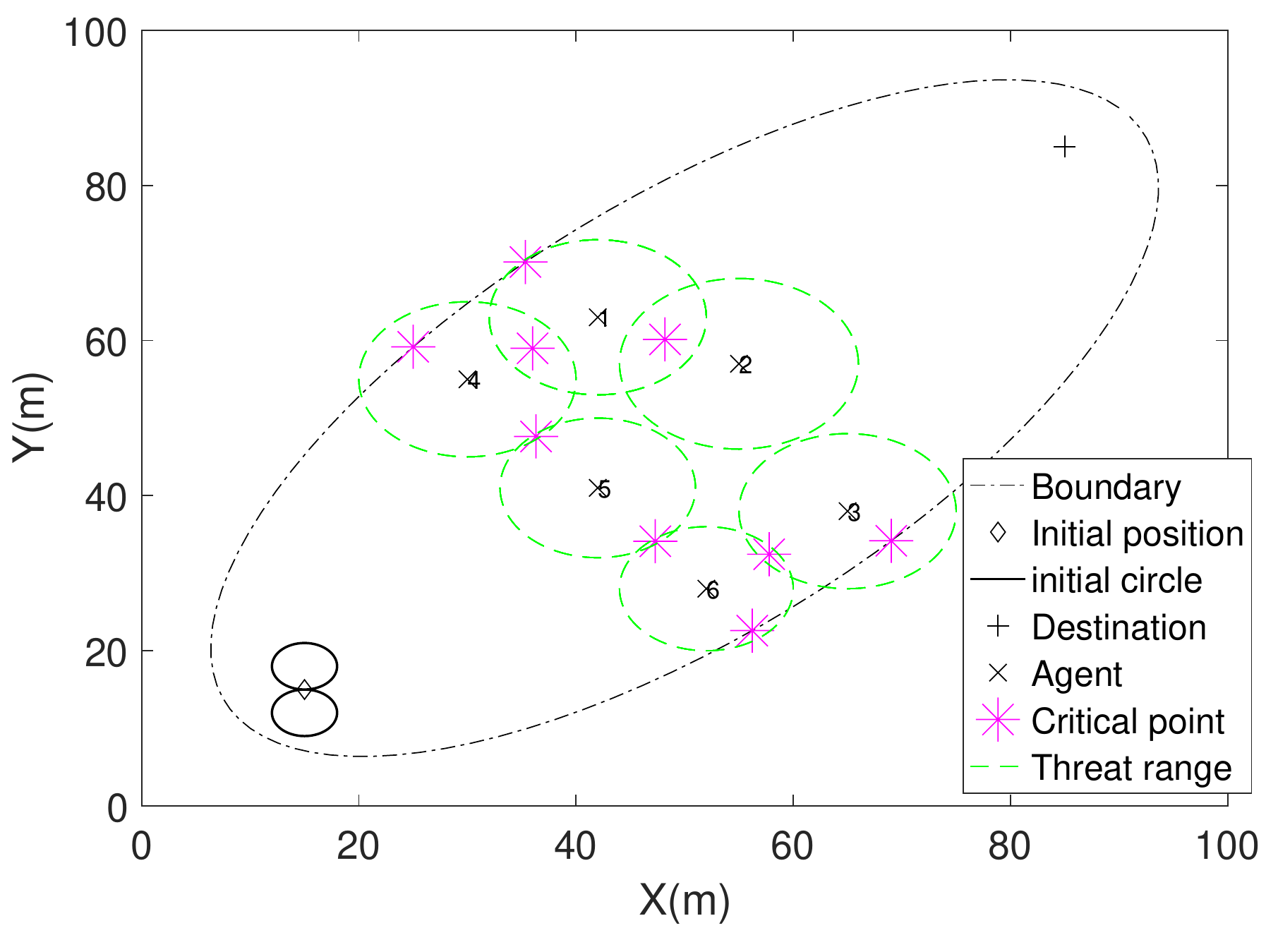}
        \caption{}
        \label{fig13a}
    \end{subfigure}
    \begin{subfigure}{0.45\textwidth}
        \includegraphics[width=\textwidth]{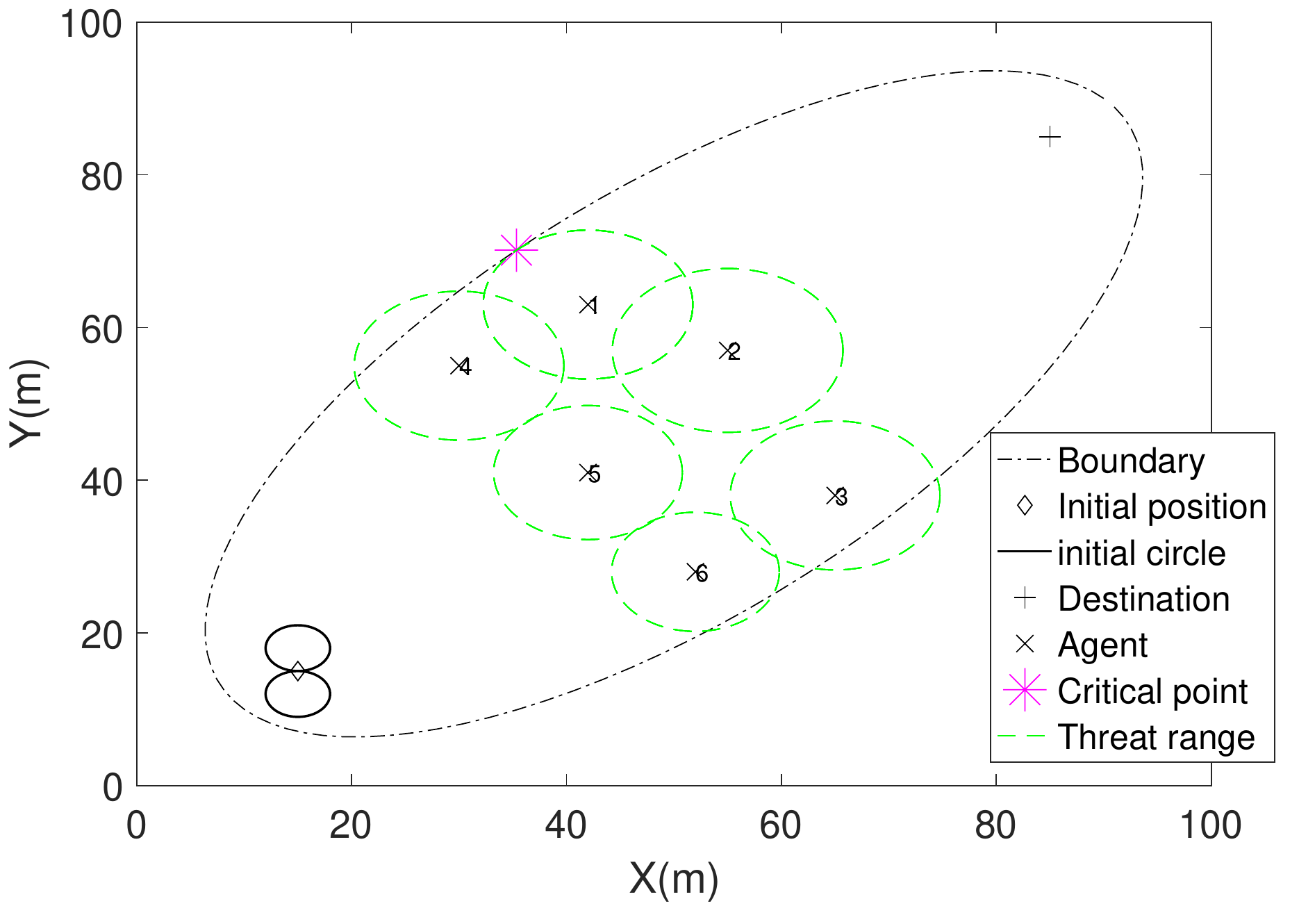}
        \caption{}
        \label{fig13b}
    \end{subfigure}
    \begin{subfigure}{0.45\textwidth}
        \includegraphics[width=\textwidth]{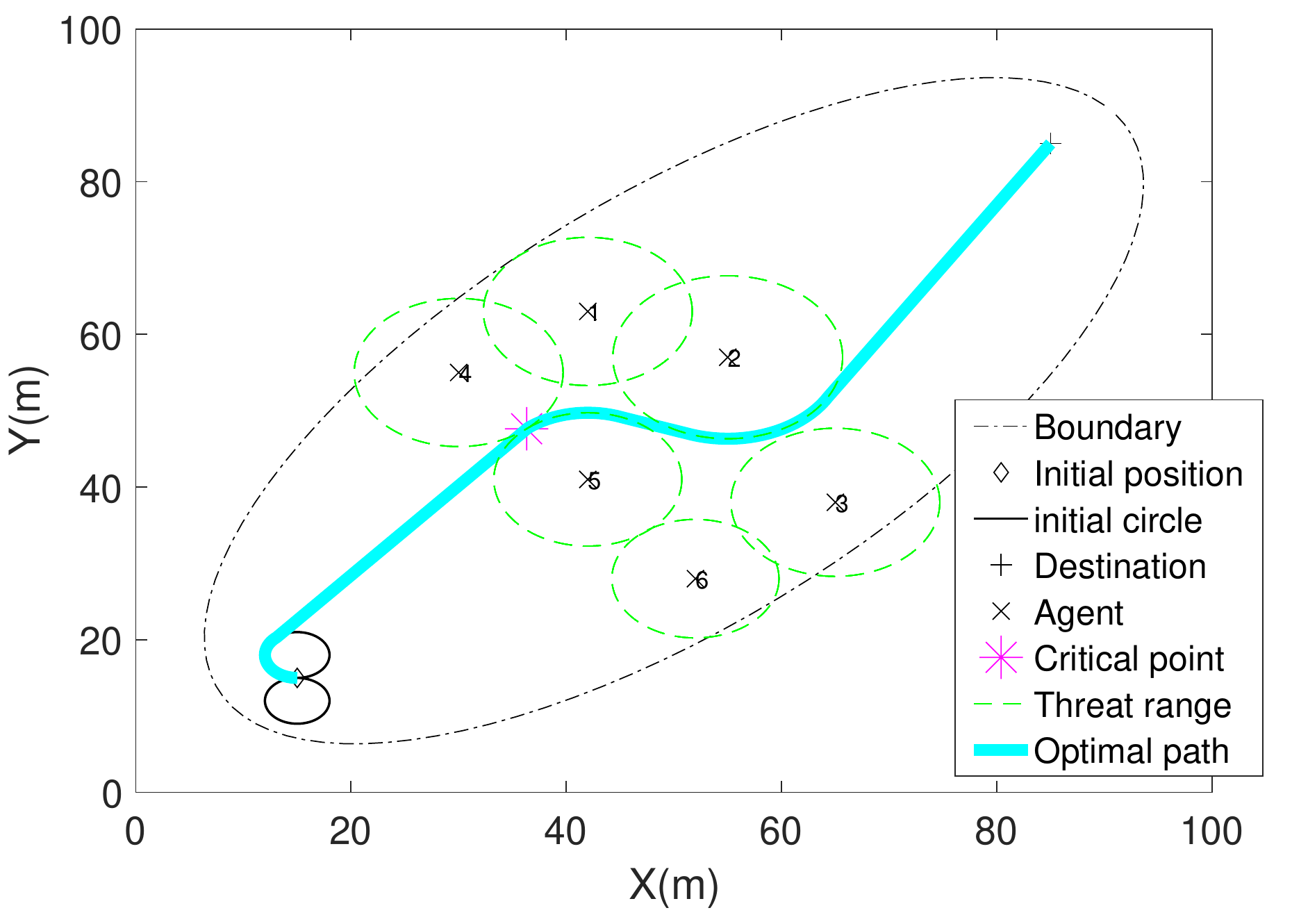}
        \caption{}
        \label{fig13c}
    \end{subfigure}
    \caption{Simulation 4 with 6 agents and non-identical threat radii. (a) $\theta=0.00$. (b) $\theta=0.02$. (c) $\theta=0.03$.}\label{fig13}
\end{figure}

\begin{figure}[h]
    \centering
    \begin{subfigure}{0.45\textwidth}
        \includegraphics[width=\textwidth]{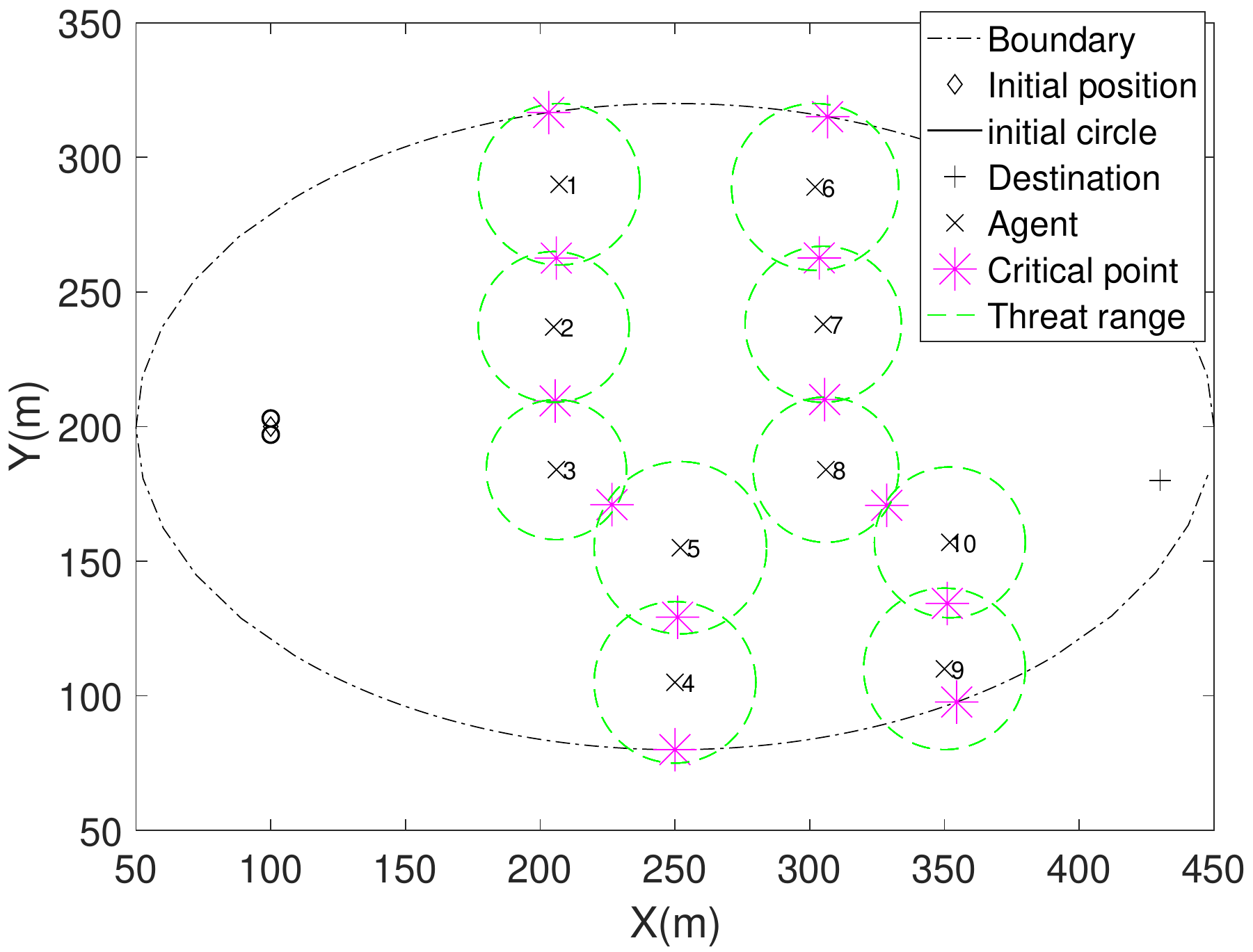}
        \caption{}
        \label{fig14a}
    \end{subfigure}
    \begin{subfigure}{0.45\textwidth}
        \includegraphics[width=\textwidth]{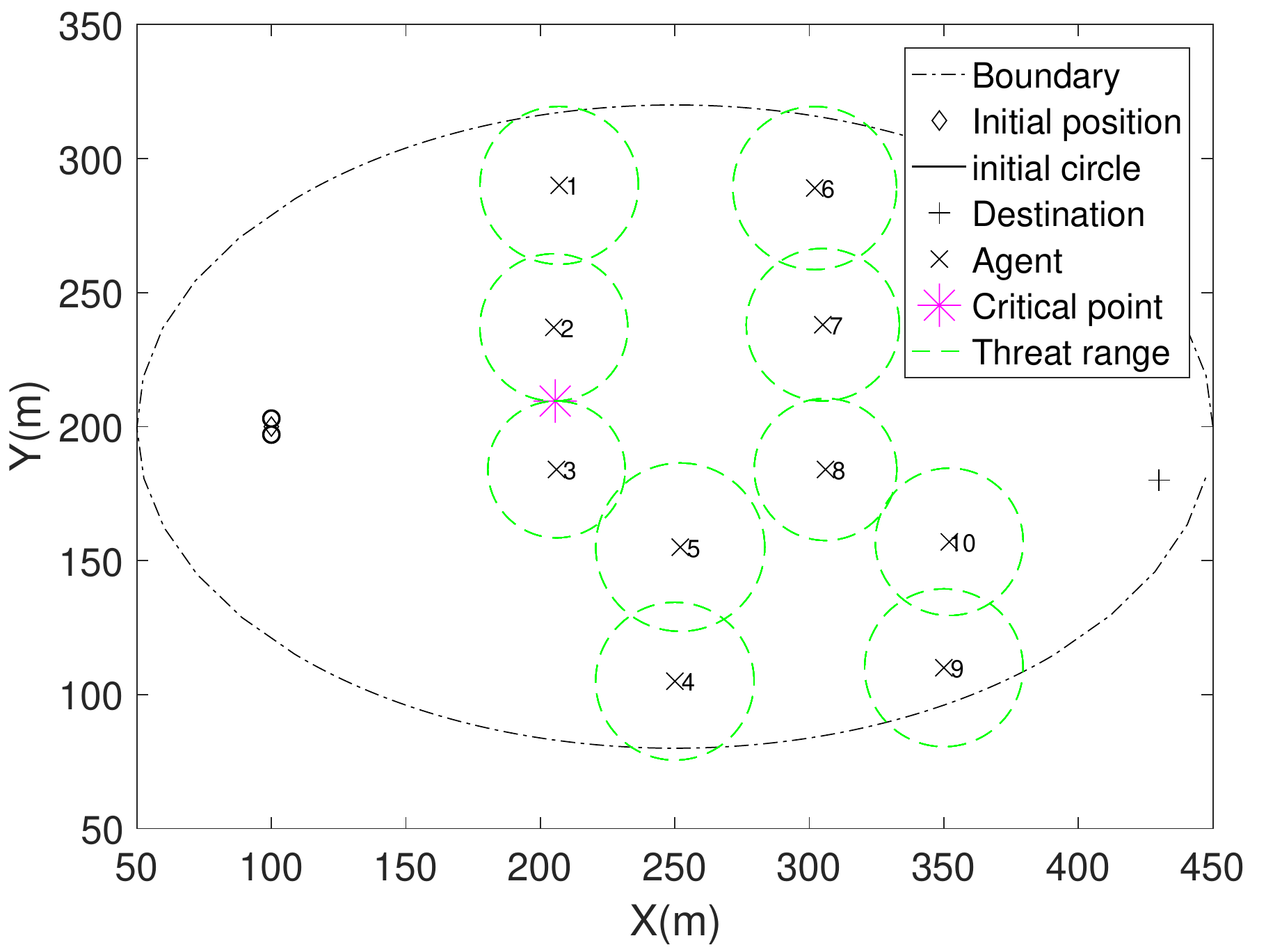}
        \caption{}
        \label{fig14b}
    \end{subfigure}
    \begin{subfigure}{0.45\textwidth}
        \includegraphics[width=\textwidth]{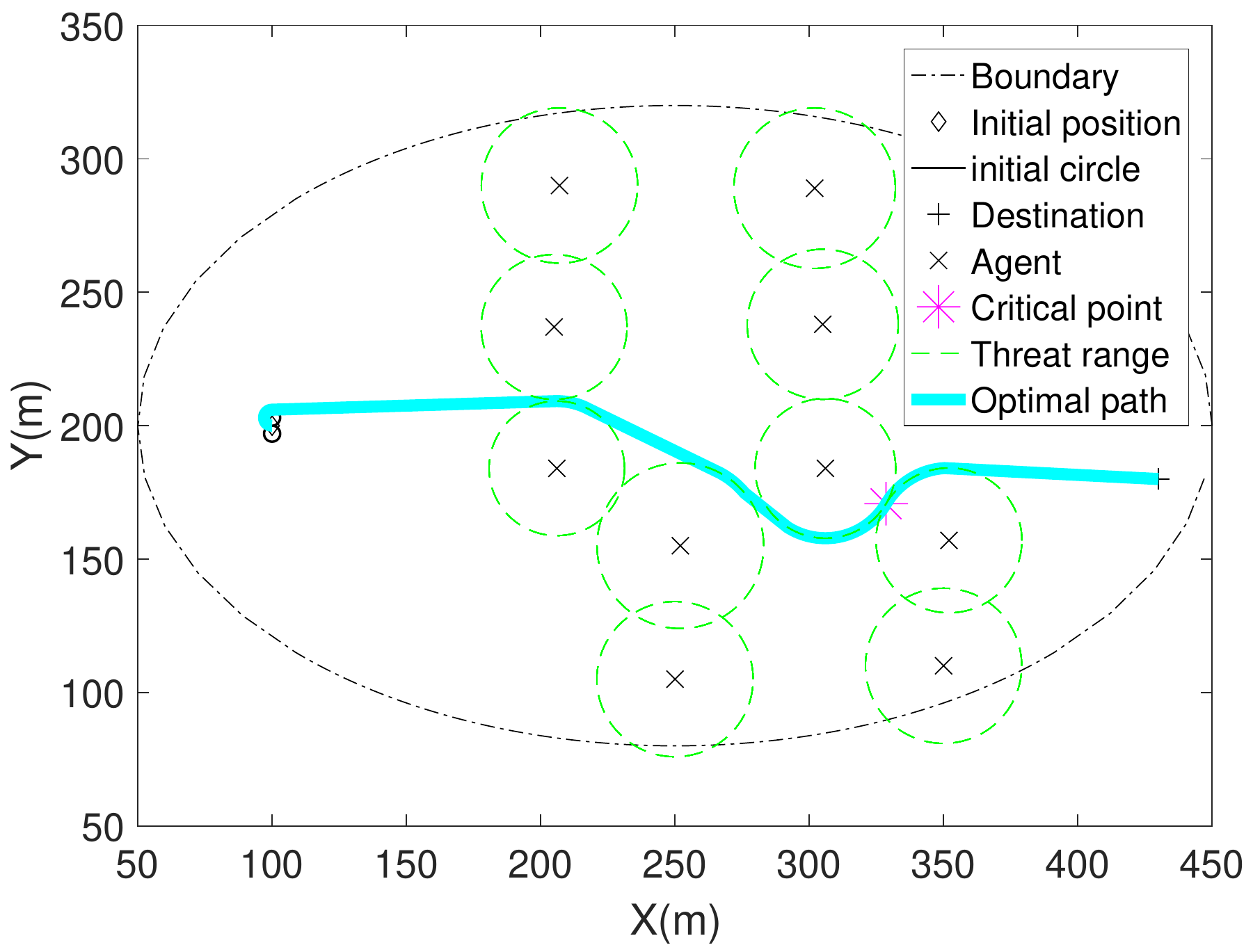}
        \caption{}
        \label{fig14c}
    \end{subfigure}
    \caption{Simulation 5 with 10 agents and non-identical threat radii. (a) $\theta=0.00$. (b) $\theta=0.03$. (c) $\theta=0.05$.}\label{fig14}
\end{figure}

\subsubsection{Comparison with existing work}
\begin{figure}[h]
    \centering
    \begin{subfigure}{0.45\textwidth}
        \includegraphics[width=\textwidth]{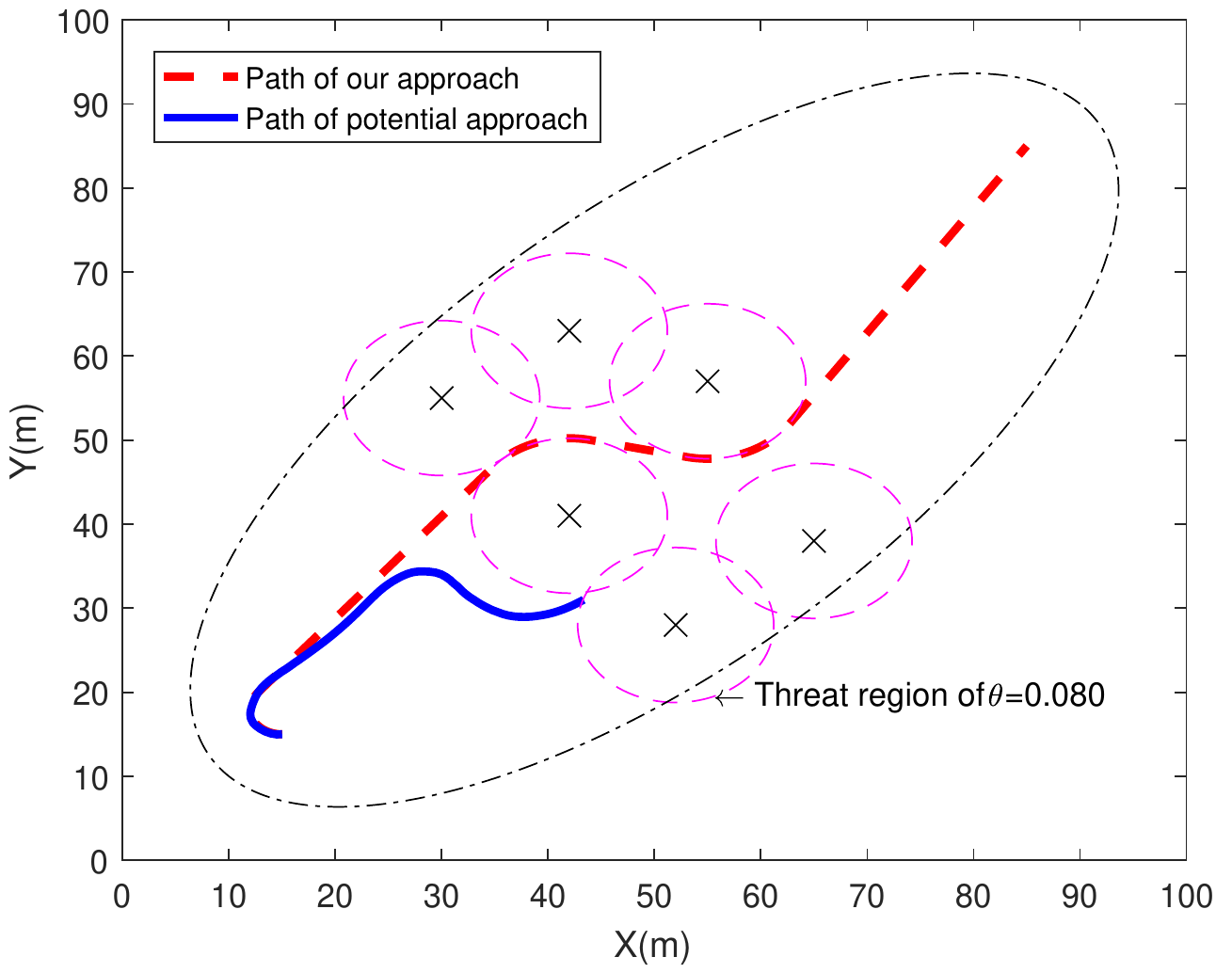}
        \caption{$R(\theta)=9.2$m.}
        \label{fig15a}
    \end{subfigure}
    \begin{subfigure}{0.45\textwidth}
        \includegraphics[width=\textwidth]{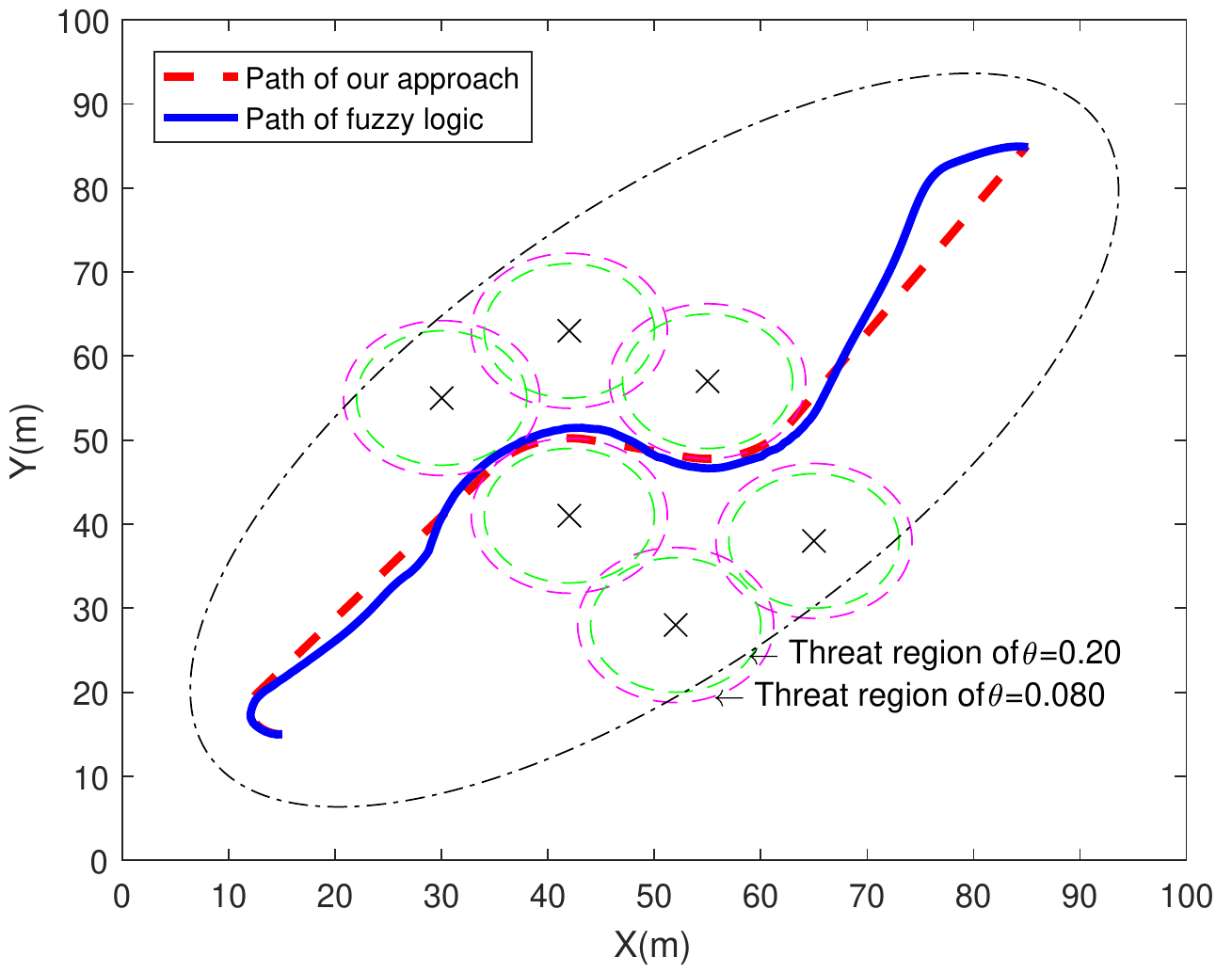}
        \caption{$R(\theta)=8$m; threat level: 0.145; path length: 118.3m.}
        \label{fig15b}
    \end{subfigure}
    \begin{subfigure}{0.45\textwidth}
        \includegraphics[width=\textwidth]{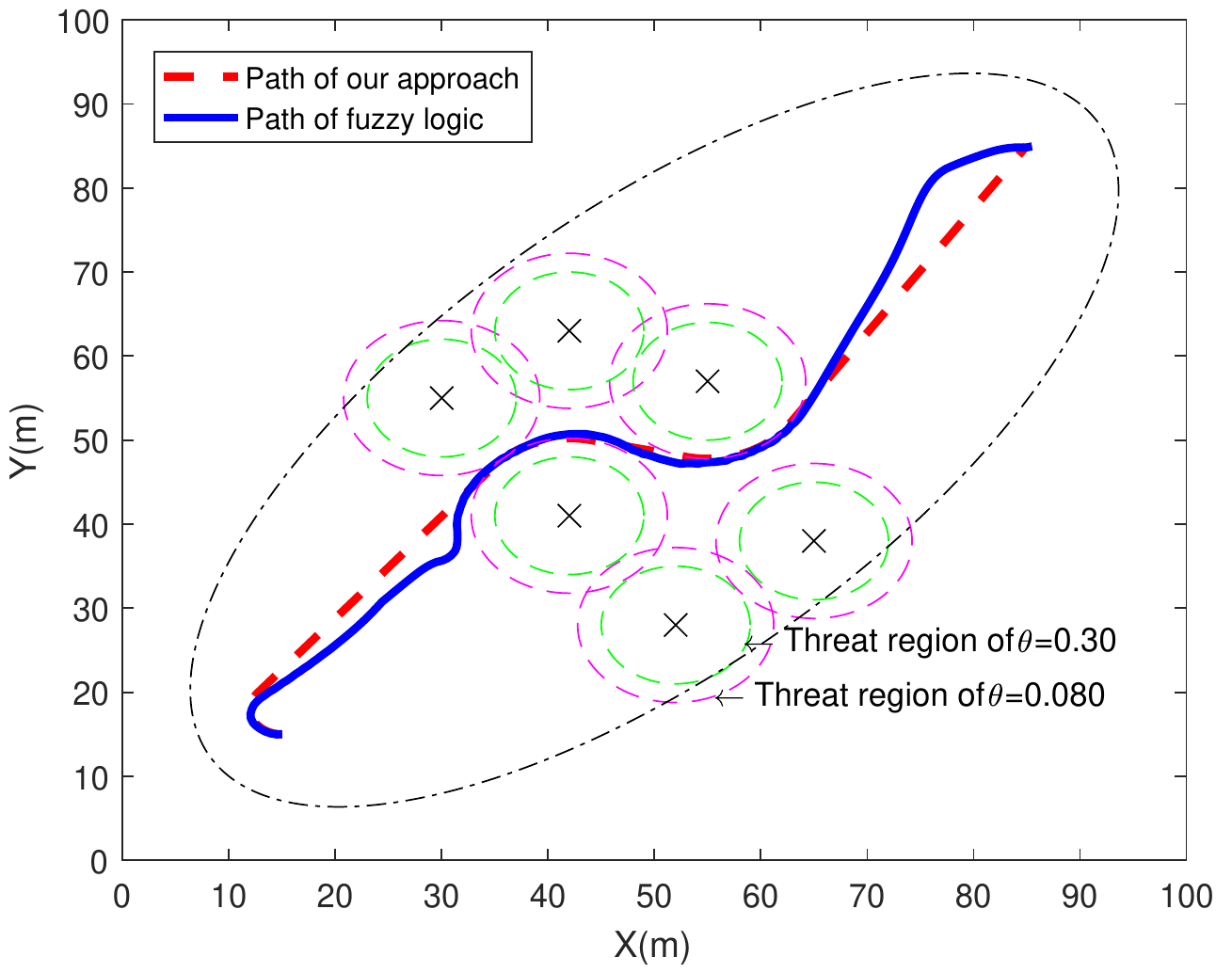}
        \caption{$R(\theta)=7$m; threat level: 0.086; path length: 117.2m.}
        \label{fig15c}
    \end{subfigure}
    \begin{subfigure}{0.45\textwidth}
        \includegraphics[width=\textwidth]{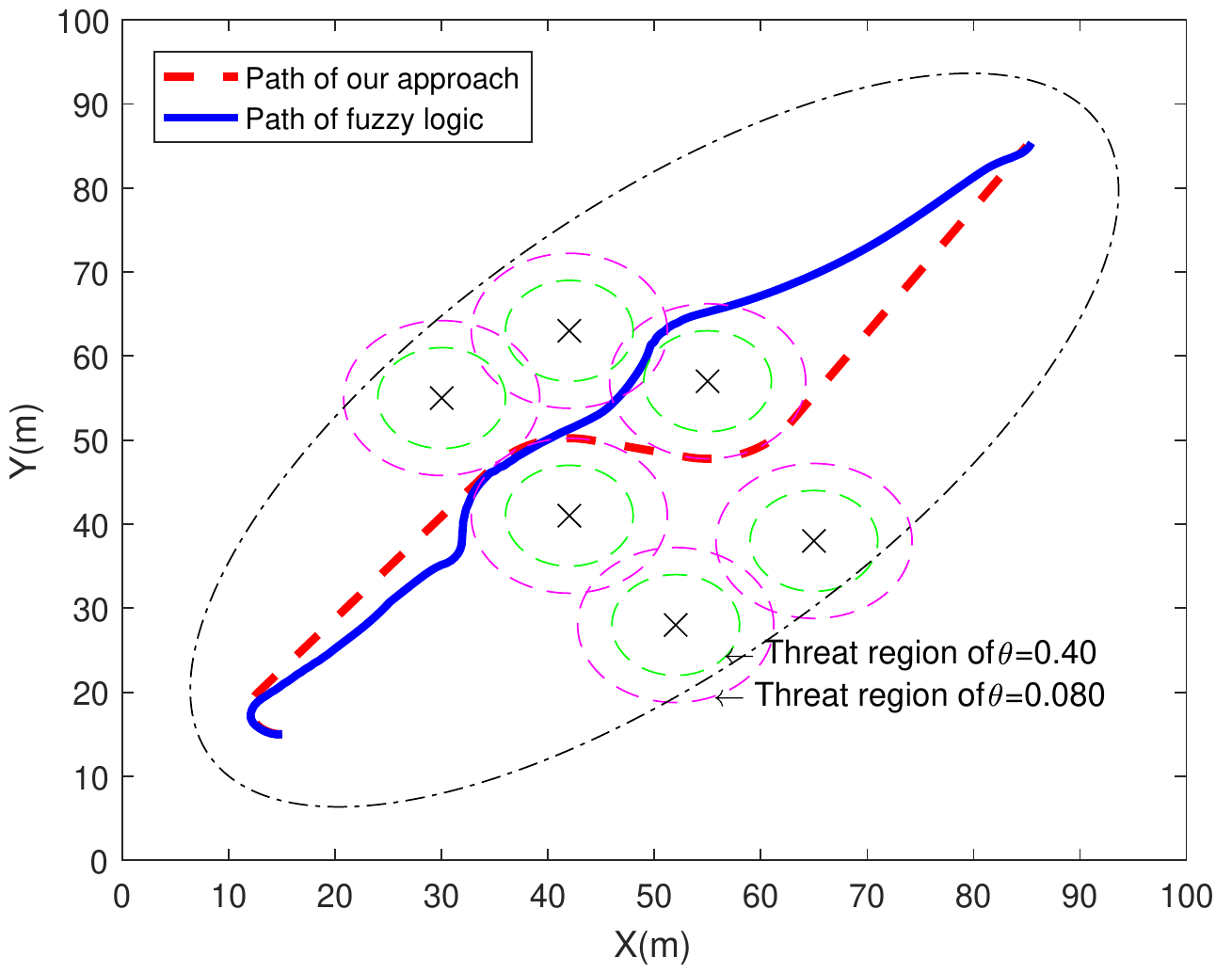}
        \caption{$R(\theta)=6$m; threat level: 0.290; path length: 108.0m.}
        \label{fig15d}
    \end{subfigure}
    \caption{Comparison with the fuzzy logic algorithm on Simulation 4.}\label{fig15}
\end{figure}

\begin{figure}[h]
    \centering
    \begin{subfigure}{0.45\textwidth}
        \includegraphics[width=\textwidth]{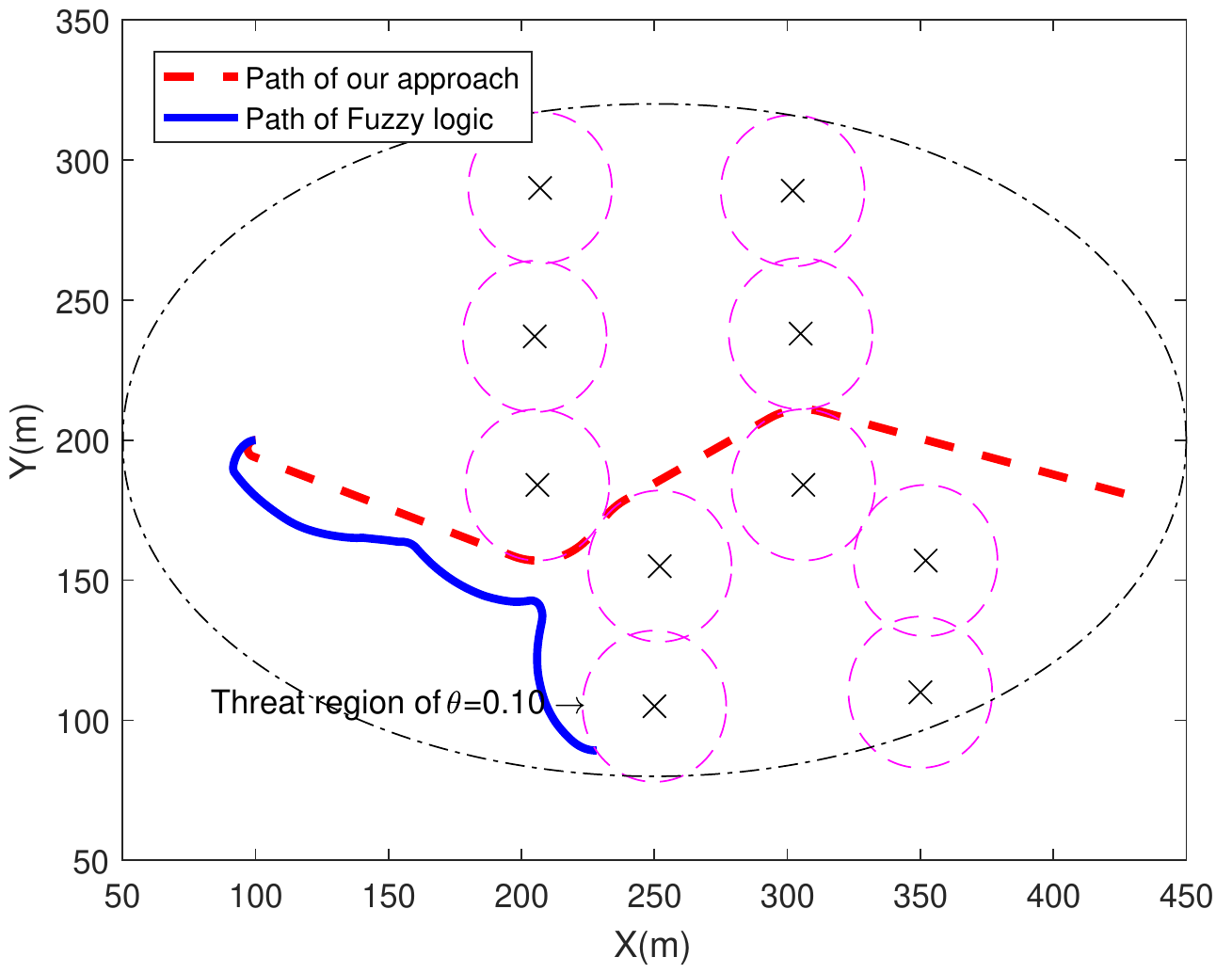}
        \caption{$R(\theta)=27$m.}
        \label{fig16a}
    \end{subfigure}
    \begin{subfigure}{0.45\textwidth}
        \includegraphics[width=\textwidth]{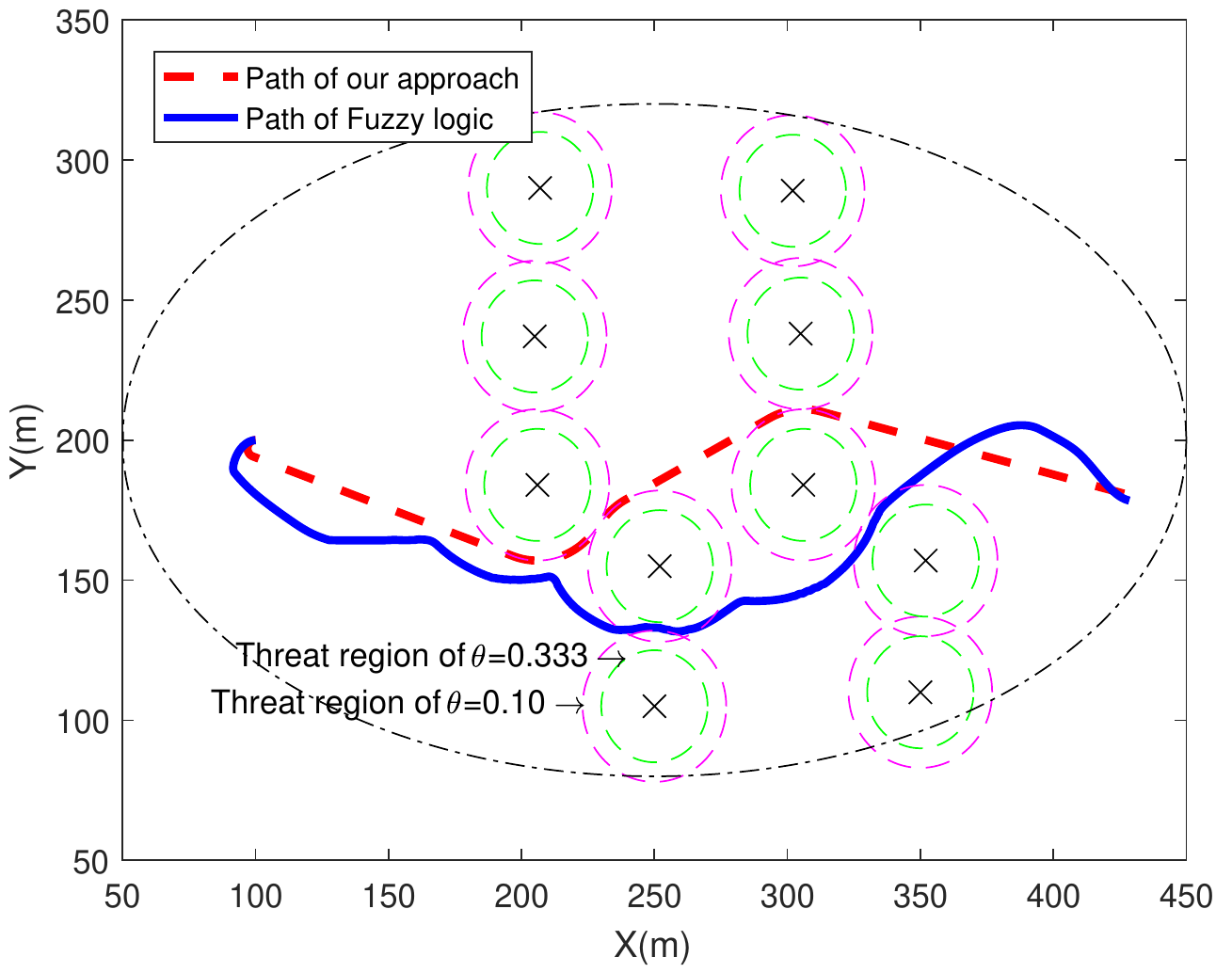}
        \caption{$R(\theta)=20$m; threat level: 0.261; path length: 407.3m.}
        \label{fig16b}
    \end{subfigure}
    \begin{subfigure}{0.45\textwidth}
        \includegraphics[width=\textwidth]{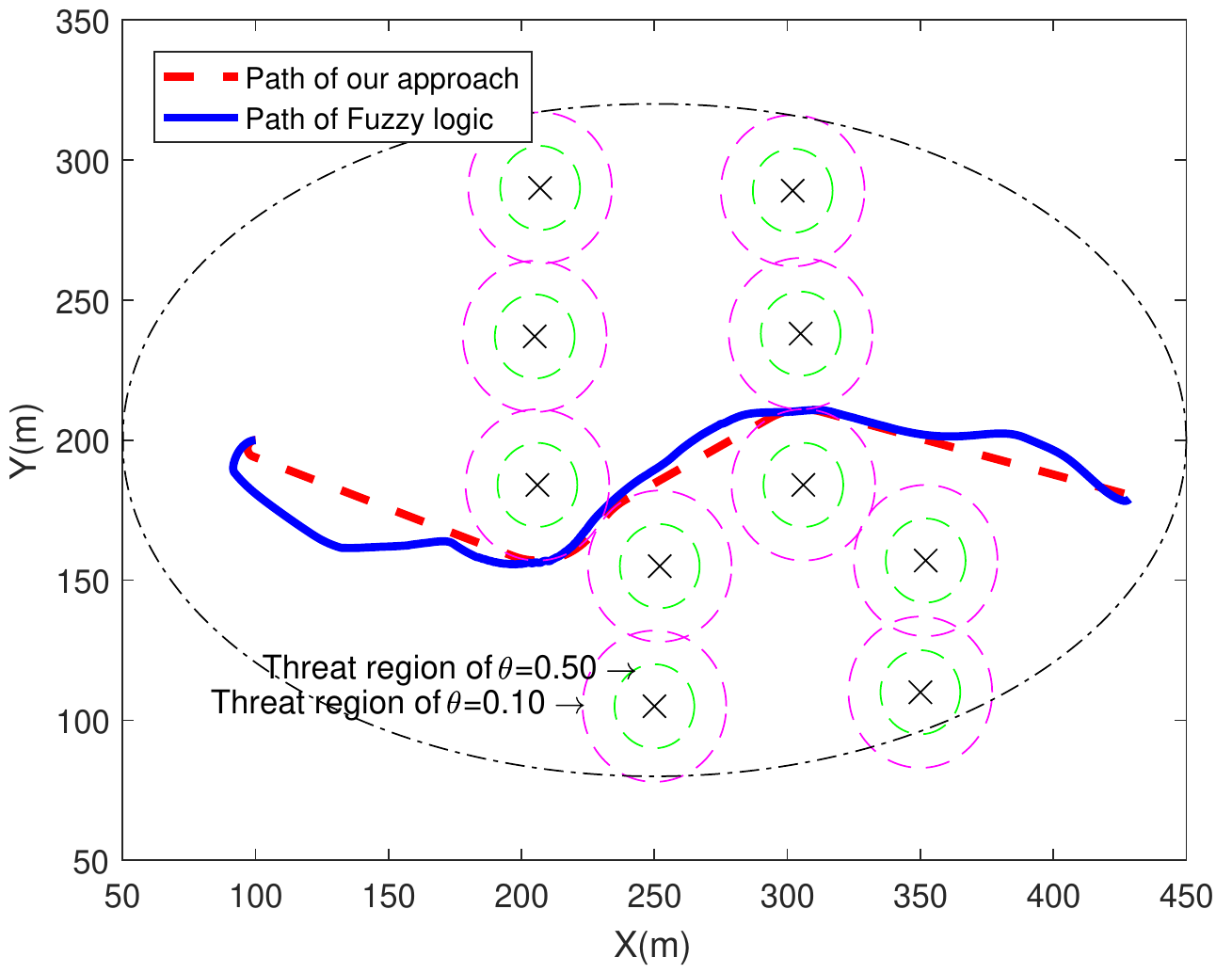}
        \caption{$R(\theta)=15$m; threat level: 0.325; path length: 373.6m.}
        \label{fig16c}
    \end{subfigure}
    \begin{subfigure}{0.45\textwidth}
        \includegraphics[width=\textwidth]{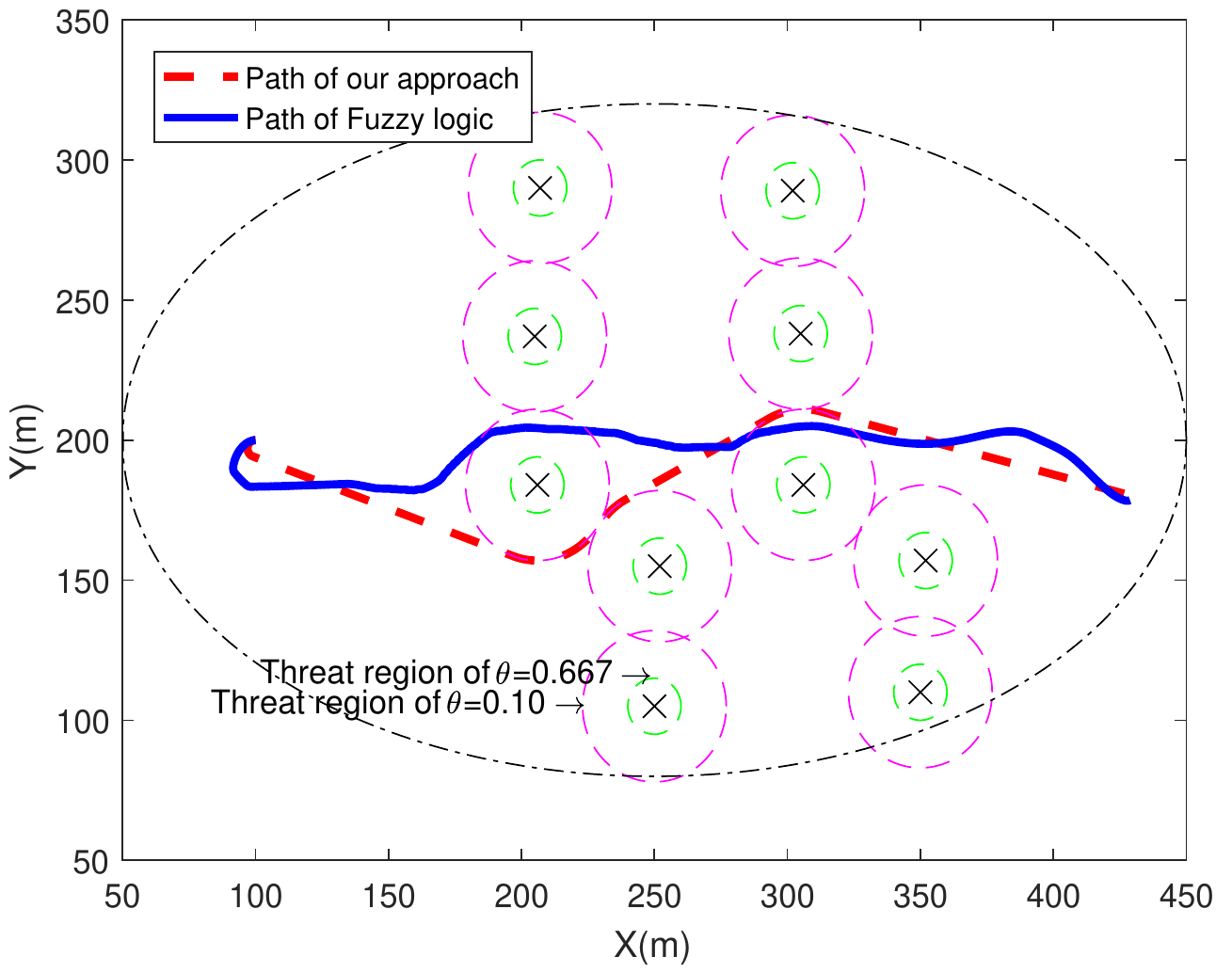}
        \caption{$R(\theta)=10$m; threat level: 0.191; path length: 389.3m.}
        \label{fig16d}
    \end{subfigure}
    \caption{Comparison with the fuzzy logic algorithm on Simulation 5.}\label{fig16}
\end{figure}

This subsection presents the comparison with a fuzzy logic approach. Fuzzy logic has been used in autonomous vehicle navigation to avoid collisions with obstacles, see e.g., \cite{llorca2011fuzzy, al2014fuzzy}. Since we use the Dubins car model to describe the motion of the aircraft, we have only one output of fuzzy logic system, i.e., angular velocity. We take front distance, left-front distance, right-front distance, distance to destination, and angle difference between the aircraft's current heading and the destination angle as inputs. Such system requires the aircraft to measure distances and know its current location.

Before applying the fuzzy logic algorithm to our problem, it is worth mentioning that such fuzzy logic algorithm requires at least one valid path from origin position to destination in the region. If we regard the threat regions with radii $R_i$ as obstacles, the field is blocked. Since the fuzzy logic system is unable to find the best threat level, we simply start from the radius corresponding to the minimum threat level found by our proposed approach, and then test some smaller radii.

We apply the fuzzy logic algorithm to Simulation 4 and 5, and the generated paths are displayed in Fig. \ref{fig15} and \ref{fig16} respectively. For the sake of comparison, the paths by our approach are also displayed. In Fig. \ref{fig15a}, the threat radius is 9.2m, which is consistent with the obtained 0.08 threat level. The fuzzy logic system fails to find a path leading the aircraft from initial position to destination. Further, we test the radii of 8, 7, and 6m, and the generated paths are demonstrated in Fig. \ref{fig15b}, \ref{fig15c} and \ref{fig15d}. We observe from these simulations that the best path generated by the fuzzy logic system is in Fig. \ref{fig15c}, whose threat level is 0.086 and path length is 117.2m. But this is still worse than our path, whose threat level and path length are respectively 0.080 and 112.1m. We execute the same procedure for Simulation 5. The best path generated by the fuzzy logic system is in Fig. \ref{fig16}, whose threat level is 0.191 and path length is 389.3m. Again the path produced by our proposed approach outperforms in threat level (0.100) and path length (365.6m).

Through these simulations we find that the threat level of a path generated by fuzzy logic system is uncontrollable. In contrast, our approach is based on a rigorous mathematical analysis and is guaranteed to find the optimal threat level from a finite set of candidates. Further, the length of the paths produced by our approach is proved to be the shortest under the threat level. The presented simulation results show  that our approach achieves better performance than the fuzzy logic algorithm  in terms of both threat level and path length.

\section{Summary}
\label{conclusion_path}
This chapter studies the problem of planning shortest viable path for unicycle robots with bounded turning radius serving as data collectors in a cluttered sensing field. We define a viable path which combines the concerns of both robotics and sensor networks. In many applications, this term is closer to reality. We formulate the problem of planning the shortest viable path for a single robot as a variant of DTSPN. Accordingly, we develop a Shortest Viable Path Planning (SVPP) algorithm. We further consider the problem of planning viable paths for multiple robots and presented a $k$-Shortest Viable Path Planning ($k$-SVPP) algorithm. We conduct simulations with different network scales, robot speeds and distributed data loads to show the performance of the proposed algorithms. Also comparisons with existing alternatives were provided. We find that SVPP and $k$-SVPP were effective to design viable paths for unicycle robots with bounded angular velocity. Compared to multihop communication, using our algorithms can save around 95\% energy for sensor nodes. Further, both increasing the robot speed and employing multiple robots are able to reduce the data collection time significantly. In this chapter, we consider several practical issues existing in the utilization of mobile robot to collect data and the presented results are meaningful in real applications. 

Moreover, we apply our approaches to the problem of point to point navigation, with the background of safe mission planning. We propose an optimization model to navigate an aircraft or a flying robot to its final destination while minimizing the maximum threat level and the length of the aircraft path. The construction of optimal paths involves a simple geometric procedure and is very computationally efficient. The effectiveness of the proposed method has been demonstrated by illustrative examples and  comparisons with existing work. It should be pointed out that we consider a 2D or planar navigation problem. An important direction of future research will be an extension of the presented planar algorithm to practically important cases of 3D threat environments.
\chapter{Energy Efficient Approach for Data Collection in Wireless Sensor Networks Using a Path Fixed Mobile Sink}\label{cluster_ms}
\minitoc

Chapter \ref{path_planning} focuses on an ideal movement pattern of MSs, i.e., \textit{\textbf{controllable}} mobility. Such model is easy to manage, while it is difficult to apply to realistic applications, although it performs well in theory. The basic reason is that in real applications, the  environment in which the MSs are moving is quite complex. It has not only obstacles as mentioned in Chapter \ref{path_planning}, but also other types of restrictions, for example, the MSs have to move only on roads in urban area. Therefore, it is necessary to study another mobility pattern, i.e., constrained mobility. Chapter \ref{cluster_ms}, \ref{cluster_cs}, \ref{cluster_um} and \ref{drone} all focus on how this kind of mobility can improve the system performance.

\section{Motivation}

As mentioned in Chapter \ref{literature}, one side effect of multiple hop communication to transmit data packets to the static BSs is the funneling effect \cite{li2007analytical}. Recent studies have shown that using MSs to collect data in WSNs can relieve the funneling effect issue \cite{mobilesinksurvey11, Gu16}. A MS traversing the sensing field can collect data from sensor nodes over a short range communication link \cite{zhao2012optimization, mottaghi2015optimizing}, and then the on-board MS transmits the collected data wirelessly to a remote center, since it has no energy limitation. Long-hop relaying is not used at sensor nodes and the energy consumption is reduced. Traversing the sensing field by MS needs to be timely and efficient because failure to visit some parts of the field leads to data loss, and infrequently visiting some areas results in long delivery delay. Besides, the trajectory planning of MS in these cases become more difficult to cope with. Furthermore, in the urban areas, the planned trajectory sometimes cannot be realized since the MS is constrained to roads. Alternatively, amounting MS on a vehicle, such as a bus, avoids some difficulties and can provide better performance for data collection. First, since the bus is already a component of the environment and its trajectory is predefined, the difficult path planning and complex control of MS's movement are avoided. Second, instead of visiting each sensor node individually, which is a time consuming task due to the low physical speed of MS, combining multihop communication with path constrained MS is able to increase the data delivery delay.

This chapter investigates using a MS, which is attached to a bus, to collect data in WSNs with nonuniform node distribution. Such WSNs exist in many applications. For example, in the case of monitoring the air pollution of a city, the industrial areas are usually deployed with more sensors than the residential areas. Also, since the areas of interest may be isolated from each other, using conventional data collection approaches is not appropriate due to the limited budget of energy resource. In this case, exploiting a MS amounted on a bus is able to relieve the bottleneck of energy at sensor nodes. Because the MS can serve the isolated areas at different time. It is like that there is a virtual static sink for each area and such sink only works at specified time duration. The specified time duration is the duration during which the MS is in the area. Instead of the coverage problem studied in publications \cite{cheng2013decentralized, savkin2012optimal, savkin2015book}, the focus here is on routing the sensory data from source nodes to MS in an energy efficient way such that the energy expenditure is balanced across the entire network.

The main contributions of this chapter are a clustering algorithm and a routing algorithm. The core of the clustering algorithm lies in the selection of cluster heads (CHs). With the aim of balancing energy consumption, we design unequal cover ranges for CHs considering the feature of nonuniform node distribution. The cover range of a CH depends on its distance to MS and the local node density. Unlike other works, the distance here is the hop distance instead of Euclidean distance. Removing the ability of measuring Euclidean distance simplifies the sensor nodes. The unequal cover ranges can make the clusters with similar distances to MS have approximate sizes such that the energy consumption by CHs can be balanced. Since the designed cover range does not exceed the single hop communication range, the cluster members (CMs) consume energy approximately also. The proposed routing algorithm associates each CH to a CH that is closer to MS' trajectory. The CH $u$ will associate to a CH $v$ only if $v$'s hop distance is no larger than $u$'s. Further, the association accounts the residual energy of CH $v$ and the number of attached CMs. We compare our approach with some existing ones through simulations and we conclude that our approach achieves longer network lifetime \footnote{There are several definitions of network lifetime in the literature. Here we adopt the definition of network lifetime as the number of rounds until the first node exhausts its energy reserve, which has been widely used.}.

The remainder of this chapter is organized as follows: Section \ref{model4} presents the network model. Section \ref{protocol4} discusses the proposed protocol in details, which is followed by some theoretical analysis. Section \ref{simulation4} provides extensive simulations to evaluate the proposed approach. Finally, Section \ref{conclusion4} summarizes this chapter. The publications related to this chapter include \cite{huang2017energy}, \cite{huang2016optimal}, \cite{huang2017vtc}.
\section{Network Model}\label{model4}
Consider a wireless sensor network consisting of $n$ static sensor nodes nonuniformly deployed in the field. A bus carrying a MS moves the predefined trajectory following its timetable. We consider the following assumptions.

\begin{enumerate}
\item The nodes as well as MS have unique IDs. 
\item All the nodes use power control to adjust the transmitting power.
\item If a node works as CH, it aggregates the received data packets within cluster into one packet; while it does not aggregate the data packets from other CHs. Further, the raw data packets and the aggregated packets have the same size.
\end{enumerate}

We consider the energy dissipation model used in previous work, e.g., \cite{CHEN09, WEI11, xiang10,  yu12, li2013coca, Shokouhifar15}:

\begin{equation}\label{transmitting_cluster}
E_t(l,d)=
\begin{cases}
l\times E_{elec}+l\times E_{fs}\times d^2,\ if\ d\leq d_0\\
l\times E_{elec}+l\times E_{mp}\times d^4,\ if\ d> d_0\\
\end{cases}
\end{equation}
where $E_t(l,d)$ is the total energy dissipated to deliver a single $l$-bit packet from a transmitter to its receiver over a single link of distance $d$. The electronic energy $E_{elec}$ depends on  electronic factors such as digital coding, modulation, filtering, and spreading of the signal. The amplifier energy in free space $E_{fs}$ or in multipath environment $E_{mp}$ depends on the distance from the transmitter to the receiver, and the threshold is $d_0$. 

For receiving data packets, the sensor nodes expand energy according to:
\begin{equation}\label{receiving}
E_r(l)=l\times E_{elec}.
\end{equation}
Note, data packet transmission and control message exchanging both follow models (\ref{transmitting_cluster}) and (\ref{receiving}).

Additionally, we assume that the energy consumption for sensing and data aggregation are, respectively,
\begin{equation}
E_s(l)=l\times E_{sens}
\end{equation}
and
\begin{equation}
E_a(l)=l\times E_{aggr}
\end{equation}
where $E_{sens}$ depends on electronic factors and $E_{aggr}$ relates to the aggregation algorithm.

\section{Routing Protocol}\label{protocol4}

The proposed protocol scheme contains two stages: initial and collecting stages. Figure \ref{fig:diagram} illustrates the protocol operation by the time line. Basically, the initial stage aims at making every node aware of the information required to operate the following procedures. The collecting stage consists of a number of data collection cycles. At the beginning of each cycle, the network constructs cluster formation. In this phase, the sensor nodes transmit control message to their neighbour nodes and build up network structure in a distributed manner. Then a certain number of data collection rounds are operated. The control messages used here are described in Table \ref{table:message}.

\begin{figure*}[t]
\begin{center}
{\includegraphics[width=0.95\textwidth]{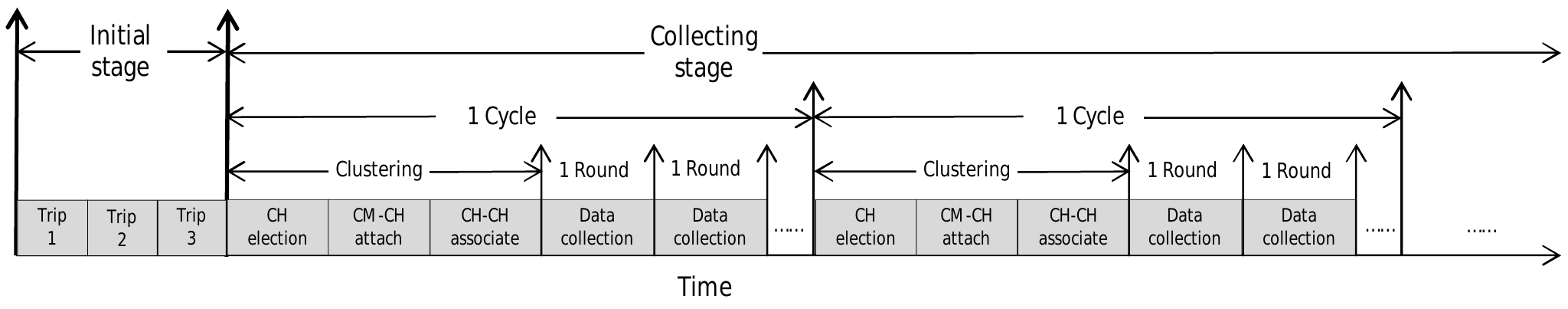}}
\caption{The operation of the proposed protocol by time line.}
\label{fig:diagram}
\end{center}
\end{figure*}

\begin{table}[t]
\begin{center}
\caption{Description of control messages}\label{table:message} 
  \begin{tabular}{ |l|  l|}
    \hline
    Message & Description \\ \hline
    Initial-Msg & hop, ID  \\\hline
    Return-Msg & hop \\\hline
    Hop-Msg & $h_{max}$, $h_{min}$ \\\hline
    CH-Compete-Msg & Energy, ID \\\hline
    CH-Win-Msg& ID \\  \hline
    CH-Quit-Msg &ID  \\\hline
    CM-CH-Offer-Msg &  ID  \\\hline
    CM-CH-Request-Msg &  ID  \\\hline
    CM-CH-Confirm-Msg &  ID  \\\hline
    CH-CH-Offer-Msg &  Energy, CM size, hop, ID \\\hline
    CH-CH-Confirm-Msg &  ID  \\\hline
    Data-Request-Msg &  ID \\\hline
    \end{tabular}
\end{center}
\end{table}

\subsection{Initial Stage}
The initial stage requires MS to make three trips on its path. The purpose here is to get hop distance to MS' trajectory as well as local node density for each node. Such information plays a significant role in the collecting stage.

Trip 1. When MS moves, it continuously broadcasts Initial-Msg containing MS ID and hop distance (the hop distance equals to 0). A node, which is within the communication range of MS, receives the packet and executes the following procedures to extract some information from the message and modify it: 1) extracting the ID in the packet as its parent node ID; 2) replacing the ID with its own ID; 3) increasing the hop distance by 1; 4) extracting the hop distance as its hop distance to MS. Then it broadcasts the modified message to the nodes within its communication range. Note, it is possible for one node to receive more than one message. If the hop distances are different, it selects the smallest one, see Figure \ref{fig:hop} for an example; otherwise, it selects the most early received one.  At the same time, it keeps the number of the received messages as its neighbour count. In the end of this trip, every node $i$ knows its hop distance ($h_i$) to MS, the parent node, as well as the neighbour count ($n_i$).

\begin{figure}[t]
    \centering
    \begin{subfigure}[t]{0.45\textwidth}
        \includegraphics[width=\textwidth]{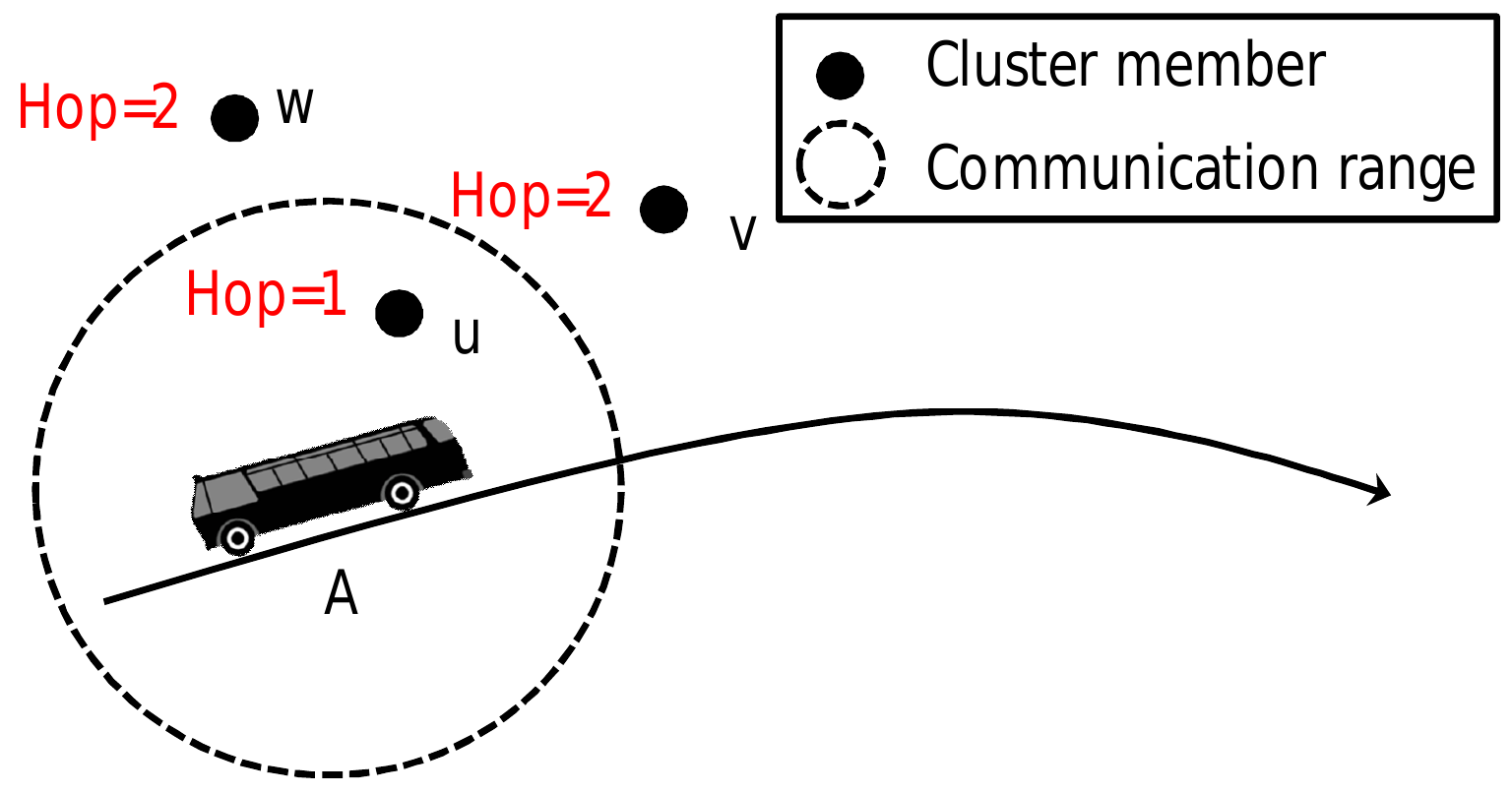}
        \caption{}
        \label{fig:hop1}
    \end{subfigure}
    \begin{subfigure}[t]{0.42\textwidth}
        \includegraphics[width=\textwidth]{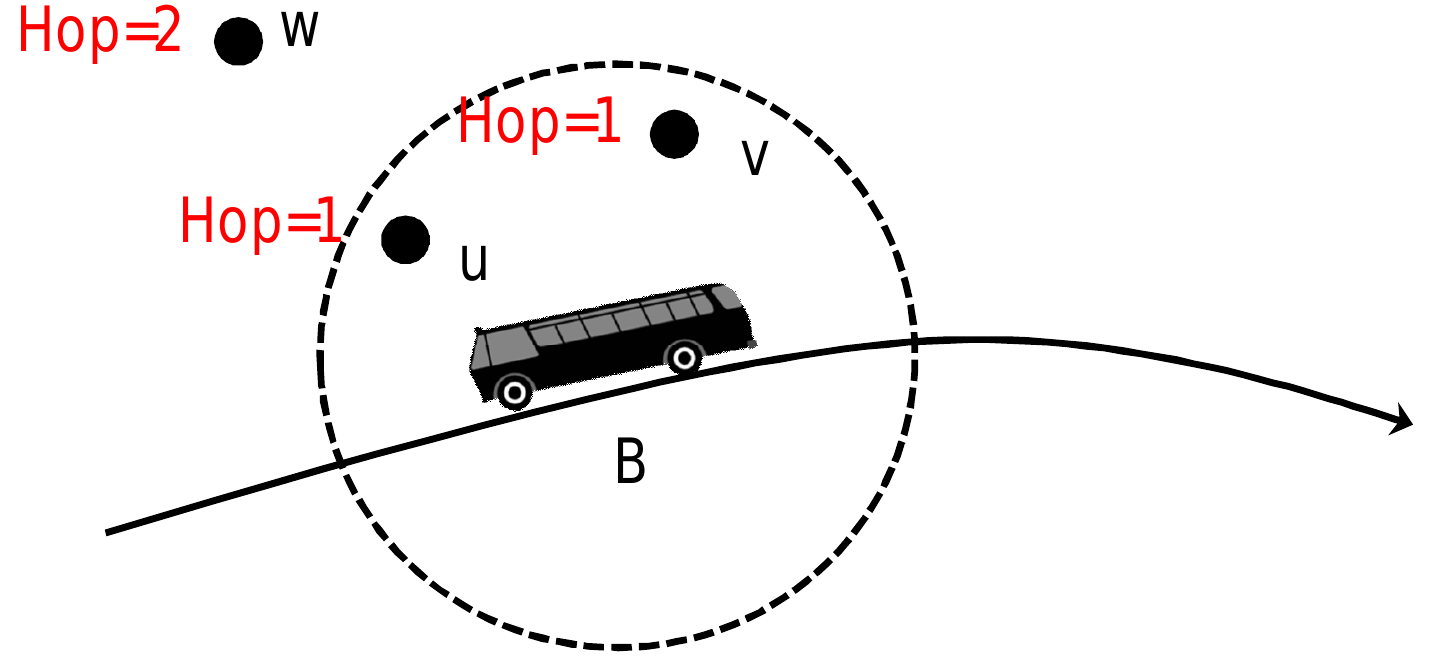}
        \caption{}
        \label{fig:hop2}
    \end{subfigure}
    \caption{Hop distance construction. When MS is at A, $u$ is within range. Hop distance of $u$ is 1. Since $v$ and $w$ are within range of $u$, their hop distances are 2. When MS is at B, $v$ is within range of MS, then the hop distance is changed to 1. $w$ never comes into range of MS but it is within range of $u$, thus its hop distance keeps 2.}\label{fig:hop}
\end{figure} 

Trip 2. Every sensor node transmits Return-Msg containing its hop distance to MS to its parent node. On this trip, the MS keeps receiving packets from the nodes nearby. At the end of Trip 2, MS knows the maximum and minimum hop distance ($h_{max}$ and $h_{min}$) from sensor nodes to itself.
 
Trip 3. MS keeps broadcasting Hop-Msg containing $h_{max}$ and $h_{min}$. The nodes receive such message extract $h_{max}$ and $h_{min}$ and forward the packet to other nodes within range. Finally, every node in the network knows $h_{max}$ and $h_{min}$.

\subsection{Collecting Stage}\label{steady_stage}
The collecting stage is the main stage of our protocol. It consists of a number of data collection cycles. At the beginning of each cycle, the sensor nodes re-organize themselves by constructing new clusters. After that, CMs send the sensory data to CHs and CHs send the aggregated data to MS directly or to another CH for relay. When the bus finishes its trip, we say one round of data collection is completed. Thus, each cycle consists of clustering and $k$ rounds of data collection. Obviously, the clustering result plays a significant role in the following data collection since it impacts on the energy consumption of both CHs and CMs. The parameter $k$ also influences the energy consumption. Small $k$ means that the network needs to reconstruct the clusters frequently, resulting in the large amount of control messages to exchange. On the other hand, large $k$ makes the network operate under one cluster structure for a long time, which may lead to the phenomenon that some CHs cannot survive the current cycle. Below, we describe the main phases in collecting stage in details.

\subsubsection{CH election}
For CH election, we extend the method in \cite{CHEN09}. In \cite{CHEN09}, the authors consider the scenario that the sensor nodes are uniformly deployed and design each CH's cover range based on its Euclidean distance to the static base station. In contrast, we consider the nonuniform deployment of sensor nodes, and we do not use Euclidean distance information since such information may not be reliable in harsh environment and measuring such information is costly. Instead, we use hop distance to perform the cover range calculation. The key equation to computer the cover range of each CH is shown below:
\begin{equation}\label{cover_range}
R_i=(1-\alpha\frac{1}{\sqrt{\rho_i}}\frac{h_{max}-h_i}{h_{max}-h_{min}})R_0
\end{equation}
where $R_i$, $\rho_i$ and $h_i$ are respectively the cover range, local node density and hop distance to MS of CH $i$, $\alpha$ is a given positive constant, and $R_0$ is the communication range of the sensor nodes. Note, $\rho_i$ is the relative node density, which is estimated based on the number of neighbour nodes, i.e., $\rho_i=n_i/\pi R_0^2/\bar{\rho}$. Here $\bar{\rho}$ is the average node density of the sensor network and can be obtained in the node placement phase. 

The fundamental idea to design the cover range like (\ref{cover_range}) is as follow. To make two CHs with the same hop to MS have approximately equal number of CMs, we try to make $\pi R_i^2 \rho_i=\pi R_j^2 \rho_j$ hold, from which we obtain $R_i \propto1/\sqrt{\rho_i}$, i.e., the CH in the dense area has a smaller cover range while the CH in the sparse area has a larger cover range, see Figure \ref{fig:idea} for an example. Besides, consistent with \cite{CHEN09}, the larger the hop distance to MS is, the larger the cover range will be. Thus, $R_i \propto h_i$. Eq. (\ref{cover_range}) works for both scenarios of uniform and nonuniform node distributions. For uniform node distribution, node densities are approximately equal across the entire network, thus $1/\sqrt{\rho_i}$ influences $R_i$ little.

\begin{figure}[t]
    \centering
    \begin{subfigure}[h]{0.45\textwidth}
        \includegraphics[width=\textwidth]{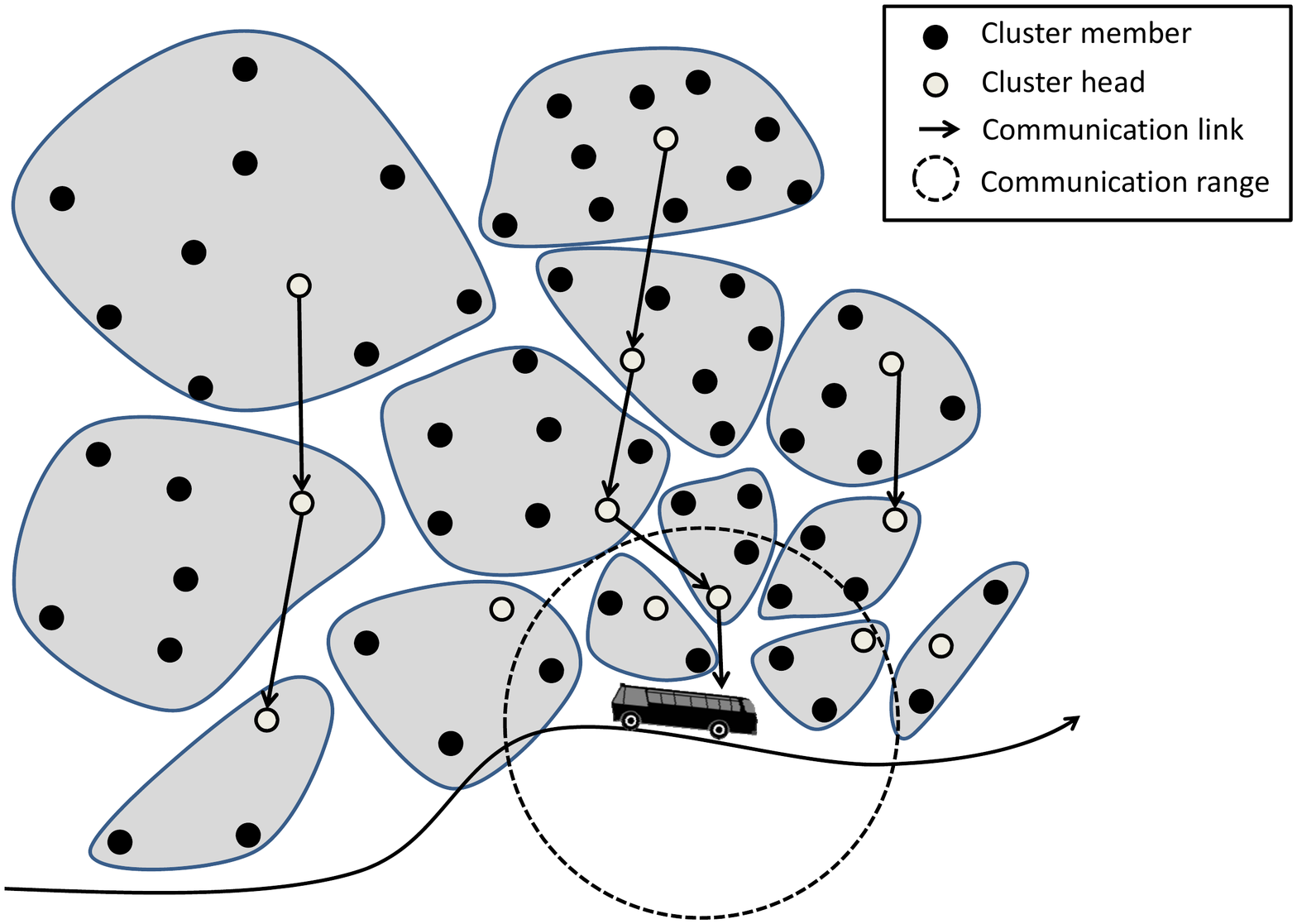}
        \caption{}
        \label{fig:nonuniform_idea}
    \end{subfigure}
    \begin{subfigure}[h]{0.45\textwidth}
        \includegraphics[width=\textwidth]{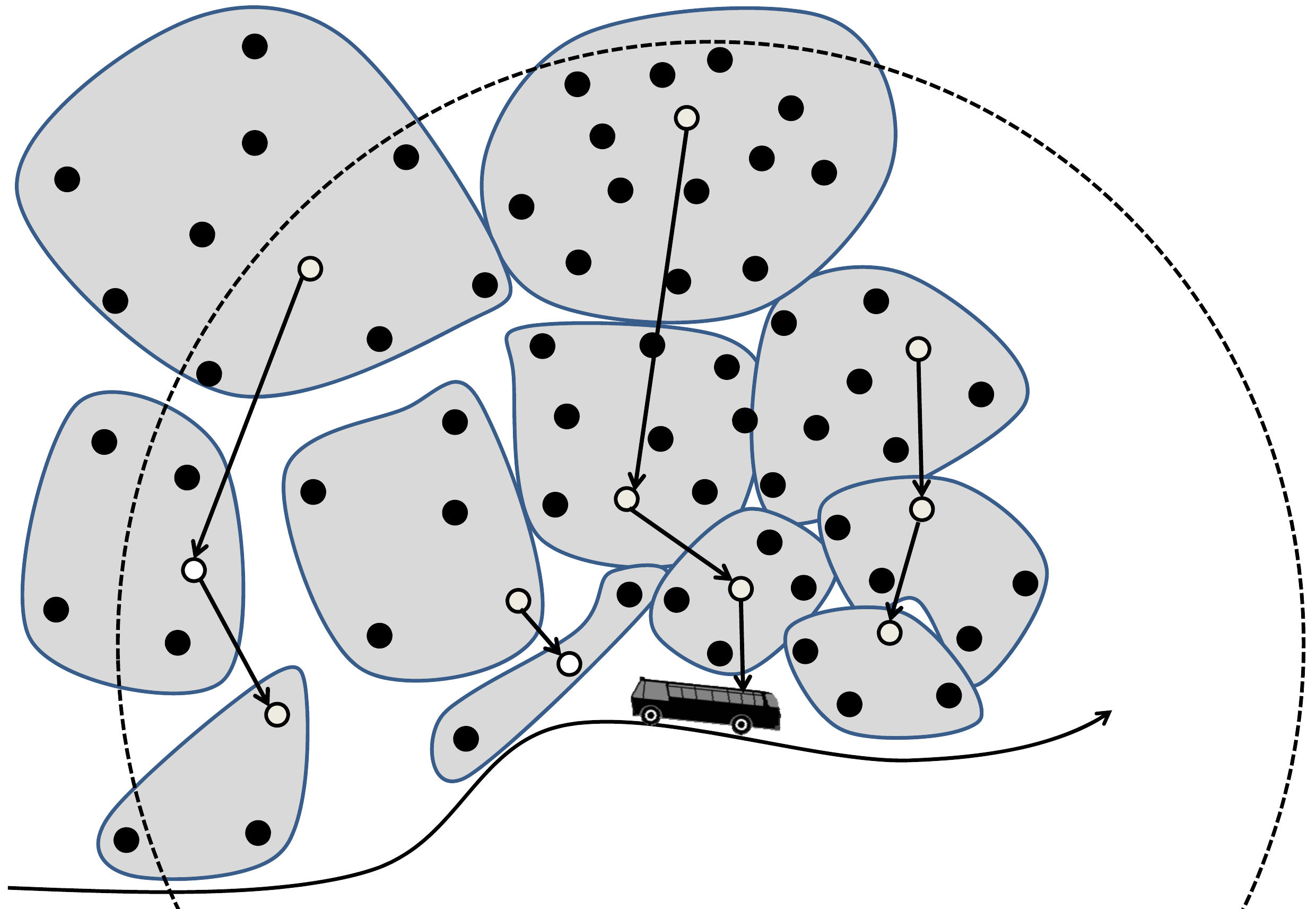}
        \caption{}
        \label{fig:uniform_idea}
    \end{subfigure}
    \caption{Cluster formation for a WSN with nonuniform node distribution. (a) Considering node distribution; (b) Not considering node distribution. When node distribution is not considered, CHs having similar distance to MS' trajectory have similar cover range. Thus, the CHs in dense areas have more CMs than those in sparse areas.}\label{fig:idea}
\end{figure} 

With this cover range in place, the network is in the phase of CH election. A set of CH-candidates is first elected. A node $i$ randomly generates $\mu \in (0,1)$ \cite{heinzelman2000energy}. Let $e_i$ be node $i$'s residual energy. If $\mu e_i >\beta$ ($\beta$ is the given threshold), node $i$ becomes a CH-candidate; otherwise, it becomes a CM. The CH-candidates turn into a CH election phase. In the CH election phase, every CH-candidate broadcasts CH-Compete-Msg containing its residual energy and ID within its cover range. The other CH-candidates within such cover range receive the message and compare the contained residual energy with their own. The competition rules are as follows. The CH-candidates with the highest residual energy within their cover range become CHs. The CHs broadcast CH-Win-Msg to inform their status. The CH-candidates that receive CH-Win-Msg quit the CH competition process and broadcast CH-Quit-Msg, because receiving CH-Win-Msg means they are covered by at least one CH's cover range. A CH-candidate that has not received CH-Win-Msg listens for CH-Quit-Msg. Once received CH-Quit-Msg, it ignores the corresponding node and check whether its residual energy is the largest in the left set. This process lasts until timer $T_{election}$ expires. The CH election algorithm is given as Algorithm \ref{algorithm1}.

\begin{algorithm}[h]
\caption{CH election}\label{algorithm1}
\begin{algorithmic}[1]
\State $\mu \longleftarrow rand(0,1)$
\If{$\mu e_i>\beta$}
\State Broadcast CH-Compete-Msg within $R_{i}$
\State Listen for CH-Compete-Msg
\State Find the message which contains the largest residual energy $e_j$
\While{timer $T_{election}$ has not expired}
\If{$e_j<e_i$}
\State Broadcast CH-Win-Msg within $R_{i}$; exit
\EndIf
\State Listen for CH-Win-Msg
\If{CH-Win-Msg is received}
\State Broadcast CH-Quit-Msg within $R_{i}$; exit
\EndIf
\State Listen for CH-Quit-Msg
\If{CH-Quit-Msg is received}
\State Get rid of the message corresponding to the ID in CH-Quit-Msg
\State From the left messages, find the one which contains the largest residual energy $e_j$
\EndIf
\EndWhile
\EndIf
\end{algorithmic}
\end{algorithm}
\subsubsection{CM-CH attachment}
After CH election, every CM needs to attach to a CH. The CM-CH attachment algorithm is shown as Algorithm \ref{algorithm2}. In this phase, CH and CM send message or hear alternatively. At the beginning, CH broadcasts a CM-CH-Offer-Msg within its cover range. Once a CM hears a CM-CH-Offer-Msg, it sends a CM-CH-Request-Msg back to the CH. Then, CH sends CM-CH-Confirm-Msg to the CMs from which it hears the CM-CH-Request-Msg. If a CM cannot hear CM-CH-Offer-Msg from any CHs, it broadcasts CM-CH-Request-Msg until it joins a CH as shown by Line 11-14 and the initial transmitting range is its own cover range. In the end of this phase, any CH (CM) knows its associated CMs (CH). We introduce several timers to support the implementation of CM-CH attachment. CM listens for CM-CH-Offer-Msg until timer $T_{offer}$ and CH listens for CM-CH-Request-Msg until timer $T_{asso}$. After timer $T_{offer}$, no matter a CM receives CM-CH-Offer-Msg or not, it starts to send CM-CH-Request-Msg and listen for CM-CH-Confirm-Msg until timer $T_{confirm}$. We display the operations of CH and CM in this phase in Figure \ref{fig:timer}. 

For intra-cluster data collection, TDMA schedule is used. The CH sets up a TDMA schedule based on the number of its CMs and transmits it back to its CMs. After the TDMA schedule is known by all CMs in the cluster, the CM-CH attachment phase completes.

\begin{algorithm}[t]
\caption{CM-CH attachment}\label{algorithm2}
\begin{algorithmic}[1]
\If{node $i$ is a CH}
\State Broadcast CM-CH-Offer-Msg within $R_{i}$.
\State Listen for CM-CH-Request-Msg until timer $T_{asso}$ expires
\State Send CM-CH-Confirm-Msg to the CMs
\Else
\State Listen for CM-CH-Offer-Msg until timer $T_{offer}$ expires.
\If{CM-CH-Offer-Msg is received}
\State Send CM-CH-Request-Msg to the CH
\State Listen for CM-CH-Confirm-Msg until timer $T_{confirm}$ expires
\Else
\While{CM-CH-Confirm-Msg not received}
\State Increase transmitting range and broadcast CM-CH-Request-Msg
\State Listen for CM-CH-Confirm-Msg until timer $T_{confirm}$ expires
\EndWhile
\EndIf 
\EndIf
\end{algorithmic}
\end{algorithm}

\begin{figure}[t]
\begin{center}
{\includegraphics[width=0.6\textwidth]{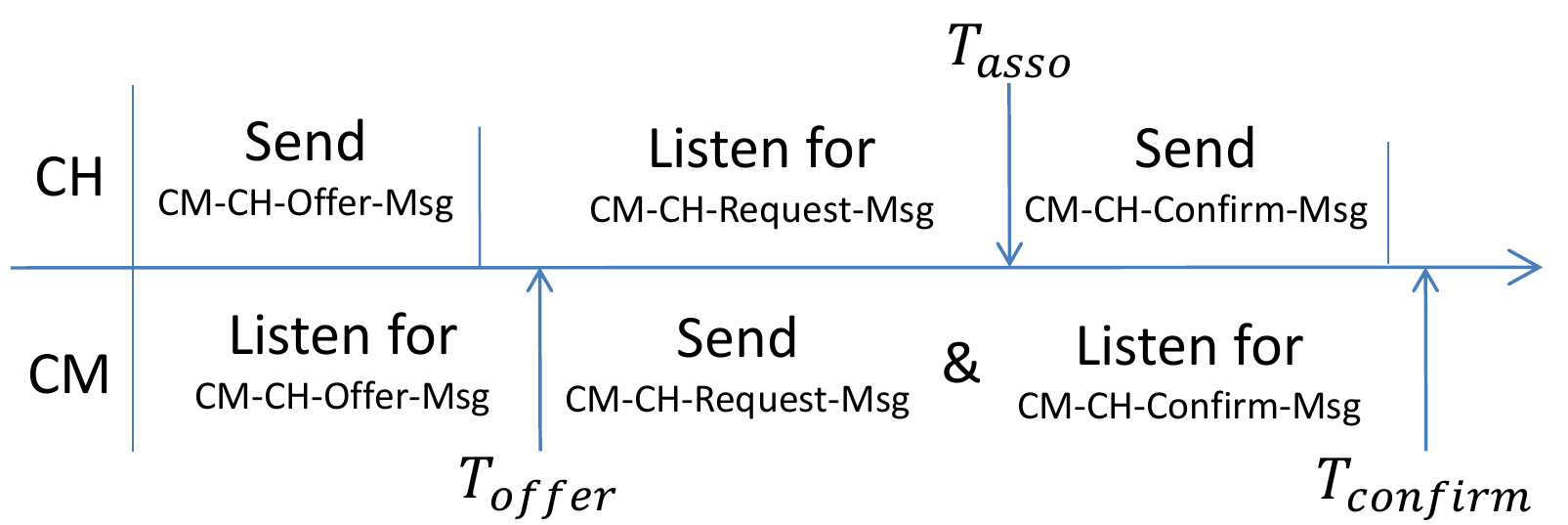}}
\caption{The operations of CH and CMs in CM-CH attachment.}
\label{fig:timer}
\end{center}
\end{figure}
\subsubsection{CH-CH association}
The next phase of the protocol is CH-CH association. The objective of CH-CH association is to find a route for each CH such that it can transmit its data to MS. Each CH broadcasts CH-CH-Offer-Msg containing its residual energy, the number of CMs attached to it (henceforce called CM size), ID and hop $l$ within a range of $\gamma R_0$ ($\gamma \geq 1$ is a parameter to adjust the transmitting range). The CH with higher hop, e.g., $l+1$, and hearing such message selects the one with the largest $\eta$ ($\eta$=residual energy/CM size) and replies a CH-CH-Confirm-Msg to the sender. The consideration behind this is that since the nodes are nonuniformly deployed, exactly the same CM sizes of two CHs cannot be guaranteed. Thus, the CH having small CM size is able to afford more data relaying tasks. Note the CH having lower hop, e.g., $l-1$, may also hear the CH-CH-Offer-Msg from the senders with hop $l$, but it will simply ignore the message. The CH-CH association algorithm is shown as Algorithm \ref{algorithm3}. Note, the CHs with hop 1 do not accept CH-CH-Offer-Msg as those with hop $l$, $2\leq l \leq h_{max}$, so Step 3-13 of Algorithm \ref{algorithm3} are skipped. 

After all the preparing procedures, the final task is to send data packets to CHs and delivery the buffered data from the CHs near MS's trajectory to MS. Different from delivering data to a static sink, the connection of a CH with MS varies since MS is moving. To assist the CHs to start and stop transmitting data packets to MS, MS periodically broadcasts Data-Request-Msg. The CH receiving Data-Request-Msg transmits data packets to MS; otherwise, it will not transmit data packets.

\begin{algorithm}[t]
\caption{CH-CH association, executed by CH $i$}\label{algorithm3}
\begin{algorithmic}[1]
\State $\eta \leftarrow$ 0, ID$\leftarrow$ 0
\State Broadcast CH-CH-Offer-Msg within $\gamma R_0$
\If{$h_i>1$}
\State Listen for CH-CH-Offer-Msg until $T_{offer}$ expires
\For{each CH-CH-Offer-Msg received from CH $j$}
\If{$\eta<$ $j$'s residual energy/CM size}
\State $\eta \leftarrow$ $j$'s residual energy/CM size
\State ID$\leftarrow$ $j$'s ID
\EndIf
\EndFor
\State Send CH-CH-Confirm-Msg to the CH with ID 
\State Listen for CH-CH-Confirm-Msg
\EndIf
\end{algorithmic}
\end{algorithm}
\section{Protocol Analysis}\label{analysis}
In this section, we present the analysis of the proposed protocol. Since the clusters and routing paths are constructed based on control message exchange, we first discuss the message complexity.

We consider the message complexity in terms of sensor nodes instead of MS. In the initial stage, all the sensor nodes forward Initial-Msg, transmit Return-Msg and forward Hop-Msg respectively. The messages add up to: $n+n+n$, i.e., the message complexity for the initial stage is $O(n)$. In the collecting stage, the worst case for CH election is that all $n$ sensor nodes are CH-candidates. In such case, each node broadcasts a CH-Compete-Msg. Suppose $m$ CHs are elected, then these CHs broadcast CH-Win-Msg; while the others broadcast CH-Quit-Msg. In the phase of CM-CH-Attachment, $m$ CM-CH-Offer-Msgs are transmitted first and then the other $n-m$ CMs send CM-CH-Request-Msgs. Third, $n-m$ CM-CH-Confirm-Msgs are returned. In the phase of CH-CH-Association, $m$ CHs broadcast CH-CH-Offer-Msg. Since the CHs with hop 1 can directly communicate with MS, they do not need to accept any CH-CH-Offer-Msg. Then, at most $m$ CH-CH-Confirm-Msgs are transmitted back. Therefore, in the worst case the messages add up to:
\begin{equation}
\underbrace{n+m+(n-m)}_\text{CH election}+\underbrace{m+2(n-m)}_\text{CM-CH}+\underbrace{2m}_\text{CH-CH}=4n+m
\end{equation}
i.e., the message complexity is $O(n)$. 

Now we consider the distribution of CHs. Due to Algorithm \ref{algorithm1}, every CH-candidate broadcasts a CH-Compete-Msg within its cover range. Any other CH-candidates within such range can receive this message. The one with the highest residual energy wins the competition and it broadcasts CH-Win-Msg. The CH-candidates that receive such message quit the competition by broadcasting CH-Quit-Msg. As mentioned in Section \ref{steady_stage}, in the competition phase, only the CH-candidates that have the highest residual energy within their own cover ranges can directly decide to be CHs; while all the other CH-candidates whose residual energies are not the largest have to wait for either CH-Win-Msg or CH-Quit-Msg from their neighbour CH-candidates, see Figure \ref{fig:competition} for an example. When CH-candidates $u$, $v$, $w$ first exchange CH-Compete-Msg, only $v$ can decide to be CH; while both $u$ and $w$ need to wait for further messages. Then, $v$ broadcasts CH-Win-Msg and $u$ receives such message and broadcasts CH-Quit-Msg. Finally, $w$ receives CH-Quit-Msg. Since $w$ has only neighbour $u$, it decides to be CH. In the end of this competition, within the cover ranges of $v$ and $w$, there is no other CHs. To make the conclusion generally: after CH competition process, it is impossible that two CHs are within each other's cover range.

\begin{figure}[t]
\begin{center}
{\includegraphics[width=0.5\textwidth]{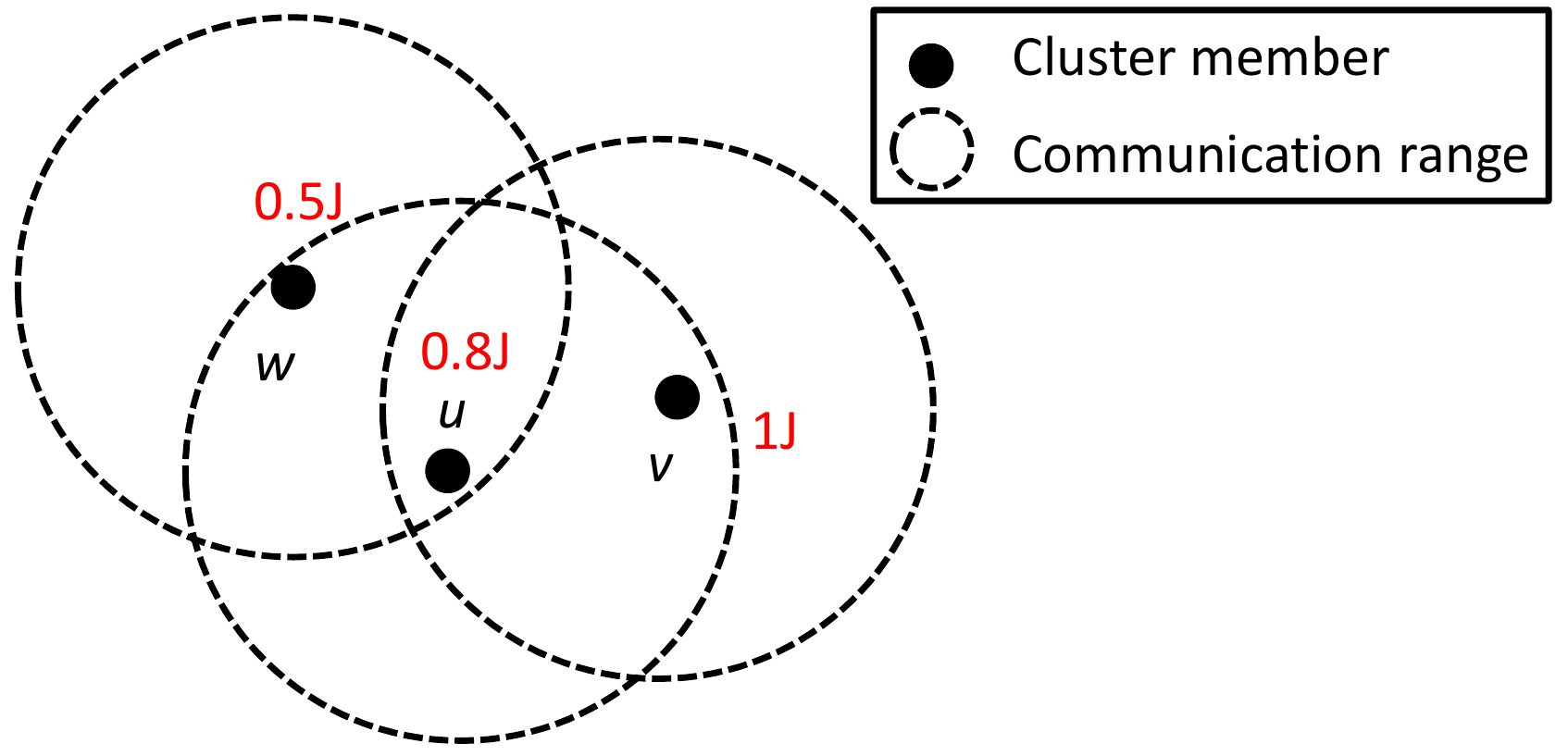}}
\caption{CH competition process. $v$ and $w$ win the competition and become CHs.}
\label{fig:competition}
\end{center}
\end{figure}

Another feature of the proposed protocol is that every sensor node is covered by exactly one CH. This feature can be obtained from the phase of CM-CH attachment. On the one hand, if a sensor node is a CH, it is covered by itself. On the other hand, if it is a CM, it listens for CM-CH-Offer-Msg and then joins the CH which sends the message. Since the CHs are elected based on a random manner, there are cases where one or several CMs cannot hear CM-CH-Offer-Msg within $T_{offer}$. In such cases, the CM proactively sends CM-CH-Request-Msg until it hears CM-CH-Confirm-Msg, and then it joins the CH. In this way, the CHs are able to cover all the CMs.

Finally, we discuss the implementation of our protocol. Since the hop distance from a node to MS' trajectory influences the cover range, our protocol requires all the nodes have the correct hop distances. The period for broadcasting Initial-Msg impacts on the hop distance construction. If a node that should receive Initial-Msg is not covered by two successful broadcasting, it may get a wrong hop distance. An example is demonstrated in Figure \ref{fig:broadcast}. We notice that this issue relates to the localizations of sensor nodes, the broadcasting frequency as well as MS moving speed. Since many applications involve randomly deployment and MS is not energy constrained in our scenario, to address this issue, we assume MS broadcasts in a high frequency such that all the nodes should hear Initial-Msg correctly, instead of making other assumptions.

\begin{figure}[t]
\begin{center}
{\includegraphics[width=0.45\textwidth]{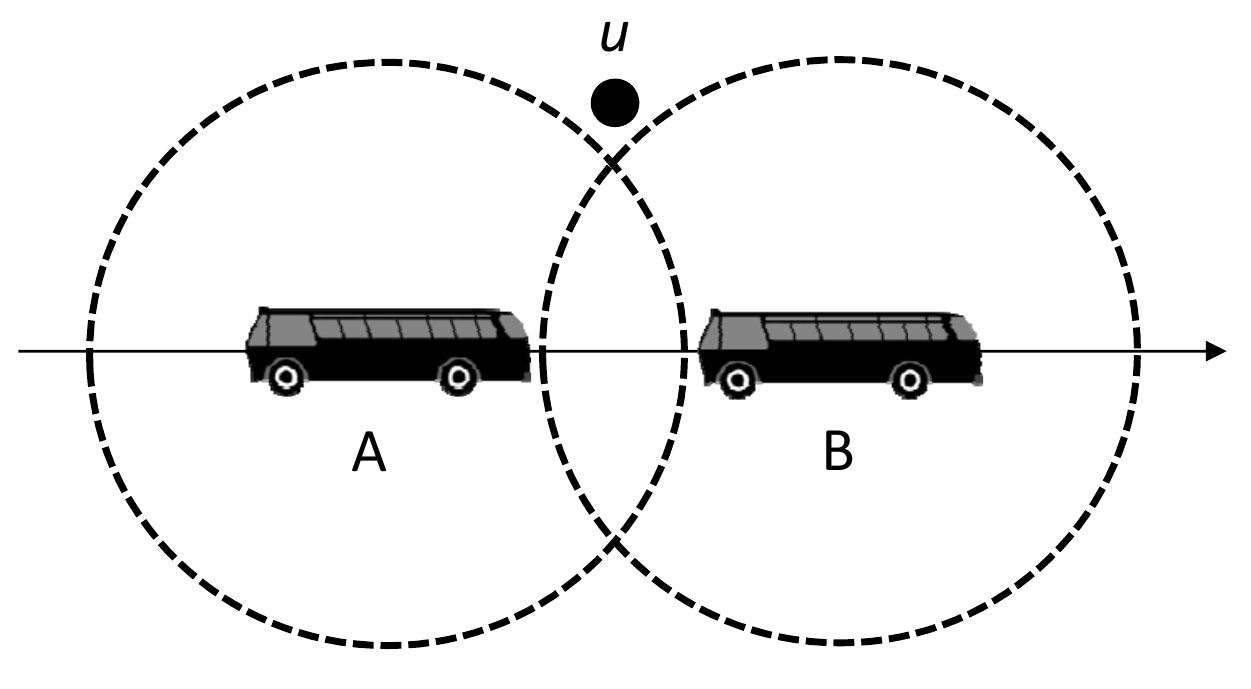}}
\caption{MS broadcasts Initial-Msg when it is at A and it broadcasts the next Initial-Msg when it arrives at B. Thus, node $u$ fails to hear such message.}
\label{fig:broadcast}
\end{center}
\end{figure}
\section{Simulation Results}\label{simulation4}
We evaluate the performance of our protocol by simulations under two scenarios. Scenario 1: uniform node distribution shown as Figure \ref{fig:uniform} and Scenario 2: nonuniform node distribution shown as Figure \ref{fig:nonuniform}. In Scenario 1, the nodes have approximate local density. In contrast, in Scenario 2, the nodes on the left part have larger densities than the right part. We simulate the proposed protocol using MATLAB with the networks shown in Figure \ref{fig:scenario} and the below mentioned parameters. Since the CH election is in a random manner, we execute simulation for each set of parameters independently 100 times. The results shown in this section are the average results.

\begin{figure}[t]
    \centering
    \begin{subfigure}[t]{0.4\textwidth}
        \includegraphics[width=\textwidth]{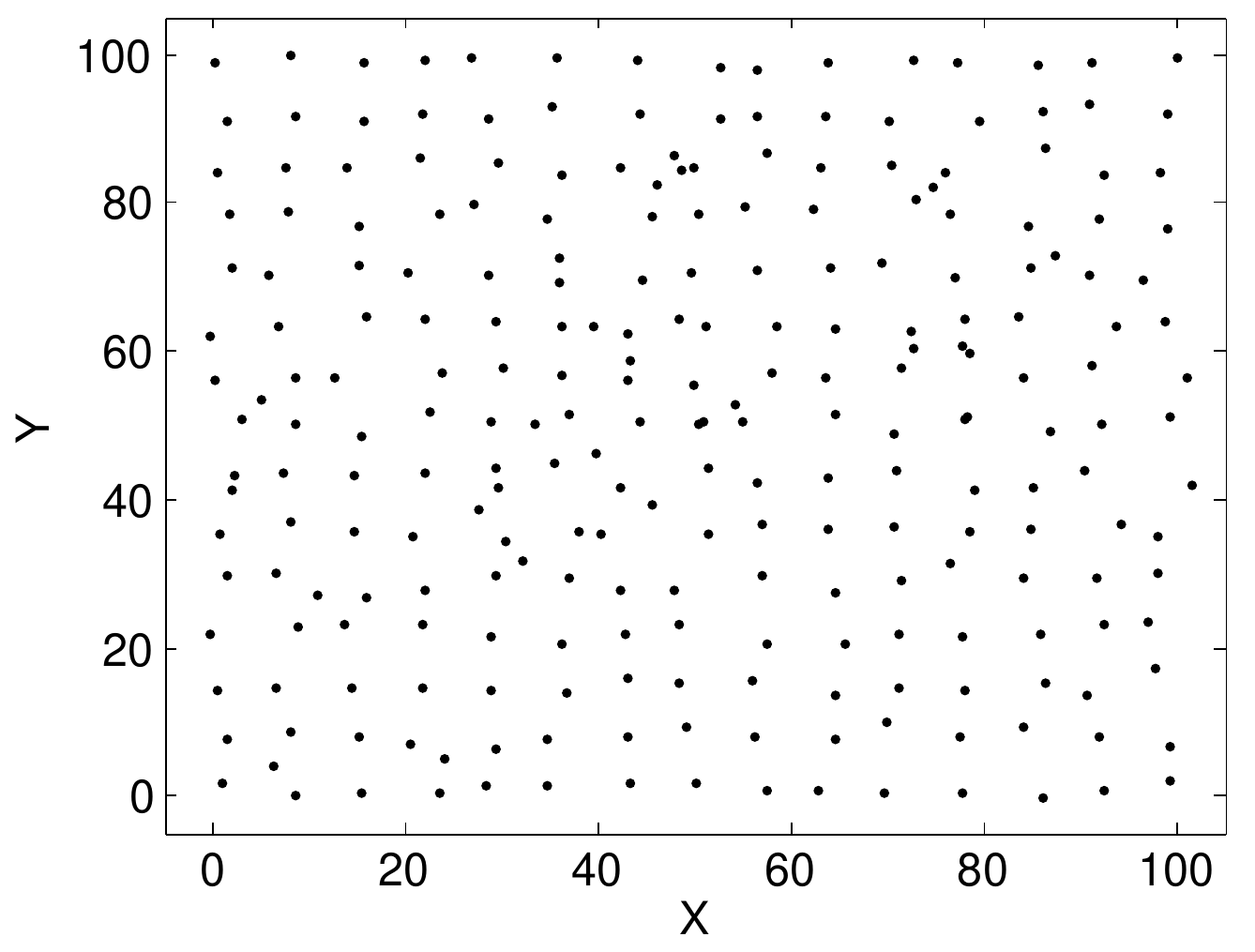}
        \caption{}
        \label{fig:uniform}
    \end{subfigure}
    ~
    \begin{subfigure}[t]{0.4\textwidth}
        \includegraphics[width=\textwidth]{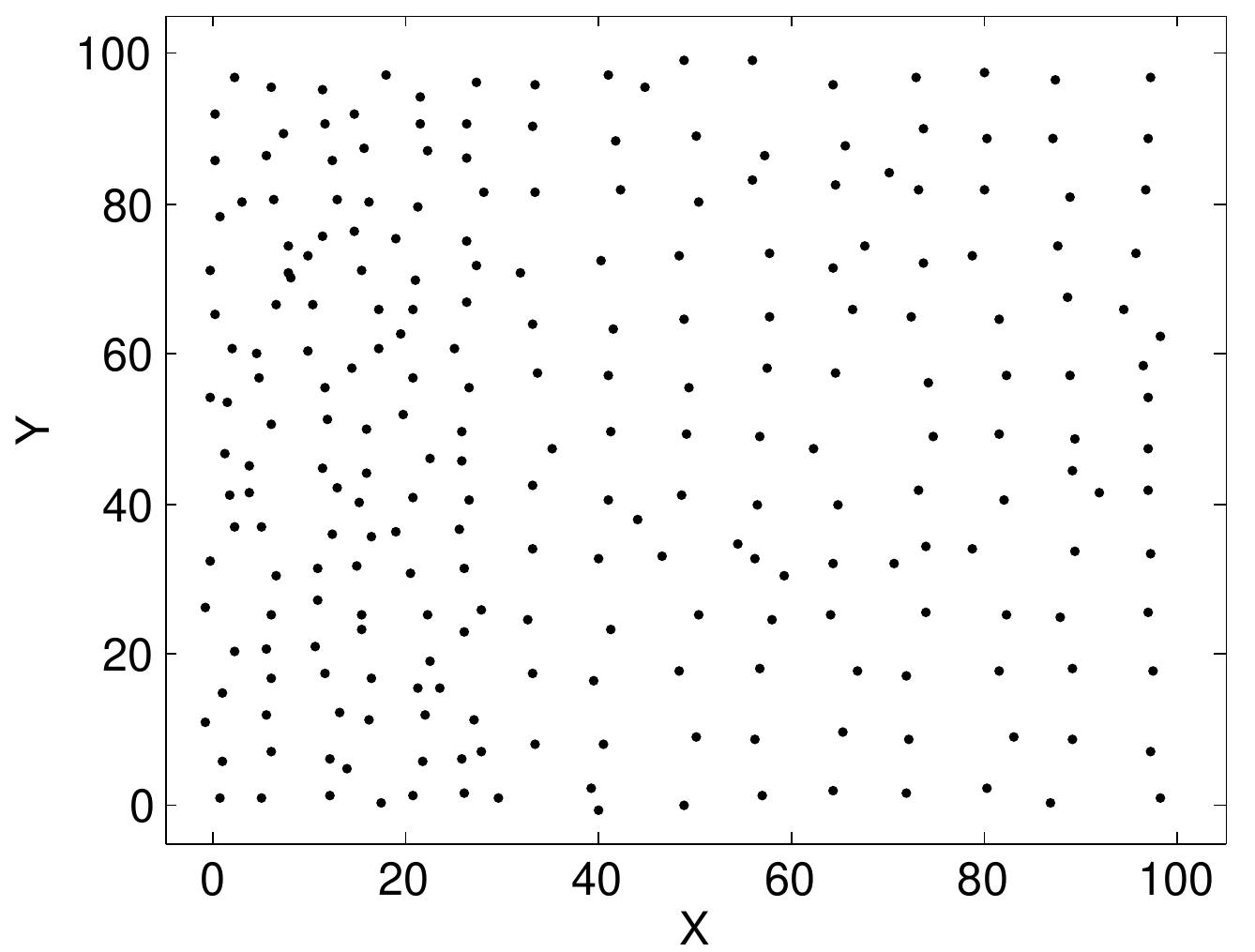}
        \caption{}
        \label{fig:nonuniform}
    \end{subfigure}
    \caption{Topology in each scenario. (a) Uniform node distribution; (b) Nonuniform node distribution.}\label{fig:scenario}
\end{figure} 

The parameters used in this section are summarized in Table \ref{table:parameter}. The parameters for energy consumption model are consistent with \cite{CHEN09, Shokouhifar15}.
\begin{table}[t]
\begin{center}
\caption{Parameters of simulations}\label{table:parameter} 
  \begin{tabular}{| l|  l|}
    \hline
    Parameter & Value\\ \hline
    Sensor field & $100 m \times 100 m$ \\
    Number of nodes & 250\\
    Data packet size & 4000 $bits$\\
    Control message size & 400 $bits$\\  
    Initial energy & $1J$\\ 
    $E_{elec}$ & $50 nJ/bit$\\ 
    $E_{fs}$ &  $10 pJ/(bit \ m^2)$\\
    $E_{mp}$ & $0.0013 pJ/(bit \ m^4)$\\
    $E_{sens}$ & $1 nJ/bit$\\ 
    $E_{aggr}$ & $5 nJ/(bit\ packet)$\\ 
    $d_0$ & $87m$\\
    \hline
    \end{tabular}
\end{center}
\end{table}

\subsection{Parameter impacts}
There are several parameters in our protocol, namely $\alpha$, $\beta$, $\gamma$, $k$, $R_0$. $\beta$ is the threshold for whether a node can act as CH-candidate. The larger the $\beta$, the smaller the possibility for a node to become CH-candidate. we fix $\beta$ as 0.1. $\gamma$ is to adjust the communication range for a CH in CH-CH-association. The larger the $\gamma$, the more neighbour CHs can be found. However, large $\gamma$ leads to energy waste, since energy consumption is a function of transmitting distance. Also too large $\gamma$ may result in that all the CHs are associated to one CH with lower hop and having the largest $\eta$ value, which heavily increases the relaying burden of this CH. we fix $\gamma$ as 2. As mentioned in Section \ref{steady_stage}, the larger the $k$, the fewer control messages need to transmit. But too large $k$ may make some CHs die within a cycle. We consider the influence of $k$ on the lifetime of a cluster.  Here, we select one cluster and test the number of rounds it can operate with the parameters in Table \ref{table:parameter}. As shown in Fig. \ref{fig:k_impact}, the lifetime increases with $k$ because larger $k$ means more energy is spent on data transmission and less on overhead. We can also see that when $k$ is smaller than 20, the lifetime raises quickly; after that it increases slightly. Thus, we fix $k$ as 30 since it gives a relative long lifetime.

\begin{figure}[t]
\begin{center}
{\includegraphics[width=0.5\textwidth]{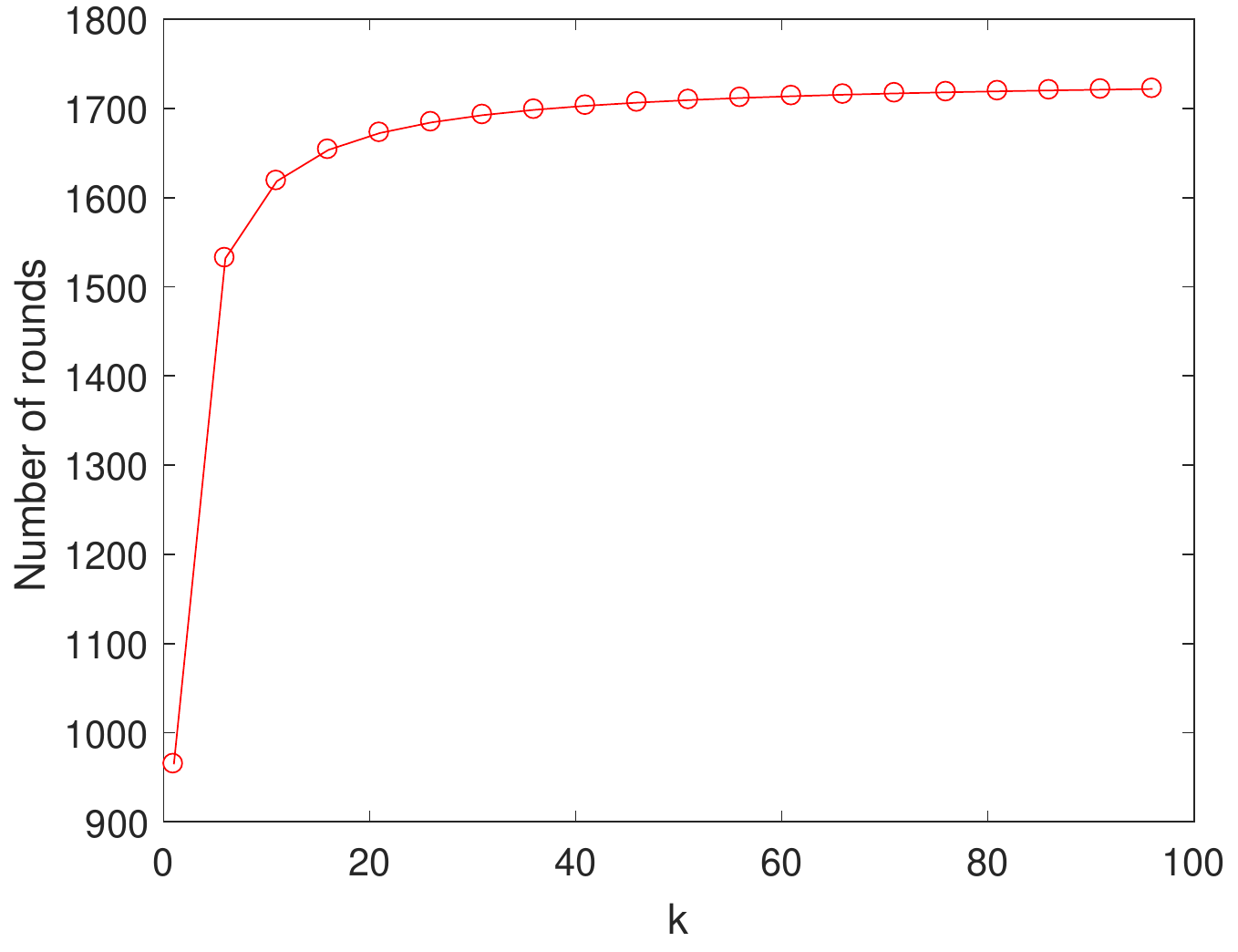}}
\caption{The impacts of $k$ on the number of rounds.}
\label{fig:k_impact}
\end{center}
\end{figure}

Next we study how $\alpha$ and $R_0$ influence the protocol performance. $\alpha$ ranges from 0.2 to 0.6 and $R_0$ is between 10 and 50. We consider the number of clusters under different parameter sets in both Scenario 1 and 2. The result is shown in Figure \ref{fig:ch}. For a fixed $\alpha$, the number of clusters decrease with the increasing of $R_0$. According to (\ref{cover_range}), we know when $\alpha$ is fixed, cover range increases with $R_0$. It testifies our design purpose, i.e., the larger the cover range, the smaller the number of clusters. Besides, for a fixed $R_0$, the number of clusters increase with the increasing of $\alpha$. The reason behind is as follows. According to (\ref{cover_range}), given $R_0$, $\alpha$ influences the range of $R_i$. The larger the $\alpha$, the smaller the lower bound of $R_i$. In other words, $R_i$ can take more values when $\alpha$ is set as a larger value. Thus, the number of clusters also increases with $\alpha$. 


\begin{figure}[t]
    \centering
    \begin{subfigure}[h]{0.45\textwidth}
        \includegraphics[width=\textwidth]{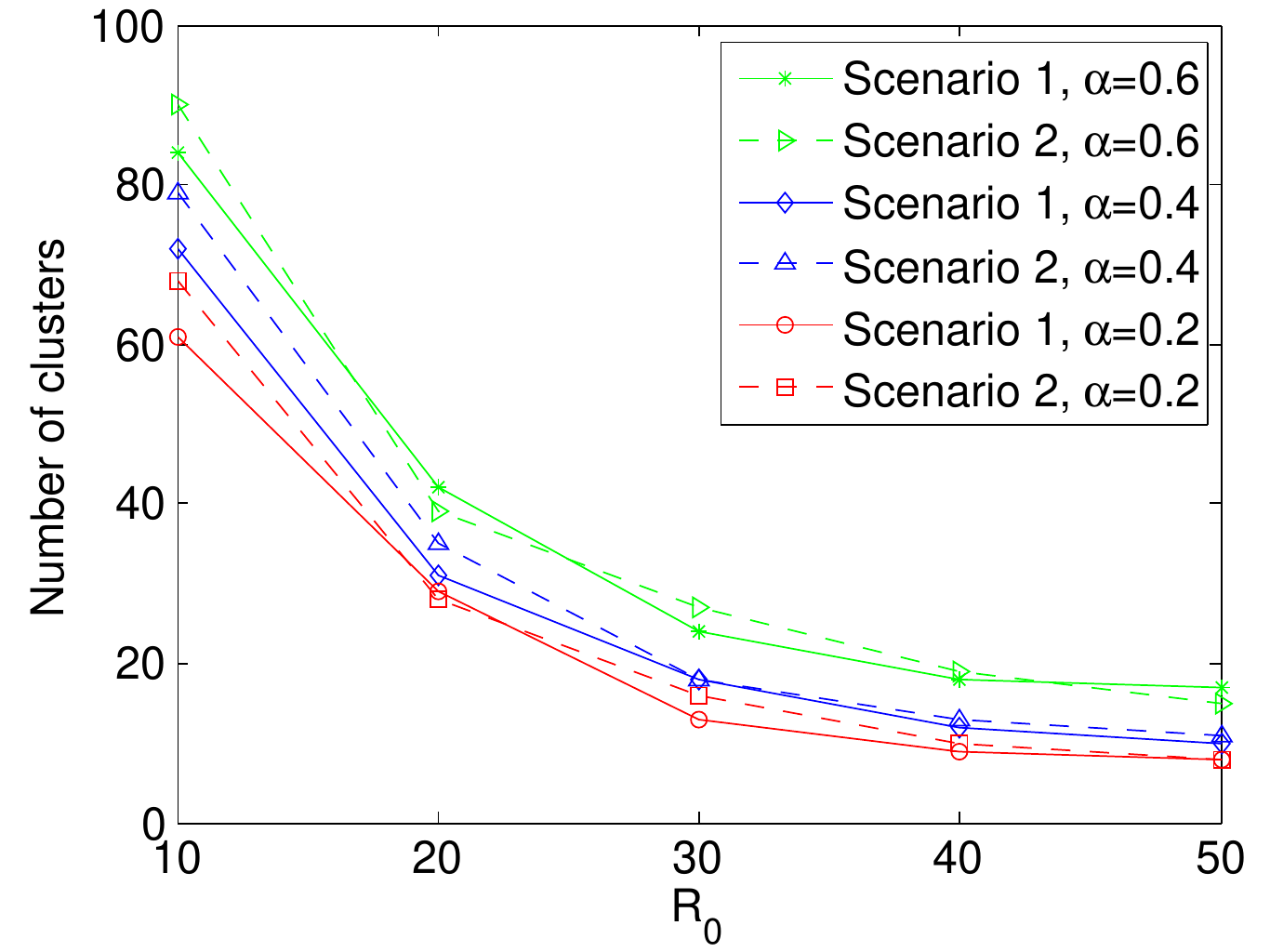}
        \caption{}
        \label{fig:ch}
        \end{subfigure}
        \
    \begin{subfigure}[h]{0.45\textwidth}
        \includegraphics[width=\textwidth]{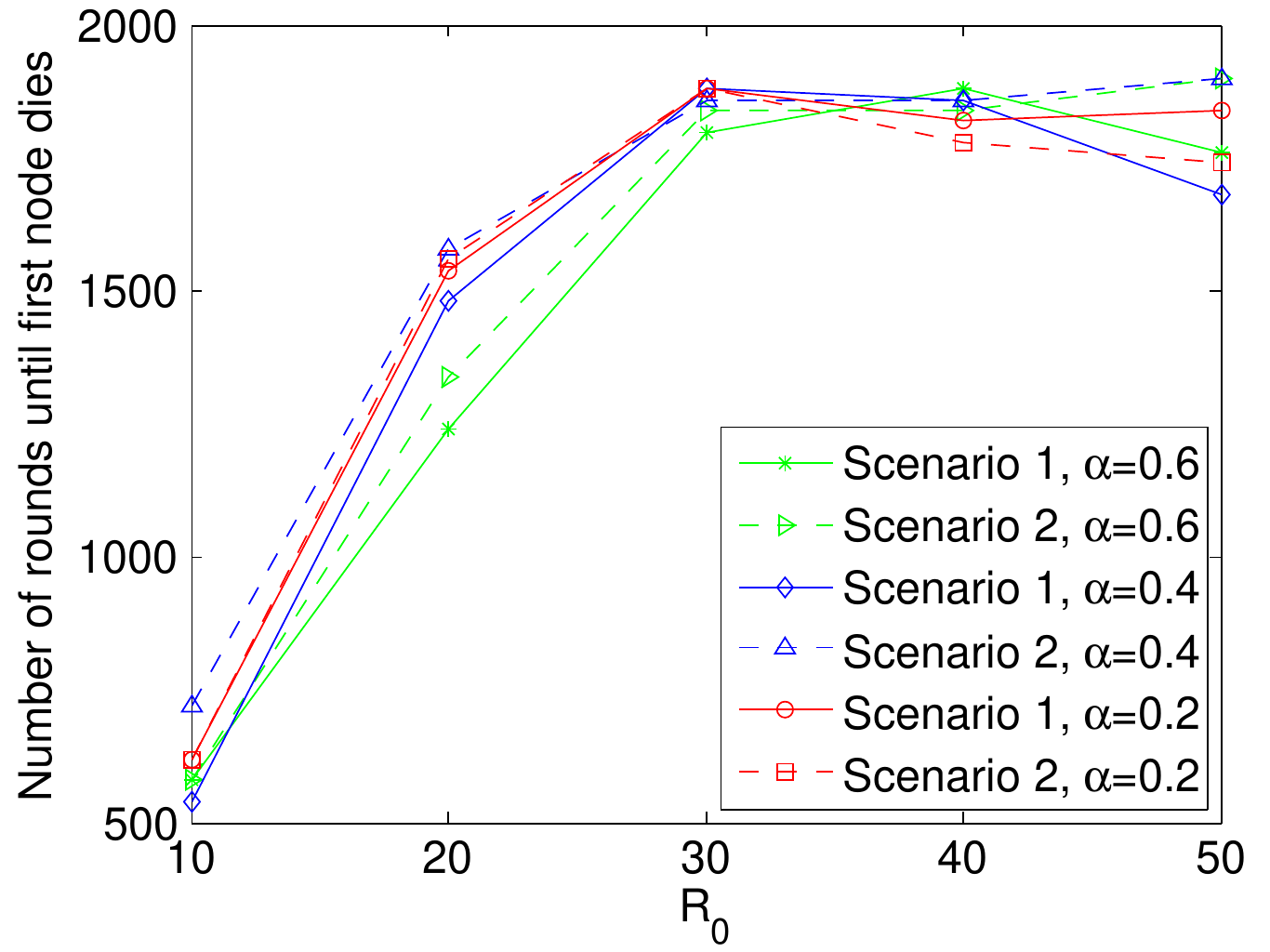}
        \caption{}
        \label{fig:round}
    \end{subfigure}
    \caption{The impacts of $\alpha$ and $R_0$ on. (a) the number of clusters; (b) the number of rounds until the first node dies.}\label{fig:average_energy_consumption}
\end{figure}

Now we consider the network lifetime under different parameter sets. The network lifetime is significant in many applications since the sensing field cannot be fully monitored once a node dies. The result is shown in Figure \ref{fig:round}. For a fixed $\alpha$, with the increase of $R_0$, the network lifetime tends to increase. The reason behind is as follows. As we assume that each CH aggregates the data packets received from its CMs into one packet, the larger the number of clusters, the more data packets are transmitted, resulting in a larger amount of energy consumption. This analysis is consistent with the result shown in Figure \ref{fig:round}, i.e., the network lifetime increase with the increase of $R_0$.


\subsection{Comparison with existing work}
Since the advantage of applying MS to collect data in WSNs has been shown in the literature, to make a fair comparison we compare our approach with those which consider the scenario of using path constrained MS. Here, we compare with \cite{GAO11} and \cite{KONS12} which are two representative approaches for data collection based on flat and cluster structure respectively. The approach in \cite{KONS12} determines the cover range for each CH according to the Euclidean distance to MS' trajectory. Thus, it requires the sensor nodes to able to extract the distance information from the received signal strength. Further, it also requires all the sensor nodes are within the range of MS. Thus, considering the size of the sensing field, MS communication range is set to be $100m$. In contrast, since our approach and \cite{GAO11} use only hop distance, the communication range of MS is set as $30m$. Since the approach in \cite{KONS12} constructs the clusters with two sizes and considering the sensing field, the communication range of sensor nodes is set as $60m$. The communication ranges of sensor nodes for the approach in \cite{GAO11} and ours are set as the same with MS. The basic features of our approach, \cite{GAO11} and \cite{KONS12} are shown in Table \ref{table:technique}. 

\begin{table}[t]
\begin{center}
\caption{Comparison of approaches}\label{table:technique} 
  \begin{tabular}{| c | c |c |c|}
    \hline
    Approach & MASP \cite{GAO11} & MobiCluster \cite{KONS12} & Ours\\\hline
     Distance measure& No & Yes & No\\
     Network structure& Flat & Cluster & Cluster\\
     MS comm. range& $30m$ & $100m$ & $30m$ \\ 
     Node comm. range& $30m$ & $60m$ & $30m$ \\
    \hline
    \end{tabular}
\end{center}
\end{table}

\begin{figure}[t]
    \centering
    \begin{subfigure}[h]{0.45\textwidth}
        \includegraphics[width=\textwidth]{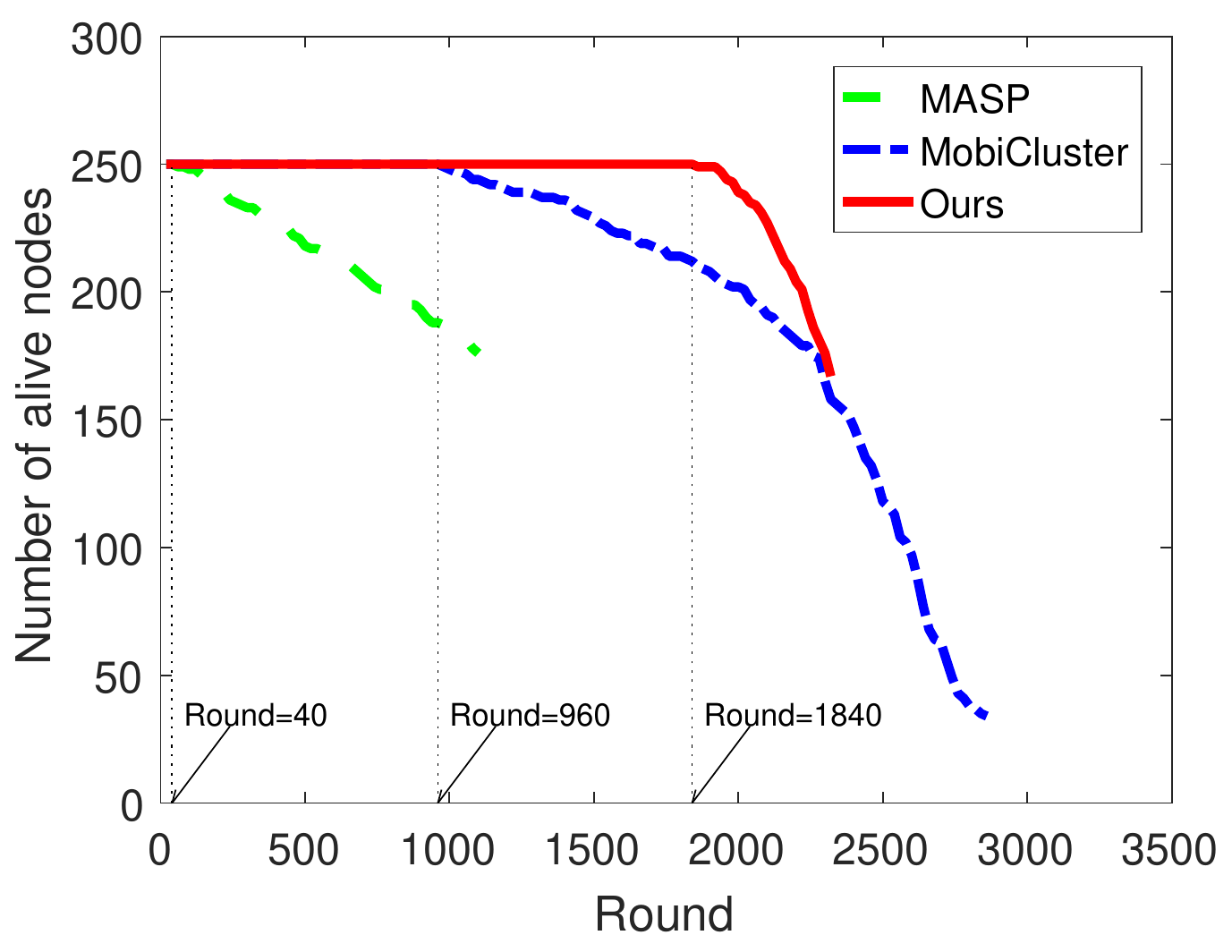}
        \caption{}
        \label{fig:alive_node_uniform}
    \end{subfigure}
    \
    \begin{subfigure}[h]{0.45\textwidth}
        \includegraphics[width=\textwidth]{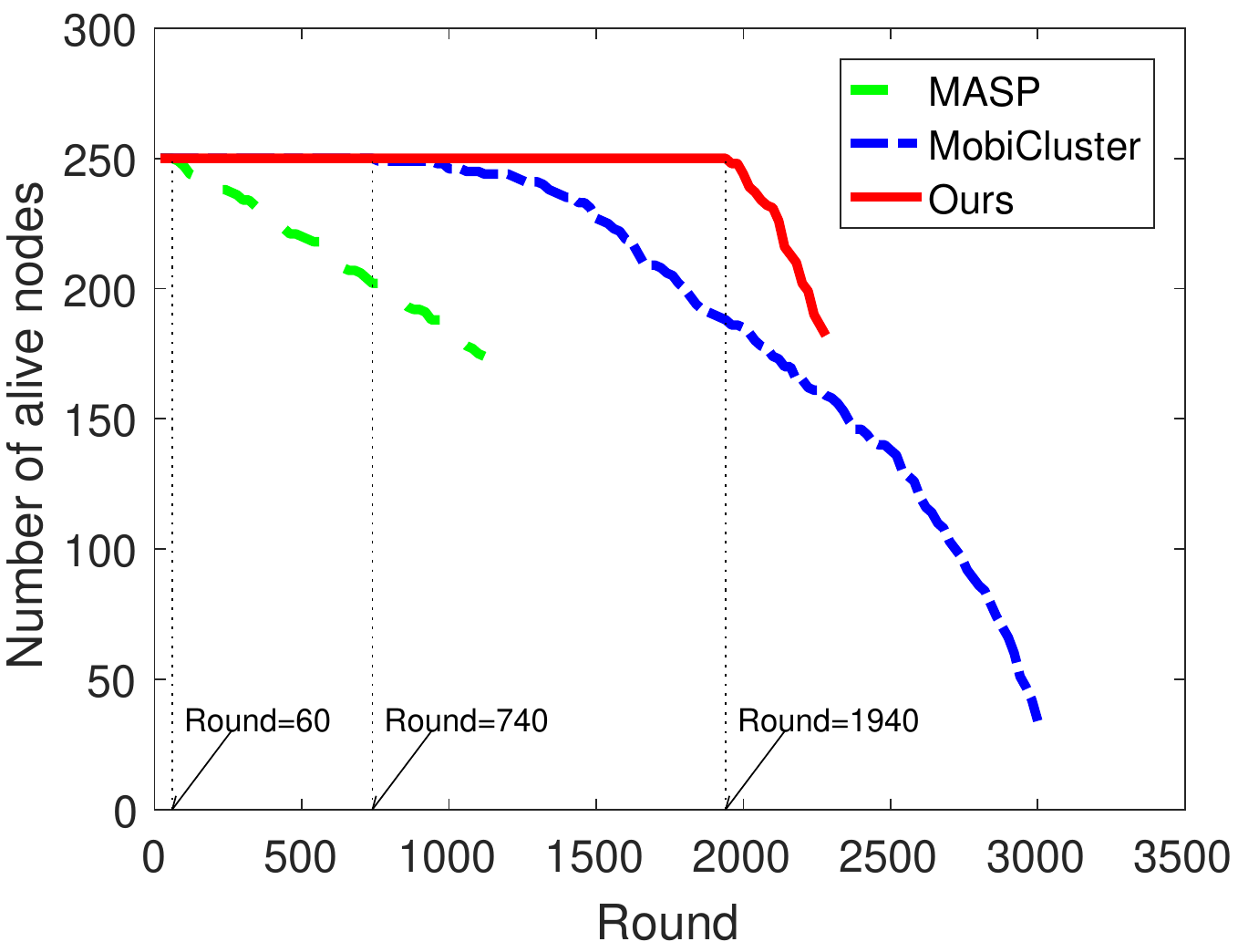}
        \caption{}
        \label{fig:alive_node_nonuniform}
    \end{subfigure}
    \caption{Number of alive nodes in each scenario. (a) Uniform node distribution; (b) Nonuniform node distribution.}\label{fig:alive_node}
\end{figure} 

Here our performance measure is still the network lifetime. The results of our approach and the compared ones are shown in Figure \ref{fig:alive_node}. In both Scenario 1 and 2, the proposed approach achieves the longest network lifetime, which guarantees the network to have good coverage of the interested areas for a long time. Besides, with the same initial energy amount, no matter if the nodes are uniformly distributed or nonuniformly distributed, the networks can achieve the similar network lifetime using the proposed approach. However, the approach of \cite{KONS12} performs differently, i.e., in the uniform node distribution case, it can operate for 960 rounds with all the nodes alive, while in the nonuniform node distribution case, it can only perform 740 rounds. The reason is that it does not take node density into account. Since the approach of \cite{GAO11} is based on flat structure, the node distribution impacts little. It performs similarly in both scenarios and both the worst, because the subsinks definitely suffer from funnelling effect issue. Comparing the performance of our approach and \cite{KONS12}, it can be seen that after the first node dies in these two cases, the number of alive nodes in our protocol drops more dramatically than that of \cite{KONS12}. This is because in our protocol the energy consumption of the sensor nodes is more balanced than that of \cite{KONS12}. Figure \ref{fig:alive_node} also plots how many rounds each protocol can operate under the connectivity requirement\footnote{Connectivity requirement: The sensory data packet at each alive sensor node can be transmitted MS. In other words, every alive sensor node should have at least one neighbour node.}.  Under this context, the approach of \cite{KONS12} can operate for the longest rounds, because the sensor nodes' communication range is larger. Our approach and the one of \cite{GAO11} cannot work that long due to relatively short communication ranges. Although it is possible to raise the operation rounds by increasing the communication range, it is not necessary since the network has already lost full coverage of the sensing field. 

\begin{figure}[t]
    \centering
    \begin{subfigure}[h]{0.45\textwidth}
        \includegraphics[width=\textwidth]{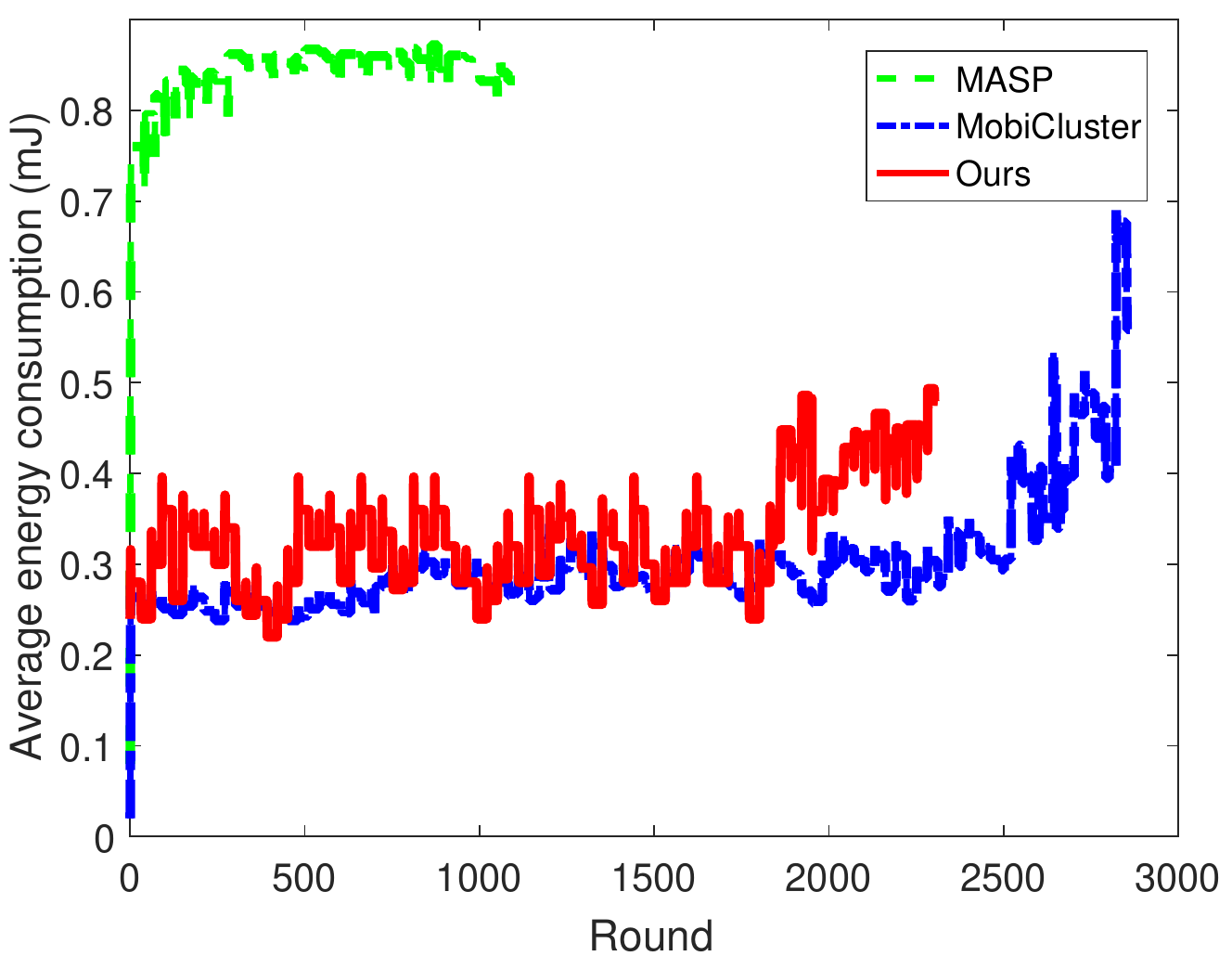}
        \caption{}
        \label{fig:average_energy_consumption_1}
        \end{subfigure}
        \
    \begin{subfigure}[h]{0.45\textwidth}
        \includegraphics[width=\textwidth]{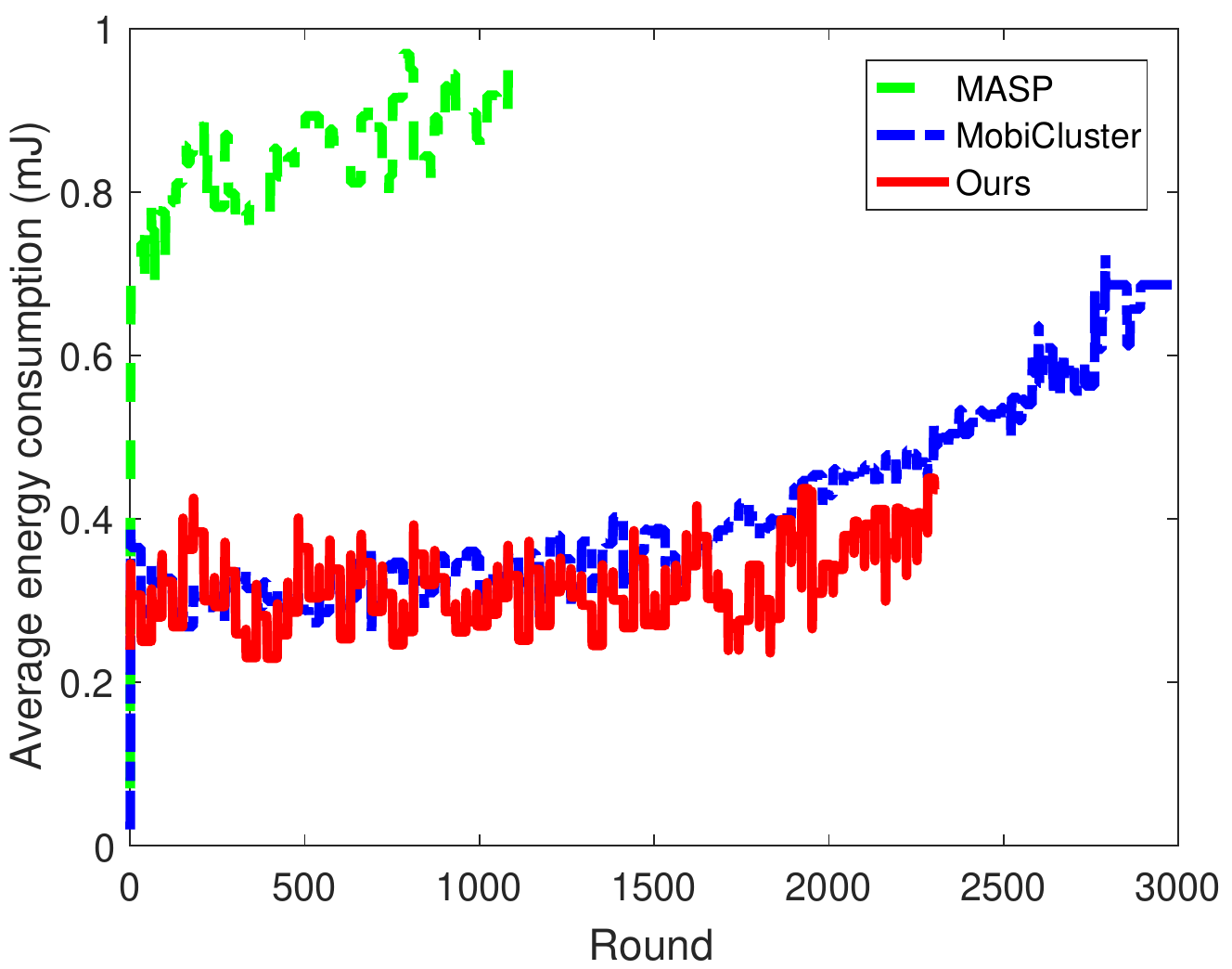}
        \caption{}
        \label{fig:average_energy_consumption_2}
    \end{subfigure}
    \caption{Average energy consumption by the alive nodes. (a) Uniform node distribution; (b) Nonuniform node distribution.}\label{fig:average_energy_consumption}
\end{figure}
\begin{figure}[t]
    \centering
    \begin{subfigure}[h]{0.45\textwidth}
        \includegraphics[width=\textwidth]{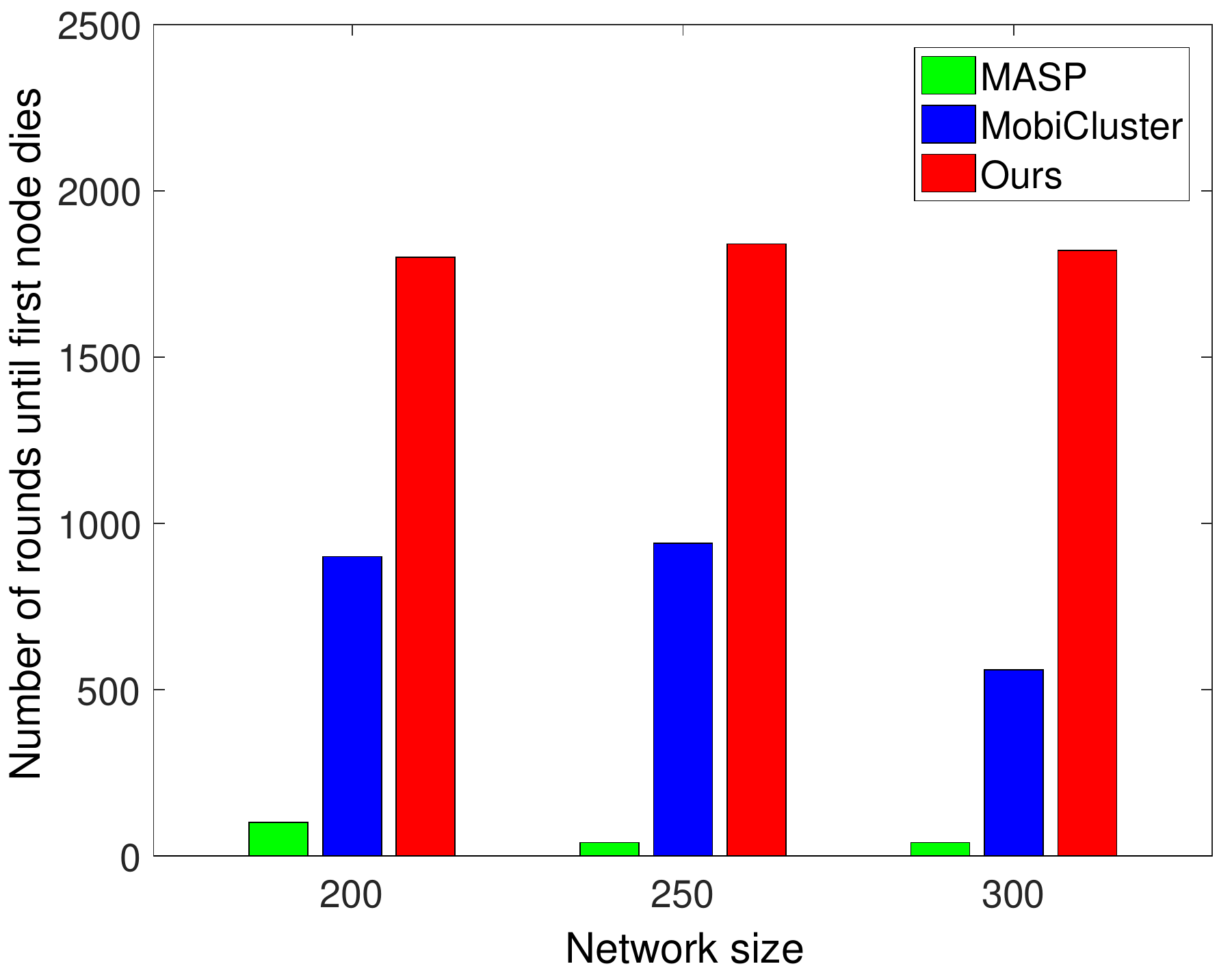}
        \caption{}
        \label{fig:alive_node_vs_network_size_uniform}
        \end{subfigure}
        \
    \begin{subfigure}[h]{0.45\textwidth}
        \includegraphics[width=\textwidth]{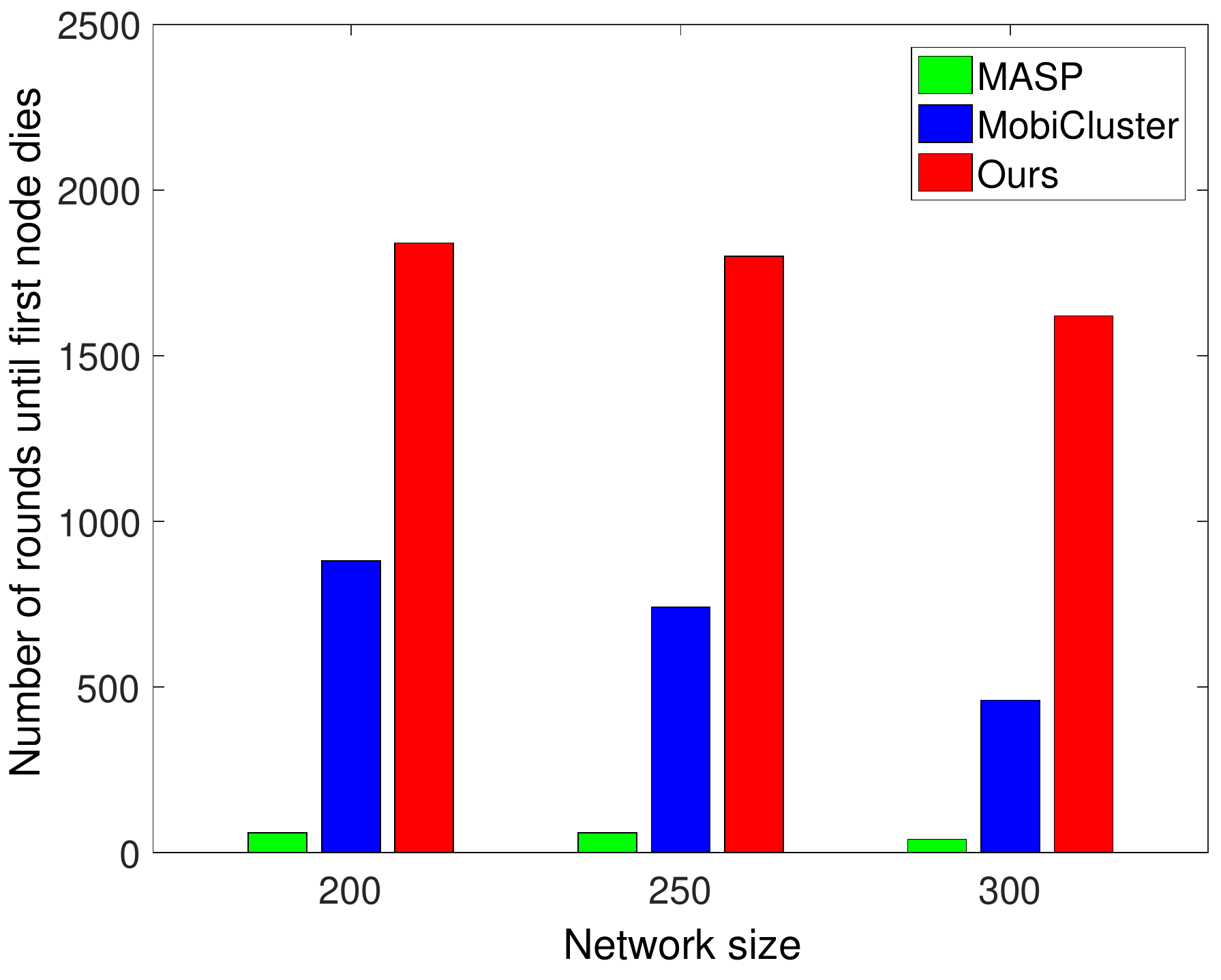}
        \caption{}
        \label{fig:alive_node_vs_network_size_nonuniform}
    \end{subfigure}
    \caption{Network lifetime over network sizes. (a) Uniform node distribution; (b) Nonuniform node distribution.}\label{fig:alive_node_vs_network_size}
\end{figure}

\begin{figure}[t]
    \centering
    \begin{subfigure}[h]{0.45\textwidth}
        \includegraphics[width=\textwidth]{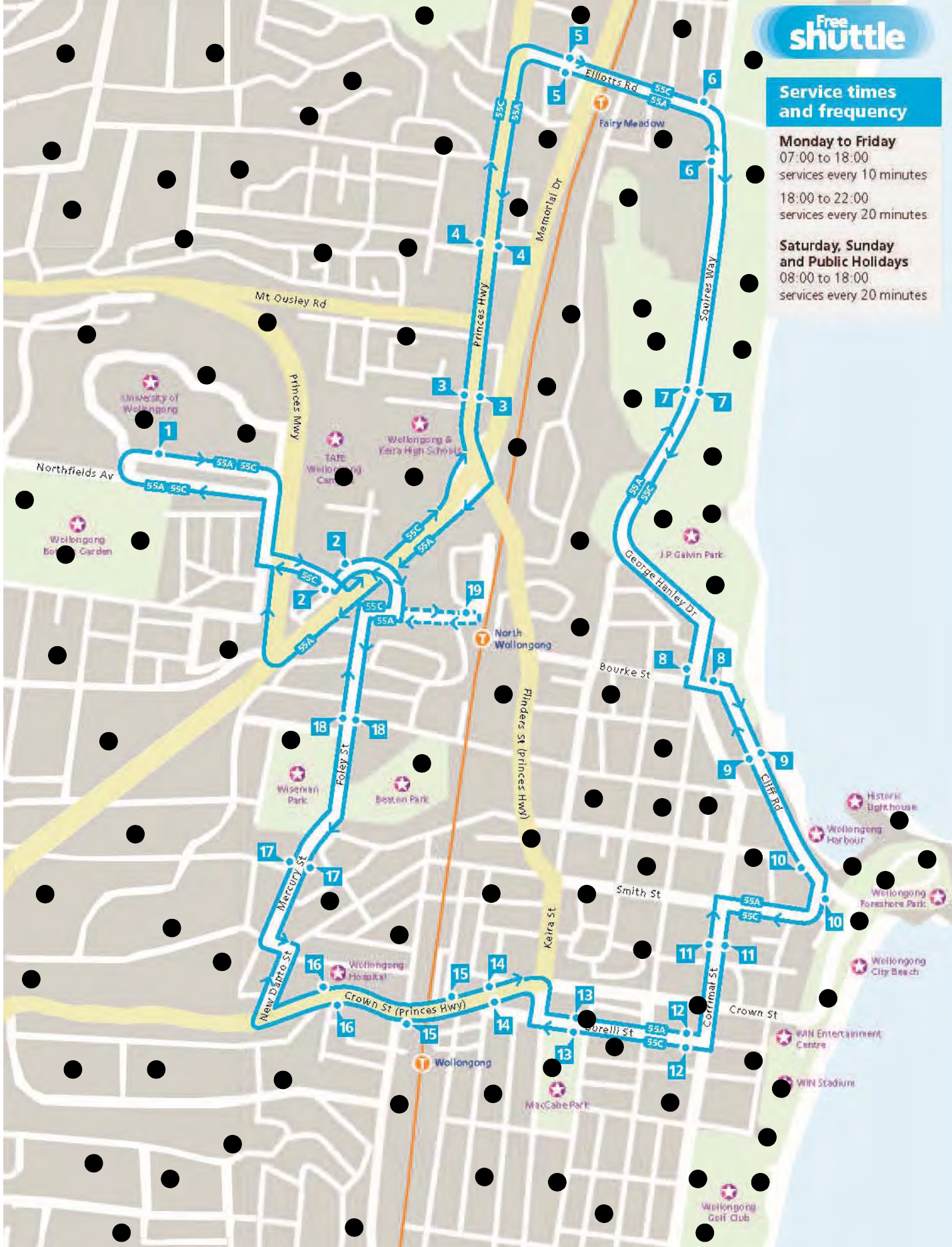}
        \caption{}
        \label{gongshuttlemap}
        \end{subfigure}
        \
    \begin{subfigure}[h]{0.45\textwidth}
        \includegraphics[width=\textwidth]{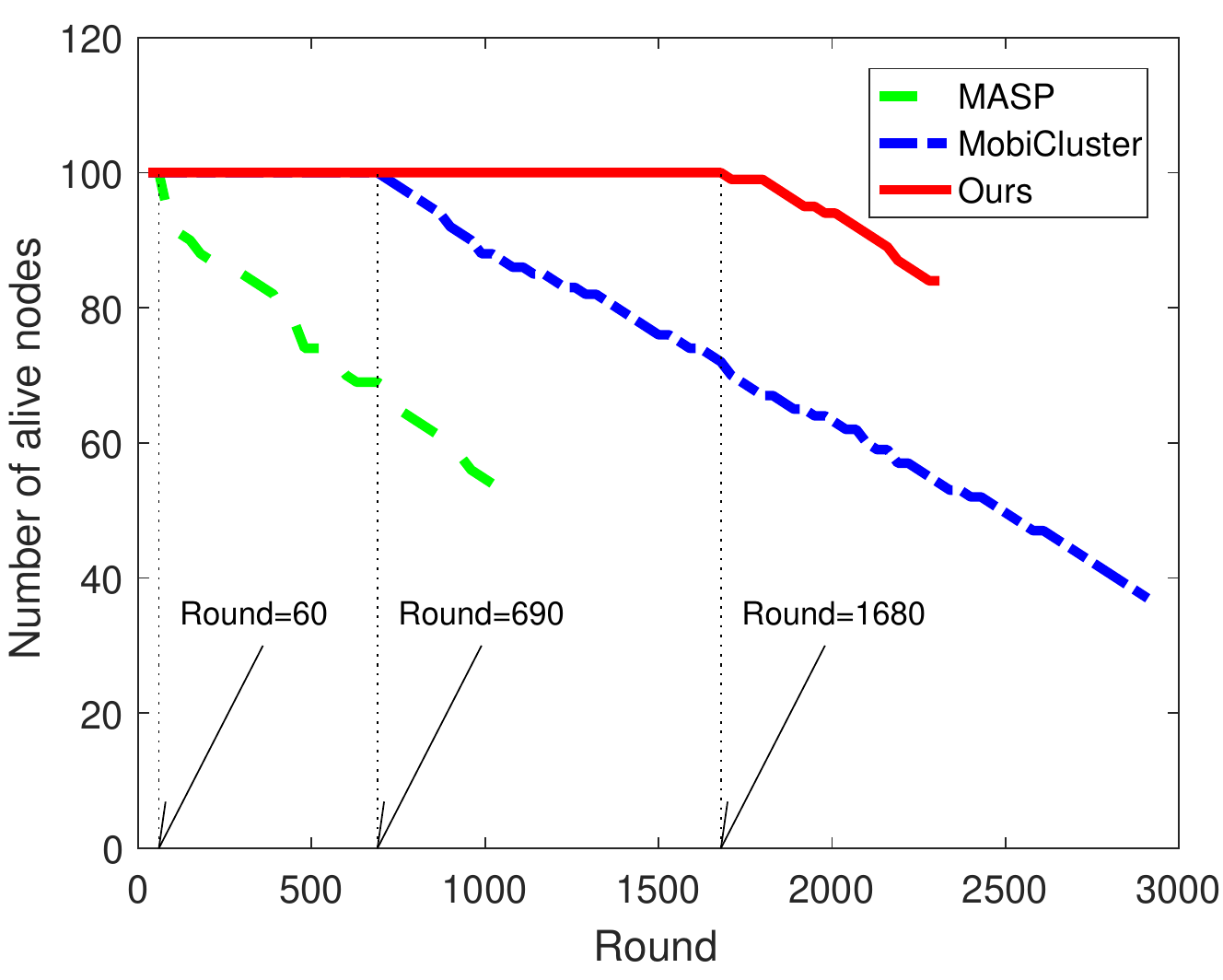}
        \caption{}
        \label{practical}
    \end{subfigure}
    \caption{(a) The sensing area covered by the route of a shuttle and a set of sensor nodes; (b) Number of alive nodes.}\label{fig:alive_node_vs_network_size}
\end{figure}

%

Accordingly, the average energy consumptions by the alive nodes in these two scenarios are displayed in Figure \ref{fig:average_energy_consumption}. The values on Round 0 represents the average energy consumption in the initial phase while those after are in the collection phase. The energy consumption at the beginning of each cycle, i.e., for clustering, is added to the first round. In the initial phase, the proposed approach consumes more energy than the alternatives, since it requires three rounds of control message exchange. In the collection phase, the average energy consumption increases when some nodes run out of energy. The fundamental reason is that the average transmission distance increases when some nodes die. Comparing Figure \ref{fig:average_energy_consumption_1} and \ref{fig:average_energy_consumption_2}, we can see that the propose approach consumes less energy than those of \cite{GAO11, KONS12} in the case of nonuniform distribution.

We further consider four other networks with 200 and 300 nodes uniformly  and nonuniformly deployed respectively. We also apply the proposed approach and the compared ones on these networks. As shown in Figure \ref{fig:alive_node_vs_network_size}, the network lifetimes with our proposed protocol are the longest. 

Moreover, we provide a simulation based on the real trace of the shuttle in Wollongong, Australia. The shuttle route is shown in Figure \ref{gongshuttlemap} and the timetable can be found on the website of Transport of NSW. 100 sensor nodes are nonuniformly deployed in this $3\times 5$ km area (the right part area is with higher density than the left part area). We apply the proposed approach as well as the two alternatives. For this practical case, the buffer overflow may occur since the shuttle does not operate at night. Here we do not consider the buffer overflow issue and only consider the network lifetime. The simulation results are shown in Figure \ref{practical}. We can see that the proposed approach achieves the longest network lifetime. 

\section{Summary}\label{conclusion4}
We propose a cluster-based routing protocol to support data collection in WSNs with nonuniformly node distribution using a path fixed MS. It involves an energy-aware and unequal clustering algorithm and an energy-aware routing algorithm. The clustering algorithm constructs clusters with unequal ranges based on the local node density and hop distance to MS's trajectory. The resulting clusters have the features: 1) at the same hop distance, the clusters in dense (sparse) area have small (large) cover ranges; 2) with the same density, the clusters close to (far away from) MS's trajectory have small (large) cover ranges. Such formation helps balance the energy consumption among clusters. In the routing algorithm, each CH selects a CH with lower hop distance to MS' trajectory and highest residual energy/CM size as the relay CH. By using the above techniques, our protocol works well for WSNs with nonuniformly node distribution and prolongs the network lifetime significantly comparing to existing work. Moreover, since Euclidean distance measure is not required, the system cost is low, which makes the proposed protocol have good scalability. The proposed approach is for a single MS. When more MSs are available, the cooperation between them may help to increase the performance further, which is our future consideration.

In this chapter, we study the cluster-based routing protocol assisting by a single ME. Local aggregation is assumed at CHs. Obviously this method loses some information in the original data. In Chapter \ref{cluster_cs}, we will further consider a recent technique, i.e., Compressive Sensing, to deal with the original data, which is able to recover the original data from a few measurements.

\chapter{The Cluster based Compressive Data Collection for Wireless Sensor Networks with a Path Fixed Mobile Sink}\label{cluster_cs}
\minitoc
This chapter still focuses on \textit{constrained} mobility. Beyond Chapter \ref{cluster_ms}, we apply a recent technique to deal with the original data. 

\section{Motivation}

According to field experiments, data communication contributes majority of energy expenditure in WSNs \cite{experiment04}. It is well known that the energy consumed for data transfer by sensor nodes relates to the packet size and the distance the packet travels. When the sensor nodes are static, the distances between them are fixed. However, it is possible to reduce the size of data, benefiting from the compressive sensing (CS) theory \cite{CS06_1, CS06_2}. In many situations the sensor nodes are densely deployed, thus the sensory readings are highly spacial/temporal correlated. The CS theory states that such spacial/temporal correlated data can be recovered from only a small number of measurements with high enough accuracy by using $l_1$-norm minimization. Many researches have already applied CS to the data collection problem in WSNs \cite{CDG09, xiang11, jia14}, due to its effectiveness in reducing the traffic load. To the best of our knowledge, these works mainly focus on the scenario of static sink. However, due to the requirements of certain applications, the WSNs may consist of several sensor islands and there may be no connectivity between them. In this case, the static sink cannot collect data from all the sensors. Therefore, designing a data collection system using CS and MS has the following advantages of reducing the number of transmissions, which in turn prolongs the network lifetime and being suitable to the large scale, especially disconnected WSNs. 


This chapter presents a complete CS based data collection strategy for delay-tolerant WSNs with constrained MS. Particularly, we consider a system consisting of a set of static sensor nodes and a MS. The constrained MS moves on a fixed path periodically. Some sensor nodes are located in proximity to MS's trajectory, which enables the communication between MS and sensor nodes. Instead of transmitting data in a real time manner, (i.e., each individual packet, involving a single reading, is sent out immediately after generation), the transmitted packets involve a set of readings. For transmitting packets, we use the hybrid CS rather than the pure CS, since the former requires less transmitted packets. We try to reduce the number of transmissions across the network by dividing the network into clusters. Cluster members (CMs) send their raw readings to their cluster heads (CHs), and CHs transmit the CS measurements. One significant problem in our scheme is the cluster radius, which impacts on the energy consumption of the entire network. As discussed in Section \ref{analysis}, larger cluster radius leads to higher intracluster energy consumption and lower intercluster energy consumption. Thus, there should be an optimal point such that the entire energy consumption is minimized. We provide an analytical model to describe the energy consumption of the network and determine the optimal cluster radius.


The main contributions are as follows. A complete scheme for data collection in WSNs using MS and hybrid CS is presented. We propose an analytical model to describe the energy consumption by all sensor nodes in the network, based on which the optimal cluster radius is figured out. Moreover, two implementations are discussed to achieve clustering the network with the optimal cluster radius. The implementations are fully distributed and do not require location information. The message complexities at each node are both $O(1)$. We further conduct extensive simulations. The results show that the proposed approach achieves around 1.2 and 2 times more network lifetime than two compared schemes.

The rest of this chapter is organized as follows. The proposed scheme is discussed in Section \ref{design}. Section \ref{analysis} presents an analytical model to figure out the optimal cluster radius. How to implement such scheme is shown in Section \ref{implementation} and the evaluation of the scheme is demonstrated in Section \ref{evaluation}. Finally, Section \ref{conclusion5} briefly summarizes the chapter and discusses some possible future work.

\section{Data Collection Approach}\label{design}
\subsection{Network Model}\label{model}
Consider a WSN consisting of $N$ uniformly deployed sensor nodes (with unique IDs) and a MS. MS passes the periphery of the sensor field (length of $L$ and width of $W$) periodically on its fixed path. It takes $T$ units of time to finish one tour.  

Let $X$ be the raw reading matrix by all nodes during the tour of MS. Thus, $X$ is a $N \times T$ matrix. $x_{nt}$, an element of $X$, is the raw reading by the $n$th ($n=1,...,N$) node at the $t$th ($t=1,...,T$) unit of time. $X_n$, a row of $X$, denotes the raw readings by the $n$th sensor nodes during the tour of MS and $X_t$, a column of $X$, denotes the raw readings by all nodes at the $t$th unit of time. Let $Y$ be the matrix of CS measurements, whose size is $M\times T$. $Y_t$, a column of $Y$, denotes the CS measurements for the $t$th unit of time. $\Phi$ is the measurement matrix and $\Psi$ is the transform matrix. The fundamental objective is to recover $X$ from $Y$.

We make the following assumptions for. The sensor nodes have the similar capabilities, i.e., sensing, computing, storage, communication, but the equipped energy resource is not necessarily identical. All sensor nodes are able to adjust transmit power based on distances \cite{ren2014data}. The adjustment only occurs in the phase of clustering; while in data transmission, a node only transmits its data packet to its neighbour nodes within 1 hop range. Other activities such as sensing, receiving packets and computing also consume energy. Compared to transmitting, the energy consumed by those activities can be omitted. 

The frequently used notations in this chapter are summarized in TABLE \ref{notations}.
\begin{table}[t]
\begin{center}
\begin{threeparttable}
\caption{Notations}\label{notations} 
  \begin{tabular}{|c| l|}
    \hline
     Notation & Description\\  \hline
     $T$ & Time for MS to finish one tour\\  \hline
     $N$ & Number of sensor nodes \\ \hline
     $X$ & Raw reading matrix\\ \hline
     $Y$ & CS measurements\\ \hline
     $M$ & Number of CS measurements \\ \hline
     $L$ & Length of the sensor field \\ \hline
     $W$ & Width of the sensor field  \\ \hline
     $R$ & Cluster radius\\ \hline
     $C$ & Number of clusters\\ \hline
     $\Phi$ & Measurement matrix \\ \hline
     $\Psi$ & Transform matrix \\ \hline
     $R_0$ & Communication range of single hop\\ \hline
     $P$ & Transmission power for single hop\\ \hline
     $p$ & The probability a sensor becomes CH\\ \hline
    \end{tabular}
\end{threeparttable}
\end{center}
\end{table}
\subsection{Approach Overview}
The cluster based hybrid CS method is used together with MS. An illustrative example of our scheme is shown in Fig. \ref{idea}. The data collection process follows the following steps. 

\begin{figure}[t]
\begin{center}
{\includegraphics[width=0.7\textwidth]{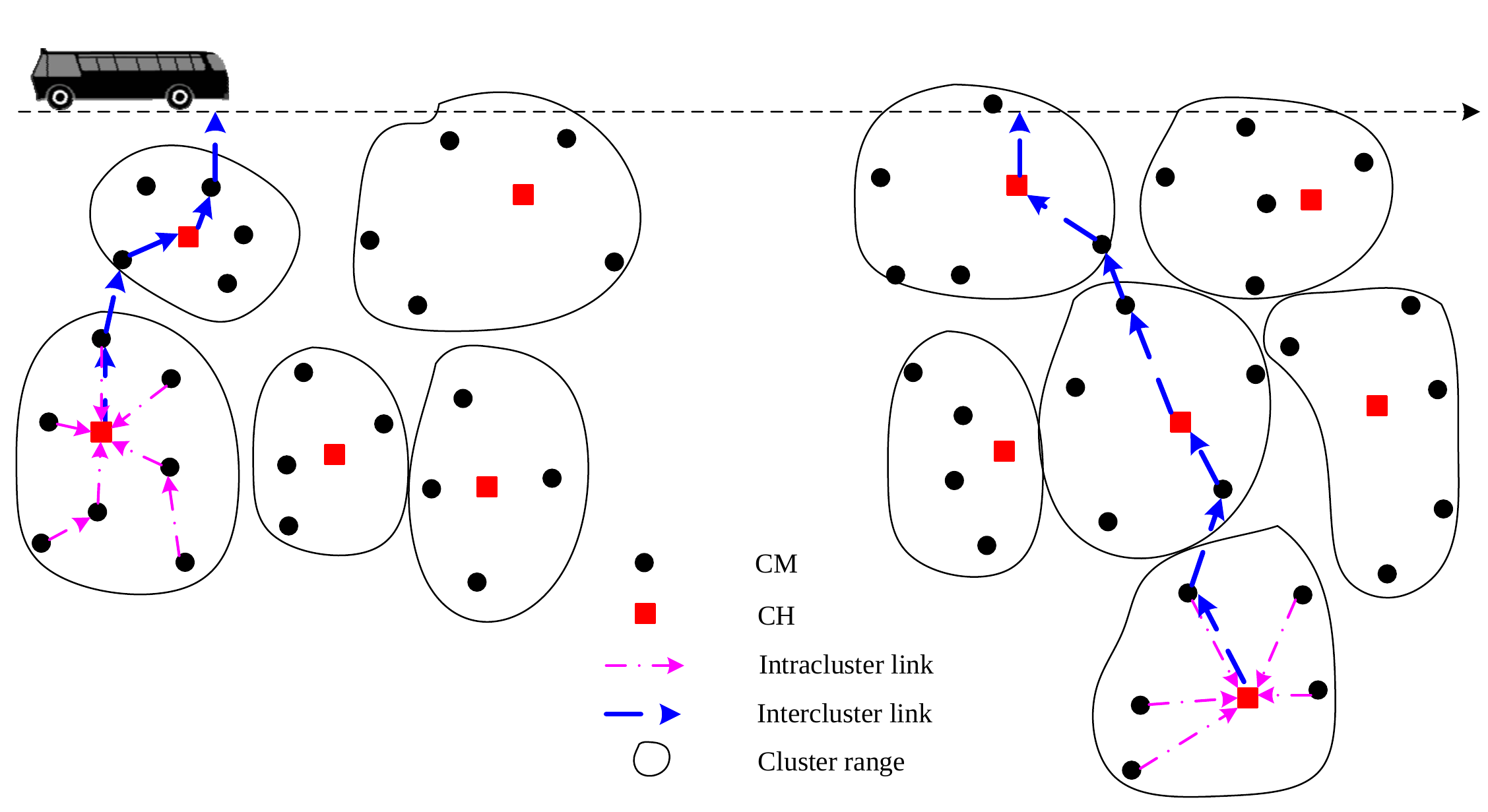}}
\caption{An example of the proposed scheme, where MS is installed on a bus.}\label{idea}
\end{center}
\end{figure}

\begin{itemize}
\item The sensor nodes are divided into a number of clusters.
\item Within clusters, CMs send raw readings to CH. CHs send CS measurements to another CH which is closer to MS's trajectory than itself.
\item The CHs, which have uploading nodes (referred to the nodes which are within the one hop range of MS), transmit CS measurements to the uploading nodes and the latter upload them when MS approaches.
\item MS recovers the raw readings via the CS measurements once it finishes the tour.
\end{itemize}

Several issues should be solved before realizing the above data collection process.
\begin{itemize}
\item How to execute clustering: 1) without the knowledge of sensor nodes' locations; and 2) how many clusters?
\item How to recover the raw readings from the CS measurements when clustering and MS are used?
\end{itemize}

In the rest, we will tackle these problems.

\subsection{Clustering}\label{clustering}

Transmitting CS measurements among CHs will not lead to the energy hole issue, benefiting from CS theory. Hence, unlike \cite{KONS12, CHEN09}, the clusters in our work can have equal radius. Let $R$ be the cluster radius. If the CHs are uniformly distributed, approximately, we have:
\begin{equation}\label{cluster_number}
C\pi R^2=LW
\end{equation}
where $C$ is the number of clusters.

As shown below, $R$ relates to the specific intracluster and intercluster transmission strategy and $R$ influences the entire network energy consumption. In Section \ref{intra_inter_cluster}, we present the communication strategies; and in Section \ref{analysis}, we discuss how to determine $R$.
\subsection{Intracluster and Intercluster Transmission}\label{intra_inter_cluster}
In CS based data collection approaches, every node transmits $M$ packets for a unit of time. The authors of \cite{xiang11} claim that such approach increases the traffic load at the nodes close to the leaf nodes. With this regard, the hybrid CS method is proposed. Further, the authors of \cite{jia14} use the hybrid CS method with clusters. In this subsection, we follow the basic idea of \cite{jia14}. However, we focus on the scenario with MS rather than the static sink case in \cite{jia14}.


A CM needs to transmit a vector of raw readings $X_{n}$ to its CH. As denoted in Section \ref{model}, $X_n$ is of size $T$. Once a CH collects all the packets from its CMs, it compresses the raw readings. Consider the $c$th ($c=1,...,C$) cluster and let $H_c$ represent this subset of nodes. The CS measurements from this cluster can be calculated as follows:
\begin{equation}\label{measurement}
Y_{H_c}= \Phi_{H_c} X_{H_c}
\end{equation}
where $\Phi_{H_c}$ is a submatrix of the measurement matrix $\Phi$ and the size is $M\times |H_c|$; $X_{H_c}$ is a submatrix of the raw reading matrix $X$ and the size is $|H_c|\times T$; and $Y_{H_c}$ is the measurements of this cluster with size $M\times T$. Here $|H_c|$ gives the number of nodes in cluster $H_c$. 

Following the popular method used in the previous work, instead of transmitting $\Phi_{H_c}$, each sensor node uses a pseudorandom number generator seeded with a ID to generate the corresponding measurement coefficients \cite{seed08}. In our work, CHs execute such function while CMs only provide their IDs in the transmitted packets. 

Once the intercluster packet is ready, a CH transmits it to another CH, which is closer to MS's trajectory. Again, hop distance is used for selecting such relay CH. As will be discussed in Section \ref{implementation}, through exchanging control messages, every sensor node knows a set of nodes which have the smaller hop distance to MS's trajectory than itself. By the same way, every CH is able to construct a set of CHs, whose hop counts are smaller than its own. In the rest, we call these CHs the \textit{closer neighbours}. The CH randomly chooses one of its closer neighbours and transmits the intercluster packet to it.

The CH, which receives intercluster packets, compresses its own packet and the received into one intercluster packet. 
The new CS measurements are the summation of its own CS measurements and those received. According to (\ref{measurement}), the sizes of the CS measurements in all the intercluster packets are the same, i.e., $M \times T$. This enables the direct compression of the CS measurements, although the number of nodes vary in different clusters. 

\subsection{Communication with MS}
The CHs which are the closest to MS's trajectory, transmit the CS measurements to the uploading nodes.
The communication between MS and the uploading nodes is a call and send mode. When MS moves on its trajectory, it broadcasts a call message to announce its presence. The uploading nodes send the CS measurements to MS. Since the uploading nodes may be close to each other, to avoid collision between them, each CH sends data packet to only one of its uploading nodes. Here, we simply assume a uploading node is able to upload all the stored CS measurements during the communication period with MS. Once MS finishes its tour, all the CS measurements have been collected. Then, the final CS measurement at MS is:
\begin{equation}\label{summation}
Y=\sum_{c=1}^C Y_{H_c}
\end{equation}

(\ref{summation}) only considers the relationship between the original CS measurements generated by each CH and the final CS measurements generated by MS. Actually, a specific $Y_{H_c}$ may not be known by MS, because additional compression may be done by relay CHs in the process of opportunistic routing. Although MS has no knowledge about where a set of CS measurements come from, it does not impact on recovery.

Note, as will be mentioned in Section \ref{analysis}, all the sensor nodes transfer data with each other using single hop communication. In particular, a CM transfers its raw readings to its CH through multihops. Intercluster transmission as well as CH-MS transfer are also executed in multihop fashion since shorter wireless link leads to low delivery delay.  
\subsection{Recovery}

Once MS gets $Y$, it solves the following problem:
\begin{equation}\label{optimization}
\begin{aligned}
&\min_{s\in R^N} \Vert s \Vert_{l_1}\\
&s.t. Y_t=\Phi X_t, X_t=\Psi s
\end{aligned}
\end{equation}

For one tour, there are totally $T$ problems of (\ref{optimization}) to be solved. After obtaining the estimated $\widehat{s}$, the raw readings are reconstructed by
\begin{equation}
\widehat{X_t}=\Psi \widehat{s}
\end{equation}

To solve (\ref{optimization}), MS needs to know the matrices of $\Phi$ and $\Psi$. $\Psi$ is only used in the data recovery process, thus, it can be generated independently at MS. $\Phi$ is used in the process of compressing raw readings. Thus, it should be identical across the network. As mentioned in Section \ref{intra_inter_cluster}, it is not necessary to make $\Phi$ on the fly, instead, MS can use the IDs in the received packets to generate such matrix.

\section{Analysis on Cluster Radius}\label{analysis}
There is a large body of work about determination of the optimal number of clusters in the literature, such as \cite{opt_number14}. 
Before discussing our analysis, it is worth mentioning the differences and similarities between our work and the existing approaches. First, in the previous work such as \cite{opt_number14}, it is assumed that the CMs transmit data packets to their CHs directly. However, as pointed out below, in our approach CMs communicate with their CHs through multihop communication, i.e., every node uses the first power level to communicate with the neighbour nodes. Second, the previous approaches use static sink to collect data from the sensor nodes; while we use MS which passes the edge of the sensor field. Further, instead of transmitting raw readings in the existing work, we adopt CS technique. So, although we both follow the framework: formulating the network energy consumption as a function of the considered variable (in our work it is the cluster radius; while in the previous it is the number of clusters), the formulation discussed here distinguish from the existing ones. In this section, we provide the analytical model to determine the radius of clusters, i.e., $R$. 


\begin{figure}[t]
\begin{center}
{\includegraphics[width=0.45\textwidth]{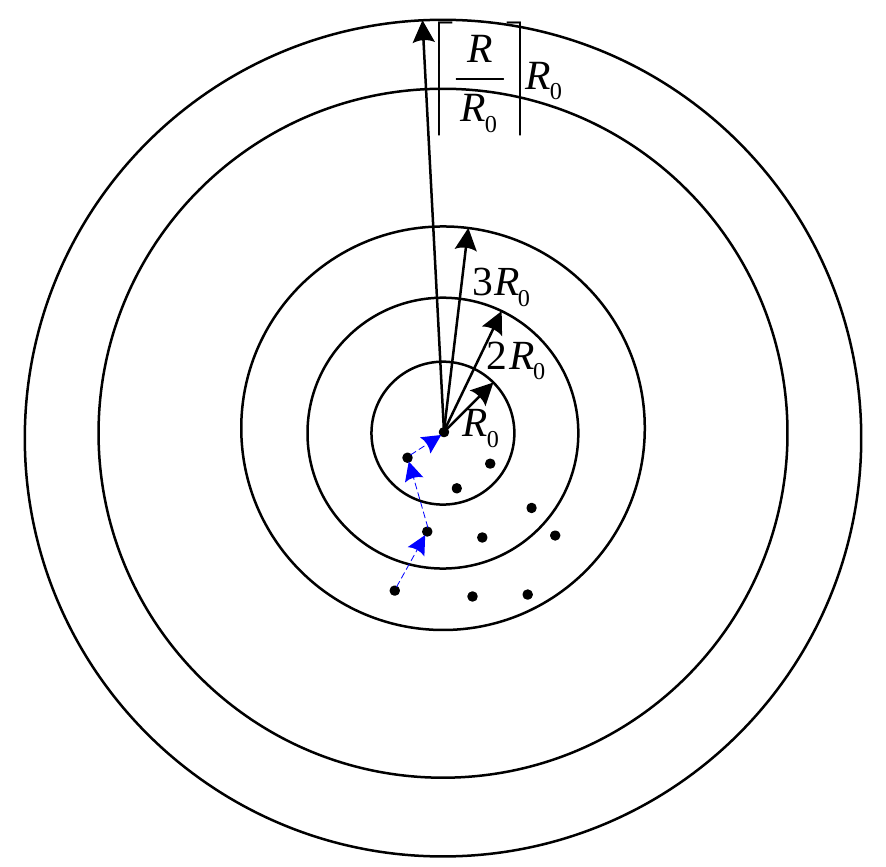}}
\caption{The layer model for a cluster. The CM in the outer layer can only send its data packet to another CM in the inner layer in an opportunisitc manner; and the CMs in the certer layer send packets directly to the CH.}\label{circles}
\end{center}
\end{figure}

Now we consider the energy consumption for intracluster transmission for one MS tour. Each CM needs to transmits a packet of size $T$. As assumed in Section \ref{clustering}, a CH is at the centre of a cluster with the radius of $R$. With the above transmission strategy, the CMs within a cluster follow the multihop fashion to transmit packet to their CH. 
The cluster area is divided into $\lceil \frac{R}{R_0}\rceil$ rings, as shown in Fig. \ref{circles}. The CMs within the first ring can communicate with the CH in one hop; those in the second need two hops to communicate with the CH; and CMs within the final can communicate with the CH in $\lceil \frac{R}{R_0}\rceil$ hops. Let $h=\lceil \frac{R}{R_0}\rceil$ be the largest hop distance between a CH and its furthest CM (see Fig. \ref{circles}) and $\rho=\frac{N}{LW}$ be the density of the sensor nodes. In the first layer, the number of CMs is $\pi R_0^2\rho-1$. Note, the 1 here refers to the CH. The number of CMs in the second layer is $\pi 2^2R_0^2\rho-\pi R_0^2\rho$. Following this pattern, the number in the final layer is $\pi h^2R_0^2\rho-\pi (h-1)^2R_0^2\rho$. Thus, the energy consumption by CMs in one cluster for one MS tour is:
\begin{equation}
\begin{aligned}
&TP((\pi R_0^2\rho-1)\cdot 1 + (\pi 2^2R_0^2\rho-\pi 1^2R_0^2\rho)\cdot 2 +\\
&\cdots+(\pi h^2R_0^2\rho-\pi (h-1)^2R_0^2\rho)\cdot h)\\
=&TP\left(\pi R_0^2\rho\left(\frac{h(h+1)(4h-1)}{6}\right)-1\right)\\
\end{aligned}
\end{equation}
Since the total number of clusters is $C=\frac{LW}{\pi R^2}$, the total energy consumption by all CMs is:
\begin{equation}
E_{intra}=TP\left(\pi R_0^2\rho\left(\frac{h(h+1)(4h-1)}{6}\right)-1\right)\frac{LW}{\pi R^2}
\end{equation}
Approximately, substituting $h=\frac{R}{R_0}$, we obtain:
\begin{equation}\label{e_intra}
E_{intra}=\frac{2NTP}{3R_0}R-\frac{NTPR_0}{6}\frac{1}{R}-\frac{LWTP}{\pi}\frac{1}{R^2}+\frac{NTP}{2}
\end{equation}

As seen from (\ref{e_intra}), the energy consumption for intracluster transmission increases with $R$. When $R$ is large enough to cover all the sensor nodes in the field, there is only one cluster. This equals to the case of the conventional strategy, i.e., every sensor node sends its raw reading to a static sink in the center, which is energy inefficient.

Next, we consider the energy consumption for intercluster transmission. Benefiting from CS theory, all the CHs transmit the same size of CS measurements, i.e., $MT$, no matter where they locate. Averagely, two CHs are $2R$ apart from each other. Thus, the packet of CS measurements may travel $\lceil \frac{2R}{R_0}\rceil$ hops in average between two CHs. Approximately, the CS measurements from the CHs (which are near the periphery of the field) to MS need to travel $\lceil \frac{R}{R_0}\rceil$ hops. There are $C=\frac{LW}{\pi R^2}$ clusters in the network. We use $\frac{2RL}{LW} C$ \footnote{Since the CHs are assumed to be uniformly distributed and $C$ CHs occupy the area of $LW$, then the number of CHs occupying the area of $2RL$ is $\frac{2RL}{LW} C$.} to estimate the number of CHs which transmit data to MS. Then, the energy consumption by the CHs can be approximated by 
\begin{equation}\label{e_inter}
\begin{aligned}
E_{inter}=&MTP\left(\left(C-\frac{2RL}{LW}C\right)\frac{2R}{R_0}+\frac{2RL}{LW}C\frac{R}{R_0}\right)\\
=&\frac{2LWMTP}{\pi R_0}\frac{1}{R}-\frac{2L}{\pi R_0}
\end{aligned}
\end{equation}

With (\ref{e_intra}) and (\ref{e_inter}), we have the total energy consumption by all the sensor nodes:
\begin{equation}\label{energy_cluster}
E=a_1R+a_2\frac{1}{R}+a_3\frac{1}{R^2}+a_4
\end{equation}
where $a_1=\frac{2NTP}{3R_0}$, $a_2=\frac{2LWMTP}{\pi R_0}-\frac{NTPR_0}{6}$, $a_3=-\frac{LWTP}{\pi}$ and $a_4=\frac{NTP}{2}-\frac{2L}{\pi R_0}$. 

Since there is at least one cluster in the network, we have the maximum of $R$, i.e., $R_{max}=\sqrt{\frac{LW}{\pi}}$. Then, we can figure out the optimal $R$ (no larger than $R_{max}$) to minimize (\ref{energy_cluster}).

The above analysis model has limitations in practice. First, for the intracluster communication, since the CH may not be exact the center, the corresponding energy consumption may not follow the layer pattern strictly, i.e., (\ref{e_intra}) is not a precise model. Second, the hop count between any two CHs may not always be $\frac{2R}{R_0}$, thus (\ref{e_inter}) also provides approximate result.
\section{Implementation}\label{implementation}
Let the optimal cluster radius be $R_*$, which is the solution of (\ref{energy_cluster}) and the corresponding energy consumption for data transmission be $E_*$. As stated in Section \ref{analysis}, to achieve $E_*$, the network should be evenly divided into a number of clusters with the radius of $R_*$. Further, the CHs should be at the centres of the clusters. According to (\ref{cluster_number}), we can calculate the optimal cluster number, $C_*$. Then, the clustering problem can be transferred to $k$-median problem \cite{kmedian}, which is to figure out $C_*$ positions for CHs such that the total hop distance from CMs to their CHs is minimized. It has been shown that $k$-median problem is NP-hard \cite{kmedian}. In this section, we provide two distributed implementations for clustering. 

\subsection{Implementation 1 (I1)}
I1 consists of: MS broadcast, CH election, CM-CH attachment, CH-CH association and uploading.

\subsubsection{MS broadcast}
In the first tour, MS broadcasts three values: $h_*=\lceil\frac{R_*}{R_0}\rceil$, $p_*=\frac{C_*}{N}$ and a hop count (initially equals to 0), within the range of $R_0$ when it moves on the predefined path. The nodes that receive such information increase the hop count by 1 and forward it to the nodes nearby. When MS finishes the tour, all the sensor nodes know $h_*$, $p_*$ and their hop counts. If more than one  message is received, a node sets its hop count as the smallest one. To avoid confusion with the below description, such hop count is referred to as intercluster hop count and it will be used to assist intercluster transmission.

\subsubsection{CH election}
The next step is to select $C_*$ CHs. Although we assume the CHs are uniformly distributed and located in the centres of the clusters in the analytical model, we point out it is not necessary to make them at centres in practice. First, selecting such CHs is a difficult task especially when the geographical information is unavailable. 
Second, in practice the role of CH will rotate within cluster. Although the CH at the centre leads to the minimum intracluster energy consumption, such situation will not last for long. 
With these observations, we present a simple method to select approximate $C_*$ CHs. Every sensor node generates a number 1 with the probability $p_*$ or 0 with the probability $1-p_*$. The sensor nodes which generate the number 1 become CHs no matter whether they will be the centres of the clusters. 

\subsubsection{CM-CH attachment}
Then, the CHs invite non-CHs into their cluster by broadcasting an Invitation message within range of $R_*$. The sensor nodes within range can hear the message and reply a Confirmation message to join the cluster. Note, it is possible that one node receives more than one Invitation. In this case, it selects the one with the strongest signal strength. It is also possible that some nodes are outside $R_*$ range from all the CHs, which are called isolated nodes. In this case, after the invitation process, i.e., timer $T_{invite}$ expires, the isolated nodes actively sends Joining request within range of $aR_*$, where $a=2,\ldots, a_m$ and $a_m=\lceil \frac{\sqrt{(L^2+W^2)}}{R_*}\rceil$. The selection of $a_m$ guarantees that no node is finally isolated. 

\subsubsection{CH-CH association}
Next, every CH needs to build up the closer neighbour set. As discussed in Section \ref{clustering}, a CH $v$ in the closer neighbours of CH $u$ has smaller intercluster hop count than $u$. The process of building up the closer neighbour set is as follows. Each CH broadcasts an Association message within range of $bR_*$, $b=1,\ldots, a_m$. $b$ starts from 1 and increases by 1 until the CH hears a Confirmation message from another CH. Note, a CH only replies the Confirmation message to the CHs with larger intercluster hop count. Further, some CHs will not receive any Confirmation message because there are no other CHs with smaller intercluster hop count. After timer $T_{association}$ expires, these CHs figure out the uploading nodes (with intercluster hop count as 1) within the clusters.

\subsubsection{Uploading}
From the second tour, CMs transmit raw readings to CHs; CHs transmit CS measurements to the uploading nodes; and the uploading nodes upload data to MS. 

To achieve longer network lifetime, the role of CH is rotated within clusters every a certain number of tours. Further, clusters will reform every a larger number of tours. The CH rotation and reclustering strategies help to balance the energy consumption among sensor nodes. The CH role can be shifted to the CM which has the largest residual energy. Reclustering can be achieved by another round of the above procedure.

\subsubsection{Discussion}
The advantage of I1 is the simpleness. It does not require location and distance information; while only uses hop information. The disadvantage is that the uniform CH distribution is not guaranteed, since CHs are randomly generated. Reselecting CHs after clustering can push CHs to the centres and then decrease intracluster energy consumption. But, this requires further overhead and only improves the energy efficiency for a short period. 

Now we consider the message complexity. 
For a CH, at first it broadcasts an Invitation message to invite non-CHs to join its cluster and then waits for Confirmation. If there are isolated nodes, suppose the number of which is $N_I$ ($N_I \ll N$), in the worst case, a CH replies at most $N_I$ Confirmation messages. To build up the closer neighbour set, a CH may broadcast at most $a_m$ Association messages and reply at most $C_*$ Confirmation messages. Therefore, the total number of messages a CH transmits is at most: $1+N_I+a_m+C_*$. Since $a_m$ is upper bounded and $C_*\ll N$ (see Section \ref{evaluation}), the message complexity in the worst case of a CH is $O(1)$. Next, we consider the case for a CM. If it is within $R_*$ range of any CHs, it replies one confirmation message to join cluster; otherwise, it may broadcast at most $a_m-1$ Joining request. Therefore, the message complexity is $O(1)$.  


\subsection{Implementation 2 (I2)}
I1 cannot guarantee the evenness of CH distribution. The reason is that due to the randomness of generating CHs, two nodes which locate in close proximity are possible to both become CHs. I2 is inspired by the observation that if the CHs are evenly distributed, they keep away from each other. So, we introduce a competing mechanism to ensure that every CH is at least a certain distance away from any other CHs.
\subsubsection{CH election}
Every node becomes a CH candidate at the beginning of a new cluster formation phase. The competing rules are as follows. A CH candidate broadcasts a Competing message involving its ID and residual energy within the range of $2R_*$. Then, it listens for Competing messages from other CH candidates. If it receives Competing messages, it compares the contained energy recording with its own (denoted by $e$). \textit{Rule 1}: if the latter is larger, it broadcasts a Win message and becomes a CH directly; otherwise, it cannot decide to be CH or not and need to check Rule 2 and Rule 3. \textit{Rule 2}: if a CH candidate receives a Win message, it broadcasts a Quit message. This is because receiving a Win message means its residual energy is not largest within the $2R_*$ range and it cannot become CH. A CH candidate which has not received win messages, listens for Quit message. \textit{Rule 3}: if a CH candidate receives a Quit message, it removes the corresponding CH candidate from the comparing list and checks whether its residual energy is the largest among the rest. This process lasts until $T_{election}$ expires. The pseudocode of the competing scheme is shown in Algorithm \ref{algorithm1}. 

\begin{algorithm}[t]
\caption{CH election by a node}\label{algorithm1}
\begin{algorithmic}[1]
\State Broadcast Competing message within $2R_*$.
\State Listen for Competing messages and organize the energy recordings into a list $E_{residual}$.
\While{timer $T_{election}$ has not expired}
\If{$e>\max(E_{residual})$}
\State Broadcast Win message within $2R_*$; exit
\EndIf
\State Listen for Win message
\If{Win message is received}
\State Broadcast Quit message within $2R_*$; exit
\EndIf
\State Listen for Quit message
\If{Quit message is received}
\State Remove the energy recording of the Quit message from $E_{residual}$.
\EndIf
\EndWhile
\end{algorithmic}
\end{algorithm}

The other steps for the cluster formation and data uploading are the same as I1. 

\subsubsection{Discussion}
We claim that with the competing mechanism, every CH is out of the $2R_*$ ranges (or $2h_*$ hops) of all the other CHs. Further, in I2, the nodes with large residual energy are likely to become CHs. 


I2 differs from I1 by the competing mechanism, which results in more overheads. We discuss the message complexity for the competing mechanism. First, a CH candidate needs to broadcast one Competing message. After that, it will broadcast either a Win message to announce its status as CH or a Quit message. Each node transmits two more messages than I1. Therefore, message complexity at a node in I2 for clustering is still $O(1)$. The benefits of such two more message transmissions will be demonstrated in the next section through simulation experiments.

\section{Evaluation}\label{evaluation}
We conduct simulation experiments to evaluate the performance of the proposed approach and compare these finding against some related methods.

\subsection{Settings}
We simulate a set of WSNs, 
where the sensor nodes are assumed to be distributed  independently  and uniformly in a field of 20 unit $\times$ 10 unit. A MS moves along the long edge and $T$=20 unit. The number of sensor nodes ($N$) ranges from 200 to 800. The number of CS measurements $M$ is set to be 40. The single hop transmission range $R_0$ is 1 unit and the corresponding power $P$ is 1 unit. All the simulations are conducted by Matlab.

\subsection{Compared algorithms}
Since the performance of reconstruction accuracy by CS for data collection in WSNs have already been investigated widely in previous work, here we pay attention to the energy efficiency of the proposed CS and MS combination approach and compare it with other related work. 
 In particular, we consider the below two approaches for comparison, both of which use MS and are location unaware.

\textit{MASP} \cite{GAO11}: formulates the data collection using MS as a Maximum Amount Shortest Path (MASP) problem. It targets on assigning sensor nodes properly to subsinks such that the overall throughput is maximized and each sensor node uses the shortest path to transmit the sensory data to its subsink. 

\textit{MobiCluster} \cite{KONS12}: uses clusters for data collection. The clusters are with unequal sizes depending on the distance to MS' trajectory. The sensor nodes have three roles: rendezvous node (RN), CH and CM. RNs are the nodes within the communication range of MS when it moves. CMs send raw readings to CHs. CHs execute an aggregation algorithm to downsize the data packet and transmit it to the CHs closer to MS's trajectory or RNs. Finally, RNs upload data to MS. 

\subsection{Results}
We first demonstrate the performance of I1 and I2. We consider the metrics of CH number and the evenness of CH distribution, which is described by the average hop distance (AHD) from CMs to CHs. When the CH numbers are the same, lower AHD means CHs are more evenly distributed.

Fig. \ref{CH_number} and Fig. \ref{AHD} compare the CH numbers and AHDs by I1, I2 and the analytical model. Both display the statistic results for 30 independent topological networks under each network size ($N$). 
As seen in Fig \ref{CH_number}, the average CH numbers by I1 and I2 are close to the analytical results ($C_*$), especially I2 when $N$ is small. The CH numbers by I1 and I2 can be either smaller or larger than $C_*$. The reason for I1 is the random CH selection. For I2, the reason for the upper side is the "margin effect". If a node near the margin of the field is with large residual energy, it is likely to become a CH. Then, the competition scheme may lead to more than $C_*$ CHs. The reason for lower side is that two CHs may be more than $2R_*$ away from each other. See Fig. \ref{example1} for an example. This occurs when a CH candidate $u$ at the edge of $2R_*$ range of CH candidate $v$ is covered by CH candidate $w$, whose residual energy is larger than $u$. Thus, $v$ and $w$, which are more $2R_*$ apart, become CHs. We also find that the range of CH numbers by I2 is more narrow than I1, which indicates that I2 is able to obtain the number of CHs more close to $C_*$. Correspondingly, the AHDs by I2 are generally smaller than I1, which indicates that I2 achieves better CH distribution than I1.

\begin{figure}[t]
    \centering
    \begin{subfigure}[b]{0.45\textwidth}
        \includegraphics[width=\textwidth]{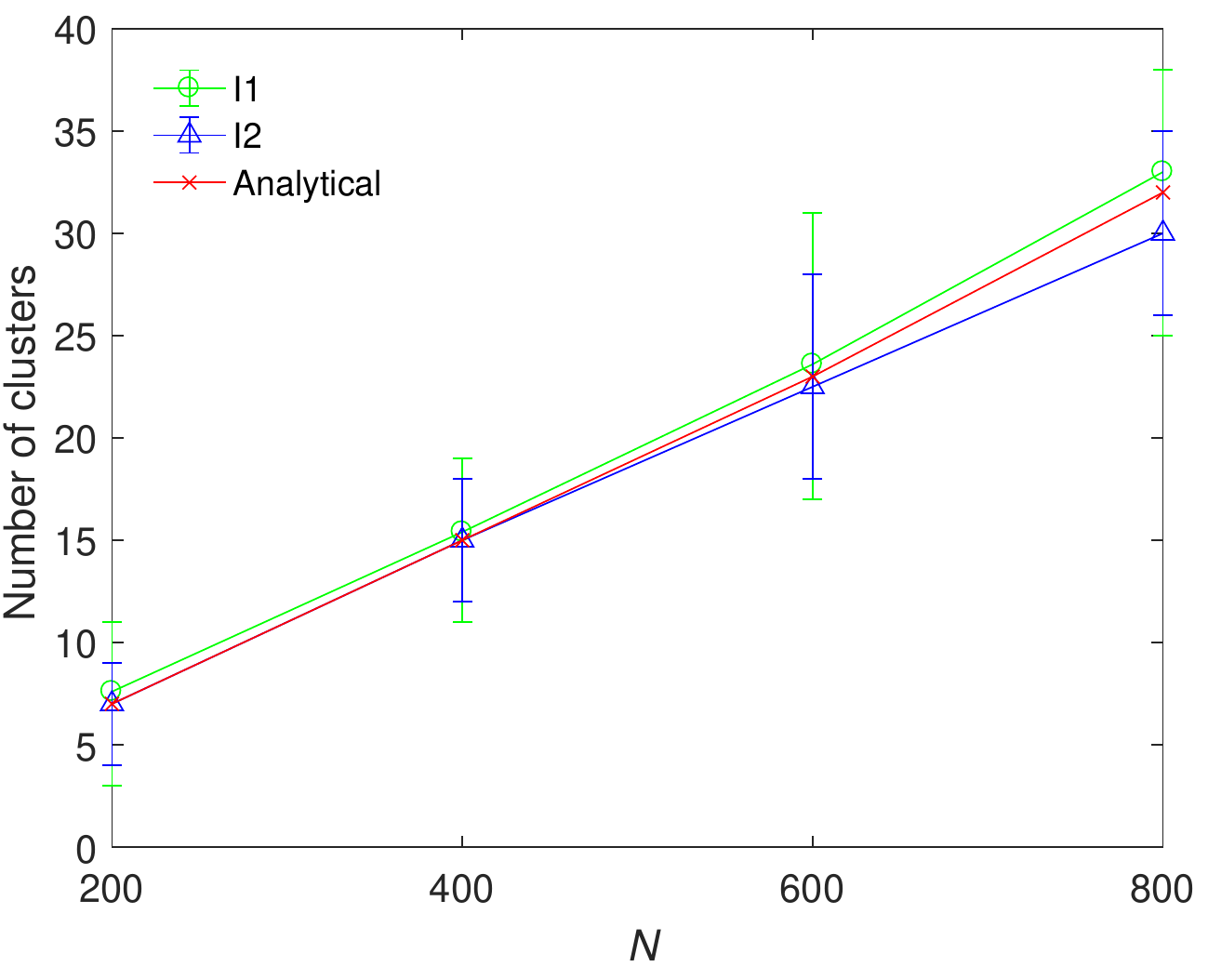}
        \caption{}
        \label{CH_number}
    \end{subfigure}
    \begin{subfigure}[b]{0.45\textwidth}
        \includegraphics[width=\textwidth]{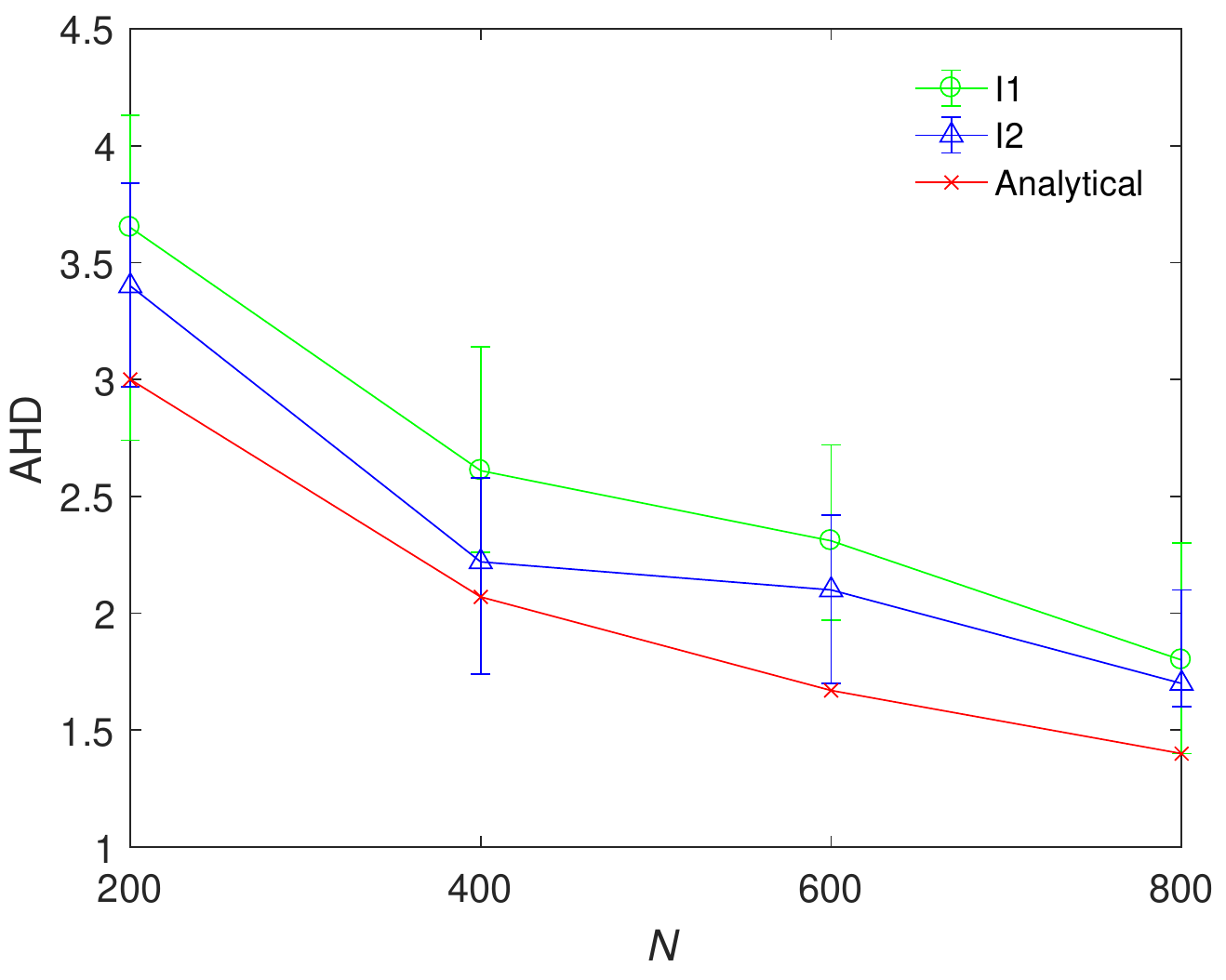}
        \caption{}
        \label{AHD}
    \end{subfigure}
    \caption{Comparison of I1 and I2. (a) CH numbers. (b) AHD.}
\end{figure}

%

\begin{figure}[t]
\begin{center}
{\includegraphics[width=0.5\textwidth]{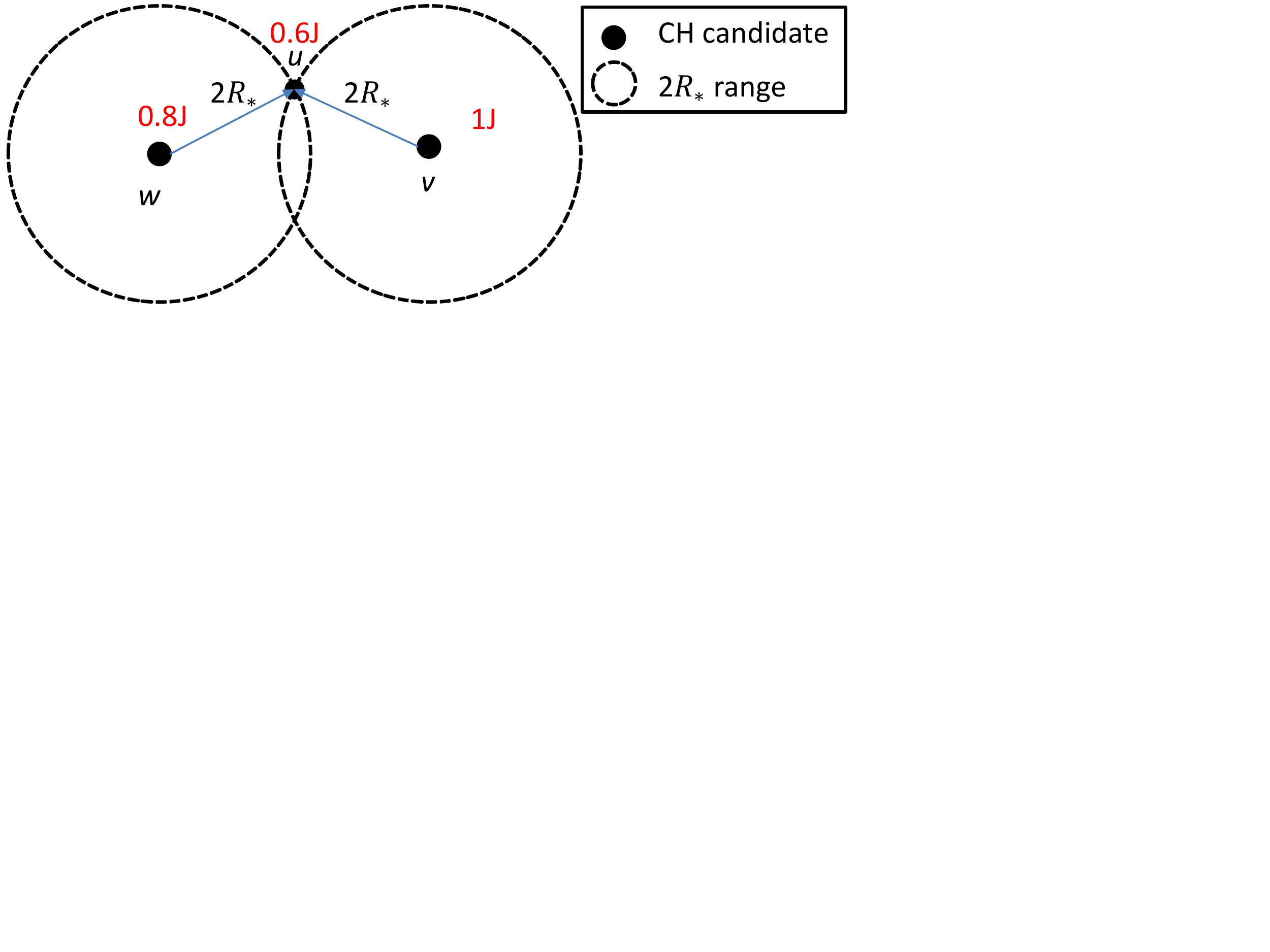}}
\caption{An example of two CHs are more than $2R_*$ apart.}\label{example1}
\end{center}
\end{figure}

Now we compare with MASP and MobiCluster. For MobiCluster, we set several parameters including the cover ranges for the clusters and the aggregation ratio. As we do not focus on the influence of these parameters on MobiCluster, we only select typical values which are suitable to our settings. Here, we set two cluster radii as 2 and 3 units and the aggregation ratio as 30\%. For simplicity, the aggregation is only adopted within clusters rather than interclusters. The probability of a node becoming a CH candidate is 0.2 \cite{CHEN09}. For each network size, we randomly select one of the 30 simulated networks, and apply I1, I2, MASP and MobiCluster to investigate the performance in energy efficiency. We focus on two metrics: average node energy consumption and network lifetime. See Fig. \ref{energy_consumption} and \ref{lifetime}.

Fig. \ref{energy_consumption} demonstrates the average energy consumption of a node per MS trip. With the increase of $N$, the energy consumption by a node in I1 and I2 decrease, which follow the trend of the analytical model. This is because that when $N$ becomes larger, the number of clusters increases (see Fig. \ref{CH_number}); while AHD becomes smaller (see Fig. \ref{AHD}). Since the multihop communication is adopted,  smaller AHD leads to lower energy consumption for intracluster transmission. Further, the intercluster energy consumption is unrelated to $N$, see Eq. (\ref{e_inter}). Thus, the total energy consumption decreases with $N$. In Fig. \ref{energy_consumption}, a node in I2 consumes less energy than that in I1 averagely, but both of them consume more than the analytical results. The reason is the same as discussed previously, i.e., the CH distributions in I1 and I2 are difficult to be precisely uniform. The energy consumption by Mobicluster has the similar trend. 
But the dropping rate is smaller than I1 and I2. The reason lies in the CS technique, the benefit of which is more obvious in the cases of large $N$. The average node energy consumption in MASP is not influenced much by $N$ and MASP performs the worst, which indicates that the flat infrastructure is less efficient than cluster infrastructure. 

\begin{figure}[t]
    \centering
    \begin{subfigure}[b]{0.45\textwidth}
        \includegraphics[width=\textwidth]{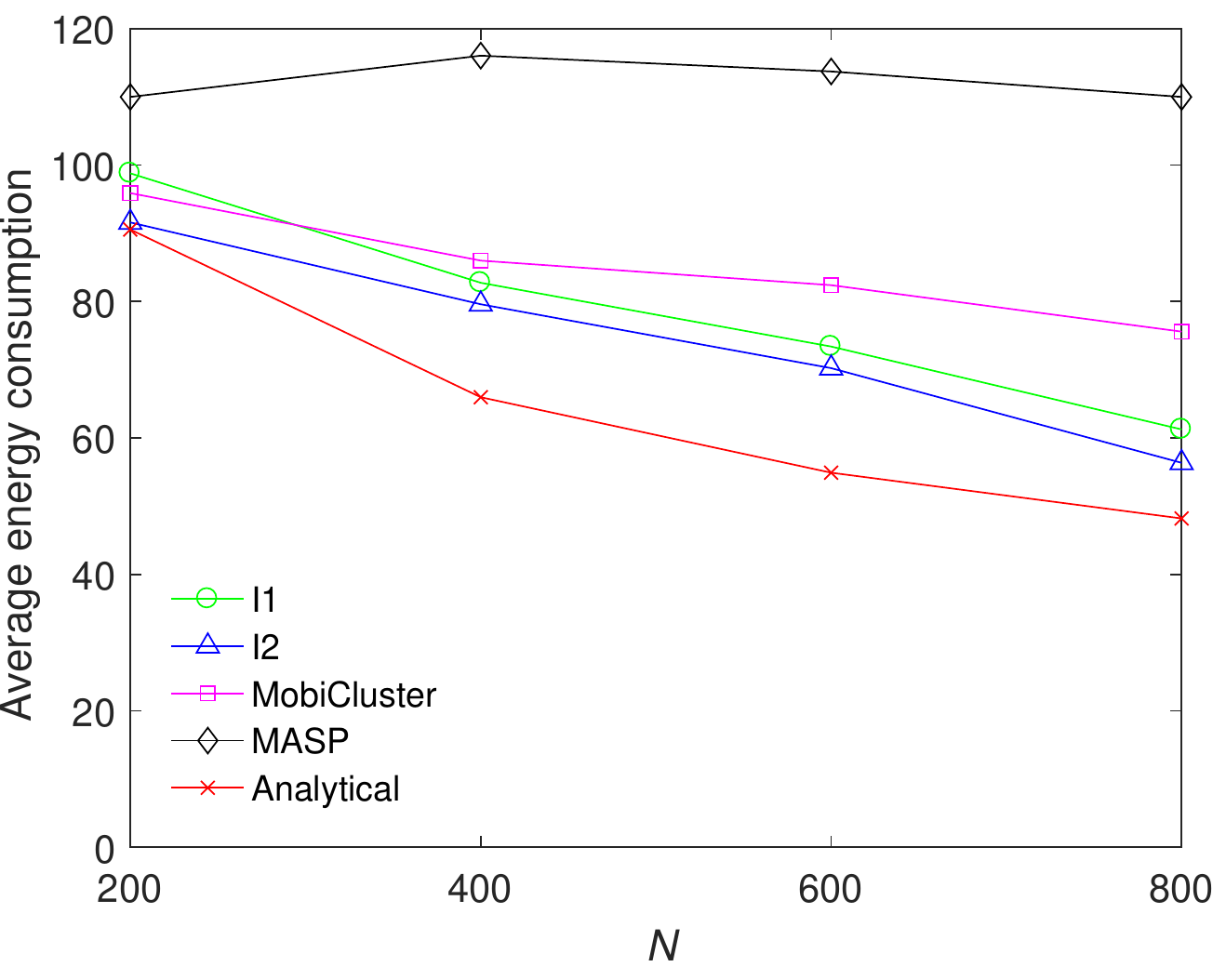}
        \caption{}
        \label{energy_consumption}
    \end{subfigure}
    \begin{subfigure}[b]{0.45\textwidth}
        \includegraphics[width=\textwidth]{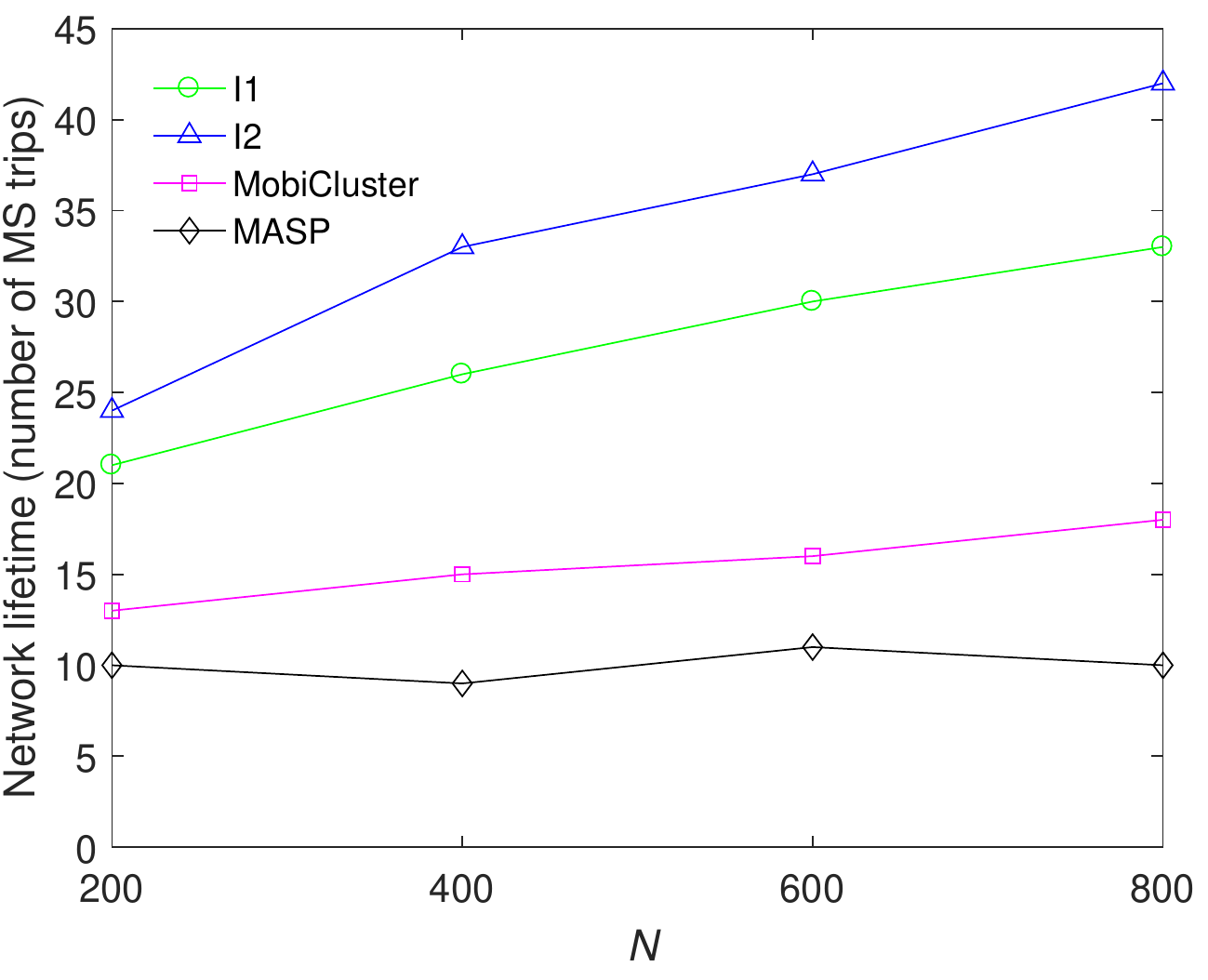}
        \caption{}
        \label{lifetime}
    \end{subfigure}
    \caption{Comparison of I1, I2, MASP, MobiCluster and the analytical model. (a) Energy consumption. (b) Network lifetime.}
\end{figure}


Fig. \ref{energy_consumption} demonstrates that the proposed approach is more energy efficient than the compared ones. But, it only provides the performance for a single MS trip. To have an insight of the performance of proposed approach, we further investigate the network lifetime\footnote{The network lifetime refers to the time duration from network deployment to the first node runs out of battery.}. The initial energy equipped by the nodes are around 5000 units
and they are not necessarily identical. 
In the following simulations, we only consider the energy consumption for data transmission. Further, for I1, I2 and MobiCluster, reclustering is triggered before each MS trip. But for MASP, the transmission routes remain the same for the whole lifetime, because it is not energy aware. Fig. \ref{lifetime} shows the network lifetimes by I1, I2, MASP and MobiCluster for the considered networks. The network lifetime by MASP is not influenced much by the node density. Its result is the lowest among the four approaches. This is because of the energy hole issue, i.e., the nodes which can communicate with MS relay large number of data packets for other nodes. The network lifetime by MobiCluster increase slightly with the network size and it performs better than MASP due to aggregation strategy. The network lifetimes by I1 and I2 increase more quickly than MobiCluster, which benefits from the utility of the hybrid CS technique. The hybrid CS technique reduces the number of the transmitted packets, thus the energy consumption by the sensor nodes is also reduced, which in turn improves the network lifetime. I2 improves the network lifetime by around 1.2 and 2 times respectively compared to MobiCluster and MASP. 


\section{Summary}\label{conclusion5}
This chapter considers the problem of data collection in delay-tolerant WSNs with constrained MS where the hybrid CS technique is used. We have presented a complete data collection strategy which accounts the characteristics of both MS and CS technique: the network is divided into clusters. We provided an analytical model for the energy consumption by the proposed data collection strategy, through which the analytical cluster radius is obtained. We further discussed two fully distributed implementations, which do not require the location information of sensor nodes. The message complexities are both $O(1)$. Extensive simulations were conducted and the comparisons with two related approaches, which do not use CS technique, were provided. The simulation results display that the proposed approach, especially I2, is more efficient than the alternatives and improves the network lifetime by about 1.2 and 2 times compared to MobiCluster and MASP. The results are also presented in \cite{huang17CS}.
 
In this chapter, we did not consider the issues associated with practical implementation, such as the packet loss and node failure. These issues influence the performance of the proposed approach, which are left for future work. Another defect is that, also similar to Chapter \ref{cluster_ms}, only a single MS is considered. In Chapter \ref{cluster_um}, multiple MSs are used.

\chapter{The Unusual Message Delivery Path Construction for Wireless Sensor Networks With Trajectory Fixed Mobile Sinks}\label{cluster_um}
\minitoc
Both Chapter \ref{cluster_ms} and \ref{cluster_cs} consider using a single MS to collect data from WSNs. As mentioned in Chapter \ref{literature}, single MS is associated with many practical issues, such as long data collection latency. In this Chapter, we focus on using multiple mobility constrained MSs to collect data from WSNs where the collection latency is intensive.
\section{Motivation}
Guaranteeing data delivery delay is critical for achieving acceptable quality of service in delay-intolerant applications, such as earthquake or volcanic eruption warning system. The detection of any unusual phenomenon needs to be transmitted to the users as soon as possible, then the users can take further actions. In this chapter, we use the term UM to represent the message containing the detected unusual phenomenon. Researchers have proposed various solutions to the delivery delay problem in WSNs. One class of approaches uses controllable M node \cite{zhao2012optimization, Liang13CSS, kumar13delay, path_opt_14, tang2015dellat, huang2016path}. The basic idea is to find an optimal path for M node such that a certain metric is optimized. These approaches are able to reduce the energy consumption by S nodes significantly, but they also pose several other challenges. Since the searching space for M node's possible position is infinite in the sensing field, it is hard to solve the path optimization problem. Besides, due to the low physical speed of M node, the room of delivery delay improvement is limited.

In this chapter, the Improved-Unusual Message Delivery Path Construction (I-UMDPC) is proposed for WSNs, which exploits predictable mobility, i.e., M nodes are attached to buses. This mobility avoids the control of M nodes' complex movements, which is suitable to the applications in urban areas. Moreover, instead of sending M nodes to visit S nodes, I-UMDPC adopts the concept of multihop communication for delivering UM, which can decrease delivery latency greatly. 

One difficulty of transmitting UM to M nodes is the management of routes. Since M nodes are moving, for a S node which has a UM to send (henceforce called source node), the routes to M nodes vary with time. To achieve successful UM delivery, S nodes need to keep track of the latest locations of M nodes. An ideal way is to flood the network once M nodes move away, such as \cite{twinroute09, stashing15}. However, it results in a huge network overhead. Thus, frequent propagation of the sink location updates should be avoided if the S nodes are not able to harvest energy from the environment. An alternative approach is based on clusters \cite{HCDD06, cluster_10, huang2017vtc} or virtual grids \cite{TTDD05, grid_09, VGDRA15}, where only cluster heads (CHs) or grid heads need to update the positions of M nodes, which significantly reduces the energy expenditure. I-UMDPC uses the cluster structure. The S nodes are divided into a number of clusters, each of which consists of one CH and several cluster members (CMs). Further, a special node is placed at each bus stop (henceforce called B node). The B nodes store the timetable of bus operations, and together with CHs, assist the source node to deliver UM. With these components, the overall procedure of I-UMDPC is as follows. Upon detecting the unusual phenomenon, the source node sends UM to its local CH (source CH). The source CH communicates with B nodes through other CHs to gather the arrival times of the coming buses. Based on these arrival times, the source CH determines a subset of B nodes as delivery targets, such that UM can be picked up within the allowed delay. When the buses come, B nodes upload UM to the on-board M nodes. 

Another difficulty of I-UMDPC lies in the uncertainty in the bus operation. Considering the traffic on road, a bus may arrival at a bus stop on time, early or late. This feature makes the timetable stored at B nodes unreliable. The stop duration at a stop is also uncertain, which depends on the passengers to get on or off. Both types of the above uncertainties impact on whether UM can be delivered in time. In this paper, we propose an optimization problem which aims at using the minimum energy to transmit UM and in the meanwhile maintaining the successful delivery probability to a high level. 

The main contribution of this chapter is I-UMDPC, which takes into account the actual features of bus operation, i.e., the uncertainties in the arrival times as well as the stop durations. The data collection system uses buses, which already exist in the environment, to carry M nodes. We conduct extensive simulations as well as practical experiments on our testbed to demonstrate the advantages of the proposed approach against the alternatives. We find that I-UMDPC is able to route UM to M nodes more reliably and efficiently than the alternatives.

The rest of this chapter is organized as follows. Section \ref{protocol} presents the proposed routing protocol in details. The optimization for selecting target B nodes is formulated. Section \ref{simu} demonstrates extensive simulations inclusive the influence from various factors. Section \ref{experiment} shows the experimental results. Finally, Section \ref{Summary} gives a brief summary of this chapter. The publications related to this chapter include \cite{huang2017umdpc}, \cite{huang2017delay}.
\section{Data Dissemination Protocol}\label{protocol}
The main objective of this work is to develop a routing method that delivers UM from source node to M nodes in an energy efficient way such that M nodes receive UM within the allowed latency. In this section, we first present the basic assumptions of the system and then describe our approach. The main notations in this chapter are summarized in TABLE \ref{notations6}.
\begin{table}[t]
\begin{center}
\begin{threeparttable}
\caption{Notations}\label{notations6} 
  \begin{tabular}{|c|l|}
    \hline
     Notation & Description\\  \hline
     $R$ & Candidate target B node set\\  \hline
     $T$ & Target B node set \\ \hline
     $c$ & The cost to transmit data packet to a B node\\ \hline
     $X$ & The instance a bus to arrive at a stop\\ \hline
     $Y$ & The duration a bus stops at a stop \\ \hline
     $\Delta$ & The maximum allowed delivery latency \\ \hline
     $t$ & The time required to transmit a data packet to a B node  \\ \hline
     $h$ & The hop distance from a CH to a B node\\ \hline
     $d$ & Distance between to node\\ \hline
     $E_t$ & Energy consumption for transmission\\ \hline
     $\alpha$ & The minimum probability of successfully received the data packet within $\Delta$\\ \hline
     $\beta$ & The time required for single hop transmission\\ \hline
     $PSD$ & The probability of successfully delivered\\ \hline
     $RSD$ & Ratio of successful delivery \\ \hline
     $RPU$ & Ratio of pick-up \\ \hline
    \end{tabular}
\end{threeparttable}
\end{center}
\end{table}

\subsection{Assumptions}
We assume the following network characteristics:
\begin{itemize}
\item S nodes are randomly deployed and remain static.
\item S nodes are equipped with limited initial energy while B nodes and M nodes do not have any energy constraints.
\item M nodes move on their predefined paths, along which static B nodes are placed. 
\item M nodes' movements follow the timetable, but they may arrive at bus stops early, on time or late.
\item M nodes stop at bus stops for a while and the durations are also uncertain.
\end{itemize}

\subsection{Improved-Unusual Message Delivery Path Construction}
The overall procedure of I-UMDPC is summarized as follows.

\begin{enumerate}
\item Announcement. The source node announces the source CH that it has data for delivery and transmits the data to the source CH.
\item Quote. The source CH sends Quote to the nearby CHs and B nodes. The nearby CHs forward the Quote until it reaches a B node. 
\item Reply. The B node transmits a Reply message to the source CH. The Reply message contains the arrival time of the coming bus.
\item Target selection. Source CH selects target B nodes.
\item Delivery. Source CH transmits UM to the targets.
\item Upload. The target B nodes upload the data to the on board M nodes when the buses arrive.
\end{enumerate}

Below we provide more details for I-UMDPC. I-UMDPC is based on a two layer structure where the higher layer consists of B nodes and CHs and the lower layer consists of CMs. Since this chapter mainly focuses on the delivery of UM, not the cluster formation, any existing approaches can be used. In this chapter, we adopt Max-Min D-Cluster Formation Algorithm \cite{maxmind}. This algorithm provides a load-balanced clustering solution by constructing D-hop clusters, i.e., any CM is at most D hops away from its CH. An illustration is shown in Fig. \ref{illustration} where D equals to 2. 

\begin{figure}[t]
\begin{center}
{\includegraphics[width=0.5\textwidth]{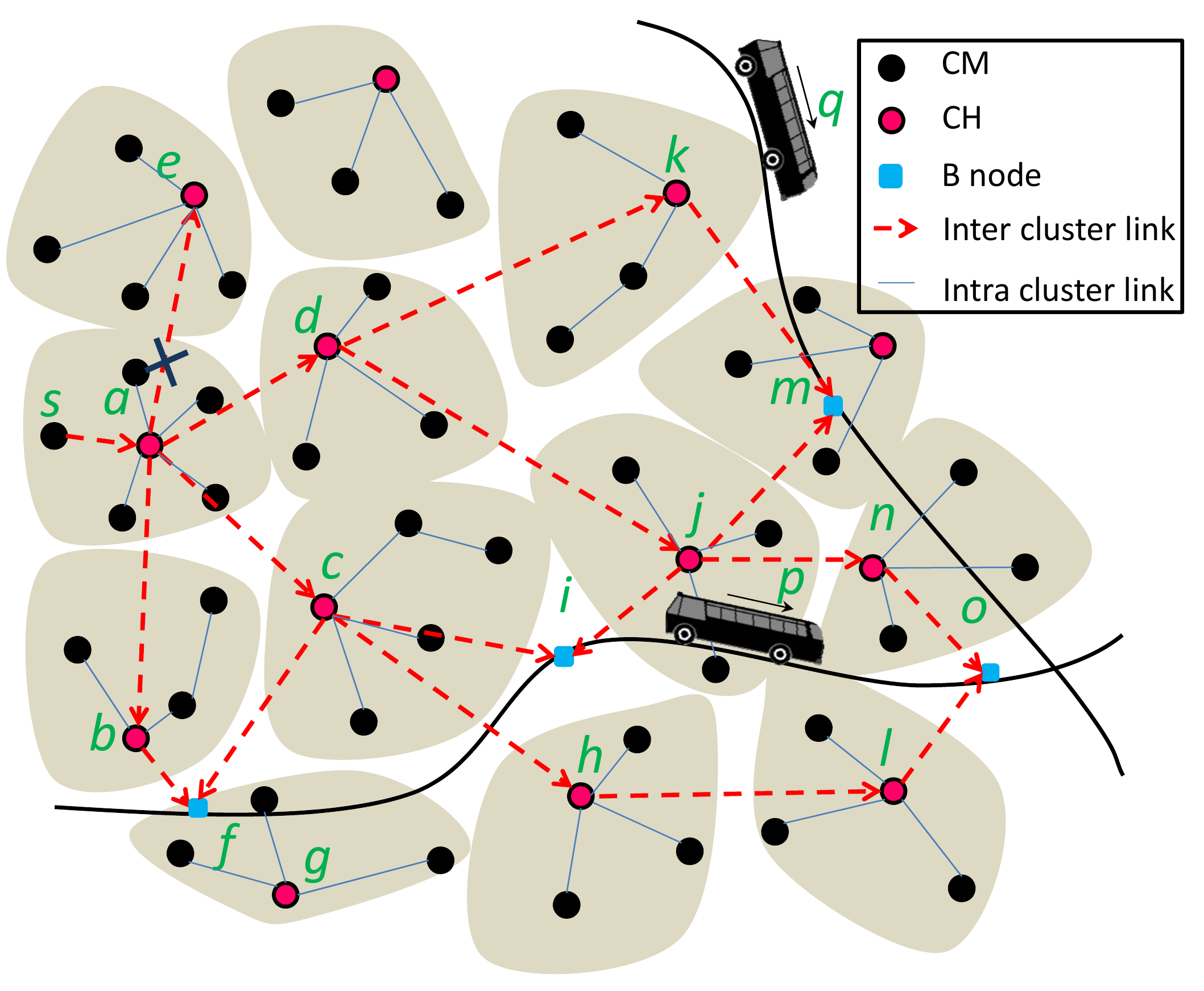}}
\caption{An illustrative example of I-UMDPC operation.}
\label{illustration}
\end{center}
\end{figure}

Any S node can be the source node, no matter if it acts as CH or CM. If a CH is the source node, Announcement is skipped; otherwise, CM announces its CH about the detection of interested event. Based on the cluster structure, every CH is able to learn the shortest path to each B node. Then, the source CH sends Quote to B nodes across the higher layer. As shown in Fig. \ref{illustration}, CH $a$ is the source CH of source node $s$ and it sends Quote to CH $b$, $c$ and $d$ rather than $e$, since $a$ is aware of the cluster structure. In Quote message, the hop count of being relayed is recorded. In the Quote released by source CH, hop count is 0, which is increased by 1 per relay. 

Once a B node receives a Quote message, in Step 3 it transmits a Reply message to the source CH through the reverse path. Such Reply message contains the B node ID (denoted by $i$), the arrival time of the coming bus (which is obtained from the store timetable), hop count $h_i$ and the energy cost $c_i$. The included information will be used in Target Selection in Step 4. 

Source CH may receive more than one Reply message. 
In Step 4, it executes Target Selection procedure, i.e., from the set of candidate targets (denoted by $R$) select final targets to deliver UM. We provide an optimization based target selection approach, which is formulated as a binary linear program. The output of this optimization is the set of the target B nodes, denoted by $T\subseteq R$. To set up our binary linear program, we define an indicator function $I(i)$, indicating whether B node $i$ belongs to $T$, i.e.,

\begin{equation}\label{indicator}
I(i)=
\begin{cases}
1,\ i\in T\\
0,\ otherwise\\
\end{cases}
\end{equation}

Based (\ref{indicator}), our objective is to minimize the total cost on delivering UM:

\begin{equation}\label{obj_func}
\min J=\sum_{i\in T} c_iI(i).
\end{equation} 
where $c_i$ is the cost to transmit data to B node $i$. 

Now, we consider the model of successful delivery at B node $i$. Let $X_i$ be the arrival instant of a bus at $i$ and $Y_i$ be the stop duration at $i$, with density functions $f_{X_i}(x_i)$ and $f_{Y_i}(y_i)$ respectively. Denote $Z_i$ as the instant when the bus departs from $i$. Then, $Z_i=X_i+Y_i$. Suppose that $X_i$ and $Y_i$ are independent, the density function of $Z_i$ is the convolution of $f_{X_i}(x_i)$ and $f_{Y_i}(y_i)$, i.e.,

\begin{equation}
f_{Z_i}(z_i) =\int_{-\infty}^{+\infty}f_{X_i}(x_i)f_{Y_i}(z_i-x_i)dx_i.
\end{equation}

Let $t_0$ be the time instant when source node detects an interested event. Let $t_i$ be the time cost to transmit a message from source CH to B node $i$. Then, the instant at which UM arrives at B node $i$ can be approximated by $t_o+3t_i$. Here, the 3 considers the time for transmitting Quote, Reply and UM. Note, this formulation neglects the time cost from source node to source CH. Let $\Delta$ be the allowed latency. For simplicity, we subtract $t_0$ from the arrival time of bus, depart time of bus, arrival time of sensory data at B node $i$. Without introducing new notations, from now on, both $X_i$ and $Z_i$ are $t_0$ free.

\begin{figure}[t]
    \centering
    \begin{subfigure}[b]{0.3\textwidth}
        \includegraphics[width=\textwidth]{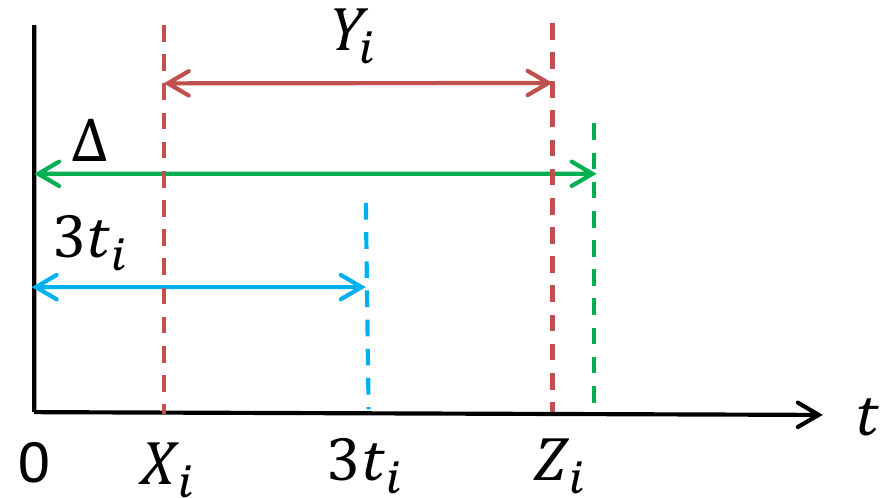}
        \caption{}
        \label{successful_delivery_1}
    \end{subfigure}
    \begin{subfigure}[b]{0.3\textwidth}
        \includegraphics[width=\textwidth]{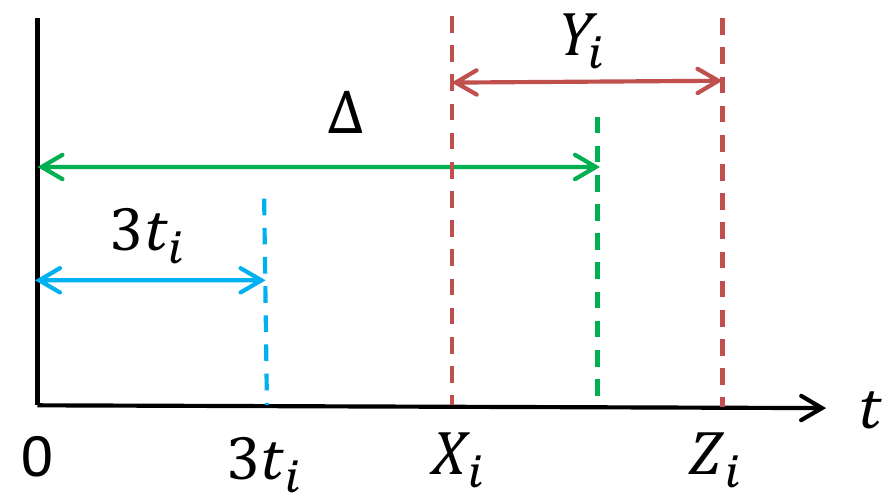}
        \caption{}
        \label{successful_delivery_2}
    \end{subfigure}
    \caption{Demonstrations of successful delivery. (a) Bus arrives early. (b) Data arrives early.}\label{successful_delivery}
\end{figure}

One fundamental requirement of successful delivery at B node $i$ is that UM should arrive at $i$ before the deadline $\Delta$: 
\begin{equation}\label{basic_requirement}
3t_i \leq \Delta.
\end{equation}
Otherwise, transmitting data to $i$ is in vain. Two specific situations of successful delivery are shown in Fig. \ref{successful_delivery}. Fig. \ref{successful_delivery_1} demonstrates Situation 1 where the bus arrives early and UM should arrive before the bus leaves. In this case, the condition is formulated as follows:
\begin{equation}\label{situation1}
X_i\leq 3t_i \leq Z_i.
\end{equation}
Fig. \ref{successful_delivery_2} demonstrates Situation 2 where UM arrives early. Then, it requires that the bus arrives before $\Delta$. In this case, the condition is formulated as follows:
\begin{equation}\label{situation2}
3t_i\leq  X_i \leq \Delta.
\end{equation}

Let $F_{X_i}(x_i)$ and $F_{Z_i}(z_i)$ be the cumulative distribution functions of $X_i$ and $Z_i$ respectively. The probability that Situation 1 occurs can be calculated by
\begin{equation}\label{P1}
P_1(i)=F_{X_i}(3t_i)(1-F_{Z_i}(3t_i)).
\end{equation} 
The probability that Situation 2 occurs can be calculated by
\begin{equation}\label{P2}
P_2(i)=F_{X_i}(\Delta)(1-F_{X_i}(3t_i)).
\end{equation}
Since Situation 1 and 2 are mutually exclusive, the probability of successful delivery at B node $i$ is
\begin{equation}\label{P}
P(i)=P_1(i)+P_2(i).
\end{equation}

We are in the position to introduce the delivery delay requirement. Let $\alpha$ be the threshold for the probability of 
successful delivery of sensory data to any M nodes. The delivery delay requirement is formulated as follows:
\begin{equation}\label{con_nonlinear}
1-\prod_{i\in T} (1-P(i))I(i) \geq \alpha.
\end{equation}

Requirement (\ref{con_nonlinear}) is also called $\alpha$ constraint and it can be transformed into a linear version as follows:
\begin{equation}\label{con_linear}
\sum_{i\in T} ln(1-P(i))I(i) \leq ln(1-\alpha).
\end{equation} 

Considering the basic requirement (\ref{basic_requirement}), we introduce a $\beta$ constraint as follow:
\begin{equation}\label{beta_constraint}
3\beta h_i I(i) \leq \Delta, i\in R.
\end{equation}
where $\beta h_i=t_i$ and $\beta$ is the time cost for one hop transmission.

With these definitions, our optimization problem looks for a set $S$ which minimizes the objective function (\ref{obj_func}) subject to the $\alpha$ constraint (\ref{con_linear}) and $\beta$ constraint (\ref{beta_constraint}). Note, the approach which does not consider the stop time uncertainty is named as UMDPC.

\section{Simulation Results}\label{simu}
The general idea of our simulations is to simulate an interested event at certain time and place, and check if the delivery approaches can successfully transmit UM to any M nodes within $\Delta$. We define that UM is successfully delivered if it is picked up by any M nodes from any target B nodes within $\Delta$. In (\ref{con_nonlinear}), we require that the probability of successfully delivered (PSD) should be no smaller than $\alpha$. 

We consider two metrics related to successful delivery:
\begin{itemize}
\item Ratio of successful delivery (RSD): the ratio of the number of UM successfully delivered to the total number of events occurred.
\item Ratio of pick-up (RPU): the ratio of the number of UMs picked up within $\Delta$ to the number of target B nodes which received UM.
\end{itemize}
RSD describes the reliability of the delivery approach. The larger the RSD, the more reliable the scheme is. RPU describes the efficiency of the delivery scheme. The higher RPU, the more efficient.

Since the B nodes are equipped with limited energy resource, w consider energy consumption on packet transmission as another metric. The energy dissipation model used in previous work, e.g., \cite{CHEN09}, is adopted:

\begin{equation}\label{transmitting}
E_t(l,d)=
\begin{cases}
l\times (E_{elec}+E_{fs}\times d^2),\ if\ d\leq d_0\\
l\times (E_{elec}+ E_{mp}\times d^4),\ if\ d> d_0\\
\end{cases}
\end{equation}
where $E_t(l,d)$ is the total energy dissipated to deliver a single $l$-bit packet from a transmitter to its receiver over a single link of distance $d$. The $E_{elec}$ depends on electronic factors such as digital coding, modulation, filtering, and spreading of the signal. The amplifier energy in free space $E_{fs}$ or in multipath environment $E_{mp}$ depends on the distance from the transmitter to the receiver. The threshold is $d_0$. The parameters in (\ref{transmitting}) are consistent with those in \cite{CHEN09}.

\subsection{Environment set up}

We test our approach in a simulated network, consisting of 400 S nodes randomly deployed in the eastern suburb of Sydney. There are several bus lines operating in this field and we select five of them (418, 395, 303, M50 and 353). We select 37 bus stops and each is associated with a B node. The B nodes store the timetable for the corresponding bus operations. They can also learn the distribution of actual arrival over long time observation. These five bus lines are with frequencies 10, 6, 13, 8 and 15 minutes respectively. We assume the arrival time and stop duration are independent and both follow Gaussian distribution. Thus, the departing time also follow Gaussian distribution, whose mean and standard deviation can be obtained from those of arrival time and stop duration.

We simulate an interested event: Event A and apply I-UMDPC, UMPDC and Stash \cite{lee2015predictive} to deliver UM. We conduct extensive simulations on Event A under different parameter sets. We display the illustrative results by the considered approaches in Fig. \ref{event_A}, where $\Delta=1.5$ minutes, $\alpha=0.997$ (three-sigma rule), and $\beta=1$ second. As seen in Fig. \ref{event_A}, the three approaches select 6, 9 and 15 target B nodes respectively.

\begin{figure}[t]
    \centering
    \begin{subfigure}[b]{0.4\textwidth}
        \includegraphics[width=\textwidth]{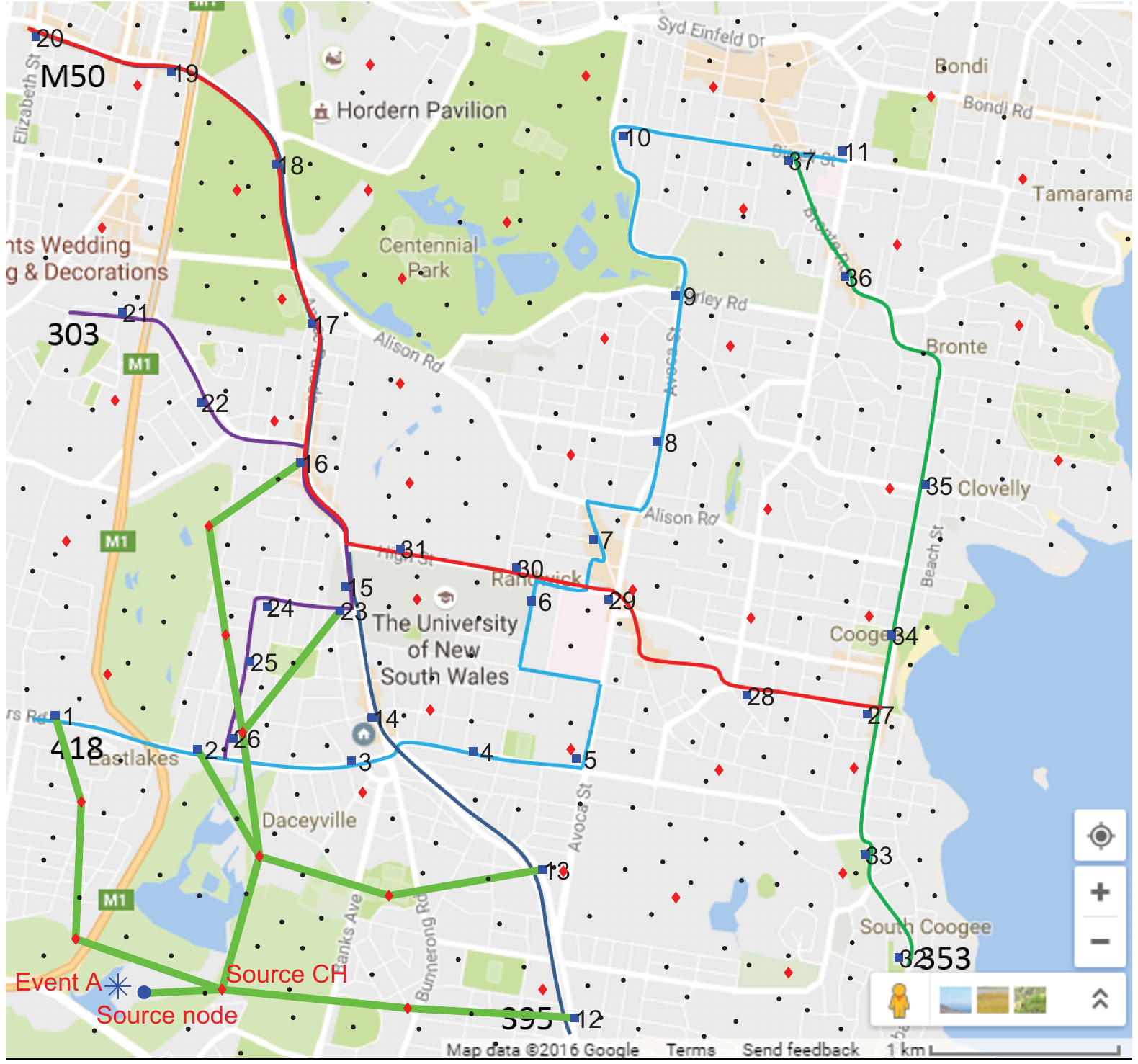}
        \caption{}
        \label{route_IUMDPC_15_1}
    \end{subfigure}
    \begin{subfigure}[b]{0.4\textwidth}
        \includegraphics[width=\textwidth]{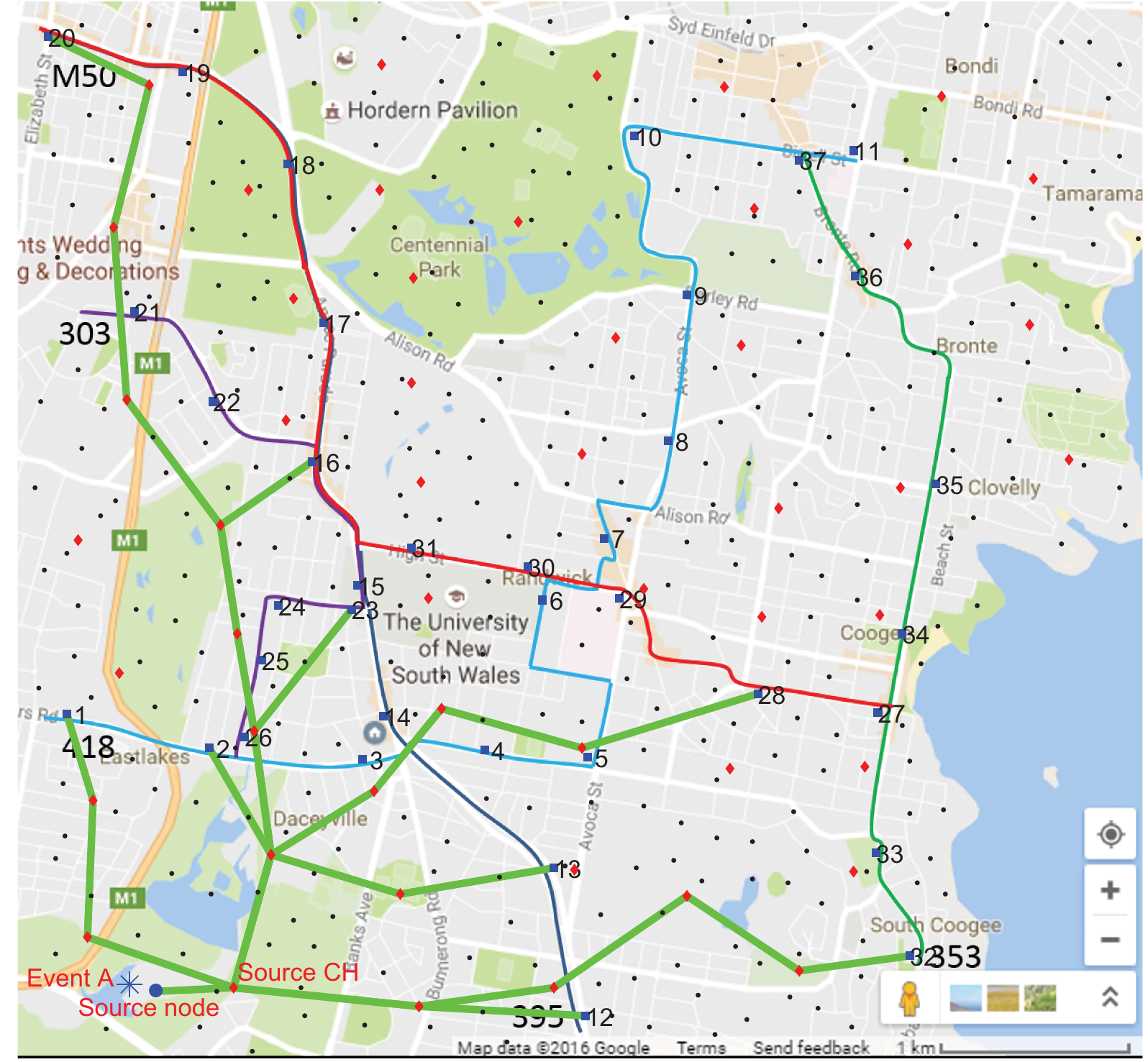}
        \caption{}
        \label{route_UMDPC_15_1}
    \end{subfigure}
    \begin{subfigure}[b]{0.4\textwidth}
        \includegraphics[width=\textwidth]{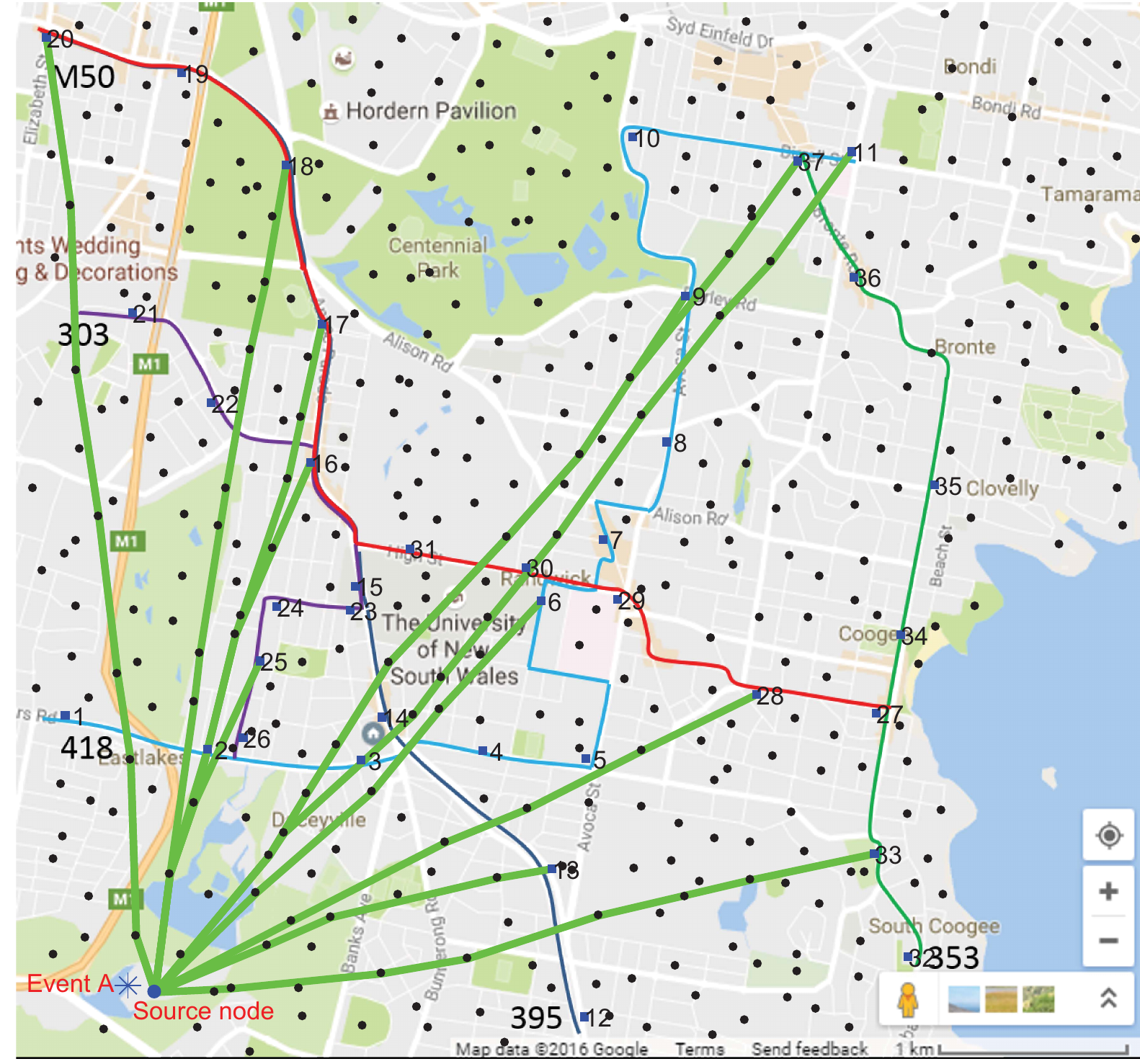}
        \caption{}
        \label{route_stash}
    \end{subfigure}
    \caption{UM delivery routes by I-UMDPC, UMDPC, and Stash for Event A ($\Delta=1.5$ minutes, $\alpha=0.997$, and $\beta=1$ second).}\label{event_A}
\end{figure} 

In the following parts, we further investigate the impacts of $\Delta$, $\alpha$ and $\beta$ on the performance of the considered approaches. Note, the number of B nodes also influences the performance of I-UMDPC, which has been reported in \cite{huang2017delay}.

\subsection{The influence of $\Delta$}\label{delta}
We investigate the influence of $\Delta$, which is an application dependent parameter. Here, $\alpha=0.997$ and $\beta=1$ second.

As seen from Fig. \ref{compare_energy_delta}, with the increase of $\Delta$, the energy consumed by I-UMDPC on delivery decreases. For UMDPC, when $\Delta$ is between 0.4 and 1.2 minutes, the energy consumption remains the same; while it decreases when $\Delta$ is between 1.3 and 1.7 minutes. This is because when $\Delta$ is smaller than 1.2 minutes, UMDPC is unable to find feasible solution of the optimization problem. Thus, all 37 B nodes are selected as targets (see Fig. \ref{compare_Bno_delta}). I-UMDPC and UMDPC consume the same energy on route construction. I-UMDPC consumes less on delivery than UMDPC. Stash spends more energy than I-UMDPC and UMDPC in both route construction and delivery. These results are consistent with the numbers of target B nodes as shown in Fig. \ref{compare_Bno_delta}.

\begin{figure}[t]
    \centering
    \begin{subfigure}[b]{0.45\textwidth}
        \includegraphics[width=\textwidth]{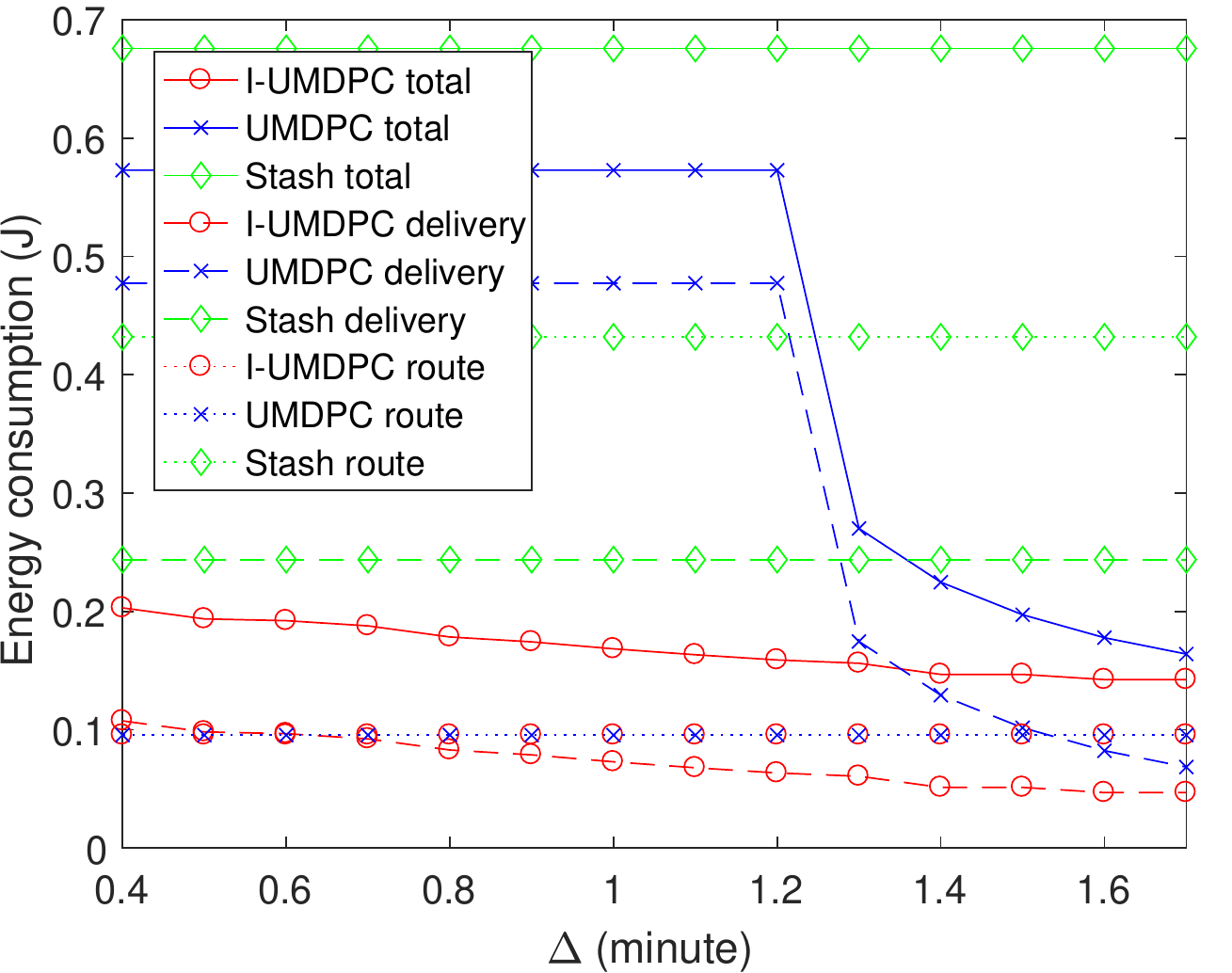}
        \caption{}
        \label{compare_energy_delta}
    \end{subfigure}
    \begin{subfigure}[b]{0.45\textwidth}
        \includegraphics[width=\textwidth]{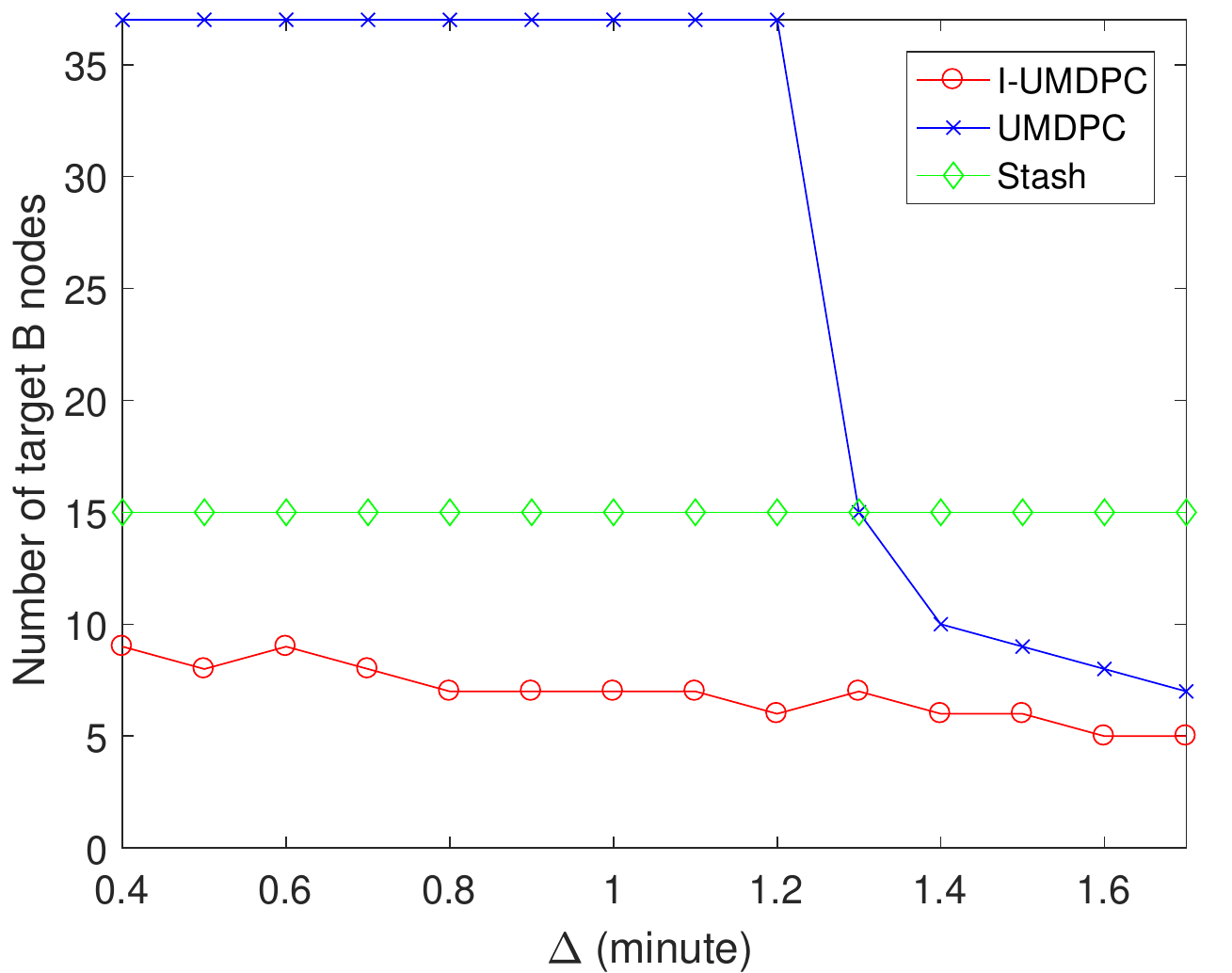}
        \caption{}
        \label{compare_Bno_delta}
    \end{subfigure}
    \begin{subfigure}[b]{0.45\textwidth}
        \includegraphics[width=\textwidth]{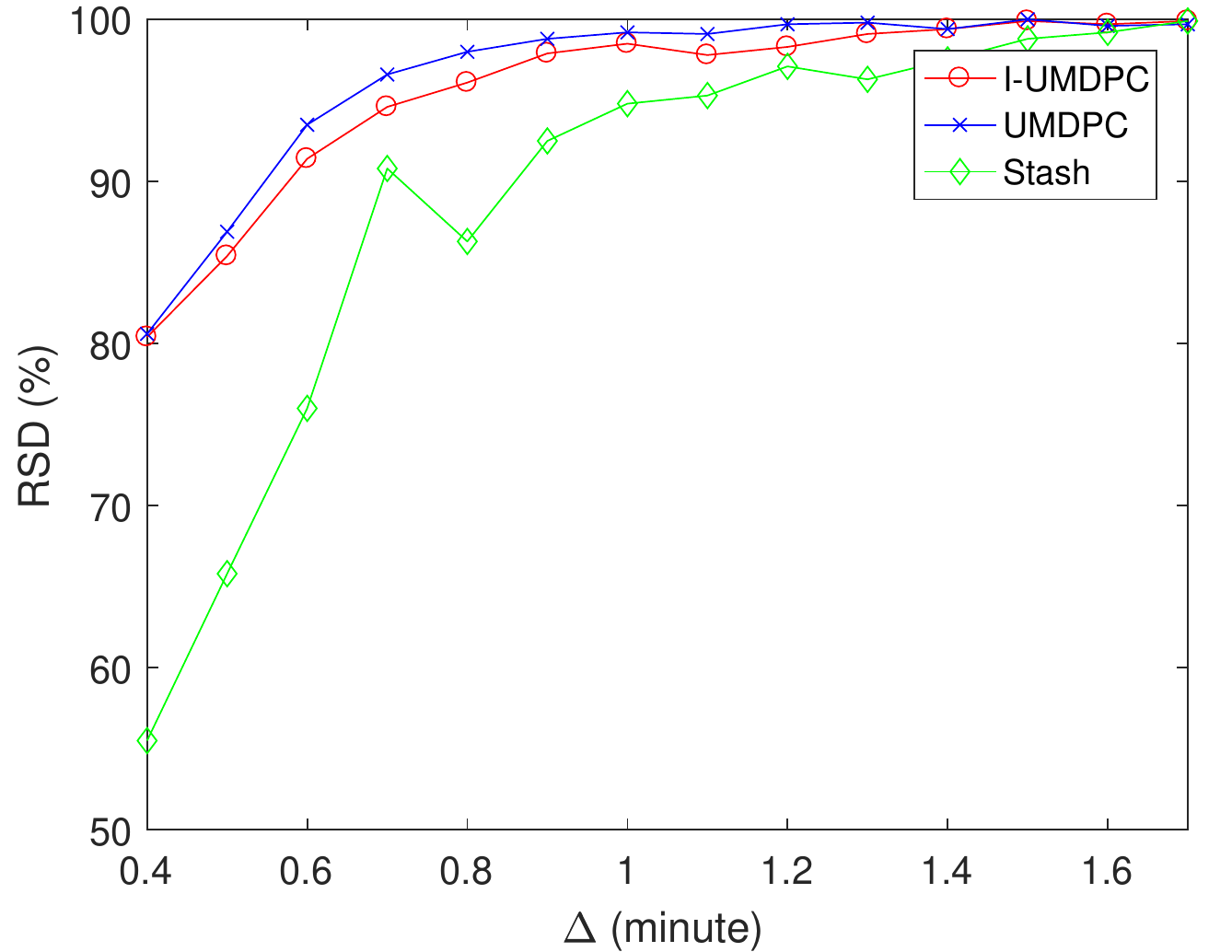}
        \caption{}
        \label{compare_RSD_delta}
    \end{subfigure}
    \begin{subfigure}[b]{0.45\textwidth}
        \includegraphics[width=\textwidth]{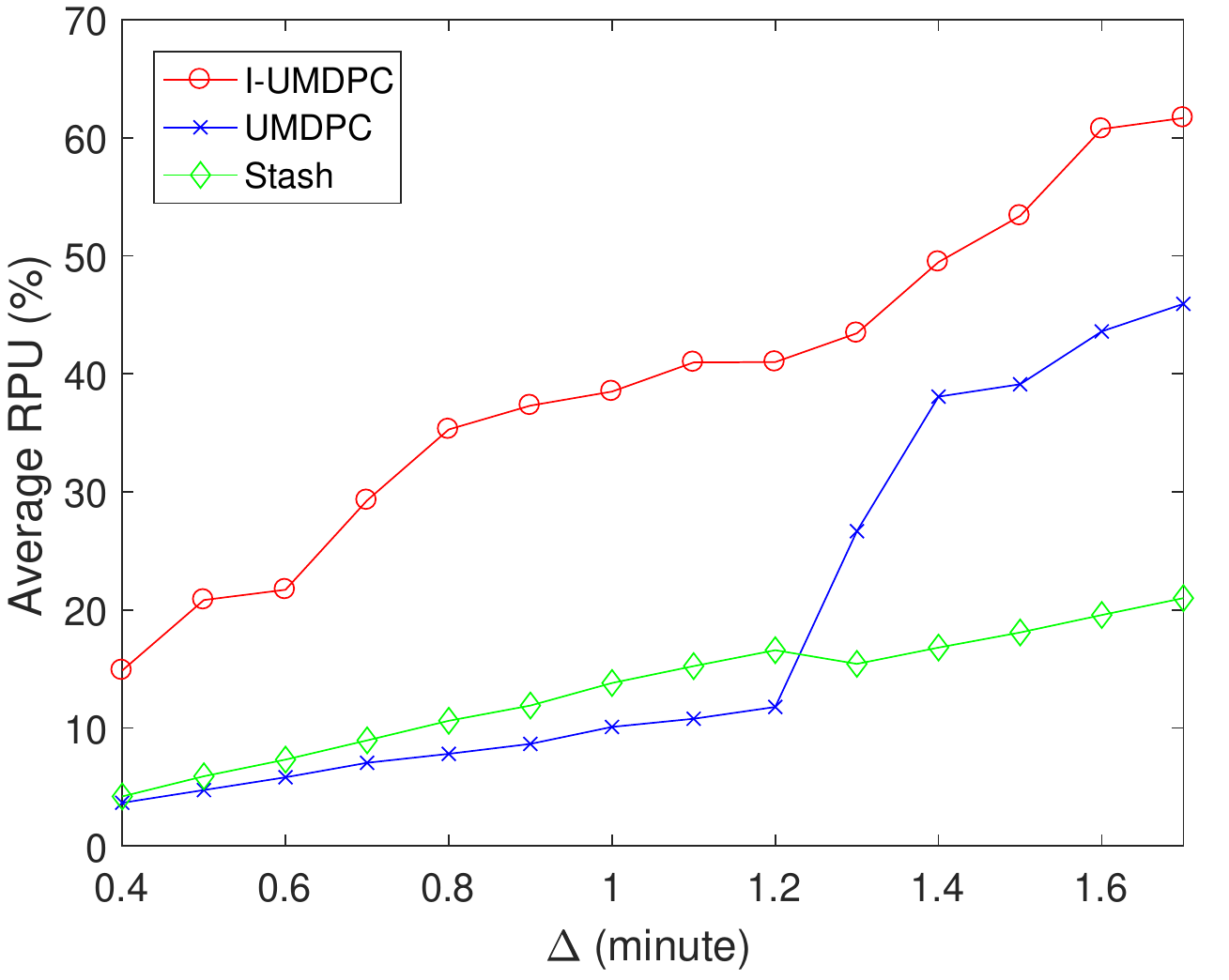}
        \caption{}
        \label{compare_RPU_delta}
    \end{subfigure}
    \caption{Comparison of I-UMDPC, UMDPC and Stash under different $\Delta$. (a) Energy consumption. (b) Number of target B nodes. (c) RSD. (d) RPU.}
\end{figure} 

%

As shown in Fig. \ref{compare_RSD_delta}, the RSDs of them increase with $\Delta$. When $\Delta$ is between 0.4 and 1.2, the RSD of UMDPC is higher than that of I-UMDPC. The reason is that all the B nodes are selected as targets by UMDPC. When $\Delta$ is larger than 1.2, I-UMDPC and UMDPC achieve similar RSD. We also notice that I-UMDPC and UMDPC perform better than Stash. The influence of $\Delta$ on RPU is demonstrated in Fig. \ref{compare_RPU_delta}. The trend of RPU is similar to RSD. When $\Delta \leq 1.2$, the RPU of UMDPC is lower than Stash; while when $\Delta \geq 1.3$, the RPU of UMDPC is larger than Stash. The reason is that Stash prefers the B nodes which are cheap to deliver to, while UMDPC focuses more on the reliability of delivery. I-UMDPC outperforms UMDPC and Stash for all the considered $\Delta$. 

%

From the above analysis, we conclude that the larger (smaller) $\Delta$, the easier (more difficult) to deliver UM successfully, leading to lower (higher) energy consumption and number of target B nodes, and higher (lower) RSD and RPU.

\subsection{The influence of $\alpha$}\label{alpha}
This section considers the impact of $\alpha$. Here, $\Delta=1.5$ minutes and $\beta=1$ second. Since I-UMDPC and UMDPC consume the same amount of energy on route construction and Stash is non-sensitive to the investigated parameters in terms of energy consumption, we only display the total amounts here and in the next section.

Here, $\alpha$ takes 0.970, 0.980, 0.990, 0.997 and 0.999. With the increasing of $\alpha$, the energy consumptions and the numbers of target B nodes by I-UMDPC and UMDPC increase as shown in Fig. \ref{compare_energy_alpha} and Fig. \ref{compare_Bno_alpha}, respectively. The reason lies in $\alpha$ constraint (\ref{con_nonlinear}). For the fixed $\Delta$ and $\beta$, $\alpha$ constraint impacts on the target B nodes. Generally, the more strict (\ref{con_nonlinear}), i.e., the larger $\alpha$, the more target B nodes are required, which further influences the energy consumption by target B nodes. Fig. \ref{compare_RSD_alpha} shows that the I-UMDPC and UMDPC become more reliable with higher $\alpha$. However, the efficiencies of them decrease with $\alpha$, see Fig. \ref{compare_RPU_alpha}. Similar to the simulations in Section \ref{delta}, Stash is non-sensitive to $\alpha$, thus the energy consumption and B node number remain constant with varying $\alpha$. RSD and RPU change slightly due to the random simulation of bus operation. 

From these simulations we can observe that I-UMDPC achieves similar reliability with UMDPC and outperforms UMDPC as well as Stash in terms of efficiency.

\begin{figure}[t]
    \centering
    \begin{subfigure}[b]{0.45\textwidth}
        \includegraphics[width=\textwidth]{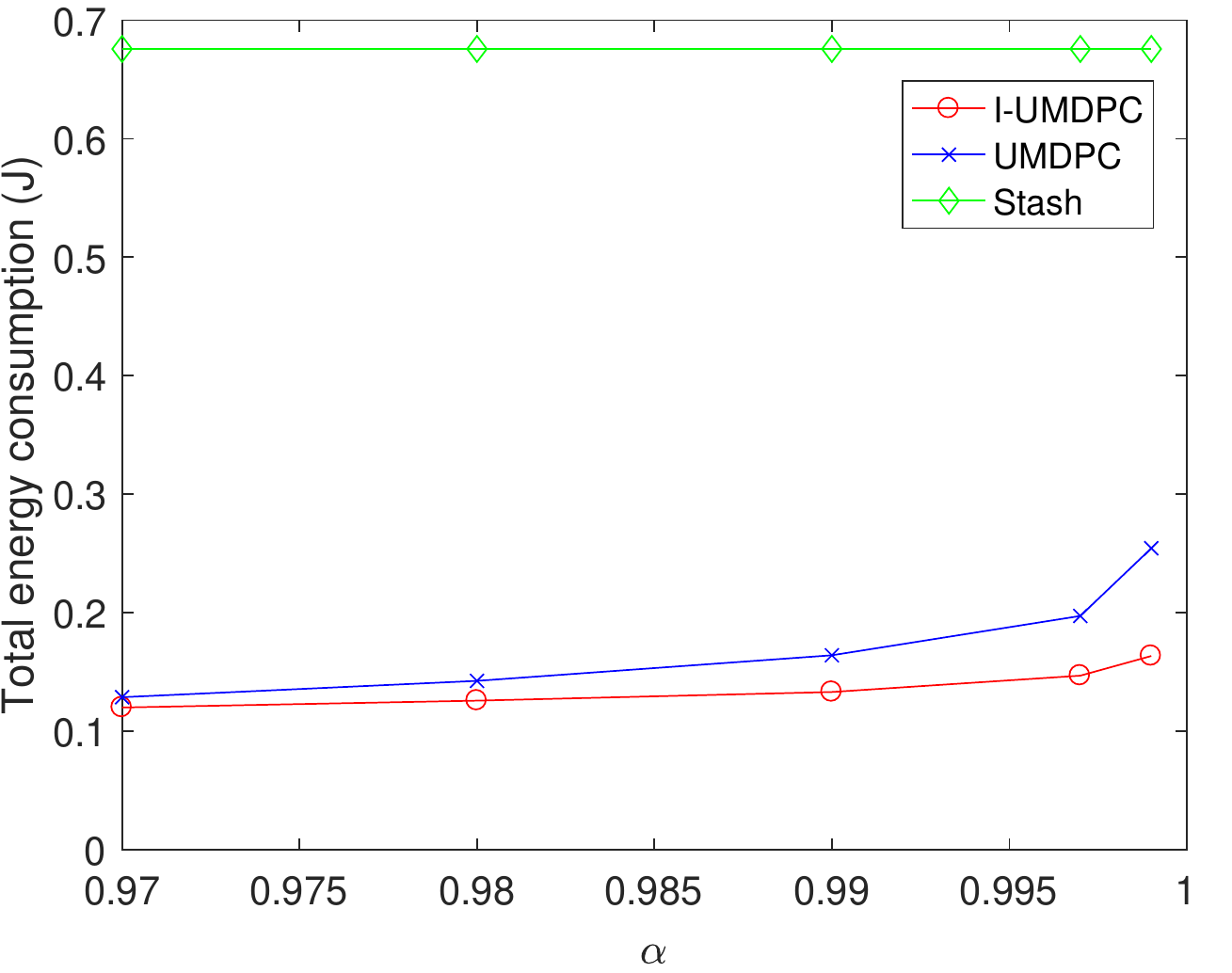}
        \caption{}
        \label{compare_energy_alpha}
    \end{subfigure}
    \begin{subfigure}[b]{0.45\textwidth}
        \includegraphics[width=\textwidth]{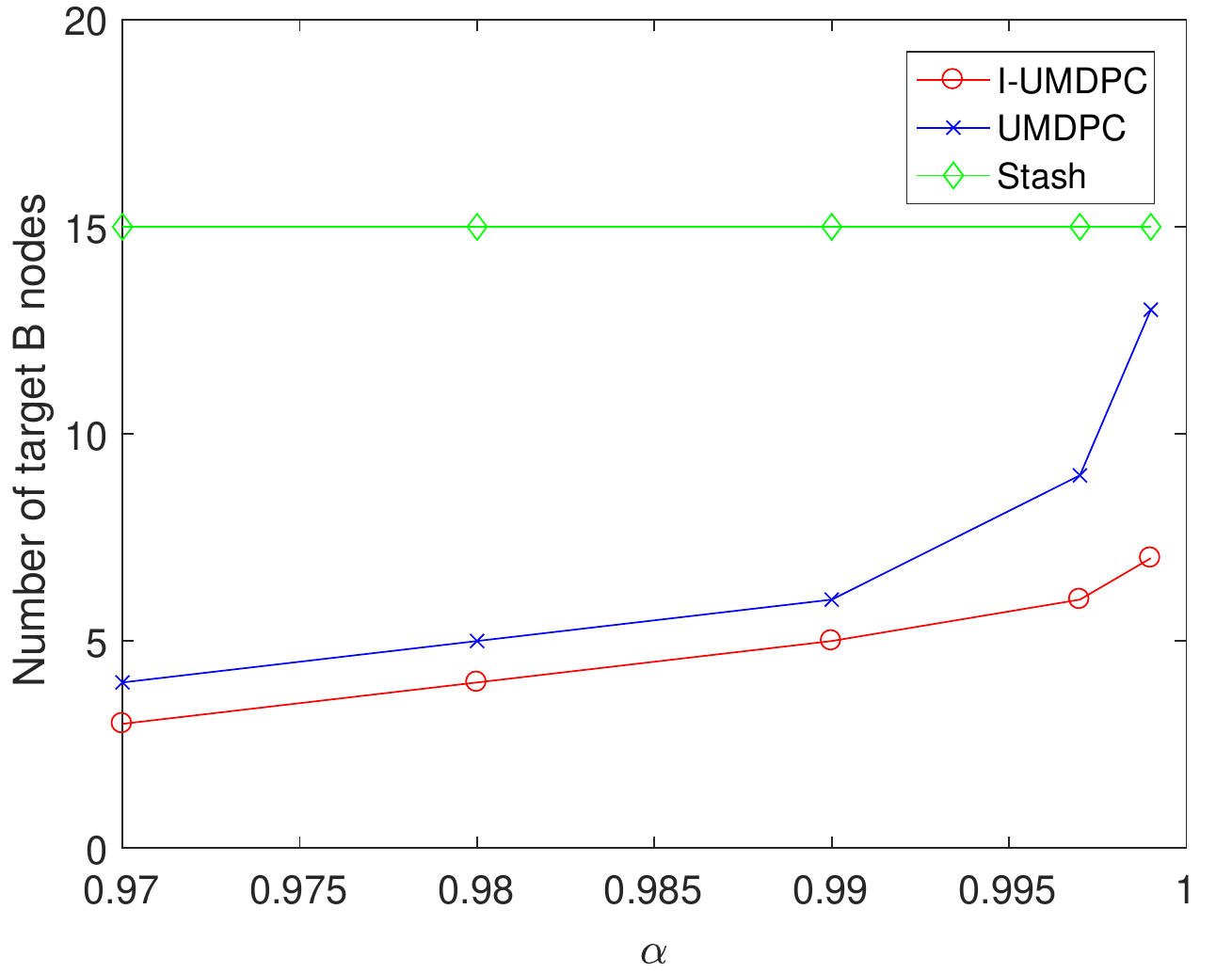}
        \caption{}
        \label{compare_Bno_alpha}
    \end{subfigure}
    \begin{subfigure}[b]{0.45\textwidth}
        \includegraphics[width=\textwidth]{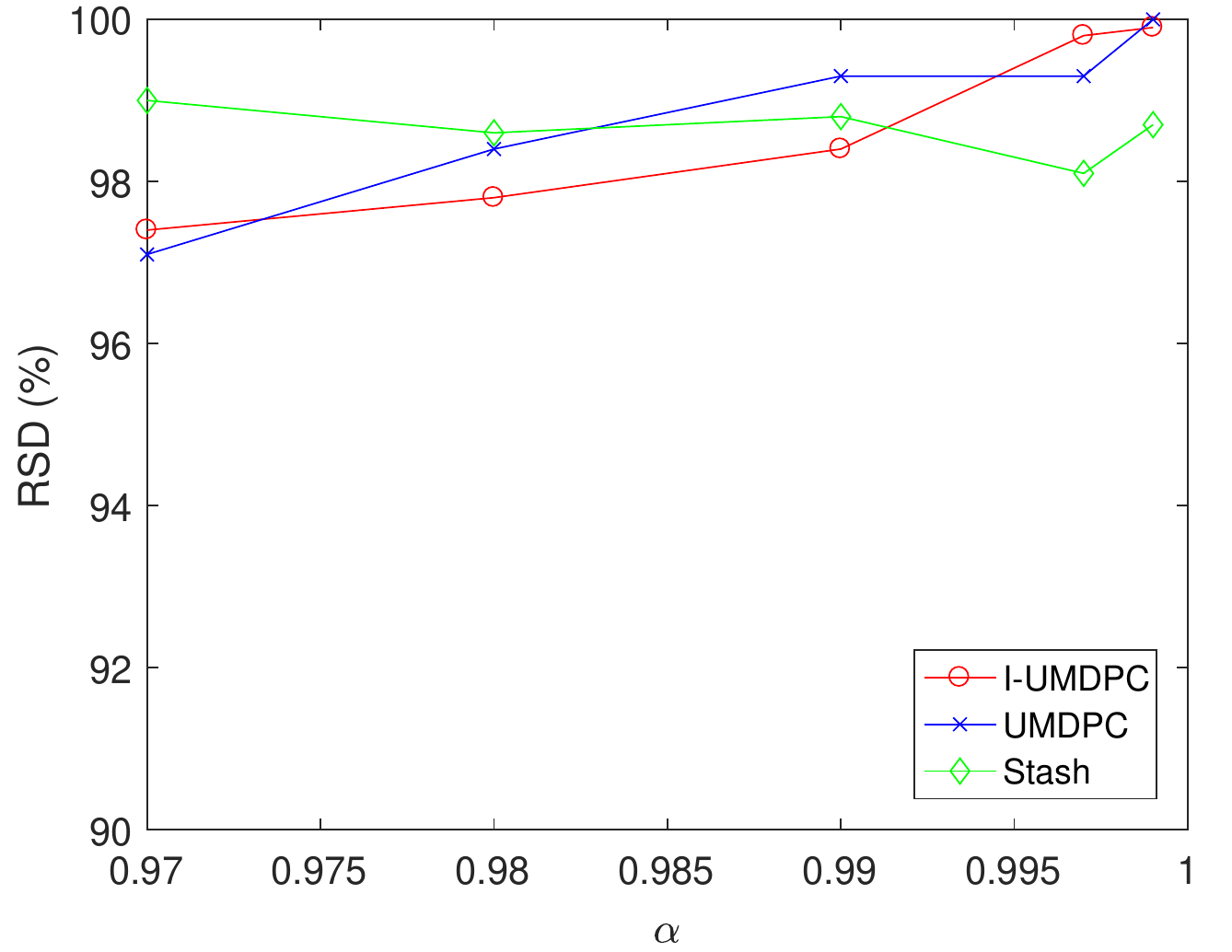}
        \caption{}
        \label{compare_RSD_alpha}
    \end{subfigure}
    \begin{subfigure}[b]{0.45\textwidth}
        \includegraphics[width=\textwidth]{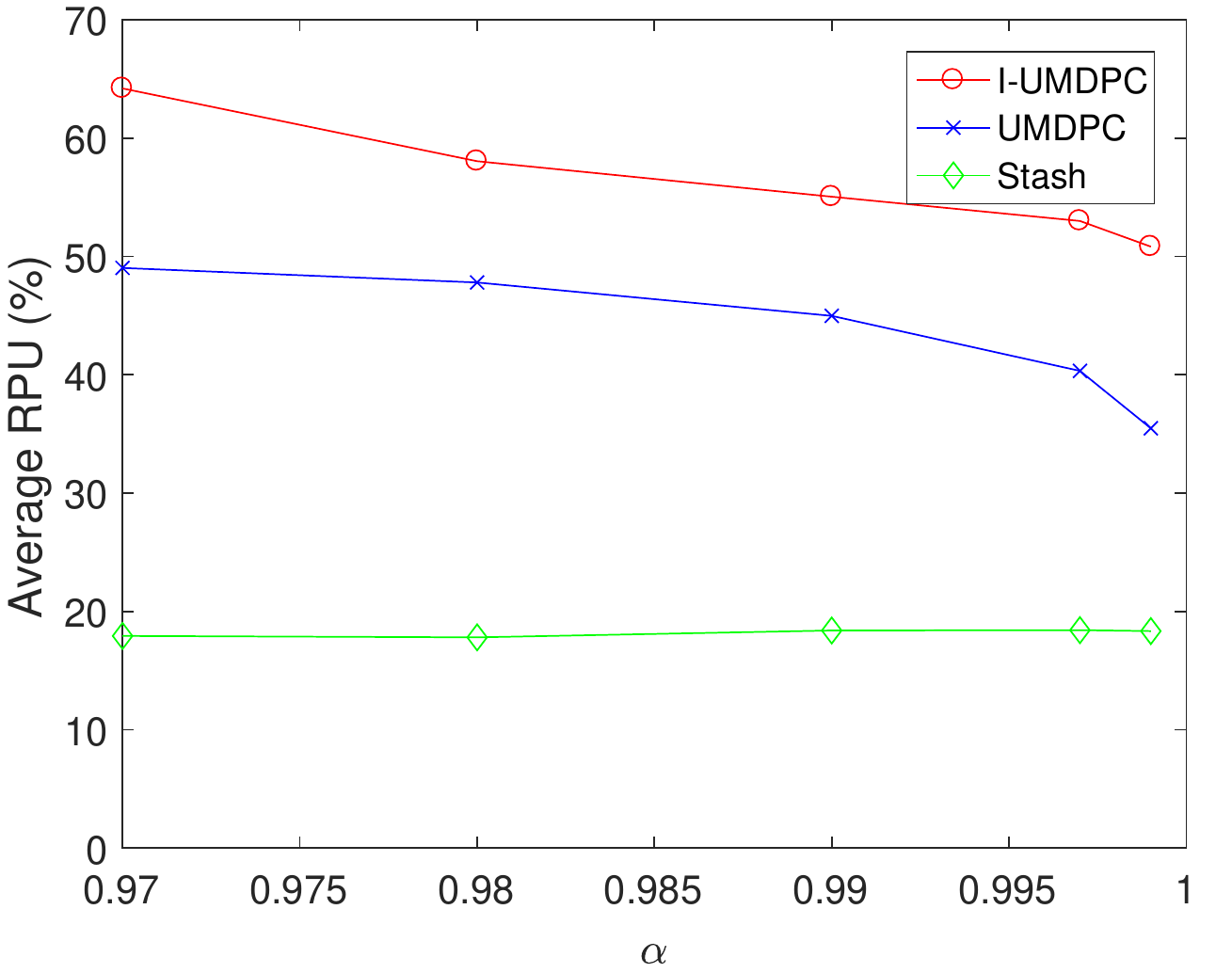}
        \caption{}
        \label{compare_RPU_alpha}
    \end{subfigure}
    \caption{Comparison of I-UMDPC, UMDPC and Stash under different $\alpha$. (a) Energy consumption. (b) Number of target B nodes. (c) RSD. (d) RPU.}
\end{figure} 

%
%
%

\subsection{The influence of $\beta$}\label{beta}
This section investigates the influence of $\beta$. $\beta$ is a parameter determining the time cost to delivery message to a B node, i.e., $t_i=\beta h_i$, for B node $i$. Moreover, it impacts on the probability of collecting a UM from a B node, see Eq. (\ref{P1}), (\ref{P2}) and (\ref{P}). For small $\beta$, the probability of Situation 2 (i.e., the probability  that UM arrives early (\ref{P2}), see Fig. \ref{successful_delivery_2}) is dominant in (\ref{P}). In this case, with the increase of $\beta$, the probability $P$ decreases. In contrast, for large $\beta$, the probability of Situation 1 (i.e., the probability  that M node arrives early (\ref{P1}), see Fig. \ref{successful_delivery_1}) is the dominant in (\ref{P}). Through simulations we find that with the increase of $\beta$, $P$ decreases. From the above analysis we can obtain that, the larger $\beta$, the smaller $P$. This is also the reason that the energy consumptions as well as the numbers of target B nodes by I-UMDPC and UMDPC raise with the increase of $\beta$, see Fig. \ref{compare_energy_beta} and \ref{compare_Bno_beta}. Another interesting phenomenon in Fig. \ref{compare_energy_beta} and \ref{compare_Bno_beta} is that when $\beta$ is larger than 2.25 seconds, target B node number by UMDPC is bounded by the number of available B nodes. This is because that there is no solution to the optimization problem in UMDPC. In the meanwhile, the energy consumption also increases and is bounded. 

\begin{figure}[t]
    \centering
    \begin{subfigure}[b]{0.45\textwidth}
        \includegraphics[width=\textwidth]{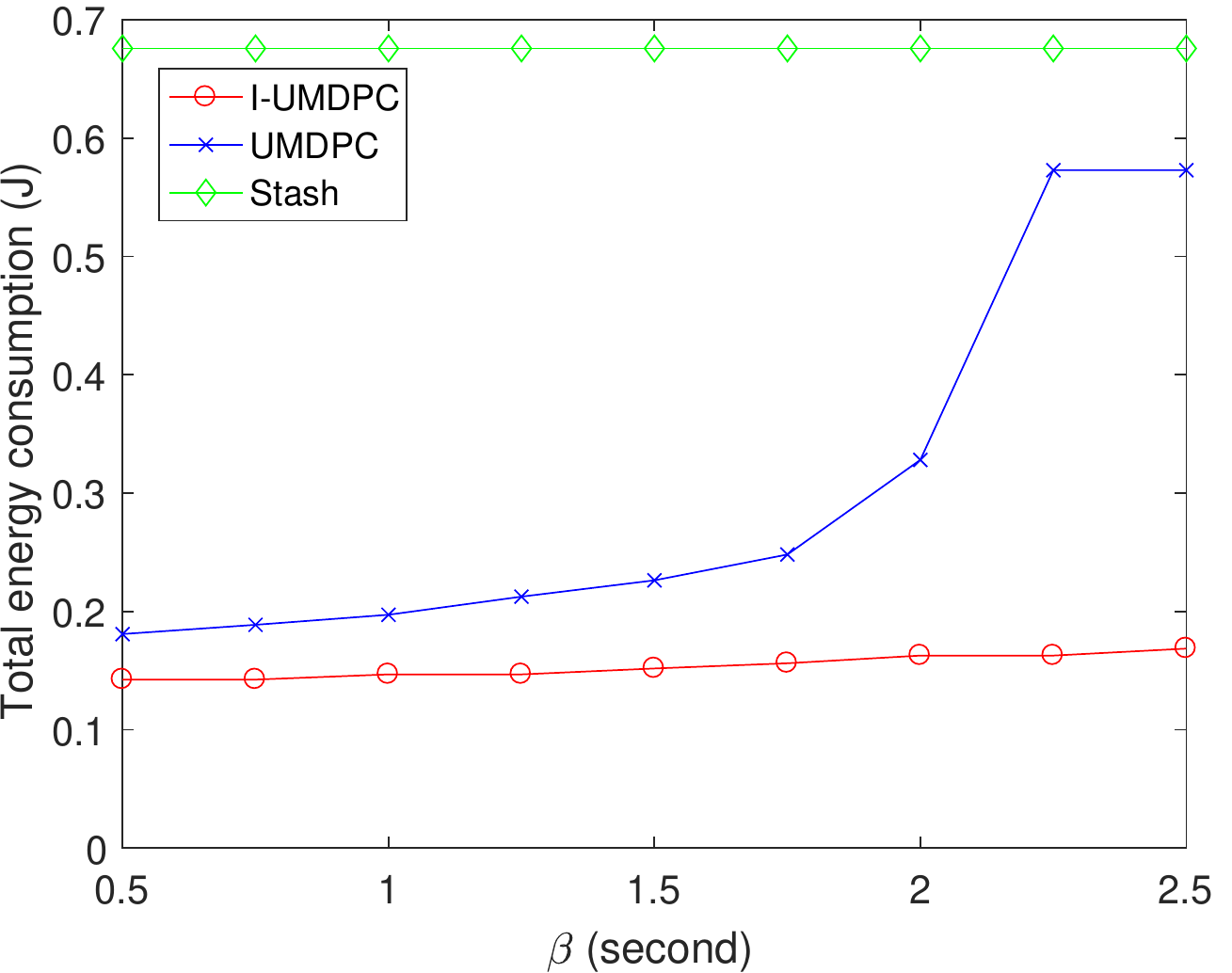}
        \caption{}
        \label{compare_energy_beta}
    \end{subfigure}
    \begin{subfigure}[b]{0.45\textwidth}
        \includegraphics[width=\textwidth]{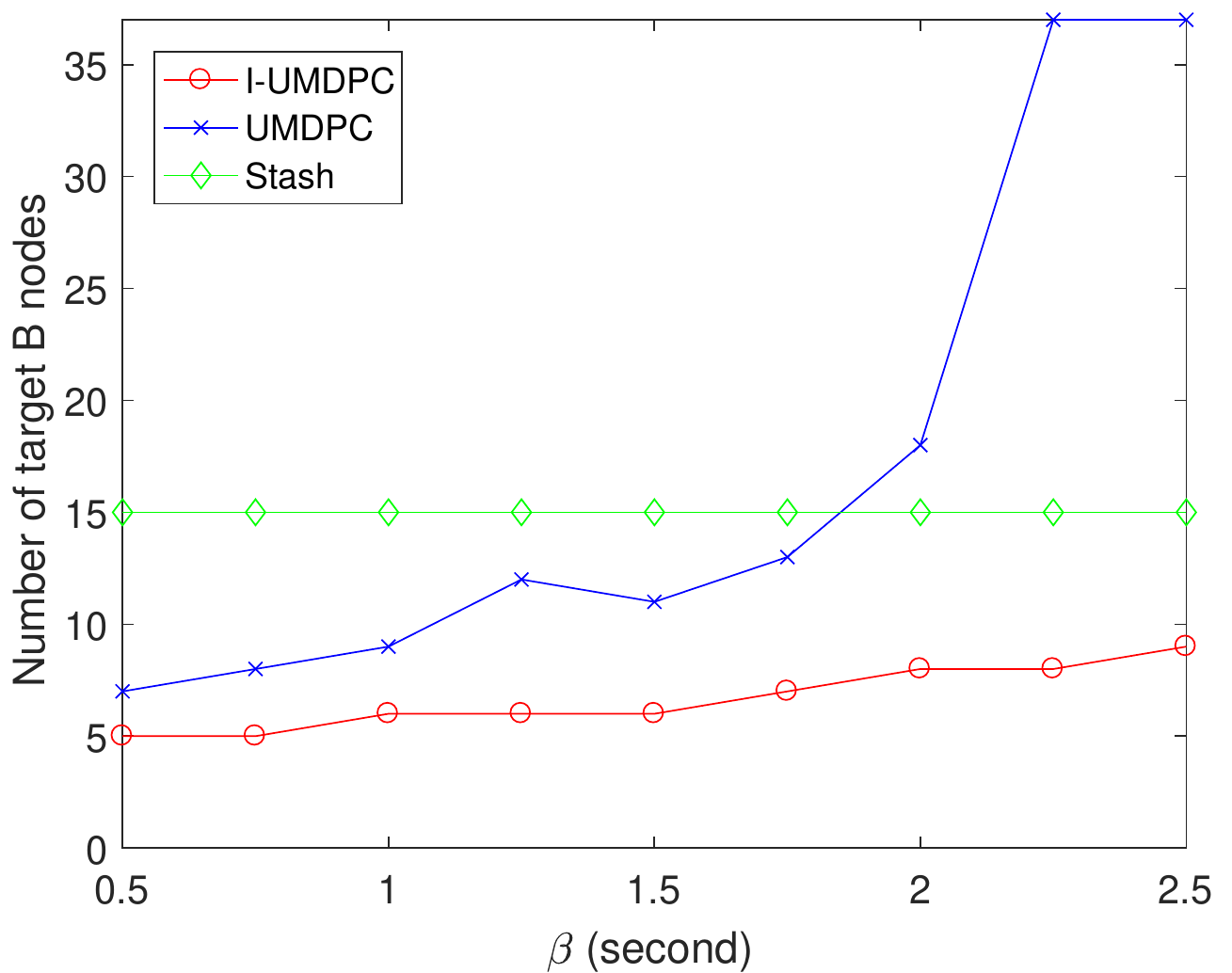}
        \caption{}
        \label{compare_Bno_beta}
    \end{subfigure}
    \begin{subfigure}[b]{0.45\textwidth}
        \includegraphics[width=\textwidth]{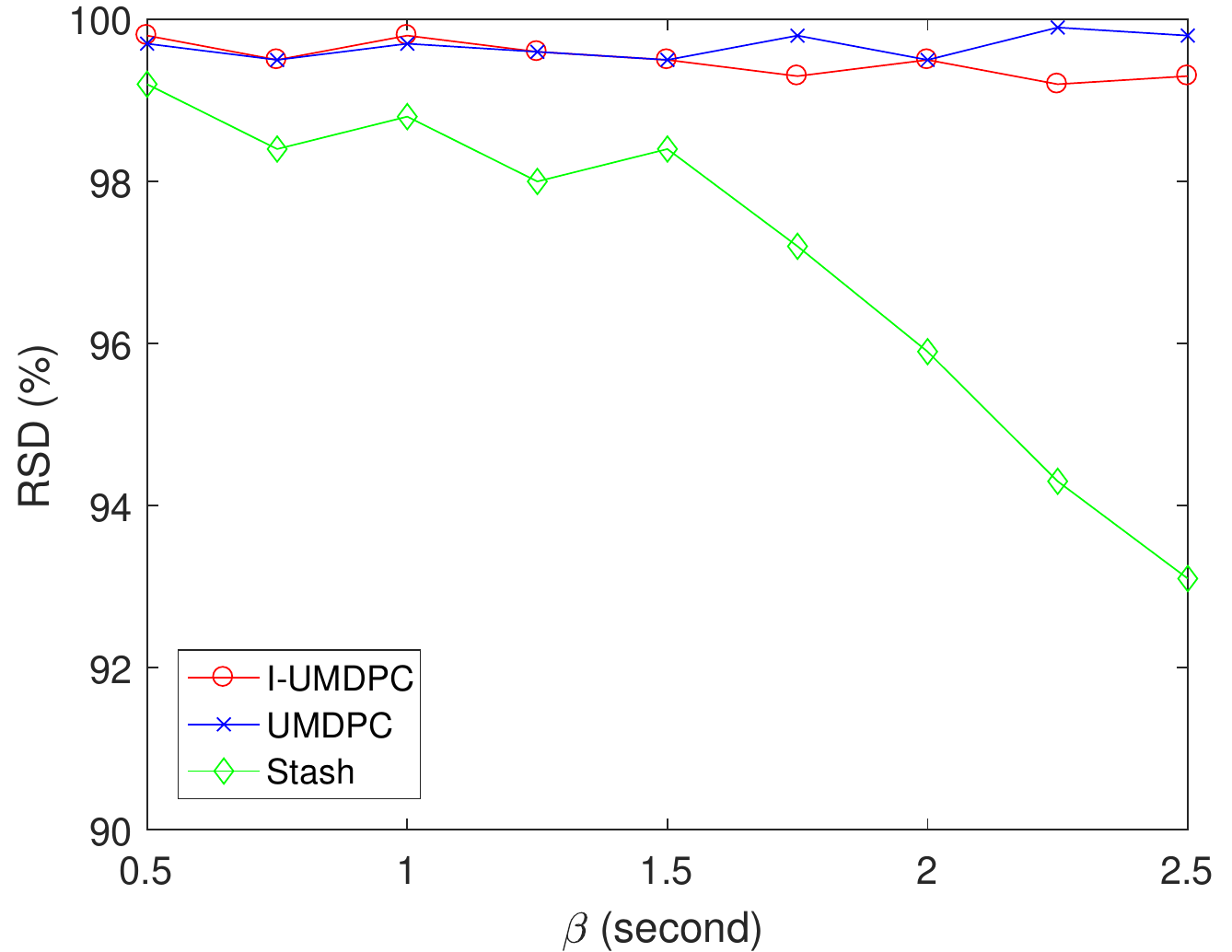}
        \caption{}
        \label{compare_RSD_beta}
    \end{subfigure}
    \begin{subfigure}[b]{0.45\textwidth}
        \includegraphics[width=\textwidth]{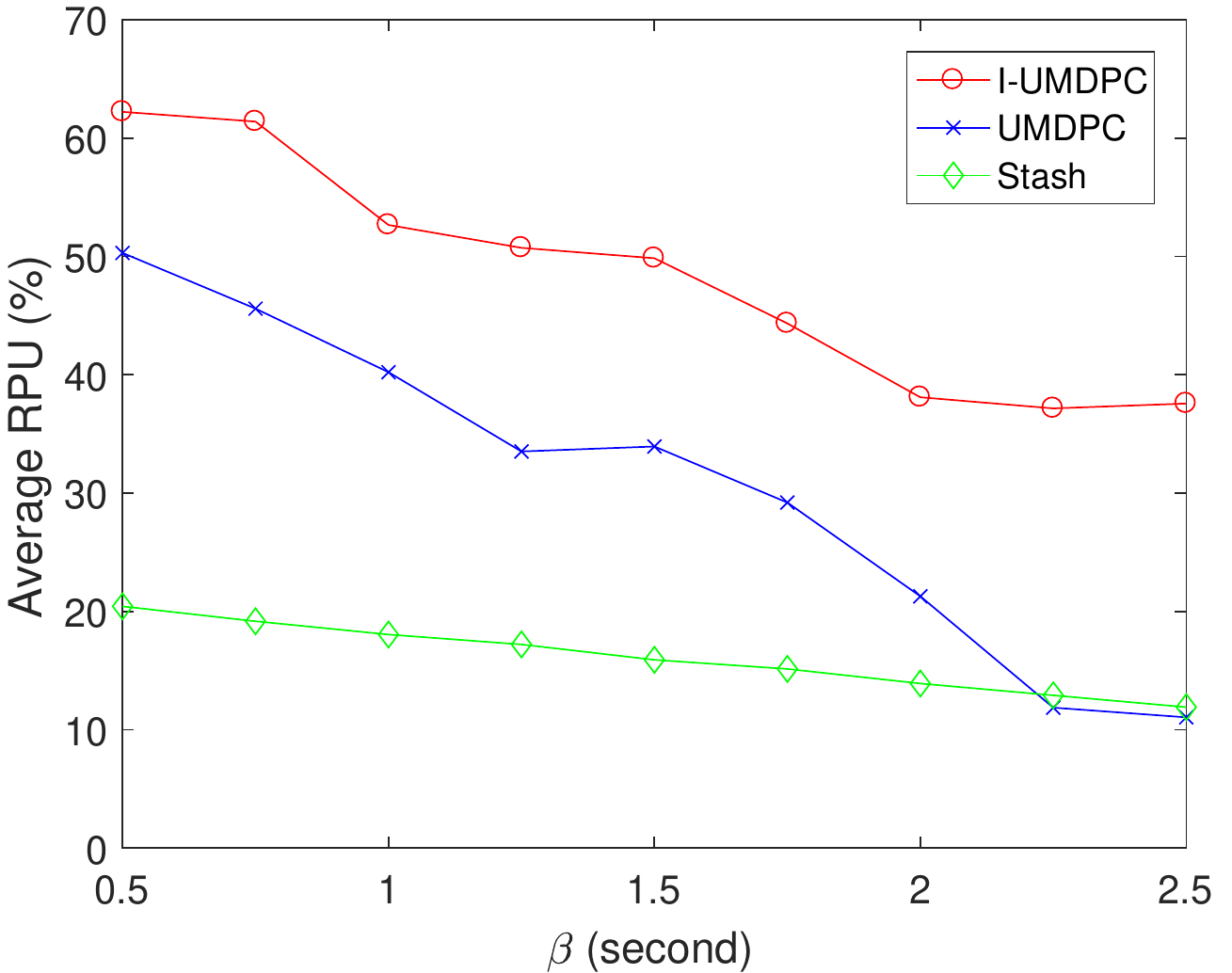}
        \caption{}
        \label{compare_RPU_beta}
    \end{subfigure}
    \caption{Comparison of I-UMDPC, UMDPC and Stash under different $\beta$. (a) Energy consumption. (b) Number of target B nodes. (c) RSD. (d) RPU.}
\end{figure} 
 
%

In terms of RSD, I-UMDPC achieves the similar performance as UMDPC when $\beta$ is smaller than 2.25 seconds. When $\beta$ takes larger values, all the B nodes are selected as targets, then RSD of UMDPC definitely achieves the best, see Fig. \ref{compare_RSD_beta}. However, the situation of RPU is different. With the increase of $\beta$, the RPUs of all the three approaches decrease. Among them, I-UMDPC performs at least 20\% better than UMDPC and around 2 times better than Stash. Therefore, I-UMDPC performs similarly to UMDPC in reliability but better than UMDPC in efficiency.

%
\section{Experimental Results} \label{experiment}
We carried out implementations of the proposed approach on our testbed, which consists of 42 ESP8266 module based S nodes, see Fig. \ref{deployment}. Each node has a DHT11 module to monitor the ambient temperature and an event is detected if the reading is over the threshold. 16 ESP8266 modules serve as B nodes. We design 4 paths (see Fig. \ref{topology}) for M nodes and 8 volunteers carrying ESP8266 modules act as buses on the paths. Each of the M nodes has a specific timetable and the B nodes along its path know it. 

\begin{figure}[t]
    \centering
    \begin{subfigure}[b]{0.6\textwidth}
        \includegraphics[width=\textwidth]{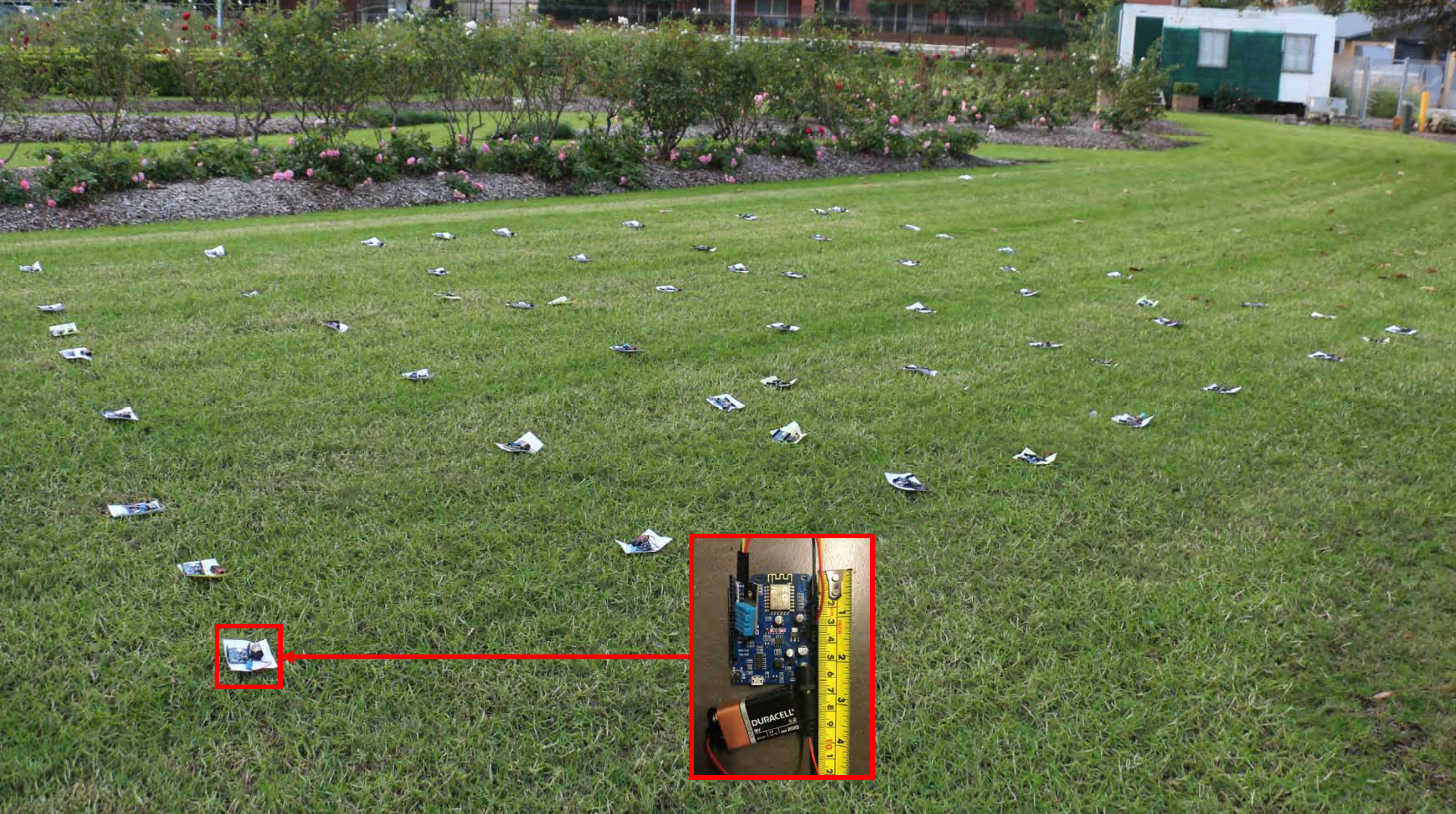}
        \caption{}
        \label{deployment}
    \end{subfigure}
    \begin{subfigure}[b]{0.4\textwidth}
        \includegraphics[width=\textwidth]{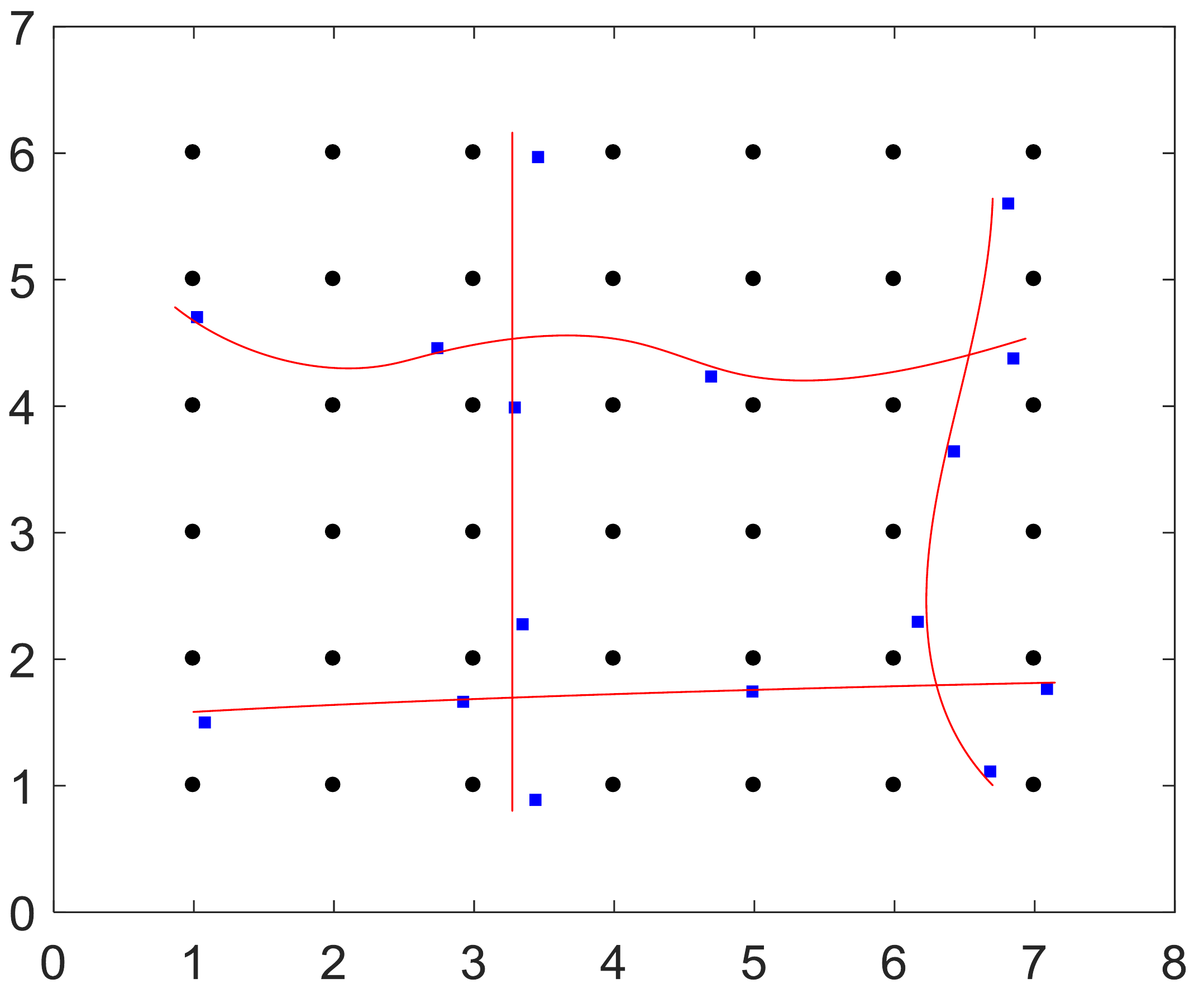}
        \caption{}
        \label{topology}
    \end{subfigure}
    \begin{subfigure}[b]{0.4\textwidth}
        \includegraphics[width=\textwidth]{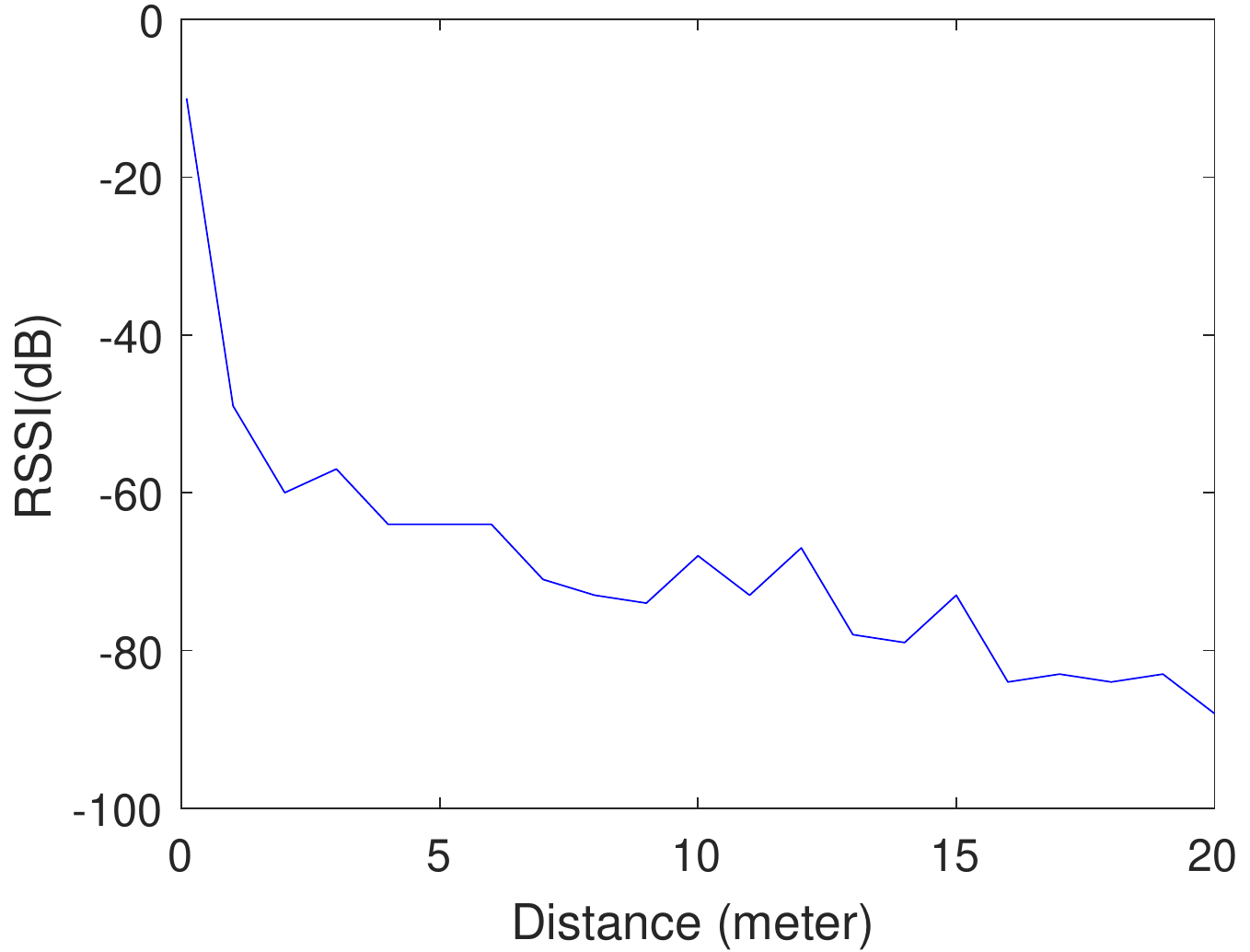}
        \caption{}
        \label{rssi}
    \end{subfigure}
    \caption{Node deployment for I-UMDPC experiments. (a) Node deployment. (b) Visualization with M nodes' paths. Black circles are S nodes; Blue square are B nodes and the red curves are M nodes' paths. (c) RSSI against distance.}\label{deployment_experiment}
\end{figure}

Since the testbed does not have enough nodes as the simulations above and clustering is not the focus of this chapter, clustering is not implemented in the experiments. We assume the topology of the network is known to the nodes. These nodes communicate with each via WiFi. To make a node communicate with only its neighbour nodes, we set a threshold for the S and B nodes, i.e., -40dB, corresponding to 1 meter, according to the tested Received Signal Strength Indicator (RSSI) in Fig. \ref{rssi}. To make B nodes communicate with M nodes only when the latter arrives, we set another threshold for B nodes, i.e., -20dB, corresponding to about 0.2 meters. In the simulations, we use $\beta h_i$ to estimate the time to transmit packet between source node and B node $i$. However, since $\beta$ is unknown in practice, that estimation is not precise. Besides, different nodes may have varying $\beta$. Thus, in the experiments, such time requirement is measured by the source node, which takes into account the time of processing and transmitting. 

We selected each S node as the source node once by making a fire near it and we totally did 42 experiments for each of the considered approaches. We set $\alpha$ as 0.997, $\Delta$ as 1 minute. We summarise the experiment results in Fig. \ref{result_experiment}. Fig. \ref{targetBnumber_experiment} shows that Stash always select 8 target B nodes since there are 8 M nodes. This number by I-UMDPC is always no larger than UMDPC since the former considers the stop time. Fig. \ref{RPU_experiment} shows the RPU of Stash is usually the smallest. In other words, the selection strategy is inefficient. But there are cases where its RPU is larger than I-UMDPC and UMDPC, for example when the source node is 24 and 25. I-UMDPC achieves the best performance in RPU among the three approaches in most experiments, while there are also cases where UMDPC is performs the best. For example, when the source node is 19, I-UMDPC chose 2 target B nodes while UMDPC chose a third one. One of the 2 B nodes failed to upload UM while the third one uploaded. Thus, the RPU of I-UMDPC is 50\% while that of UMDPC is 67\%. In summary, I-UMDPC performed the best in 93\% cases, which  showed its effectiveness.

\begin{figure}[t]
    \centering
    \begin{subfigure}[b]{0.7\textwidth}
        \includegraphics[width=\textwidth]{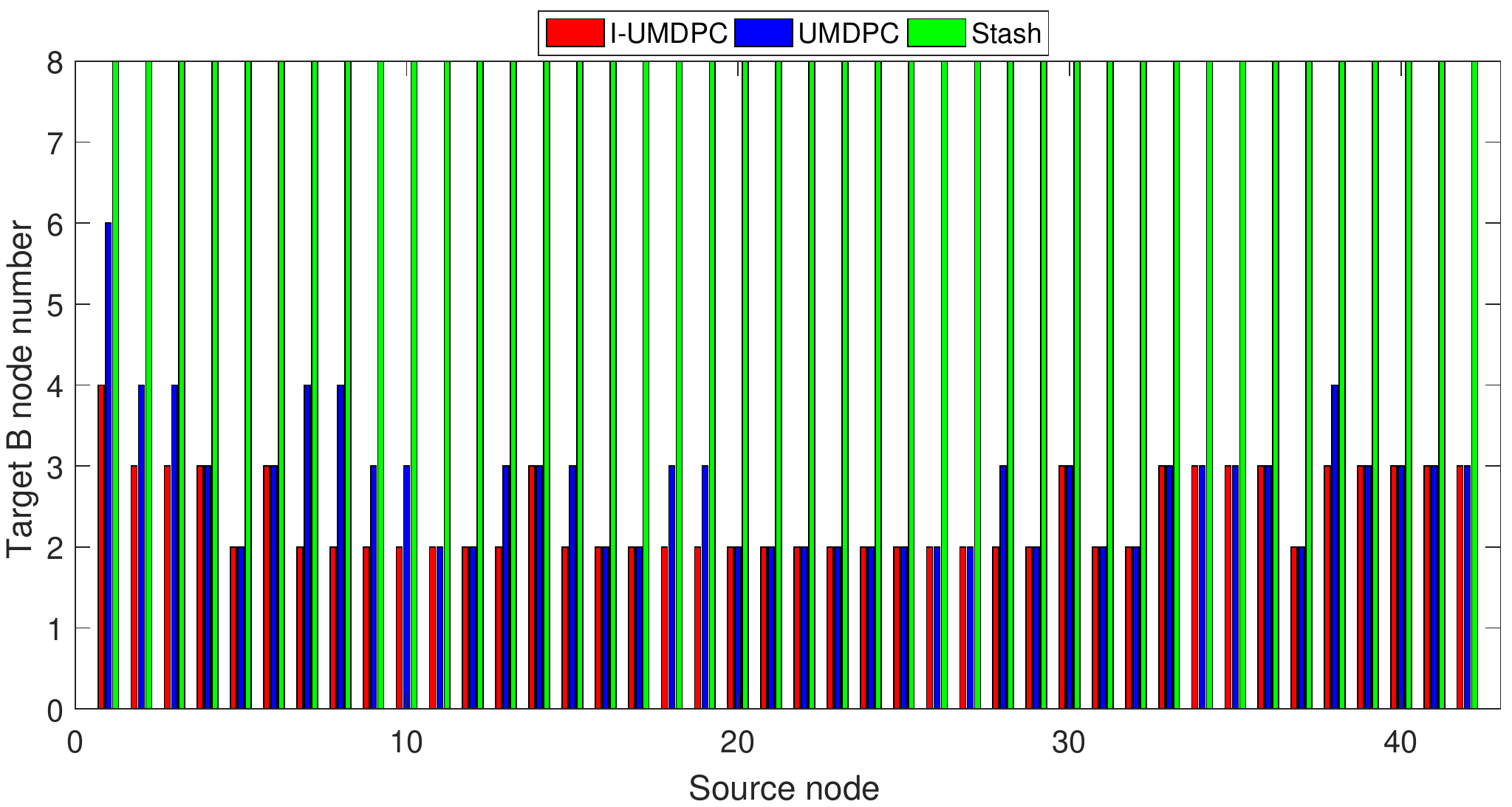}
        \caption{}
        \label{targetBnumber_experiment}
    \end{subfigure}
    \begin{subfigure}[b]{0.7\textwidth}
        \includegraphics[width=\textwidth]{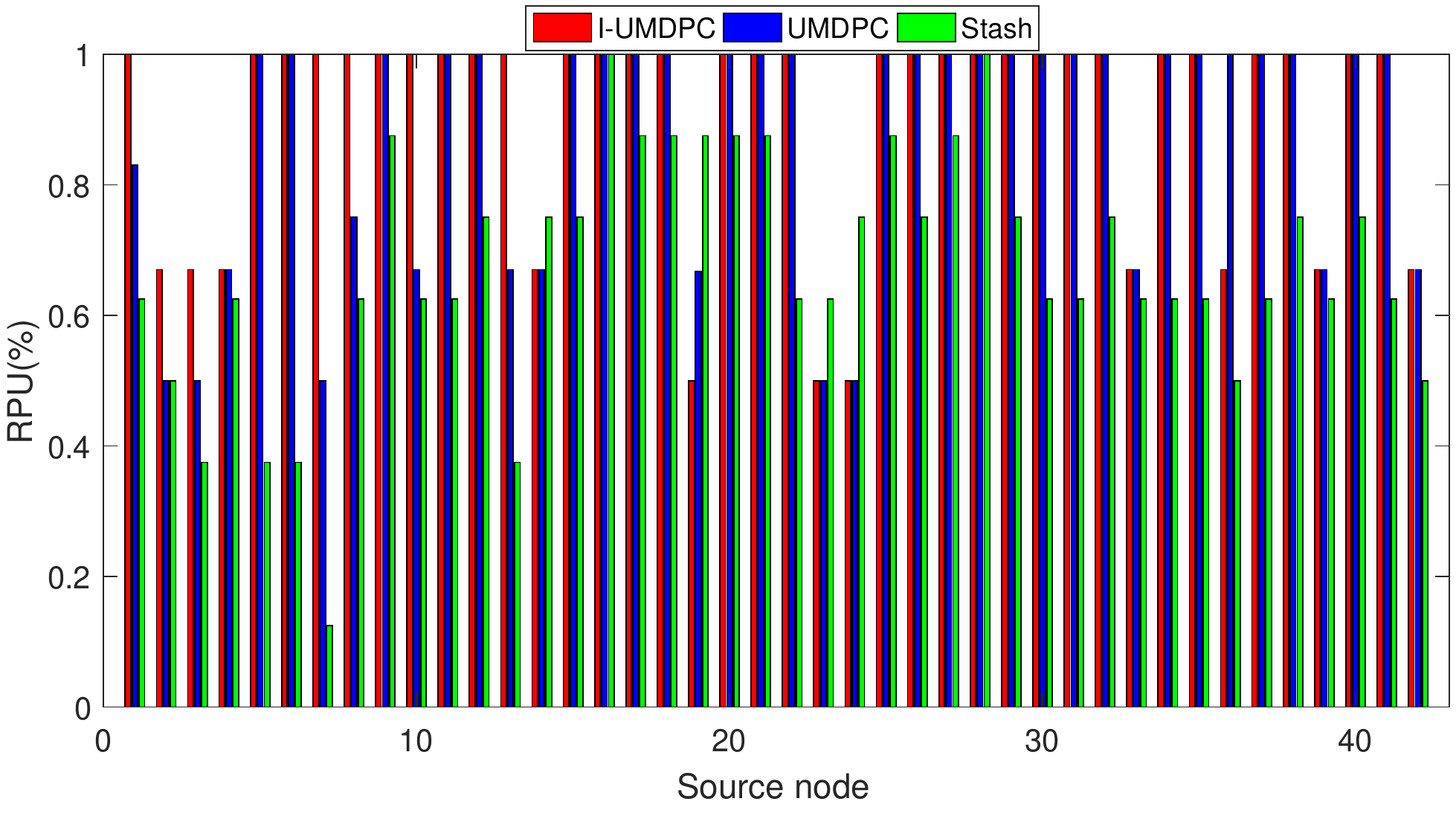}
        \caption{}
        \label{RPU_experiment}
    \end{subfigure}
    \caption{Comparison of I-UMDPC, UMDPC and Stash in real experiments. (a) Target B node number. (b) RPU.}\label{result_experiment}
\end{figure}

\section{Summary}\label{Summary}
This chapter considers the applications of wireless sensor networks where the users require fresh report of an event. We define that the report message is fresh if it is delivered to M nodes within the given latency. We present an Improved-Unusual Message Delivery Path Construction approach (I-UMDPC) for UM delivery. I-UMDPC combines cluster-based routing, M nodes which are attached to buses and an optimization based target B nodes selection procedure. The features of bus operation, i.e., the uncertain arrival time and stop duration at a bus, are formulated by random variables and accounted in the optimization problem. Through extensive simulations as well as experiments on our testbed, we show that I-UMDPC is able to deal with the uncertainties and deliver UM to M nodes with higher reliability and efficiency than the alternatives.

Although the models used here are quite practical, some aspects still need to be investigated, such as the failure in transmission, which is a key factor impacting on the collection latency.

\chapter{Optimized Deployment of Autonomous Drones to Improve User Experience in Cellular Networks}\label{drone}
\minitoc
Distinguishing from Chapters \ref{path_planning}, \ref{cluster_ms}, \ref{cluster_cs} and \ref{cluster_um}, this chapter considers the cellular networks. Similar to Chapters \ref{cluster_ms}, \ref{cluster_cs} and \ref{cluster_um}, we aim at using mobility \textit{constrained} drones to improve user experience.

\section{Motivation}\label{motivation}
Explosive demands for mobile data are driving mobile operators to respond to the challenging requirements of higher capacity and improved quality of user experience (QoE) \cite{nakamura2013trends}. Deploying more Base Stations (BSs) is able to meet the increasing traffic demand. This solution, however, may not only result in more cost for equipments and site rental, but also bring with other issues, such as a high percentage of BSs having low utility in non-peak hours. In this context, the utilization of autonomous drones, which work as flying BSs, could be a more efficient solution than network densification. 

Recently, the US government has approved a resolution to increase commercial use of drones\footnote{https://www.voanews.com/a/trump-ok-test-program-expand-domestic-drone-flights/4085752.html}. Arguably, drones will play a more significantly important role in our daily life in the future. In this chapter, we focus on one of the key issues of content delivery using drones: drone deployment. Generally, the drone deployment problem has been studied to find out the optimal 3D positions for drones to serve user equipments (UEs) on a 2D plane. The altitudes of drones are restricted by local regulations. For example, in the US, the maximum allowable altitude is 120 meters above the ground \cite{H_usa}. In Australia, the drones should be more than 30 meters away from people \cite{H_australia}. 
Here, we consider a scenario where drones are flying at a fixed altitude within the allowed range. The drone deployment problem is a bit similar to the problem of optimal sensor placement in control \cite{savkin2001problem, savkin2002hybrid}. The difference is that in control the optimal sensor placement is done in the time domain, while the drone deployment considered in this chapter is in the spatial domain.

Different from most existing work where UEs are assumed to be randomly distributed or following a predefined distribution in a 2D environment \cite{ding2016performance}, here we consider a more realistic scenario. We focus on urban environment, and only consider outdoor UEs. We propose a street graph and the UEs to be served by drones are near streets. This assumption is reasonable in urban environment as when users enter buildings, they can switch to Wi-Fi to access the network. Furthermore, instead of focusing on modelling the movements of UEs, we propose a UE density function model. Such model reflects the traffic demand at a certain position on the street graph during a certain time period. In this chapter, we build up the UE density functions based on a realistic dataset collected from a social discovery mobile App: Momo\footnote{http://www.immomo.com}. We assume that the drones are wirelessly connected to the existing BSs \cite{backhaul2012} via high frequency radios, which do not interfere with the low frequency radios used by drones and UEs.

Commercial drones often rely on batteries to power their rotors and the on-board electronic modules, such as sensors and radios \cite{gupta2016survey}. Hence, the flying time allowed by the battery is limited. In this chapter, we introduce the idea that the drones can recharge their battery from the existing powerlines\footnote{http://www.sbs.com.au/news/article/2017/08/23/powerlines-charge-drones-vic-students}. For safety, we assume that the drones can stop on the utility poles (instead of on powerlines) and recharge themselves. Then, the routine of a drone is serving UEs, flying to a utility pole, recharging on the utility pole, and then flying back to its serving position. According to the field experiments using Phantom drones \cite{azade2017}, the power for flying is over 140 watts; while the typical power for transmitting information through radios is usually around 250 milliwatts \cite{3GPP}, which is three orders of magnitude less than the former case. Compared to flying, the energy consumption caused by wireless transmission is neglected. To guarantee a certain time for serving UEs, the positions of drones should be well managed such that they are able to get recharged before running out of battery, since such positions impact the time spent on flying.

From the literature review in Section \ref{literature} we can see that, most existing approaches are proactive, which are based on some assumed UE distribution \cite{yang2017proactive, becvar2017performance, chen2017caching}. However, the real situation cannot be fully reflected by such distribution model, due to the high dynamics of realistic UEs. In contrast, the reactive approach is based on what is happening concurrently. The technique of on-line crowdsensing \cite{crowdsensing11} has been available to collect the dynamic information of UEs. Thus, the reactive approach may be more effective in real applications to to provide better service to the dynamic UEs. Before moving forward to the reactive approach, what we do in this chapter falls into the proactive group. But different from those assuming a specific distribution of UEs, we make use of a collected dataset called Momo, which can better reflect the real UE distribution.

In this chapter, we study four problems about drone deployment from simple to difficult:
\begin{itemize}
\item Single Drone Deployment (SDD): where to place a single drone in the area of interest to maximize the effectively served UE number;
\item $k$ Drones Deployment (kDD): given $k$ available drones, where to deploy them such that the effectively served UE number is maximized;
\item Energy aware $k$ Drones Deployment (EkDD): given $k$ available drones, the energy consumption and recharging models, where to deploy them such that the effectively served UE number is maximized;
\item Minimum Drones Deployment (MinDD): what is the minimum number of drones and where to deploy them such that a preferred UE coverage level can be achieved, subject to the inner drone distance constraint.
\end{itemize}

Clearly, the kDD problem is the general case of SDD by extending the number of drones to $k$ from 1, and the EkDD problem is the general case of kDD by taking into account the flying time of drones, which has not been considered in the context of using drones to serve UEs so far in the existing work. The minimum number of drones to achieve a certain user coverage level is a problem in which the Internet Service Providers (ISPs) are interested, which helps ISPs to consider the trade-off between the investment and benefit. Since we assume there a sufficient number of drones available, when drones are nearly out of battery, battery fresh drones would be deployed to replace them. Thus, the flying time constraint is not considered in MinDD.

We address SDD by finding the maximum coverage street point over the street graph, which is relatively easy. For kDD and EkDD, we prove that they can be reduced to the well known max $k$-cover problem, which is NP-hard \cite{feige1998threshold}. For MinDD, we prove that it can be reduced to the set cover problem, which is also NP-hard \cite{feige1998threshold}. Thus, we develop three greedy algorithms to solve kDD, EkDD and MinDD respectively.

To the best of our knowledge, this is the first work to study drone deployment constrained to street graph and with the consideration of battery lifetime constraint. We summarize our contributions as follows:
\begin{itemize}
\item We propose a street graph model describing the urban area of interest; and a UE density function reflecting the traffic demand at a certain position on street graph during a certain period of time.
\item We formulate a series of drone deployment problems and prove that they are NP-hard.
\item We design greedy algorithms to solve the proposed problems and provide theoretical analysis on the approximation factors.
\item Extensive simulations are conducted to verify the effectiveness of the proposed approaches.
\end{itemize}

The rest content is organized as follows. In Section \ref{model}, we present the system model, and in Section \ref{statement}, we formally formulate the problems. Section \ref{solution} presents the proposed solutions, which is followed by our theoretical analysis in Section \ref{analysis}. In Section \ref{simulation}, we conduct extensive simulations to evaluate the proposed approaches based on the realistic network dataset of Momo. Finally, Section \ref{conclusion7} concludes this chapter and discusses future work.

\section{System Model}\label{model}
We consider an urban area with UEs that cannot be served by the existing BSs due to capacity limitation or some malfunction of the infrastructure. We deploy drones to serve these UEs. The drones are assumed to be wirelessly connected to existing BSs via high frequency radios and there is a central station which manages the drones. 

In this section, we present the system model including a proposed street graph associated with the UE density function, and the wireless communication model. 

We first introduce the street graph model. Let ${\cal G} (V,E,\rho)$ be the street graph of the considered area, where $V$ and ${E}$ are the sets of street points and edges between two neighbour street points respectively. Each street point is associated with a UE density function $\rho_{v,t}$, $v\in V$, describing the number of UEs at $v$ during the $t^{th}$ ($t\in [1,T]$) time slot (the total time window we consider is $T$ time slots). Since $\rho$ is a function of time, the street graph $\cal G$ also varies with time. To simplify the statement, we omit the symbol $t$ in the below descriptions.

\begin{definition}\label{graph_distance}
The \textbf{graph distance} between two street points $v$ and $w$ on the graph, i.e., $g(v,w)$, is defined as the length of the shortest path between $v$ and $w$.  
\end{definition}

In the example shown in Fig. \ref{graph}, $g(A,B)$ is the length of the path represented by the red line segments; while $g(A,C)$ is the length of the path represented by the blue line segment. 

\begin{figure}[t]
\begin{center}
{\includegraphics[width=0.4\textwidth]{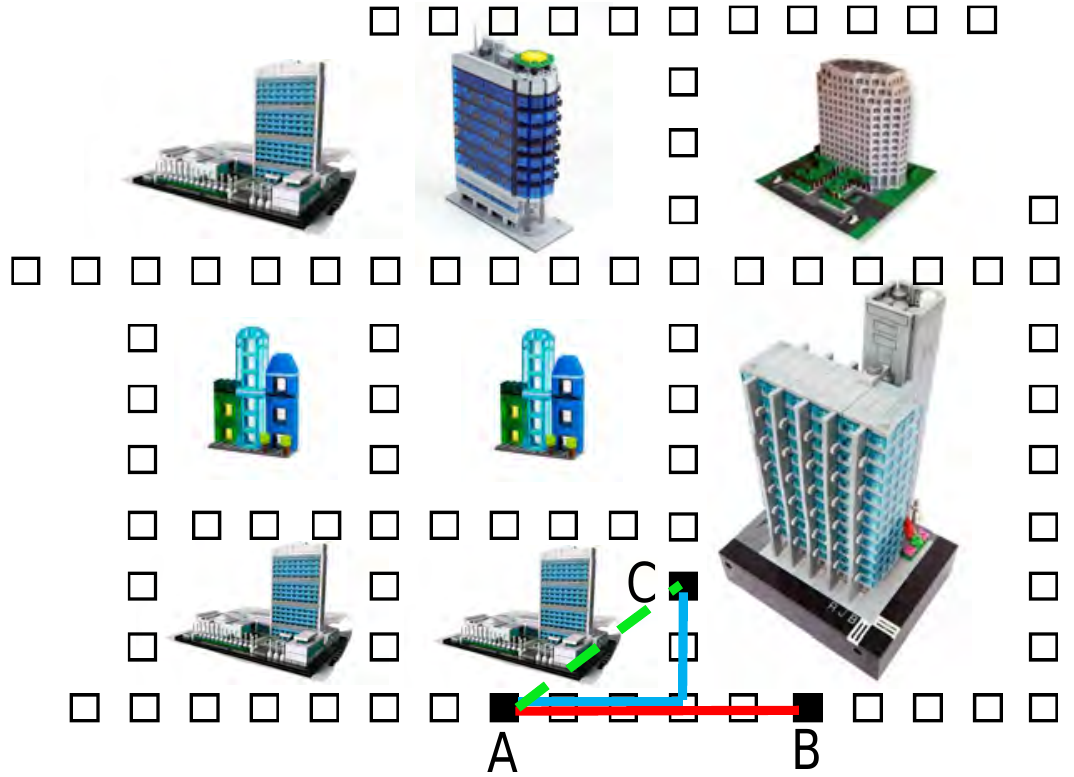}}
\caption{A street graph. The squares are street points. The length of the path represented by the red line segments is the graph distance between A and B; and that between A and C is the blue line.}\label{graph}
\end{center}
\end{figure}

\begin{remark}\label{remark1}
The physical ground distance between two street points on the street graph is always no larger than the corresponding graph distance, according to the triangle inequality theorem, which states that the sum of the lengths of two sides of a triangle must always be greater than the length of the third side.
\end{remark}

\begin{definition}\label{graph_distance}
Let $h$ be the altitude of the drone. The \textbf{spatial graph distance} between the drone at $(v,h)$ and a point $w$ on the graph, i.e., $d(v,h,w)$, is defined as:
\begin{equation}\label{physical_space_distance}
d(v,h,w)=\sqrt{g(v,w)^2+h^2}
\end{equation} 
\end{definition}

\begin{remark}\label{remark2}
According to Remark \ref{remark1}, $d(v,h,w)$ is always no smaller than the corresponding physical spatial distance.
\end{remark}

In this chapter, we do not focus on optimizing the altitude of the drones\footnote{This may involve dealing with specific regulations and policies.}. For simplicity, we fix $h$ in this chapter and then the $h$ in the spatial graph distance can be omitted without confusion. To avoid hitting tall buildings, we assume the drones can only fly and hover head over the streets.

Now, we consider the wireless communication model between drones and UEs, i.e., the drones-to-UEs links rather than the BSs-to-drones wireless backhaul links. The drones may have LoS or NLoS with UEs depending on the terrain. Consider a drone hovering over street point $v$ and a UE at street point $w$. Adopting the realistic 3GPP propagation model \cite{3GPP}, we have the path loss of signal from the drone to the UE: 
\begin{equation}\label{pathloss_simple}
\begin{aligned}
PL^{\lambda}(d(v,w))=A^{\lambda}+B^{\lambda}\log(d(v,w))\\
\end{aligned}
\end{equation}
where $\lambda \in \{LoS, NLoS\}$. 

The received power by the UE is computed by:
\begin{equation}\label{received_power}
S(v,w)=\frac{P_{tx}}{PL^{\lambda}(d(v,w))}
\end{equation}
where $P_{tx}$ is the transmit power of the drone. 

Then, the signal to noise ratio (SNR) is:
\begin{equation}\label{SNR}
SNR(v,w)=\frac{S(v,w)}{N_0}
\end{equation}
where $N_0$ is the power of noise. 

To meet the Quality of Service (QoS) requirement, we require that the SNR at any UEs should not be lower than a given threshold $\alpha$:
\begin{equation}\label{SNR_alpha}
SNR(v,w) \geq \alpha
\end{equation} 
otherwise, the received signal cannot be demodulated by the normal UEs.

From (\ref{pathloss_simple}), (\ref{received_power}) and (\ref{SNR}), we find that whether (\ref{SNR_alpha}) can be satisfied depends on $d(v,w)$ ($g(v,w)$, as $h$ is fixed in this chapter). Moreover, we find that the graph distance is upper bounded, beyond which (\ref{SNR_alpha}) will never be met. Let $g_{max}$ denote such upper bound.
\begin{remark}\label{remark3}
As will be shown later, we use the NLoS situation to compute $g_{max}$, which is the worst case. Then, according to Remark \ref{remark2}, if a drone is placed at $v$, the UEs within graph distance $g_{max}$ on the street graph always satisfy (\ref{SNR_alpha}), i.e., can be served by the drone.
\end{remark}

When multiple drones are used, the signal to interference and noise ratio (SINR) at a UE is computed by:
\begin{equation}\label{SINR}
SINR(v,w)=\frac{S(v,w)}{I+N_0}
\end{equation}
where $I$ is the total interference from all the other drones.
\section{Problem Statement}\label{statement}
In this section, we formulate four drone deployment problems.
\subsection{Single Drone Deployment (SDD)}
We start from the simplest problem: Single Drone Deployment (SDD) i.e., where to deploy a single drone in the area of interest such that the effectively served UE number by the drone is maximized.

Before formulating the problem, we introduce the below definitions.

\begin{definition}\label{covered_set}
The \textbf{covered street point set} is defined as the set of street points, which are covered by the drone. If the drone is deployed at $v$, the covered street point set can be computed by: $S(v)=\{w|g(w,v)\leq g_{max}\}$. Correspondingly, we define the \textbf{covered UE set} as the set of UEs which can be served by the drone if it is placed at $v$, i.e., $U(v)$.
\end{definition}

\begin{remark}\label{remark4}
Let $\cal U$ denote the set of all UEs on the street graph, then we have ${\cal U}=\bigcup_{v\in V} U(v)$.
\end{remark}

\begin{definition}\label{benefit}
The \textbf{benefit} achieved by placing a drone at $v$ is defined as the number of UEs it can serve, i.e., $B(v)=\sum_{w\in S(v)} \rho_w$.
\end{definition}

\begin{remark}\label{remark5}
From Definition \ref{covered_set}, \ref{benefit}, we have $|U(v)|=B(v)$, $\forall\ v\in V$.
\end{remark}

We introduce a binary variable $z(v)$ indicating whether the drone is deployed at $v$:
\begin{equation}\label{z}
z(v)=
\begin{cases}
\begin{aligned}
&1, \text{if the drone is deployed at}\ v\\
&0, \text{otherwise}
\end{aligned}
\end{cases}
\end{equation}

With Definition \ref{covered_set}, \ref{benefit} and (\ref{z}), Single Drone Deployment (SDD) can be formulated as follows:
\begin{equation}\label{optimization1}
\max_{z(v)} \left( \sum_{w\in S(v)} \rho_w \right) z(v)
\end{equation}
subject to 
\begin{equation}\label{constraint1}
S(v)=\{w|g(w,v)\leq g_{max}\}
\end{equation}

\subsection{$k$ Drones Deployment (kDD)}
Now we consider the second problem: $k$ Drones Deployment (kDD), i.e., where to deploy $k$ drones such that the effectively served UE number is maximized subject to the minimum distance between any to drones. Another constraint, i.e., the battery constraint, will be considered in Section \ref{P4}.

In the scenario with multiple drones, a UE may hear signals from more than one of them, which causes interference. In this chapter, we require that the minimum distance between any two drones is $\beta$. Note that, the system parameter $\beta$ controls the intensity of interference and a lower $\beta$ may lead to higher interference.

We formulate the $k$ Drones Deployment (kDD) problem as follows:
\begin{equation}\label{optimization2}
\max_{z(v)} \sum_{v\in V}\left( \sum_{w\in S(v)} \rho_w \right) z(v)
\end{equation}
subject to
\begin{equation}\label{constraint2_1}
S(v)=\{w|g(w,v)\leq g_{max}\}
\end{equation}
\begin{equation}\label{constraint2_2}
\sum_{v\in V} z(v)=k
\end{equation}
\begin{equation}\label{constraint2_3}
g(v,w)>\beta,\ \forall\ z(v)=1, z(w)=1, v\neq w
\end{equation}

As seen in the objective function (\ref{optimization2}), we want to maximize the benefits of drone deployment. 
In comparison to the SDD problem, here (\ref{constraint2_3}) constrains the inner distance between any pair of drones.

\subsection{Energy aware $k$ Drones Deployment (EkDD)}\label{P4}
The third problem is named as Energy aware $k$ Drones Deployment (EkDD), which aims at deploying $k$ drones such that the effectively served UE number is maximized, subject to the inner drone distance constraint and the drone battery lifetime constraint. 

The off-the-shelf drones often rely on the preloaded battery to power their rotors and the on-board electronics models. Considering the battery lifetime constraint, we require an efficient duty management scheme for the available drones to guarantee their service. 

Typically, there are three states of a drone: Serve (S), Fly (F) and Recharge (R). Denote $\lambda_S$, $\lambda_F$, and $\lambda_R$ as the percentages of a time slot during which a drone is in S, F, and R respectively. It is apparent that, $\lambda_S+\lambda_F+\lambda_R=100\%$. Among these three states, we assume that only S contributes to network performance, because when a drone is in F or R, the wireless backhaul links may be unstable or unavailable.

We assume that there are a few street points on the graph, which allow the drones to recharge. In practice, such positions can be placed on the utility poles. Let $V_R \subset V$ be the set of recharging positions.

Let $s$ be the speed of the drones. Given $\lambda_S$, $\lambda_F$, and $\lambda_R$, we can compute a maximum graph distance $g_R$ from the recharging positions, beyond which a drone cannot guarantee its serving time, i.e., $\lambda_S$ percent of the full cycle of time slots (S $\rightarrow$ F $\rightarrow$ R $\rightarrow$ F). For simplicity, we assume that to reach a recharging position $v_R$ from its current serving position $v$, a drone first flies over the street to the position whose projection is $v_R$, it then lands on the utility pole (whose altitude is $h_R$) to recharge. When the recharging process finishes, it ascends to $h$ and then flies to the serving position $v$ following streets. Consequently, we have
\begin{equation}
2(g(v,w)+h-h_R) \leq s\lambda_F\cdot 1,
\end{equation}
from which we can obtain
\begin{equation}
g_R=\frac{1}{2}s\lambda_F+h_R-h.
\end{equation}

Now, we are ready to formulate the problem with the constraint of recharging positions, which is shown below.

\begin{equation}\label{optimization4}
\max_{z(v)} \lambda_S\sum_{v\in V}\left( \sum_{w\in S(v)} \rho_w \right) z(v)
\end{equation}
subject to
\begin{equation}\label{constraint4_1}
S(v)=\{w|g(w,v)\leq g_{max}\}
\end{equation}
\begin{equation}\label{constraint4_2}
\sum_{v\in V} z(v)=k
\end{equation}
\begin{equation}\label{constraint4_3}
g(v,w)>\beta,\ \forall\ z(v)=1, z(w)=1, v\neq w
\end{equation}
\begin{equation}\label{constraint4_4}
g(v,v_R)\leq g_R,\ \forall\ z(v)=1,\ \exists\ v_R \in V_R 
\end{equation}

Here, we still aim at finding the optimal positions for $k$ drones such that the total number of served UEs can be maximized and in the meanwhile any two drones are at least $\beta$ apart from each other and any drone must be within $g_R$ from the nearest recharging position.

\subsection{Minimum Drones Deployment (MinDD)}
The above three problems all consider the scenario when drones are limited. The final problem is to seek the minimum number of drones and their deployments such that an ISP preferred level of coverage of the UEs can be served (say 98\%), which is called Minimum Drones Deployment (MinDD). Since we assume that we have a sufficient number of drones, the battery constraint can be ignored but the inner drone distance constrained is still considered.

MinDD can be formulated as follows:

\begin{equation}\label{optimization3}
\min_{z(v)} \sum_{v\in V} z(v)
\end{equation}
subject to
\begin{equation}\label{constraint3_1}
|\bigcup_{\{v|z(v)=1\}}U(v)|\geq\gamma |{\cal U}|
\end{equation}
\begin{equation}\label{constraint3_2}
g(v,w)>\beta,\ \forall\ z(v)=1, z(w)=1, v\neq w
\end{equation}
where $\gamma$ in (\ref{constraint3_1}) is the ISP preferred level of coverage, i.e., the effectively served UE ratio should not be less than $\gamma$. 


\section{Proposed solutions}\label{solution}
In this section, we present the solutions to the four considered problems.
\subsection{Solution to SDD}
\label{solution_single_case}
It is relatively easy to solve SDD using the maximum coverage street point over the street graph.

To find the solution, we can search the set of $V$. For each $v\in V$, find $S(v)$ and compute $B(v)$. Then, the $v$ having the largest $B(v)$ is the optimal place. Note, the solution to the problem may not be unique due to the discrete nature of the UE number. 

\subsection{Solution to kDD and EkDD}\label{solution_multiple_case}
First, we characterize the kDD and EkDD problems as NP-hard ones in Theorem \ref{theorem1}.

\begin{Theorem}\label{theorem1}
The problems of kDD and EkDD are NP-hard.
\end{Theorem}
\begin{proof}
If we ignore constraint (\ref{constraint2_3}) in kDD, and constraint (\ref{constraint4_3}) and (\ref{constraint4_4}) in EkDD, the two problems are reduced to an instance of max $k$-cover problem, i.e., given a universe $\cal U$ and sets $U(v)\in {\cal U}$ ($v\in V$), we are looking for a set ${\cal V}$ consisting of $k$ elements of $v$, such that the union of $U(v)$ has maximum cardinality. As shown in \cite{feige1998threshold}, the max $k$-cover problem is NP-hard. So, both of kDD and EkDD are also NP-hard.
\end{proof}

We propose greedy algorithms to solve kDD and EkDD respectively as shown in Algorithms \ref{algorithm1} and \ref{algorithm2}. The basic idea of Algorithm \ref{algorithm1} follows that in \cite{korte2012combinatorial}, i.e, in each iteration, we pick a street point which 1) adds a maximum benefit to the current benefit. Since we consider the inner drone distance constraint, we have to further check 2) whether the street point is $\beta$ away from all the already picked street points.

Clearly, the set ${\cal V}=\{v|z(v)=1\}$, i.e., the set of street points to deploy drones, is with a cardinality of $|{\cal V}|=k$. 
Let $C$ be the set of UEs which have already been served by drones. Initially, both $\cal V$ and $C$ are empty sets. In each iteration, we first compute the optimal street point $v$ in $V$. Second, we check whether (\ref{constraint2_3}) in kDD is satisfied. If yes, we add the served UEs by $v$ to the set of $C$ and add $v$ to the set of $\cal V$. Otherwise, we remove $v$ from $V$.

The intuition of Algorithm \ref{algorithm2} is similar to Algorithm \ref{algorithm1}. For EkDD,  with regards to $\beta$ constraint, we have another condition to check, i.e., 3) whether the selected street point is within $g_R$ from the nearest recharging position, see Line 7 of Algorithm \ref{algorithm2}. 

\begin{algorithm}[t]
\caption{The greedy algorithm for kDD}\label{algorithm1}
\begin{algorithmic}[1]
\Require $V$, $U$, $\beta$ and $k$
\Ensure $\cal V$
\State $C=\emptyset$.
\State ${\cal V}=\emptyset$.
\State $i=1$.
\While{$i \leq k$}
\State $\max_{v\in \{v|v\in V, v\notin {\cal V}\}} |C\bigcup U(v)|$.
\If {$g(v,w)> \beta$, $\forall\ w\in {\cal V}$}
\State $C=C\bigcup U(v)$.
\State ${\cal V}={\cal V}\bigcup \{v\}$.
\EndIf
\State $V=V\setminus v$.
\State $i=i+1$.
\EndWhile
\end{algorithmic}
\end{algorithm}

\begin{algorithm}[t]
\caption{The greedy algorithm for EkDD}\label{algorithm2}
\begin{algorithmic}[1]
\Require $V$, $U$, $\beta$, $k$, $V_R$ and $g_R$
\Ensure $\cal V$
\State $C=\emptyset$.
\State ${\cal V}=\emptyset$.
\State $i=1$.
\While{$i \leq k$}
\State $\max_{v\in \{v|v\in V, v\notin {\cal V}\}} |C\bigcup U(v)|$.
\If {$g(v,w)> \beta$, $\forall\ w\in {\cal V}$}
\If {$g(v,v_R)\leq g_R$, $\exists\ v_R \in V_R$}
\State $C=C\bigcup U(v)$.
\State ${\cal V}={\cal V}\bigcup \{v\}$.
\EndIf
\EndIf
\State $V=V\setminus v$.
\State $i=i+1$.
\EndWhile
\end{algorithmic}
\end{algorithm}

One interesting point worth mentioning is that not all of these $k$ drones can serve simultaneously if the battery constraint is considered, instead, they need to fly to the recharging positions alternately. Let $n$ be the number of drones which must recharge to make the drones work sustainablely. Then, we have that the amount of consumed energy should be no larger than the recharged energy, i.e.,
\begin{equation}\label{recharge_number}
p(k-n) \leq qn
\end{equation}
where $p$ is the energy consumption per time slot and $q$ is the recharging energy per time slot. Here, we only consider the dominant part of the power usage for resisting the gravity and omit that for other marginal parts of energy consumption such as wind. As a result, flying and hovering of drones consume the same power.

From (\ref{recharge_number}), we can obtain:
\begin{equation}\label{n_k}
n \geq \frac{q}{q+p} k
\end{equation}
Then, the drones can be divided into $\lfloor\frac{q+p}{q}\rfloor$ groups to recharge by turns. 

\subsection{Solution to MinDD}\label{minimum_number}
We characterize the MinDD problem as a NP-hard one in Theorem \ref{theorem2}.
\begin{Theorem}\label{theorem2}
The problem of MinDP is NP-hard.
\end{Theorem}
\begin{proof}
If we ignore constraint (\ref{constraint3_2}) and set $\gamma=100\%$ in MinDD, the problem is reduced to an instance of set cover problem, i.e., given a universe $\cal U$ and sets $U(v)\in {\cal U}$ ($v\in V$), we are looking for a collection ${\cal V}$ of the minimum number of sets from $U$, whose union is the entire universe $\cal U$. Formally, $\cal V$ is a set cover if $\bigcup_{v\in {\cal V}}U(v)={\cal U}$. We try to minimize $|{\cal V}|$. As shown in \cite{feige1998threshold}, the set cover problem is NP-hard, then the generalized and constrained version of set cover problem, i.e., MinDD, is also NP-hard.
\end{proof}

We design a greedy algorithm for MinDD, which is shown in Algorithm \ref{algorithm3}. The basic idea of Algorithm \ref{algorithm3} is similar to Algorithm \ref{algorithm1}, i.e., in each iteration, it picks the street point which leads to the maximum coverage of UEs and at the same time satisfying the inner drone distance constraint.  Algorithm \ref{algorithm3} will terminate when $\gamma$ percent of all the UEs in $\cal U$ have been covered by the drones. The output $\cal V$ gives the projections of drones and the number of drones is $|\cal V|$.


\begin{algorithm}[t]
\caption{The greedy algorithm for MinDD}\label{algorithm3}
\begin{algorithmic}[1]
\Require $V$, $U$, $\cal U$ and $\beta$
\Ensure ${\cal V}$
\State $C=\emptyset$.
\State ${\cal V}=\emptyset$.
\While{$|C| < \gamma |{\cal U}|$}
\State $\max_{v\in \{v|v\in V, v\notin {\cal V}\}} |C\bigcup U(v)|$.
\If {$g(v,w)> \beta$, $\forall\ w\in {\cal V}$}
\State $C=C\bigcup U(v)$.
\State ${\cal V}={\cal V}\bigcup \{v\}$.
\EndIf
\State $V=V\setminus v$.
\EndWhile
\end{algorithmic}
\end{algorithm}

\section{Analysis of the proposed algorithms} \label{analysis}
This section analyses the proposed greedy algorithms.

We first consider Algorithms \ref{algorithm1} and \ref{algorithm2}.
Let $OPT$ be the number of served UEs by an optimal solution. Let $c_i$ be the number of served UEs by the $i^{th}$ ($i\leq k$) picked street point, and $x_i$ be the number of served UEs by $i$ already picked street points, i.e., $x_i=\sum_{j=1}^i c_j$. Then, the remaining unserved UE number can be computed by $y_i=OPT-x_i$. Further, $x_0=0$, then $y_0=OPT$.

Since the optimal solution uses $k$ sets to serve $OPT$ UEs, we find that: in the $(i+1)^{th}$ step, some street point must be able to serve at least $\frac{1}{k}$ of the remaining uncovered UEs from $OPT$, i.e., 
\begin{equation}\label{ana1}
x_{i+1}\geq \frac{y_i}{k}
\end{equation}
In particular, when $i=0$, we have:
\begin{equation}\label{ana2}
x_1\geq \frac{OPT}{k},
\end{equation}
since $y_0=OPT$.

\begin{Claim}\label{claim1}
$y_{i+1}\leq (1-\frac{1}{k})^{i+1}OPT$.
\end{Claim}
\begin{proof}
We use the method of Mathematical Induction to prove the claim.

When $i=0$, from (\ref{ana2}) we have
\begin{equation}\label{ana4}
\begin{aligned}
&x_1 \geq \frac{OPT}{k}\\
\Rightarrow &OPT-x_1 \leq OPT-\frac{OPT}{k}\\
\Rightarrow &y_1 \leq (1-\frac{1}{k}) OPT
\end{aligned}
\end{equation}
i.e., the claim is true when $i$ takes 0.

Further, we assume $y_{i}\leq (1-\frac{1}{k})^{i}OPT$ and then we derive $y_{i+1}\leq (1-\frac{1}{k})^{i+1}OPT$ below:
\begin{equation}\label{ana3}
\begin{aligned}
&x_{i+1}\leq y_{i}-y_{i+1}\\
\Rightarrow &y_{i+1}\leq y_{i}-x_{i+1}\\
\Rightarrow &y_{i+1}\leq y_{i}-\frac{y_i}{k}\ (\text{using}\ (\ref{ana1}))\\
\Rightarrow &y_{i+1}\leq (1-\frac{1}{k})^{i+1}OPT.
\end{aligned}
\end{equation}

Therefore, (\ref{ana4}) and (\ref{ana3}) prove Claim \ref{claim1}.
\end{proof}

Now we are in the position to present the main results. 

\begin{Theorem}\label{theorem3}
Algorithm \ref{algorithm1} (Algorithm \ref{algorithm2}) is a $(1-\frac{1}{e})$ approximation of kDD (EkDD).
\end{Theorem}

\begin{proof}
From Claim \ref{claim1}, we have
\begin{equation}
\begin{aligned}
y_{k} &\leq (1-\frac{1}{k})^{k}OPT\\
&\leq \frac{1}{e}OPT
\end{aligned}
\end{equation}
Hence, we have
\begin{equation}
x_k=OPT-y_k\geq (1-\frac{1}{e})OPT
\end{equation}
i.e., the number of served UEs by the $k$ street points picked by Algorithm \ref{algorithm1} (Algorithm \ref{algorithm2}) is at least $(1-\frac{1}{e})$ fraction of the optimal solution $OPT$.
\end{proof}

The above results can also be used in the analysis of Algorithm \ref{algorithm3}. Let $k^*$ denote the optimal solution to MinDD, i.e., $k^*$ is the minimum number of required drones to serve at least $\gamma$ of all the UEs. Let $m=\gamma |{\cal U}|$. Clearly, if we set $k=k^*$ in kDD or EkDD, $OPT=m$. 
\begin{Theorem}\label{theorem4}
Algorithm \ref{algorithm3} is a $(\ln (m)+1)$ approximation of MinDD.
\end{Theorem}
\begin{proof}
From Claim \ref{claim1}, we have $y_i\leq (1-\frac{1}{k^*})^i m$. After picking $k^* \ln (\frac{m}{k^*})$ street points, the remaining unserved UE number is:
\begin{equation}
\begin{aligned}
y_i&\leq (1-\frac{1}{k^*})^{k^* \ln (\frac{m}{k^*})} m\\
&\leq \frac{1}{e}^{\ln (\frac{n}{k^*})} m \\
&=k^*\\
\end{aligned}
\end{equation}

The worst case of serving those remaining unserved UEs (at most $k^*$) is to use at most $k^*$ drones, i.e., each drone serves one UE. Then, we have the minimum number of drones obtained by Algorithm \ref{algorithm3} is
\begin{equation}
\begin{aligned}
k^* \ln (\frac{m}{k^*})+k^* &\leq k^* (\ln (\frac{m}{k^*})+1)\\
&\leq k^* (\ln (m)+1)\\
\end{aligned}
\end{equation}

Therefore, the minimum number of drones to serve at least $\gamma$ of all the UEs obtained by Algorithm \ref{algorithm3} is a $(\ln(m)+1)$ approximation of the optimal solution.
\end{proof}

This section analyses the approximation factors for the proposed greedy algorithms. In the next section, we will evaluation these algorithms through extensive simulations.
\section{Evaluation}\label{simulation}

In this section, we evaluate our proposed drone deployment strategies based on a collected dataset from a mobile App (Momo). We first introduce the parameters used in the simulation (Section \ref{setup}). Then, we provide details on the realistic dataset in use (Section \ref{dataset}). Next, we show the metrics for evaluation (Section \ref{metric}), followed by the compared approaches (Section \ref{compared_approach}). We present the extensive simulation results are shown (Section \ref{results}). We end this section by discussing the advantages and disadvantages of the proposed approaches (Section \ref{discussion}).
\subsection{Simulation Setup}\label{setup}
Table \ref{table_parameter} provides a quick reference for the used parameters in the simulation.

According to the QoS constraint, we can determine $g_{max}$ from (\ref{physical_space_distance}), (\ref{pathloss_simple}), (\ref{received_power}), (\ref{SNR}) and (\ref{SNR_alpha}). We vary $\alpha$ from 10 to 20 $dB$ and calculate the corresponding $g_{max}$ for both LoS and NLoS cases, as shown in Fig. \ref{G_alpha}. We can find that under the same $\alpha$, the LoS case has a larger $g_{max}$ than of NLoS case. To determine the value of $g_{max}$, we further show the probability of LoS under varying distance between transmitter and receiver in Fig. \ref{p_los}, using the model obtained by realistic experiments in \cite{3GPP}. From Fig. \ref{p_los}, we can see that when the distance between the transmitter and receiver is greater than 100 meters, the probability of LoS is less than 20\%. Thereby, we use the NLoS case to compute $g_{max}$. We select $\alpha=15$ $dB$ and the corresponding $g_{max}$ becomes 95 $m$.

\begin{table}[t]
\begin{center}
\caption{Parameter configuration}\label{table_parameter} 
  \begin{tabular}{| l | l | l |}
    \hline
    Notation & Value & Description\\\hline
    $A^{LoS}$ & 103.8 &\multirow{2}{*}{LoS path loss parameters}\\\cline{1-2}
    $B^{LoS}$ & 20.9 & \\\hline
    $A^{NLoS}$ & 145.4&\multirow{2}{*}{NLoS path loss parameters}\\\cline{1-2}
    $B^{NLoS}$ & 37.5&\\\hline
    $P_{tx}$ & 20 $dBm$ & Drone transmission power\\\hline
    $N_0$ & -104 $dBm$& Noise power\\\hline
    $h$ & 50 $m$& Drone altitude\\\hline
    $\alpha$ & 15 $dB$& SNR threshold\\\hline
    $\beta$ & 0 - $3g_{max}$ & Inner drone distance\\\hline
    $\gamma$ & 90\% - 98\% & Level of UE coverage \\\hline
    $s$ & 4 - 8$m/s$ & Flying speed of drones\\\hline
    $W$ & 100 $MHz$ & Bandwidth\\\hline
    $W_{max}$ & 2 $MHz$ & Upper bound of bandwidth at a UE\\\hline
    $h_R$ & 10 $m$ & Height of utility pole\\\hline
	$\lambda_S$ & 45\% & Percentage of one time slot for serving\\\hline
	$\lambda_F$& 5\% & Percentage of one time slot for flying\\\hline
	$\lambda_R$& 50\% & Percentage of one time slot for recharging\\\hline
    \end{tabular}
\end{center}
\end{table}

\begin{figure}[t]
    \centering
    \begin{subfigure}[b]{0.45\textwidth}
        \includegraphics[width=\textwidth]{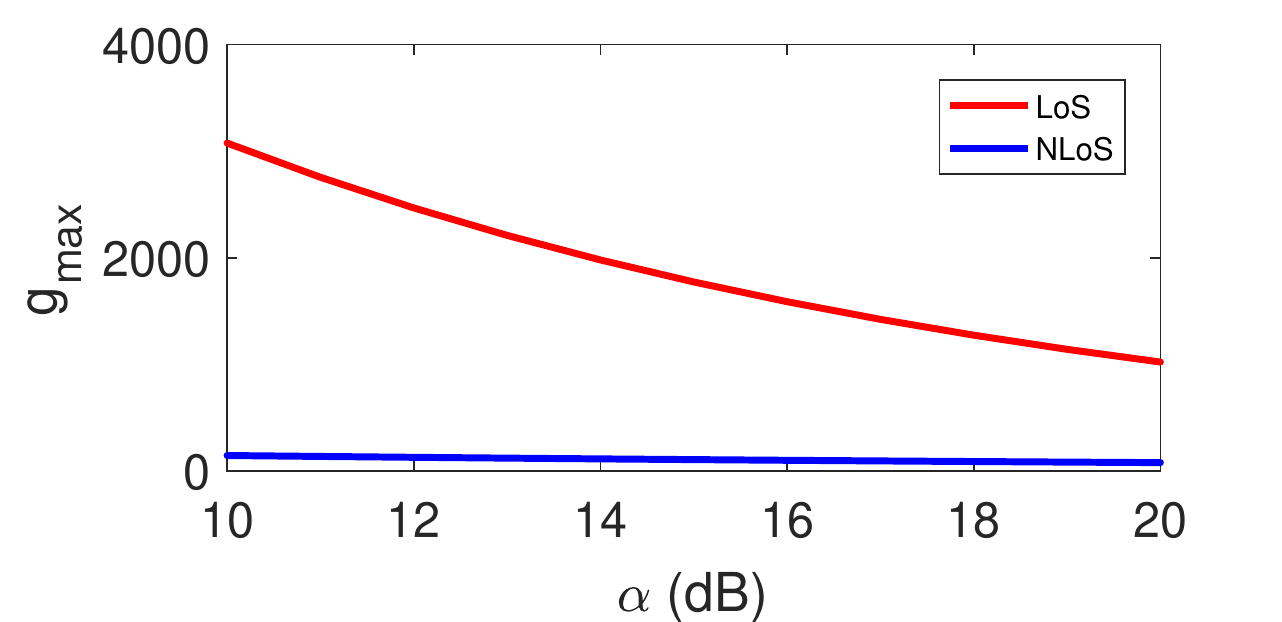}
        \caption{}
        \label{G_alpha}
    \end{subfigure}
    \begin{subfigure}[b]{0.45\textwidth}
        \includegraphics[width=\textwidth]{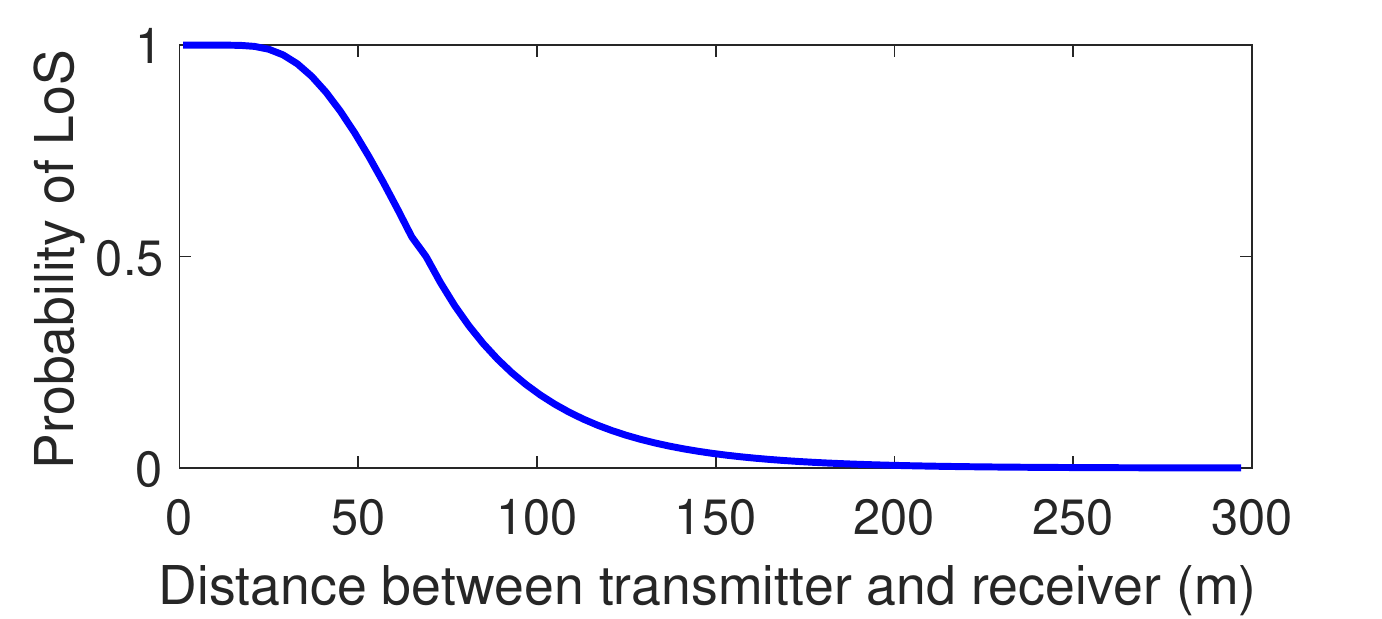}
        \caption{}
        \label{p_los}
    \end{subfigure}
    \begin{subfigure}[b]{0.45\textwidth}
        \includegraphics[width=\textwidth]{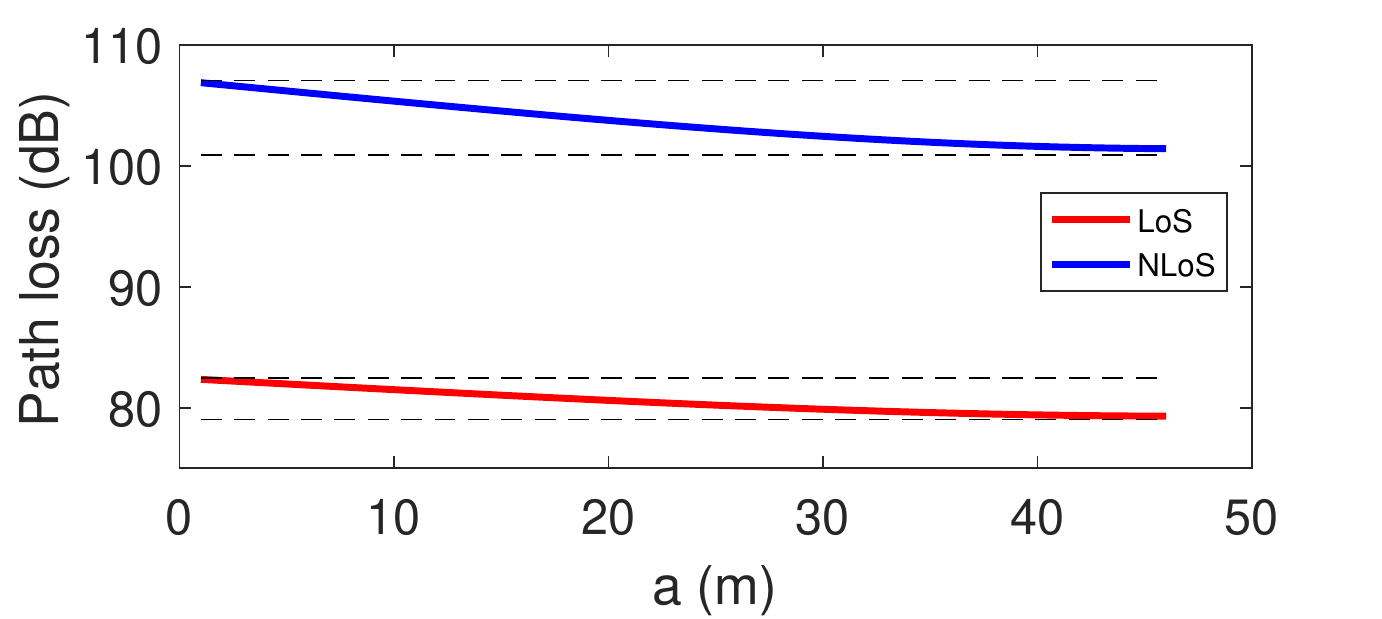}
        \caption{}
        \label{p_los_bound}
    \end{subfigure}
    \caption{(a) $g_{max}$ vs $\alpha$. (b) Probability of LoS vs distance. (c) Path loss.}
\end{figure}

Now we discuss how accurate of using graph distance to calculate path loss, instead of the physical distance. Consider two street segments $a$ and $b$ and they are perpendicular. The graph distance between the two end points is $a+b$ while the physical distance is $\sqrt{a^2+b^2}$. Let $g_{max}=a+b$ and we compute the path loss corresponding to the physical distance. Fig. \ref{p_los_bound} shows the path loss of both LoS and NLoS when $a$ is between 1 and $g_{max}/2$ meters. The physical distance varies with $a$, while the graph distance is fixed to $g_{max}$. Obviously, the largest gap occurs when $a=g_{max}/2$, and the corresponding path loss of LoS is only about 3 $dB$ lower, and that of NLoS is only about 6 $dB$ lower than the value calculated by the graph distance $g_{max}$, both of which incur small errors in practice.
In short, approximating the physical distance by the graph distance has a small impact on the calculation of path loss.

\subsection{Dataset}\label{dataset}
To get the UE density function, we make use of the dataset of Momo. When a Momo user has an update, the information of his ID, timestamp, latitude and longitude is sent to the server. The Momo dataset contains approximately 150 million such updates in a period of 38 days, from 21/5/2012 to 27/6/2012 \cite{chen2013and}. We extract a subset of this dataset, based on which we build up the UE density function. 

The Momo dataset consists of the updates of world-wide users. To make the selected dataset suitable to our problem, we only focus on the updates by the users in a small residential community in Beijing, China. The latitude of this area is from 39.9176N to 39.9242N and the longitude is from 116.4406E to 116.4501E, which is about 1059 $\times$ 721 $m^2$, as shown in Fig. \ref{map}. Further, we build up a discrete street graph according to the area map as shown in Fig. \ref{user_distribution_521}. Each street is represented by a number of discrete points. 
From the whole dataset, we select the updates whose locations fall into the considered latitude and longitude range. Based on our observation, these updates belong to indoor and outdoor UEs. We remove the indoor updates by selecting those falling into the neighbourhoods of the street points. Then, the selected dataset only consists of the updates belonging to the UEs near the streets. We demonstrate the distribution of such UEs on 21/5/2012 in Fig. \ref{user_distribution_521}.

\begin{figure}[t]
    \centering
    \begin{subfigure}[b]{0.45\textwidth}
        \includegraphics[width=\textwidth]{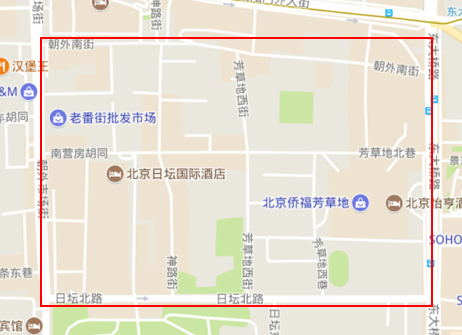}
        \caption{}
        \label{map}
    \end{subfigure}
    \begin{subfigure}[b]{0.5\textwidth}
        \includegraphics[width=\textwidth]{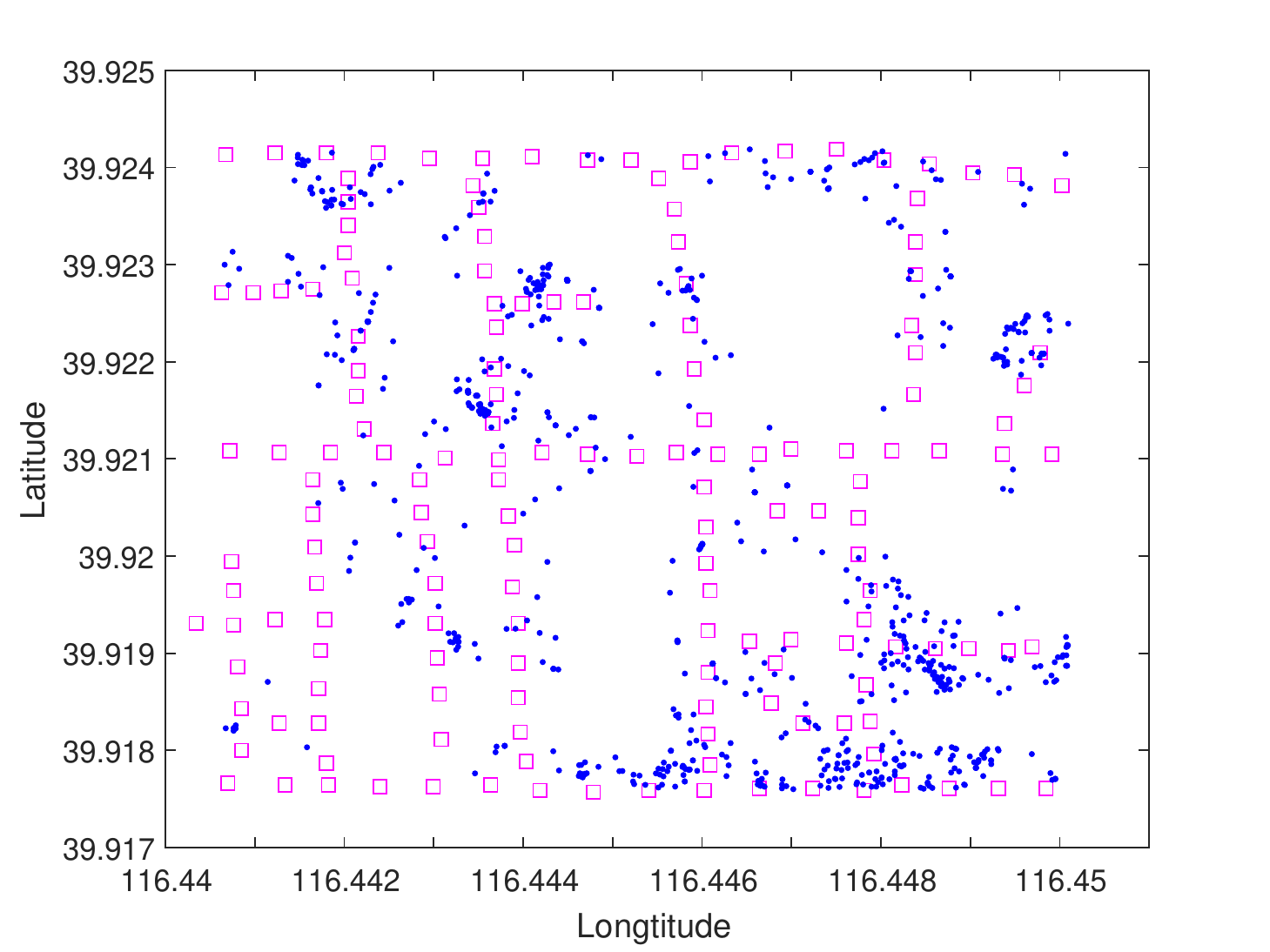}
        \caption{}
        \label{user_distribution_521}
    \end{subfigure}
    \caption{(a) The considered residential community in Beijing. (b) UE distribution on the street graph on 21/5/2012.}
\end{figure}

The selected subset of data covers more than five weeks including three types of UE pattens on weekdays, weekends, and public holidays. Since there are only three days belonging to category of public holidays, in this chapter we only consider the UE pattens on weekdays and weekends. We respectively take the average numbers of UEs for weekdays and weekends and show them in Fig. \ref{UEnumber}, where the duration of a time slot is one hour. 
It can be seen that on weekdays the average UE number is mostly larger than weekends. The average total UE number on weekdays is 1003 per day while that on weekends is 727. On weekdays, the UE number increases steadily from 8:00 and arrives at the peak (about 70 UEs) at 16:00, after which it decreases. The UE variation for weekends is different: it increases slowly from 8:00 to 16:00, and from 16:00 to 23:00 the UE number remains at around 40. After 23:00, much later than weekdays, it drops down. If we set 40 as a UE number threshold to decide the usage of drones: on weekdays we need to send drones to serve UEs from 10:00 to 22:59, i.e., 13 hours; while on weekends the drones should work for 8 hours between 15:00 and 22:59. We need to mention that the number of UEs provided by Fig. \ref{UEnumber} shows only the UEs using Momo. The actual number of UEs should also include those not using Momo. Although the dataset we have cannot provide us the true total number of UEs, it presents the realistic traffic pattens and UE distributions, which is much more closer to reality than random distribution. To make the results more sensible, we introduce a scalar to scale up the UE number according to \cite{scale15}, which is set as 5 in the simulations.

\begin{figure}[t]
\begin{center}
{\includegraphics[width=0.53\textwidth]{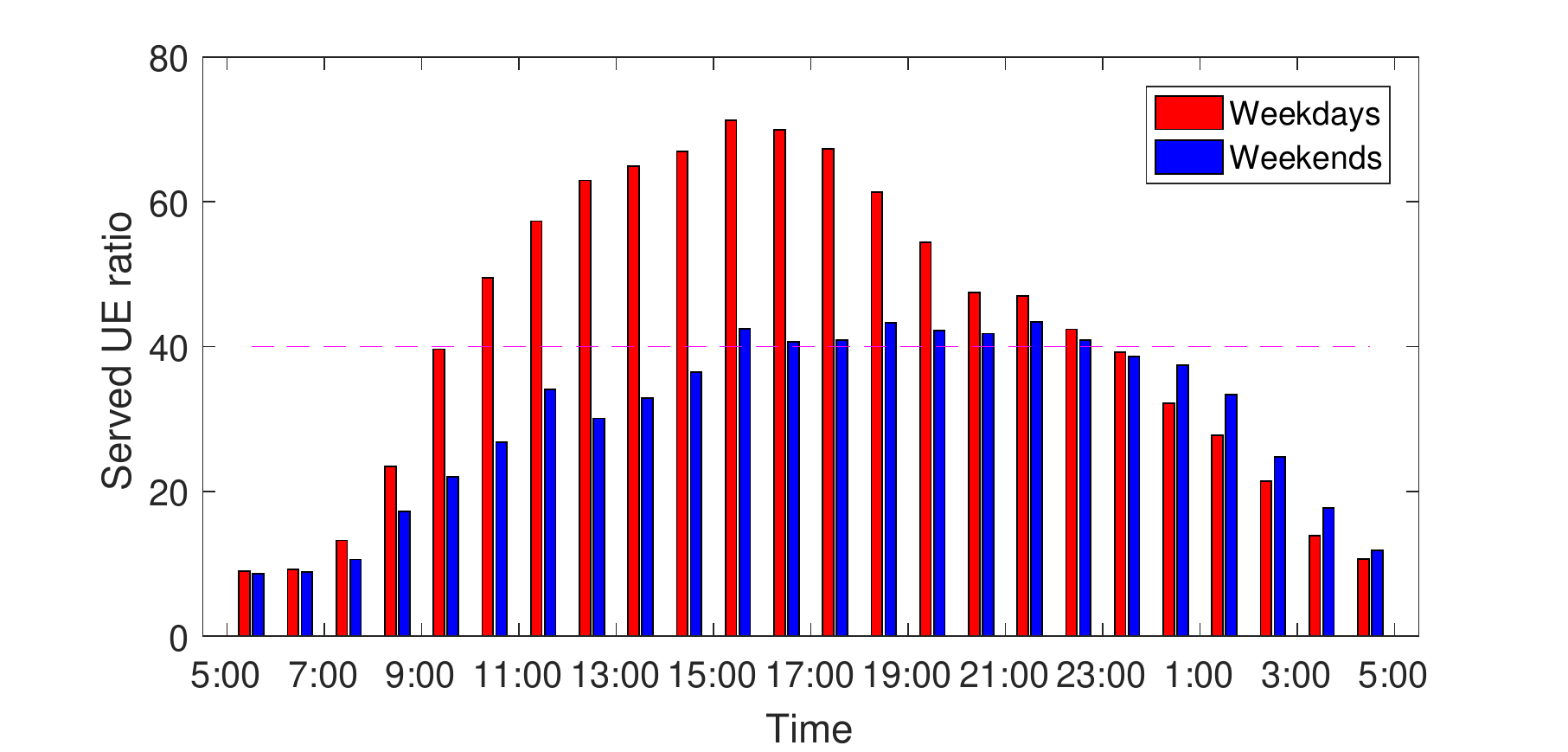}}
\caption{Average UE numbers using Momo on weekdays and weekends in the considered area.}\label{UEnumber}
\end{center}
\end{figure}

\subsection{Metrics}\label{metric}
In the evaluations, we consider the performance metrics:
\begin{itemize} 
\item Served UE ratio: The served UE ratio is defined as $|\bigcup_{v\in {\cal V}}U(v)|/|{\cal U}|\times 100\%$. 
\item Spectral efficiency ($SE$): a measure of how efficiently a limited frequency spectrum is utilized by PHY layer protocol. The spectral efficiency at a UE is computed by $SE=log_2(1+SINR)$, where SINR is defined in (\ref{SINR}). In this chapter, we consider the average spectral efficiency:
\begin{equation}
ASE=\frac{1}{M} \sum_{i=1}^M log_2(1+SINR_i)
\end{equation} 
\item Number of drones: the number of drones required to serve at least $\gamma$ percent of UEs. 
\item Served UE per drone: the average number of UEs served by one drone. 
\item Network capacity ($NC$): the number of bits that can be generated in an unit area of 1 $km^2$, given the bandwidth $W$. The network capacity is a metric measuring the overall performance in both PHY and MAC layers, whose unit is $Mbps/km^2$. For simplicity, we assume that the bandwidth of a drone is evenly allocated to its served UEs and there is a limit on the bandwidth that each UE can get, i.e., $W_{max}$. As a result, $W_{i}=\min(\frac{W}{M_j}, W_{max})$, where $M_j$ is the number of UEs served by drone $j$. Suppose there are $k$ drones in 1 $km^2$, then the network capacity can be computed by:
\begin{equation}
NC=\sum_{j=1}^k \sum_{i=1}^{M_j} SE_i\cdot W_i
\end{equation}
\end{itemize}

\subsection{Comparing Approaches}\label{compared_approach}
As mentioned in Sections \ref{motivation}, there are several related work on the topic of drone deployment problem. For the scenario of placing a single drone, the most relevant work is  \cite{alzenad20173d}. Aiming at serving the maximum UEs, the authors formulate a mixed integer non-linear problem to find the optimal position of the drone. Further, the authors formulate a second order cone problem to shorten the cover range of the drone to save transmission energy. 
For a fair comparison, we only compare our work with the results of the mixed integer non-linear problem (MINLP), since shortening a bit transmission range does not contribute much to the total energy consumption by the drone. We acknowledge that the energy consumed for drone movement is the dominant factor.

There are also some related work about placing multiple drones, such as \cite{sharma2016uav} and \cite{kalantari2016number}. However, both of them divide the area of interest into a set of zones or subareas, which are quite different from our basic model, i.e., the street graph. So we demonstrate the comparison of our approach with \cite{alzenad20173d} for the case of single drone deployment; and the comparison with max $k$-cover (without inner drone distance constraint) for the case of multiple drone deployment. 

\subsection{Simulation results}\label{results}
We present simulation results to evaluate the performance of the proposed solutions respectively in this part.
\subsubsection{Evaluation of SDD}

We first show the performance of single drone deployment method. As discussed in Section \ref{dataset}, on weekdays we deploy a drone from 10:00 to 22:59 for 13 hours; while on weekends the drone work for 8 hours between 15:00 and 22:59. Except \cite{alzenad20173d} (named as Approach 1), we also compare with a random deployment (named as Random). 

We respectively display the served UE ratios on weekdays and weekends in Fig. \ref{average_UE_weekdays_5} and \ref{average_UE_weekends}. We can see that our proposed method achieves similar performance compared with \cite{alzenad20173d} in terms of served UE ratio. In theory, the 2D projection of a drone is constrained on street in our approach; while there is no such limitation in \cite{alzenad20173d}. So the optimum solution by \cite{alzenad20173d} achieves no worse performance in served UE ratio than ours. However, \cite{alzenad20173d} uses MOSEK solver to address the MINLP and some solutions are local optimums, which are not competitive to our solutions. Further, both of them outperform the random deployment. 

It is worthy of mentioning that, as there is no limitation on the 2D projection of the drone, over 70\% of the solutions by \cite{alzenad20173d} locate off the streets (see the illustrative example shown in Fig. \ref{deployment_example}, although the deployment by \cite{alzenad20173d} serves more UEs, the projection is off street), which have high probability of hitting tall buildings if applied in urban environment directly. In contract, our solutions are on streets and can be applied directly to realistic networks, thanks to the street graph model.

\begin{figure}[t]
    \centering
    \begin{subfigure}[b]{0.47\textwidth}
        \includegraphics[width=\textwidth]{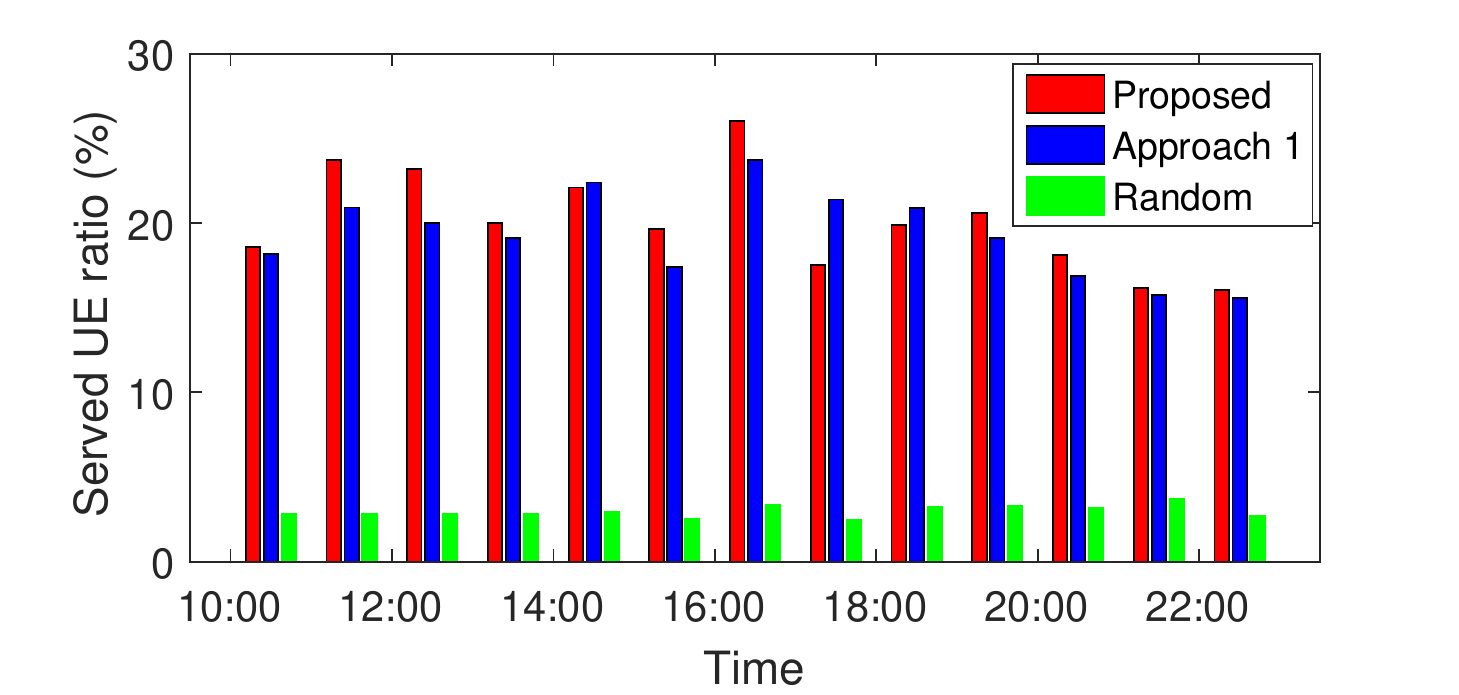}
        \caption{}
        \label{average_UE_weekdays_5}
    \end{subfigure}
    \begin{subfigure}[b]{0.47\textwidth}
        \includegraphics[width=\textwidth]{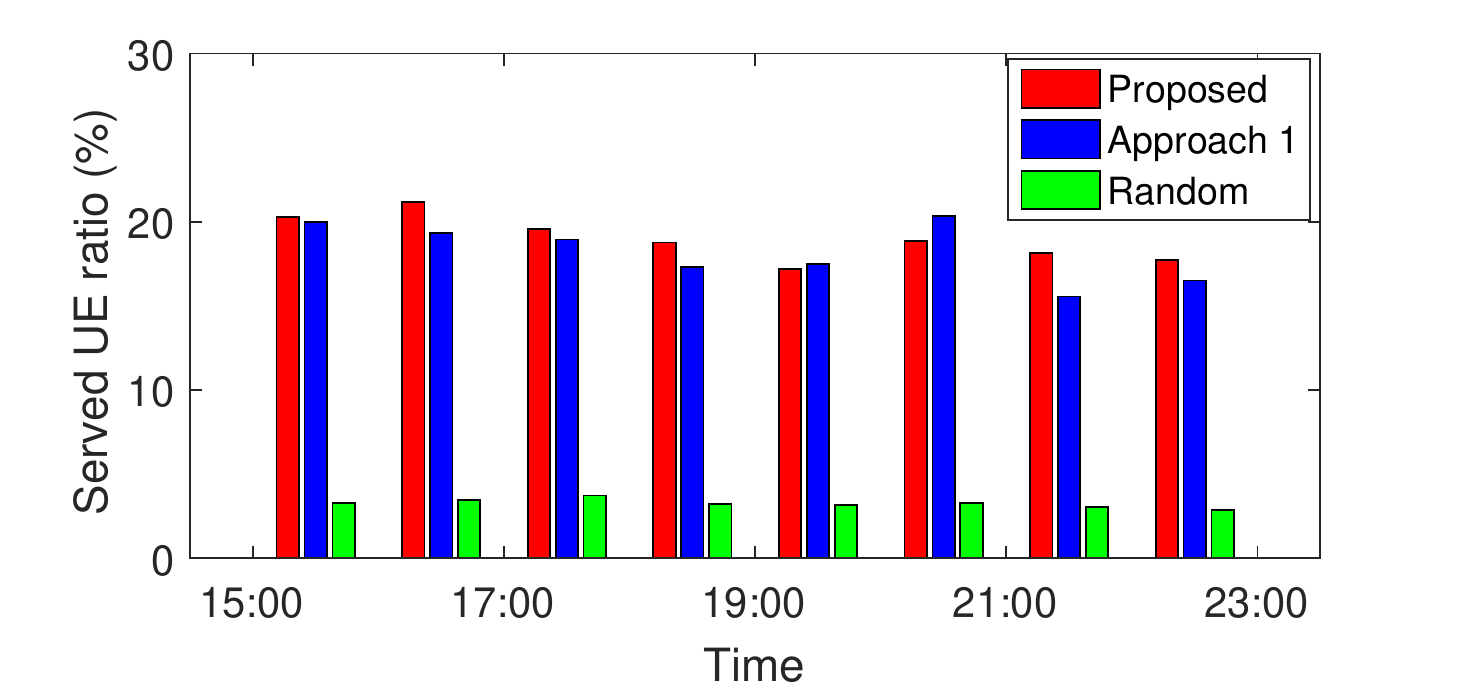}
        \caption{}
        \label{average_UE_weekends}
    \end{subfigure}
    \caption{(a) Average served UE ratio on weekdays. (b) Average served UE ratio on weekends.}
\end{figure}

\begin{figure}[t]
\begin{center}
{\includegraphics[width=0.45\textwidth]{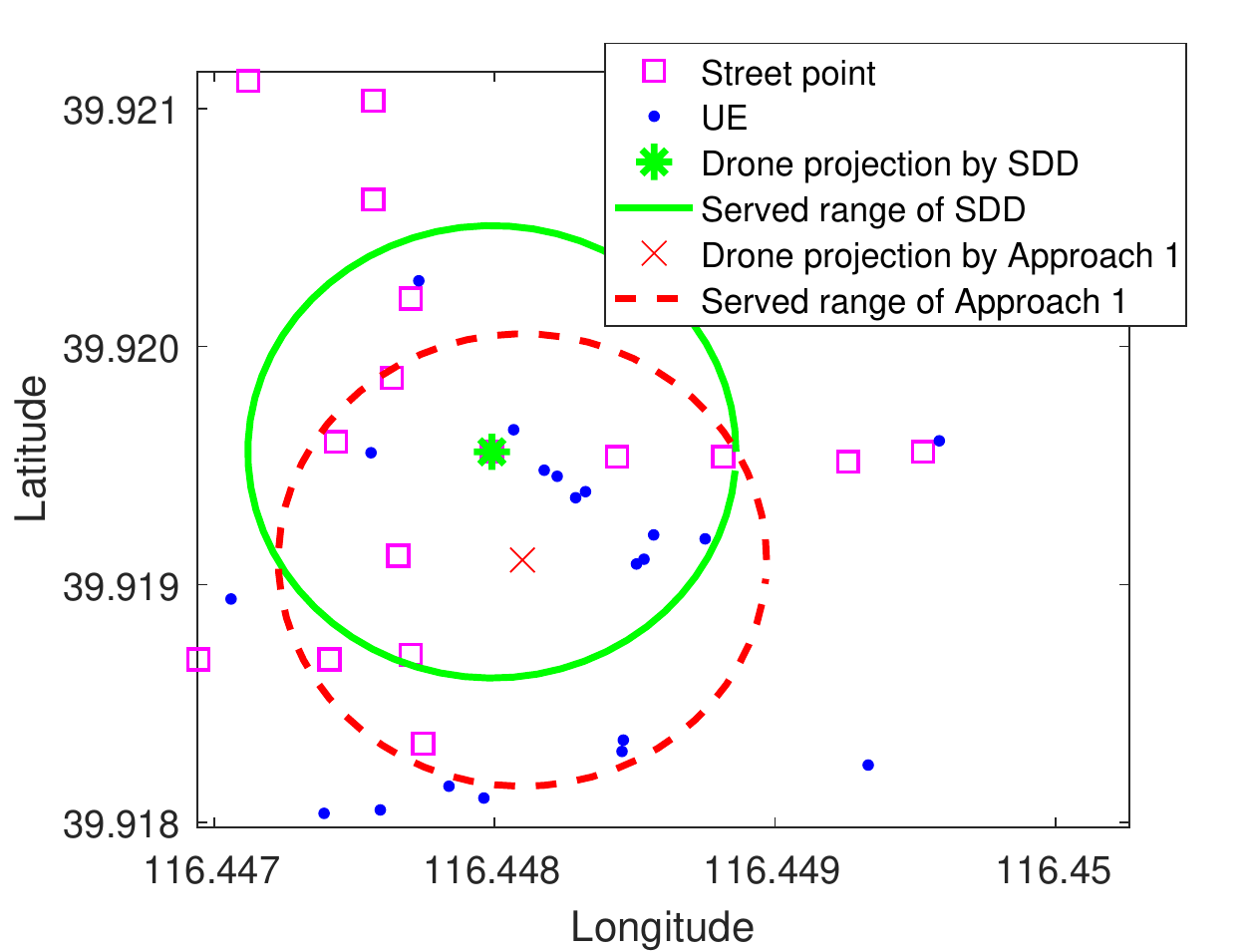}}
\caption{Illustrative example of 2D projections by the proposed approach and Approach 1.}\label{deployment_example}
\end{center}
\end{figure}

\subsubsection{Evaluation of kDD}

\begin{figure}[t]
    \centering
    \begin{subfigure}[b]{0.47\textwidth}
        \includegraphics[width=\textwidth]{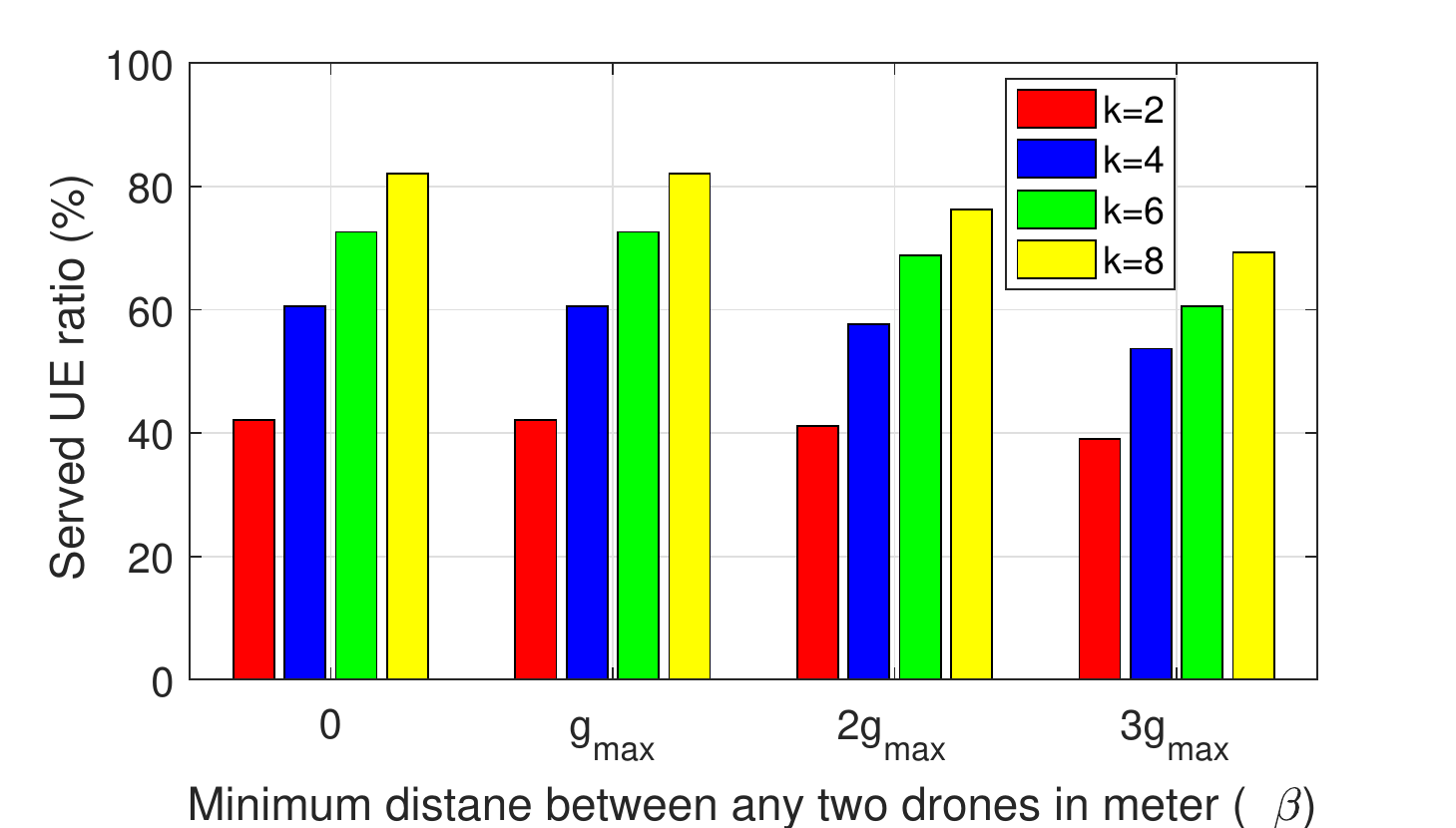}
        \caption{}
        \label{UE_ratio_beta_k_17}
    \end{subfigure}
    \begin{subfigure}[b]{0.47\textwidth}
        \includegraphics[width=\textwidth]{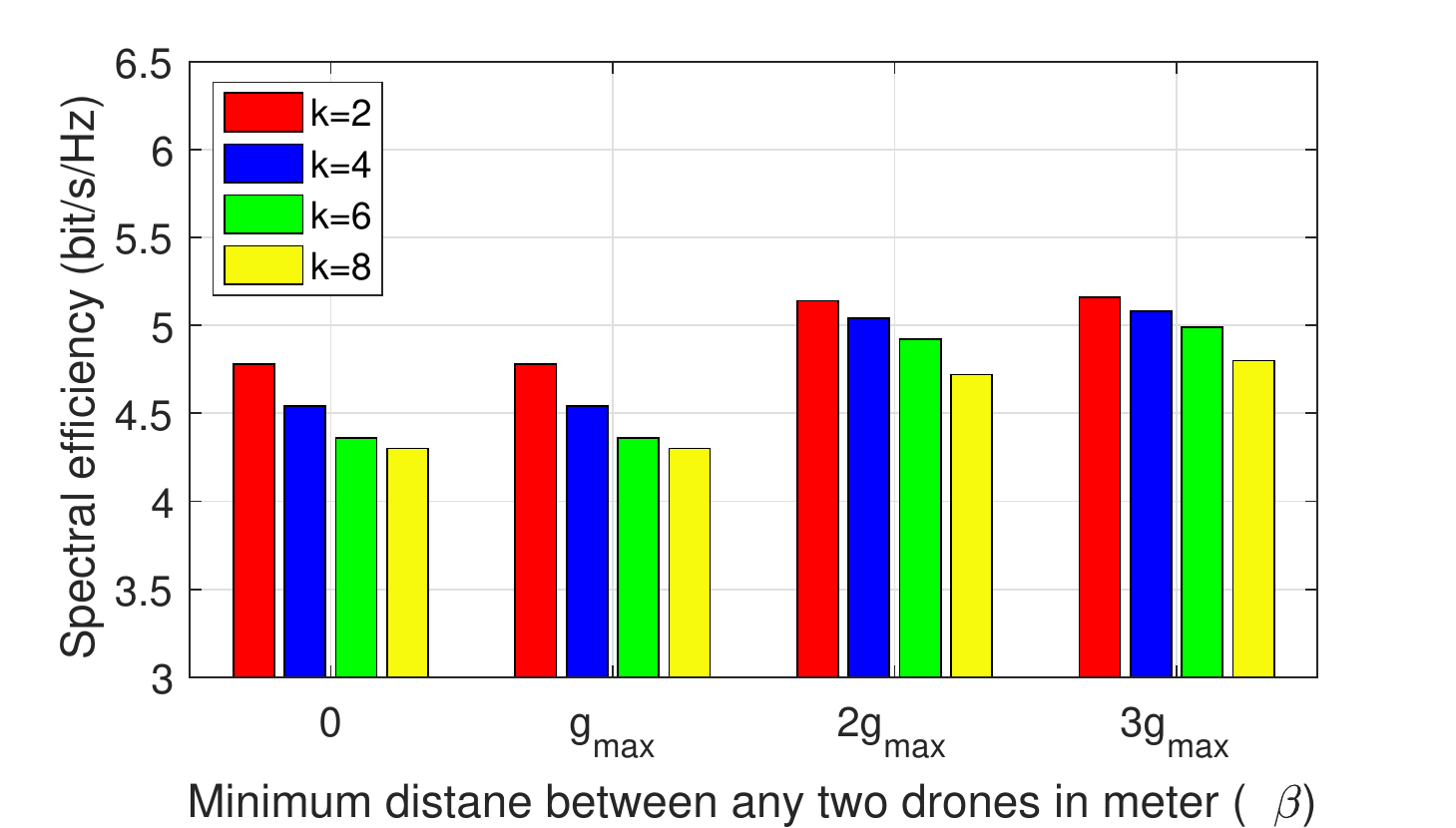}
        \caption{}
        \label{rate_beta_k_17}
    \end{subfigure}
    \caption{(a) Average served UE ratio on weekdays by multiple drones. The average number of UEs during the peak hour on weekdays is 350. (b) Average spectral efficiency on weekdays by multiple drones.}
\end{figure}

We demonstrate simulation results for deploying multiple drones. When multiple drones are used, as mentioned before, the interference should be taken into account to assess the user experience. Two main factors influence the interference intensity: the number of drones $k$, and the distance threshold between two drones $\beta$. Here, we vary $k$ from 2 to 8 and $\beta$ from $0$ to $3g_{max}$. Note that, when $\beta$ is 0, we are solving the max $k$-cover problem, i.e., without the consideration of inner drone distance constraint, because constraint (\ref{constraint4_3}) is always satisfied.

To have an insight of the impacts of these two factors, we focus on the performance on the peak hour, i.e., from 15:00 to 15:59, on weekdays. We present the served UE ratio and spectral efficiency in Fig. \ref{UE_ratio_beta_k_17} and \ref{rate_beta_k_17} for the peak hour. Note, the results shown here are the average values of 27 weekdays. From these results we can see that: 
\begin{itemize}
\item With the increase of $k$, the served UE ratio increases; while the average spectral efficiency decreases. Because more UEs are covered by drones when $k$ is increased, and in the same time, the interference is also  raised.
\item With the increase of $\beta$, the served UE ratio slightly decreases; while the average spectral efficiency increases. This is because increasing $\beta$ means increasing the minimum distance between any two drones, which reduces the coverage performance; but in the meanwhile, it also reduces the strength of interference at a UE from other drones, which leads to larger spectral efficiency.
\item An interesting finding is that the results when $\beta$ takes $g_{max}$ are very similar to those when $\beta=0$ in terms of both served UE ratio and average spectral efficiency. In other words, the standard greedy algorithm for max $k$-cover problem inherently avoids placing two drones closely to each other.
\end{itemize}

\begin{figure}[t]
    \centering
    \begin{subfigure}[b]{0.47\textwidth}
        \includegraphics[width=\textwidth]{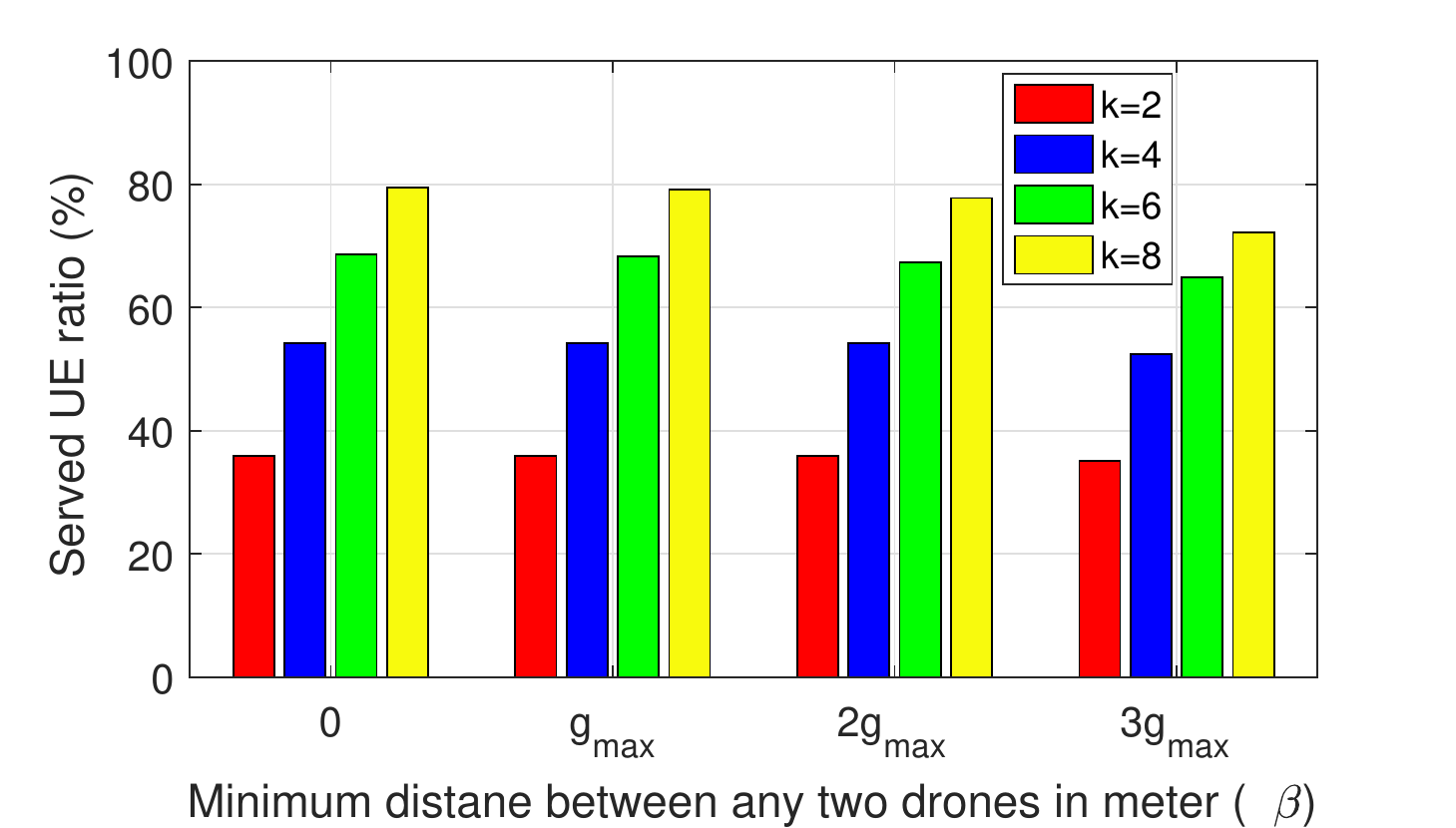}
        \caption{}
        \label{UE_ratio_beta_k_17_weekends}
    \end{subfigure}
    \begin{subfigure}[b]{0.47\textwidth}
        \includegraphics[width=\textwidth]{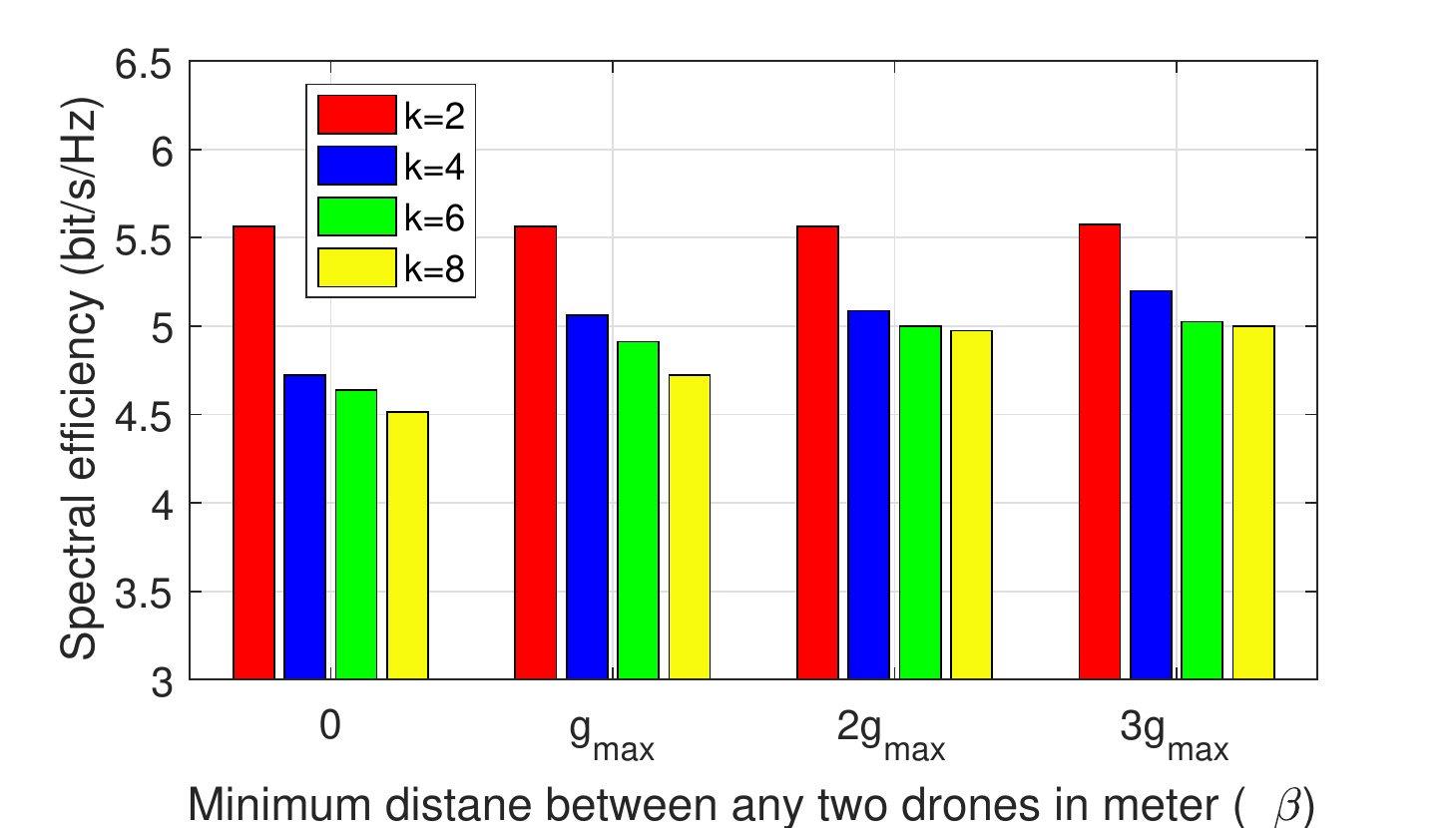}
        \caption{}
        \label{rate_beta_k_17_weekends}
    \end{subfigure}
    \caption{(a) Average served UE ratio on weekends by multiple drones. The average number of UEs during the peak hour on weekends is 210. (b) Average spectral efficiency on weekends by multiple drones.}
\end{figure}

The results for the same peak hour on the weekends are shown in Fig. \ref{UE_ratio_beta_k_17_weekends} and \ref{rate_beta_k_17_weekends}. The performance trends are similar to those of weekdays. Comparing the results of weekdays and weekends, we can find that under the same $\beta$ and $k$, the served UE ratio on weekdays is similar to weekends (see Fig. \ref{UE_ratio_beta_k_17} and \ref{UE_ratio_beta_k_17_weekends}). But the number of served UEs on weekdays is much larger than weekends due to the dense distribution of UEs on weekdays. The spectral efficiency on weekdays is about 10\% lower than that on weekends in average (see Fig. \ref{rate_beta_k_17} and \ref{rate_beta_k_17_weekends}). The reason behind this phenomenon is that the distribution of UEs on weekends is more sparse than weekdays. Under the same $\beta$ and $k$, the deployments of drones is also more sparse, which leads to lower interference, then higher spectral efficiency.

\subsubsection{Evaluation of EkDD}
In this part, we demonstrate the performance of the solution to the problem of EkDD. 

We assume that the operation power and the energy recharging rate are the same, i.e., $p=q$. From (\ref{n_k}) we have $n\geq \frac{1}{2}k$, i.e., the drones should be  divided into 2 groups. Consider the time slot of 1 hour in this chapter and the state-of-the-art commercial drones, such as DJI\footnote{https://www.dji.com}, whose the flying time are around 30 minutes. So in this chapter, we set $\lambda_S$, $\lambda_F$, and $\lambda_R$ to 45\%, 5\% and 50\% respectively. We consider that there are four utility poles located at the corners of graph, see Fig. \ref{snapshot} for an example. Here, we fix $\beta$ as $g_{max}$.

As discussed in Section \ref{P4}, the flying speed is the main factor impacting on the drone deployment. Here, we consider various practical flying speeds $4,5,6,7,8 m/s$, and present the corresponding performance in served UE ratio for the considered peak hour in Fig. \ref{P4_served_UE_ratio}. The results of kDP are also displayed for comparison, indicated by 'Inf'. For a fair comparison, in EkDD, $k=8$; while in kDD, $k=4$. Then, in both the cases, 4 drones can serve UEs simultaneously. From Fig. \ref{P4_served_UE_ratio} we can see that with the increase of $s$, the served UE ratio increases first and then remain at a steady level. Because a larger $s$ means wider operation radius for drones, and thus a weaker constraint on the positions of drones. When $s$ is larger than $6m/s$, the recharging constraint becomes invalid. Fig. \ref{snapshot} shows the drone projections for four cases of $s=4m/s$, $s=5m/s$, $s=6m/s$ and 'Inf', from which we can find that the constraint (\ref{constraint4_4}) pulls the drones' positions closer to the recharging positions compared to the non-constraint case of 'Inf'. When $s$ takes 7 or $8m/s$, the positions of drones are the same with the cases of $6m/s$ and 'Inf'. So we do not display them in Fig. \ref{snapshot}.

\begin{figure}[t]
    \centering
    \begin{subfigure}[b]{0.47\textwidth}
        \includegraphics[width=\textwidth]{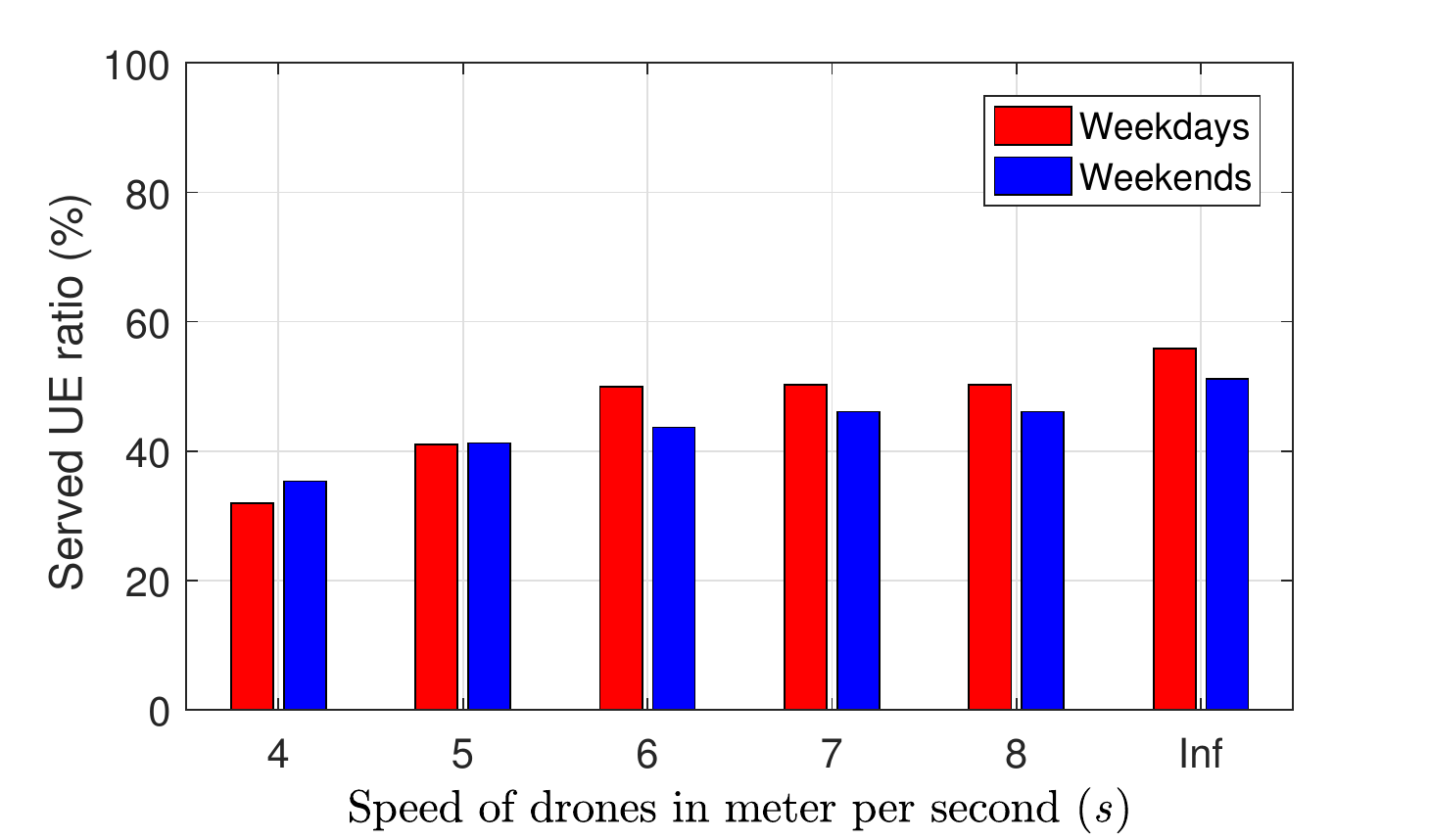}
        \caption{}
        \label{P4_served_UE_ratio}
    \end{subfigure}
    \begin{subfigure}[b]{0.47\textwidth}
        \includegraphics[width=\textwidth]{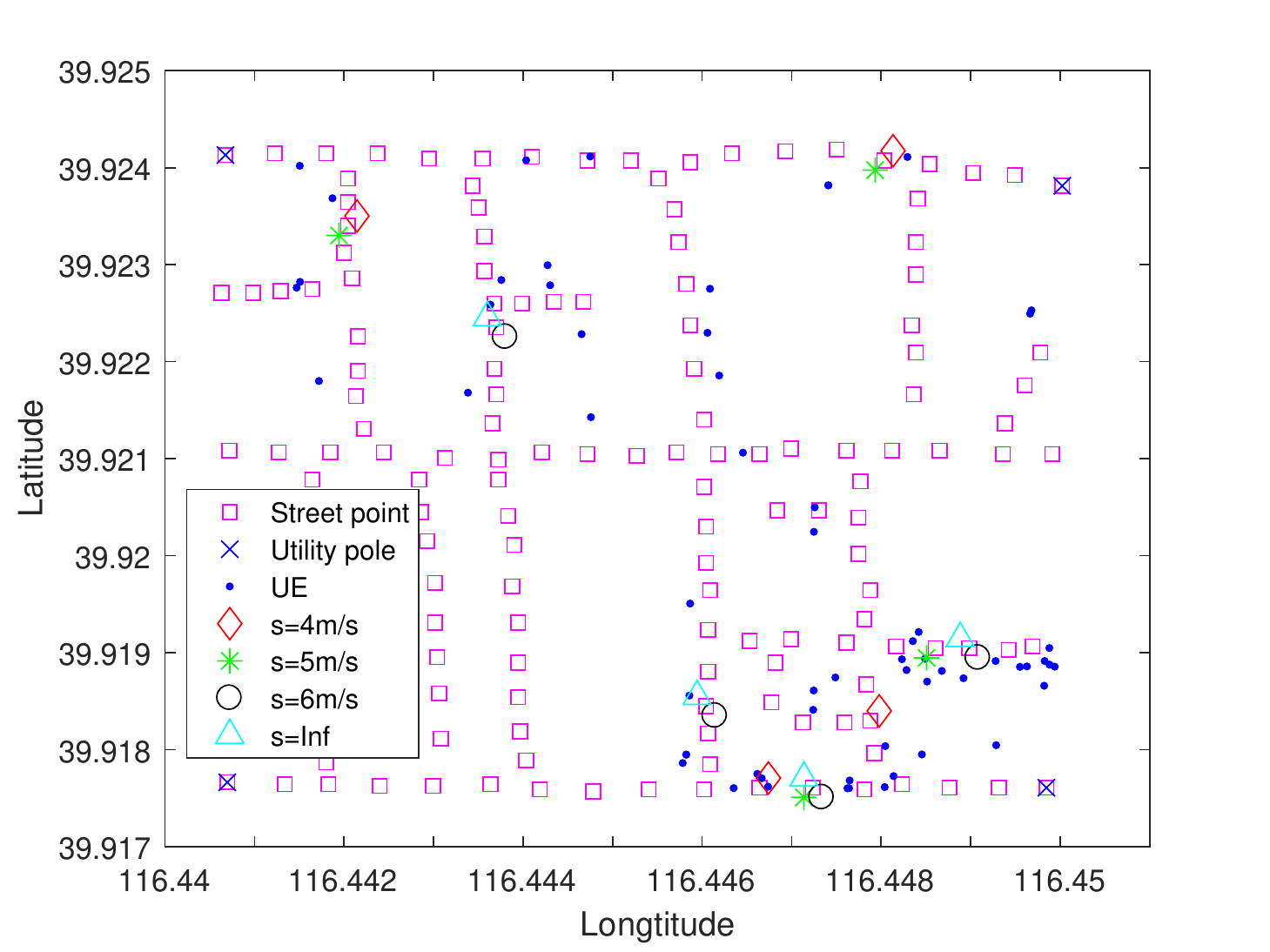}
        \caption{}
        \label{snapshot}
    \end{subfigure}
    \caption{(a) Served UE ratio against flying speed. The total numbers of UEs are 350 and 210 for the peak hour on weekdays and weekends respectively. (b) Drone projections in the cases of $4m/s$, $5m/s$, $6m/s$, and 'Inf' for the peak hour on a weekday.}
\end{figure}

In short, taking into account the energy issue, the performance reflects a more practical scenario. The drones with higher flying speed can achieve better performance due to the relaxed constraint on its position to recharge its battery.

\subsubsection{Evaluation of MinDD}
In this part, we show how many drones are needed to serve at least $\gamma$ percent of UEs. Here, we fix $\beta$ as $g_{max}$ and show the relationship between $\gamma$ and the minimum number of required drones, number of UEs served per drone, and network capacity. We vary $\gamma$ from 90\% to 98\%. The results are shown in Fig. \ref{evaluation4}.

\begin{figure}[t]
    \centering
    \begin{subfigure}[b]{0.47\textwidth}
        \includegraphics[width=\textwidth]{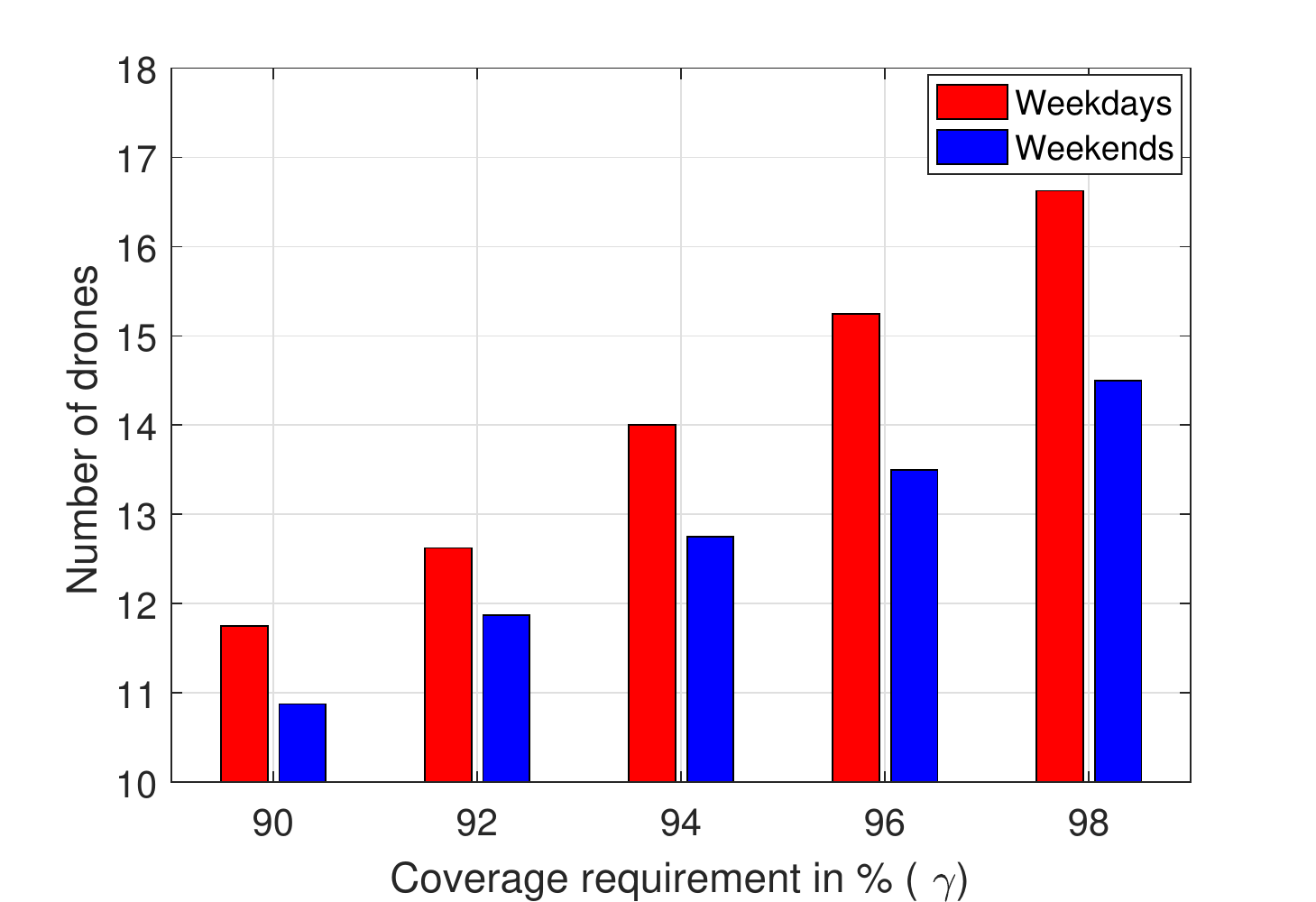}
        \caption{}
        \label{droneNo_gamma_17}
    \end{subfigure}
    \begin{subfigure}[b]{0.47\textwidth}
        \includegraphics[width=\textwidth]{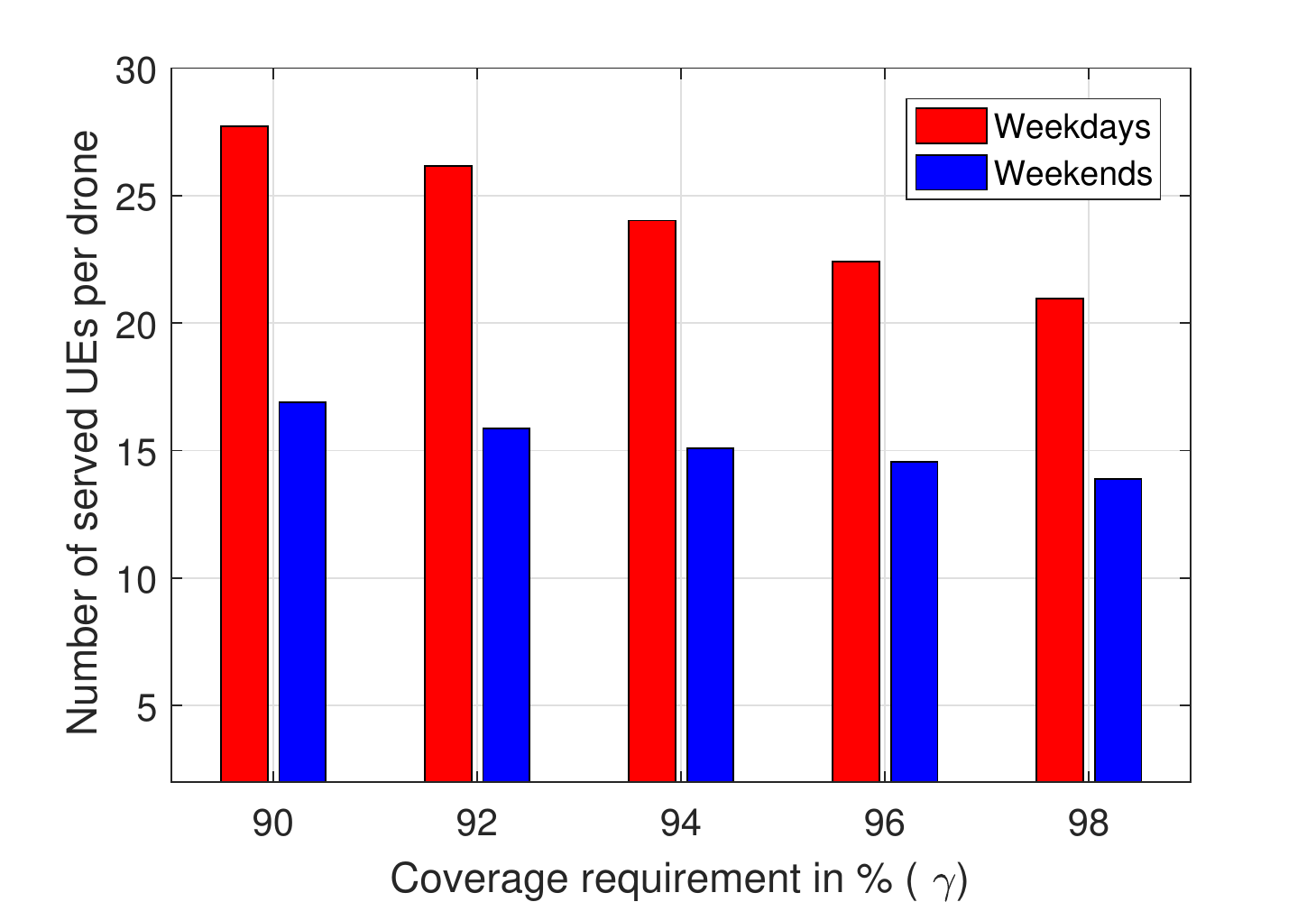}
        \caption{}
        \label{UE_per_drone_gamma_17}
    \end{subfigure}
    \begin{subfigure}[b]{0.47\textwidth}
        \includegraphics[width=\textwidth]{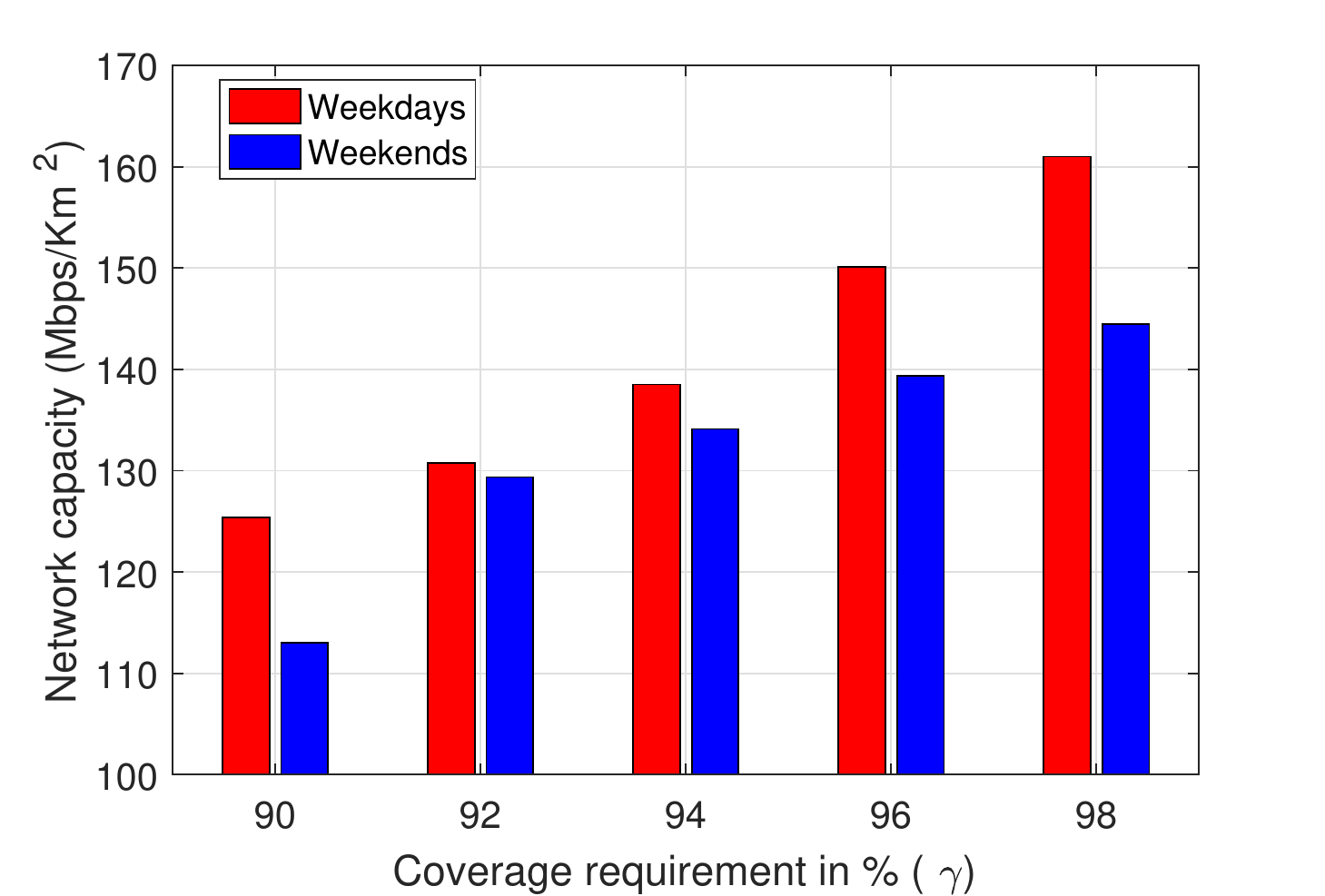}
        \caption{}
        \label{network_capacity_17}
    \end{subfigure}
    \caption{(a) Average minimum number of drones to serve $\gamma$ percent of UEs. (b) Average number of served UEs per drone against $\gamma$. (c) Network capacity against $\gamma$.}\label{evaluation4}
\end{figure}

\begin{itemize}
\item Fig. \ref{droneNo_gamma_17} displays the required drone numbers with various $\gamma$ on weekdays and weekends for the considered peak hour. We can see that the minimum number of drones increases with $\gamma$. 
Since the number of UEs on weekends are much smaller than weekdays, fewer drones are needed on weekends than on weekdays to achieve the same level of UE coverage. 
\item Fig. \ref{UE_per_drone_gamma_17} shows that, with the increase of $\gamma$, the average number of UEs served per drone decreases. In other words, the increase of served UE number is slow than the increase of drone number. We can also see that on weekdays the number of UEs served per drone is much higher than weekends, because UEs are sparsely distributed on weekends. 

\item Fig. \ref{network_capacity_17} presents the network capacity against $\gamma$. We can see that the network capacities on weekdays and weekends both increase with $\gamma$. Since the total number of UEs on weekends is much smaller than that on weekdays, hence the resource at drones is not fully used. So the network capacity on weekdays is larger than that on weekends.
\end{itemize}

\subsection{Discussion}\label{discussion}
\subsubsection{Advantages} 
From the above simulation results we can find that to serve more UEs, more drones are required and the network capacity can be improved. However, the efficiency of drone usage, i.e., the average number of served UEs per drone, decreases. Therefore, from the view of ISPs, there is a trade-off between the investment, i.e., the number of drones, and performance the gain, i.e., the network capacity.  

\subsubsection{Limitations} 
By use of practical models and reasonable network assumptions, the proposed solutions to the drone deployment problems provide useful guideline for practical drone deployment. In particular, all the drones are deployed over head streets, which can avoid collision with tall buildings. Thus, it is straightforward to apply our approaches to realistic applications. Moreover, we consider the energy constraint of drones in the optimization problems, which in our best knowledge has not been addressed by the existing work.

\section{Summary}\label{conclusion7}
In this chapter, we adopt a UE density model based on the collected dataset Momo. However, the modern wireless traffic demand is quite dynamic. Thus, more effective models to predict the time-variant UE density are for further study. Besides, we have not touched the performance impact of drone altitude, which is another significant factor in practical drone deployment. Although it is a regulation issue, finding the optimal altitude within the allowed range may improve the number of served UEs. Further, the objective function of EkDD is a rough estimation of the served UE number, because some drones can fly back to the serving positions using less time than $\lambda_F$ of a time slot. The precise formulation requires to consider the serving positions in the objective function, which makes the problem more complex. Moreover, in terms of radio resource management (RRM), we simply allocate the available bandwidth evenly to UEs. Optimal allocation of radio resource is worth investigating to achieve a larger network capacity \cite{abdelnasser2016resource, jafari2015study}. 

From simple to complex, in this chapter we discussed how to deploy a given number of drones to maximize the effectively served UE number; and how to determine the minimum number of drones to achieve a certain level of UE coverage. Different from existing work, the drone deployment is formulated based on a street graph model in this chapter. The street graph, associated with UE density function (obtained from a realistic dataset Momo), is close to reality. Further, we proved that the problems are NP-hard, and then greedy algorithms were proposed to solve the problems. The effectiveness of our approaches was verified by extensive simulations on the Momo dataset. Since both the street graph model and the UE density function are quite realistic, the results provided in this chapter provide valuable guidelines for realistic applications. The main content is presented in \cite{huang17drone}.
\chapter{Conclusion and Future Work}\label{conclusion}
\minitoc
\section{Contribution}
This report has addressed several challenges of introducing mobility into mobile networks, including how to track the mobile sinks and how to make use of them to improve the system performance. We conclude the report by highlighting the contributions and providing some possible future work. 

\begin{itemize}
\item \textbf{Viable Path Planning for Data Collection Robots in a Sensing Field with Obstacles}: The proposal of a path planning approach resolves several practical issues not having been sufficiently addressed so far when mobile sinks are used to collect data. We propose SVPP and $k$-SVPP and show that they are effective to design viable paths for unicycle mobile sinks with bounded angular velocity and can save much energy for the nodes through various simulations.
\item \textbf{An Energy Efficient Approach for Data Collection in Wireless Sensor Networks with Nonuniform Node Distribution Using Public Transportation Vehicles}: is an approach using a single mobile sink with fixed path for data collection. It aims at balancing the energy consumption to make the network operate as long as possible with all nodes alive. An energy-aware unequal clustering algorithm and an energy-aware routing algorithm are designed. Theoretical analysis is also provided.
\item \textbf{A Cluster based Compressive Data Collection for Wireless Sensor Networks with Mobile Sinks}: Similarly, this also considers the scenario of using a single mobile sink with fixed path. The difference lies in that we combine the technique of compressive sensing (CS) and clustering: within clusters, raw reading is transmitted; while CS measurement is transmitted between clusters and MS. We present an analytical model to describe the energy consumed by the nodes, based on which we figure out the optimal cluster radius. Two distributed implementations are presented and simulations are conducted. 
\item \textbf{Unusual Message Delivery Path Construction for Wireless Sensor Networks With Mobile Sinks}: The proposal of an algorithm targets on delivery unusual message to the mobile sinks within the allowed latency. Upon detecting any unusual message, the source node transmits the information to a set of selected target bus stop nodes. When buses pass by, the information is uploaded. We formulate the bus stop nodes selection as an integer programming problem. Some realistic features such as the uncertainties in arrival time and stopping duration are accounted. By simulations and experiments, we show the proposed approach can deliver the unusual message to mobile sinks within the allowed latency with higher reliability and efficiency than the alternatives. 
\item \textbf{Optimized Deployment of Autonomous Drones to Improve User Experience in Cellular Networks}: We formulate the constrained drone placement problems based on a novel street graph model associated with the UE density function (built up based on the real dataset). We prove that the problems are NP-hard and provide greedy solutions. Theoretical analysis on the approximation factors of the greedy algorithms is provided. Simulations on the real dataset are conducted. Since the models used here are quite realistic, the results can serve suitable guidelines in practice.
\end{itemize}
\section{Future work}
The research questions addressed in this report have created new opportunities for further research. We highlight some of them in this section.

In Chapter \ref{path_planning}, we use controllable mobile sinks to collect data from sensor nodes. We assume both the sensor nodes and obstacles are static. However, in practice, the obstacles can also be mobile. Thus, the development of algorithms which can tackle the mobile obstacles for data collection is one future direction.

In Chapter \ref{cluster_ms} and \ref{cluster_cs}, we only consider the scenario of using single mobile sink which follows a predefined trajectory to collect data from sensor nodes. The extension to multiple mobile sinks with fixed trajectories is worth studying, since it is promising in large scale networks.

In Chapter \ref{cluster_um}, we assume the wireless transmission between two nodes is reliable. However, in practice, the failure in transmission occurs frequently. Thus, accounting such feature is one possible means to improve our current approach.

In Chapter \ref{drone}, we present a primary study of using autonomous drones to serve mobile users in cellular networks. There are various direction we can consider. First, we have not touched the impact of drone altitude, which is believed to be another important factor influencing QoS. Second, the radio resource is assumed to be evenly allocated to the mobile users, which is a naive scheme. Designing a comprehensive strategy can improve the overall network capacity. Third, what we have studied is the problem of proactive deployment of autonomous drones based on a collected dataset Momo. However, in reality, the actual UE distribution may not exactly the same with Momo. Therefore, it is necessary to develop reactive strategies which can deal with the dynamic UE distribution.
\appendix
\chapter{Construction of tangents}\label{app}
We first describe the algorithm to construct tangents between a convex set and a outside point (\textit{Point-Conv}) and then between two convex sets (\textit{Conv-Conv}).

Let $Q$ be a convex set with a boundary approximated by $N$ points and $p$ be a point outside $Q$. The line connecting $p$ and one point on $Q$ is called a tangent candidate. \textit{Point-Conv} calculates the angle ($\alpha\in[-\pi,\pi]$) between X-axis direction and the each tangent candidate. The tangent candidate whose $\alpha$ value is the maximum (minimum) is a tangent between $p$ and $Q$, e.g., Figure \ref{fig:point_to_convex}. As there are $N$ tangent candidate to be checked, the time complexity of \textit{Point-Conv} is $O(N)$.

\begin{figure}[!h]
    \centering
    \begin{subfigure}[t]{0.2\textwidth}
        \includegraphics[width=\textwidth]{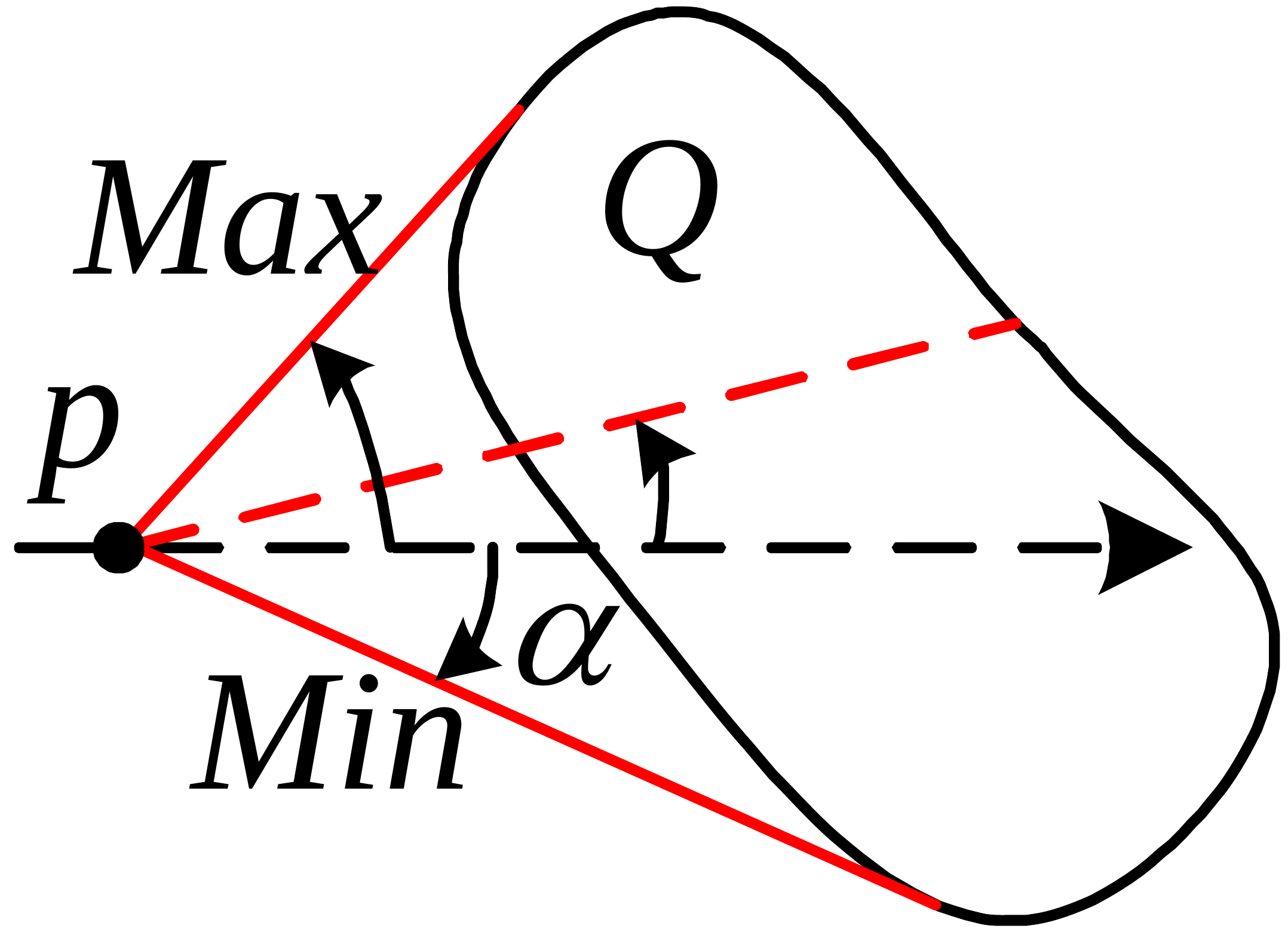}
        \caption{}
        \label{fig:point_to_convex}
    \end{subfigure}
    ~
    \begin{subfigure}[t]{0.25\textwidth}
        \includegraphics[width=\textwidth]{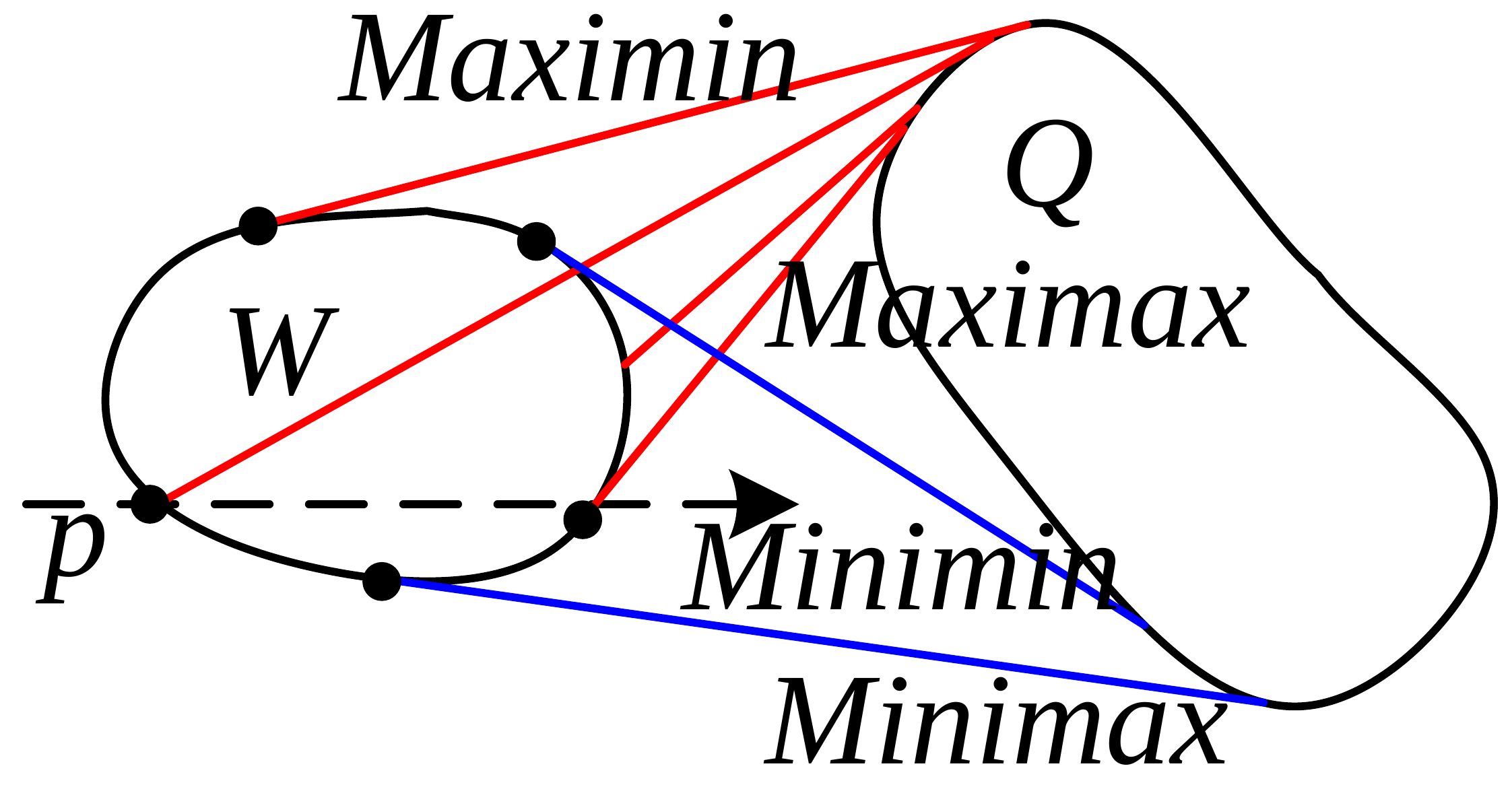}
        \caption{}
        \label{fig:convex_to_convex}
    \end{subfigure}
    \caption{(a) Two tangents are between a point and a convex set shown as the two solid lines marked by $Max$ and $Min$; (b) Four tangents are between two convex sets, marked by $Maximax$, $Maximin$, $Minimax$, and $Minimin$.}\label{fig:tangents}
\end{figure} 

Let $W$ be a convex set with a boundary approximated by $M$ points. Following Assumption \ref{position}, $W$ and $Q$ do not overlap. \textit{Conv-Conv} works as follows. For each point $p$ on the boundary of $W$, it uses \textit{Point-Conv} to determine a tangent with maximum $\alpha$ for $p$ and $Q$, e.g., see Figure \ref{fig:convex_to_convex}. Among these maxima $\alpha$ values, \textit{Point-Conv} finds the maximum and minimum values, which correspond to two tangents between $W$ and $Q$ (see the tangents marked by $Maximax$ and $Maximin$ in Figure \ref{fig:convex_to_convex}). At the same time, for each point on the boundary of $W$, \textit{Conv-Conv} can also determine a tangent with minimum $\alpha$ value. Among these minima values, it can find the maximum and minimum values, which correspond to the other two tangents. Since \textit{Conv-Conv} calls \textit{Point-Conv} $M$ times to get all the $\alpha$ values, the time complexity of \textit{Conv-Conv} is $O(MN)$. In reality, if a point on one convex set cannot 'see' the other convex set, none of the tangent candidates from this point will be a tangent, such as point $p$ shown in Figure \ref{fig:convex_to_convex}. This fact reduces the number of tangent candidates.\label{AppendixA}
\chapter{Critical Points under Non-identical Threat Functions}
Consider two agents $P_i$ and $P_j$ with radius $R_i$ and $R_j$ (see Fig. \ref{fig6}). 
Let $P_*$ $(x_*,y_*)$ be the type 2 critical point, and $D$ and $d$ be respectively the distance between $P_i$ and $P_j$, and $P_*$ and $P_i$. 

\begin{figure}[h]
\begin{center}
{\includegraphics[width=0.35\textwidth]{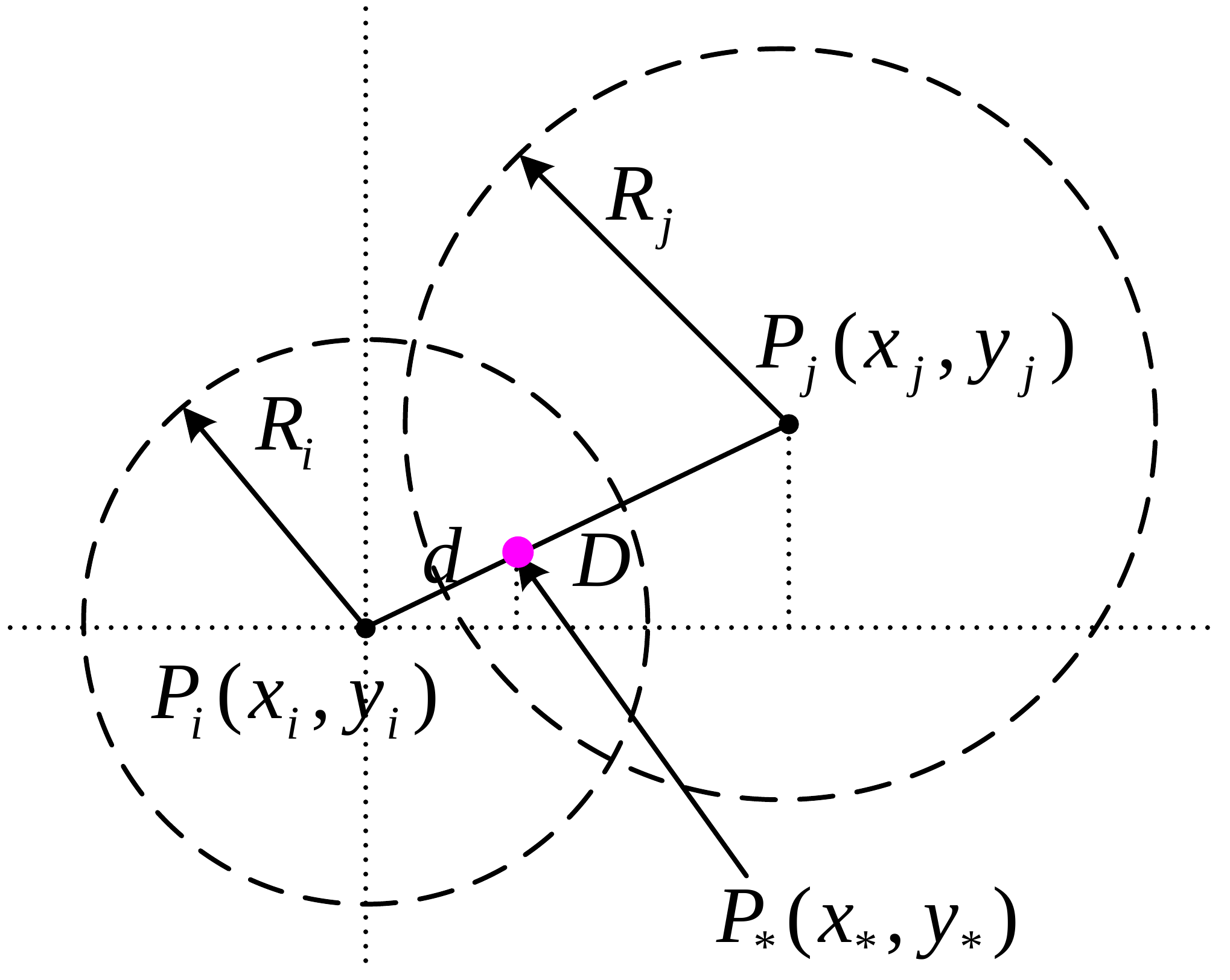}}
\caption{Calculating critical point of type 2 under non-identical threat radii.}
\label{fig6}
\end{center}
\end{figure}

According to Definition \ref{define_3_2} and model (\ref{eq3}), we have 

\begin{equation}\label{bili}
1-\frac{d}{R_i}=1-\frac{D-d}{R_j}.
\end{equation}
According to the properties of similar triangles, we have
\begin{equation}\label{similar_triangles}
\frac{x_{*}-x_{i}}{x_{j}-x_{i}}=\frac{y_{*}-y_{i}}{y_{j}-y_{i}}=\frac{d}{D}.
\end{equation}
Thus, the location of type 2 critical point can be obtained from Eq. (\ref{bili}) and (\ref{similar_triangles}):
\begin{equation}\label{cp_coordinates}
x_{*}=\frac{R_i}{R_i+R_j}(x_{j}-x_{i})+x_{i}, y_{*}=\frac{R_i}{R_i+R_j}(y_{j}-y_{i})+y_{i}
\end{equation}

Having these critical points, the corresponding threat levels can be calculated based on the threat function $f_i$, where the input is the distance between a critical point and the agent. \label{AppendixB}
\chapter{Tangent Lines}

Since any tangent line is determined by two tangent points, here we only provide the equation for tangent points.

There are two common tangent lines between a circle and a point outside. This covers the scenario of finding the tangent lines between $\mathcal{F}$ and agents. Let $A=(x_A, y_A)$ be the point and $(x_{t1},y_{t1})$ and $(x_{t2},y_{t2})$ be the tangent points, see, e.g., Fig. \ref{fig7}. Since $(x_{t1},y_{t1})$ is on the circle centred at $P_i$, we have

\begin{equation}\label{on_ciecle}
(x_{t1}-x_{i})^2+(y_{t1}-y_{i})^2=R_i^2.
\end{equation}
Since the line connecting $A$ and $(x_{t1},y_{t1})$ is perpendicular to the line connecting $P_i$ and $(x_{t1},y_{t1})$, and further suppose both of the lines have slopes, we know their slopes are opposite reciprocals:
\begin{equation}\label{reciprocals}
\frac{y_{t1}-y_A}{x_{t1}-x_A}\cdot \frac{y_{t1}-y_i}{x_{t1}-x_i}=-1
\end{equation}
Solving Eq. (\ref{on_ciecle}) and (\ref{reciprocals}), we obtain

\begin{equation}\label{tp1}
\begin{aligned}
x_{t1,t2}=\frac{R_i^2(x_A-x_i)\pm R_i(y_A-y_{i})\sqrt{D^2-R_i^2}}{D^2}+x_i\\
y_{t1,t2}=\frac{R_i^2(y_A-y_{i})\mp R_i(x_{A}-x_{i})\sqrt{D^2-R_i^2}}{D^2}+y_{i}
\end{aligned}
\end{equation}
where $D=\sqrt{(x_{A}-x_{i})^2+(y_{A}-y_{i})^2}$.

\begin{figure}[t]
\begin{center}
{\includegraphics[width=0.5\textwidth]{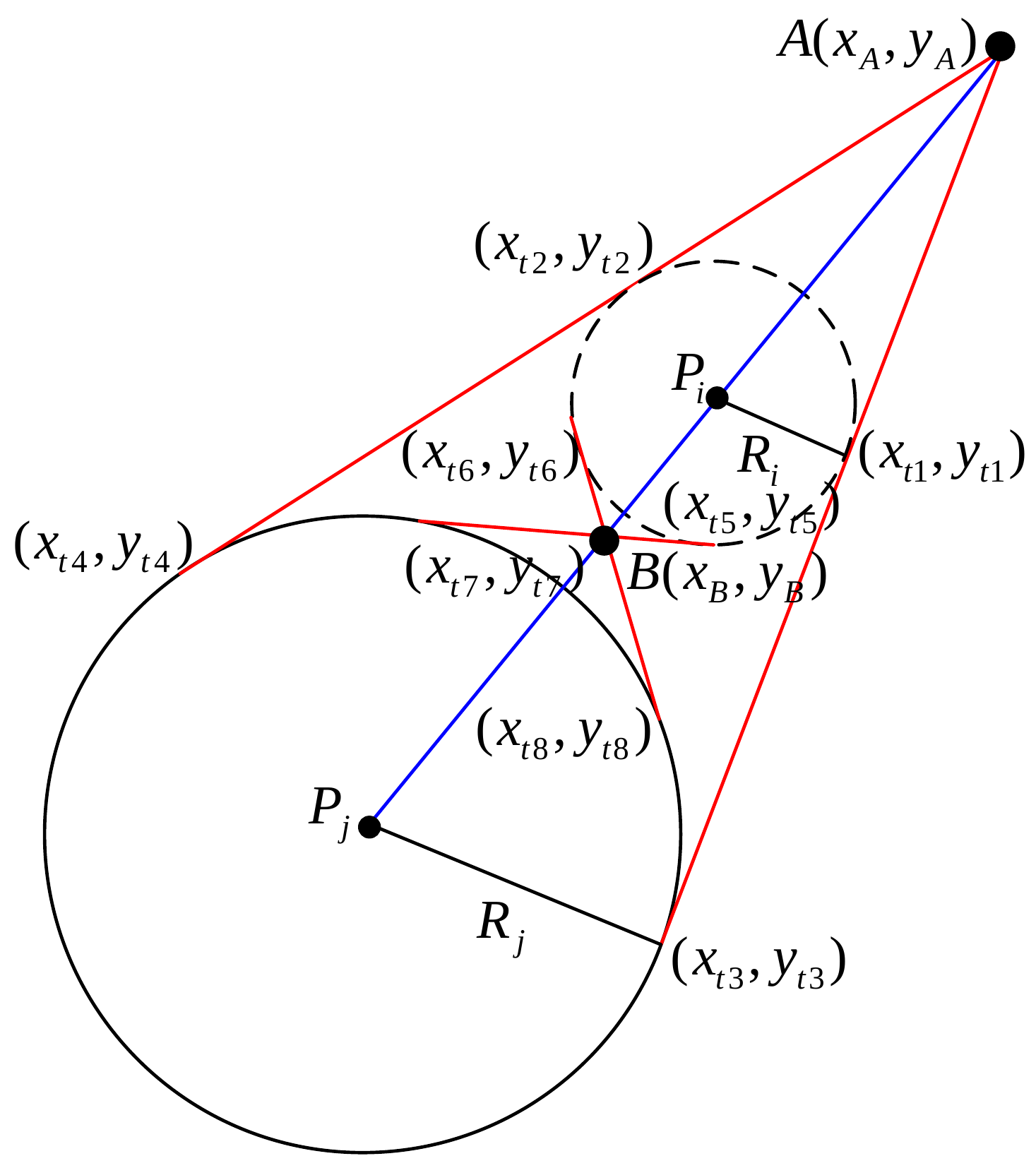}}
\caption{Tangent lines between 1) a circle and a point outside and 2) two circles.}
\label{fig7}
\end{center}
\end{figure}

There are four common tangent lines between two non-overlapped circles. This covers the scenario of finding the tangent lines between any two of initial circles and circles centred at agents. The formulation can be derived from the above case. Consider another agent $P_j(x_{j}, y_{j})$ and let $(x_{tk},y_{tk})$, $k=1,\cdots,8$ be the tangent points. To calculate the tangent points for the outer common tangent lines, we first find a point $A$ on the line passing both $P_i$ and $P_j$, see Fig. \ref{fig7}. 

Since $A$, $P_i$ and $P_j$ are on the same line, we have
\begin{equation}\label{same_line}
\frac{y_A-y_{i}}{x_A-x_{i}}=\frac{y_A-y_{j}}{x_A-x_{j}}.
\end{equation}

Based on the characteristics of similar triangles, we have $\frac{|AP_i|}{|AP_j|}=\frac{R_i}{R_j}$, i.e.,
\begin{equation}\label{ratio}
\frac{\sqrt{(x_A-x_{i})^2+(y_A-y_{i})^2}}{\sqrt{(x_A-x_{j})^2+(y_A-y_{j})^2}}=\frac{R_i}{R_j}
\end{equation}

By solving Eq. (\ref{same_line}) and (\ref{ratio}), we obtain the coordinates of $A$:
\begin{equation}
x_A=\frac{x_{i}R_j-x_{j}R_i}{R_j-R_i}, y_A=\frac{y_{i}R_j-y_{j}R_i}{R_j-R_i}
\end{equation}

With $A$, the tangent points on $P_i$ for the outer tangent lines can be calculated by Eq. (\ref{tp1}) directly. Further to find the tangent points on $P_j$ for the outer tangent lines, we can simply modify (\ref{tp1}) by replacing $(x_i,y_i)$ and $R_i$ with $(x_j,y_j)$ and $R_j$. 

For the tangent points of the inner common tangent lines, we need to find another outside point $B=(x_B,y_B)$, i.e., the intersecting point of two inner common tangent lines, see Fig. 7. 
 Using the same idea of finding $A$, the coordinates of $B$ are
\begin{equation}
x_B=\frac{x_{i}R_j+x_{j}R_i}{R_j+R_i},
y_B=\frac{y_{i}R_j+y_{j}R_i}{R_j+R_i}
\end{equation}

Having $B$, the tangent points for inner common tangent lines can be calculated by replacing $A$ with $B$ in the formulations for tangent points of the outer common tangent lines.

With these tangent points, the tangent lines can be constructed. Then $\mathcal{G}(\Theta(k))$ can be built up by taking all the tangent points, $p(0)$ and $\mathcal{F}$ as vertices and the tangent lines and arcs as edges, following Definition III.7.\label{AppendixC}
\small
\bibliographystyle{ieeetr}
\bibliography{mybibfile}

\begin{thebibliography}{100}

\bibitem{mainwaring2002habitat}
A.~Mainwaring, D.~Culler, J.~Polastre, R.~Szewczyk, and J.~Anderson, ``Wireless
  sensor networks for habitat monitoring,'' in {\em the 1st ACM international
  workshop on Wireless sensor networks and applications}, pp.~88--97, Acm,
  2002.

\bibitem{kim2007health}
S.~Kim, S.~Pakzad, D.~Culler, J.~Demmel, G.~Fenves, S.~Glaser, and M.~Turon,
  ``Health monitoring of civil infrastructures using wireless sensor
  networks,'' in {\em the 6th International Conference on Information
  Processing in Sensor Networks}, pp.~254--263, ACM, 2007.

\bibitem{zhao2011wireless}
G.~Zhao, ``Wireless sensor networks for industrial process monitoring and
  control: A survey,'' {\em Network Protocols and Algorithms}, vol.~3, no.~1,
  pp.~46--63, 2011.

\bibitem{butun2014survey}
I.~Butun, S.~D. Morgera, and R.~Sankar, ``A survey of intrusion detection
  systems in wireless sensor networks,'' {\em IEEE Communications Surveys \&
  Tutorials}, vol.~16, no.~1, pp.~266--282, 2014.

\bibitem{saha2007disaster}
S.~Saha and M.~Matsumoto, ``A framework for disaster management system and
  {WSN} protocol for rescue operation,'' in {\em TENCON}, pp.~1--4, IEEE, 2007.

\bibitem{UAV2016survey}
L.~Gupta, R.~Jain, and G.~Vaszkun, ``Survey of important issues in {UAV}
  communication networks,'' {\em IEEE Communications Surveys \& Tutorials},
  vol.~18, no.~2, pp.~1123--1152, 2016.

\bibitem{jiang2005perpetual}
X.~Jiang, J.~Polastre, and D.~Culler, ``Perpetual environmentally powered
  sensor networks,'' in {\em the 4th International Symposium on Information
  Processing in Sensor Networks}, p.~65, IEEE Press, 2005.

\bibitem{kansal2007power}
A.~Kansal, J.~Hsu, S.~Zahedi, and M.~B. Srivastava, ``Power management in
  energy harvesting sensor networks,'' {\em ACM Transactions on Embedded
  Computing Systems}, vol.~6, no.~4, p.~32, 2007.

\bibitem{park2008energy}
G.~Park, T.~Rosing, M.~D. Todd, C.~R. Farrar, and W.~Hodgkiss, ``Energy
  harvesting for structural health monitoring sensor networks,'' {\em Journal
  of Infrastructure Systems}, vol.~14, no.~1, pp.~64--79, 2008.

\bibitem{gilbert2008comparison}
J.~M. Gilbert and F.~Balouchi, ``Comparison of energy harvesting systems for
  wireless sensor networks,'' {\em International Journal of Automation and
  Computing}, vol.~5, no.~4, pp.~334--347, 2008.

\bibitem{seah2009wireless}
W.~K. Seah, Z.~A. Eu, and H.-P. Tan, ``Wireless sensor networks powered by
  ambient energy harvesting {(WSN-HEAP)}-survey and challenges,'' in {\em 1st
  International Conference on Wireless Communication, Vehicular Technology,
  Information Theory and Aerospace \& Electronic Systems Technology}, pp.~1--5,
  IEEE, 2009.

\bibitem{vullers2010energy}
R.~J. Vullers, R.~Van~Schaijk, H.~J. Visser, J.~Penders, and C.~Van~Hoof,
  ``Energy harvesting for autonomous wireless sensor networks,'' {\em IEEE
  Solid-State Circuits Magazine}, vol.~2, no.~2, pp.~29--38, 2010.

\bibitem{xie2013wireless}
L.~Xie, Y.~Shi, Y.~T. Hou, and A.~Lou, ``Wireless power transfer and
  applications to sensor networks,'' {\em IEEE Wireless Communications},
  vol.~20, no.~4, pp.~140--145, 2013.

\bibitem{xie2012making}
L.~Xie, Y.~Shi, Y.~T. Hou, and H.~D. Sherali, ``Making sensor networks
  immortal: An energy-renewal approach with wireless power transfer,'' {\em
  IEEE/ACM Transactions on networking}, vol.~20, no.~6, pp.~1748--1761, 2012.

\bibitem{cheng2009comlett}
T.~M. Cheng and A.~V. Savkin, ``A distributed self-deployment algorithm for the
  coverage of mobile wireless sensor networks,'' {\em IEEE Communications
  Letters}, vol.~13, no.~11, 2009.

\bibitem{cheng2011sweep}
T.~M. Cheng, A.~V. Savkin, and F.~Javed, ``Decentralized control of a group of
  mobile robots for deployment in sweep coverage,'' {\em Robotics and
  Autonomous Systems}, vol.~59, no.~7, pp.~497--507, 2011.

\bibitem{cheng2011barriersweep}
T.~M. Cheng and A.~V. Savkin, ``Decentralized control for mobile robotic sensor
  network self-deployment: Barrier and sweep coverage problems,'' {\em
  Robotica}, vol.~29, no.~2, pp.~283--294, 2011.

\bibitem{flood02}
B.~Williams and T.~Camp, ``Comparison of broadcasting techniques for mobile ad
  hoc networks,'' in {\em the 3rd ACM International Symposium on Mobile Ad Hoc
  Networking \& Computing}, pp.~194--205, ACM, 2002.

\bibitem{gossip06}
Z.~J. Haas, J.~Y. Halpern, and L.~Li, ``Gossip-based ad hoc routing,'' {\em
  IEEE/ACM Transactions on Networking (ToN)}, vol.~14, no.~3, pp.~479--491,
  2006.

\bibitem{al2004routing}
J.~N. Al-Karaki and A.~E. Kamal, ``Routing techniques in wireless sensor
  networks: a survey,'' {\em IEEE Wireless Communications}, vol.~11, no.~6,
  pp.~6--28, 2004.

\bibitem{intanagonwiwat2000directed}
C.~Intanagonwiwat, R.~Govindan, and D.~Estrin, ``Directed diffusion: A scalable
  and robust communication paradigm for sensor networks,'' in {\em the 6th
  Annual International Conference on Mobile Computing and Networking},
  pp.~56--67, ACM, 2000.

\bibitem{ye2001scalable}
F.~Ye, A.~Chen, S.~Lu, and L.~Zhang, ``A scalable solution to minimum cost
  forwarding in large sensor networks,'' in {\em the 10th International
  Conference on Computer Communications and Networks}, pp.~304--309, IEEE,
  2001.

\bibitem{schurgers2001energy}
C.~Schurgers and M.~B. Srivastava, ``Energy efficient routing in wireless
  sensor networks,'' in {\em Military Communications Conference}, vol.~1,
  pp.~357--361, IEEE, 2001.

\bibitem{CDG09}
C.~Luo, F.~Wu, J.~Sun, and C.~W. Chen, ``Compressive data gathering for
  large-scale wireless sensor networks,'' in {\em the 15th Annual International
  Conference on Mobile Computing and Networking}, pp.~145--156, ACM, 2009.

\bibitem{CS06_1}
E.~J. Cand{\`e}s, J.~Romberg, and T.~Tao, ``Robust uncertainty principles:
  Exact signal reconstruction from highly incomplete frequency information,''
  {\em IEEE Transactions on Information Theory}, vol.~52, no.~2, pp.~489--509,
  2006.

\bibitem{CS06_2}
E.~J. Candes and T.~Tao, ``Near-optimal signal recovery from random
  projections: Universal encoding strategies?,'' {\em IEEE Transactions on
  Information Theory}, vol.~52, no.~12, pp.~5406--5425, 2006.

\bibitem{heinzelman2000energy}
W.~R. Heinzelman, A.~Chandrakasan, and H.~Balakrishnan, ``Energy-efficient
  communication protocol for wireless microsensor networks,'' in {\em the 33rd
  Annual Hawaii International Conference on System Sciences}, pp.~1--10, IEEE,
  2000.

\bibitem{estrin1999next}
D.~Estrin, R.~Govindan, J.~Heidemann, and S.~Kumar, ``Next century challenges:
  Scalable coordination in sensor networks,'' in {\em the 5th Annual ACM/IEEE
  International Conference on Mobile Computing and Networking}, pp.~263--270,
  ACM, 1999.

\bibitem{younis2004heed}
O.~Younis and S.~Fahmy, ``Heed: a hybrid, energy-efficient, distributed
  clustering approach for ad hoc sensor networks,'' {\em IEEE Transactions on
  Mobile Computing}, vol.~3, no.~4, pp.~366--379, 2004.

\bibitem{chen2009unequal}
G.~Chen, C.~Li, M.~Ye, and J.~Wu, ``An unequal cluster-based routing protocol
  in wireless sensor networks,'' {\em Wireless Networks}, vol.~15, no.~2,
  pp.~193--207, 2009.

\bibitem{xiang11}
L.~Xiang, J.~Luo, and A.~Vasilakos, ``Compressed data aggregation for energy
  efficient wireless sensor networks,'' in {\em SECON}, pp.~46--54, IEEE, 2011.

\bibitem{jia14}
R.~Xie and X.~Jia, ``Transmission-efficient clustering method for wireless
  sensor networks using compressive sensing,'' {\em IEEE Transactions on
  Parallel and Distributed Systems}, vol.~25, no.~3, pp.~806--815, 2014.

\bibitem{CDG16}
D.~Ebrahimi and C.~Assi, ``On the interaction between scheduling and
  compressive data gathering in wireless sensor networks,'' {\em IEEE
  Transactions on Wireless Communications}, vol.~15, no.~4, pp.~2845--2858,
  2016.

\bibitem{TSP85}
E.~L. Lawler, ``The traveling salesman problem: a guided tour of combinatorial
  optimization,'' {\em WILEY-INTERSCIENCE SERIES IN DISCRETE MATHEMATICS},
  1985.

\bibitem{TSP94}
G.~Reinelt, {\em The traveling salesman: computational solutions for TSP
  applications}.
\newblock Springer-Verlag, 1994.

\bibitem{L92}
G.~Laporte, ``The traveling salesman problem: An overview of exact and
  approximate algorithms,'' {\em European Journal of Operational Research},
  vol.~59, no.~2, pp.~231--247, 1992.

\bibitem{TSPN07}
B.~Yuan, M.~Orlowska, and S.~Sadiq, ``On the optimal robot routing problem in
  wireless sensor networks,'' {\em IEEE Transactions on Knowledge and Data
  Engineering}, vol.~19, no.~9, pp.~1252--1261, 2007.

\bibitem{TSPN13CSS}
L.~He, J.~Pan, and J.~Xu, ``A progressive approach to reducing data collection
  latency in wireless sensor networks with mobile elements,'' {\em IEEE
  Transactions on Mobile Computing}, vol.~12, no.~7, pp.~1308--1320, 2013.

\bibitem{NB91}
C.~E. Noon and J.~C. Bean, ``An efficient transformation of the generalized
  traveling salesman problem,'' Tech. Rep. 91-26, Department of Industrial and
  Operations Engineering, University of Michigan, Ann Arbor, MI, USA, 1991.

\bibitem{SFB05}
K.~Savla, E.~Frazzoli, and F.~Bullo, ``On the point-to-point and traveling
  salesperson problems for {Dubins'} vehicle,'' in {\em American Control
  Conference}, pp.~786--791, IEEE, 2005.

\bibitem{tang2005motion}
Z.~Tang and U.~Ozguner, ``Motion planning for multitarget surveillance with
  mobile sensor agents,'' {\em IEEE Transactions on Robotics}, vol.~21, no.~5,
  pp.~898--908, 2005.

\bibitem{mcgee2006path}
T.~G. McGee and J.~K. Hedrick, ``Path planning and control for multiple point
  surveillance by an unmanned aircraft in wind,'' in {\em American Control
  Conference}, pp.~6--pp, IEEE, 2006.

\bibitem{rathinam2007resource}
S.~Rathinam, R.~Sengupta, and S.~Darbha, ``A resource allocation algorithm for
  multivehicle systems with nonholonomic constraints,'' {\em IEEE Transactions
  on Automation Science and Engineering}, vol.~4, no.~1, pp.~98--104, 2007.

\bibitem{savla2008traveling}
K.~Savla, E.~Frazzoli, and F.~Bullo, ``Traveling salesperson problems for the
  {Dubins} vehicle,'' {\em IEEE Transactions on Automatic Control}, vol.~53,
  no.~6, pp.~1378--1391, 2008.

\bibitem{cons2014integrating}
M.~S. Cons, T.~Shima, and C.~Domshlak, ``Integrating task and motion planning
  for unmanned aerial vehicles,'' {\em Unmanned Systems}, vol.~2, no.~01,
  pp.~19--38, 2014.

\bibitem{shima2014motion}
T.~Shima, P.~Isaiah, and Y.~Gottlieb, ``Motion planning and task assignment for
  unmanned aerial vehicles cooperating with unattended ground sensors,'' tech.
  rep., TECHNION RESEARCH AND DEVELOPMENT FOUNDATION LTD HAIFA (ISRAEL), 2014.

\bibitem{isaiah2015motion}
P.~Isaiah and T.~Shima, ``Motion planning algorithms for the {Dubins}
  travelling salesperson problem,'' {\em Automatica}, vol.~53, pp.~247--255,
  2015.

\bibitem{visitingcircle2015}
A.~Wichmann and T.~Korkmaz, ``Smooth path construction and adjustment for
  multiple mobile sinks in wireless sensor networks,'' {\em Computer
  Communications}, vol.~72, pp.~93--106, 2015.

\bibitem{lumelsky1987path}
V.~J. Lumelsky and A.~A. Stepanov, ``Path-planning strategies for a point
  mobile automaton moving amidst unknown obstacles of arbitrary shape,'' {\em
  Algorithmica}, vol.~2, no.~1, pp.~403--430, 1987.

\bibitem{lumelsky1990incorporating}
V.~J. Lumelsky and T.~Skewis, ``Incorporating range sensing in the robot
  navigation function,'' {\em IEEE Transactions on Systems, Man, and
  Cybernetics}, vol.~20, no.~5, pp.~1058--1069, 1990.

\bibitem{kamon1997sensory}
I.~Kamon and E.~Rivlin, ``Sensory-based motion planning with global proofs,''
  {\em IEEE Transactions on Robotics and Automation}, vol.~13, no.~6,
  pp.~814--822, 1997.

\bibitem{langer2007k}
R.~A. Langer, L.~S. Coelho, and G.~H. Oliveira, ``K-bug, a new bug approach for
  mobile robot's path planning,'' in {\em IEEE International Conference on
  Control Applications}, pp.~403--408, IEEE, 2007.

\bibitem{van2008planning}
J.~Van Den~Berg and M.~Overmars, ``Planning time-minimal safe paths amidst
  unpredictably moving obstacles,'' {\em The International Journal of Robotics
  Research}, vol.~27, no.~11-12, pp.~1274--1294, 2008.

\bibitem{toibero2009stable}
J.~M. Toibero, F.~Roberti, and R.~Carelli, ``Stable contour-following control
  of wheeled mobile robots,'' {\em Robotica}, vol.~27, no.~1, pp.~1--12, 2009.

\bibitem{MTS11}
A.~S. Matveev, H.~Teimoori, and A.~V. Savkin, ``A method for guidance and
  control of an autonomous vehicle in problems of border patrolling and
  obstacle avoidance,'' {\em Automatica}, vol.~47, no.~3, pp.~515--524, 2011.

\bibitem{matveev2012real}
A.~S. Matveev, C.~Wang, and A.~V. Savkin, ``Real-time navigation of mobile
  robots in problems of border patrolling and avoiding collisions with moving
  and deforming obstacles,'' {\em Robotics and Autonomous systems}, vol.~60,
  no.~6, pp.~769--788, 2012.

\bibitem{savkin2013simple}
A.~V. Savkin and C.~Wang, ``A simple biologically inspired algorithm for
  collision-free navigation of a unicycle-like robot in dynamic environments
  with moving obstacles,'' {\em Robotica}, vol.~31, no.~6, pp.~993--1001, 2013.

\bibitem{savkin2014seeking}
A.~V. Savkin and C.~Wang, ``Seeking a path through the crowd: Robot navigation
  in unknown dynamic environments with moving obstacles based on an integrated
  environment representation,'' {\em Robotics and Autonomous Systems}, vol.~62,
  no.~10, pp.~1568--1580, 2014.

\bibitem{matveev2015globally}
A.~S. Matveev, M.~C. Hoy, and A.~V. Savkin, ``A globally converging algorithm
  for reactive robot navigation among moving and deforming obstacles,'' {\em
  Automatica}, vol.~54, pp.~292--304, 2015.

\bibitem{HMS15}
M.~Hoy, A.~S. Matveev, and A.~V. Savkin, ``Algorithms for collision-free
  navigation of mobile robots in complex cluttered environments: a survey,''
  {\em Robotica}, vol.~33, pp.~463--497, 3 2015.

\bibitem{MSHW15}
A.~V. Savkin, A.~S. Matveev, M.~Hoy, and C.~Wang, ``Safe robot navigation among
  moving and steady obstacles,'' {\em Elsevier}, 2015.

\bibitem{MES04}
A.~A. Somasundara, A.~Ramamoorthy, M.~B. Srivastava, {\em et~al.}, ``Mobile
  element scheduling for efficient data collection in wireless sensor networks
  with dynamic deadlines,'' in {\em Real-Time Systems Symposium}, pp.~296--305,
  IEEE, 2004.

\bibitem{MES07}
A.~A. Somasundara, A.~Ramamoorthy, M.~B. Srivastava, {\em et~al.}, ``Mobile
  element scheduling with dynamic deadlines,'' {\em IEEE Transactions on Mobile
  Computing}, vol.~6, no.~4, pp.~395--410, 2007.

\bibitem{PBS05}
Y.~Gu, D.~Bozdag, E.~Ekici, F.~{\"O}zg{\"u}ner, and C.-G. Lee, ``Partitioning
  based mobile element scheduling in wireless sensor networks,'' in {\em
  SECON}, pp.~386--395, Citeseer, 2005.

\bibitem{rendezvous07}
G.~Xing, T.~Wang, Z.~Xie, and W.~Jia, ``Rendezvous planning in
  mobility-assisted wireless sensor networks,'' in {\em Real-Time Systems
  Symposium}, pp.~311--320, IEEE, 2007.

\bibitem{rendezvous12}
C.~Konstantopoulos, G.~Pantziou, D.~Gavalas, A.~Mpitziopoulos, and B.~Mamalis,
  ``A rendezvous-based approach enabling energy-efficient sensory data
  collection with mobile sinks,'' {\em IEEE Transactions on Parallel and
  Distributed Systems}, vol.~23, no.~5, pp.~809--817, 2012.

\bibitem{adaptivestop14}
A.~Kinalis, S.~Nikoletseas, D.~Patroumpa, and J.~Rolim, ``Biased sink mobility
  with adaptive stop times for low latency data collection in sensor
  networks,'' {\em Information Fusion}, vol.~15, pp.~56--63, 2014.

\bibitem{MA07}
M.~Ma and Y.~Yang, ``Sencar: an energy-efficient data gathering mechanism for
  large-scale multihop sensor networks,'' {\em IEEE Transactions on Parallel
  and Distributed Systems}, vol.~18, no.~10, pp.~1476--1488, 2007.

\bibitem{pathconstrained11IEEE}
N.~Vlajic, D.~Stevanovic, and G.~Spanogiannopoulos, ``Strategies for improving
  performance of {IEEE} 802.15. 4/{ZigBee} {WSNs} with path-constrained mobile
  sink (s),'' {\em Computer Communications}, vol.~34, no.~6, pp.~743--757,
  2011.

\bibitem{pathconstrained07}
L.~Song and D.~Hatzinakos, ``Architecture of wireless sensor networks with
  mobile sinks: Sparsely deployed sensors,'' {\em IEEE transactions on
  Vehicular Technology}, vol.~56, no.~4, pp.~1826--1836, 2007.

\bibitem{pathconstrained11}
S.~Gao, H.~Zhang, and S.~K. Das, ``Efficient data collection in wireless sensor
  networks with path-constrained mobile sinks,'' {\em IEEE Transactions on
  Mobile Computing}, vol.~10, no.~4, pp.~592--608, 2011.

\bibitem{DR59}
G.~B. Dantzig and J.~H. Ramser, ``The truck dispatching problem,'' {\em
  Management Science}, vol.~6, no.~1, pp.~80--91, 1959.

\bibitem{TV14}
P.~Toth and D.~Vigo, {\em Vehicle routing: problems, methods, and
  applications}, vol.~18.
\newblock SIAM, 2014.

\bibitem{TO05}
Z.~Tang and U.~Ozguner, ``Motion planning for multitarget surveillance with
  mobile sensor agents,'' {\em IEEE Transactions on Robotics}, vol.~21, no.~5,
  pp.~898--908, 2005.

\bibitem{RSD07}
S.~Rathinam, R.~Sengupta, and S.~Darbha, ``A resource allocation algorithm for
  multivehicle systems with nonholonomic constraints,'' {\em IEEE Transactions
  on Automation Science and Engineering}, vol.~4, no.~1, pp.~98--104, 2007.

\bibitem{FHK78}
G.~N. Frederickson, M.~S. Hecht, and C.~E. Kim, ``Approximation algorithms for
  some routing problems,'' in {\em SIAM Journal on Computing}, vol.~7,
  pp.~178--193, 1978.

\bibitem{BI09}
D.~Bhadauria and V.~Isler, ``Data gathering tours for mobile robots,'' in {\em
  IEEE/RSJ International Conference on Intelligent Robots and Systems},
  pp.~3868--3873, 2009.

\bibitem{KAUWWT12}
D.~Kim, B.~H. Abay, R.~Uma, W.~Wu, W.~Wang, and A.~O. Tokuta, ``Minimizing data
  collection latency in wireless sensor network with multiple mobile
  elements,'' in {\em INFOCOM}, pp.~504--512, IEEE, 2012.

\bibitem{KUAWWT14}
D.~Kim, R.~Uma, B.~H. Abay, W.~Wu, W.~Wang, and A.~O. Tokuta, ``Minimum latency
  multiple data mule trajectory planning in wireless sensor networks,'' {\em
  IEEE Transactions on Mobile Computing}, vol.~13, no.~4, pp.~838--851, 2014.

\bibitem{SONG07}
L.~Song and D.~Hatzinakos, ``Architecture of wireless sensor networks with
  mobile sinks: Sparsely deployed sensors,'' {\em IEEE Transactions on
  Vehicular Technology}, vol.~56, no.~4, pp.~1826--1836, 2007.

\bibitem{MEHRABI16}
A.~Mehrabi and K.~Kim, ``Maximizing data collection throughput on a path in
  energy harvesting sensor networks using a mobile sink,'' {\em IEEE
  Transactions on Mobile Computing}, vol.~15, no.~3, pp.~690--704, 2016.

\bibitem{CHAK03}
A.~Chakrabarti, A.~Sabharwal, and B.~Aazhang, ``Using predictable observer
  mobility for power efficient design of sensor networks,'' in {\em Information
  Processing in Sensor Networks}, pp.~129--145, Springer, 2003.

\bibitem{SOMA06}
A.~A. Somasundara, A.~Kansal, D.~D. Jea, D.~Estrin, and M.~B. Srivastava,
  ``Controllably mobile infrastructure for low energy embedded networks,'' {\em
  IEEE Transactions on Mobile Computing}, vol.~5, no.~8, pp.~958--973, 2006.

\bibitem{GAO11}
S.~Gao, H.~Zhang, and S.~K. Das, ``Efficient data collection in wireless sensor
  networks with path-constrained mobile sinks,'' {\em IEEE Transactions on
  Mobile Computing}, vol.~10, no.~4, pp.~592--608, 2011.

\bibitem{MAMALIS09}
B.~Mamalis, D.~Gavalas, C.~Konstantopoulos, and G.~Pantziou, ``Clustering in
  wireless sensor networks,'' {\em RFID and Sensor Networks: Architectures,
  Protocols, Security and Integrations, Y. Zhang, LT Yang, J. Chen, eds},
  pp.~324--353, 2009.

\bibitem{bandyopadhyay2003energy}
S.~Bandyopadhyay and E.~J. Coyle, ``An energy efficient hierarchical clustering
  algorithm for wireless sensor networks,'' in {\em INFOCOM}, vol.~3,
  pp.~1713--1723, IEEE, 2003.

\bibitem{KONS12}
C.~Konstantopoulos, G.~Pantziou, D.~Gavalas, A.~Mpitziopoulos, and B.~Mamalis,
  ``A rendezvous-based approach enabling energy-efficient sensory data
  collection with mobile sinks,'' {\em IEEE Transactions on Parallel and
  Distributed Systems}, vol.~23, no.~5, pp.~809--817, 2012.

\bibitem{CHEN09}
G.~Chen, C.~Li, M.~Ye, and J.~Wu, ``An unequal cluster-based routing protocol
  in wireless sensor networks,'' {\em Wireless Networks}, vol.~15, no.~2,
  pp.~193--207, 2009.

\bibitem{WEI11}
D.~Wei, Y.~Jin, S.~Vural, K.~Moessner, and R.~Tafazolli, ``An energy-efficient
  clustering solution for wireless sensor networks,'' {\em IEEE Transactions on
  Wireless Communications}, vol.~10, no.~11, pp.~3973--3983, 2011.

\bibitem{wu2008avoiding}
X.~Wu, G.~Chen, and S.~K. Das, ``Avoiding energy holes in wireless sensor
  networks with nonuniform node distribution,'' {\em IEEE Transactions on
  Parallel and Distributed Systems}, vol.~19, no.~5, pp.~710--720, 2008.

\bibitem{cardei2008non}
M.~Cardei, Y.~Yang, and J.~Wu, ``Non-uniform sensor deployment in mobile
  wireless sensor networks,'' in {\em International Symposium on a World of
  Wireless, Mobile and Multimedia Networks}, pp.~1--8, IEEE, 2008.

\bibitem{gps04}
Y.-C. Tseng, S.-P. Kuo, H.-W. Lee, and C.-F. Huang, ``Location tracking in a
  wireless sensor network by mobile agents and its data fusion strategies,''
  {\em The Computer Journal}, vol.~47, no.~4, pp.~448--460, 2004.

\bibitem{RSSI08}
F.~Viani, L.~Lizzi, P.~Rocca, M.~Benedetti, M.~Donelli, and A.~Massa, ``{Object
  Tracking Through RSSI Measurements in Wireless Sensor Networks},'' {\em
  Electronics Letters}, vol.~44, no.~10, pp.~653--654, 2008.

\bibitem{aoa06}
L.~Girod, M.~Lukac, V.~Trifa, and D.~Estrin, ``The design and implementation of
  a self-calibrating distributed acoustic sensing platform,'' in {\em
  International Conference on Embedded Networked Sensor Systems}, pp.~71--84,
  ACM, 2006.

\bibitem{toa03}
D.~Niculescu and B.~Nath, ``Ad hoc positioning system {(APS)} using {AOA},'' in
  {\em INFOCOM}, vol.~3, pp.~1734--1743, IEEE, 2003.

\bibitem{tdoa05}
N.~Patwari, J.~N. Ash, S.~Kyperountas, A.~O. Hero~III, R.~L. Moses, and N.~S.
  Correal, ``Locating the nodes: Cooperative localization in wireless sensor
  networks,'' {\em IEEE Signal Processing Magazine}, vol.~22, no.~4,
  pp.~54--69, 2005.

\bibitem{TPM14}
D.~Dhanapala and A.~Jayasumana, ``Topology preserving maps -- extracting layout
  maps of wireless sensor networks from virtual coordinates,'' {\em IEEE/ACM
  Transactions on Networking}, vol.~22, no.~3, pp.~784--797, 2014.

\bibitem{anchor11}
D.~C. Dhanapala and A.~P. Jayasumana, ``Anchor selection and topology
  preserving maps in {WSNs} -- a directional virtual coordinate based
  approach,'' in {\em Local Computer Networks}, pp.~571--579, 2011.

\bibitem{gunathillake2016maximum}
A.~Gunathillake, A.~V. Savkin, and A.~P. Jayasumana, ``Maximum likelihood
  topology maps for wireless sensor networks using an automated robot,'' in
  {\em Local Computer Networks}, pp.~339--347, IEEE, 2016.

\bibitem{gunathillake2017topology}
A.~Gunathillake, A.~V. Savkin, and A.~P. Jayasumana, ``Topology mapping
  algorithm for {2D} and {3D} wireless sensor networks based on maximum
  likelihood estimation,'' {\em Computer Networks}, 2017.

\bibitem{cluster03}
H.~Yang and B.~Sikdar, ``A protocol for tracking mobile targets using sensor
  networks,'' in {\em Sensor Network Protocols and Applications}, pp.~71--81,
  IEEE, 2003.

\bibitem{tree04}
W.~Zhang and G.~Cao, ``Optimizing tree reconfiguration for mobile target
  tracking in sensor networks,'' in {\em INFOCOM}, vol.~4, pp.~2434--2445,
  IEEE, 2004.

\bibitem{particle10}
N.~Ahmed, M.~Rutten, T.~Bessell, S.~S. Kanhere, N.~Gordon, and S.~Jha,
  ``Detection and tracking using particle-filter-based wireless sensor
  networks,'' {\em IEEE Transactions on Mobile Computing}, vol.~9, no.~9,
  pp.~1332--1345, 2010.

\bibitem{binary08}
P.~M. Djuri{\'c}, M.~Vemula, and M.~F. Bugallo, ``Target tracking by particle
  filtering in binary sensor networks,'' {\em IEEE Transactions on Signal
  Processing}, vol.~56, no.~6, pp.~2229--2238, 2008.

\bibitem{disk05}
J.~Yick, B.~Mukherjee, and D.~Ghosal, ``Analysis of a prediction-based mobility
  adaptive tracking algorithm,'' in {\em International Conference on Broadband
  Networks}, pp.~753--760, IEEE, 2005.

\bibitem{ekf79}
L.~Ljung, ``Asymptotic behavior of the extended {Kalman} filter as a parameter
  estimator for linear systems,'' {\em IEEE Transactions on Automatic Control},
  vol.~24, no.~1, pp.~36--50, 1979.

\bibitem{likelihood68}
R.~Mortensen, ``Maximum-likelihood recursive nonlinear filtering,'' {\em
  Journal of Optimization Theory and Applications}, vol.~2, no.~6,
  pp.~386--394, 1968.

\bibitem{ekf09energy}
J.~Lin, W.~Xiao, F.~L. Lewis, and L.~Xie, ``Energy-efficient distributed
  adaptive multisensor scheduling for target tracking in wireless sensor
  networks,'' {\em IEEE Transactions on Instrumentation and Measurement},
  vol.~58, no.~6, pp.~1886--1896, 2009.

\bibitem{ekf12missing}
J.~Hu, Z.~Wang, H.~Gao, and L.~K. Stergioulas, ``Extended {Kalman} filtering
  with stochastic nonlinearities and multiple missing measurements,'' {\em
  Automatica}, vol.~48, no.~9, pp.~2007--2015, 2012.

\bibitem{ekf12sparsity}
E.~Masazade, M.~Fardad, and P.~K. Varshney, ``Sparsity-promoting extended
  {Kalman} filtering for target tracking in wireless sensor networks,'' {\em
  IEEE Signal Processing Letters}, vol.~19, no.~12, pp.~845--848, 2012.

\bibitem{rekf99}
I.~R. Petersen and A.~V. Savkin, {\em Robust {Kalman} Filtering for Signals and
  Systems with Large Uncertainties}.
\newblock Boston, MA: Birkhauser, 1999.

\bibitem{rekf09}
A.~G. Kallapur, I.~R. Petersen, and S.~G. Anavatti, ``A discrete-time robust
  extended {Kalman} filter for uncertain systems with sum quadratic
  constraints,'' {\em IEEE Transactions on Automatic Control}, vol.~54, no.~4,
  pp.~850--854, 2009.

\bibitem{rekf04location}
P.~N. Pathirana, A.~V. Savkin, and S.~Jha, ``Location estimation and trajectory
  prediction for cellular networks with mobile base stations,'' {\em IEEE
  Transactions on Vehicular Technology}, vol.~53, no.~6, pp.~1903--1913, 2004.

\bibitem{pathirana2005node}
P.~N. Pathirana, N.~Bulusu, A.~V. Savkin, and S.~Jha, ``Node localization using
  mobile robots in delay-tolerant sensor networks,'' {\em IEEE Transactions on
  Mobile Computing}, vol.~4, no.~3, pp.~285--296, 2005.

\bibitem{binary11}
J.~Chen, K.~Cao, K.~Li, and Y.~Sun, ``Distributed sensor activation algorithm
  for target tracking with binary sensor networks,'' {\em Cluster Computing},
  vol.~14, no.~1, pp.~55--64, 2011.

\bibitem{binary05}
W.~Kim, K.~Mechitov, J.-Y. Choi, and S.~Ham, ``On target tracking with binary
  sensor networks,'' in {\em the 4th International Symposium on Information
  Processing in Sensor Networks}, pp.~301--308, IEEE, 2005.

\bibitem{TPMtrack13}
Y.~Jiang, D.~Dhanapala, and A.~Jayasumana, ``Tracking and prediction of
  mobility without physical distance measurements in sensor networks,'' in {\em
  IEEE International Conference on Communications}, pp.~1845--1850, IEEE, 2013.

\bibitem{gunathillake2017mobile}
A.~Gunathillake and A.~V. Savkin, ``Mobile robot navigation for emergency
  source seeking using sensor network topology maps,'' in {\em 36th Chinese
  Control Conference}, pp.~6027--6030, IEEE, 2017.

\bibitem{ashanie16}
A.~Gunathillake, A.~V. Savkin, and A.~Jayasumana, ``Decentralized time-based
  target searching algorithm using sensor network topology maps,'' in {\em the
  41st Conference on Local Computer Networks Workshops}, pp.~173--180, IEEE,
  2016.

\bibitem{ashanie17}
A.~Gunathillake, A.~V. Savkin, and A.~Jayasumana, ``Robust {Kalman} filter
  based decentralized target search and prediction with topology maps,'' {\em
  IET Wireless Sensor Systems}, November 2017.

\bibitem{bor2016new}
I.~Bor-Yaliniz and H.~Yanikomeroglu, ``The new frontier in {RAN} heterogeneity:
  Multi-tier drone-cells,'' {\em IEEE Communications Magazine}, vol.~54,
  no.~11, pp.~48--55, 2016.

\bibitem{shafiq2016characterizing}
M.~Z. Shafiq, L.~Ji, A.~X. Liu, J.~Pang, S.~Venkataraman, and J.~Wang,
  ``Characterizing and optimizing cellular network performance during crowded
  events,'' {\em IEEE/ACM Transactions on Networking}, vol.~24, no.~3,
  pp.~1308--1321, 2016.

\bibitem{nakamura2013trends}
T.~Nakamura, S.~Nagata, A.~Benjebbour, Y.~Kishiyama, T.~Hai, S.~Xiaodong,
  Y.~Ning, and L.~Nan, ``Trends in small cell enhancements in lte advanced,''
  {\em IEEE Communications Magazine}, vol.~51, no.~2, pp.~98--105, 2013.

\bibitem{bhushan2014network}
N.~Bhushan, J.~Li, D.~Malladi, R.~Gilmore, D.~Brenner, A.~Damnjanovic,
  R.~Sukhavasi, C.~Patel, and S.~Geirhofer, ``Network densification: the
  dominant theme for wireless evolution into {5G},'' {\em IEEE Communications
  Magazine}, vol.~52, no.~2, pp.~82--89, 2014.

\bibitem{ge20165g}
X.~Ge, S.~Tu, G.~Mao, C.-X. Wang, and T.~Han, ``{5G} ultra-dense cellular
  networks,'' {\em IEEE Wireless Communications}, vol.~23, no.~1, pp.~72--79,
  2016.

\bibitem{bennis2013cellular}
M.~Bennis, M.~Simsek, A.~Czylwik, W.~Saad, S.~Valentin, and M.~Debbah, ``When
  cellular meets {WiFi} in wireless small cell networks,'' {\em IEEE
  Communications Magazine}, vol.~51, no.~6, pp.~44--50, 2013.

\bibitem{hwang2013holistic}
I.~Hwang, B.~Song, and S.~S. Soliman, ``A holistic view on hyper-dense
  heterogeneous and small cell networks,'' {\em IEEE Communications Magazine},
  vol.~51, no.~6, pp.~20--27, 2013.

\bibitem{ghazzai2016optimized}
H.~Ghazzai, E.~Yaacoub, M.-S. Alouini, Z.~Dawy, and A.~Abu-Dayya, ``Optimized
  lte cell planning with varying spatial and temporal user densities,'' {\em
  IEEE Transactions on Vehicular Technology}, vol.~65, no.~3, pp.~1575--1589,
  2016.

\bibitem{zhao2017approximation}
W.~Zhao, S.~Wang, C.~Wang, and X.~Wu, ``Approximation algorithms for cell
  planning in heterogeneous networks,'' {\em IEEE Transactions on Vehicular
  Technology}, vol.~66, no.~2, pp.~1561--1572, 2017.

\bibitem{arnold2010power}
O.~Arnold, F.~Richter, G.~Fettweis, and O.~Blume, ``Power consumption modeling
  of different base station types in heterogeneous cellular networks,'' in {\em
  Future Network and Mobile Summit}, pp.~1--8, IEEE, 2010.

\bibitem{ashraf2011sleep}
I.~Ashraf, F.~Boccardi, and L.~Ho, ``Sleep mode techniques for small cell
  deployments,'' {\em IEEE Communications Magazine}, vol.~49, no.~8, 2011.

\bibitem{marsan2009optimal}
M.~A. Marsan, L.~Chiaraviglio, D.~Ciullo, and M.~Meo, ``Optimal energy savings
  in cellular access networks,'' in {\em International Conference on
  Communications Workshops}, pp.~1--5, IEEE, 2009.

\bibitem{soh2013energy}
Y.~S. Soh, T.~Q. Quek, M.~Kountouris, and H.~Shin, ``Energy efficient
  heterogeneous cellular networks,'' {\em IEEE Journal on Selected Areas in
  Communications}, vol.~31, no.~5, pp.~840--850, 2013.

\bibitem{hossain2014evolution}
E.~Hossain, M.~Rasti, H.~Tabassum, and A.~Abdelnasser, ``Evolution toward {5G}
  multi-tier cellular wireless networks: An interference management
  perspective,'' {\em IEEE Wireless Communications}, vol.~21, no.~3,
  pp.~118--127, 2014.

\bibitem{ye2013user}
Q.~Ye, B.~Rong, Y.~Chen, M.~Al-Shalash, C.~Caramanis, and J.~G. Andrews, ``User
  association for load balancing in heterogeneous cellular networks,'' {\em
  IEEE Transactions on Wireless Communications}, vol.~12, no.~6,
  pp.~2706--2716, 2013.

\bibitem{singh2014joint}
S.~Singh and J.~G. Andrews, ``Joint resource partitioning and offloading in
  heterogeneous cellular networks,'' {\em IEEE Transactions on Wireless
  Communications}, vol.~13, no.~2, pp.~888--901, 2014.

\bibitem{facebook2014}
M.~Zuckerberg, ``Connecting the world from the sky,'' 2014.

\bibitem{google2014}
S.~Katikala, ``Google™ project loon,'' {\em InSight: Rivier Academic
  Journal}, vol.~10, no.~2, pp.~1--6, 2014.

\bibitem{chae2015iot}
H.~Chae, J.~Park, H.~Song, Y.~Kim, and H.~Jeong, ``The {IoT} based automate
  landing system of a drone for the round-the-clock surveillance solution,'' in
  {\em International Conference on Advanced Intelligent Mechatronics},
  pp.~1575--1580, IEEE, 2015.

\bibitem{mozaffari2015drone}
M.~Mozaffari, W.~Saad, M.~Bennis, and M.~Debbah, ``Drone small cells in the
  clouds: Design, deployment and performance analysis,'' in {\em GLOBECOM},
  pp.~1--6, IEEE, 2015.

\bibitem{chi2012civil}
T.-Y. Chi, Y.~Ming, S.-Y. Kuo, C.-C. Liao, {\em et~al.}, ``Civil {UAV} path
  planning algorithm for considering connection with cellular data network,''
  in {\em the 12th International Conference on Computer and Information
  Technology}, pp.~327--331, IEEE, 2012.

\bibitem{al2014modeling}
A.~Al-Hourani, S.~Kandeepan, and A.~Jamalipour, ``Modeling air-to-ground path
  loss for low altitude platforms in urban environments,'' in {\em GLOBECOM},
  pp.~2898--2904, IEEE, 2014.

\bibitem{al2014optimal}
A.~Al-Hourani, S.~Kandeepan, and S.~Lardner, ``Optimal {LAP} altitude for
  maximum coverage,'' {\em IEEE Wireless Communications Letters}, vol.~3,
  no.~6, pp.~569--572, 2014.

\bibitem{bor2016efficient}
R.~I. Bor-Yaliniz, A.~El-Keyi, and H.~Yanikomeroglu, ``Efficient {3-D}
  placement of an aerial base station in next generation cellular networks,''
  in {\em International Conference on Communications}, pp.~1--5, IEEE, 2016.

\bibitem{alzenad20173d}
M.~Alzenad, A.~El-Keyi, F.~Lagum, and H.~Yanikomeroglu, ``{3D} placement of an
  unmanned aerial vehicle base station {(UAV-BS)} for energy-efficient maximal
  coverage,'' {\em IEEE Wireless Communications Letters}, vol.~6, no.~4,
  pp.~434--437, 2017.

\bibitem{sharma2016uav}
V.~Sharma, M.~Bennis, and R.~Kumar, ``{UAV}-assisted heterogeneous networks for
  capacity enhancement,'' {\em IEEE Communications Letters}, vol.~20, no.~6,
  pp.~1207--1210, 2016.

\bibitem{mozaffari2016efficient}
M.~Mozaffari, W.~Saad, M.~Bennis, and M.~Debbah, ``Efficient deployment of
  multiple unmanned aerial vehicles for optimal wireless coverage,'' {\em IEEE
  Communications Letters}, vol.~20, no.~8, pp.~1647--1650, 2016.

\bibitem{kalantari2016number}
E.~Kalantari, H.~Yanikomeroglu, and A.~Yongacoglu, ``On the number and {3D}
  placement of drone base stations in wireless cellular networks,'' in {\em
  IEEE Vehicular Technology Conference}, pp.~1--6, IEEE, 2016.

\bibitem{rohde2012interference}
S.~Rohde and C.~Wietfeld, ``Interference aware positioning of aerial relays for
  cell overload and outage compensation,'' in {\em IEEE Vehicular Technology
  Conference}, pp.~1--5, IEEE, 2012.

\bibitem{yang2017proactive}
P.~Yang, X.~Cao, C.~Yin, Z.~Xiao, X.~Xi, and D.~Wu, ``Proactive drone-cell
  deployment: Overload relief for a cellular network under flash crowd
  traffic,'' {\em IEEE Transactions on Intelligent Transportation Systems},
  vol.~18, no.~10, pp.~2877 -- 2892, 2017.

\bibitem{becvar2017performance}
Z.~Becvar, M.~Vondra, P.~Mach, J.~Plachy, and D.~Gesbert, ``Performance of
  mobile networks with {UAVs}: Can flying base stations substitute ultra-dense
  small cells?,'' in {\em 23th European Wireless Conference}, pp.~1--7, VDE,
  2017.

\bibitem{chen2017caching}
M.~Chen, M.~Mozaffari, W.~Saad, C.~Yin, M.~Debbah, and C.~S. Hong, ``Caching in
  the sky: Proactive deployment of cache-enabled unmanned aerial vehicles for
  optimized quality-of-experience,'' {\em IEEE Journal on Selected Areas in
  Communications}, vol.~35, no.~5, pp.~1046--1061, 2017.

\bibitem{savkin2004coordinated}
A.~V. Savkin, ``Coordinated collective motion of groups of autonomous mobile
  robots: Analysis of {Vicsek's} model,'' {\em IEEE Transactions on Automatic
  Control}, vol.~49, no.~6, pp.~981--982, 2004.

\bibitem{matveev2003problem}
A.~S. Matveev and A.~V. Savkin, ``The problem of state estimation via
  asynchronous communication channels with irregular transmission times,'' {\em
  IEEE Transactions on Automatic Control}, vol.~48, no.~4, pp.~670--676, 2003.

\bibitem{matveev2004problem}
A.~S. Matveev and A.~V. Savkin, ``The problem of {LQG} optimal control via a
  limited capacity communication channel,'' {\em Systems \& Control Letters},
  vol.~53, no.~1, pp.~51--64, 2004.

\bibitem{matveev2007analogue}
A.~S. Matveev and A.~V. Savkin, ``An analogue of {Shannon} information theory
  for detection and stabilization via noisy discrete communication channels,''
  {\em SIAM Journal on Control and Optimization}, vol.~46, no.~4,
  pp.~1323--1367, 2007.

\bibitem{ren2008distributed}
W.~Ren and R.~W. Beard, {\em Distributed consensus in multi-vehicle cooperative
  control}.
\newblock Springer, 2008.

\bibitem{dimarogonas2012distributed}
D.~V. Dimarogonas, E.~Frazzoli, and K.~H. Johansson, ``Distributed
  event-triggered control for multi-agent systems,'' {\em IEEE Transactions on
  Automatic Control}, vol.~57, no.~5, pp.~1291--1297, 2012.

\bibitem{cao2013overview}
Y.~Cao, W.~Yu, W.~Ren, and G.~Chen, ``An overview of recent progress in the
  study of distributed multi-agent coordination,'' {\em IEEE Transactions on
  Industrial informatics}, vol.~9, no.~1, pp.~427--438, 2013.

\bibitem{azuma2013broadcast}
S.-i. Azuma, R.~Yoshimura, and T.~Sugie, ``Broadcast control of multi-agent
  systems,'' {\em Automatica}, vol.~49, no.~8, pp.~2307--2316, 2013.

\bibitem{cassandras2013optimal}
C.~G. Cassandras, X.~Lin, and X.~Ding, ``An optimal control approach to the
  multi-agent persistent monitoring problem,'' {\em IEEE Transactions on
  Automatic Control}, vol.~58, no.~4, pp.~947--961, 2013.

\bibitem{shamma2008cooperative}
J.~Shamma, {\em Cooperative control of distributed multi-agent systems}.
\newblock John Wiley \& Sons, 2008.

\bibitem{ian02survey}
I.~F. Akyildiz, W.~Su, Y.~Sankarasubramaniam, and E.~Cayirci, ``Wireless sensor
  networks: a survey,'' {\em Computer Networks}, vol.~38, no.~4, pp.~393--422,
  2002.

\bibitem{mobilesinksurvey11}
M.~Di~Francesco, S.~K. Das, and G.~Anastasi, ``Data collection in wireless
  sensor networks with mobile elements: A survey,'' {\em ACM Transactions on
  Sensor Networks}, vol.~8, no.~1, p.~7, 2011.

\bibitem{OR15survey}
A.~Boukerche and A.~Darehshoorzadeh, ``Opportunistic routing in wireless
  networks: Models, algorithms, and classifications,'' {\em ACM Computing
  Surveys}, vol.~47, no.~2, p.~22, 2015.

\bibitem{feng2013survey}
D.~Feng, C.~Jiang, G.~Lim, L.~J. Cimini, G.~Feng, and G.~Y. Li, ``A survey of
  energy-efficient wireless communications,'' {\em IEEE Communications Surveys
  \& Tutorials}, vol.~15, no.~1, pp.~167--178, 2013.

\bibitem{Gu16}
Y.~Gu, F.~Ren, Y.~Ji, and J.~Li, ``The evolution of sink mobility management in
  wireless sensor networks: A survey,'' {\em IEEE Communications Surveys \&
  Tutorials}, vol.~18, no.~1, pp.~507--524, 2016.

\bibitem{savkin2015book}
A.~V. Savkin, T.~M. Cheng, Z.~Xi, F.~Javed, A.~S. Matveev, and H.~Nguyen, {\em
  Decentralized coverage control problems for mobile robotic sensor and
  actuator networks}.
\newblock John Wiley \& Sons, 2015.

\bibitem{hasan2011green}
Z.~Hasan, H.~Boostanimehr, and V.~K. Bhargava, ``Green cellular networks: A
  survey, some research issues and challenges,'' {\em IEEE Communications
  surveys \& tutorials}, vol.~13, no.~4, pp.~524--540, 2011.

\bibitem{damnjanovic2011survey}
A.~Damnjanovic, J.~Montojo, Y.~Wei, T.~Ji, T.~Luo, M.~Vajapeyam, T.~Yoo,
  O.~Song, and D.~Malladi, ``A survey on {3GPP} heterogeneous networks,'' {\em
  IEEE Wireless Communications}, vol.~18, no.~3, 2011.

\bibitem{aliu2013survey}
O.~G. Aliu, A.~Imran, M.~A. Imran, and B.~Evans, ``A survey of self
  organisation in future cellular networks,'' {\em IEEE Communications Surveys
  \& Tutorials}, vol.~15, no.~1, pp.~336--361, 2013.

\bibitem{huang17survey}
H.~Huang, A.~V. Savkin, M.~Ding, and C.~Huang, ``Data collection and energy
  charging by mobile robots in wireless sensor networks: A survey,'' {\em
  submitted to IEEE Communications Surveys \& Tutorials}.

\bibitem{tpm10topology}
D.~C. Dhanapala and A.~P. Jayasumana, ``Topology preserving maps from virtual
  coordinates for wireless sensor networks,'' in {\em Local Computer Networks},
  pp.~136--143, IEEE, 2010.

\bibitem{james1998nonlinear}
M.~R. James and I.~R. Petersen, ``Nonlinear state estimation for uncertain
  systems with an integral constraint,'' {\em IEEE Transactions on Signal
  Processing}, vol.~46, no.~11, pp.~2926--2937, 1998.

\bibitem{peterson2000robust}
I.~R. Peterson, V.~Ugronovskii, and A.~V. Savkin, {\em Robust control design
  using H-infinity methods}.
\newblock Springer-Verlag, London, 2000.

\bibitem{mobility97}
M.~M. Zonoozi and P.~Dassanayake, ``User mobility modeling and characterization
  of mobility patterns,'' {\em IEEE Journal on Selected Areas in
  Communications}, vol.~15, no.~7, pp.~1239--1252, 1997.

\bibitem{VRP92}
G.~Laporte, ``The vehicle routing problem: An overview of exact and approximate
  algorithms,'' {\em European journal of operational research}, vol.~59, no.~3,
  pp.~345--358, 1992.

\bibitem{dubins1957}
L.~E. Dubins, ``On curves of minimal length with a constraint on average
  curvature, and with prescribed initial and terminal positions and tangents,''
  {\em American Journal of Mathematics}, pp.~497--516, 1957.

\bibitem{andrey11unicycle_2}
A.~S. Matveev, H.~Teimoori, and A.~V. Savkin, ``Navigation of a unicycle-like
  mobile robot for environmental extremum seeking,'' {\em Automatica}, vol.~47,
  no.~1, pp.~85--91, 2011.

\bibitem{andrey11unicycle_3}
A.~S. Matveev, H.~Teimoori, and A.~V. Savkin, ``Range-only measurements based
  target following for wheeled mobile robots,'' {\em Automatica}, vol.~47,
  no.~1, pp.~177--184, 2011.

\bibitem{savkin2013reactive}
A.~V. Savkin and M.~Hoy, ``Reactive and the shortest path navigation of a
  wheeled mobile robot in cluttered environments,'' {\em Robotica}, vol.~31,
  no.~2, pp.~323--330, 2013.

\bibitem{obstaclefree10}
A.~K. Kumar and K.~M. Sivalingam, ``Energy-efficient mobile data collection in
  wireless sensor networks with delay reduction using wireless communication,''
  in {\em the 2nd International Conference on Communication Systems and
  Networks}, pp.~1--10, IEEE, 2010.

\bibitem{rendezvous08}
G.~Xing, T.~Wang, W.~Jia, and M.~Li, ``Rendezvous design algorithms for
  wireless sensor networks with a mobile base station,'' in {\em the 9th ACM
  International Symposium on Mobile Ad Hoc Networking and Computing},
  pp.~231--240, ACM, 2008.

\bibitem{OOD10}
K.~J. Obermeyer, P.~Oberlin, and S.~Darbha, ``Sampling-based roadmap methods
  for a visual reconnaissance {UAV},'' in {\em AIAA Conference on Guidance,
  Navigation and Control}, 2010.

\bibitem{IH13}
J.~T. Isaacs and J.~P. Hespanha, ``{Dubins} traveling salesman problem with
  neighborhoods: a graph-based approach,'' {\em Algorithms}, vol.~6, no.~1,
  pp.~84--99, 2013.

\bibitem{huang2017viable}
H.~Huang and A.~V. Savkin, ``Viable path planning for data collection robots in
  a sensing field with obstacles,'' {\em Computer Communications}, vol.~111,
  pp.~84--96, 2017.

\bibitem{savkin2017optimal}
A.~Savkin and H.~Huang, ``Optimal aircraft planar navigation in static threat
  environments,'' {\em IEEE Transactions on Aerospace and Electronic Systems},
  vol.~53, no.~5, pp.~2413--2426, 2017.

\bibitem{huang2016path}
H.~Huang and A.~V. Savkin, ``Path planning algorithms for a mobile robot
  collecting data in a wireless sensor network deployed in a region with
  obstacles,'' in {\em 35th Chinese Control Conference}, pp.~8464--8467, IEEE,
  2016.

\bibitem{savkin2016problem}
A.~V. Savkin and H.~Huang, ``The problem of minimum risk path planning for
  flying robots in dangerous environments,'' in {\em 35th Chinese Control
  Conference}, pp.~5404--5408, IEEE, 2016.

\bibitem{manchester2006circular}
I.~R. Manchester and A.~V. Savkin, ``Circular-navigation-guidance law for
  precision missile/target engagements,'' {\em Journal of Guidance, Control,
  and Dynamics}, vol.~29, no.~2, pp.~314--320, 2006.

\bibitem{kim2011minimum}
K.~B. Kim and B.~K. Kim, ``Minimum-time trajectory for three-wheeled
  omnidirectional mobile robots following a bounded-curvature path with a
  referenced heading profile,'' {\em IEEE Transactions on Robotics}, vol.~27,
  no.~4, pp.~800--808, 2011.

\bibitem{MSF08}
I.~R. Manchester, A.~V. Savkin, and F.~A. Faruqi, ``Method for
  optical-flow-based precision missile guidance,'' {\em IEEE Transactions on
  Aerospace and Electronic Systems}, vol.~44, no.~3, pp.~835--851, 2008.

\bibitem{piazzi2007mmb}
A.~Piazzi, C.~G.~L. Bianco, and M.~Romano, ``$\eta^3$-splines for the smooth
  path generation of wheeled mobile robots,'' {\em IEEE Transactions on
  Robotics}, vol.~23, no.~5, pp.~1089--1095, 2007.

\bibitem{ST10}
A.~V. Savkin and H.~Teimoori, ``Bearings-only guidance of a unicycle-like
  vehicle following a moving target with a smaller minimum turning radius,''
  {\em IEEE Transactions on Automatic Control}, vol.~55, no.~10,
  pp.~2390--2395, 2010.

\bibitem{zarchan2012tactical}
P.~Zarchan, {\em Tactical and strategic missile guidance}.
\newblock American Institute of Aeronautics and Astronautics, 2012.

\bibitem{shima2011intercept}
T.~Shima, ``Intercept-angle guidance,'' {\em Journal of Guidance Control and
  Dynamics}, vol.~34, no.~2, p.~484, 2011.

\bibitem{gottlieb2017multi}
Y.~Gottlieb, J.~Manathara, and T.~Shima, ``Multi-target motion planning amidst
  obstacles for autonomous aerial and ground vehicles,'' {\em Journal of
  Intelligent \& Robotic Systems}, pp.~1--22, 2017.

\bibitem{yang2002optimal}
G.~Yang and V.~Kapila, ``Optimal path planning for unmanned air vehicles with
  kinematic and tactical constraints,'' in {\em the 41st IEEE Conference on
  Decision and Control}, vol.~2, pp.~1301--1306, IEEE, 2002.

\bibitem{struik2012lectures}
D.~J. Struik, {\em Lectures on classical differential geometry}.
\newblock Courier Corporation, 2012.

\bibitem{ren2014data}
X.~Ren, W.~Liang, and W.~Xu, ``Data collection maximization in renewable sensor
  networks via time-slot scheduling,'' {\em IEEE Transactions on Computers},
  vol.~64, no.~7, pp.~1870--1883, 2015.

\bibitem{VJ15}
P.~Vana and J.~Faigl, ``On the {Dubins} traveling salesman problem with
  neighborhoods,'' in {\em IEEE/RSJ International Conference on Intelligent
  Robots and Systems}, pp.~4029--4034, IEEE, 2015.

\bibitem{FGM82}
A.~M. Frieze, G.~Galbiati, and F.~Maffioli, ``On the worst-case performance of
  some algorithms for the asymmetric traveling salesman problem,'' {\em
  Networks}, vol.~12, no.~1, pp.~23--39, 1982.

\bibitem{Dijkstra}
E.~W. Dijkstra, ``A note on two problems in connexion with graphs,'' {\em
  Numerische Mathematik}, vol.~1, no.~1, pp.~269--271, 1959.

\bibitem{energy_consumption}
F.~El-Moukaddem, E.~Torng, G.~Xing, and S.~Kulkarni, ``Mobile relay
  configuration in data-intensive wireless sensor networks,'' {\em IEEE
  Transactions on Mobile Computing}, vol.~12, no.~2, pp.~261--273, 2013.

\bibitem{PLZQ10}
Y.~Peng, Z.~Li, W.~Zhang, and D.~Qiao, ``Prolonging sensor network lifetime
  through wireless charging,'' in {\em Real-time Systems Symposium},
  pp.~129--139, IEEE, 2010.

\bibitem{ZWL12}
S.~Zhang, J.~Wu, and S.~Lu, ``Collaborative mobile charging for sensor
  networks,'' in {\em The 9th International Conference on Mobile Adhoc and
  Sensor Systems}, pp.~84--92, IEEE, 2012.

\bibitem{DWCXL14}
H.~Dai, X.~Wu, G.~Chen, L.~Xu, and S.~Lin, ``Minimizing the number of mobile
  chargers for large-scale wireless rechargeable sensor networks,'' {\em
  Computer Communications}, vol.~46, pp.~54 -- 65, 2014.

\bibitem{kurs2007wireless}
A.~Kurs, A.~Karalis, R.~Moffatt, J.~D. Joannopoulos, P.~Fisher, and
  M.~Solja{\v{c}}i{\'c}, ``Wireless power transfer via strongly coupled
  magnetic resonances,'' {\em Science}, vol.~317, no.~5834, pp.~83--86, 2007.

\bibitem{IMM04}
T.~Inanc, M.~K. Muezzinoglu, K.~Misovec, and R.~M. Murray, ``Framework for
  low-observable trajectory generation in presence of multiple radars,'' {\em
  Journal of Guidance, Control, and Dynamics}, vol.~31, no.~6, pp.~1740--1749,
  2008.

\bibitem{ZAB1}
M.~Zabarankin, S.~Uryasev, and R.~Murphey, ``Aircraft routing under the risk of
  detection,'' {\em Naval Research Logistics (NRL)}, vol.~53, no.~8,
  pp.~728--747, 2006.

\bibitem{BA1}
B.~Jiang, A.~Bishop, B.~Anderson, and S.~P. Drake, ``Path planning for
  minimizing detection,'' in {\em the 19th IFAC World Congress Cape Town, South
  Africa}, pp.~10200--10206, 2014.

\bibitem{BA2}
B.~Jiang, A.~Bishop, B.~Anderson, and S.~P. Drake, ``Optimal path planning and
  sensor placement for mobile target detection,'' {\em Automatica}, vol.~60,
  pp.~127--139, 2015.

\bibitem{KH03}
J.~Kim and J.~P. Hespanha, ``Discrete approximations to continuous
  shortest-path: Application to minimum-risk path planning for groups of
  {UAVs},'' in {\em the 42nd IEEE Conference on Decision and Control}, vol.~2,
  pp.~1734--1740, IEEE, 2003.

\bibitem{MOR}
W.~M. Carlyle, J.~O. Royset, and R.~K. Wood, ``Routing military aircraft with a
  constrained shortest-path algorithm,'' {\em Military Operations Research},
  vol.~14, no.~3, pp.~31--52, 2009.

\bibitem{ZAB2}
R.~Murphey, S.~Uryasev, and M.~Zabarankin, ``Optimal path planning in a threat
  environment,'' in {\em Recent Developments in Cooperative Control and
  Optimization}, pp.~349--406, Springer, 2004.

\bibitem{KF99}
A.~Kolmogorov and S.~Fomin, {\em Elements of the theory of functions and
  functional analysis}.
\newblock Courier Dover, N. Chemsford, MA, 1999.

\bibitem{DUB57}
L.~E. Dubins, ``On curves of minimal length with a constraint on average
  curvature, and with prescribed initial and terminal positions and tangents,''
  {\em American Journal of mathematics}, vol.~79, no.~3, pp.~497--516, 1957.

\bibitem{LA92}
Y.-H. Liu and S.~Arimoto, ``Path planning using a tangent graph for mobile
  robots among polygonal and curved obstacles: Communication,'' {\em The
  International Journal of Robotics Research}, vol.~11, no.~4, pp.~376--382,
  1992.

\bibitem{BFS}
S.~S. Skiena, {\em The Algorithm Design Manual}.
\newblock Springer Publishing Company, Incorporated, 2nd~ed., 2008.

\bibitem{llorca2011fuzzy}
D.~F. Llorca, V.~Milan{\'e}s, I.~P. Alonso, M.~Gavil{\'a}n, I.~G. Daza,
  J.~P{\'e}rez, and M.~{\'A}. Sotelo, ``Autonomous pedestrian collision
  avoidance using a fuzzy steering controller,'' {\em IEEE Transactions on
  Intelligent Transportation Systems}, vol.~12, no.~2, pp.~390--401, 2011.

\bibitem{al2014fuzzy}
A.~Al-Mayyahi and W.~Wang, ``Fuzzy inference approach for autonomous ground
  vehicle navigation in dynamic environment,'' in {\em IEEE International
  Conference on Control System, Computing and Engineering}, pp.~29--34, IEEE,
  2014.

\bibitem{li2007analytical}
J.~Li and P.~Mohapatra, ``Analytical modeling and mitigation techniques for the
  energy hole problem in sensor networks,'' {\em Pervasive and Mobile
  Computing}, vol.~3, no.~3, pp.~233--254, 2007.

\bibitem{zhao2012optimization}
M.~Zhao and Y.~Yang, ``Optimization-based distributed algorithms for mobile
  data gathering in wireless sensor networks,'' {\em IEEE Transactions on
  Mobile Computing}, vol.~11, no.~10, pp.~1464--1477, 2012.

\bibitem{mottaghi2015optimizing}
S.~Mottaghi and M.~R. Zahabi, ``Optimizing leach clustering algorithm with
  mobile sink and rendezvous nodes,'' {\em AEU-International Journal of
  Electronics and Communications}, vol.~69, no.~2, pp.~507--514, 2015.

\bibitem{cheng2013decentralized}
T.~M. Cheng and A.~V. Savkin, ``Decentralized control of mobile sensor networks
  for asymptotically optimal blanket coverage between two boundaries,'' {\em
  IEEE Transactions on Industrial Informatics}, vol.~9, no.~1, pp.~365--376,
  2013.

\bibitem{savkin2012optimal}
A.~V. Savkin, F.~Javed, and A.~S. Matveev, ``Optimal distributed blanket
  coverage self-deployment of mobile wireless sensor networks,'' {\em IEEE
  Communications Letters}, vol.~16, no.~6, pp.~949--951, 2012.

\bibitem{huang2017energy}
H.~Huang and A.~V. Savkin, ``An energy efficient approach for data collection
  in wireless sensor networks using public transportation vehicles,'' {\em
  AEU-International Journal of Electronics and Communications}, vol.~75,
  pp.~108--118, 2017.

\bibitem{huang2016optimal}
H.~Huang and A.~V. Savkin, ``Optimal path planning for a vehicle collecting
  data in a wireless sensor network,'' in {\em 35th Chinese Control
  Conference}, pp.~8460--8463, IEEE, 2016.

\bibitem{huang2017vtc}
H.~Huang and A.~V. Savkin, ``Data collection in nonuniformly deployed wireless
  sensor networks by public transportation vehicles,'' in {\em IEEE Vehicular
  Technology Conference}, IEEE, 2017.

\bibitem{xiang10}
X.~Min, S.~Wei-Ren, J.~Chang-Jiang, and Z.~Ying, ``Energy efficient clustering
  algorithm for maximizing lifetime of wireless sensor networks,'' {\em
  AEU-International Journal of Electronics and Communications}, vol.~64, no.~4,
  pp.~289--298, 2010.

\bibitem{yu12}
J.~Yu, Y.~Qi, G.~Wang, and X.~Gu, ``A cluster-based routing protocol for
  wireless sensor networks with nonuniform node distribution,'' {\em
  AEU-International Journal of Electronics and Communications}, vol.~66, no.~1,
  pp.~54--61, 2012.

\bibitem{li2013coca}
H.~Li, Y.~Liu, W.~Chen, W.~Jia, B.~Li, and J.~Xiong, ``Coca: Constructing
  optimal clustering architecture to maximize sensor network lifetime,'' {\em
  Computer Communications}, vol.~36, no.~3, pp.~256--268, 2013.

\bibitem{Shokouhifar15}
M.~Shokouhifar and A.~Jalali, ``A new evolutionary based application specific
  routing protocol for clustered wireless sensor networks,'' {\em
  AEU-International Journal of Electronics and Communications}, vol.~69, no.~1,
  pp.~432--441, 2015.

\bibitem{experiment04}
R.~Szewczyk, A.~Mainwaring, J.~Polastre, J.~Anderson, and D.~Culler, ``An
  analysis of a large scale habitat monitoring application,'' in {\em the 2nd
  International Conference on Embedded Networked Sensor Systems}, pp.~214--226,
  ACM, 2004.

\bibitem{seed08}
J.~Haupt, W.~U. Bajwa, M.~Rabbat, and R.~Nowak, ``Compressed sensing for
  networked data,'' {\em IEEE Signal Processing Magazine}, vol.~25, no.~2,
  pp.~92--101, 2008.

\bibitem{opt_number14}
D.~Kumar, ``Performance analysis of energy efficient clustering protocols for
  maximising lifetime of wireless sensor networks,'' {\em IET Wireless Sensor
  Systems}, vol.~4, no.~1, pp.~9--16, 2014.

\bibitem{kmedian}
V.~Arya, N.~Garg, R.~Khandekar, A.~Meyerson, K.~Munagala, and V.~Pandit,
  ``Local search heuristics for k-median and facility location problems,'' {\em
  SIAM Journal on Computing}, vol.~33, no.~3, pp.~544--562, 2004.

\bibitem{huang17CS}
H.~Huang and A.~V. Savkin, ``The cluster based compressive data collection for
  wireless sensor networks with mobile sinks,'' {\em submitted to IEEE Internet
  of Things Journal}.

\bibitem{Liang13CSS}
L.~He, J.~Pan, and J.~Xu, ``A progressive approach to reducing data collection
  latency in wireless sensor networks with mobile elements,'' {\em IEEE
  Transactions on Mobile Computing}, vol.~12, no.~7, pp.~1308--1320, 2013.

\bibitem{kumar13delay}
A.~K. Kumar, K.~M. Sivalingam, and A.~Kumar, ``On reducing delay in mobile data
  collection based wireless sensor networks,'' {\em Wireless Networks},
  vol.~19, no.~3, pp.~285--299, 2013.

\bibitem{path_opt_14}
H.~Salarian, K.-W. Chin, and F.~Naghdy, ``An energy-efficient mobile-sink path
  selection strategy for wireless sensor networks,'' {\em IEEE Transactions on
  Vehicular Technology}, vol.~63, no.~5, pp.~2407--2419, 2014.

\bibitem{tang2015dellat}
J.~Tang, H.~Huang, S.~Guo, and Y.~Yang, ``Dellat: Delivery latency minimization
  in wireless sensor networks with mobile sink,'' {\em Journal of Parallel and
  Distributed Computing}, vol.~83, pp.~133--142, 2015.

\bibitem{twinroute09}
R.~Wohlers, N.~Trigoni, R.~Zhang, and S.~Ellwood, ``Twinroute: Energy-efficient
  data collection in fixed sensor networks with mobile sinks,'' in {\em the
  10th International Conference on Mobile Data Management: Systems, Services
  and Middleware}, pp.~192--201, IEEE, 2009.

\bibitem{stashing15}
H.~Lee, M.~Wicke, B.~Kusy, O.~Gnawali, and L.~Guibas, ``Predictive data
  delivery to mobile users through mobility learning in wireless sensor
  networks,'' {\em IEEE Transactions on Vehicular Technology}, vol.~64, no.~12,
  pp.~5831--5849, 2015.

\bibitem{HCDD06}
C.-J. Lin, P.-L. Chou, and C.-F. Chou, ``{HCDD}: hierarchical cluster-based
  data dissemination in wireless sensor networks with mobile sink,'' in {\em
  International Conference on Wireless and Mobile Computing, Networking and
  Communications}, pp.~1189--1194, ACM, 2006.

\bibitem{cluster_10}
B.~Nazir and H.~Hasbullah, ``Mobile sink based routing protocol ({MSRP}) for
  prolonging network lifetime in clustered wireless sensor network,'' in {\em
  International Conference on Computer Applications and Industrial
  Electronics}, pp.~624--629, IEEE, 2010.

\bibitem{TTDD05}
H.~Luo, F.~Ye, J.~Cheng, S.~Lu, and L.~Zhang, ``{TTDD}: Two-tier data
  dissemination in large-scale wireless sensor networks,'' {\em Wireless
  Networks}, vol.~11, no.~1-2, pp.~161--175, 2005.

\bibitem{grid_09}
K.~Kweon, H.~Ghim, J.~Hong, and H.~Yoon, ``Grid-based energy-efficient routing
  from multiple sources to multiple mobile sinks in wireless sensor networks,''
  in {\em the 4th International Symposium on Wireless Pervasive Computing},
  pp.~1--5, IEEE, 2009.

\bibitem{VGDRA15}
A.~W. Khan, A.~H. Abdullah, M.~A. Razzaque, and J.~I. Bangash, ``{VGDRA}: a
  virtual grid-based dynamic routes adjustment scheme for mobile sink-based
  wireless sensor networks,'' {\em IEEE Sensors Journal}, vol.~15, no.~1,
  pp.~526--534, 2015.

\bibitem{huang2017umdpc}
H.~Huang, A.~V. Savkin, and C.~Huang, ``{I-UMDPC}: The improved-unusual message
  delivery path construction for wireless sensor networks with mobile sinks,''
  {\em IEEE Internet of Things Journal}, vol.~4, no.~5, pp.~1528--1536, 2017.

\bibitem{huang2017delay}
H.~Huang, A.~V. Savkin, and C.~Huang, ``Delay-aware data collection in wireless
  sensor networks with mobile nodes,'' in {\em 36th Chinese Control
  Conference}, pp.~8875--8878, IEEE, 2017.

\bibitem{maxmind}
A.~D. Amis, R.~Prakash, T.~H. Vuong, and D.~T. Huynh, ``Max-min d-cluster
  formation in wireless ad hoc networks,'' in {\em the 19th Annual Joint
  Conference IEEE Computer and Communications Societies}, vol.~1, pp.~32--41,
  IEEE, 2000.

\bibitem{lee2015predictive}
H.~Lee, M.~Wicke, B.~Kusy, O.~Gnawali, and L.~Guibas, ``Predictive data
  delivery to mobile users through mobility learning in wireless sensor
  networks,'' {\em IEEE Transactions on Vehicular Technology}, vol.~64, no.~12,
  pp.~5831--5849, 2015.

\bibitem{H_usa}
``Part 107 of the federal aviation regulations,'' 2016.

\bibitem{H_australia}
``Civil aviation legislation amendment (part 101) regulation,'' 2016.

\bibitem{savkin2001problem}
A.~V. Savkin, R.~J. Evans, and E.~Skafidas, ``The problem of optimal robust
  sensor scheduling,'' {\em Systems \& Control Letters}, vol.~43, no.~2,
  pp.~149--157, 2001.

\bibitem{savkin2002hybrid}
A.~V. Savkin and R.~J. Evans, {\em Hybrid dynamical systems: controller and
  sensor switching problems}.
\newblock Birkhauser, Boston, 2002.

\bibitem{ding2016performance}
M.~Ding, P.~Wang, D.~L{\'o}pez-P{\'e}rez, G.~Mao, and Z.~Lin, ``Performance
  impact of los and nlos transmissions in dense cellular networks,'' {\em IEEE
  Transactions on Wireless Communications}, vol.~15, no.~3, pp.~2365--2380,
  2016.

\bibitem{backhaul2012}
S.~Rohde and C.~Wietfeld, ``Interference aware positioning of aerial relays for
  cell overload and outage compensation,'' in {\em IEEE Vehicular Technology
  Conference}, pp.~1--5, IEEE, 2012.

\bibitem{gupta2016survey}
L.~Gupta, R.~Jain, and G.~Vaszkun, ``Survey of important issues in {UAV}
  communication networks,'' {\em IEEE Communications Surveys \& Tutorials},
  vol.~18, no.~2, pp.~1123--1152, 2016.

\bibitem{azade2017}
A.~Fotouhi, M.~Ding, and M.~Hassan, ``Understanding autonomous drone
  maneuverability for internet of things applications,'' in {\em WoWMoM
  Workshops}, IEEE, 2017.

\bibitem{3GPP}
``{3GPP TR} 36.828: Further enhancements to {LTE} time division duplex ({TDD})
  for ({DL-UL}) interference management and traffic adaptation,'' 2012.

\bibitem{crowdsensing11}
R.~K. Ganti, F.~Ye, and H.~Lei, ``Mobile crowdsensing: current state and future
  challenges,'' {\em IEEE Communications Magazine}, vol.~49, no.~11, 2011.

\bibitem{feige1998threshold}
U.~Feige, ``A threshold of ln \textit{n} for approximating set cover,'' {\em
  Journal of the ACM}, vol.~45, no.~4, pp.~634--652, 1998.

\bibitem{korte2012combinatorial}
B.~Korte, J.~Vygen, B.~Korte, and J.~Vygen, {\em Combinatorial optimization},
  vol.~2.
\newblock Springer, 2012.

\bibitem{chen2013and}
T.~Chen, M.~A. Kaafar, and R.~Boreli, ``The where and when of finding new
  friends: Analysis of a location-based social discovery network,'' in {\em
  ICWSM}, 2013.

\bibitem{scale15}
D.~L{\'o}pez-P{\'e}rez, M.~Ding, H.~Claussen, and A.~H. Jafari, ``Towards 1
  {Gbps/UE} in cellular systems: Understanding ultra-dense small cell
  deployments,'' {\em IEEE Communications Surveys \& Tutorials}, vol.~17,
  no.~4, pp.~2078--2101, 2015.

\bibitem{abdelnasser2016resource}
A.~Abdelnasser and E.~Hossain, ``Resource allocation for an {OFDMA} {Cloud-RAN}
  of small cells underlaying a macrocell,'' {\em IEEE Transactions on Mobile
  Computing}, vol.~15, no.~11, pp.~2837--2850, 2016.

\bibitem{jafari2015study}
A.~H. Jafari, D.~L{\'o}pez-P{\'e}rez, M.~Ding, and J.~Zhang, ``Study on
  scheduling techniques for ultra dense small cell networks,'' in {\em IEEE
  Vehicular Technology Conference}, pp.~1--6, IEEE, 2015.

\bibitem{huang17drone}
H.~Huang, A.~V. Savkin, M.~Ding, and M.~A. Kaafar, ``Optimized deployment of
  autonomous drones to improve user experience in cellular networks,'' {\em
  submitted to IEEE Transactions on Mobile Computing}.

\end{thebibliography}
\end{document}